\documentclass[11pt]{article}
\usepackage{amsmath}
\usepackage{amsthm}
\usepackage{amssymb}
\usepackage{algorithm}
\usepackage{subfig}
\usepackage{color}
\usepackage[english]{babel}
\usepackage{graphicx}
\usepackage{wrapfig,epsfig}
\usepackage{epstopdf}
\usepackage{url}
\usepackage{graphicx}
\usepackage{color}
\usepackage{epstopdf}
\usepackage{algpseudocode}
\usepackage{scrextend}
\usepackage[T1]{fontenc}
\usepackage{bbm}
\usepackage{comment}

\usepackage{tikz}
\usepackage{hyperref}  
\hypersetup{colorlinks=true,citecolor=blue,linkcolor=blue} 
\usetikzlibrary{arrows}
\usepackage[margin=1in]{geometry}
\graphicspath{{./figs/}}

\author{
  Zhao Song\thanks{Work done while visiting IBM Almaden, and supported in part by UTCS TAship (CS361 Spring 17 Introduction to Computer Security).}\\
  \texttt{zhaos@utexas.edu}\\
  UT-Austin
  \and
  David P. Woodruff\\
  \texttt{dpwoodru@us.ibm.com}\\
  IBM Almaden
  \and
  Peilin Zhong\thanks{Supported in part by Simons Foundation, and NSF CCF-1617955.}\\
  \texttt{peilin.zhong@columbia.edu}\\
  Columbia University
}

\date{}
\title{Relative Error Tensor Low Rank Approximation}
\newtheorem{theorem}{Theorem}[section]
\newtheorem{lemma}[theorem]{Lemma}
\newtheorem{definition}[theorem]{Definition}

\newtheorem{corollary}[theorem]{Corollary}
\newtheorem{conjecture}[theorem]{Conjecture}
\newtheorem{assumption}[theorem]{Assumption}

\newtheorem{fact}[theorem]{Fact}
\newtheorem{remark}[theorem]{Remark}
\newtheorem{claim}[theorem]{Claim}

\newtheorem{hypothesis}[theorem]{Hypothesis}

\newcommand{\wh}{\widehat}
\newcommand{\wt}{\widetilde}
\newcommand{\ov}{\overline}
\newcommand{\eps}{\epsilon}

\newcommand{\LHS}{\mathrm{LHS}}

\renewcommand{\varepsilon}{\epsilon}
\renewcommand{\tilde}{\wt}
\renewcommand{\hat}{\wh}

\DeclareMathOperator*{\E}{{\bf {E}}}

\DeclareMathOperator{\OPT}{OPT}

\DeclareMathOperator{\poly}{poly}

\DeclareMathOperator{\nnz}{nnz}

\DeclareMathOperator{\rank}{rank}
\DeclareMathOperator{\tucker}{tucker}
\DeclareMathOperator{\train}{train}

\DeclareMathOperator{\constraints}{constraints}
\DeclareMathOperator{\degree}{degree}
\DeclareMathOperator{\variables}{variables}

\DeclareMathOperator{\vect}{vec}
\DeclareMathOperator{\tr}{tr}
\DeclareMathOperator{\RAM}{RAM}

\newcommand{\SAT}{{\sf 3SAT}~}
\newcommand{\ESAT}{{\sf E3SAT}~}
\newcommand{\ESATB}{{\sf E3SAT(B)}~}
\newcommand{\MAX}{{\sf MAX}}

\newcommand{\CNF}{{\sf CNF}~}

\newcommand{\kSUM}{{\sf k-SUM}~}
\newcommand{\kClique}{{\sf k-Clique}~}
\newcommand{\ETH}{{\sf ETH}~}
\newcommand{\NP}{{\bf{NP}}}
\newcommand{\RP}{{\bf{RP}}}

\makeatletter
\newcommand*{\RN}[1]{\expandafter\@slowromancap\romannumeral #1@}
\makeatother

\usepackage{lineno}

\usepackage{etoolbox}

\makeatletter

\newcommand{\define}[4][ignore]{%
  \ifstrequal{#1}{ignore}{}{
  \@namedef{thmtitle@#2}{#1}}%
  \@namedef{thm@#2}{#4}%
  \@namedef{thmtypen@#2}{lemma}%
  \newtheorem{thmtype@#2}[theorem]{#3}%
  \newtheorem*{thmtypealt@#2}{#3~\ref{#2}}%
}

\newcommand{\state}[1]{%
  \@namedef{curthm}{#1}
  \@ifundefined{thmtitle@#1}{
  \begin{thmtype@#1}
    }{
  \begin{thmtype@#1}[\@nameuse{thmtitle@#1}]
  }
    \label{#1}
    \@nameuse{thm@#1}
  \end{thmtype@#1}
  \@ifundefined{thmdone@#1}{
  \@namedef{thmdone@#1}{stated}%
  }{}
}

\newcommand{\restate}[1]{%
  \@namedef{curthm}{#1}
  \@ifundefined{thmtitle@#1}{
    \begin{thmtypealt@#1}
    }{
  \begin{thmtypealt@#1}[\@nameuse{thmtitle@#1}]
  }
    \@nameuse{thm@#1}
  \end{thmtypealt@#1}
  \@ifundefined{thmdone@#1}{
  \@namedef{thmdone@#1}{stated}%
  }{}
}

\newcommand{\thmlabel}[1]{
  \@ifundefined{thmdone@\@nameuse{curthm}}{\label{#1}
    }{\tag*{\eqref{#1}}}
}
\makeatother

\begin{document}

\begin{titlepage}
  \maketitle
  \begin{abstract}
%

We consider relative error low rank approximation of {\it tensors} with respect to the Frobenius norm. Namely, given an order-$q$ tensor $A \in \mathbb{R}^{\prod_{i=1}^q n_i}$, output a rank-$k$ tensor $B$ for which $\|A-B\|_F^2 \leq (1+\epsilon) \OPT$, where $\OPT = \inf_{\textrm{rank-}k~A'} \|A-A'\|_F^2$. Despite much success on obtaining relative error low rank approximations for matrices, no such results were known for tensors for arbitrary $(1+\epsilon)$-approximations. 
One structural issue is that there may be no rank-$k$ tensor $A_k$ achieving the above infinum. Another, 
computational issue, is that an efficient relative error low rank approximation algorithm for tensors would allow one to compute the rank of a tensor, which is NP-hard. We bypass these two issues via (1) bicriteria and (2) parameterized complexity solutions:
\begin{enumerate}
\item We give an algorithm  which outputs a rank $k' = O((k/\epsilon)^{q-1})$ tensor $B$ for which $\|A-B\|_F^2 \leq (1+\epsilon) \OPT$ in $\nnz(A) + n \cdot \poly(k/\epsilon)$ time in the real $\RAM$ model,
whenever either $A_k$ exists or $\OPT > 0$. Here $\nnz(A)$ denotes the number
of non-zero entries in $A$. If both $A_k$ does not exist and $\OPT = 0$, then $B$ instead satisfies $\|A-B\|_F^2 < \gamma$, where $\gamma$ is any positive,
arbitrarily small function of $n$. 
\item We give an algorithm for any $\delta >0$ which outputs a rank $k$ tensor $B$ for which $\|A-B\|_F^2 \leq (1+\epsilon) \OPT$ and runs in $ ( \nnz(A) + n \poly(k/\epsilon) + \exp(k^2/\epsilon) ) \cdot n^\delta$ time in the unit cost $\RAM$ model, whenever $\OPT > 2^{- O(n^\delta)}$ and there is a rank-$k$ tensor $B = \sum_{i=1}^k u_i \otimes v_i \otimes w_i$ for which $\|A-B\|_F^2 \leq (1+\epsilon/2)\OPT$ and $\|u_i\|_2, \|v_i\|_2, \|w_i\|_2 \leq 2^{O(n^\delta)}$. If $\OPT \leq 2^{- \Omega(n^\delta)}$, then $B$ instead satisfies $\|A-B\|_F^2 \leq 2^{- \Omega(n^\delta)}$.
\end{enumerate}
Our first result is polynomial time, and in fact input sparsity time, in $n, k,$ and $1/\epsilon$, for any $k\geq 1$ and any $0 < \epsilon <1$, while our second result is fixed parameter tractable in $k$ and $1/\epsilon$. For outputting a rank-$k$ tensor, or even a bicriteria solution with rank-$Ck$ for a certain constant $C > 1$, we show a $2^{\Omega(k^{1-o(1)})}$ time lower bound under the Exponential Time Hypothesis. 

Our results are based on an ``iterative existential argument'', and also give the first relative error low rank approximations for tensors for a large number of error measures for which nothing was known. In particular, we give the first relative error approximation algorithms on tensors for: column row and tube subset selection, entrywise $\ell_p$-low rank approximation for $1 \leq p < 2$, low rank approximation with respect to sum of Euclidean norms of faces or tubes, weighted low rank approximation, and low rank approximation in distributed and streaming models. We also obtain several new results for matrices, such as $\nnz(A)$-time CUR decompositions, improving the previous $\nnz(A)\log n$-time CUR decompositions, which may be of independent interest.

  \end{abstract}
  \thispagestyle{empty}
\end{titlepage}

\newpage
{\hypersetup{linkcolor=black}
\tableofcontents
}

\newpage


\newcommand{\eqdef}{\mathbin{\stackrel{\rm def}{=}}}
\section{Introduction}\label{sec:intro}
Low rank approximation of matrices is one of the most well-studied problems in randomized numerical linear algebra. Given an $n \times d$ matrix $A$ with real-valued entries, we want to output a rank-$k$ matrix $B$ for which $\|A-B\|$ is small, under a given norm. While this problem can be solved exactly using the singular value decomposition for some norms like the spectral and Frobenius norms, the time complexity is still $\min(nd^{\omega-1}, dn^{\omega-1})$, where $\omega \approx 2.376$ is the exponent of matrix multiplication \cite{s69,cw87,w12}.
This time complexity is prohibitive when $n$ and $d$ are large. By now there are a number of approximation algorithms for this problem, with the Frobenius norm
\footnote{Recall the Frobenius norm $\|A\|_F$ of a matrix $A$ is
  $ (\sum_{i = 1}^n \sum_{j=1}^d A_{i,j}^2  )^{1/2}$.}
  being one of the most common error measures.
  Initial solutions \cite{fkv04,am07} to this problem were based on sampling and achieved additive error in terms of $\epsilon \|A\|_F$, where $\epsilon > 0$ is an approximation parameter, which can be arbitrarily larger than the optimal cost $\OPT = \min_{\textrm{rank-}k \ B} \|A-B\|_F^2$. Since then a number of solutions based on the technique of oblivious sketching \cite{s06,cw13,mm13,nn13} as well as sampling based on non-uniform distributions \cite{dmm06,dmm06b,dmm08,dmmw12}, have been proposed which achieve the stronger notion of {\it relative error}, namely, which output a rank-$k$ matrix $B$ for which $\|A-B\|_F^2 \leq (1+\epsilon) \OPT$ with high probability. It is now known how to output a factorization of such a $B = U \cdot V$, where $U$ is $n \times k$ and $V$ is $k \times d$, in $\nnz(A) + (n+d) \poly(k/\epsilon)$ time \cite{cw13,mm13,nn13}. 
  Such an algorithm is optimal, up to the $\poly(k/\epsilon)$ factor, as any algorithm achieving relative error must read almost all of the entries.

  Tensors are often more useful than matrices for capturing higher order relations in data. Computing low rank factorizations of approximations of tensors is the primary task of interest in a number of applications, such as in psychology\cite{k83}, chemometrics \cite{p00,sbg04}, neuroscience \cite{aabby07,kb09,clkgar15}, computational biology \cite{cv15,sc15}, natural language processing \cite{cyym14,lzbj14,lzmb15,bnrv15}, computer vision \cite{vt02,wa03,sh05,hsa05,hd08,adtl09,ply10,lfclly16,clz17}, computer graphics \cite{vt04,wwsya05,vas09}, security \cite{acky05,acy06,kb06}, cryptography \cite{fs99,s12,kyfd15,shwpjj16} data mining \cite{ks08,rs10,kabo10,m11}, machine learning applications such as learning hidden Markov models, reinforcement learning, community detection, multi-armed bandit, ranking models, neural network, Gaussian mixture models and Latent Dirichlet allocation
  \cite{mr05,afhkl12,hk13,galb13,absv14,aghkt14,aghk14,bcv14,jo14,ghk15,pbl15,jsa15,ala16,agmr17,zsjbd17},
programming languages \cite{rtp16},
  signal processing \cite{w94,dd98,c09,cmdz15}, and other applications \cite{ycs11,lmwy13,os14,zczj14,stls14,ycs16,rnss16}.

Despite the success for matrices, the situation for order-$q$ tensors for $q > 2$ is much less understood. There are a number of works based on alternating minimization \cite{cc70,h70,fmps13,ft15,zg01,bs15} gradient descent or Newton methods \cite{es09,zg01}, methods based on the Higher-order SVD (HOSVD) \cite{lmv00} which provably incur $\Omega(\sqrt{n})$-inapproximability for Frobenius norm error \cite{lmv00a}, the power method or orthogonal iteration method \cite{lmv00a}, additive error guarantees in terms of the flattened (unfolded) tensor rather than the original tensor \cite{mmd08}, tensor trains \cite{o11}, the tree Tucker decomposition \cite{o09}, or methods specialized to orthogonal tensors \cite{km11,aghkt14,mhg15,wtsa15,wa16,swz16}. There are also a number of works on the problem of tensor completion, that is, recovering a low rank tensor from missing entries \cite{wm01,akdm10,tshk11,lmwy13,mh13,j14,bm16}.
There is also another line of work using the sum of squares (SOS) technique to study tensor problems \cite{bks15,gm15,hss15,hsss16,mss16,ps17,ts17}, other recent work on tensor PCA \cite{a12a,a12b,rm14,jmz15,adgm16,zx17}, and work applying smoothed analysis to tensor decomposition \cite{bcmv14}.
Several previous works also consider more robust norms than the Frobenius norm for tensors, e.g., the $R_1$ norm ($\ell_1$-$\ell_2$-$\ell_2$ norm in our work) \cite{hd08}, $\ell_1$-PCA \cite{ply10}, entry-wise $\ell_1$ regularization \cite{ggh14}, M-estimator loss \cite{yfs16}, weighted approximation \cite{p97,tk11,lrhg13}, tensor-CUR \cite{ost08,mmd08,cc10,fmmn11,ft15}, or robust tensor PCA \cite{gq14,lfclly16,clz17}. 

Some of the above works, such as ones based on the tensor power method or alternating minimization,
require incoherence or orthogonality assumptions. Others, such as those based on the simultaneous
SVD, require an assumption on the minimum singular value. See the monograph of Moitra \cite{m14} for further discussion. Unlike the situation for matrices,
there is no work for tensors that is able
to achieve the following natural relative error guarantee:
given a $q$-th order tensor $A \in \mathbb{R}^{n^{\otimes q}}$ and an arbitrary accuracy
parameter $\epsilon > 0$, 
output a rank-$k$
tensor $B$ for which 
\begin{equation}\label{eqn:guarantee}
\|A-B\|_F^2 \leq (1+\epsilon)\OPT,
\end{equation}
where $\OPT = \inf_{\textrm{rank-}k \ B'}\|A-B'\|_F^2$, and where recall the rank
of a tensor $B$ is the minimal integer $k$ for which $B$ can be expressed
as $\sum_{i=1}^k u_i \otimes v_i \otimes w_i$.
A third order tensor, for example, has rank which is an 
integer in $\{0, 1, 2, \ldots, n^2\}$. We note that \cite{bcv14} is able to achieve 
a relative error $5$-approximation for third order tensors, and an $O(q)$-approximation
for $q$-th order tensors, though it cannot achieve a $(1+\epsilon)$-approximation. 
We compare our work to \cite{bcv14} in Section \ref{sec:comparison} below.  

For notational simplicity,
we will start by assuming third order tensors with all dimensions of equal size, but
we extend all of our main theorems below to tensors of any constant order $q > 3$ and dimensions of different sizes.

The first caveat regarding (\ref{eqn:guarantee}) for tensors is that an optimal rank-$k$ solution
may not even exist! This is a well-known problem for tensors (see, e.g., \cite{khl89,p00,kds08,ste06,ste08} and more details in section 4 of \cite{sl08}),
for which for any rank-$k$
tensor $B$, there always exists another rank-$k$ tensor $B'$ for which
$\|A-B'\|_F^2 < \|A-B\|_F^2$. If $\OPT = 0$, then in this case for any rank-$k$ tensor $B$,
necessarily $\|A-B\|_F^2 > 0$, and so (\ref{eqn:guarantee}) cannot be satisfied.
This fact was known to algebraic geometers as early as the 19th century, which they refer
to as the fact that the locus of $r$-th secant planes to a Segre variety may not define a (closed) algebraic variety \cite{sl08,l12}. It is also known as the phenomenon underlying the concept of {\it border rank}\footnote{\url{https://en.wikipedia.org/wiki/Tensor\_rank\_decomposition\#Border\_rank}}\cite{b80,b86,bcs97,k98,l06}.
 In this case
it is natural to allow the algorithm to output an arbitrarily small $\gamma > 0$ amount of additive error.
Note that unlike several additive error algorithms for matrices, the additive error here can
in fact be an arbitrarily small positive function of $n$. If, however, $\OPT > 0$, then for any
$\epsilon > 0$, there exists a rank-$k$ tensor $B$ for which $\|A-B\|_F^2 \leq (1+\epsilon)\OPT$,
and in this case we should still require the algorithm to output a relative-error solution. If an optimal
rank-$k$ solution $B$ exists, then as for matrices, it is natural to require the algorithm to output a
relative-error solution.

Besides the above definitional issue,
a central reason that (\ref{eqn:guarantee}) has not been achieved is that
computing the rank of a third order tensor is well-known to be NP-hard \cite{h90,hl13}. Thus, if
one had such a polynomial time procedure for solving the problem above,
one could determine the rank of $A$ by running the procedure on each
$k \in \{0, 1, 2, \ldots, n^2\}$, and check for the first value of $k$
for which $\|A-B\|_F^2 = 0$, thus determining the rank of $A$.
However, it is unclear if approximating the tensor rank is hard. This question will also be answered in this work. 

%
The main question which we address is how to define a meaningful notion of (\ref{eqn:guarantee})
for the case of tensors and whether it is possible to obtain provably efficient algorithms which achieve this
guarantee, without any assumptions on the tensor itself.
Besides (\ref{eqn:guarantee}), there are many other notions of
relative error for low rank approximation of matrices for which provable guarantees for tensors
are unknown, such as tensor CURT, $R_1$ norm, and the weighted and $\ell_1$ norms mentioned above.
Our goal is to provide a general technique to obtain algorithms for many of these variants as well.
\subsection{Our Results}\label{sec:our_results}
To state our results, we first consider the case when a rank-$k$ solution $A_k$ exists, that is, there exists
a rank-$k$ tensor $A_k$ for which $\|A-A_k\|_F^2 = \OPT$.

We first give a poly$(n, k, 1/\epsilon)$-time $(1+\epsilon)$-relative error approximation algorithm for any $0 < \epsilon <1$ and any $k\geq 1$,
  but allow the output tensor
  $B$ to be of rank $O((k/\epsilon)^2)$ (for general $q$-order tensors, the output rank is $O((k/\epsilon)^{q-1})$,
  whereas we measure the cost of $B$ with respect to rank-$k$ tensors.
  Formally, $\|A-B\|_F^2 \leq (1+\epsilon) \|A-A_k\|_F^2$.
  In fact, our algorithm can be implemented in $\nnz(A) + n \cdot \poly(k/\epsilon)$ time in the real-$\RAM$ model,
  where $\nnz(A)$
  is the number of non-zero entries of $A$. Such an algorithm is optimal for any relative error algorithm,
  even bicriteria ones.

  If $A_k$ does not exist, then our output $B$ instead satisfies $\|A-B\|_F^2 \leq (1+\epsilon)\OPT + \gamma$,
  where $\gamma$ is an arbitrarily small additive error. Since $\gamma$ is arbitrarily small,
  $(1+\epsilon)\OPT + \gamma$ is still a relative error whenever $\OPT > 0$.
 Our theorem is as follows.
  \begin{theorem}[A Version of Theorem \ref{thm:f_bicriteria_algorithm_bit}, bicriteria]\label{thm:bicriteria}
    Given a $3$rd order tensor $A \in \mathbb{R}^{n\times n\times n}$, if $A_k$ exists then
    there is a randomized algorithm running in
    $\nnz(A) + n \cdot \poly(k/\epsilon)$ time
    which outputs a (factorization of a) rank-$O(k^2/\epsilon^2)$ tensor $B$ for which $\|A-B\|_F^2 \leq (1+\epsilon)\|A-A_k\|_F^2$.
    If $A_k$ does not exist, then the algorithm outputs a rank-$O(k^2/\epsilon^2)$
    tensor $B$ for which $\|A-B\|_F^2 \leq (1+\epsilon)\OPT + \gamma$, where $\gamma > 0$ 
    is an arbitrarily small positive function of $n$.
    In both cases, the success probability is at least $2/3$.
    \end{theorem}
  One of the main applications of matrix low rank approximation is parameter reduction, as one can
  store the matrix using fewer parameters in factored form or more quickly multiply by the matrix if given in factored form, as well as remove directions that correspond to noise. 
  In such applications, it is not essential that the low rank approximation
  have rank exactly $k$, since one still has a significant parameter reduction with a matrix of slightly larger
  rank. This same motivation applies to tensor low rank approximation; we obtain both space and time savings
  by representing a tensor in factored form, and in such applications bicriteria applications suffice. Moreover,
  the extremely efficient $\nnz(A) + n \cdot \poly(k/\epsilon)$ time algorithm we obtain may outweigh the need for
  outputting a tensor of rank exactly $k$. Bicriteria algorithms are common for coping with hardness; 
  see e.g., results on robust low rank approximation of matrices \cite{dv07,FFSS07,cw15focs}, sparse recovery \cite{ckps16},
  clustering \cite{m152,hsu2016greedy}, and approximation algorithms more generally. 

  We note that 
  there are other applications, such as unique tensor decomposition in the method of moments, see, e.g., \cite{bcv14},
  where one may have a hard rank constraint of $k$ for the output. However, in such applications the so-called Tucker
  decomposition is still a useful dimensionality-reduction analogue of the SVD and our techniques for proving 
  Theorem \ref{thm:bicriteria} can also be used for obtaining Tucker decompositions, see Section \ref{sec:other_rank}. 

  We next consider the case when the rank parameter $k$ is small, and we try to obtain rank-$k$ solutions
  which are efficient for small values of $k$. As before, we first suppose that $A_k$ exists.

  If $A_k = \sum_{i=1}^k u_i \otimes v_i \otimes w_i$ and the norms
  $\|u_i\|_2, \|v_i\|_2,$ and $\|w_i\|_2$ are bounded by $2^{\poly(n)}$, we can return a rank-$k$ solution $B$
  for which $\|A-B\|_F^2 \leq (1+\eps)\|A-A_k\|_F^2 + 2^{-\poly(n)}$,
  in $f(k,1/\epsilon) \cdot \poly(n)$ time in the standard
  unit cost $\RAM$ model with words of size $O(\log n)$ bits.
  Thus, our algorithm is {\it fixed parameter tractable} in $k$ and
  $1/\epsilon$, and in fact remains polynomial time for any values of $k$ and $1/\epsilon$ for which
  $k^2/\epsilon = O(\log n)$. This is motivated by a number of low rank approximation applications in which
  $k$ is typically small. The additive error of $2^{-\poly(n)}$ is only needed in order to write
  down our solution $B$ in the unit cost $\RAM$ model, since in general the entries of $B$ may be irrational, even
  if the entries of $A$ are specified by $\poly(n)$ bits. If instead we only want to output an approximation
  to the value $\|A-A_k\|_F^2$, then we can output a number $Z$ for which $\OPT \leq Z \leq (1+\epsilon) \OPT$, that
  is, we do not incur additive error.

  When $A_k$ does not exist, there still exists a rank-$k$ tensor $\tilde{A}$ for which 
  $\|A-\tilde{A}\|_F^2 \leq \OPT + \gamma$. We require
  there exists such a $\tilde{A}$ for which if $\tilde{A} = \sum_{i=1}^k u_i \otimes v_i \otimes w_i$, then
  the norms $\|u_i\|_2$, $\|v_i\|_2$, and $\|w_i\|_2$ are bounded by $2^{\poly(n)}$.

  The assumption in the previous two paragraphs that the factors of $A_k$ and of $\tilde{A}$ have norm bounded by $2^{\poly(n)}$
  is necessary in certain cases, e.g., if $\OPT = 0$ and we are to write down the factors in $\poly(n)$ time.
  An abridged version of our theorem is as follows.
\begin{theorem}[Combination of Theorem~\ref{thm:f_main_algorithm} and~\ref{thm:f_main_algorithm_bit}, rank-$k$]\label{thm:smallk}
  Given a $3$rd order tensor $A \in \mathbb{R}^{n\times n\times n}$, for any $\delta >0$, if $A_k = \sum_{i=1}^k u_i \otimes v_i \otimes w_i$ exists
  and each of $\|u_i\|_2, \|v_i\|_2,$ and $\|w_i\|_2$ is bounded by $2^{O(n^\delta)}$, then there is a randomized algorithm running in
  $O( \nnz(A) + n\poly(k,1/\epsilon) + 2^{O(k^2/\epsilon)} ) \cdot n^{\delta}$ time in the unit cost $\RAM$ model with words of size $O(\log n)$ bits\footnote{The entries of $A$ are assumed to fit in $n^{\delta}$ words.},
  which outputs a (factorization of a) rank-$k$ tensor $B$ for which $\|A-B\|_F^2 \leq (1+\epsilon)\|A-A_k\|_F^2 + 2^{- O(n^\delta)}$. Further, we can output a number
  $Z$ for which $\OPT \leq Z \leq (1+\epsilon) \OPT$ in the same amount of time. When $A_k$ does not exist, if there exists a rank-$k$
  tensor $\tilde{A}$ for which $\|A-\tilde{A}\|_F^2 \leq \OPT + 2^{-O(n^\delta)}$ and $\tilde{A} = \sum_{i=1}^k u_i \otimes v_i \otimes w_i$ is
  such that the norms $\|u_i\|_2$, $\|v_i\|_2$, and $\|w_i\|_2$ are bounded by $2^{O(n^\delta)}$, then we can output a (factorization of a) rank-$k$
  tensor $\tilde{A}$ for which $\|A-\tilde{A}\|_F^2 \leq (1+\epsilon)\OPT + 2^{-O(n^\delta)}$.
\end{theorem}
Our techniques for proving Theorem \ref{thm:bicriteria} and Theorem \ref{thm:smallk}
open up avenues for many other problems in linear algebra on tensors. We now define the problems and state our results for them.

There is a long line of research on matrix column subset selection and CUR decomposition \cite{dmm08,bmd09,dr10,bdm11,fegk13,bw14,ws15,abf16,swz17} under operator, Frobenius, and entry-wise $\ell_1$ norm. It is natural to consider tensor column subset selection or tensor-CURT\footnote{T denotes the tube which is the column in $3$rd dimension of tensor.}, however most previous works either give error bounds in terms of the tensor flattenings \cite{dmm08}, assume the original tensor has certain properties \cite{ost08,ft15,tm17}, consider the exact case which assumes the tensor has low rank \cite{cc10}, or only fit a high dimensional cross-shape to the tensor rather than to all of its entries \cite{fmmn11}. Such works are not able to provide a $(1+\epsilon)$-approximation guarantee as in the matrix case without assumptions.
We consider tensor column, row, and tube subset selection, with the goal being to find three matrices: a subset $C \in \mathbb{R}^{n \times c}$ of columns of $A$, a subset $R \in \mathbb{R}^{n \times r}$ of rows of $A$, and a subset $T\in \mathbb{R}^{n\times t}$ of tubes of $A$, such that there exists a tensor $U\in \mathbb{R}^{c\times r\times t}$ for which
\begin{equation}\label{eq:intro_curt}
\| U(C,R,T) -A \|_{\xi} \leq \alpha \| A_k - A \|_{\xi} + \gamma,
\end{equation}
where $\gamma=0$ if $A_k$ exists and $\gamma=2^{-\poly(n)}$ otherwise, $\alpha>1$ is the approximation ratio, $\xi$ is either Frobenius norm or Entry-wise $\ell_1$ norm, and $ U(C,R,T) = \sum_{i=1}^c \sum_{j=1}^r \sum_{l=1}^t U_{i,j,l} \cdot C_i \otimes R_j \otimes T_l$. In tensor CURT decomposition, we also want to output $U$.

We provide a (nearly) input sparsity time algorithm for this, together with 
%
an alternative input sparsity time algorithm
which chooses slightly larger factors $C,R$, and $T$. 

%
To do this, we combine Theorem \ref{thm:bicriteria} with the following theorem which, given a factorization of a rank-$k$ tensor $B$, obtains $C$, $U$, $R$, and $T$ in terms of it:

\begin{theorem}[Combination of Theorem~\ref{thm:f_curt_algorithm_input_sparsity} and \ref{thm:f_curt_algorithm_optimal_samples}, $\| \|_F$-norm, CURT decomposition]\label{thm:intro_f_curt}
  Given a $3$rd order tensor $A\in \mathbb{R}^{n\times n \times n}$, let $k\geq 1$, and let $U_B,V_B,W_B\in \mathbb{R}^{n\times k}$ be given. There is an algorithm running in $O(\nnz(A) \log n) + \wt{O}( n^2 )\poly( k, 1/\epsilon) $ time (respectively, $O(\nnz(A)) + n \poly(k,1/\epsilon)$ time) which outputs a subset $C\in \mathbb{R}^{n\times c}$ of columns of $A$, a subset $R\in \mathbb{R}^{n\times r}$ of rows of $A$,
  a subset $T\in \mathbb{R}^{n\times t}$ of tubes of $A$, together with a tensor $U\in \mathbb{R}^{c\times r\times t}$ with $\rank(U)=k$ such that $c=r=t=O(k/\epsilon)$ (respectively, $c = r = t = O(k\log k+k/\epsilon)$), and $\|U(C,R,T) - A \|_F^2 \leq (1+\epsilon) \|U_B\otimes V_B \otimes W_B -A \|_F^2$ holds with probability at least $9/10$.
\end{theorem}

Combining Theorems \ref{thm:smallk} and \ref{thm:intro_f_curt} (with $B$ being a $(1+O(\epsilon))$-approximation
to $A$) we achieve Equation \eqref{eq:intro_curt} with $\alpha=(1+\epsilon)$ and $\xi=F$ with the {\it optimal} number of columns, rows, tubes, and rank of $U$ (we mention our
matching lower bound later), though the running time has an $2^{O(k^2/\epsilon)}$ term in it. 
We note that instead 
combining Theorem \ref{thm:bicriteria} and Theorem \ref{thm:intro_f_curt}
gives a bicriteria result for CURT without a $2^{O(k^2/\epsilon)}$ term in the running time, though it is suboptimal in the number of columns, rows, tubes, and rank of $U$.

We also obtain several algorithms for tensor entry-wise $\ell_p$ norm low-rank approximation,
as well as results for asymmetric tensor norms, which are natural extensions of the matrix $\ell_1$-$\ell_2$ norm. Here, for a tensor $A$, 
$\| A\|_v = \sum_{i} ( \sum_{j,k} (A_{i,j,k} )^2 )^\frac{1}{2} $ and $\| A\|_u = \sum_{i,j} ( \sum_{k} (A_{i,j,k} )^2 )^\frac{1}{2}$. 

\begin{theorem}[Combination of Theorem~\ref{thm:l1_bicriteria_algorithm_rank_k2_nearly_input_sparsity_time} ($\|\|_1$-norm), Theorem~\ref{thm:lp_bicriteria_algorithm_rank_k2_nearly_input_sparsity_time} ($\|\|_p$-norm, $p\in (0,1)$) Theorem~\ref{thm:lv_l122_polyklogn_approx_algorithm} ($\| \|_v$-norm or $\ell_1$-$\ell_2$-$\ell_2$), Theorem~\ref{thm:lu_l112_polyklogn_approx_algorithm} ($\| \|_u$-norm or $\ell_1$-$\ell_1$-$\ell_2$)]
Given a $3$rd order tensor $A\in \mathbb{R}^{n\times n \times n}$, for any $k\geq 1$, let $r=\wt{O}(k^2)$. If $A_k$ exists then there is an algorithm which runs in $\nnz(A) \cdot t + \wt{O}(n) \poly(k)$ time and  outputs a (factorization of a) rank-$r$ tensor $B$ for which $ \| B - A \|_{\xi} \leq \poly(k,\log n)  \cdot  \| A_k - A \|_{\xi}$ holds. If $A_k$ does not exist, we have $\|B-A\|_{\xi} \leq \poly(k,\log n)  \cdot \OPT + \gamma $, where $\gamma$ is an arbitrarily small positive function of $n$. The success probability is at least $9/10$. For $\xi = 1$ or $p$, $t=\wt{O}(k)$; for $\xi=v$, $t=O(1)$; for $\xi=u$, $t=O(n)$.
\end{theorem}

As in the case of Frobenius norm, we can get rank-$k$ and CURT algorithms for the above norms. Our results for asymmetric norms can be extended to $\ell_p$-$\ell_2$-$\ell_2$, $\ell_p$-$\ell_p$-$\ell_2$, and families of M-estimators.

We also obtain the following result for weighted tensor low-rank approximation.
\begin{theorem}[Informal Version of Theorem~\ref{thm:w_r_distinct_2d_cols}, weighted]
 Suppose we are given a third order tensor $A\in \mathbb{R}^{n\times n\times n}$, as well as a tensor $ W\in \mathbb{R}^{n\times n \times n}$ with $r$ distinct rows and $r$ distinct columns. Suppose there is a rank-$k$ tensor $A'\in \mathbb{R}^{n\times n\times n}$ for which $\| W\circ (A'-A)\|_F^2 = \OPT$ and one can write $A'=\sum_{i=1}^k u_i \otimes v_i \otimes w_i$ for $\|u_i\|_2$, $\|v_i\|_2$, and $\|w_i\|_2$ bounded by $2^{n^\delta}$. Then there is an algorithm running in $( \nnz(A)+\nnz(W)  + n 2^{\wt{O}(r^2k^2/\epsilon)} ) \cdot n^\delta$ time and outputting $n\times k$ matrices $U_1,U_2,U_3 $ for which 
$
\left\|  W\circ \left( U_1 \otimes U_2 \otimes U_3 - A \right) \right\|_F^2 \leq (1+\epsilon) \OPT
$
with probability at least $2/3$.
\end{theorem}

We next strengthen H{\aa}stad's NP-hardness 
to show that even approximating tensor rank is hard (we note at the time of H{\aa}stad's NP-hardness, there was no PCP theorem available; nevertheless we need to do additional work here):
\begin{theorem}[Informal Version of Theorem~\ref{thm:approximate_tensor_rank_is_eth_hard}]
  Let $q \geq 3$. Unless the Exponential Time Hypothesis ($\mathsf{ETH}$) fails, there is an absolute constant $c_0>1$ for which distinguishing if a tensor in $\mathbb{R}^{n^q}$ has rank at most $k$, or at least $c_0 \cdot k$, requires
$2^{\delta k^{1-o(1)}}$ time, for a constant $\delta>0$. 
\end{theorem}
Under random-\ETH \cite{f02,gl04,rsw16}, an average case hardness assumption for \SAT, we can replace the $k^{1-o(1)}$ in the exponent above with a $k$. We also obtain hardness in terms of $\epsilon$:
\begin{theorem}[Informal Version of Corollary~\ref{cor:two_to_the_one_over_eps_to_the_forth}]
  Let $q \geq 3$. Unless \ETH fails, there is no algorithm running in $2^{o(1/\epsilon^{1/4})}$ time
  which, given a tensor $A \in \mathbb{R}^{n^q}$, outputs a rank-$1$ tensor $B$
  for which $\|A-B\|_F^2 \leq (1+\epsilon)\OPT$. 
\end{theorem}
As a side result worth stating,
our analysis improves the best matrix CUR decomposition algorithm under Frobenius norm \cite{bw14},
providing the first optimal $\nnz(A)$-time algorithm:
\begin{theorem}[Informal Version of Theorem~\ref{thm:f_matrix_cur_algorithm}, Matrix CUR decomposition]
  There is an algorithm, which given a matrix $A\in \mathbb{R}^{n\times d}$ and an integer $k \geq 1$,
  runs in $O(\nnz(A)) + (n+d)\poly(k,1/\epsilon)$ time and outputs three matrices: $C\in \mathbb{R}^{n \times c}$
  containing $c$ columns of $A$, $R\in \mathbb{R}^{r\times d}$ containing $r$ rows of $A$, and
  $U\in \mathbb{R}^{c\times r}$ with $\rank(U)=k$ for which $r=c=O(k/\epsilon)$ and
$\|CUR - A \|_F^2 \leq (1+\epsilon) \min_{\mathrm{rank-}k~A_k} \| A_k - A \|_F^2,$
holds with probability at least $9/10$. 
\end{theorem}
\subsection{Our Techniques}\label{sec:our_techniques}
Many of our proofs, in particular those for Theorem \ref{thm:bicriteria} and Theorem \ref{thm:smallk},
are based on what we call an ``iterative existential proof'', which we then turn into an algorithm in
two different ways depending if we are proving Theorem \ref{thm:bicriteria} or Theorem \ref{thm:smallk}.

Henceforth, we assume $A_k$ exists; otherwise replace $A_k$ with a suitably good tensor $\tilde{A}$ in what follows.
Since $A_k = \sum_{i=1}^k U_i^* \otimes V_i^* \otimes W_i^*$\footnote{For simplicity, we define $U\otimes V\otimes W=\sum_{i=1}^k U_i \otimes V_i \otimes W_i$, where $U_i$ is the $i$-th column of $U$.}, we can create three $n \times k$ matrices $U^*$, $V^*$, and $W^*$ whose columns are the vectors $U_i^*$, $V_i^*$,
  and $W_i^*$, respectively. Now we consider the three different flattenings (or unfoldings) of $A_k$, which express $A_k$ as an $n \times n^2$ matrix. Namely, by thinking of $A_k$ as the sum of outer products, we can write
  the three flattenings of $A_k$ as $U^* \cdot Z_1$, $V^* \cdot Z_2$, and $W^* \cdot Z_3$, where the rows of $Z_1$ are $\vect( V^*_i \otimes W^*_i)$
  \footnote{$\vect(V^*_i \otimes W^*_i)$ denotes a row vector that has length $n_1 n_2$ where   $V^*_i$ has length $n_1$ and $W^*_i$ has length $n_2$.}
  (
  For simplicity, we write $Z_1 = (V^{*\top} \odot W^{*\top})$.
  \footnote{$(V^{*\top} \odot W^{*\top} )$ denotes a $k\times n_1n_2$ matrix where the $i$-th row is $\vect( V^*_i \otimes W^*_i)$, where length $n_1$ vector $V_i^*$ is the $i$-th column of $n_1\times k$ matrix $V^{*}$, and length $n_2$  vector $W^*_i$ is the $i$-th column of $n_2 \times k$ matrix $W^*$, $\forall i\in [k]$.}
  ),
  the rows of $Z_2$ are $\vect(U^*_i \otimes W^*_i)$, and the rows of $Z_3$ are $\vect( U^*_i \otimes V^*_i)$, for $i \in [k] \eqdef \{1, 2, \ldots, k\}$.
  Letting the three corresponding flattenings of the input tensor $A$ be
  $A_1, A_2,$ and $A_3$, by the symmetry of the Frobenius norm, we have
$  \|A-B\|_F^2 = \|A_1-U^*Z_1\|_F^2 = \|A_2-V^*Z_2\|_F^2 = \|A_3-W^*Z_3\|_F^2.$

  Let us consider the hypothetical regression problem $\min_U \|A_1 - UZ_1\|_F^2$. Note that we do not know $Z_1$, but we will not need to.
  Let $r = O(k/\epsilon)$, and suppose $S_1$
  is an $n^2 \times r$ matrix of i.i.d. normal random variables with mean $0$ and variance $1/r$, denoted $N(0, 1/r)$. Then by standard results
  for regression (see, e.g., \cite{w14} for a survey), if $\hat{U}$ is the minimizer to the smaller regression problem
  $\hat{U} = \textrm{argmin}_U \|UZ_1S_1 - A_1 S_1\|_F^2$, then
  \begin{equation}\label{eqn:first}
    \|A_1 -\hat{U}Z_1\|_F^2 \leq (1+\epsilon) \textrm{min}_U \|A_1 - UZ_1\|_F^2. 
  \end{equation}
  Moreover,
  $\hat{U} = A_1 S_1 (Z_1 S_1)^{\dagger}$. Although we do not know know $Z_1$, this implies $\hat{U}$ is in the column span of $A_1 S_1$, which we do know, since
  we can flatten $A$ to compute $A_1$ and then compute $A_1 S_1$. Thus, this hypothetical regression argument gives us an existential statement - there
  exists a good rank-$k$ matrix $\hat{U}$ in the column span of $A_1 S_1$.
  We could similarly define $\hat{V} = A_2 S_2 (Z_2 S_2)^\dagger$ and $\hat{W} = A_3 S_3 (Z_3 S_3)^\dagger$ as solutions to the analogous regression problems
  for the other two flattenings of $A$, which are in the column spans of $A_2 S_2$ and $A_3 S_3$, respectively. Given $A_1 S_1$, $A_2 S_2$, and
  $A_3 S_3$, which we know, we could hope there is a good rank-$k$ tensor in the span of the rank-$1$ tensors 
  \begin{equation}\label{eqn:span}
    \{( A_1 S_1)_a \otimes ( A_2 S_2 )_b \otimes ( A_3 S_3)_c\}_{a,b,c \in [r]}.
    \end{equation}
  However, an immediate issue arises. First, note that our hypothetical regression problem guarantees that $\|A_1 - \hat{U}Z_1\|_F^2 \leq (1+\epsilon)\|A-A_k\|_F^2$,
  and therefore since the rows of $Z_1$ are of the special form $\vect( V^*_i \otimes W^*_i )$, we can perform a ``retensorization'' to create
  a rank-$k$ tensor $B = \sum_i \hat{U}_i \otimes V^*_i \otimes W^*_i$ from the matrix $\hat{U}Z_1$ for which $\|A-B\|_F^2 \leq (1+\epsilon)\|A-A_k\|_F^2$.
  While we do not know $\hat{U}$, since it is in the column span of $A_1 S_1$, it implies that $B$ is in the span of the rank-$1$ tensors
  $\{( A_1 S_1 )_a \otimes V^*_b \otimes W^*_c\}_{a \in [r], b,c \in [k]}.$ Analogously, we have that there is a good rank-$k$ tensor $B$ in the span of the
  rank-$1$ tensors $\{U^*_a \otimes (A_2 S_2)_b \otimes W^*_c\}_{a, c \in [k], b \in [r]}$, and a good rank-$k$ tensor $B$ in the span of the
  rank-$1$ tensors $\{U^*_a \otimes V^*_b \otimes (A_3 S_3 )_c\}_{a,b \in [k], c \in [r]}$. However, we do not know $U^*$ or $V^*$, and it is not clear
  there is a rank-$k$ tensor $B$ for which {\it simultaneously} its first factors are in the column span of $ A_1 S_1$, its second factors are in the
  column span of $ A_2 S_2$, and its third factors are in the column span of $ A_3 S_3$, i.e., whether there is a good rank-$k$ tensor B in the span
  of rank-$1$ tensors in (\ref{eqn:span}).

  We fix this by an iterative argument. Namely, we first compute $A_1 S_1$, and write $\hat{U} = A_1 S_1 (Z_1 S_1)^{\dagger}$. We now redefine $Z_2$ {\it with
    respect to $\hat{U}$}, so the rows of $Z_2$ are $\vect(\hat{U}_i \otimes W^*_i)$ for $i \in [k]$, and consider the regression problem
  $\min_V \|A_2 - VZ_2\|_F^2$. While we do not know $Z_2$, if $S_2$ is an $n^2 \times r$ matrix of i.i.d. Gaussians, we again have the statement that
  $\hat{V} = A_2 S_2 (Z_2 S_2)^{\dagger}$ satisfies
  \begin{eqnarray*}\label{eqn:second}
    \|A_2 - \hat{V}Z_2\|_F^2 & \leq & (1+\epsilon) \mathrm{min}_V \|A_2 - VZ_2\|_F^2  \textrm{  by the regression guarantee with Gaussians} \ \\
    &\leq & (1+\epsilon)\|A_2 - V^*Z_2\|_F^2 \textrm{  since }V^* \textrm{ is no better than the minimizer }V \ \\
    & = & (1+\epsilon)\|A_1-\hat{U}Z_1\|_F^2 \textrm{ by retensorizing and flattening along a different dimension } \ \\
    & \leq & (1+\epsilon)^2 \mathrm{min}_U \|A_1 - UZ_1\|_F^2 \textrm{ by (\ref{eqn:first})}\\
    & = & (1+\epsilon)^2 \|A - A_k\|_F^2 \textrm{ by definition of }Z_1 \ .
  \end{eqnarray*}
  Now we can retensorize $\hat{V} Z_2$ to obtain a rank-$k$ tensor $B$ for which $\|A-B\|_F^2 = \|A_2 - \hat{V} Z_2\|_F^2 \leq (1+\epsilon)^2 \|A-A_k\|_F^2$. Note
  that since the columns of $\hat{V}$ are in the span of $A_2 S_2$, and the rows of $Z_2$ are $\vect(\hat{U}_i \otimes W^*_i)$ for $i \in [k]$, where the columns
  of $\hat{U}$ are in the span of $A_1 S_1$, it follows that $B$ is in the span of rank-$1$ tensors
$  \{(A_1 S_1)_a \otimes ( A_2 S_2 )_b \otimes \hat{V}_c\}_{a,b \in [r], c \in [k]}.$

  Suppose we now redefine $Z_3$ so that it is now an $r^2 \times n^2$ matrix with rows $\vect( ( A_1 S_1)_a \otimes ( A_2 S_2)_b )$ for all pairs $a,b \in [r]$,
  and consider the regression problem $\min_W \|A_3 - WZ_3\|_F^2$. Now observe that since {\it we know } $Z_3$, and since we can form $A_3$ by
  flattening $A$, we can solve for $W \in \mathbb{R}^{n \times r^2}$ in polynomial time by solving a regression problem. Retensorizing $WZ_3$ to a tensor $B$,
  it follows that we have found a rank-$r^2 = O(k^2/\epsilon^2)$ tensor $B$ for which $\|A-B\|_F^2 \leq (1+\epsilon)^2 \|A-A_k\|_F^2 = (1+O(\epsilon))\|A-A_k\|_F^2$,
  and the result follows by adjusting $\epsilon$ by a constant factor.

  To obtain the $\nnz(A) + n \poly(k/\epsilon)$ running time guarantee of Theorem \ref{thm:bicriteria}, while we can replace $S_1$ and $S_2$ with compositions
  of a sparse CountSketch matrix and a Gaussian matrix (see chapter 2 of \cite{w14} for a survey), enabling us to compute $A_1 S_1$ and $A_2 S_2$ in
  $\nnz(A) + n \poly(k/\epsilon)$ time, we still need to solve the regression problem $\min_W \|A_3 - WZ_3\|_F^2$ quickly, and note that we cannot even
  write down $Z_3$ without spending $r^2 n^2$ time. Here we use a different random matrix $S_3$ called TensorSketch, which was introduced in
  \cite{p13,pp13}, but for which we will need the stronger properties of a subspace embedding and approximate matrix product shown
  to hold for it in \cite{anw14}. Given the latter properties, we can instead solve the regression problem $\min_W \|A_3S_3 - WZ_3S_3\|_F^2$, and
  importantly $A_3S_3$ and $Z_3 S_3$ can be computed in $\nnz(A) + n \poly(k/\epsilon)$ time. Finally, this small problem can be solved in $n \poly(k/\epsilon)$
  time.

  If we want to output a rank-$k$ solution as in Theorem \ref{thm:smallk}, then we need to introduce indeterminates at several places in the preceding argument
  and run a generic polynomial optimization procedure which runs in time exponential in the number of indeterminates. Namely, we write
  $\hat{U}$ as $A_1 S_1 X_1$, where $X_1$ is an $r \times k$ matrix of indeterminates, we write $\hat{V}$ as $A_2 S_2 X_2$, where $X_2$ is an $r \times k$ matrix
  of indeterminates, and we write $\hat{W}$ as $A_3 S_3 X_3$, where $X_3$ is an $r \times k$ matrix of indeterminates. When executing the above iterative
  argument, we let the rows of $Z_1$ be the vectors $\vect(V^*_i \otimes W^*_i)$, the rows of $Z_2$ be the vectors $\vect(\hat{U}_i \otimes W^*_i)$, and the rows of $Z_3$
  be the vectors $\vect(\hat{U}_i \otimes V_i)$. Then $\hat{U}$ is a $(1+\epsilon)$-approximate minimizer to $\min_U \|A_1 - UZ_1\|_F$, while
  $\hat{V}$ is a $(1+\epsilon)$-approximate minimizer to $\min_V \|A_2 - VZ_2\|_F$, while $\hat{W}$ is a $(1+\epsilon)$-approximate minimizer to
  $\min_W \|A_3 - WZ_3\|_F$. Note that by assigning $X_1 = (Z_1 S_1)^{\dagger}$, $X_2 = (Z_2 S_2)^{\dagger}$, and $X_3 = (Z_3 S_3)^{\dagger}$, it follows
  that the rank-$k$ tensor $B = \sum_{i=1}^k (A_1 S_1 X_1)_i \otimes (A_2 S_2 X_2)_i \otimes (A_3 S_3 X_3)_i$ satisfies
  $\|A-B\|_F^2 \leq (1+\epsilon)^3 \|A-A_k\|_F^2$, as desired. Note that here the rows of $Z_2$ are a function of $X_1$, while the rows of $Z_3$ are a function
  of both $X_1$ and $X_2$. What is important for us though is that it suffices to minimize the degree-6 polynomial
    $\sum_{a,b,c \in [n]} (\sum_{i=1}^k (A_1 S_1 X_1)_{a,i} \cdot (A_2 S_2 X_2)_{b,i} \cdot (A_3 S_3 X_3)_{c,i} - A_{a,b,c}  )^2$,
  over the $3rk = O(k^2/\epsilon)$ indeterminates $X_1, X_2, X_3$, since we know there exists an assignment to $X_1, X_2$, and $X_3$
  providing a $(1+O(\epsilon))$-approximate solution, and any solution $X_1, X_2$, and $X_3$ found by minimizing the above polynomial will be no worse than
  that solution. This polynomial can be minimized up to additive $2^{-\poly(n)}$ additive error in $\poly(n)$ time \cite{r92a,bpr96} assuming the entries of $U^*, V^*$, and
  $W^*$ are bounded by $2^{\poly(n)}$, as assumed in Theorem \ref{thm:smallk}. Similar arguments can be made for obtaining a relative error approximation
  to the actual value $\OPT$ as well as handling the case when $A_k$ does not exist.

To optimize the running time to $\nnz(A)$, we can choose CountSketch matrices $T_1, T_2, T_3$ of $t=\poly(k,1/\epsilon) \times n$ dimensions and reapply the above iterative argument. Then it suffices to minimize this small size degree-6 polynomial
$\sum_{a,b,c \in [t]} (\sum_{i=1}^k (T_1 A_1 S_1 X_1)_{a,i} \cdot (T_2 A_2 S_2 X_2)_{b,i} \cdot (T_3 A_3 S_3 X_3)_{c,i} - (A(T_1,T_2,T_3))_{a,b,c}  )^2$,
over the $3rk = O(k^2/\epsilon)$ indeterminates $X_1, X_2, X_3$. Outputting $A_1S_1X_1$, $A_2S_2X_2$, $A_3S_3X_3$ then provides a $(1+\epsilon)$-approximate solution.

Our iterative existential argument provides a general framework for obtaining low rank approximation results
for tensors for many other error measures as well.

\subsection{Other Low Rank Approximation Algorithms Following Our Framework.}

\paragraph{Column, row, tube subset selection, and CURT decomposition.}
In tensor column, row, tube subset selection, the goal is to find three matrices: a subset
$C$ of columns of $A$, a subset $R$ of rows of $A$, and
a subset $T$ of tubes of $A$, such that there exists a small tensor $U$ for which $\| U(C,R,T) - A \|_F^2 \leq (1+\epsilon) \OPT$. 
We first choose two Gaussian matrices $S_1$ and $S_2$ with $s_1=s_2=O(k/\epsilon)$ columns, and form a matrix $Z_3' \in \mathbb{R}^{(s_1s_2) \times n^2}$ with $(i,j)$-th row equal to the vectorization of $ (A_1 S_1)_i \otimes (A_2 S_2)_j$. Motivated
by the regression problem $\min_W \|A_3-WZ_3'\|_F$,
we sample $d_3= O(s_1 s_2 /\epsilon)$ columns from $A_3$ and let $D_3$ denote this selection matrix. There are a few
ways to do the sampling depending on the tradeoff between the number of columns and running time, which we describe
below. Proceeding
iteratively, we write down $Z_2'$ by setting its $(i,j)$-th row to the vectorization of $ ( A_1 S_1)_i \otimes ( A_3 D_3)_j$. We then sample $d_2 =O(s_1d_3/\epsilon)$ columns from $A_2$ and let $D_2$ denote that selection matrix. Finally, we define $Z_1' $ by setting its $(i,j)$-th row to be the vectorization of $ (A_2 D_2)_i \otimes (A_3 D_3)_j$. We obtain $C=A_1 D_1$, $R=A_2 D_2$ and $T= A_3 D_3$. For the sampling steps, we can use a generalized matrix column subset selection
technique, which extends a column subset selection technique of \cite{bw14} in the context of CUR
decompositions to the case when $C$ is not necessarily a subset of the input. 
This gives $O(\nnz(A)\log n) + \wt{O}(n^2) \poly(k,1/\epsilon)$ time. Alternatively, we can use
a technique we develop called tensor leverage score sampling described below,
yielding $O(\nnz(A)) + n \poly(k,1/\epsilon)$ time.
%
%


A body of work in the matrix case has focused on finding the best possible number of columns and rows of a CUR
decomposition, and
we can ask the same question for tensors. It turns out that if one is given the factorization
$\sum_{i=1}^k (U_B)_i \otimes (V_B)_i \otimes (W_B)_i$ of a rank-$k$ tensor $B \in \mathbb{R}^{n \times n \times n}$
with $U_B , V_B , W_B \in \mathbb{R}^{n \times k}$, then one
can find a set $C$ of $O(k/\epsilon)$ columns, a set $R$ of $O(k/\epsilon)$ rows, and
a set $T$ of $O(k/\epsilon)$ tubes of $A$, together with a rank-$k$ tensor $U$ for which
$\|U(C,R,T) - A\|_F^2 \leq (1+\epsilon)\|A-B\|_F^2$. This is based on an iterative argument, where the initial
sampling (which needs to be our generalized matrix column subset selection rather than tensor leverage
score sampling to achieve optimal bounds) is done
with respect to $V_B^{\top} \odot W_B^{\top}$, and then an iterative argument is carried out. Since we show
a matching lower bound on the number of columns, rows, tubes and rank of $U$, these parameters are tight. The
algorithm is efficient if one is given a rank-$k$ tensor $B$ which is a $(1+O(\epsilon))$-approximation to $A$;
if not
then one can use Theorem~\ref{thm:f_main_algorithm_bit} and and this step will be exponential time in $k$. If one
just wants $O(k\log k + k/\epsilon)$ columns, rows, and tubes, then one can achieve $O(\nnz(A)) + n \poly(k,1/\epsilon)$
time, if one is given $B$. 

\paragraph{Column-row, row-tube, tube-column face subset selection, and CURT decomposition.}

  In tensor column-row, row-tube, tube-column face subset selection, the goal is to find three tensors: a subset $C\in \mathbb{R}^{c\times n \times n}$ of row-tube faces of $A$, a subset $R\in \mathbb{R}^{n\times r \times n}$ of tube-column faces of $A$, and a subset $T\in \mathbb{R}^{n\times n\times t}$ of column-row faces of $A$, such that there exists a tensor $U\in \mathbb{R}^{tn \times cn \times rn}$ with small rank for which $\| U(T_1,C_2,R_3) - A\|_F^2 \leq (1+\epsilon) \OPT$, where $T_1\in \mathbb{R}^{n\times tn}$ denotes the matrix obtained by flattening the tensor $T$ along the first dimension, $C_2\in \mathbb{R}^{n\times cn}$ denotes the matrix obtained by flattening the tensor $C$ along the second dimension, and $R_3\in \mathbb{R}^{n\times rn}$ denotes the matrix obtained by flattening the tensor $T$ along the third dimension.

  We solve this problem by first choosing two Gaussian matrices $S_1$ and $S_2$ with $s_1=s_2=O(k/\epsilon)$ columns, and then forming matrix $U_3\in \mathbb{R}^{n\times s_1s_2}$ with $(i,j)$-th column equal to $(A_1S_1)_i$, as well as matrix $V_3\in \mathbb{R}^{n\times s_1 s_2}$ with $(i,j)$-th column equal to $(A_2S_2)_j$. 
Inspired 
by the regression problem $\min_{W\in \mathbb{R}^{n\times s_1s_2}} \| V_3 \cdot (W^\top \odot U_3^\top) - A_2\|_F$,
we sample $d_3= O(s_1 s_2 /\epsilon)$ rows from $A_2$ and let $D_3 \in \mathbb{R}^{n\times n}$ denote this selection matrix. In other words, $D_3$ selects $d_3$ tube-column faces from the original tensor $A$. Thus, we obtain a small regression problem: $\min_W \| D_3 V_3 \cdot (W^\top \odot U_3^\top) - D_3 A_2\|_F$. By retensorizing the objective function, we obtain the problem $\min_W \| U_3 \otimes (D_3 V_3) \otimes W - A(I,D_3,I) \|_F$. Flattening the objective function along the third dimension, we obtain $\min_W \| W \cdot (U_3^\top \odot (D_3 V_3)^\top) -  (A(I,D_3,I))_3 \|_F$ which has optimal solution $(A(I,D_3,I))_3 (U_3^\top \odot (D_3 V_3)^\top)^\dagger$. Let $W'$ denote $A(I,D_3,I))_3$. In the next step, we fix $W_2= W'(U_3^\top \odot (D_3 V_3)^\top)^\dagger$ and $U_2=U_3$, and consider the objective function $\min_V \|U_2 \cdot ( V^\top \odot W_2^\top) - A_1 \|_F$. Applying a similar argument, we obtain $V'= (A(D_2,I,I))_2$ and $U'=( A(I,I,D_1)_1)$. Let $C$ denote $A(D_2,I,I)$, $R$ denote $A(I,D_3,I)$, and $T$ denote $A(I,I,D_1)$. Overall, this algorithm selects $\poly(k,1/\epsilon)$ faces from each dimension.

Similar to our column-based CURT decomposition,
our face-based CURT decomposition has the property that if one is given the factorization $\sum_{i=1}^k (U_B)_i\otimes (V_B)_i \otimes (W_B)_i$ of a rank-$k$ tensor $B\in \mathbb{R}^{n\times n\times n}$ with $U_B,V_B,W_B\in \mathbb{R}^{n\times k}$ which is a $(1+O(\epsilon))$-approximation to $A$, then one can find a set $C$ of $O(k/\epsilon)$ row-tube faces, a set $R$ of $O(k/\epsilon)$ tube-column faces, and a set $T$ of $O(k/\epsilon)$ column-row faces of $A$, together with a $\rank$-$k$ tensor $U$ for which $\| U(T_1,C_2,R_3) -A \|_F^2 \leq (1+\epsilon) \OPT$. 

\paragraph{Tensor multiple regression and tensor leverage score sampling.} In the above we need to consider standard problems for matrices in the context of tensors. Suppose we are given a matrix $A\in \mathbb{R}^{n_1 \times n_2 n_3}$ and a matrix $B= (V^\top \odot W^\top) \in \mathbb{R}^{k \times n_2n_3}$ with rows $(V_i \otimes W_i)$ for
an $n_2 \times k$ matrix $V$ and $n_3 \times k$ matrix $W$. Using \textsc{TensorSketch} \cite{p13,pp13,anw14} one can solve multiple regression $\min_{U} \| U B - A\|_F$ without forming $B$ in $O(n_2 + n_3) \poly(k,1/\epsilon)$ time, rather than the na\"ive $O(n_2n_3)\poly(k,1/\epsilon)$ time. However, this does not immediately help us if
we would like to sample columns of such a matrix $B$
proportional to its leverage scores. Even if we apply \textsc{TensorSketch} to compute a $k \times k$
change of basis matrix $R$ in
$O(n_2 + n_3)\poly(k, \log(n_2 n_3))$ time, for which the leverage scores of $B$ are (up to a constant factor)
the squared column norms of $R^{-1} B$, there are still $n_2n_3$ leverage scores and we cannot write them all down!
Nevertheless, we show we can still sample by them by using that the matrix of
interest is formed via a tensor product, which can be rewritten as a matrix multiplication which we never need to explicily materialize. 
In more detail, for the $i$-th row $e_iR^{-1}$ of $R^{-1}$,
we create a matrix $V^{'i}$ by scaling each of the columns of $V^\top$ entrywise by the entries of $z$.
The squared norms of $e_iR^{-1}B$ are exactly the squared entries of $(V^{'i})W^{\top}$. We cannot
compute this matrix product, but we can first sample a column of it
proportional to its squared norm and then sample an entry in that column proportional to its square.
To sample a column, we compute $G (V^{'i})W^{\top}$ for a Gaussian matrix $G$ with
$O(\log n_3)$ rows by computing $G \cdot V^{'i}$, then computing $(G \cdot V^{'i}) \cdot W^{\top}$,
which is $O(n_2 + n_3)\poly(k, \log(n_2 n_3))$ total time.
After sampling a column, we compute the column exactly
and sample a squared entry. We do this for each $i \in [k]$, first sampling an $i$ proportional to
$\|GV^{'i}W^{\top}\|_F^2$, then running the above scheme on that $i$. The $\poly(\log n)$ factor in
the running time can be replaced by $\poly(k)$ if one wants to avoid a $\poly(\log n)$ dependence
in the running time.

%

\paragraph{Entry-wise $\ell_1$ low-rank approximation.}
  We consider the problem of entrywise $\ell_1$-low rank approximation of an $n \times n \times n$ tensor $A$, namely, the problem
  of finding a rank-$k$ tensor $B$ for which $\|A-B\|_1 \leq \poly(k, \log n) \OPT$, where
  $\OPT = \inf_{\textrm{rank-}k \ B}\|A-B\|_1$, and where for a tensor $A$, $\|A\|_1 = \sum_{i,j,k} |A_{i,j,k}|$.
  Our iterative existential argument
  can be applied in much the same way as for the Frobenius norm. We iteratively flatten $A$ along each of its three
  dimensions, obtaining $A_1$, $A_2$, and $A_3$ as above, and iteratively build a good rank-$k$ solution $B$ of the form
  $ (A_1S_1X_1) \otimes (A_2 S_2 X_2) \otimes (A_3 S_3 X_3)$, where now the $S_i$ are matrices of i.i.d. Cauchy random
  variables or sparse matrices of Cauchy random variables and the $X_i$ are $O(k \log k) \times k$ matrices of indeterminates.
  For a matrix $C$ and a matrix $S$ of i.i.d. Cauchy random variables with $k$ columns, it is known \cite{swz17} that the column
  span of $CS$ contains
  a $\poly(k \log n)$-approximate rank-$k$ space with respect to the entrywise $\ell_1$-norm for $C$.
  In the case of tensors, we must perform an iterative flattening and retensorizing argument to guarantee there exists
  a tensor $B$ of the form above. Also, if we insist on outputting a rank-$k$ solution as opposed to a bicriteria solution,
  $\| (A_1 S_1 X_1) \otimes (A_2 S_2 X_2) \otimes (A_3 S_3 X_3) -A\|_1$ is not a polynomial of the
  $X_i$, and if we introduce sign variables for the $n^3$ absolute values,
  the running time of the polynomial solver will be $2^{\# \textrm{ of variables}} = 2^{\Omega(n^3)}$. We perform additional dimensionality
  reduction by Lewis weight sampling \cite{cp15} from the flattenings to reduce the problem size to $\poly(k)$. This small problem still has
  $\tilde{O}(k^3)$ sign variables,
  and to obtain a $2^{\tilde{O}(k^2)}$ running time we relax the reduced problem to a Frobenius norm problem,
  mildly increasing the approximation factor by another $\poly(k)$ factor.

   Combining the iterative existential argument with techniques in \cite{swz17}, we also obtain an $\ell_1$ CURT decomposition algorithm (which is similar to the Frobenius norm result in Theorem~\ref{thm:intro_f_curt}), which can find $\wt{O}(k)$ columns, $\wt{O}(k)$ rows, $\wt{O}(k)$ tubes, and a tensor $U$. 
Our algorithm starts from a given factorization of a rank-$k$ tensor $B = U_B \otimes V_B \otimes W_B$ found above. We compute a sampling and rescaling diagonal matrix $D_1$ according to the Lewis weights of matrix $B_1=(V_B^\top \odot W_B^\top)$, where $D_1$ has $\wt{O}(k)$ nonzero entries. Then we iteratively construct $B_2$, $D_2$, $B_3$ and $D_3$. Finally we have $C=A_1 D_1$ (selecting $\wt{O}(k)$ columns from $A$), $R=A_2 D_2$ (selecting $\wt{O}(k)$ rows from $A$), $T=A_3 D_3$ (selecting $\wt{O}(k)$ tubes from $A$) and tensor $U = ( (B_1 D_1)^\dagger) \otimes ( (B_2 D_2)^\dagger ) \otimes ( (B_3 D_3)^\dagger )$.

We have similar results for entry-wise $\ell_p$, $1 \leq p < 2$, via analogous techniques.

\paragraph{$\ell_1$-$\ell_2$-$\ell_2$ low-rank approximation (sum of Euclidean norms of faces).} For an
$n \times n \times n$ tensor $A$, in $\ell_1$-$\ell_2$-$\ell_2$ low rank approximation we seek a \textrm{rank-}$k$ tensor $B$ for which $ \| A - B \|_v \leq \poly(k,\log n) \OPT$, where $\OPT=\inf_{\textrm{rank-}k~B}  \| A - B \|_v$ and where $\| A\|_v = \sum_{i} ( \sum_{j,k} (A_{i,j,k} )^2 )^\frac{1}{2} $ for a tensor $A$. This norm is asymmetric, i.e., not invariant under permutations to its coordinates, and we cannot flatten the tensor along each of its dimensions while preserving its cost. Instead, we embed the problem to a new problem with a symmetric norm. Once we have a symmetric norm, we apply an iterative existential argument. We choose an oblivious sketching matrix (the $M$-Sketch in \cite{cw15soda}) $S\in \mathbb{R}^{s\times n}$ with $s=\poly(k,\log n)$, and reduce the original problem to $\| S (A-B) \|_v$,
by losing a small approximation factor. Because $s$ is small, we can then turn the $\ell_1$ part of the problem to $\ell_2$ by losing another $\sqrt{s}$ in the approximation, so that now the problem is a Frobenius norm problem. We then apply our iterative existential argument to the problem $\| S( \sum_{i=1}^k U^*_i \otimes (\wh{A}_2 S_2 X_2)_i \otimes (\wh{A}_3 S_3 X_3)_i - A )\|_F$ where $U^*$ is a fixed matrix and $\wh{A} = SA$, and output a bicriteria solution.

\paragraph{$\ell_1$-$\ell_1$-$\ell_2$ low-rank approximation (sum of Euclidean norms of tubes).} For an $n \times n \times n$ tensor $A$, in the $\ell_1$-$\ell_1$-$\ell_2$ low rank approximation problem we seek a \textrm{rank-}$k$ tensor $B$ for which $ \| A - B \|_u \leq \poly(k,\log n) \OPT$, where $\OPT=\inf_{\textrm{rank-}k~B}  \| A - B \|_u$ and $\| A\|_u = \sum_{i,j} ( \sum_{k} (A_{i,j,k} )^2 )^\frac{1}{2}$. The main difficulty in this problem is that the norm is asymmetric, and we cannot flatten the tensor along all three dimensions. To reduce the problem to a problem with a symmetric norm, we choose random Gaussian matrices $S\in \mathbb{R}^{n \times s}$ with $s=O(n)$. By Dvoretzky's theorem \cite{d61}, for all tensors $A$, $\| A S\|_1 \approx \|A\|_u$, which reduces our problem to $\min_{\text{rank-}k~B}\| (A - B)S \|_1$. Via an iterative existential argument, we obtain a generalized version of entrywise $\ell_1$ low rank approximation, $\| (  (\wh{A}_1 S_1 X_1) \otimes (\wh{A}_2 S_2 X_2) \otimes (A_3S_3X_3) -  A) S\|_1$, where $\wh{A} = AS$ is an $n\times n \times s$ size tensor. Finally, we can either use a polynomial system solver to obtain a rank-$k$ solution, or output a bicriteria solution.

\paragraph{Weighted low-rank approximation.} We also consider weighted low rank approximation.
Given an $n \times n \times n$ tensor $A$ and an $n\times n \times n$ tensor $W$ of weights, we want to find a rank-$k$
  tensor $B$ for which $\| W\circ(A-B)\|_F^2 \leq (1+\epsilon) \OPT$, where
  $\OPT = \inf_{\textrm{rank-}k \ B} \| W\circ(A-B)\|_F^2$ and where for a tensor $A$,
  $\| W\circ A \|_F = (\sum_{i,j,k} W_{i,j,k}^2 A_{i,j,k}^2)^{\frac{1}{2}}$.
  We provide two algorithms based on different assumptions on the weight tensor $W$.
  The first algorithm assumes that $W$ has $r$ distinct faces on each of its three dimensions.
  We flatten $A$ and $W$ along each of its three dimensions, obtaining $A_1,A_2,A_3$ and $W_1,W_2,W_3$. Because each $W_i$ has $r$ distinct rows, combining the ``{\it guess a sketch}'' technique from \cite{rsw16} with our iterative argument,
  we can create matrices $U_1, U_2,$ and $U_3$ in terms of $O(rk^2/\epsilon)$ total indeterminates and for
  which a solution to the objective function
  $\|W \circ( \sum_{i=1}^k (U_1)_i \otimes (U_2)_i \otimes (U_3)_i - A) \|_F^2 $, together with $O(r)$ side
  constraints, gives a $(1+\varepsilon)$-approximation.
  We can solve the latter problem in $\poly(n) \cdot 2^{\wt{O}(rk^2/\epsilon)}$ time.
  Our second algorithm assumes $W$ has $r$ distinct faces in two dimensions. Via a pigeonhole argument,
  the third dimension will have at most $2^{\wt{O}(r)}$ distinct faces.
  We again use $O(rk^2/\epsilon)$ variables to express $U_1$ and $U_2$, but now express $U_3$ in terms of these
  variables, which is necessary since $W_3$ could have an exponential number of distinct rows, ultimately causing
  too many variables needed to express $U_3$ directly. We again arrive at the objective function
  $\|W \circ( \sum_{i=1}^k (U_1)_i \otimes (U_2)_i \otimes (U_3)_i - A) \|_F^2 $, but now have $2^{\wt{O}(r)}$
  side constraints, coming from the fact that $U_3$ is a rational function of the variables created for $U_1$
  and $U_2$ and we need to clear denominators. Ultimately, the running time is $2^{\wt{O}(r^2k^2/\epsilon)}$.

\paragraph{Computational Hardness.}
  Our $2^{\delta k^{1-o(1)}}$ time hardness for $c$-approximation
  in Theorem \ref{thm:approximate_tensor_rank_is_eth_hard}
  is shown via a reduction from approximating {\sf MAX-3SAT} to
  approximating {\sf MAX-E3SAT},
  where the latter problem has the property that each clause in the satisfiability
  instance has exactly $3$ literals (in {\sf MAX-3SAT} some clauses may have $2$
  literals). Then, a reduction \cite{t01} from approximating
  {\sf MAX-E3SAT} to approximating {\sf MAX-E3SAT(B)}
  is performed, for a constant $B$ which provides an upper bound on the number of
  clauses each literal can occur in. Given an instance $\phi$ to {\sf MAX-E3SAT(B)},
  we create a $3$rd order tensor $T$ as H{\aa}stad does using $\phi$ \cite{h90}.
  While H{\aa}stad's reduction guarantees that the rank of $T$ is at most $r$ if
  $\phi$ is satisfiable, and at least $r+1$ otherwise, we can show that if $\phi$
  is not satisfiable then its rank is at least the minimal size of a set of variables
  which is guaranteed to intersect every unsatisfied clause in any unsatisfiable
  assignment. Since if $\phi$ is not satisfiable, there are at least a linear fraction
  of clauses in $\phi$ that are unsatisfied under any assignment by the inapproximability
  of {\sf MAX-E3SAT(B)}, and since each literal occurs in at most $B$ clauses
  for a constant $B$, it follows that the rank of $T$ when $\phi$ is not satisfiable
  is at least $c_0r$ for a constant $c_0 > 1$. Further, under \ETH, our reduction implies
  one cannot approximate {\sf MAX-E3SAT(B)}, and thus approximate the rank of a tensor
  up to a factor $c_0$, in less than $2^{\delta k^{1-o(1)}}$ time. We need the near-linear
  size reduction of \MAX-\SAT to \MAX-\ESAT of \cite{mr10} to get our strongest result.

  The $2^{\Omega(1/\epsilon^{1/4})}$ time hardness for $(1+\epsilon)$-approximation for rank-$1$ tensors in
  Theorem \ref{thm:approximate_rank1_is_eth_hard} strengthens the NP-hardness for rank-$1$
  tensor computation in Section 7 of \cite{hl13}, where instead of assuming the NP-hardness of
  the {\sf Clique} problem, we assume \ETH. Also, the proof in \cite{hl13} did not explicitly bound
  the approximation error; we do this for a $\poly(1/\epsilon)$-sized tensor (which can be padded
  with $0$s to a $\poly(n)$-sized tensor) to rule out $(1+\epsilon)$-approximation in $2^{o(1/\epsilon^{1/4})}$ time.

  The same hard instance above shows, assuming \ETH, that $2^{\Omega(1/\epsilon^{1/2})}$ time is necessary for
  $(1+\varepsilon)$-approximation to the spectral norm of a symmetric rank-$1$ tensor
  (see Section~\ref{sec:hardness_symmetric_tensor_eigenvalue} and Section~\ref{sec:hardness_singular_spectral_rank1}).

  Assuming \ETH, the $2^{1/\epsilon^{1-o(1)}}$-hardness \cite{swz17} for matrix $\ell_1$-low rank approximation gives the same
  hardness for tensor entry-wise $\ell_1$ and $\ell_1$-$\ell_1$-$\ell_2$ low rank approximation. Also,
  under \ETH, we strengthen the NP-hardness in \cite{cw15focs} to a $2^{1/\epsilon^{\Omega(1)}}$-hardness for $\ell_1$-$\ell_2$-low
  rank approximation of a matrix, which gives the same hardness for tensor $\ell_1$-$\ell_2$-$\ell_2$ low rank approximation.

\paragraph{Hard Instance.} We extend the previous matrix CUR hard instance \cite{bw14} to $3$rd order tensors by planting
multiple rotations of the hard instance for matrices into a tensor. We show $C$ must select $\Omega(k/\epsilon)$
columns from $A$, $R$ must select $\Omega(k/\epsilon)$ rows from $A$, and $T$ must select $\Omega(k/\epsilon)$ tubes
from $A$. Also the tensor $U$ must have rank at least $k$. This generalizes to $q$-th order tensors.

\begin{algorithm}[t!]\caption{Main Meta-Algorithm}\label{alg:intro_algorithm}
\begin{algorithmic}[1]{
\Procedure{\textsc{TensorLowRankApproxBicriteria}}{$A,n,k,\epsilon$} \Comment{Theorem \ref{thm:bicriteria}}
\State Choose sketching matrices $S_2$,$S_3$(Composition of Gaussian and CountSketch.)
\State Choose sketching matrices $T_2$,$T_3$(CountSketch.)
\State Compute $T_2 A_2 S_2$, $T_3 A_3 S_3$.
\State Construct $\wh{V}$ by setting $(i,j)$-th column to be $(A_2 S_2)_i$.
\State Construct $\wh{W}$ by setting $(i,j)$-th column to be $(A_3 S_3)_j$.
\State Construct matrix $B$ by setting $(i,j)$-th row of $B$ is vectorization of $(T_2 A_2 S_2)_i \otimes (T_3 A_3 S_3)_j $.
\State Solve $ \min_{U} \| U B - (A(I,T_2,T_3) )_1\|_F^2$.
\State \Return $\wh{U}$, $\wh{V}$, and $\wh{W}$.
\EndProcedure
\Procedure{\textsc{TensorLowRankApprox}}{$A,n,k,\epsilon$} \Comment{Theorem \ref{thm:smallk}}
\State Choose sketching matrices $S_1$,$S_2$,$S_3$(Composition of Gaussian and CountSketch.)
\State Choose sketching matrices $T_1$,$T_2$,$T_3$(CountSketch.)
\State Compute $T_1 A_1 S_1$, $T_2 A_2 S_2$, $T_3 A_3 S_3$.
\State Solve  $ \min_{X_1, X_2, X_3} \| (T_1 A_1 S_1 X_1) \otimes (T_2 A_2 S_2 X_2) \otimes (T_3 A_3 S_3 X_3) - A(T_1,T_2,T_3)\|_F^2$.
\State \Return $A_1 S_1 X_1$, $A_2 S_2 X_2$, and $A_3S_3 X_3$.
\EndProcedure}
\end{algorithmic}
\end{algorithm}

\paragraph{Optimal matrix CUR decomposition.}
We also improve the $\nnz(A)\log n + (n+d) \poly(\log n, k, $ $1/\epsilon)$ running time
of \cite{bw14} for CUR decomposition of $A \in \mathbb{R}^{n \times d}$ to $\nnz(A) + (n+d)\poly(k,1/\epsilon)$,
while selecting the optimal number of columns, rows, and a rank-$k$ matrix $U$. Using \cite{cw13,mm13,nn13},
we find
a matrix $\wh{U}$ with $k$ orthonormal columns in $\nnz(A)+n\poly(k/\varepsilon)$ time
for which ${\min}_{V}\|\wh{U}V-A\|_F^2\leq (1+\varepsilon)\|A-A_k\|_F^2.$
Let $s_1=\wt{O}(k/\varepsilon^2)$ and $S_1\in\mathbb{R}^{s_1\times n}$ be a sampling/rescaling matrix by
the leverage scores of $\wh{U}.$ By strengthening the affine embedding analysis of \cite{cw13}
to leverage score sampling (the analysis of \cite{cw13} gives a weaker analysis for affine embeddings
using leverage scores which does not allow approximation in the sketch space to translate to approximation
in the original space), with probability at least $0.99$, for all $ X'$
which satisfy $\|S_1\wh{U}X'-S_1A\|_F^2\leq (1+\varepsilon')\min_{X}\|S_1\wh{U}X-S_1A\|_F^2$, we have
$\|\wh{U}X'-A\|_F^2\leq (1+\varepsilon)\min_{X}\|\wh{U}X-A\|_F^2,$ where $\varepsilon'=0.0001\varepsilon.$
Applying our generalized row subset selection procedure, we
can find $Y,R$ for which
$\|S_1\wh{U}YR-S_1A\|_F^2\leq (1+\varepsilon')\min_{X}\|S_1\wh{U}X-S_1A\|_F^2,$
where $R$ contains $O(k/\varepsilon')=O(k/\varepsilon)$ rescaled rows of $S_1A$. A key point
is that rescaled rows of $S_1A$ are also rescaled rows of $A$.
Then, $\|\wh{U}YR-A\|_F^2\leq (1+\varepsilon)\min_{X}\|\wh{U}X-A\|_F^2$. Finding $Y,R$ can be done in
$d\poly(s_1/\varepsilon)=d\poly(k/\varepsilon)$ time. Now set $\wh{V}=YR$. We can choose $S_2$ to be a sampling/rescaling matrix, and then 
find $C,Z$ for which
$\|CZ\wh{V}S_2-AS_2\|_F^2\leq (1+\varepsilon')\min_X \|X\wh{V}S_2-AS_2\|_F^2$ in a similar way, where $C$ contains $O(k/\varepsilon)$ rescaled columns of $AS_2$, and
thus also of $A$.
We thus have 
$\|CZYR-A\|_F^2\leq (1+O(\varepsilon))\|A-A_k\|_F^2.$

\paragraph{Distributed and streaming settings.}
Since our algorithms use linear sketches, they are implementable in distributed and streaming models.
We use random variables with limited independence to succinctly store the sketching matrices \cite{cw13,kvw14,kn14,w14,swz17}.

\paragraph{Extension to other notions of tensor rank.} This paper focuses on the standard CP rank, or canonical rank, of a tensor. As mentioned, due to border rank issues, the best rank-$k$ solution does not exist in certain cases. There are other notions of tensor rank considered in some applications which do not suffer from this problem, e.g., the tucker rank \cite{kc07,pc08,mh09,zw13,yc14}, and the train rank \cite{o11,otz11,zwz16,ptbd16}). We also show observe that our techniques can be applied to these notions of rank.

\subsection{Comparison to \cite{bcv14}}\label{sec:comparison}
In \cite{bcv14}, the authors show for a third order $n_1 \times n_2 \times n_3$ tensor $A$ 
how to find a rank-$k$ tensor $B$ for which 
$\|A-B\|_F^2 \leq 5\OPT$ in $\poly(n_1 n_2 n_3) \exp(\poly(k))$ time. They generalize this to
$q$-th order tensors to find a rank-$k$ tensor $B$ for which
$\|A-B\|_F^2 = O(q) \OPT$ in $\poly(n_1 n_2 \cdots n_q) \exp(\poly(qk))$ time. 

In contrast, we obtain a rank-$k$ tensor $B$ for which $\|A-B\|_F^2 \leq (1+\epsilon)\OPT$ in
$\nnz(A) + n \cdot \poly(k/\epsilon) + \exp((k^2/\epsilon) \poly(q))$ time for every order $q$. Thus, we obtain a
$(1+\epsilon)$ instead of an $O(q)$ approximation. The $O(q)$ approximation in \cite{bcv14} seems
inherent since the authors apply triangle inequality $q$ times, each time losing a constant factor.
This seems necessary since their argument is based on the span of the top $k$ principal
components in the SVD in each flattening separately containing a good space to project onto for a given mode. In contrast, 
our iterative existential argument chooses the space to project onto in successive modes {\it adaptively}
as a function of spaces chosen for previous modes, and thus we obtain a $(1+\epsilon)^{O(q)} = (1+O(\epsilon q))$-approximation,
which becomes a $(1+\epsilon)$-approximation after replacing $\epsilon$ with $\epsilon/q$. 
Also, importantly, our algorithm runs in $\nnz(A) + n \cdot \poly(k/\epsilon) + \exp((k^2/\epsilon) \poly(q))$ 
time and there are multiple hurdles we overcome to achieve this, as described in Section \ref{sec:our_techniques} above. 

\subsection{An Algorithm and a Roadmap}

\paragraph{Roadmap}
Section~\ref{sec:notation} introduces notation and definitions. Section~\ref{sec:preli} includes several useful tools. We provide our Frobenius norm low rank approximation algorithms in Section~\ref{sec:f}. Section~\ref{sec:f_general_order} extends our results to general $q$-th order tensors. Section~\ref{sec:l1} has our results for entry-wise $\ell_1$ norm low rank approximation. Section~\ref{sec:lp} has our results for entry-wise $\ell_p$ norm low rank approximation. Section~\ref{sec:w} has our results for weighted low rank approximation.  Section~\ref{sec:lvu} has our results for asymmetric norm low rank approximation algorithms. We present our hardness results in Section~\ref{sec:hardness} and Section~\ref{sec:hardinstance}. Section~\ref{sec:distributed} and Section~\ref{sec:streaming} extend the results to distributed and streaming settings. Section~\ref{sec:other_rank} extends our techniques from tensor rank to other notions of tensor rank including tensor tucker rank and tensor train rank.





\newpage

\appendix

\newpage

\section{Notation}\label{sec:notation}

\begin{figure}[!t]
  \centering
    \includegraphics[width=0.5\textwidth]{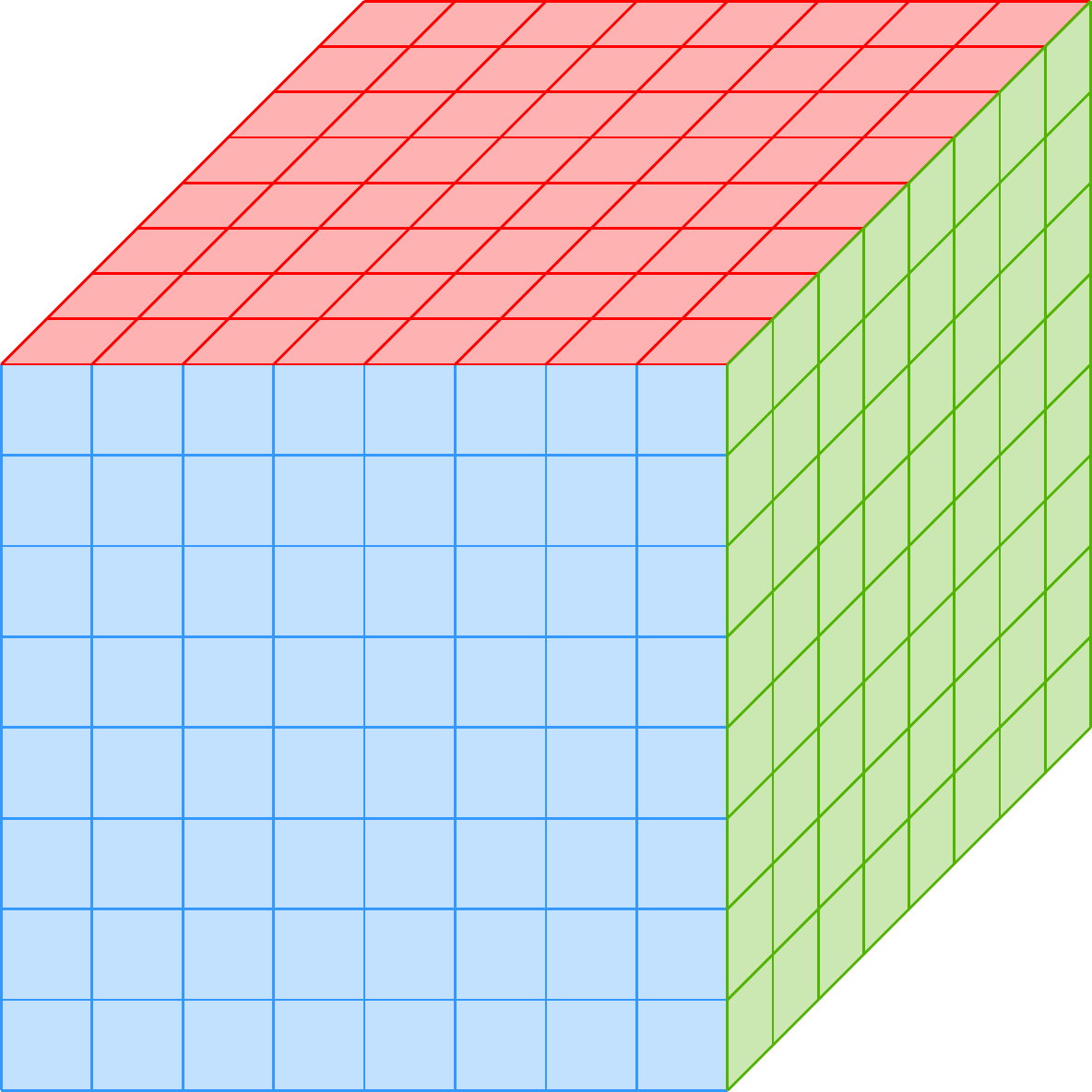}
    \caption{A $3$rd order tensor with size $8\times 8\times 8$.}
\end{figure}

For an $n\in \mathbb{N}_{+}$, let $[n]$ denote the set $\{1,2,\cdots,n\}$.

For any function $f$, we define $\wt{O}(f)$ to be $f\cdot \log^{O(1)}(f)$. In addition to $O(\cdot)$ notation, for two functions $f,g$, we use the shorthand $f\lesssim g$ (resp. $\gtrsim$) to indicate that $f\leq C g$ (resp. $\geq$) for an absolute constant $C$. We use $f\eqsim g$ to mean $cf\leq g\leq Cf$ for constants $c,C$. 

For a matrix $A$, we use $\|A\|_2$ to denote the spectral norm of $A$. For a tensor $A$, let $\| A\|$ and $\| A\|_2$ (which we sometimes use interchangeably)
denote the spectral norm of tensor $A$,
\begin{align*}
\| A \| = \sup_{x,y,z \neq 0} \frac{|A(x,y,z)|}{ \| x \| \cdot \| y \| \cdot \| z \|}.
\end{align*}
Let $\| A\|_F$ denote the Frobenius norm of a matrix/tensor $A$, i.e., $\|A\|_F$ is the square root of sum of squares of all the entries of $A$. For $1\leq p<2$, we use $\| A \|_p$ to denote the entry-wise $\ell_p$-norm of a matrix/tensor $A$, i.e., $\|A\|_p$ is the $p$-th root of the sum of $p$-th powers of the absolute values of the entries of $A$. $\|A\|_1$ will be an important special case of $\|A\|_p$, which corresponds to the sum of absolute values of all of the entries.

Let $\nnz(A)$ denote the number of nonzero entries of $A$. Let $\det(A)$ denote the determinant of a square matrix $A$. Let $A^\top$ denote the transpose of $A$. Let $A^\dagger$ denote the Moore-Penrose pseudoinverse of $A$. Let $A^{-1}$ denote the inverse of a full rank square matrix.


For a 3rd order tensor $A \in \mathbb{R}^{n\times n \times n}$, its $x$-mode fibers are called column fibers ($x=1$), row fibers ($x=2$) and tube fibers ($x=3$). For tensor $A$, we use $A_{*,j,l}$ to denote its $(j,l)$-th column, we use $A_{i,*,l}$ to denote its $(i,l)$-th row, and we use $A_{i,j,*}$ to denote its $(i,j)$-th tube.

A tensor $A$ is symmetric if and only if for any $i,j,k$, $A_{i,j,k} = A_{i,k,j} = A_{j,i,k} = A_{j,k,i} = A_{k,i,j} = A_{k,j,i}$.

For a tensor $A\in \mathbb{R}^{n_1 \times n_2 \times n_3}$, we use $\top$ to denote rotation (3 dimensional transpose) so that $A^\top \in \mathbb{R}^{n_3\times n_1 \times n_2}$.
 For a tensor $A\in \mathbb{R}^{n_1 \times n_2 \times n_3}$ and matrix $B\in \mathbb{R}^{n_3 \times k}$, we define the tensor-matrix dot product to be $A \cdot B \in \mathbb{R}^{n_1\times n_2 \times k}$.

We use $\otimes$ to denote outer product, $\circ$ to denote entrywise product, and $\cdot$ to denote dot product. Given two column vectors $u,v \in \mathbb{R}^n$, let $ u\otimes v  \in \mathbb{R}^{n\times n}$ and $(u \otimes v)_{i,j} = u_i \cdot v_j$, $u^\top v= \sum_{i=1}^n u_i v_i \in \mathbb{R}$ and $(u \circ v )_i = u_i v_i$.

\begin{figure}[!t]
  \centering
    \includegraphics[width=1.0\textwidth]{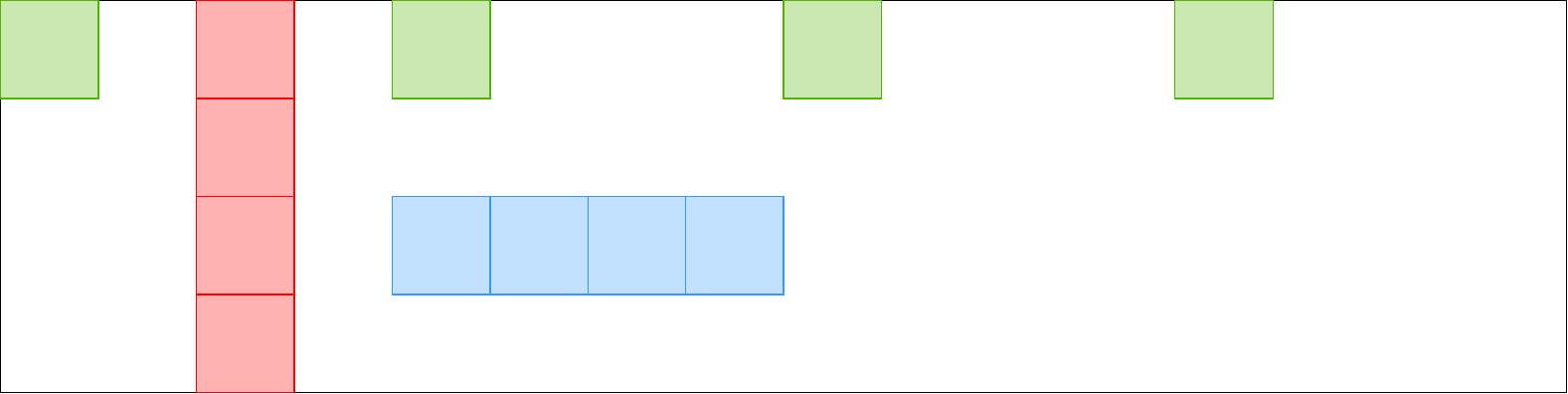}
    \caption{Flattening. We flatten a third order $4 \times 4 \times 4$ tensor along the $1$st dimension to obtain a $4 \times 16$ matrix. The red blocks correspond to a column in the original third order tensor, the blue blocks correspond to a row in the original third order tensor, and the green blocks correspond to a tube in the original third order tensor. }
\end{figure}

\begin{definition}[$\otimes$ product for vectors]
Given $q$ vectors $ u_1\in \mathbb{R}^{n_1}$, $u_2 \in \mathbb{R}^{n_2}$, $\cdots$, $u_q\in \mathbb{R}^{n_q}$, we use $u_1\otimes u_2 \otimes \cdots \otimes u_q$ to denote an $n_1 \times n_2 \times \cdots \times n_q$ tensor such that, for each $(j_1,j_2,\cdots,j_q)\in [n_1]\times [n_2]\times \cdots \times [n_q]$,
\begin{align*}
( u_1\otimes u_2 \otimes \cdots \otimes u_q )_{j_1,j_2,\cdots,j_q} = (u_1)_{j_1} (u_2)_{j_2} \cdots (u_q)_{j_q},
\end{align*}
where $(u_i)_{j_i}$ denotes the $j_i$-th entry of vector $u_i$.
\end{definition}

\begin{definition}[$\vect()$, convert tensor into a vector]\label{def:vect}
Given a tensor $A\in \mathbb{R}^{n_1 \times n_2 \times \cdots \times n_q}$, let $\vect(A)\in\mathbb{R}^{1\times \prod_{i=1}^q n_i}$ be a row vector, such that the $t$-$th$ entry of $\vect(A)$ is $A_{j_1,j_2,\cdots,j_q}$ where $t=(j_1-1)\prod_{i=2}^q n_i+(j_2-1)\prod_{i=3}^q n_i+\cdots+(j_{q-1}-1) n_q+j_{q}.$
\end{definition}
For example if $u=\begin{bmatrix}1\\2\end{bmatrix},v=\begin{bmatrix}3\\4\\5\end{bmatrix}$ then $\vect(u\otimes v)=\begin{bmatrix}3&4&5&6&8&10\end{bmatrix}.$

\begin{definition}[$\otimes$ product for matrices]\label{def:otimes_product}
Given $q$ matrices $U_1\in \mathbb{R}^{n_1 \times k}$, $U_2 \in \mathbb{R}^{n_2 \times k}$, $\cdots$, $U_q\in \mathbb{R}^{n_q\times k}$, we use $U_1 \otimes U_2 \otimes \cdots \otimes U_q$ to denote an $n_1 \times n_2 \times \cdots \times n_q$ tensor which can be written as,
\begin{align*}
U_1 \otimes U_2 \otimes \cdots \otimes U_q = \sum_{i=1}^k (U_1)_i \otimes (U_2)_i \otimes \cdots \otimes (U_q)_i \in \mathbb{R}^{n_1 \times n_2 \times \cdots \times n_q},
\end{align*}
where $(U_j)_i$ denotes the $i$-th column of matrix $U_j \in \mathbb{R}^{n_j \times k}$.
\end{definition}

\begin{definition}[$\odot$ product for matrices]\label{def:odot_product}
Given $q$ matrices $U_1 \in\mathbb{R}^{k\times n_1}$, $U_2\in \mathbb{R}^{k\times n_2}$, $\cdots$, $U_q \in \mathbb{R}^{k\times n_q}$, we use $U_1 \odot U_2 \odot \cdots \odot U_q$ to denote a $k \times \prod_{j=1}^q n_j$ matrix where the $i$-th row of $U_1 \odot U_2 \odot \cdots \odot U_q$ is the vectorization of $(U_1)^i\otimes (U_2)^i \otimes \cdots \otimes (U_q)^i$, i.e.,
\begin{align*}
U_1 \odot U_2 \odot \cdots \odot U_q = \begin{bmatrix}
\vect( (U_1)^1 \otimes (U_2)^1 \otimes \cdots \otimes (U_q)^1 ) \\
\vect( (U_1)^2 \otimes (U_2)^2 \otimes \cdots \otimes (U_q)^2 ) \\
\cdots \\
\vect( (U_1)^k \otimes (U_2)^k \otimes \cdots \otimes (U_q)^k )
\end{bmatrix}
\in \mathbb{R}^{k\times \prod_{j=1}^q n_j}.
\end{align*}
where $(U_j)^i\in\mathbb{R}^{n_j}$ denotes the $i$-th row of matrix $U_j\in \mathbb{R}^{k\times n_j}$.
\end{definition}

\begin{definition}[Flattening vs unflattening/retensorizing]
Suppose we are given three matrices $U\in \mathbb{R}^{n_1 \times k}$, $V\in \mathbb{R}_{n_2 \times k}$, $W\in \mathbb{R}^{n_3 \times k}$. Let tensor $A\in \mathbb{R}^{n_1 \times n_2 \times n_3}$ denote $U\otimes V \otimes W$. Let $A_1 \in \mathbb{R}^{n_1 \times n_2 n_3}$ denote a matrix obtained by flattening tensor $A$ along the $1$st dimension. Then $A_1 = U \cdot B$, where $B = V^\top \odot W^\top \in \mathbb{R}^{k\times n_2 n_3}$ denotes the matrix for which the $i$-th row is $\vect( V_i \otimes W_i ), \forall i \in [k]$. We let the ``flattening'' be the operation that obtains $A_1$ by $A$. Given $A_1 = U \cdot B$, we can obtain tensor $A$ by unflattening/retensorizing $A_1$. We let ``retensorization'' be the operation that obtains $A$ from $A_1$. Similarly, let $A_2\in \mathbb{R}^{n_2 \times n_1 n_3}$ denote a matrix obtained by flattening tensor $A$ along the $2$nd dimension, so $A_2 = V\cdot C$, where $C= W^\top \odot U^\top \in \mathbb{R}^{k\times n_1 n_3}$ denotes the matrix for which the $i$-th row is $\vect(W_i \otimes U_i), \forall i\in [k]$. Let $A_3 \in \mathbb{R}^{n_3 \times n_1 n_2}$ denote a matrix obtained by flattening tensor $A$ along the $3$rd dimension. Then, $A_3 = W \cdot D$, where $D= U^\top \odot V^\top \in \mathbb{R}^{k\times n_1 n_2}$ denotes the matrix for which the $i$-th row is $\vect(U_i \otimes V_i), \forall i\in [k]$.
\end{definition}

\begin{definition}[ $(\cdot,\cdot,\cdot)$ operator for tensors and matrices]\label{def:bracket}
Given tensor $A\in \mathbb{R}^{n_1 \times n_2 \times n_3}$ and three matrices $B_1\in \mathbb{R}^{n_1\times d_1}$, $B_2 \in \mathbb{R}^{n_2\times d_2}$, $B_3\in \mathbb{R}^{n_3 \times d_3}$, we define tensors $A(B_1,I,I)\in \mathbb{R}^{d_1\times n_2\times n_3}$, $A(I,B_2,I)\in \mathbb{R}^{n_1 \times d_2 \times n_3}$, $A(I,I,B_3)\in \mathbb{R}^{n_1 \times n_2 \times d_3}$, $A(B_1,B_2,I)\in \mathbb{R}^{d_1 \times d_2 \times n_3}$, $A(B_1,B_2,B_3)\in \mathbb{R}^{d_1 \times d_2 \times d_3}$ as follows,
\begin{align*}
A(B_1,I,I)_{i,j,l} & ~ = ~ \sum_{i'=1}^{n_1} A_{i',j,l} (B_1)_{i',i}, & \forall (i,j,l) \in [d_1]\times [n_2] \times [n_3] \\
A(I,B_2,I)_{i,j,l} & ~ = ~ \sum_{j'=1}^{n_2} A_{i,j',l} (B_2)_{j',j}, & \forall (i,j,l) \in [n_1]\times [d_2] \times [n_3] \\
A(I,I,B_3)_{i,j,l} & ~ = ~ \sum_{l'=1}^{n_3} A_{i,j,l'} (B_3)_{l',l}, & \forall (i,j,l) \in [n_1]\times [n_2] \times [d_3] \\
A(B_1,B_2,I)_{i,j,l} & ~ = ~ \sum_{i'=1}^{n_1}\sum_{j'=1}^{n_2} A_{i',j',l} (B_1)_{i',i} (B_2)_{j',j}, & \forall (i,j,l) \in [d_1]\times [d_2] \times [n_3] \\
A(B_1,B_2,B_3)_{i,j,l} & ~ = ~ \sum_{i'=1}^{n_1} \sum_{j'=1}^{n_2}\sum_{l'=1}^{n_3} A_{i',j',l'} (B_1)_{i',i} (B_2)_{j',j} (B_3)_{l',l}, & \forall (i,j,l) \in [d_1]\times [d_2] \times [d_3]
\end{align*}
\end{definition}
Note that $B_1^\top A = A(B_1,I,I)$, $A B_3 = A(I,I,B_3)$ and $B_1^\top A B_3 = A(B_1,I,B_3)$.
In our paper, if $\forall i\in[3], B_i$ is either a rectangular matrix or a symmetric matrix, then we sometimes use $A(B_1,B_2,B_3)$ to denote $A(B_1^\top,B_2^\top,B_3^\top)$ for simplicity. Similar to the $(\cdot,\cdot,\cdot)$ operator on $3$rd order tensors, we can define the $(\cdot,\cdot,\cdots,\cdot)$ operator on higher order tensors.

\begin{figure}[!t]
  \centering
    \includegraphics[width=1.0\textwidth]{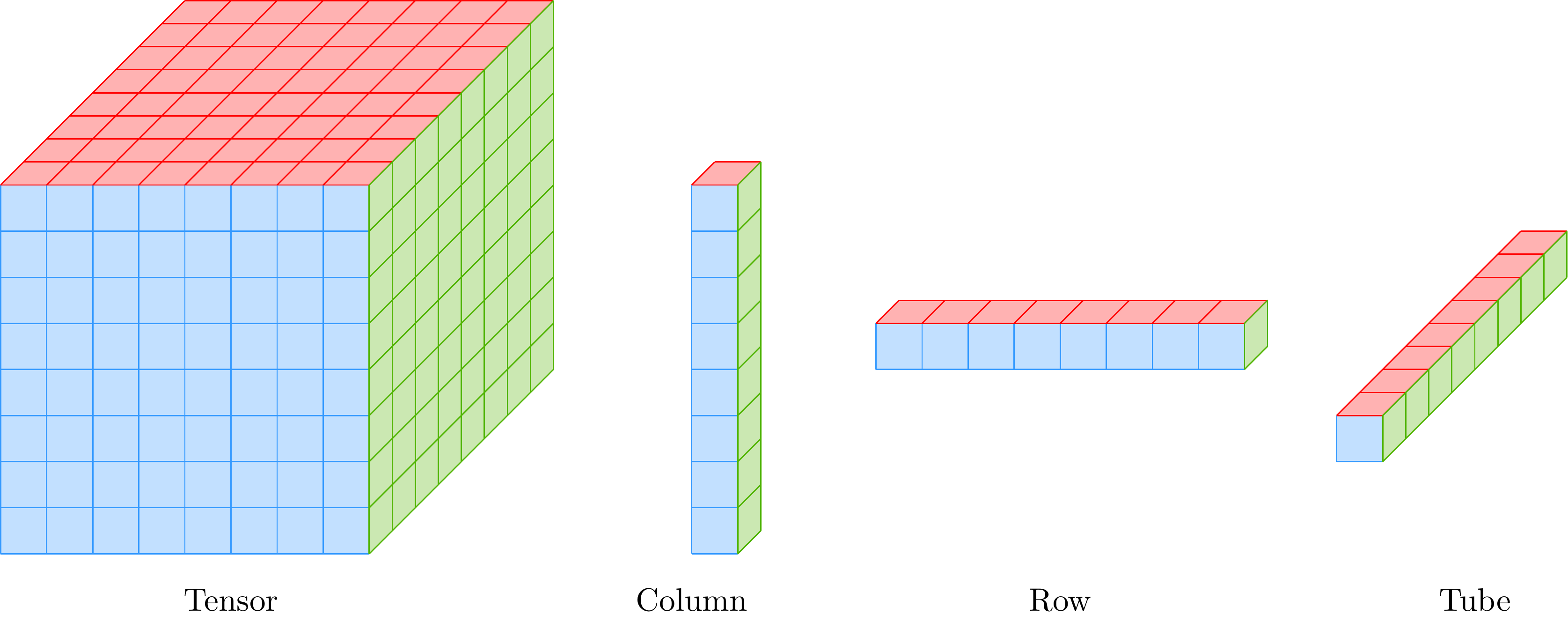}
    \caption{A $3$rd order tensor contains $n^2$ columns, $n^2$ rows, and $n^2$ tubes.}
\end{figure}

For the matrix case, $\underset{\rank-k\ A'}{\min}\|A-A'\|_F^2$ always exists. However, this is not true for tensors~\cite{sl08}.
For convenience, we redefine the notation of $\OPT$ and $\min$.
\begin{definition}
  Given tensor $A\in\mathbb{R}^{n_1\times n_2\times n_3},k>0,$ if $\underset{\rank-k\ A'}{\min}\|A-A'\|_F^2$ does not exist,
  then we define $\OPT=\underset{\rank-k\ A'}{\inf}\|A-A'\|_F^2+\gamma$ for sufficiently small $\gamma>0$, which can
  be an arbitrarily small positive function of $n$.
We let $\underset{\rank-k\ A'}{\min}\|A-A'\|_F^2$ be the value of $\OPT$, and we let $\underset{\rank-k\ A'}{\arg\min}\|A-A'\|_F^2$ be a $\rank-k$ tensor $A_k\in\mathbb{R}^{n_1\times n_2\times n_3}$ which satisfies $\|A-A_k\|_F^2=\OPT.$
\end{definition}

\section{Preliminaries}\label{sec:preli}
Section~\ref{sec:pre_subspace_embeddings_approximate_matrix_product} provides the definitions for Subspace Embeddings and Approximate Matrix Product. We introduce the definition for Tensor-CURT decomposition in Section~\ref{sec:pre_tensor_curt}. Section~\ref{sec:pre_decision_solver} presents a tool which we call a ``polynomial system verifier''. Section~\ref{sec:pre_lower_bound_on_cost} introduces a tool which is able to determine the minimum nonzero value of the absolute value of a polynomial evaluated on a set, provided the polynomial is never equal to $0$ on that set. Section~\ref{sec:pre_relaxation} shows how to relax an $\ell_p$ problem to an $\ell_2$ problem. We provide definitions for CountSketch and Gaussian transforms in Section~\ref{sec:def_count_sketch_gaussian}. We present Cauchy and $p$-stable transforms in Section~\ref{sec:def_cauchy_pstable}. We introduce leverage scores and Lewis weights in Section~\ref{sec:def_leverage_score} and Section~\ref{sec:def_lewis_weights}. Finally, we explain an extension of CountSketch, which is called \textsc{TensorSketch} in Section~\ref{sec:def_tensor_sketch}. 

\subsection{Subspace Embeddings and Approximate Matrix Product}\label{sec:pre_subspace_embeddings_approximate_matrix_product}

\begin{definition}[Subspace Embedding]\label{def:subspace_embedding}
A $(1\pm\epsilon)$ $\ell_2$-subspace embedding for the column space of an $n\times d$ matrix $A$ is a matrix $S$ for which for all $x\in \mathbb{R}^d$, $\| SA x\|_2^2 = (1\pm \epsilon) \| A x\|_2^2$.
\end{definition}
\begin{definition}[Approximate Matrix Product]\label{def:approximate_matrix_product}
  Let $0<\epsilon<1$ be a given approximation parameter. Given matrices $A$ and $B$, where $A$ and $B$ each have $n$ rows, the goal is to output a matrix $C$ so that $\| A^\top B - C\|_F \leq \epsilon \| A \|_F \| B \|_F$. Typically $C$
  has the form $A^\top S^\top S B$, for a random matrix $S$ with a small
  number of rows. See, e.g., Lemma 32 of \cite{cw13} for a number of example matrices $S$
  with $O(\epsilon^{-2})$ rows for which this property holds.
%
\end{definition}

\begin{figure}[!]
  \centering
    \includegraphics[width=0.8\textwidth]{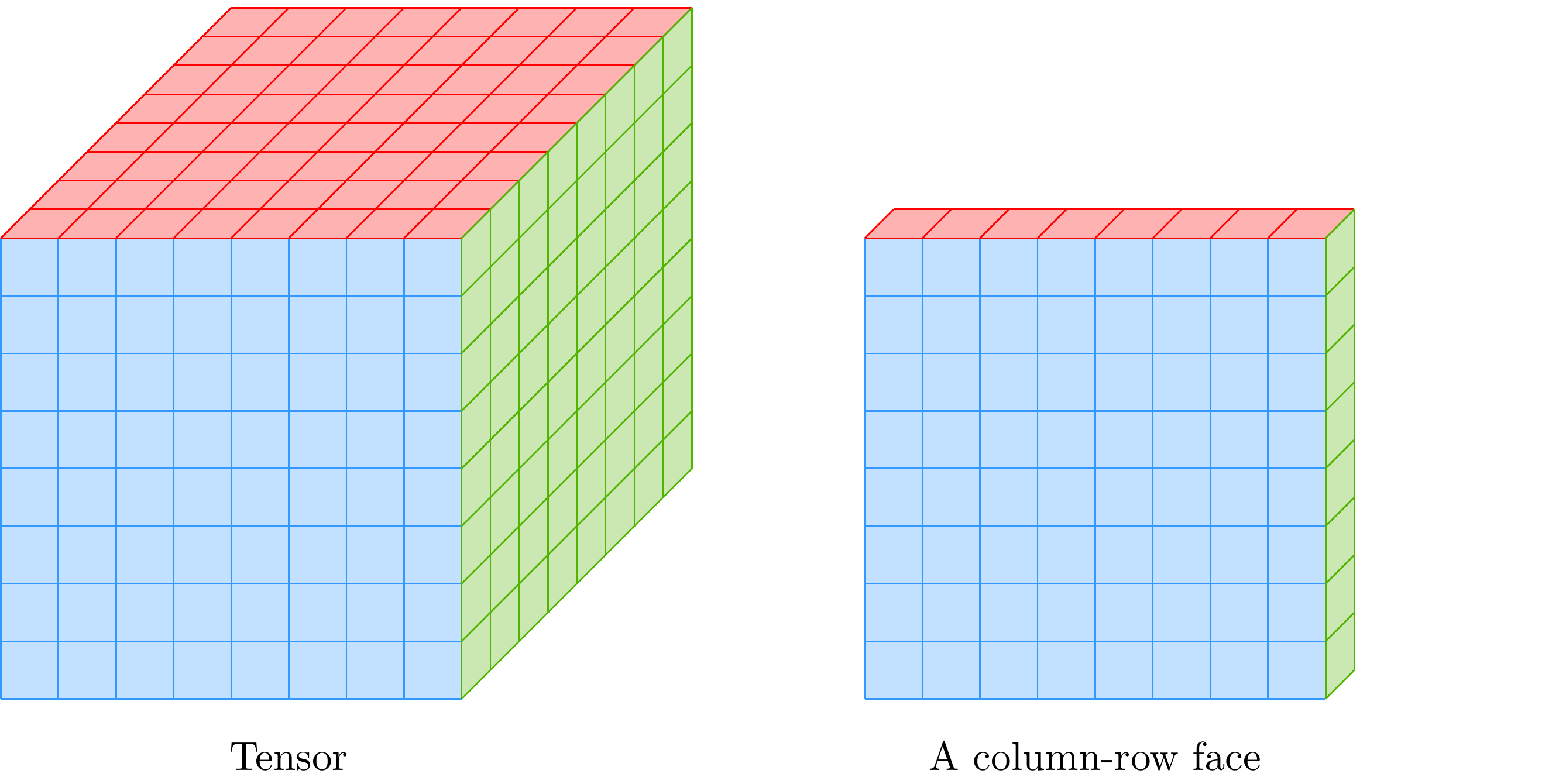}
    \includegraphics[width=0.8\textwidth]{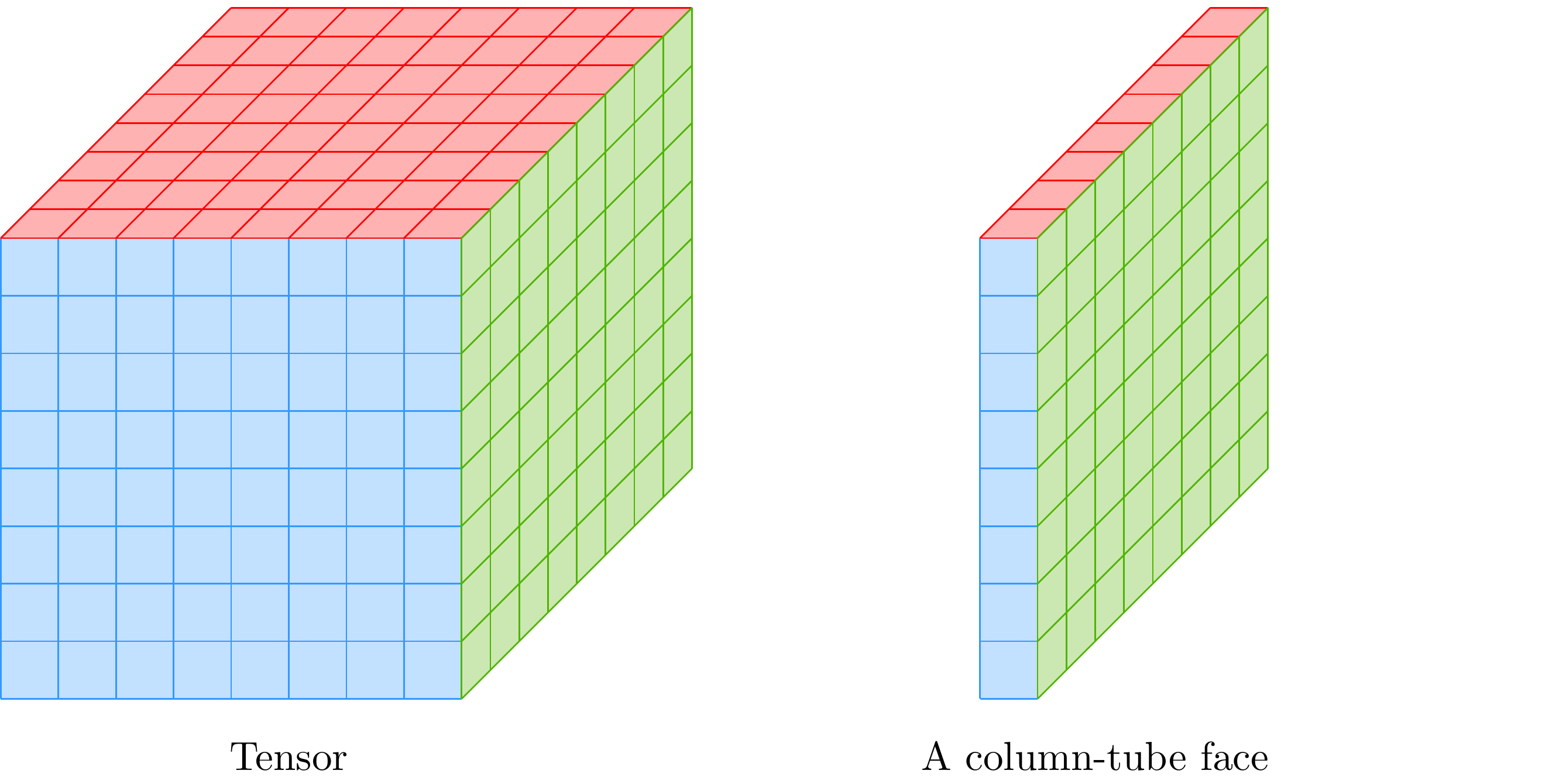}
    \includegraphics[width=0.8\textwidth]{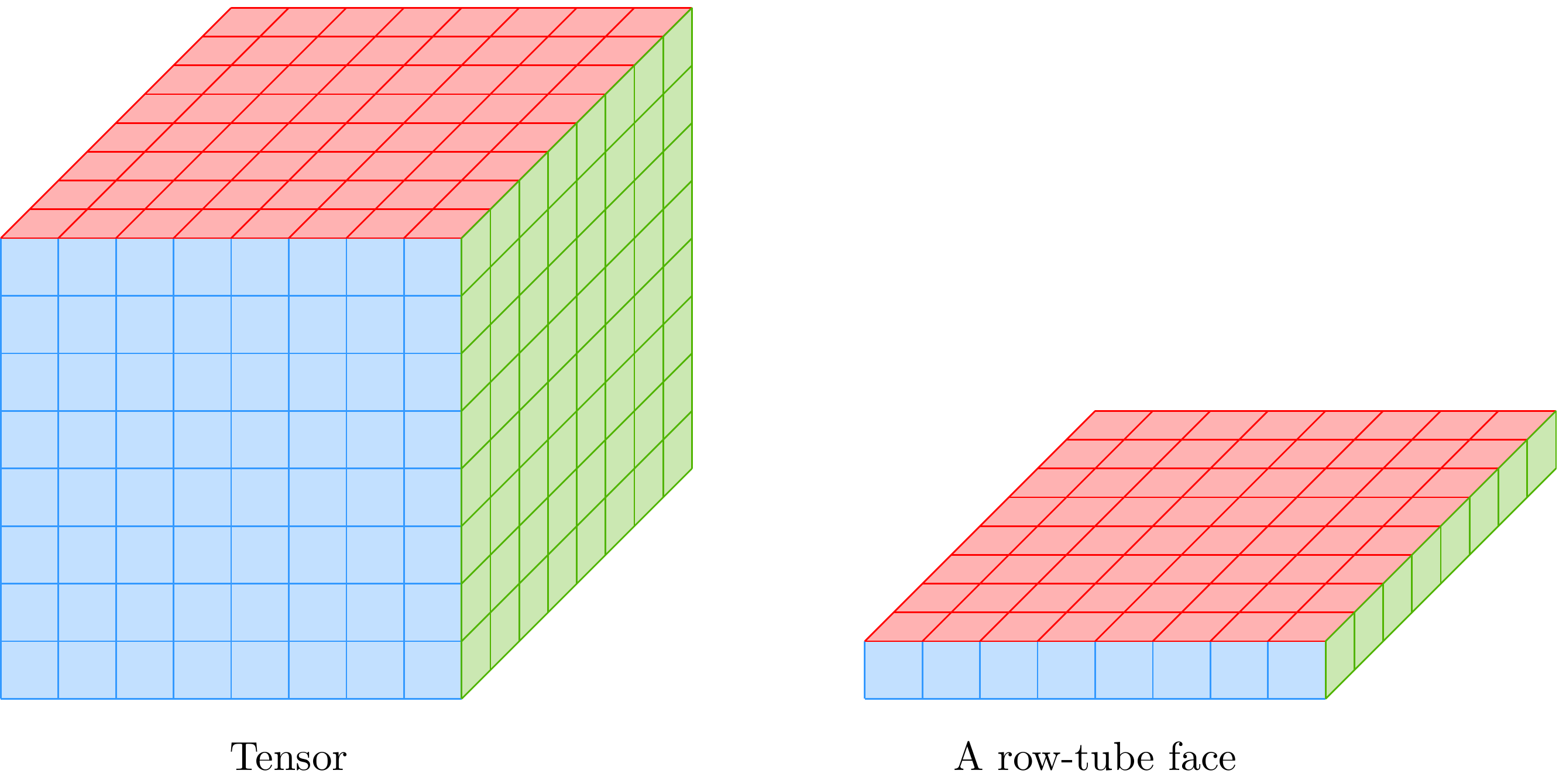}
    \caption{A third order tensor has three types of faces: the column-row faces, the column-tube faces, and the row-tube faces}
\end{figure}

\subsection{Tensor CURT decomposition}\label{sec:pre_tensor_curt}
We first review matrix CUR decompositions:
\begin{definition}[Matrix CUR, exact]
Given a matrix $A\in \mathbb{R}^{n\times d}$, we choose $C\in \mathbb{R}^{n\times c}$ to be a subset of columns of $A$ and $R\in \mathbb{R}^{r \times n}$ to be a subset of rows of $A$. If there exists a matrix $U\in \mathbb{R}^{c\times r}$ such that $A$ can be written as,
\begin{align*}
CUR = A,
\end{align*}
then we say $C,U,R$ is matrix $A$'s CUR decomposition.
\end{definition}

\begin{definition}[Matrix CUR, approximate]
Given a matrix $A\in \mathbb{R}^{n\times d}$, a parameter $k\geq 1$, an approximation ratio $\alpha >1$, and a norm $\| \|_{\xi}$, we choose $C\in \mathbb{R}^{n\times c}$ to be a subset of columns of $A$ and $R\in \mathbb{R}^{r \times n}$ to be a subset of rows of $A$. Then if there exists a matrix $U\in \mathbb{R}^{c\times r}$ such that,
\begin{align*}
\| CUR - A \|_{\xi} \leq \alpha \min_{\rank-k~A_k} \| A_k - A\|_{\xi},
\end{align*}
where $\| \|_{\xi}$ can be operator norm, Frobenius norm or Entry-wise $\ell_1$ norm, we say that $C, U, R$ is
matrix $A$'s approximate CUR decomposition, and sometimes just refer to this as a CUR decomposition.
\end{definition}

\begin{definition}[\cite{b11}]\label{def:subspace_best}
Given matrix $A\in \mathbb{R}^{m\times n}$, integer $k$, and matrix $C\in \mathbb{R}^{m\times r}$ with $r>k$, we define the matrix $\Pi_{C,k}^{\xi}(A) \in \mathbb{R}^{m\times n}$ to be the best approximation to $A$ (under the $\xi$-norm) within the column space of $C$ of rank at most $k$; so, $\Pi_{C,k}^{\xi}(A) \in \mathbb{R}^{m\times n}$ minimizes the residual $\|A - \wh{A} \|_{\xi}$, over all $\wh{A}\in \mathbb{R}^{m\times n}$ in the column space of $C$ of rank at most $k$.
\end{definition}

We define the following notion of tensor-CURT decomposition.
 \begin{definition}[Tensor CURT, exact]
Given a tensor $A\in \mathbb{R}^{ n_1 \times n_2 \times n_3} $, we choose three sets of pair of coordinates $S_1 \subseteq [n_2] \times [n_3], S_2 \subseteq [n_1] \times [n_3], S_3 \subseteq [n_1] \times [n_2]$. We define $c=|S_1|$, $r=|S_2|$ and $t=|S_3|$. Let $C \in \mathbb{R}^{ n_1 \times c}$ denote a subset of columns of $A$, $R \in \mathbb{R}^{ n_2 \times r}$ denote a subset of rows of $A$, and $T \in \mathbb{R}^{ n_3 \times t}$ denote a subset of tubes of $A$. If there exists a tensor $U\in \mathbb{R}^{c \times r \times t}$ such that $A$ can be written as
\begin{align*}
( ( ( U \cdot T^\top )^\top \cdot R^\top )^\top \cdot C^\top )^\top   = A,
\end{align*}
or equivalently,
\begin{align*}
U(C,R,T) = A,
\end{align*}
or equivalently,
\begin{align*}
\forall (i,j,l) \in [n_1] \times [n_2] \times [n_3], A_{i,j,l} = \sum_{u_1=1}^{c} \sum_{u_2=1}^{r} \sum_{u_3=1}^{t} U_{u_1,u_2,u_3} C_{i,u_1} R_{j,u_2} T_{l,u_3},
\end{align*}
then we say $C,U,R, T$ is tensor $A$'s CURT decomposition.
 \end{definition}

\begin{definition}[Tensor CURT, approximate]\label{def:tensor_curt}
Given a tensor $A\in \mathbb{R}^{ n_1 \times n_2 \times n_3} $, for some $k\geq 1$, for some approximation $\alpha>1$, for some norm $\| \|_{\xi}$, we choose three sets of pair of coordinates $S_1 \subseteq [n_2] \times [n_3], S_2 \subseteq [n_1] \times [n_3], S_3 \subseteq [n_1] \times [n_2]$. We define $c=|S_1|$, $r=|S_2|$ and $t=|S_3|$. Let $C \in \mathbb{R}^{ n_1 \times c}$ denote a subset of columns of $A$, $R \in \mathbb{R}^{ n_2 \times r}$ denote a subset of rows of $A$, and $T \in \mathbb{R}^{ n_3 \times t}$ denote a subset of tubes of $A$. If there exists a tensor $U\in \mathbb{R}^{c \times r \times t}$ such that
\begin{align*}
\| U(C, R, T ) - A \|_{\xi} \leq \alpha \min_{\rank-k~A_k} \| A_k - A \|_{\xi},
\end{align*}
where $\|\|_{\xi}$ is operator norm, Frobenius norm or Entry-wise $\ell_1$ norm, then we refer to $C, U, R, T$ as an approximate
CUR decomposition of $A$, and sometimes just refer to this as a CURT decomposition of $A$.
\end{definition}

\begin{figure}[!t]
  \centering
    \includegraphics[width=0.8\textwidth]{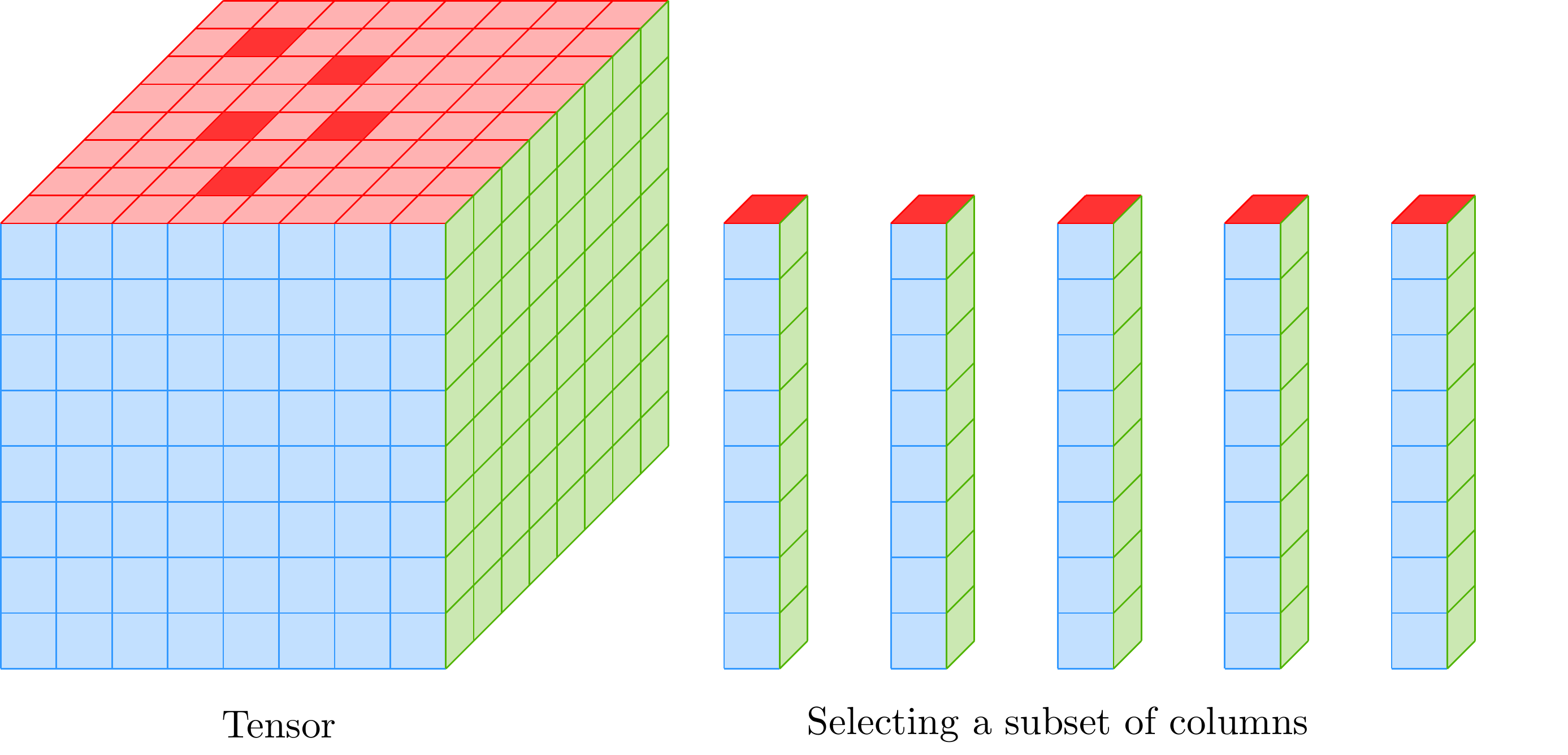}
    \includegraphics[width=0.8\textwidth]{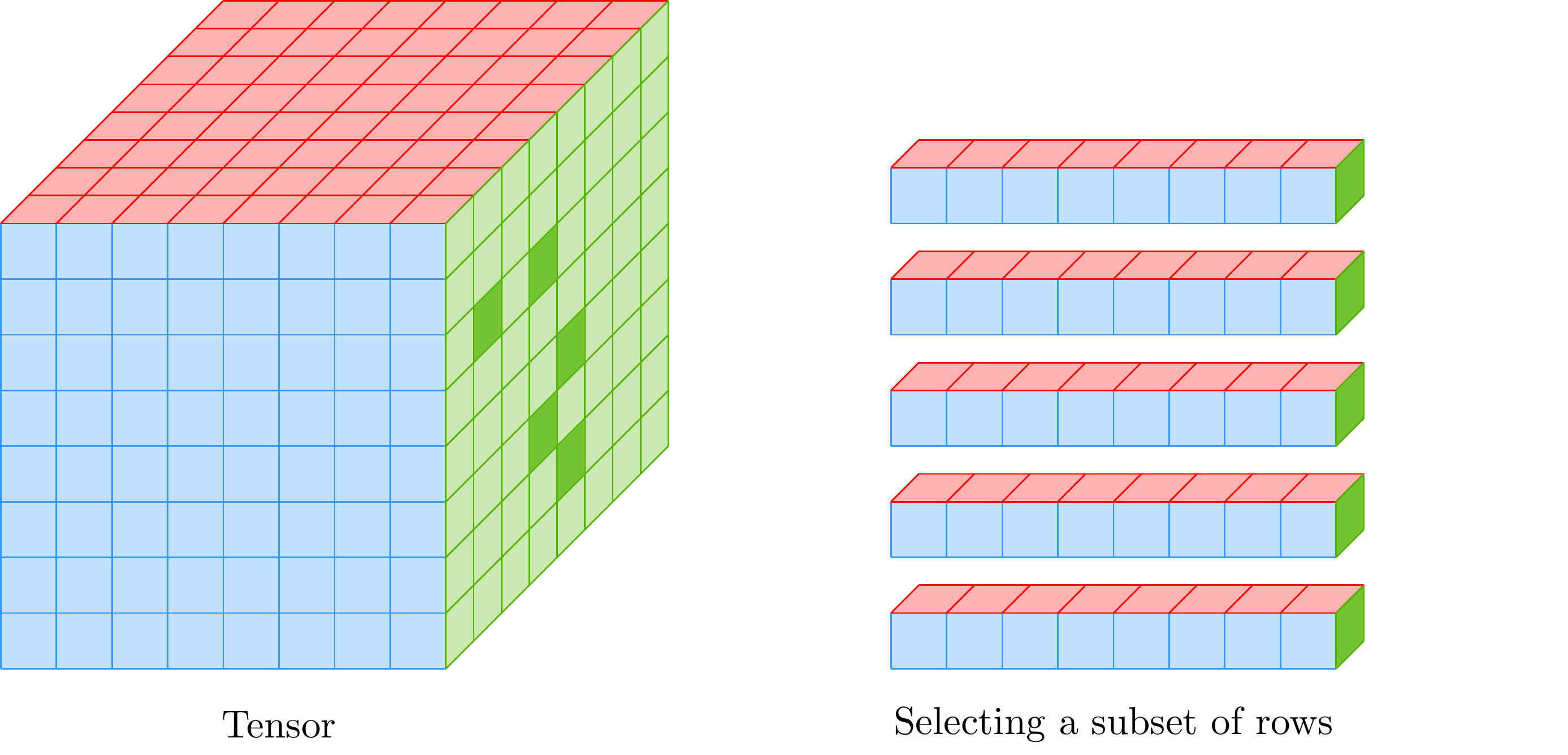}
    \includegraphics[width=0.8\textwidth]{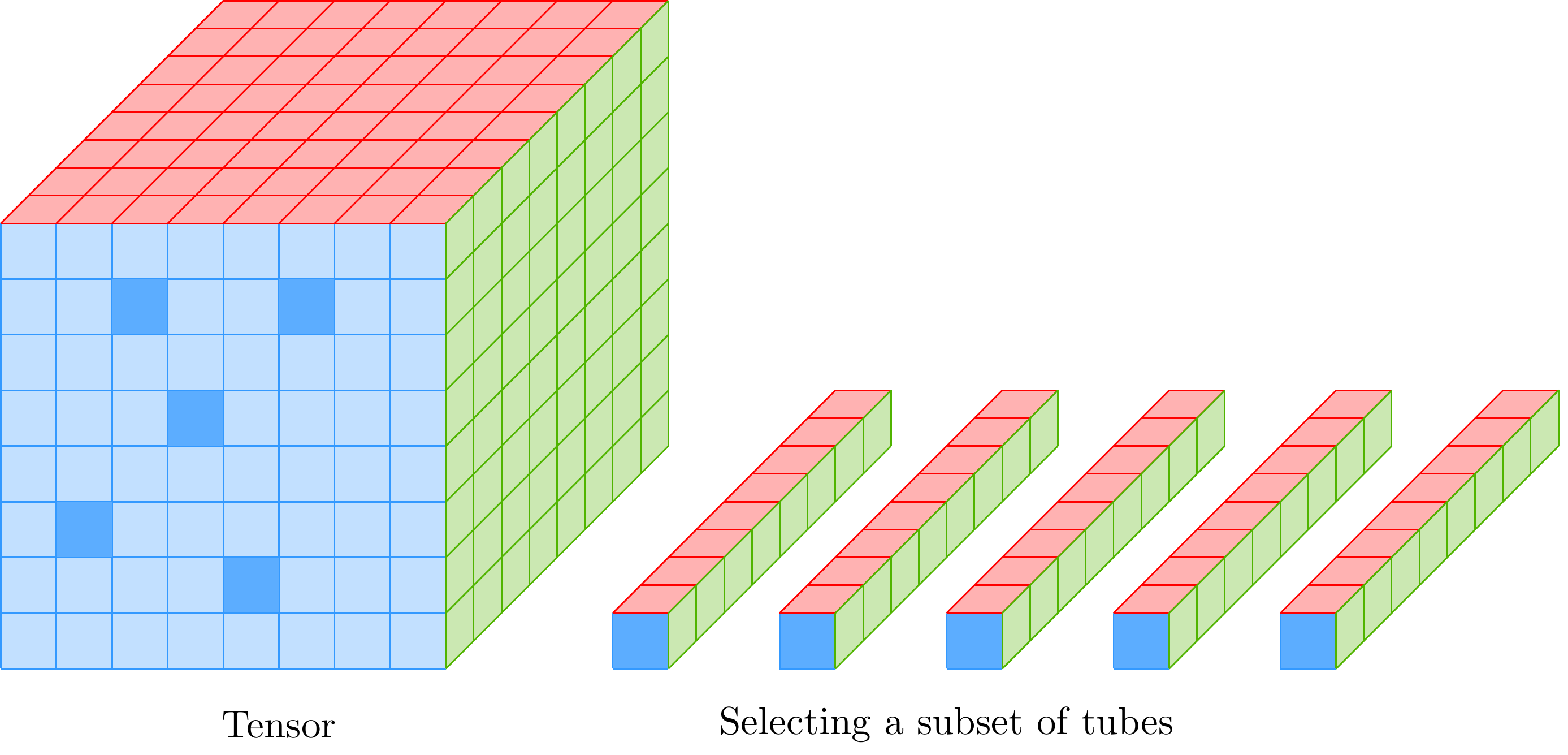}
    \caption{Column subset selection, row subset selection and tube subset selection.}
\end{figure}

Recently, \cite{tm17} studied a very different face-based tensor-CUR decomposition, which selects faces from tensors rather than columns. To achieve their results, \cite{tm17} need to make several incoherence assumptions on the original tensor. Their sample complexity depends on $\log n$, and they only sample two of the three dimensions. We will provide more general face-based tensor CURT decompositions.

 \begin{definition}[Tensor (face-based) CURT, exact]
Given a tensor $A\in \mathbb{R}^{ n_1 \times n_2 \times n_3} $, we choose three sets of coordinates $S_1 \subseteq [n_1], S_2 \subseteq [n_2], S_3 \subseteq [n_3]$. We define $c=|S_1|$, $r=|S_2|$ and $t=|S_3|$. Let $C \in \mathbb{R}^{ c\times n_2 \times n_3}$ denote a subset of row-tube faces of $A$, $R \in \mathbb{R}^{ n_1 \times r \times n_3}$ denote a subset of column-tube faces of $A$, and $T \in \mathbb{R}^{ n_1 \times n_2 \times t}$ denote a subset of column-row faces of $A$. Let $C_2\in \mathbb{R}^{n_2 \times cn_3}$ denote the matrix obtained by flattening the tensor $C$ along the second dimension. Let $R_3\in \mathbb{R}^{n_3 \times rn_1}$ denote the matrix obtained by flattening the tensor $R$ along the third dimension. Let $T_1 \in \mathbb{R}^{n_1 \times tn_2}$ denote the matrix obtained by flattening the tensor $T$ along the first dimension. If there exists a tensor $U\in \mathbb{R}^{t n_2 \times cn_3 \times r n_1 }$ such that $A$ can be written as
\begin{align*}
 \sum_{i=1}^{t n_2} \sum_{j=1}^{cn_3} \sum_{l=1}^{rn_1}   U_{i,j,l} (T_1)_{l} \otimes (C_2)_{i} \otimes (R_3)_{j} = A,
\end{align*}
\begin{align*}
U(T_1,C_2,R_3) = A,
\end{align*}
or equivalently,
\begin{align*}
\forall (i',j',l') \in [n_1] \times [n_2] \times [n_3], A_{i,j,l} =\sum_{i=1}^{t n_1} \sum_{j=1}^{cn_3} \sum_{l=1}^{rn_2}  U_{i,j,l} (T_1)_{i',i} (C_2)_{j',j} (R_3)_{l',l} ,
\end{align*}
then we say $C,U,R, T$ is tensor $A$'s (face-based) CURT decomposition.
 \end{definition}

\begin{definition}[Tensor (face-based) CURT, approximate]\label{def:tensor_face_curt}
Given a tensor $A\in \mathbb{R}^{ n_1 \times n_2 \times n_3} $, for some $k\geq 1$, for some approximation $\alpha>1$, for some norm $\| \|_{\xi}$,we choose three sets of coordinates $S_1 \subseteq [n_1], S_2 \subseteq [n_2], S_3 \subseteq [n_3]$. We define $c=|S_1|$, $r=|S_2|$ and $t=|S_3|$. Let $C \in \mathbb{R}^{ c\times n_2 \times n_3}$ denote a subset of row-tube faces of $A$, $R \in \mathbb{R}^{ n_1 \times r \times n_3}$ denote a subset of column-tube faces of $A$, and $T \in \mathbb{R}^{ n_1 \times n_2 \times t}$ denote a subset of column-row faces of $A$. Let $C_2\in \mathbb{R}^{n_2 \times cn_3}$ denote the matrix obtained by flattening the tensor $C$ along the second dimension. Let $R_3\in \mathbb{R}^{n_3 \times rn_1}$ denote the matrix obtained by flattening the tensor $R$ along the third dimension. Let $T_1 \in \mathbb{R}^{n_1 \times tn_2}$ denote the matrix obtained by flattening the tensor $T$ along the first dimension.  If there exists a tensor $U\in \mathbb{R}^{t n_2 \times cn_3 \times r n_1 }$ such that
\begin{align*}
\| U(T_1,C_2,R_3 ) - A \|_{\xi} \leq \alpha \min_{\rank-k~A_k} \| A_k - A \|_{\xi},
\end{align*}
where $\|\|_{\xi}$ is operator norm, Frobenius norm or Entry-wise $\ell_1$ norm, then we refer to $C, U, R, T$ as an approximate
CUR decomposition of $A$, and sometimes just refer to this as a (face-based) CURT decomposition of $A$.
\end{definition}

\subsection{Polynomial system verifier}\label{sec:pre_decision_solver}

We use the polynomial system verifiers independently developed by Renegar \cite{r92a,r92b} and Basu $et~al.$ \cite{bpr96}.
\begin{theorem}[Decision Problem \cite{r92a,r92b,bpr96}]
  \label{thm:decision_solver}

Given a real polynomial system $P(x_1, x_2, \cdots, x_v)$ having $v$ variables and $m$ polynomial constraints $f_i (x_1, x_2, \cdots, x_v) \Delta_i 0, \forall i \in [m]$, where $\Delta_i$ is any of the ``standard relations'': $\{ >, \geq, =, \neq, \leq, < \}$, let $d$ denote the maximum degree of all the polynomial constraints and let $H$ denote the maximum bitsize of the coefficients of all the polynomial constraints. Then in
\begin{equation*}
(m d)^{O(v)} \poly(H),
\end{equation*}
time one can determine if there exists a solution to the polynomial system $P$.
\end{theorem}

\begin{figure}[!t]
  \centering
    \includegraphics[width=1.0\textwidth]{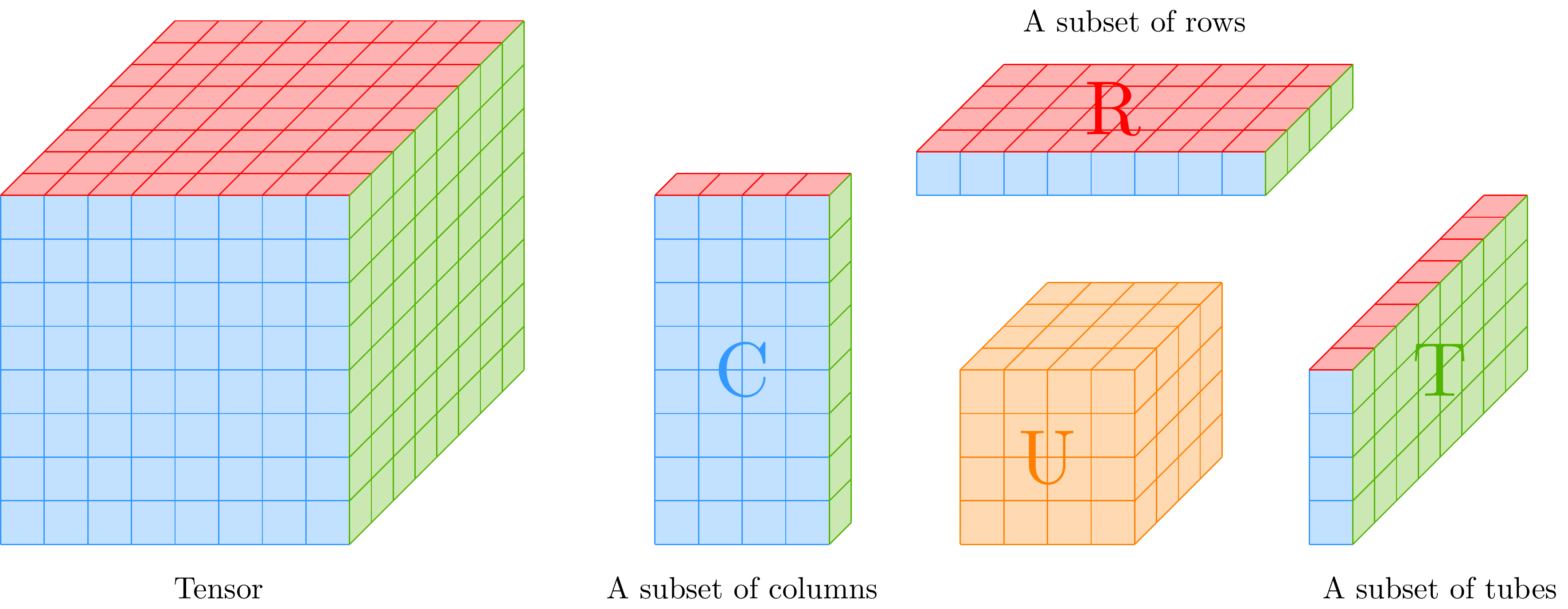}
    \caption{An example tensor CURT decomposition.}
\end{figure}
Recently, this technique has been used to solve a number of low-rank approximation and matrix factorization problems \cite{agkm12,m13,cw15focs,bdl16,rsw16,swz17}.

\subsection{Lower bound on the cost of a polynomial system}\label{sec:pre_lower_bound_on_cost}
An important result we use is the following lower bound on the minimum value attained by a polynomial restricted to a compact connected component of a basic closed semi-algebraic subset of $\mathbb{R}^v$.
\begin{theorem}[\cite{jpt13}]
  \label{thm:minimum_positive}
Let $T=\{ x\in \mathbb{R}^v | f_1(x) \geq 0, \cdots,$ $ f_{\ell}(x) \geq 0,  f_{\ell+1}(x) = 0, \cdots, f_m(x) = 0 \}$ be defined by polynomials $f_1, \cdots, f_m \in \mathbb{Z}[x_1, \cdots, x_v ]$ with $n \geq 2$, degrees bounded by an even integer $d$, and coefficients of absolute value at most $H$, and let $C$ be a compact connected (in the topological sense) component of $T$. Let $g \in \mathbb{Z}[x_1, \cdots, x_v]$ be a polynomial of degree at most $d$ and coefficients of absolute value bounded by $H$. Then, the minimum value that $g$ takes over $C$
satisfies that if it is not zero, then
its absolute value is greater than or equal to
\begin{equation*}
(2^{4-v/2} \widetilde{H} d^v)^{-v 2^v d^v},
\end{equation*}
where $\widetilde{H} = \max\{ H, 2v+2m\}$.
\end{theorem}

While the above theorem involves notions from topology, we shall apply
it in an elementary way. Namely, in our setting $T$ will be bounded
and so every connected component, which is by definition closed, will also
be bounded and therefore compact. As the connected components partition $T$
the theorem will just be applied to give a global minimum value of $g$ on $T$
provided that it is non-zero.

\subsection{Frobenius norm and $\ell_2$ relaxation }\label{sec:pre_relaxation}

\begin{theorem}[Generalized rank-constrained matrix approximations, Theorem 2 in \cite{ft07}]\label{thm:reduce_to_frobenius}
Given matrices $A\in \mathbb{R}^{n\times d}$, $B\in \mathbb{R}^{n\times p}$, and $C\in \mathbb{R}^{q\times d}$, let the SVD of $B$ be $B=U_B\Sigma_B V_B^\top$ and the SVD of $C$ be $C=U_C\Sigma_C V_C^\top$. Then,
\begin{equation*}
B^\dagger ( U_B U_B^\top A V_C C_C^\top )_k C^\dagger = \underset{ \rank-k ~X\in \mathbb{R}^{p\times q} }{\arg \min} \| A - B X C \|_F,
\end{equation*}
where $(U_B U_B^\top A V_C V_C^\top)_k \in \mathbb{R}^{p\times q}$ is of rank at most $k$ and denotes the best rank-$k$ approximation to $U_B U_B^\top A V_C V_C^\top \in \mathbb{R}^{p\times d}$ in Frobenius norm.
\end{theorem}

\begin{claim}[$\ell_2$ relaxation of $\ell_p$-regression]\label{cla:ell2_relax_ell1_regression}
Let $p\in [1,2)$. For any $A\in \mathbb{R}^{n\times d}$ and $b\in \mathbb{R}^n$, define $x^* =\underset{x\in \mathbb{R}^d}{\arg\min} \| A x - b\|_p $ and $x'=\underset{x \in \mathbb{R}^d}{\arg\min} \| A x - b\|_2$. Then,
\begin{equation*}
\| Ax^* -b\|_p \leq \| A x' - b\|_p \leq n^{1/p-1/2} \cdot \| A x^* - b\|_p.
\end{equation*}
\end{claim}

\begin{claim}[(Matrix) Frobenius norm relaxation of $\ell_p$-low rank approximation]\label{cla:matrix_frobenius_relax_ell1_lowrank}
Let $p\in [1,2)$ and for any matrix $A\in \mathbb{R}^{n\times d}$, define $A^* = \underset{\rank-k~B\in \mathbb{R}^{n\times d}}{\arg \min} \| B - A\|_p$ and $A' =\underset{\rank-k~B \in \mathbb{R}^{n\times d} }{\arg \min} \| B - A \|_F$. Then
\begin{align*}
\| A^* - A \|_p \leq \| A'- A \|_p \leq (nd)^{1/p-1/2} \| A^* - A \|_p.
\end{align*}
\end{claim}

\begin{claim}[(Tensor) Frobenius norm relaxation of $\ell_p$-low rank approximation]\label{cla:tensor_frobenius_relax_ell1_lowrank}
Let $p\in [1,2)$ and for any matrix $A\in \mathbb{R}^{n_1 \times n_2 \times n_3}$, define
\begin{align*}
A^* = \underset{\rank-k~B\in \mathbb{R}^{n_1\times n_2 \times n_3}}{\arg \min} \| B - A\|_p
\end{align*}
 and
\begin{align*}
A' =\underset{\rank-k~B \in \mathbb{R}^{n_1\times n_2 \times n_3} }{\arg \min} \| B - A \|_F.
\end{align*}
Then
\begin{align*}
\| A^* - A \|_p \leq \| A'- A \|_p \leq (n_1 n_2 n_3)^{1/p-1/2} \| A^* - A \|_p.
\end{align*}
\end{claim}

\subsection{CountSketch and Gaussian transforms}\label{sec:def_count_sketch_gaussian}
\begin{definition}[Sparse embedding matrix or CountSketch transform]\label{def:count_sketch_transform}
A CountSketch transform is defined to be $\Pi=\sigma \cdot \Phi D\in \mathbb{R}^{m\times n}$. Here, $\sigma$ is a scalar, $D$ is an $n\times n$ random diagonal matrix with each diagonal entry independently chosen to be $+1$ or $-1$ with equal probability, and $\Phi\in \{0,1\}^{m\times n}$ is an $m\times n$ binary matrix with $\Phi_{h(i),i}=1$ and all remaining entries $0$, where $h:[n]\rightarrow [m]$ is a random map such that for each $i\in [n]$, $h(i) = j$ with probability $1/m$ for each $j \in [m]$. For any matrix $A\in \mathbb{R}^{n\times d}$, $\Pi A$ can be computed in $O(\nnz(A))$ time. For any tensor $A\in \mathbb{R}^{n\times d_1 \times d_2}$, $\Pi A$ can be computed in $O(\nnz(A))$ time. Let $\Pi_1, \Pi_2, \Pi_3$ denote three CountSktech transforms. For any tensor $A\in \mathbb{R}^{n_1 \times n_2 \times n_3}$, $A(\Pi_1,\Pi_2,\Pi_3)$ can be computed in $O(\nnz(A))$ time.
\end{definition}
If the above scalar $\sigma$ is not specified in the context, we assume the scalar $\sigma$ to be $1$.

\begin{definition}[Gaussian matrix or Gaussian transform]\label{def:gaussian_transform}
Let $S=\sigma \cdot G \in \mathbb{R}^{m\times n}$ where $\sigma$ is a scalar, and each entry of $G\in \mathbb{R}^{m\times n}$ is chosen independently from the standard Gaussian distribution. For any matrix $A\in \mathbb{R}^{n\times d}$, $SA$ can be computed in $O(m \cdot \nnz(A))$ time. For any tensor $A\in \mathbb{R}^{n \times d_1 \times d_2}$, $SA$ can be computed in $O(m \cdot \nnz(A))$ time.
\end{definition}
If the above scalar $\sigma$ is not specified in the context, we assume the scalar $\sigma$ to be $1/\sqrt{m}$.
%
In most places, we can combine CountSketch and Gaussian
transforms to achieve the following:
\begin{definition}[CountSketch + Gaussian transform]\label{def:fast_gaussian_transform}
Let $S' = S \Pi$, where $\Pi\in \mathbb{R}^{t\times n}$ is the CountSketch transform (defined in Definition~\ref{def:count_sketch_transform}) and $S\in \mathbb{R}^{m \times t}$ is the Gaussian transform (defined in Definition~\ref{def:gaussian_transform}). For any matrix $A\in \mathbb{R}^{n\times d}$, $S'A$ can be computed in $O(\nnz(A) + dtm^{\omega-2})$ time, where $\omega$ is the matrix multiplication exponent.
\end{definition}

\begin{lemma}[Affine Embedding - Theorem 39 in~\cite{cw13}]\label{lem:affine_embedding}
Given matrices $A\in\mathbb{R}^{n\times r},B\in\mathbb{R}^{n\times d}$, and $\rank(A)=k$, let $m=\poly(k/\varepsilon)$, $S\in\mathbb{R}^{m\times n}$ be a sparse embedding matrix (Definition~\ref{def:count_sketch_transform}) with scalar $\sigma=1$. Then with probability at least $0.999$, $\forall X\in\mathbb{R}^{r\times d},$ we have
\begin{align*}
(1-\varepsilon)\cdot\|AX-B\|_F^2\leq \|S(AX-B)\|_F^2\leq (1+\varepsilon)\|AX-B\|_F^2.
\end{align*}
\end{lemma}

\begin{lemma}[see, e.g., Lemma 10 in version 1 of~\cite{bwz16}\footnote{\label{fot:bwz16} \url{https://arxiv.org/pdf/1504.06729v1.pdf}}]\label{lem:gaussian_sketch_for_regression}
Let $m=\Omega(k/\varepsilon),$ $S=\frac{1}{\sqrt{m}}\cdot G$, where $G\in\mathbb{R}^{m\times n}$ is a random matrix where each entry is an i.i.d Gaussian $N(0,1)$. Then with probability at least $0.998$, $S$ satisfies $(1\pm1/8)$ Subspace Embedding (Definition~\ref{def:subspace_embedding}) for any fixed matrix $C\in\mathbb{R}^{n\times k},$ and it also satisfies $O(\sqrt{\varepsilon/k})$ Approximate Matrix Product (Definition~\ref{def:approximate_matrix_product}) for any fixed matrix $A$ and $B$ which has the same number of rows.
\end{lemma}

\begin{lemma}[see, e.g., Lemma 11 in version 1 of~\cite{bwz16}\footref{fot:bwz16}]\label{lem:count_sketch_for_regression}
Let $m=\Omega(k^2+k/\varepsilon),$ $\Pi\in\mathbb{R}^{m\times n}$, where $\Pi$ is a sparse embedding matrix (Definition~\ref{def:count_sketch_transform}) with scalar $\sigma=1$, then with probability at least $0.998$,  $S$ satisfies $(1\pm1/8)$ Subspace Embedding (Definition~\ref{def:subspace_embedding}) for any fixed matrix $C\in\mathbb{R}^{n\times k},$ and it also satisfies $O(\sqrt{\varepsilon/k})$ Approximate Matrix Product (Definition~\ref{def:approximate_matrix_product}) for any fixed matrix $A$ and $B$ which has the same number of rows.
\end{lemma}

\begin{lemma}[see, e.g., Lemma 12 in version 1 of~\cite{bwz16}\footref{fot:bwz16}]\label{lem:gaussian_count_sketch_for_regression}
Let $m_2=\Omega(k^2+k/\varepsilon),$ $\Pi\in\mathbb{R}^{m_2\times n}$, where $\Pi$ is a sparse embedding matrix (Definition~\ref{def:count_sketch_transform}) with scalar $\sigma=1$. Let $m_1=\Omega(k/\varepsilon),$ $S=\frac{1}{\sqrt{m_1}}\cdot G$, where $G\in\mathbb{R}^{m_1\times m_2}$ is a random matrix where each entry is an i.i.d Gaussian $N(0,1)$. Let $S'=S\Pi$. Then with probability at least $0.99$,  $S'$ is a $(1\pm1/3)$ Subspace Embedding (Definition~\ref{def:subspace_embedding}) for any fixed matrix $C\in\mathbb{R}^{n\times k},$ and it also satisfies $O(\sqrt{\varepsilon/k})$ Approximate Matrix Product (Definition~\ref{def:approximate_matrix_product}) for any fixed matrix $A$ and $B$ which have the same number of rows.
\end{lemma}

\begin{theorem}[Theorem 36 in~\cite{cw13}]\label{thm:multiple_regression_sketch}
Given $A\in\mathbb{R}^{n\times k},B\in\mathbb{R}^{n\times d},$ suppose $S\in\mathbb{R}^{m\times n}$ is such that $S$ is a $(1\pm\frac{1}{\sqrt{2}})$ Subspace Embedding for $A$, and satisfies $O(\sqrt{\varepsilon/k})$ Approximate Matrix Product for matrices $A$ and $C$ where $C$ with $n$ rows, where $C$ depends on $A$ and $B$. If
\begin{align*}
\wh{X}=\arg\min_{X\in\mathbb{R}^{k\times d}} \|SAX-SB\|_F^2,
\end{align*}
then
\begin{align*}
\|A\wh{X}-B\|_F^2\leq (1+\varepsilon)\min_{X\in\mathbb{R}^{k\times d}}\|AX-B\|_F^2.
\end{align*}
\end{theorem}

\subsection{Cauchy and $p$-stable transforms}\label{sec:def_cauchy_pstable}

\begin{definition}[Dense Cauchy transform]
Let $S = \sigma \cdot C \in \mathbb{R}^{m\times n}$ where $\sigma$ is a scalar, and each entry of $C\in\mathbb{R}^{m\times n}$ is chosen independently from the standard Cauchy distribution. For any matrix $A\in \mathbb{R}^{n\times d}$, $SA$ can be computed in $O(m\cdot \nnz(A))$ time.
\end{definition}
\begin{definition}[Sparse Cauchy transform]
Let $\Pi = \sigma \cdot S C\in \mathbb{R}^{m \times n}$, where $\sigma$ is a scalar, $S\in \mathbb{R}^{m\times n}$ has each column chosen independently and uniformly from the $m$ standard basis vectors of $\mathbb{R}^{m}$, and $C\in \mathbb{R}^{n\times n}$ is a diagonal matrix with diagonals chosen independently from the standard Cauchy distribution.  For any matrix $A\in \mathbb{R}^{n\times d}$, $\Pi A$ can be computed in $O(\nnz(A))$ time. For any tensor $A\in \mathbb{R}^{n\times d_1 \times d_2}$, $\Pi A$ can be computed in $O(\nnz(A))$ time. Let $\Pi_1 \in \mathbb{R}^{m_1 \times n_1},\Pi_2\in \mathbb{R}^{m_2 \times n_2},\Pi_3\in \mathbb{R}^{m_3 \times n_3}$ denote three sparse Cauchy transforms. For any tensor $A\in \mathbb{R}^{n_1\times n_2\times n_3}$, $A(\Pi_1,\Pi_2,\Pi_3)\in \mathbb{R}^{m_1 \times m_2 \times m_3}$ can be computed in $O(\nnz(A))$ time.
\end{definition}

\begin{definition}[Dense $p$-stable transform]
Let $p\in (1,2)$. Let $S = \sigma \cdot C \in \mathbb{R}^{m\times n}$, where $\sigma$ is a scalar, and each entry of $C\in\mathbb{R}^{m\times n}$ is chosen independently from the standard $p$-stable distribution. For any matrix $A\in \mathbb{R}^{n\times d}$, $SA$ can be computed in $O(m \nnz(A) )$ time.
\end{definition}

\begin{definition}[Sparse $p$-stable transform]
Let $p\in (1,2)$. Let $\Pi = \sigma \cdot S C\in \mathbb{R}^{m\times n}$, where $\sigma$ is a scalar, $S\in \mathbb{R}^{m\times n}$ has each column chosen independently and uniformly from the $m$ standard basis vectors of $\mathbb{R}^{m}$, and $C\in \mathbb{R}^{n\times n}$ is a diagonal matrix with diagonals chosen independently from the standard $p$-stable distribution. For any matrix $A\in \mathbb{R}^{n\times d}$, $\Pi A$ can be computed in $O(\nnz(A))$ time. For any tensor $A\in \mathbb{R}^{n\times d_1 \times d_2}$, $\Pi A$ can be computed in $O(\nnz(A))$ time. Let $\Pi_1 \in \mathbb{R}^{m_1 \times n_1},\Pi_2\in \mathbb{R}^{m_2 \times n_2},\Pi_3\in \mathbb{R}^{m_3 \times n_3}$ denote three sparse $p$-stable transforms. For any tensor $A\in \mathbb{R}^{n_1\times n_2\times n_3}$, $A(\Pi_1,\Pi_2,\Pi_3)\in \mathbb{R}^{m_1 \times m_2 \times m_3}$ can be computed in $O(\nnz(A))$ time.
\end{definition}

\subsection{Leverage scores}\label{sec:def_leverage_score}
\begin{definition}[Leverage scores]
Let $U\in \mathbb{R}^{n\times k}$ have orthonormal columns, and let $p_i = u_i^2 / k$, where $u_i^2 = \| e_i^\top U \|_2^2$ is the $i$-th leverage score of $U$. 
\end{definition}

\begin{definition}[Leverage score sampling]
Given $A\in\mathbb{R}^{n\times d}$ with rank $k$, let $U\in \mathbb{R}^{n\times k}$ be an orthonormal basis of the column space of $A$, and for each $i$ let $p_i$ be the squared row norm of the $i$-th row of $U$, i.e., the $i$-th leverage score. Let $k\cdot p_i$ denote the $i$-th leverage score of $U$ scaled by $k$. Let $\beta>0$ be a constant and $q=(q_1, \cdots,q_n)$ denote a distribution such that, for each $i\in [n]$, $q_i \geq \beta p_i$. Let $s$ be a parameter. Construct an $n\times s$ sampling matrix $B$ and an $s\times s$ rescaling matrix $D$ as follows. Initially, $B = 0^{n\times s}$ and $D = 0^{s\times s}$. For each column $j$ of $B$, $D$, independently, and with replacement, pick a row index $i\in [n]$ with probability $q_i$, and set $B_{i,j}=1$ and $D_{j,j}=1/\sqrt{q_i s}$. We denote this procedure \textsc{Leverage score sampling} according to the matrix $A$.
\end{definition}

\subsection{Lewis weights}\label{sec:def_lewis_weights}
We follow the exposition of Lewis weights from \cite{cp15}.
\begin{definition}
For a matrix $A$, let $a_i$ denote the $i^{\text{th}}$ row of $A$, where $a_i(=(A^i)^\top )$ is a column vector. The statistical leverage score of a row $a_i$ is
\begin{align*}
\tau_i(A) \overset{\mathrm{def}}{=} a_i^\top (A^\top A)^{-1} a_i = \| (A^\top A )^{-1/2} a_i \|_2^2.
\end{align*}
For a matrix $A$ and norm $p$, the $\ell_p$ Lewis weights $w$ are the unique weights such that for each row $i$ we have
\begin{align*}
w_i = \tau_i ( W^{1/2-1/p} A ).
\end{align*}
or equivalently,
\begin{align*}
a_i^\top (A^\top W^{1-2/p} A)^{-1} a_i = w_i^{2/p}.
\end{align*}
\end{definition}


\begin{lemma}[Lemma~2.4 of \cite{cp15} and Lemma~7 of \cite{clmmps15}]\label{lem:compute_lewis_weight}
Given a matrix $A\in\mathbb{R}^{n\times d}$, $n\geq d$, for any constant $C>0,4>p\geq 1$, there is an algorithm which can compute $C$-approximate $\ell_p$ Lewis weights for every row $i$ of $A$ in $O((nnz(A)+d^\omega\log d)\log n) $ time, where $\omega < 2.373$ is the matrix multiplication exponent\cite{s69,cw87,w12}.
\end{lemma}

\begin{lemma}[Theorem~7.1 of \cite{cp15}]\label{lem:num_samples}
Given matrix $A\in\mathbb{R}^{n\times d}$ ($n\geq d$) with $\ell_p$ ($4>p\geq 1$) Lewis weights $w$, for any set of sampling probabilities $p_i$, $\sum_i p_i=N$,
\begin{align*}
p_i\geq f(d,p)w_i,
\end{align*}
if $S\in\mathbb{R}^{N\times n}$ has each row chosen independently as the $i^{\text{th}}$ standard basis vector, multiplied by $1/p_i^{1/p}$, with probability $p_i/N$. Then, overall with probability at least $0.999$,
\begin{align*}
\forall x\in\mathbb{R}^d, \frac12\|Ax\|_p^p\leq\|SAx\|_p^p\leq 2\|Ax\|_p^p.
\end{align*}
Furthermore, if $p=1$, $N=O(d\log d)$. If $1<p<2$, $N=O(d\log d\log \log d)$. If $2\leq p<4$, $N=O(d^{p/2}\log d)$.
\end{lemma}

\begin{lemma}
  Given matrix $A\in\mathbb{R}^{n\times d}$ ($n\geq d$), there is an
  algorithm to compute a diagonal matrix $D=SS_1$ with $N$ nonzero entries in $O(n\poly(d))$ time such that, with probability at least $0.999$, for all $x\in\mathbb{R}^d$
\begin{align*}
\frac{1}{10}\|DAx\|_p^p\leq \|Ax\|_p^p\leq 10\|DAx\|_p^p,
\end{align*}
where $S,S_1$ are two sampling/rescaling matrices.
Furthermore, if $p=1$, then $N=O(d\log d)$. If $1<p<2$, then $N=O(d\log d\log \log d)$. If $2\leq p<4$, then $N=O(d^{p/2}\log d)$.
\end{lemma}

Given a matrix $A\in\mathbb{R}^{n\times d}$ ($n\geq d$), by Lemma~\ref{lem:num_samples} and Lemma~\ref{lem:compute_lewis_weight}, we can compute a sampling/rescaling matrix $S$ in $O((nnz(A)+d^\omega\log d)\log n) $ time with $\wt{O}(d)$ nonzero entries such that
\begin{align*}
\forall x\in\mathbb{R}^d, \frac12\|Ax\|_p^p\leq\|SAx\|_p^p\leq 2\|Ax\|_p^p.
\end{align*}
Sometimes, $\poly(d)$ is much smaller than $\log n$. In this case, we are also able to compute such a sampling/rescaling matrix $S$ in $n\poly(d)$ time in an alternative way.

To do so, we run one of the input sparsity $\ell_p$ embedding algorithms (see e.g., \cite{mm13}) to compute a well conditioned basis $U$ of the column span of $A$ in $n\poly(d/\varepsilon)$ time. By sampling according to the well conditioned basis (see e.g. \cite{c05,ddhkm09,w14}), we can compute a sampling/rescaling matrix $S_1$ such that $(1-\varepsilon)\|Ax\|_p^p\leq\|S_1Ax\|_p^p\leq (1+\varepsilon)\|Ax\|_p^p$ where $\varepsilon\in(0,1)$ is an arbitrary constant. Notice that $S_1$ has $\poly(d/\varepsilon)$ nonzero entries, and thus $S_1A$ has size $\poly(d/\varepsilon)$. Next, we apply Lewis weight sampling according to $S_1A$, and we obtain a sampling/rescaling matrix $S$ for which
\begin{align*}
\forall x\in\mathbb{R}^d, (1-\frac{1}{3})\|S_1Ax\|_p^p\leq\|SS_1Ax\|_p^p\leq (1+\frac{1}{3})\|S_1Ax\|_p^p.
\end{align*}
This implies that
\begin{align*}
\forall x\in\mathbb{R}^d, \frac{1}{2}\|Ax\|_p^p\leq\|SS_1Ax\|_p^p\leq 2\|Ax\|_p^p.
\end{align*}
Note that $SS_1$ is still a sampling/rescaling matrix according to $A$, and the number of non-zero entries is $\wt{O}(d)$. The total running time is thus $n\poly(d/\epsilon)$, as desired.

\subsection{\texorpdfstring{\textsc{TensorSketch}}~}\label{sec:def_tensor_sketch}

Let $\phi(v_1,v_2,\cdots,v_q)$ denote the function that maps $q$ vectors($u_i\in \mathbb{R}^{n_i}$) to the $\prod_{i=1}^q n_i$-dimensional vector formed by $v_1\otimes v_2 \otimes \cdots \otimes u_q$.

We first give the definition of \textsc{TensorSketch}. Similar definitions can be found in previous work \cite{p13,pp13,anw14,wtsa15}.
\begin{definition}[\textsc{TensorSketch} \cite{p13}]\label{def:tensor_sketch}
Given $q$ points $v_1,v_2,\cdots, v_q$ where for each $i\in [q], v_i \in \mathbb{R}^{n_i}$, let $m$ be the target dimension. The \textsc{TensorSketch} transform is specified using $q$ $3$-wise independent hash functions, $h_1,\cdots,h_q$, where for each $i\in [q]$, $h_i : [n_i]\rightarrow [m]$, as well as $q$ $4$-wise independent sign functions $s_1, \cdots, s_q$, where for each $i\in [q]$, $s_i : [n_i] \rightarrow \{-1,+1\}$.

\textsc{TensorSketch} applied to $v_1,\cdots,v_q$ is then \textsc{CountSketch} applied to $\phi(v_1,\cdots,v_q)$ with hash function $H: [\prod_{i=1}^q n_i] \rightarrow [m]$ and sign functions $S : [\prod_{i=1}^q n_i] \rightarrow \{-1,+1\}$ defined as follows:
\begin{align*}
H(i_1,\cdots,i_q) = h_1(i_1) + h_2(s_2) + \cdots + h_q(i_q) {\pmod m},
\end{align*}
and
\begin{align*}
S(i_1,\cdots,i_q) = s_1(i_1) \cdot s_2(i_2) \cdot \cdots \cdot s_q(i_q).
\end{align*}
Using the Fast Fourier Transform, \textsc{TensorSketch}($v_1,\cdots,v_q$) can be computed in $O(\sum_{i=1}^q (\nnz(v_i) + m\log m) )$ time.
\end{definition}
Note that Theorem 1 in \cite{anw14} only defines $\phi(v) = v\otimes v\otimes \cdots \otimes v$. Here we state a stronger version of Theorem 1 than in \cite{anw14}, though the proofs are identical; a formal derivation can be found in \cite{dw17}.
\begin{theorem}[Generalized version of Theorem 1 in \cite{anw14}]\label{thm:theorem1_anw14}
Let $S$ be the $(\prod_{i=1}^q n_i) \times m$ matrix such that \textsc{TensorSketch} ($v_1,v_2,\cdots,v_q$) is $\phi(v_1,v_2,\cdots,v_q) S$ for a randomly selected \textsc{TensorSketch}. The matrix $S$ satisfies the following two properties.

Property \RN{1} (Approximate Matrix Product). Let $A$ and $B$ be matrices with $\prod_{i=1}^q n_i$ rows. For $m\geq (2+3^q) / (\epsilon^2 \delta)$, we have
\begin{align*}
\Pr[  \| A^\top S S^\top B - A^\top B \|_F^2 \leq \epsilon^2 \| A \|_F^2 \| B \|_F^2 ] \geq 1-\delta.
\end{align*}

Property \RN{2} (Subspace Embedding). Consider a fixed $k$-dimensional subspace $V$. If $m \geq k^2 (2+3^q)/ (\epsilon^2 \delta)$, then with probability at least $1-\delta$, $\| x S \|_2 = (1\pm \epsilon) \| x \|_2$ simultaneously for all $x\in V$.
\end{theorem}

\newpage
\section{Frobenius Norm for Arbitrary Tensors}\label{sec:f}

Section~\ref{sec:f_1_plus_epsilon} presents a Frobenius norm tensor low-rank approximation algorithm with $(1+\epsilon)$-approximation ratio. Section~\ref{sec:f_input_sparsity_reduction} introduces a tool which is able to reduce the size of the objective function from $n^3$ to $\poly(k,1/\epsilon)$. Section~\ref{sec:f_tensor_multiple_regression} introduces a new problem called tensor multiple regression. Section~\ref{sec:f_bicriteria_algorithm} presents several bicriteria algorithms. Section~\ref{sec:f_generalized_matrix_row} introduces a powerful tool which we call generalized matrix row subset selection. Section~\ref{sec:f_columns_rows_tubes_subset_selection} presents an algorithm that is able to select a batch of columns, rows and tubes from a given tensor, and those samples are also able to form a low-rank solution. Section~\ref{sec:f_curt} presents several useful tools for tensor problems, and also two $(1+\epsilon)$-approximation CURT decomposition algorithms: one has the optimal sample complexity, and the other has the optimal running time. Section~\ref{sec:f_solving_small_problems} shows how to solve the problem if the size of the objective function is small. Section~\ref{sec:f_general_order} extends several techniques from $3$rd order tensors to general $q$-th order tensors, for any $q\geq 3$. Finally, in Section~\ref{sec:f_matrix_cur} we also provide a new matrix CUR decomposition algorithm, which is faster than \cite{bw14}.

For simplicity of presentation, we assume $A_k$ exists in theorems (e.g., Theorem~\ref{thm:f_main_algorithm}) which concern outputting a $\rank$-$k$ solution, as well as the theorems (e.g., Theorem~\ref{thm:f_bicriteria}, Theorem~\ref{thm:f_bicriteria_better}, Theorem~\ref{thm:f_bicriteria_best}) which concern outputting a bicriteria solution (the output rank is larger than $k$). For each of the bicriteria theorems, we can obtain a more detailed version when $A_k$ does not exist, like Theorem~\ref{thm:bicriteria} in Section~\ref{sec:intro} (by instead considering a tensor sufficiently close to $A_k$ in objective function value). Note that the theorems for column, row, tube subset selection Theorem~\ref{thm:f_fast_curt_without_u} and Theorem~\ref{thm:f_fast_curt_with_u_but_bicriteria} also belong to this first category. In the second category, for each of the rank-$k$ theorems we can obtain a more detailed version handling all cases, even when $A_k$ does not exist, like Theorem~\ref{thm:smallk} in Section~\ref{sec:intro} (by instead considering a tensor sufficiently close to $A_k$ in objective function value).

Several other tensor results or tools (e.g., Theorem~\ref{thm:f_tensor_multiple_regression}, Lemma~\ref{lem:f_input_sparsity_reduction}, Theorem~\ref{thm:f_curt_algorithm_input_sparsity}, Theorem \ref{thm:f_curt_algorithm_optimal_samples}, Theorem~\ref{thm:f_generalized_matrix_row}, Theorem~\ref{thm:f_fast_tensor_leverage_score_general_order}) that we build in this section do not belong to the above two categories. It means those results do not depend on whether $A_k$ exists or not and whether $\OPT$ is zero or not.

\subsection{$(1+\epsilon)$-approximate low-rank approximation}\label{sec:f_1_plus_epsilon}

\begin{algorithm}[!]\caption{Frobenius Norm Low-rank Approximation}\label{alg:f_main_algorithm}
\begin{algorithmic}[1]
\Procedure{\textsc{FLowRankApprox}}{$A,n,k,\epsilon$} \Comment{Theorem \ref{thm:f_main_algorithm}}
\State $s_1\leftarrow s_2 \leftarrow s_3 \leftarrow O(k/\epsilon)$.
\State Choose sketching matrices $S_1\in \mathbb{R}^{n^2 \times s_1}$, $S_2 \in \mathbb{R}^{n^2 \times s_2}$, $S_3\in \mathbb{R}^{n^2 \times s_3}$. \Comment{Definition \ref{def:fast_gaussian_transform}}
\State Compute $A_i S_i,\forall i\in [3]$. \label{sta:f_main_compute_AiSi}
\State $Y_1, Y_2, Y_3, C \leftarrow$\textsc{FInputSparsityReduction}($A,A_1S_1,A_2S_2,A_3S_3,n,s_1,s_2,s_3,k,\epsilon$). \label{sta:f_main_input_sparsity_reduction} \Comment{Algorithm~\ref{alg:f_input_sparsity_reduction}}
\State Create variables for $X_i \in \mathbb{R}^{s_i \times k}, \forall i\in [3]$.
\State Run polynomial system verifier for $\| (Y_1 X_1) \otimes (Y_2 X_2) \otimes (Y_3 X_3) -C\|_F^2$. \label{sta:f_main_solve_small_problem}
\State \Return $A_1 S_1 X_1$, $A_2 S_2 X_2$, and $A_3S_3 X_3$.
\EndProcedure
\end{algorithmic}
\end{algorithm}

\begin{theorem}\label{thm:f_main_algorithm}
Given a 3rd order tensor $A\in \mathbb{R}^{n\times n \times n}$, for any $k\geq 1,\varepsilon\in(0,1)$, there exists an algorithm which takes $O(\nnz(A)) + n \poly(k,1/\epsilon) + 2^{O(k^2/\epsilon)}$ time and outputs three matrices $U\in \mathbb{R}^{n\times k}$, $V\in \mathbb{R}^{n\times k}$, $W\in \mathbb{R}^{n\times k}$ such that
\begin{align*}
\left\| \sum_{i=1}^k U_i \otimes V_i \otimes W_i -A \right\|_F^2 \leq (1+\epsilon) \underset{\rank-k~A_k}{\min} \| A_k -A \|_F^2
\end{align*}
holds with probability $9/10$.
\end{theorem}
\begin{proof}

Given any tensor $A\in \mathbb{R}^{n_1\times n_2 \times n_3}$, we define three matrices $ A_1 \in \mathbb{R}^{n_1 \times n_2 n_3}, A_2 \in \mathbb{R}^{n_2 \times n_3 n_1}, A_3 \in \mathbb{R}^{n_3 \times n_1 n_2}$ such that, for any $i\in [n_1], j \in [n_2], l \in [n_3]$,
\begin{align*}
A_{i,j,l} = ( A_{1})_{i, (j-1) \cdot n_3 + l} = ( A_{2} )_{ j, (l-1) \cdot n_1 + i } = ( {A}_3)_{l, (i-1) \cdot n_2 + j }.
\end{align*}

We define $\OPT$ as
\begin{align*}
\OPT=\underset{\rank-k~A'}{\min} \| A' -A \|_F^2.
\end{align*}

Suppose the optimal $A_k=U^*\otimes V^*\otimes W^*.$
We fix $V^* \in \mathbb{R}^{n\times k}$ and $W^* \in \mathbb{R}^{n\times k}$. We use $V_1^*, V_2^*, \cdots, V_k^*$ to denote the columns of $V^*$ and $W_1^*, W_2^*, \cdots, W_k^*$ to denote the columns of $W^*$.

We consider the following optimization problem,
\begin{align*}
\min_{U_1, \cdots, U_k \in \mathbb{R}^n } \left\| \sum_{i=1}^k U_i \otimes V_i^* \otimes W_i^* - A \right\|_F^2,
\end{align*}
which is equivalent to
\begin{align*}
\min_{U_1, \cdots, U_k \in \mathbb{R}^n } \left\|
\begin{bmatrix}
U_1 & U_2 & \cdots & U_k
\end{bmatrix}
\begin{bmatrix}
 V_1^* \otimes W_1^*  \\
 V_2^* \otimes W_2^*  \\
\cdots \\
 V_k^* \otimes W_k^*
\end{bmatrix}
- A \right\|_F^2.
\end{align*}

We use matrix $Z_1$ to denote $\begin{bmatrix} \mathrm{vec}(V_1^* \otimes W_1^* ) \\ \mathrm{vec}(V_2^* \otimes W_2^* ) \\ \cdots \\ \mathrm{vec}(V_k^* \otimes W_k^* ) \end{bmatrix} \in \mathbb{R}^{k\times n^2}$ and matrix $U$ to denote $\begin{bmatrix} U_1 & U_2 & \cdots & U_k \end{bmatrix}$. Then we can obtain the following equivalent objective function,
\begin{align*}
\min_{U \in \mathbb{R}^{n\times k} } \| U Z_1  - A_1 \|_F^2.
\end{align*}
Notice that $\min_{U \in \mathbb{R}^{n\times k} } \| U Z_1  - A_1 \|_F^2=\OPT$, since $A_k=U^*Z_1$.

Let $S_1^\top \in\mathbb{R}^{s_1\times n^2}$ be a sketching matrix defined in Definition~\ref{def:fast_gaussian_transform}, where $s_1=O(k/\varepsilon)$. We obtain the following optimization problem,
\begin{align*}
\min_{U \in \mathbb{R}^{n\times k} } \| U Z_1 S_1 - A_1 S_1 \|_F^2.
\end{align*}
Let $ \wh{U} \in \mathbb{R}^{n\times k}$ denote the optimal solution to the above optimization problem. Then $\wh{U} = A_1 S_1 (Z_1 S_1)^\dagger$. By Lemma~\ref{lem:gaussian_count_sketch_for_regression} and Theorem~\ref{thm:multiple_regression_sketch}, we have

\begin{align*}
\| \wh{U} Z_1  - A_1  \|_F^2 \leq (1+\epsilon) \underset{U\in \mathbb{R}^{n\times k}}{\min} \| U Z_1 - A_1 \|_F^2 = (1+\epsilon) \OPT,
\end{align*}

which implies
\begin{align*}
\left\| \sum_{i=1}^k \wh{U}_i \otimes V_i^* \otimes W_i^* - A \right\|_F^2 \leq (1+\epsilon) \OPT.
\end{align*}
To write down $\wh{U}_1, \cdots, \wh{U}_k$, we use the given matrix $A_1$, and we create $s_1 \times k$ variables for matrix $(Z_1 S_1)^\dagger$.

As our second step, we fix $\wh{U} \in \mathbb{R}^{n\times k}$ and $W^* \in \mathbb{R}^{n\times k}$, and we convert tensor $A$ into matrix $A_2$. Let matrix $Z_2$ denote $\begin{bmatrix} \mathrm{vec} ( \wh{U}_1 \otimes W_1^* ) \\ \mathrm{vec} ( \wh{U}_2 \otimes W_2^* ) \\ \cdots \\ \mathrm{vec} ( \wh{U}_k \otimes W_k^* ) \end{bmatrix}$. We consider the following objective function,
\begin{align*}
\min_{V \in \mathbb{R}^{n\times k} } \| V Z_2 -A_2  \|_F^2,
\end{align*}
for which the optimal cost is at most $(1+\epsilon) \OPT$.

Let $S_2^\top \in\mathbb{R}^{s_2\times n^2}$ be a sketching matrix defined in Definition~\ref{def:fast_gaussian_transform}, where $s_2=O(k/\varepsilon)$.
We sketch $S_2$ on the right of the objective function to obtain the new objective function,
\begin{align*}
\underset{V\in \mathbb{R}^{n\times k} }{\min} \| V Z_2 S_2 - A_2 S_2 \|_F^2.
\end{align*}
Let $\wh{V} \in \mathbb{R}^{n\times k}$ denote the optimal solution of the above problem. Then $\wh{V} = A_2 S_2 (Z_2 S_2)^\dagger$. By Lemma~\ref{lem:gaussian_count_sketch_for_regression} and Theorem~\ref{thm:multiple_regression_sketch}, we have,
\begin{align*}
\| \wh{V} Z_2 - A_2 \|_F^2 \leq (1+\epsilon ) \underset{V\in \mathbb{R}^{n\times k} }{\min} \| V Z_2  - A_2 \|_F^2 \leq  (1+\epsilon)^2 \OPT,
\end{align*}
which implies
\begin{align*}
\left\| \sum_{i=1}^k \wh{U}_i \otimes \wh{V}_i \otimes W_i^* - A \right\|_F^2 \leq (1+\epsilon )^2 \OPT.
\end{align*}
To write down $\wh{V}_1, \cdots, \wh{V}_k$, we need to use the given matrix $A_2 \in \mathbb{R}^{n^2 \times n}$, and we need to create $s_2\times k$ variables for matrix $(Z_2 S_2)^\dagger$.

As our third step, we fix the matrices $\wh{U} \in \mathbb{R}^{n\times k}$ and $\wh{V}\in \mathbb{R}^{n \times k}$. We convert tensor $A\in \mathbb{R}^{n\times n \times n}$ into matrix $A_3 \in \mathbb{R}^{n^2 \times n}$. Let matrix $Z_3$ denote $ \begin{bmatrix} \mathrm{vec} ( \wh{U}_1 \otimes \wh{V}_1 ) \\ \mathrm{vec} ( \wh{U}_2 \otimes \wh{V}_2 ) \\ \cdots \\ \mathrm{vec} ( \wh{U}_k \otimes \wh{V}_k ) \end{bmatrix}$. We consider the following objective function,
\begin{align*}
\underset{W\in \mathbb{R}^{n\times k} }{\min} \| W Z_3 - A_3 \|_F^2,
\end{align*}
which has optimal cost at most $(1+\epsilon)^2 \OPT$.

Let $S_3^\top \in\mathbb{R}^{s_3\times n^2}$ be a sketching matrix defined in Definition~\ref{def:fast_gaussian_transform}, where $s_3=O(k/\varepsilon)$.
We sketch $S_3$ on the right of the objective function to obtain a new objective function,
\begin{align*}
\underset{ W \in \mathbb{R}^{n\times k} }{ \min } \| W Z_3 S_3 - A_3 S_3 \|_F^2.
\end{align*}
Let $\wh{W} \in \mathbb{R}^{n\times k}$ denote the optimal solution of the above problem. Then $\wh{W} = A_3 S_3 (Z_3 S_3)^\dagger$. By Lemma~\ref{lem:gaussian_count_sketch_for_regression} and Theorem~\ref{thm:multiple_regression_sketch}, we have,
\begin{align*}
\| \wh{W} Z_3 - A_3 \|_F^2 \leq (1+\epsilon) \underset{W\in \mathbb{R}^{n\times k} }{\min} \| W Z_3 - A_3 \|_F^2 \leq (1+\epsilon)^3 \OPT.
\end{align*}
Thus, we have
\begin{align*}
\min_{X_1,X_2,X_3} \left\| \sum_{i=1}^k (A_1 S_1 X_1)_i \otimes (A_2S_2 X_2)_i \otimes (A_3S_3 X_3)_i - A \right\|_F^2 \leq (1+\epsilon)^3 \OPT.
\end{align*}
Let $V_1=A_1S_1,V_2=A_2S_2,V_3=A_3S_3,$ we then apply Lemma \ref{lem:f_input_sparsity_reduction}, and we obtain $\wh{V}_1,\wh{V}_2,\wh{V}_3,C$. We then apply Theorem~\ref{thm:f_solving_small_problems}. Correctness follows by rescaling $\varepsilon$ by a constant factor.

\paragraph{Running time.} Due to Definition~\ref{def:fast_gaussian_transform}, the running time of line~\ref{sta:f_main_compute_AiSi} is $O(\nnz(A))+n\poly(k)$.  The running time of line~\ref{sta:f_main_input_sparsity_reduction} is shown by Lemma \ref{lem:f_input_sparsity_reduction}, and the running time of line~\ref{sta:f_main_solve_small_problem} is shown by Theorem~\ref{thm:f_solving_small_problems}.
\end{proof}

\begin{theorem}\label{thm:f_main_algorithm_bit}
  Suppose we are given a $3$rd order $n\times n\times n$ tensor $A$ such that each entry can be written using $n^\delta$ bits, where $\delta > 0$ is a given,
  value which can be arbitrarily small (e.g., we could have $n^\delta$ being $O(\log n)$).
  Define $\OPT={\inf}_{\rank-k~A_k} \|  A_k - A \|_F^2$. For any $k\geq 1$, and for any $0<\epsilon<1$, define $n^{\delta'} = O( n^{\delta} 2^{O(k^2/\epsilon)} )$.
$\mathrm{(\RN{1})}$ If $\OPT>0$, and there exists a rank-$k$ $A_k=U^*\otimes V^* \otimes W^*$ tensor, with size $n\times n\times n$, such that $\| A_k-A \|_F^2 = \OPT$, and $\max (\|U^*\|_F,\|V^*\|_F,\|W^*\|_F) \leq 2^{O(n^{\delta'})}$, then there exists an algorithm that takes $(\nnz(A)+ n\poly(k,1/\epsilon)   + 2^{{O}(k^2/\epsilon)} ) n^{{\delta}}$ time in the unit cost $\RAM$ model with word size $O(\log n)$ bits\footnote{The entries of $A$ are assumed to fit in $n^{\delta}$ words.} and outputs three $n \times k$ matrices $U,V,W$ such that
\begin{align}\label{eq:f_thm_1}
\left\|  U \otimes V \otimes W - A  \right\|_F^2 \leq (1+\epsilon) \OPT
\end{align}
holds with probability $9/10$, and each entry of each of $U,V,W$ fits in $n^{\delta'}$ bits.

$\mathrm{(\RN{2})}$ If $\OPT>0$, and $A_k$ does not exist, and there exist three $n\times k$ matrices $U',V',W'$ for which $\max (\|U'\|_F,\|V'\|_F,\|W'\|_F) \leq 2^{O(n^{\delta'})}$ and $\|U' \otimes V'\otimes W' - A  \|_F^2 \leq (1+\epsilon/2) \OPT$, then we can find $U,V,W$ such that \eqref{eq:f_thm_1} holds.

$\mathrm{(\RN{3})}$ If $\OPT=0$ and $A_k$ does exist, and there exists a solution $U^*,V^*,W^*$ such that each entry can be written by $n^{\delta'}$ bits, then we can obtain \eqref{eq:f_thm_1}.

$\mathrm{(\RN{4})}$ If $\OPT=0$, and there exist three $n\times k$ matrices $U,V,W$ such that $\max (\|U\|_F,\|V\|_F,\|W\|_F) $  $\leq 2^{O(n^{\delta'})}$ and
 \begin{align}\label{eq:f_thm_2}
\left\|  U \otimes V \otimes W - A \right\|_F^2 \leq (1+\epsilon) \OPT + 2^{-\Omega(n^{\delta'})} = 2^{-\Omega(n^{\delta'})} ,
\end{align}
then we can output $U,V,W$ such that \eqref{eq:f_thm_2} holds.

Further if $A_k$ exists, we can output a number $Z$ for which $\OPT \leq Z \leq (1+\epsilon) \OPT$.
For all the cases above,
the algorithm runs in the same time as (\RN{1}) and succeeds with probability at least $9/10$.
\end{theorem}
\begin{proof}
This follows by the discussion in Section~\ref{sec:intro}, Theorem~\ref{thm:f_main_algorithm} and Theorem~\ref{thm:f_solving_small_problems} in Section~\ref{sec:f_solving_small_problems}.

Part (\RN{1}) Suppose $\delta>0$ and $A_k = U^*\otimes V^*\otimes W^*$ exists and each of $\|U^*\|_F$, $\|V^*\|_F$, and $\| W^* \|_F$ is bounded by $2^{O(n^{\delta'})}$. We assume the computation model is the unit cost $\RAM$ model with words of size $O(\log n)$ bits, and allow each number of the input tensor $A$ to be written using $n^\delta$ bits. For the case when $\OPT$ is nonzero, using the proof of Theorem~\ref{thm:f_main_algorithm} and Theorems~\ref{thm:f_solving_small_problems}, \ref{thm:minimum_positive}, there exists a lower bound on the cost $\OPT$, which is at least $2^{-O(n^\delta) 2^{O(k^2/\epsilon)}}$. We can round each entry of matrices $U^*,V^*,W^*$ to be an integer expressed using $O(n^{\delta'})$ bits to obtain $U',V',W'$.
Using the triangle inequality and our lower bound on $\OPT$, it follows that $U',V',W'$ provide a $(1+\epsilon)$-approximation.

 Thus, applying Theorem~\ref{thm:f_main_algorithm} by fixing $U',V',W'$ and using Theorem~\ref{thm:f_solving_small_problems} at the end, we can output three matrices $U,V, W$, where each entry can be written using $n^{\delta'}$ bits, so that we satisfy $\| U \otimes V \otimes W - A\|_F^2 \leq (1+\epsilon) \OPT$.

For the running time, since each entry of the input is bounded by $n^{\delta}$ bits, due to Theorem~\ref{thm:f_main_algorithm}, we need $(\nnz(A)+n\poly(k/\varepsilon))\cdot n^{\delta}$ time to reduce the size of the problem to $\poly(k/\varepsilon)$ size (with each number represented using $O(n^\delta)$ bits). According to Theorem~\ref{thm:f_solving_small_problems}, the running time of using a polynomial system verifier to get the solution is $2^{{O}(k^2/\epsilon)} n^{O(\delta')}=2^{{O}(k^2/\epsilon)} n^{O(\delta)}$ time. Thus the total running time is $(\nnz(A)+n\poly(k/\varepsilon))n^{\delta}+2^{{O}(k^2/\epsilon)}\cdot n^{O(\delta)}$.

Part (\RN{2}) is similar to Part (\RN{1}). Part (\RN{3}) is trivial to prove since there exists a solution which can be written down in the bit model, so we obtain a $(1+\epsilon)$-approximation. Part (\RN{4}) is also very similar to Part (\RN{2}).

\end{proof}

\subsection{Input sparsity reduction}\label{sec:f_input_sparsity_reduction}

\begin{algorithm}[h]\caption{Reducing the Size of the Objective Function from $\poly(n)$ to $\poly(k)$}\label{alg:f_input_sparsity_reduction}
\begin{algorithmic}[1]
\Procedure{\textsc{FInputSparsityReduction}}{$A,V_1,V_2,V_3,n,b_1,b_2,b_3,k,\epsilon$} \Comment{Lemma \ref{lem:f_input_sparsity_reduction}}
\State $c_1\leftarrow c_2 \leftarrow c_3 \leftarrow \poly(k,1/\epsilon)$.
\State Choose sparse embedding matrices $T_1\in \mathbb{R}^{c_1 \times n}$, $T_2 \in \mathbb{R}^{c_2 \times n}$, $T_3\in \mathbb{R}^{c_3 \times n}$. \Comment{Definition~\ref{def:count_sketch_transform}}
\State $\wh{V}_i \leftarrow T_i V_i \in \mathbb{R}^{c_i \times b_i}, \forall i\in [3]$.
\State $C\leftarrow A(T_1,T_2,T_3) \in \mathbb{R}^{c_1\times c_2 \times c_3}$.
\State \Return $\wh{V}_1$, $\wh{V}_2$, $\wh{V}_3$ and $C$.
\EndProcedure
\end{algorithmic}
\end{algorithm}

\begin{lemma}\label{lem:f_input_sparsity_reduction}
Let $\poly(k,1/\epsilon) \geq b_1b_2b_3\geq k$. Given a tensor $A\in\mathbb{R}^{n\times n\times n}$ and three matrices $V_1\in \mathbb{R}^{n\times b_1}$, $V_2 \in \mathbb{R}^{n\times b_2}$, and $V_3 \in \mathbb{R}^{n\times b_3}$, there exists an algorithm that takes $O(\nnz(A) + \nnz(V_1) + \nnz(V_2) + \nnz(V_3) )=O(\nnz(A)+n\poly(k/\varepsilon))$ time and outputs a tensor $C\in \mathbb{R}^{c_1\times c_2\times c_3}$ and three matrices $\wh{V}_1\in \mathbb{R}^{c_1\times b_1}$, $\wh{V}_2 \in \mathbb{R}^{c_2\times b_2}$ and $\wh{V}_3 \in \mathbb{R}^{c_3 \times b_3}$ with $c_1=c_2=c_3=\poly(k,1/\epsilon)$, such that with probability at least $0.99$, for all $\alpha>0,X_1,X'_1\in\mathbb{R}^{b_1\times k}, X_2,X'_2\in\mathbb{R}^{b_2\times k}, X_3,X'_3\in\mathbb{R}^{b_3\times k}$ satisfy that,{\small
\begin{align*}
\left\| \sum_{i=1}^k (\wh{V}_1 X_1')_i \otimes (\wh{V}_2 X_2')_i \otimes (\wh{V}_3 X_3')_i - C \right\|_F^2 \leq \alpha \left\| \sum_{i=1}^k (\wh{V}_1 X_1)_i \otimes (\wh{V}_2 X_2)_i \otimes (\wh{V}_3 X_3)_i - C \right\|_F^2,
\end{align*}}
then,{\small
\begin{align*}
\left\| \sum_{i=1}^k (V_1 X_1')_i \otimes (V_2 X_2')_i \otimes ( V_3 X_3')_i - A \right\|_F^2 \leq (1+\epsilon) \alpha \left\| \sum_{i=1}^k ({V}_1 X_1)_i \otimes ({V}_2 X_2)_i \otimes ({V}_3 X_3)_i - A \right\|_F^2.
\end{align*}}

\end{lemma}
\begin{proof}

Let $X_1\in\mathbb{R}^{b_1\times k}, X_2\in\mathbb{R}^{b_2\times k}, X_3\in\mathbb{R}^{b_3\times k}.$
First, we define $Z_1 = ( (V_2 X_2)^\top  \odot (V_3 X_3)^\top )\in \mathbb{R}^{k\times n^2}$. (Note that, for each $i\in [k]$, the $i$-th row of matrix $Z_1$ is $\vect( (V_2 X_2)_i \otimes (V_3 X_3)_i )$.) 
Then, by flattening we have
\begin{align*}
 \left\| \sum_{i=1}^k ({V}_1 X_1)_i \otimes ({V}_2 X_2)_i \otimes ({V}_3 X_3)_i - A \right\|_F^2=\|V_1 X_1 \cdot Z_1 - A_1\|_F^2.
\end{align*}
We choose a sparse embedding matrix (Definition~\ref{def:count_sketch_transform}) $T_1\in \mathbb{R}^{c_1\times n}$ with $c_1=\poly(k,1/\epsilon)$ rows. Since $V_1$ has $b_1\leq \poly(k/\varepsilon)$ columns, according to Lemma~\ref{lem:affine_embedding} with probability $0.999$, for all $X_1 \in \mathbb{R}^{b_1 \times k},Z\in \mathbb{R}^{k\times n^2}$,
\begin{align*}
(1-\epsilon)  \| V_1 X_1 Z - A_1\|_F^2 \leq \| T_1 V_1 X_1 Z - T_1 A_1\|_F^2 \leq (1+\epsilon)  \| V_1 X_1 Z - A_1\|_F^2.
\end{align*}
Therefore, we have
 \begin{align*}
  \|T_1V_1 X_1 \cdot Z_1 - T_1 A_1 \|_F^2=(1\pm\varepsilon)\left\| \sum_{i=1}^k ({V}_1 X_1)_i \otimes ({V}_2 X_2)_i \otimes ({V}_3 X_3)_i - A \right\|_F^2.
 \end{align*}

Second, we unflatten matrix $T_1A_1\in \mathbb{R}^{c_1 \times n^2}$ to obtain a tensor $A'\in \mathbb{R}^{c_1\times n\times n}$. Then we flatten $A'$ along the second direction to obtain $A_2 \in \mathbb{R}^{n\times c_1 n}$. We define $Z_2 = (T_1 V_1 X_1)^\top \odot (V_3 X_3)^\top \in \mathbb{R}^{k\times c_1n}$. Then, by flattening,
\begin{align*}
 \|V_2 X_2 \cdot Z_2 - A_2 \|_F^2 = & ~\|T_1V_1 X_1 \cdot Z_1 - T_1 A_1 \|_F^2 \\
 = & ~ (1\pm\varepsilon)\left\| \sum_{i=1}^k ({V}_1 X_1)_i \otimes ({V}_2 X_2)_i \otimes ({V}_3 X_3)_i - A \right\|_F^2.
\end{align*}
We choose a sparse embedding matrix (Definition~\ref{def:count_sketch_transform}) $T_2\in \mathbb{R}^{c_2 \times n}$ with $c_2 = \poly(k,1/\epsilon)$ rows. Then according to Lemma~\ref{lem:affine_embedding} with probability $0.999$, for all $X_2\in \mathbb{R}^{b_2\times k}$, $Z\in \mathbb{R}^{k\times c_1 n}$,
\begin{align*}
(1-\epsilon) \| V_2 X_2 Z - A_2 \|_F^2 \leq \| T_2 V_2 X_2 Z - T_2 A_2 \|_F^2 \leq (1+\epsilon) \| V_2 X_2 Z - A_2 \|_F^2.
\end{align*}
Therefore, we have
\begin{align*}
\| T_2 V_2 X_2 \cdot Z_2 - T_2 A_2 \|_F^2= & ~ (1\pm\varepsilon)\|V_2 X_2 \cdot Z_2 - A_2 \|_F^2 \\
= & ~ (1\pm \varepsilon)^2\left\| \sum_{i=1}^k ({V}_1 X_1)_i \otimes ({V}_2 X_2)_i \otimes ({V}_3 X_3)_i - A \right\|_F^2.
\end{align*}

Third, we unflatten matrix $T_2 A_2\in \mathbb{R}^{c_2 \times c_1 n}$ to obtain a tensor $A''(=A(T_1,T_2,I))\in \mathbb{R}^{c_1 \times c_2 \times n}$. Then we flatten tensor $A''$ along the last direction (the third direction) to obtain matrix $A_3\in \mathbb{R}^{n \times c_1 c_2}$. We define $Z_3 = (T_1 V_1 X_1)^\top \odot (T_2 V_2 X_2)^\top \in \mathbb{R}^{k \times c_1 c_2}$. Then, by flattening, we have
\begin{align*}
\|V_3 X_3 \cdot Z_3 - A_3 \|_F^2 = & ~ \| T_2 V_2 X_2 \cdot Z_2 - T_2 A_2 \|_F^2 \\
= & ~ (1\pm \varepsilon)^2\left\| \sum_{i=1}^k ({V}_1 X_1)_i \otimes ({V}_2 X_2)_i \otimes ({V}_3 X_3)_i - A \right\|_F^2.
\end{align*}
We choose a sparse embedding matrix (Definition~\ref{def:count_sketch_transform}) $T_3 \in \mathbb{R}^{c_3 \times n}$ with $c_3= \poly(k,1/\epsilon)$ rows. Then according to Lemma~\ref{lem:affine_embedding} with probability $0.999$, for all $X_3\in \mathbb{R}^{b_3 \times k}$, $Z \in \mathbb{R}^{k\times c_1 c_2}$,
\begin{align*}
(1-\epsilon) \| V_3 X_3 Z - A_3 \|_F^2 \leq \| T_3 V_3 X_3 Z - T_3 A_3 \|_F^2 \leq (1+\epsilon) \| V_3 X_3 Z - A_3 \|_F^2.
\end{align*}
Therefore, we have
\begin{align*}
 \| T_3 V_3 X_3 \cdot Z_3 - T_3 A_3 \|_F^2=(1\pm \varepsilon)^3\left\| \sum_{i=1}^k ({V}_1 X_1)_i \otimes ({V}_2 X_2)_i \otimes ({V}_3 X_3)_i - A \right\|_F^2.
\end{align*}
Note that 
 \begin{align*}
\| T_3 V_3 X_3 \cdot Z_3 - T_3 A_3 \|_F^2=\left\| \sum_{i=1}^k (T_1 V_1 X_1)_i \otimes (T_2 V_2 X_2)_i \otimes (T_3 V_3 X_3)_i  - A(T_1,T_2,T_3) \right\|_F^2,
 \end{align*}
 and thus, we have $\forall X_1\in\mathbb{R}^{b_1\times k},X_2\in\mathbb{R}^{b_2\times k},X_3\in\mathbb{R}^{b_3\times k}$
\begin{align*}
& \left\| \sum_{i=1}^k (T_1 V_1 X_1)_i \otimes (T_2 V_2 X_2)_i \otimes (T_3 V_3 X_3)_i  - A(T_1,T_2,T_3) \right\|_F^2\\
=& (1\pm\varepsilon)^3\left\| \sum_{i=1}^k ({V}_1 X_1)_i \otimes ({V}_2 X_2)_i \otimes ({V}_3 X_3)_i - A \right\|_F^2.
\end{align*}
Let $\wh{V}_i$ denote $T_i V_i$, for each $i\in [3]$. Let $C\in \mathbb{R}^{c_1 \times c_2 \times c_3}$ denote $A(T_1,T_2,T_3)$. For $\alpha>1$, if
\begin{align*}
\left\| \sum_{i=1}^k (\wh{V}_1 X_1')_i \otimes (\wh{V}_2 X_2')_i \otimes (\wh{V}_3 X_3')_i - C \right\|_F^2 \leq \alpha \left\| \sum_{i=1}^k (\wh{V}_1 X_1)_i \otimes (\wh{V}_2 X_2)_i \otimes (\wh{V}_3 X_3)_i - C \right\|_F^2,
\end{align*}
then
\begin{align*}
& ~ \left\| \sum_{i=1}^k (V_1 X_1')_i \otimes (V_2 X_2')_i \otimes (V_3 X_3')_i - C \right\|_F^2\\
\leq & ~\frac{1}{(1-\varepsilon)^3}\left\| \sum_{i=1}^k (\wh{V}_1 X_1')_i \otimes (\wh{V}_2 X_2')_i \otimes (\wh{V}_3 X_3')_i - C \right\|_F^2\\
\leq & ~\frac1{(1-\varepsilon)^3}\alpha\left\| \sum_{i=1}^k (\wh{V}_1 X_1)_i \otimes (\wh{V}_2 X_2)_i \otimes (\wh{V}_3 X_3)_i - C \right\|_F^2\\
\leq & ~\frac{(1+\varepsilon)^3}{(1-\varepsilon)^3}\alpha\left\| \sum_{i=1}^k (V_1 X_1)_i \otimes (V_2 X_2)_i \otimes (V_3 X_3)_i - C \right\|_F^2
\end{align*}

By rescaling $\varepsilon$ by a constant, we complete the proof of correctness.

\paragraph{Running time.}
According to Section~\ref{sec:def_count_sketch_gaussian}, for each $i\in [3]$, $T_i V_i$ can be computed in $O(\nnz(V_i))$ time, and $A(T_1,T_2,T_3)$ can be computed in $O(\nnz(A))$ time.

By the analysis above, the proof is complete.
\end{proof}

\subsection{Tensor multiple regression}\label{sec:f_tensor_multiple_regression}

\begin{algorithm}[h]\caption{Frobenius Norm Tensor Multiple Regression}\label{alg:f_tensor_multiple_regression}
\begin{algorithmic}[1]
\Procedure{\textsc{FTensorMultipleRegression}}{$A,U,V,d,n,k$} \Comment{Theorem~\ref{thm:f_tensor_multiple_regression}}
\State $s \leftarrow O(k^2 + k/\epsilon)$.
\State Choose $S \in \mathbb{R}^{n^2 \times s}$ to be a \textsc{TensorSketch}. \Comment{Definition~\ref{def:tensor_sketch}}
\State Compute $A \cdot S$.
\State Compute $B \cdot S$. \Comment{$B= U^\top \odot V^\top$}
\State $W\leftarrow (AS) (BS)^\dagger$
\State \Return $W$.
\EndProcedure
\end{algorithmic}
\end{algorithm}

\begin{theorem}\label{thm:f_tensor_multiple_regression}
Given matrices $A\in \mathbb{R}^{d\times n^2}$, $U,V\in \mathbb{R}^{n\times k}$, let $B\in \mathbb{R}^{k\times n^2}$ denote $U^\top \odot V^\top$. There exists an algorithm that takes $O( \nnz(A) + \nnz(U) + \nnz(V) + d \poly (k,1/\epsilon) )$ time and outputs a matrix $W'\in \mathbb{R}^{d\times k}$ such that,
\begin{align*}
\| W' B - A \|_F^2 \leq (1+\epsilon) \min_{W\in \mathbb{R}^{d\times k} }\| WB - A \|_F^2.
\end{align*}
\end{theorem}
\begin{proof}
We choose a \textsc{TensorSketch} (Definition~\ref{def:tensor_sketch}) $S\in \mathbb{R}^{n^2 \times s}$ to reduce the problem to a smaller problem,
\begin{align*}
\min_{W\in \mathbb{R}^{d\times k}} \| W B S - A S\|_F^2.
\end{align*}
Let $W'$ denote the optimal solution to the above problem. Following a similar proof to that in Section~\ref{sec:f_leverage_score_multiple_regression}, if $S$ is a $(1\pm 1/2)$-subspace embedding and satisfies $\sqrt{\epsilon/k}$-approximate matrix product, then $W'$ provides a $(1+\epsilon)$-approximation to the original problem. By Theorem~\ref{thm:theorem1_anw14}, we have $s=O(k^2+k/\epsilon)$.

\paragraph{Running time.} According to Definition~\ref{def:tensor_sketch}, $BS$ can be computed in $O(\nnz(U)+\nnz(V))+\poly(k/\varepsilon)$ time. Notice that each row of $S$ has exactly $1$ nonzero entry, thus $AS$ can be computed in $O(\nnz(A))$ time. Since $BS\in\mathbb{R}^{k\times s}$ and $AS\in\mathbb{R}^{d\times s}$, $\min_{W\in \mathbb{R}^{d\times k}} \| W B S - A S\|_F^2$ can be solved in $d\poly(sk)=d\poly(k/\varepsilon)$ time.
\end{proof}

\subsection{Bicriteria algorithms}\label{sec:f_bicriteria_algorithm}

\subsubsection{Solving a small regression problem}

\begin{lemma}\label{lem:tensor_regression}
Given tensor $A\in \mathbb{R}^{n\times n \times n}$ and three matrices $U\in \mathbb{R}^{n\times s_1}, V\in \mathbb{R}^{n\times s_2}$ and $W\in \mathbb{R}^{n\times s_3}$, there exists an algorithm that takes $O(\nnz(A) + n \poly(s_1,s_2,s_3,1/\varepsilon))$ time and outputs $\alpha'\in \mathbb{R}^{s_1 \times s_2 \times s_3}$ such that
\begin{align*}
\left\| \sum_{i=1}^{s_1} \sum_{j=1}^{s_2} \sum_{l=1}^{s_3} \alpha'_{i,j,l} \cdot U_i \otimes V_j \otimes W_l - A \right\|_F^2 \leq (1+\epsilon) \min_{\alpha \in \mathbb{R}^{s_1 \times s_2 \times s_3} } \left\| \sum_{i=1}^{s_1} \sum_{j=1}^{s_2} \sum_{l=1}^{s_3} \alpha_{i,j,l} \cdot U_i \otimes V_j \otimes W_l - A \right\|_F^2.
\end{align*}
holds with probability at least $.99$.
\end{lemma}
\begin{proof}
We define $\wt{b}\in \mathbb{R}^{n^3}$ to be the vector where the $i+(j-1)n+(l-1)n^2$-th entry of $\wt{b}$ is $A_{i,j,l}$. We define $\wt{A} \in \mathbb{R}^{n^3 \times s_1 s_2 s_3}$ to be the matrix where the $(i+(j-1)n+(l-1)n^2,i'+(j'-1)s_2+(l'-1)s_2s_3)$ entry is $U_{i',i}\cdot V_{j',j} \cdot W_{l',l}$.
This problem is equivalent to a linear regression problem,
\begin{align*}
\min_{x \in \mathbb{R}^{s_1 s_2 s_3}} \| \wt{A} x - \wt{b} \|_2^2,
\end{align*}
where $\wt{A} \in \mathbb{R}^{n^3 \times s_1 s_2 s_3}, \wt{b} \in \mathbb{R}^{n^3}$. Thus, it can be solved fairly quickly using recent work \cite{cw13,mm13,nn13}. However, the running time of this na\"ively is $\Omega(n^3)$, since we have to write down each entry of $\wt{A}$. In the next few paragraphs, we show how to improve the running time to $\nnz(A) + n\poly(s_1,s_2,s_3)$.

Since $\alpha\in\mathbb{R}^{s_1\times s_2\times s_3},$ $\alpha$ can be always written as
$\alpha=X_1\otimes X_2\otimes X_3,$
where $X_1\in\mathbb{R}^{s_1\times s_1s_2s_3},X_2\in\mathbb{R}^{s_2\times s_1s_2s_3},X_3\in\mathbb{R}^{s_3\times s_1s_2s_3}$, we have
\begin{align*}
\min_{\alpha \in \mathbb{R}^{s_1 \times s_2 \times s_3} } \left\| \sum_{i=1}^{s_1} \sum_{j=1}^{s_2} \sum_{l=1}^{s_3} \alpha_{i,j,l} \cdot U_i \otimes V_j \otimes W_l - A \right\|_F^2=\underset{X_3\in\mathbb{R}^{s_3\times s_1s_2s_3}}{\underset{X_2\in\mathbb{R}^{s_2\times s_1s_2s_3}}{\underset{X_1\in\mathbb{R}^{s_1\times s_1s_2s_3}}{\min}}}\left\|(UX_1)\otimes(VX_2)\otimes (WX_3)-A\right\|_F^2.
\end{align*}
By Lemma~\ref{lem:f_input_sparsity_reduction}, we can reduce the problem size $n\times n\times n$ to a smaller problem that has size $t_1\times t_2 \times t_3$,
 \begin{align*}
& ~ \min_{X_1,X_2,X_3} \left\| \sum_{i=1}^{s_1s_2s_3}  ( T_1 U X_1)_i \otimes ( T_2 V X_2)_i \otimes (T_3 WX_3)_i - A(T_1,T_2,T_3) \right\|_F^2 \\
\end{align*}
where $T_1\in \mathbb{R}^{t_1 \times n}$, $T_2\in \mathbb{R}^{t_2 \times n}$, $T_3\in \mathbb{R}^{t_3 \times n},t_1=t_2=t_3=\poly(s_1s_2s_3/\varepsilon)$. Notice that
 \begin{align*}
& ~ \min_{X_1,X_2,X_3} \left\| \sum_{i=1}^{s_1s_2s_3}  ( T_1 U X_1)_i \otimes ( T_2 V X_2)_i \otimes (T_3 WX_3)_i - A(T_1,T_2,T_3) \right\|_F^2 \\
=& ~ \min_{\alpha \in \mathbb{R}^{s_1 \times s_2 \times s_3} } \left\| \sum_{i=1}^{s_1} \sum_{j=1}^{s_2} \sum_{l=1}^{s_3} \alpha_{i,j,l} \cdot (T_1U)_i \otimes (T_2V)_j \otimes (T_3W)_l - A(T_1,T_2,T_3) \right\|_F^2.
\end{align*}
Let
\begin{align*}
\alpha'=\underset{\alpha \in \mathbb{R}^{s_1 \times s_2 \times s_3}}{\arg\min}\left\| \sum_{i=1}^{s_1} \sum_{j=1}^{s_2} \sum_{l=1}^{s_3} \alpha_{i,j,l} \cdot (T_1U)_i \otimes (T_2V)_j \otimes (T_3W)_l - A(T_1,T_2,T_3) \right\|_F^2,
\end{align*}
then we have
\begin{align*}
\left\| \sum_{i=1}^{s_1} \sum_{j=1}^{s_2} \sum_{l=1}^{s_3} \alpha'_{i,j,l} \cdot U_i \otimes V_j \otimes W_l - A \right\|_F^2 \leq (1+\epsilon) \min_{\alpha \in \mathbb{R}^{s_1 \times s_2 \times s_3} } \left\| \sum_{i=1}^{s_1} \sum_{j=1}^{s_2} \sum_{l=1}^{s_3} \alpha_{i,j,l} \cdot U_i \otimes V_j \otimes W_l - A \right\|_F^2.
\end{align*}
Again, according to Lemma~\ref{lem:f_input_sparsity_reduction}, the total running time is then $O(\nnz(A) + n \poly(s_1,s_2,s_3,1/\varepsilon))$.

\end{proof}

\begin{lemma}\label{lem:f_input_sparsity_for_regression}
Given tensor $A\in \mathbb{R}^{n\times n\times n}$, and two matrices $U\in \mathbb{R}^{n\times s}, V\in \mathbb{R}^{n\times s}$ with $\rank(U)=r_1,\rank(V)=r_2$, let $T_1\in\mathbb{R}^{t_1\times n},T_2\in\mathbb{R}^{t_2\times n}$ be two sparse embedding matrices (Definition~\ref{def:count_sketch_transform}) with $t_1=\poly(r_1/\varepsilon),t_2=\poly(r_2/\varepsilon)$. Then with probability at least $0.99$, $\forall X\in\mathbb{R}^{n\times s}$,
\begin{align*}
(1-\eps)\|U\otimes V\otimes X-A\|_F^2\leq\|T_1U\otimes T_2V\otimes X-A(T_1,T_2,I)\|_F^2\leq (1+\eps)\|U\otimes V\otimes X-A\|_F^2.
\end{align*}
\end{lemma}
\begin{proof}
Let $X\in\mathbb{R}^{n\times s}.$ We define $Z_1 = ( V^\top  \odot X^\top )\in \mathbb{R}^{s\times n^2}$. We choose a sparse embedding matrix (Definition~\ref{def:count_sketch_transform}) $T_1\in \mathbb{R}^{t_1\times n}$ with $t_1=\poly(r_1/\epsilon)$ rows. According to Lemma~\ref{lem:affine_embedding} with probability $0.999$, for all $Z\in \mathbb{R}^{s\times n^2}$,
\begin{align*}
(1-\epsilon)  \| U Z - A_1\|_F^2 \leq \| T_1 UZ - T_1 A_1\|_F^2 \leq (1+\epsilon)  \| T_1U Z - A_1\|_F^2.
\end{align*}
It means that
\begin{align*}
(1-\epsilon)  \| U Z_1 - A_1\|_F^2 \leq \| T_1 UZ_1 - T_1 A_1\|_F^2 \leq (1+\epsilon)  \| T_1U Z_1 - A_1\|_F^2.
\end{align*}
Second, we unflatten matrix $T_1A_1\in \mathbb{R}^{t_1 \times n^2}$ to obtain a tensor $A'\in \mathbb{R}^{t_1\times n\times n}$. Then we flatten $A'$ along the second direction to obtain $A'_2 \in \mathbb{R}^{n\times t_1 n}$. We define $Z_2 = ((T_1 U)^\top \odot X^\top) \in \mathbb{R}^{s\times t_1n}$. Then, by flattening,
\begin{align*}
 \|V \cdot Z_2 - A'_2 \|_F^2=\|T_1U \cdot Z_1 - T_1 A_1 \|_F^2=(1\pm\varepsilon)\|U\otimes V\otimes X-A\|_F^2.
\end{align*}
We choose a sparse embedding matrix (Definition~\ref{def:count_sketch_transform}) $T_2\in \mathbb{R}^{t_2 \times n}$ with $t_2 = \poly(r_2/\epsilon)$ rows. Then according to Lemma~\ref{lem:affine_embedding} with probability $0.999$, for all $Z\in \mathbb{R}^{s\times t_1n}$,
\begin{align*}
(1-\epsilon) \| VZ - A'_2 \|_F^2 \leq \| T_2 VZ - T_2 A'_2 \|_F^2 \leq (1+\epsilon) \| V Z - A'_2 \|_F^2.
\end{align*}
Thus,
\begin{align*}
 \|T_2V \cdot Z_2 - T_2A'_2 \|_F^2=(1\pm\varepsilon)^2\|U\otimes V\otimes X-A\|_F^2.
\end{align*}
After rescaling $\varepsilon$ by a constant, with probability at least $0.99$, $\forall X\in\mathbb{R}^{n\times s}$,
\begin{align*}
(1-\eps)\|U\otimes V\otimes X-A\|_F^2\leq\|T_1U\otimes T_2V\otimes X-A(T_1,T_2,I)\|_F^2\leq (1+\eps)\|U\otimes V\otimes X-A\|_F^2.
\end{align*}
\end{proof}

\subsubsection{Algorithm \RN{1}}
We start with a slightly unoptimized bicriteria low rank approximation algorithm.

\begin{algorithm}[h]\caption{Frobenius Norm Bicriteria Low Rank Approximation Algorithm, rank-$O(k^3/\epsilon^3)$}
\begin{algorithmic}[1]
\Procedure{\textsc{FTensorLowRankBicriteriaCubicRank}}{$A,n,k$} \Comment{Theorem~\ref{thm:f_bicriteria}}
\State $s_1 \leftarrow s_2 \leftarrow s_3 \leftarrow O(k/\epsilon)$.
\State $t_1 \leftarrow t_2 \leftarrow t_3 \leftarrow \poly(k/\epsilon)$.
\State Choose $S_i \in \mathbb{R}^{n^2 \times s_i}$ to be a Sketching matrix, $\forall i\in [3]$. \Comment{Definition~\ref{def:fast_gaussian_transform}}
\State Choose $T_i \in \mathbb{R}^{t_i \times n}$ to be a Sketching matrix, $\forall i\in [3]$. \Comment{Definition~\ref{def:count_sketch_transform}}
\State Compute $U \leftarrow T_1 \cdot (A_1 \cdot S_1)$, $ V\leftarrow T_2 \cdot (A_2 \cdot S_2)$, $W\leftarrow T_3 \cdot (A_3 \cdot S_3)$.
\State Compute $C \leftarrow A(T_1,T_2,T_3)$.
\State $X\leftarrow$\textsc{FTensorRegression}($C,U,V,W,t_1,s_1,t_2,s_2,t_3,s_3$). \Comment{Linear regression}
\State \Return $X(A_1 S_1,A_2 S_2, A_3S_3)$.
\EndProcedure
\end{algorithmic}
\end{algorithm}

\begin{theorem}\label{thm:f_bicriteria}
Given a 3rd order tensor $A\in \mathbb{R}^{n\times n \times n}$, for any $k\geq 1, \epsilon \in (0,1)$, let $r=O(k^3/\epsilon^3)$. There exists an algorithm that takes $O(\nnz(A) + n\poly(k,1/\epsilon))$ time and outputs three matrices $U \in \mathbb{R}^{n\times r}, V\in \mathbb{R}^{n\times r}$, $W\in \mathbb{R}^{n\times r}$ such that
\begin{align*}
\left\| \sum_{i=1}^r U_i \otimes V_i \otimes W_i - A \right\|_F^2 \leq (1+\epsilon) \underset{\rank-k~A_k}{\min} \| A_k - A \|_F^2
\end{align*}
holds with probability $9/10$.
\end{theorem}
\begin{proof}
At the end of Theorem \ref{thm:f_main_algorithm}, we need to run a polynomial system verifier. This is why we obtain exponential in $k$ running time. Instead of running the polynomial system verifier, we can use Lemma \ref{lem:tensor_regression}. This reduces the running time to be polynomial in all parameters: $n,k,1/\epsilon$. However, the output tensor has rank $(k/\epsilon)^3$ (Here we mean that we do not obtain a better decomposition than $(k/\epsilon)^3$ components). According to Section~\ref{sec:def_count_sketch_gaussian}, for each $i$, $A_iS_i$ can be computed in $O(\nnz(A))+n\poly(k/\varepsilon)$ time. Then $T_i(A_iS_i)$ can be computed in $n\poly(k,1/\epsilon)$ time and $A(T_1,T_2,T_3)$ also can be computed in $O(\nnz(A))$ time. The running time for the regression is $\poly(k/\varepsilon)$.
\end{proof}

Now we present an optimized bicriteria algorithm.
\begin{algorithm}[h]\caption{Frobenius Norm Low Rank Approximation Algorithm, rank-$O(k^2/\epsilon^2)$}
\begin{algorithmic}[1]
\Procedure{\textsc{FTensorLowRankBicriteriaQuadraticRank}}{$A,n,k$} \Comment{Theorem~\ref{thm:f_bicriteria_better}}
\State $s_1 \leftarrow s_2 \leftarrow O(k/\epsilon)$.
\State Choose $S_i \in \mathbb{R}^{n^2 \times s_i}$ to be a sketching matrix, $\forall i\in [3]$. \Comment{Definition~\ref{def:fast_gaussian_transform}}
\State Compute $A_1 \cdot S_1$, $A_2 \cdot S_2$.
\State Form $\wh{U}$ by using $A_1 S_1$ according to Equation~\eqref{eq:f_bicriteria_rank_k2eps2_U}.
\State Form $\wh{V}$ by using $A_2 S_2$ according to Equation~\eqref{eq:f_bicriteria_rank_k2eps2_V}.
\State $\wh{W}\leftarrow$\textsc{FTensorMultipleRegression}($A,\wh{U},\wh{V},n,n,s_1s_2$). \Comment{Algorithm~\ref{alg:f_tensor_multiple_regression}}
\State \Return $\wh{U},\wh{V},\wh{W}$.
\EndProcedure
\Procedure{\textsc{FTensorLowRankBicriteriaQuadraticRank}}{$A,n,k$} \Comment{Theorem~\ref{thm:f_bicriteria_better}}
\State $s_1 \leftarrow s_2 \leftarrow O(k/\epsilon)$.
\State $t_1 \leftarrow t_2 \leftarrow \poly(k/\varepsilon)$.
\State Choose $S_i\in \mathbb{R}^{n^2 \times s_i}$ to be a Sketching matrix, $\forall i\in [2]$. \Comment{Definition~\ref{def:fast_gaussian_transform}}
\State Choose $T_i\in \mathbb{R}^{t_i \times n}$ to be a Sketching matrix, $\forall i\in [2]$. \Comment{Definition~\ref{def:count_sketch_transform}}
\State Form $\wh{U}$ by using $A_1 S_1$ according to Equation~\eqref{eq:f_bicriteria_rank_k2eps2_U}.
\State Form $\wh{V}$ by using $A_2 S_2$ according to Equation~\eqref{eq:f_bicriteria_rank_k2eps2_V}.
\State Compute $C\leftarrow A(T_1,T_2,I)$. \Comment{$C\in \mathbb{R}^{t_1 \times t_2 \times n}$}
\State Compute $B\leftarrow ( T_1\wh{U})^\top \odot (T_2 \wh{V})^\top$.
\State $\wh{W}\leftarrow\underset{X\in\mathbb{R}^{n\times s_1s_2}}{\arg\min}\|XB-C_3\|_F^2.$
\State \Return $\wh{U},\wh{V},\wh{W}$.
\EndProcedure
\end{algorithmic}
\end{algorithm}

\begin{theorem}\label{thm:f_bicriteria_better}
Given a 3rd order tensor $A\in \mathbb{R}^{n\times n \times n}$, for any $k\geq 1, \epsilon \in (0,1)$, let $r=O(k^2/\epsilon^2)$. There exists an algorithm that takes $O(\nnz(A)+ n\poly(k,1/\epsilon))$ time and outputs three matrices $U \in \mathbb{R}^{n\times r}, V\in \mathbb{R}^{n\times r}$, $W\in \mathbb{R}^{n\times r}$ such that
\begin{align*}
\left\| \sum_{i=1}^r U_i \otimes V_i \otimes W_i - A \right\|_F^2 \leq (1+\epsilon) \underset{\rank-k~A_k}{\min} \| A_k - A \|_F^2
\end{align*}
holds with probability $9/10$.
\end{theorem}
Note that there are two different ways to implement algorithm \textsc{FTensorLowRankBicriteriaQuadraticRank}. We present the proofs for both of them here.

Approach \RN{1}.
\begin{proof}
Let $\OPT=\underset{\rank-k\ A_k}{\min}\|A_k-A\|_F^2.$
According to Theorem~\ref{thm:f_main_algorithm}, we know that there exists a sketching matrix $S_3\in\mathbb{R}^{n^2\times s_3}$ where $s_3=O(k/\varepsilon)$, such that
\begin{align*}
\min_{X_1\in\mathbb{R}^{s_1\times k},X_2\in\mathbb{R}^{s_2\times k},X_3\in\mathbb{R}^{s_3\times k}}\left\| \sum_{l=1}^k (A_1 S_1 X_1)_l \otimes (A_2S_2 X_2)_l \otimes (A_3S_3 X_3)_l - A \right\|_F^2 \leq (1+\epsilon) \OPT
\end{align*}
Now we fix an $l$ and we have:
\begin{align*}
&(A_1 S_1 X_1)_l \otimes (A_2S_2 X_2)_l \otimes (A_3S_3 X_3)_l\\
=&\left(\sum_{i=1}^{s_1} (A_1S_1)_i (X_1)_{i,l}\right)\otimes \left(\sum_{j=1}^{s_2} (A_2S_2)_j (X_2)_{j,l}\right)\otimes (A_3S_3 X_3)_l\\
=&\sum_{i=1}^{s_1}\sum_{j=1}^{s_2} (A_1S_1)_i \otimes (A_2S_2)_j \otimes (A_3S_3 X_3)_l(X_1)_{i,l}(X_2)_{j,l}
\end{align*}
Thus, we have
\begin{align}\label{eq:A1S1otimesA2S2}
\min_{X_1,X_2,X_3}\left\| \sum_{i=1}^{s_1}\sum_{j=1}^{s_2} (A_1S_1)_i \otimes (A_2S_2)_j \otimes \left( \sum_{l=1}^k(A_3S_3 X_3)_l(X_1)_{i,l}(X_2)_{j,l}\right) - A \right\|_F^2 \leq (1+\epsilon) \OPT.
\end{align}

We use matrices $A_1 S_1 \in \mathbb{R}^{n\times s_1}$ and $A_2 S_2 \in \mathbb{R}^{n\times s_2}$ to construct a matrix $B \in \mathbb{R}^{s_1s_2 \times n^2}$ in the following way: each row of $B$ is the vector corresponding to the matrix generated by the $\otimes$ product between one column vector in $A_1 S_1$ and the other column vector in $A_2 S_2$, i.e.,
\begin{align}\label{eq:f_bicriteria_rank_k2eps2_form_B}
B^{i+ (j-1)s_1} = \mathrm{vec}( (A_1 S_1)_i \otimes (A_2 S_2)_j  ), \forall i\in [s_1], j\in [s_2] ,
\end{align}
where $(A_1 S_1)_i$ denotes the $i$-th column of $A_1S_1$ and $(A_2 S_2)_j$ denote the $j$-th column of $A_2 S_2$.

We create matrix $\wh{U} \in \mathbb{R}^{n\times s_1 s_2}$ by copying matrix $A_1 S_1$ $s_2$ times, i.e.,
\begin{align}\label{eq:f_bicriteria_rank_k2eps2_U}
\wh{U} = \begin{bmatrix}
A_1 S_1 & A_1 S_1 & \cdots & A_1 S_1
\end{bmatrix}.
\end{align}
We create matrix $\wh{V} \in \mathbb{R}^{n\times s_1 s_2}$ by copying the $i$-th column of $A_2 S_2$ a total of $s_1$ times, into columns $(i-1)s_1, \cdots, i s_1$ of $\wh{V}$, for each $i\in [s_2]$, i.e.,
\begin{align}\label{eq:f_bicriteria_rank_k2eps2_V}
\wh{V} = \begin{bmatrix}
(A_2 S_2)_1 & \cdots &(A_2 S_2)_1 & (A_2 S_2)_2 & \cdots &(A_2 S_2)_2 & \cdots & (A_2 S_2)_{s_2} & \cdots &(A_2 S_2)_{s_2}
\end{bmatrix}.
\end{align}
Thus, we can use $\wh{U}$ and $\wh{V}$ to represent $B$,
\begin{align*}
B = ( \wh{U}^\top \odot \wh{V}^\top ) \in \mathbb{R}^{s_1 s_2 \times n^2}.
\end{align*}

According to Equation~\eqref{eq:A1S1otimesA2S2}, we have:
\begin{align*}
\underset{W\in \mathbb{R}^{n\times s_1 s_2} }{\min} \| WB - A_3 \|_F^2\leq (1+\varepsilon)\OPT.
\end{align*}

Next, we want to find matrix $W\in \mathbb{R}^{n\times s_1 s_2}$ by solving the following optimization problem,
\begin{align*}
\underset{W\in \mathbb{R}^{n\times s_1 s_2} }{\min} \| WB - A_3 \|_F^2.
\end{align*}
Note that $B$ has size $s_1s_2\times n^2$. Na\"{i}vely writing down $B$ already requires $\Omega(n^2)$ time. In order to achieve nearly linear time in $n$, we cannot write down $B$.
We choose $S_3 \in \mathbb{R}^{n_1 n_2 \times s_3}$ to be a \textsc{TensorSketch} (Definition~\ref{def:tensor_sketch}). In order to solve multiple regression, we need to set $s_3 = O( (s_1 s_2)^2+ (s_1s_2)/\epsilon)$. Let $\wh{W}$ denote the optimal solution to $\| WB S_3 - A_3 S_3 \|_F^2$. Then $\wh{W} = (A_3 S_3) (BS_3)^\dagger$. Since each row of $S_3$ has exactly $1$ nonzero entry, $A_3S_3$ can be computed in $O(\nnz(A))$ time. Since $B = ( \wh{U}^\top \odot \wh{V}^\top )$, according to Definition~\ref{def:tensor_sketch}, $BS_3$ can be computed in $n \poly(s_1s_2/\epsilon)=n\poly(k/\varepsilon)$ time. By Theorem~\ref{thm:f_tensor_multiple_regression}, we have
\begin{align*}
\| \wh{W} B - A_3 \|_F^2 \leq (1+\epsilon) \min_{W\in \mathbb{R}^{n\times s_1 s_2} } \| W B - A_3 \|_F^2.
\end{align*}
Thus, we have
\begin{align*}
\|\wh{U}\otimes\wh{V}\otimes\wh{W}-A\|_F^2\leq (1+\varepsilon)\OPT.
\end{align*}
According to Definition~\ref{def:fast_gaussian_transform}, $A_1S_1,A_2S_2$ can be computed in $O(\nnz(A)+\poly(k/\varepsilon))$ time. Te total running time is thus $O(\nnz(A)+\poly(k/\varepsilon))$.
\end{proof}

Approach \RN{2}.
\begin{proof}
Let $\OPT=\underset{\rank-k\ A_k}{\min}\|A_k-A\|_F^2.$
Choose sketching matrices (Definition~\ref{def:fast_gaussian_transform}) $S_1\in \mathbb{R}^{n^2\times s_1}$, $S_2\in \mathbb{R}^{n^2\times s_2}$, $S_3 \in \mathbb{R}^{n^2\times s_3}$, and sketching matrices (Definition~\ref{def:count_sketch_transform}) $T_1\in \mathbb{R}^{t_1 \times n}$ and $T_2 \in \mathbb{R}^{t_2 \times n}$ with $s_1 =s_2 = s_3=O(k/\epsilon),t_1 =t_2=\poly(k/\varepsilon)$.
We create matrix $\wh{U} \in \mathbb{R}^{n\times s_1 s_2}$ by copying matrix $A_1 S_1$ $s_2$ times, i.e.,
\begin{align*}
\wh{U} = \begin{bmatrix}
A_1 S_1 & A_1 S_1 & \cdots & A_1 S_1
\end{bmatrix}.
\end{align*}
We create matrix $\wh{V} \in \mathbb{R}^{n\times s_1 s_2}$ by copying the $i$-th column of $A_2 S_2$ a total of $s_1$ times, into columns $(i-1)s_1, \cdots, i s_1$ of $\wh{V}$, for each $i\in [s_2]$, i.e.,
\begin{align*}
\wh{V} = \begin{bmatrix}
(A_2 S_2)_1 & \cdots &(A_2 S_2)_1 & (A_2 S_2)_2 & \cdots &(A_2 S_2)_2 & \cdots & (A_2 S_2)_{s_2} & \cdots &(A_2 S_2)_{s_2}
\end{bmatrix}.
\end{align*}

As we proved in Approach \RN{1}, we have
\begin{align*}
\min_{X\in\mathbb{R}^{n\times s_1s_2}}\|\wh{U}\otimes\wh{V}\otimes X-A\|_F^2\leq (1+\varepsilon)\OPT.
\end{align*}
Let $B=( (T_1\wh{U})^\top \odot (T_2 \wh{V})^\top)\in\mathbb{R}^{s_1s_2\times t_1t_2},$ and flatten $A(T_1,T_2,I)$ along the third direction to obtain $C_3\in\mathbb{R}^{n\times t_1t_2}$.
Let
\begin{align*}
\wh{W}=\underset{X\in\mathbb{R}^{n\times s_1s_2}}{\arg\min}\|T_1\wh{U}\otimes T_2\wh{V}\otimes X-A(T_1,T_2,I)\|_F^2=\underset{X\in\mathbb{R}^{n\times s_1s_2}}{\arg\min}\|XB-C_3\|_F^2.
\end{align*}
Let
\begin{align*}
W^*=\underset{X\in\mathbb{R}^{n\times s_1s_2}}{\arg\min}\|\wh{U}\otimes\wh{V}\otimes X-A\|_F^2.
\end{align*}
According to Lemma~\ref{lem:f_input_sparsity_for_regression},
\begin{align*}
&\|\wh{U}\otimes\wh{V}\otimes \wh{W}-A\|_F^2\\
\leq &\frac{1}{1-\varepsilon}\|T_1\wh{U}\otimes T_2\wh{V}\otimes \wh{W}-A(T_1,T_2,I)\|_F^2\\
\leq &\frac{1}{1-\varepsilon}\|T_1\wh{U}\otimes T_2\wh{V}\otimes W^*-A(T_1,T_2,I)\|_F^2\\
\leq &\frac{1+\varepsilon}{1-\varepsilon}\|\wh{U}\otimes \wh{V}\otimes W^*-A\|_F^2\\
\leq &\frac{(1+\varepsilon)^2}{1-\varepsilon}\OPT.
\end{align*}

According to Definition~\ref{def:fast_gaussian_transform}, $A_1S_1,A_2S_2$ can be computed in $O(\nnz(A)+\poly(k/\varepsilon))$ time. The total running time is thus $O(\nnz(A)+\poly(k/\varepsilon))$. Since $T_1,T_2$ are sparse embedding matrices, $T_1\wh{U},T_2\wh{V}$ can be computed in $O(\nnz(A)+\poly(k/\varepsilon))$ time. The total running time is in $O(\nnz(A)+\poly(k/\varepsilon))$.
\end{proof}

\begin{theorem}\label{thm:f_bicriteria_algorithm_bit}
    Given a $3$rd order tensor $A \in \mathbb{R}^{n\times n\times n}$, for any $k\geq 1$ and any $0 < \epsilon <1$, if $A_k$ exists then
    there is a randomized algorithm running in
    $\nnz(A) + n \cdot \poly(k/\epsilon)$ time
    which outputs a rank-$O(k^2/\epsilon^2)$ tensor $B$ for which $\|A-B\|_F^2 \leq (1+\epsilon)\|A-A_k\|_F^2$.
    If $A_k$ does not exist, then the algorithm outputs a rank-$O(k^2/\epsilon^2)$
    tensor $B$ for which $\|A-B\|_F^2 \leq (1+\epsilon)\OPT + \gamma$, where $\gamma$ is an arbitrarily
    small positive function of $n$. In both cases, the algorithm succeeds with probability at least $9/10$.
\end{theorem}
\begin{proof}
If $A_k$ exists, then the proof directly follows the proof of Theorem~\ref{thm:f_main_algorithm} and Theorem~\ref{thm:f_bicriteria_better}. If $A_k$ does not exist, then for any $\gamma>0,$ there exist $U^*\in\mathbb{R}^{n\times k},V^*\in\mathbb{R}^{n\times k},W^*\in\mathbb{R}^{n\times k}$ such that
\begin{align*}
\|U^*\otimes V^*\otimes W^*-A\|_F^2\leq \inf_{\rank-k~A'}\|A-A'\|_F^2+\frac{1}{10}\gamma.
\end{align*}
Then we just regard $U^*\otimes V^*\otimes W^*$ as the ``best'' $\rank~k$ approximation to $A$, and follow the same argument as in the proof of Theorem~\ref{thm:f_main_algorithm} and the proof of Theorem~\ref{thm:f_bicriteria_better}. We can finally output a tensor $B\in\mathbb{R}^{n\times n\times n}$ with rank-$O(k^2/\epsilon^2)$ such that
\begin{align*}
\|B-A\|_F^2 &\leq~ (1+\varepsilon)\|U^*\otimes V^*\otimes W^*-A\|_F^2\\
&\leq~ (1+\varepsilon)\left( \inf_{\rank-k~A'}\|A-A'\|_F^2+\frac{1}{10}\gamma\right)\\
&\leq~(1+\varepsilon) \inf_{\rank-k~A'}\|A-A'\|_F^2+\gamma
\end{align*}
where the first inequality follows by the proof of Theorem~\ref{thm:f_main_algorithm} and the proof of theorem~\ref{thm:f_bicriteria_better}. The second inequality follows by our choice of $U^*,V^*,W^*$. The third inequality follows since $1+\varepsilon<2$ and $\gamma>0$.
\end{proof}




\subsubsection{$\poly(k)$-approximation to multiple regression}

\begin{lemma}[(1.4) and (1.9) in \cite{rv09}]\label{lem:rv09_tool}
 Let $s\geq k$. Let $U\in \mathbb{R}^{n \times k}$ denote a matrix that has orthonormal columns, and $S \in \mathbb{R}^{s\times n} $ denote an i.i.d. $N(0,1/s)$ Gaussian matrix. Then $SU$ is also an $s \times k$ i.i.d. Gaussian matrix with each entry draw from $N(0,1/s)$, and furthermore, we have with arbitrarily large constant probability,
\begin{align*}
\sigma_{\max}(SU) = O(1) \text{~and~}  \sigma_{\min}(SU) = \Omega(1/\sqrt{s}).
\end{align*}
\end{lemma}
\begin{proof}
Note that $\sqrt{s} - \sqrt{k-1}= \frac{ s-k-1 }{  \sqrt{s} + \sqrt{k-1} } = \Omega(1/\sqrt{s})$.
\end{proof}

\begin{lemma}\label{lem:f_exact_k_gaussian}
Given matrices $A\in \mathbb{R}^{n\times k}$, $B\in \mathbb{R}^{n\times d}$, let $S\in \mathbb{R}^{s\times n}$ denote a standard Gaussian $N(0,1)$ matrix with $s=k$. Let $X^*=\underset{X\in \mathbb{R}^{k\times d} }{\min} \| A X - B \|_F$. Let $X'=\underset{X\in \mathbb{R}^{k\times d}}{\min}\| SA X - SB\|_F$. Then, we have that
\begin{align*}
\| A X' - B \|_F \leq O(\sqrt{k}) \| A X^* - B \|_F,
\end{align*}
holds with probability at least $0.99$.
\end{lemma}
\begin{proof}
Let $X^* \in \mathbb{R}^{k\times d}$ denote the optimal solution such that
\begin{align*}
\| A X^* - B \|_F = \underset{X \in \mathbb{R}^{k\times d} }{\min} \| AX - B \|_F.
\end{align*}

Consider a standard Gaussian matrix $S \in \mathbb{R}^{k\times n}$ scaled by $1/\sqrt{k}$ with exactly $k$ rows. Then for any $X \in \mathbb{R}^{k\times d}$, by the triangle inequality, we have
\begin{align*}
  \| SAX - SB \|_F \leq \| SAX - SA X^* \|_F + \| SA X^* - SB \|_F,
\end{align*}
and
\begin{align*}
 \| SAX - SB \|_F  \geq \| SAX - SA X^* \|_F - \| SA X^* - SB \|_F.
\end{align*}
We first show how to bound $\| SAX - SA X^* \|_F$, and then show how to bound $\| SA X^* - SB \|_F$.

Note that Lemma~\ref{lem:rv09_tool} implies the following result,
\begin{claim}\label{cla:f_exact_k_gaussian_subspace_embedding}
For any $X \in \mathbb{R}^{k\times d}$, with probability $0.999$, we have
\begin{align*}
\frac{ 1 }{ \sqrt{k} } \| A X - AX^* \|_F \lesssim \| SA X - SA X^* \|_F \lesssim \| A X - A X^* \|_F.
\end{align*}
\end{claim}
\begin{proof}
First, we can write $A= UR \in \mathbb{R}^{n \times k}$ where $U\in \mathbb{R}^{n\times k}$ has orthonormal columns and $R \in \mathbb{R}^{k\times k}$. It gives,
\begin{align*}
  \| SAX - SAX^* \|_F = \| SU ( R X - R X^* ) \|_F.
\end{align*}
Second, applying Lemma~\ref{lem:rv09_tool} to $SU \in \mathbb{R}^{s\times k}$ completes the proof.
\end{proof}

Using Markov's inequality, for any fixed matrix $AX^* - B$, choosing a Gaussian matrix $S$, we have that
\begin{align*}
  \| S A X^* -S B \|_F^2 = O(\|A X^* - B\|_F^2)
\end{align*}
holds with probability at least $0.999$. This is equivalent to
\begin{align}\label{eq:f_exact_k_gaussian_fixed_AXstar_minus_B}
\|S A X^* - S B\|_F = O(\| A X^* - B \|_F),
\end{align}
 holding with probability at least $0.999$.

Let  $ X' = \underset{X \in \mathbb{R}^{k\times d} }{\arg\min} \| S A X - S B \|_F$. Putting it all together, we have
\begin{align*}
  & ~\| AX'-B \|_F  \\
 \leq & ~ \| AX'- AX^* \|_F + \| AX^* - B \|_F &\text{~by~triangle~inequality} \\
 \leq & ~ O(\sqrt{k}) \| S AX'- SA X^* \|_F + \|A X^* - B\|_F &\text{~by~Claim~\ref{cla:f_exact_k_gaussian_subspace_embedding}} \\
  \leq & ~ O(\sqrt{k}) \| S AX'- S B \|_F + O(\sqrt{k}) \| SA X^* - SB \|_F+ \|A X^* - B\|_F &\text{~by~triangle~inequality} \\
  \leq & ~ O(\sqrt{k}) \| S AX^*- S B \|_F + O(\sqrt{k}) \| SA X^* - SB \|_F+ \|A X^* - B\|_F &\text{~by~definition~of~}X' \\
 \leq & ~ O(\sqrt{k}) \| A X^* - B \|_F. & \text{~by~Equation~\eqref{eq:f_exact_k_gaussian_fixed_AXstar_minus_B}}
\end{align*}
\end{proof}

\subsubsection{Algorithm \RN{2}}

\begin{theorem}\label{thm:f_bicriteria_best}
Given a $3$rd order tensor $A\in \mathbb{R}^{n\times n\times n}$, for any $k\geq 1$, let $r=k^2$. There exists an algorithm which takes $O(\nnz(A) k) + n \poly(k)$ time and outputs three matrices $U,V,W\in \mathbb{R}^{n\times r}$ such that,
\begin{align*}
\left\| \sum_{i=1}^r U_i \otimes V_i \otimes W_i - A \right\|_F \leq \poly(k) \min_{\rank-k~A'} \| A' - A \|_F
\end{align*}
holds with probability $9/10$.
\end{theorem}
\begin{proof}
Let $\OPT=\underset{\rank-k~A'}{\min}\| A' - A \|_F,$
we fix $V^*\in \mathbb{R}^{n \times k}, W^*\in \mathbb{R}^{n\times k}$ to be the optimal solution of the original problem. We use $Z_1 = (V^{*\top}\odot W^{*\top})\in \mathbb{R}^{k\times n^2}$ to denote the matrix where the $i$-th row is the vectorization of $V_i^* \otimes W_i^*$. Let $A_1 \in \mathbb{R}^{n\times n^2}$ denote the matrix obtained by flattening tensor $A\in \mathbb{R}^{n\times n \times n}$ along the first direction. Then, we have
\begin{align*}
\min_{U} \| U Z_1 -A_1 \|_F \leq \OPT.
\end{align*}
Choosing an $N(0,1/k)$ Gaussian sketching matrix $S_1\in \mathbb{R}^{n^2\times s_1}$ with $s_1=k$, we can obtain the smaller problem,
\begin{align*}
\min_{U\in \mathbb{R}^{n\times k} } \| U Z_1 S_1 - A_1 S_1 \|_F.
\end{align*}
Define $\wh{U} = A_1S_1 (Z_1 S_1)^\dagger$. Define $\alpha=O(\sqrt{k})$. By Lemma~\ref{lem:f_exact_k_gaussian}, we have
\begin{align*}
\|\wh{U} Z_1 - A_1 \|_F \leq \alpha \OPT.
\end{align*}
Second, we fix $\wh{U}$ and $W^*$. Define $Z_2,A_2$ similarly as above. Choosing an $N(0,1/k)$ Gaussian sketching matrix $S_2\in \mathbb{R}^{n^2\times s_2}$ with $s_2=k$, we can obtain another smaller problem,
\begin{align*}
\min_{V\in \mathbb{R}^{n\times k}} \| V Z_2 S_2 - A_2 S_2 \|_F.
\end{align*}
Define $\wh{V} = A_2 S_2 (Z_2 S_2)^\dagger$. By Lemma~\ref{lem:f_exact_k_gaussian} again, we have
\begin{align*}
\| \wh{V} Z_2 - A_2 \|_F \leq \alpha^2 \OPT.
\end{align*}
Thus, we now have
\begin{align*}
\min_{X_1,X_2,W}\|A_1S_1X_1\otimes A_2S_2X_2 \otimes W -A\|_F\leq \alpha^2\OPT
\end{align*}

We use a similar idea as in the proof of Theorem~\ref{thm:f_bicriteria_better}.
We create matrix $\wt{U} \in \mathbb{R}^{n\times s_1 s_2}$ by copying matrix $A_1 S_1$ $s_2$ times, i.e.,
\begin{align*}
\wt{U} = \begin{bmatrix}
A_1 S_1 & A_1 S_1 & \cdots & A_1 S_1
\end{bmatrix}.
\end{align*}
We create matrix $\wt{V} \in \mathbb{R}^{n\times s_1 s_2}$ by copying the $i$-th column of $A_2 S_2$ a total of $s_1$ times, into columns $(i-1)s_1, \cdots, i s_1$ of $\wt{V}$, for each $i\in [s_2]$, i.e.,
\begin{align*}
\wt{V} = \begin{bmatrix}
(A_2 S_2)_1 & \cdots &(A_2 S_2)_1 & (A_2 S_2)_2 & \cdots &(A_2 S_2)_2 & \cdots & (A_2 S_2)_{s_2} & \cdots &(A_2 S_2)_{s_2}
\end{bmatrix}.
\end{align*}
We have
\begin{align*}
\min_{X\in\mathbb{R}^{n\times s_1s_2}} \|\wt{U}\otimes \wt{V}\otimes X-A\|_F\leq \alpha^2\OPT.
\end{align*}

Choose $T_i \in \mathbb{R}^{t_i\times n}$ to be a sparse embedding matrix (Definition~\ref{def:count_sketch_transform}) with $t_i = \poly(k/\varepsilon)$, for each $i\in [2]$. By applying Lemma~\ref{lem:f_input_sparsity_for_regression}, we have, if $W'$ satisfies,
\begin{align*}
 \|T_1\wt{U}\otimes T_2\wt{V}\otimes W'-A(T_1,T_2,I)\|_F= \min_{X\in\mathbb{R}^{n\times s_1s_2}}\|T_1\wt{U}\otimes T_2\wt{V}\otimes X-A(T_1,T_2,I)\|_F
\end{align*}
then,
\begin{align*}
 \|\wt{U}\otimes \wt{V}\otimes W'-A\|_F\leq  (1+\epsilon) \min_{X\in\mathbb{R}^{n\times s_1s_2}} \| \wt{U}\otimes \wt{V}\otimes X-A\|_F\leq (1+\varepsilon)\alpha^2 \OPT.
\end{align*}
Thus, we only need to solve
\begin{align*}
\min_{X\in\mathbb{R}^{n\times s_1s_2}}\|T_1\wt{U}\otimes T_2\wt{V}\otimes X-A(T_1,T_2,I)\|_F.
\end{align*}
which is similar to the proof of Theorem~\ref{thm:f_bicriteria_better}. Therefore, we complete the proof of correctness.
For the running time, $A_1S_1,A_2S_2$ can be computed in $O(\nnz(A)k)$ time, $T_1\wt{U},T_2\wt{V}$ can be computed in $n\poly(k)$ time. The final regression problem can be computed in $n\poly(k)$ running time.
\end{proof}

\subsection{Generalized matrix row subset selection}\label{sec:f_generalized_matrix_row}
Note that in this section, the notation $\Pi_{C,k}^{\xi}$ is given in Definition~\ref{def:subspace_best}.

\begin{algorithm}[h]\caption{Generalized Matrix Row Subset Selection: Constructing $R$ with $r=O(k+k/\epsilon)$ Rows and a rank-$k$ $U\in \mathbb{R}^{k\times r}$}\label{alg:f_generalized_matrix_row}
\begin{algorithmic}[1]
\Procedure{\textsc{GeneralizedMatrixRowSubsetSelection}}{$A,C,n,m,k,\epsilon$} \Comment{Theorem \ref{thm:f_generalized_matrix_row}}
\State $Y, \Phi, \Delta \leftarrow$ \textsc{ApproxSubspaceSVD}($A,C,k$). \Comment{Claim~\ref{cla:f_generalized_approx_subspace} and Lemma 3.12 in \cite{bw14}} 
\State $B \leftarrow Y \Delta$.
\State $Z_2, D \leftarrow \textsc{QR}(B)$.  \Comment{$Z_2 \in \mathbb{R}^{m\times k}$, $Z_2^\top Z_2 = I_k$, $D\in \mathbb{R}^{k\times k}$}
\State $h_2 \leftarrow 8 k\ln (20k)$.
\State $\Omega_2, D_2 \leftarrow \textsc{RandSampling}(Z_2, h_2,1)$ \Comment{Definition 3.6 in \cite{bw14}}
\State $M_2 \leftarrow Z_2^\top \Omega_2 D_2 \in \mathbb{R}^{k\times h_2}$.
\State $U_{M_2}, \Sigma_{M_2}, V_{M_2}^\top \leftarrow \textsc{SVD}(M_2)$. \Comment{$\rank(M_2)=k$ and $V_{M_2}\in \mathbb{R}^{h_2\times k}$}
\State $r_1 \leftarrow 4k$.
\State $S_2\leftarrow$ \textsc{BSSSamplingSparse}($V_{M_2}, ( (A^\top - A^\top Z_2 Z_2^\top ) \Omega_2 D_2 )^\top , r_1, 0.5$) \Comment{Lemma 4.3 in \cite{bw14}} 
\State $R_1 \leftarrow (A^\top \Omega_2 D_2 S_2)^\top \in \mathbb{R}^{r_1 \times n}$ containing rescaled rows from $A$.
\State $r_2 \leftarrow 4820k/\epsilon$.
\State $R_2 \leftarrow$ \textsc{AdaptiveRowsSparse}($A,Z_2,R_1,r_2$) \Comment{Lemma 4.5 in \cite{bw14}} 
\State $R\leftarrow [R_1^\top, R_2^\top]^\top$. \Comment{$R\in \mathbb{R}^{ (r_1+r_2)\times n }$ containing $r=4k+4820k/\epsilon$ rescaled rows of $A$. }
\State Choose $W\in \mathbb{R}^{\xi\times m}$ to be a randomly chosen sparse subspace embedding with $\xi = \Omega(k^2\epsilon^{-2})$.
\State $U\leftarrow \Phi^{-1}\Delta D^{-1} (WC \Phi^{-1} \Delta D^{-1})^\dagger W A R^\dagger = \Phi^{-1} \Delta \Delta^\top (WC)^\dagger WAR^\dagger$.
\State \Return $R$, $U$.
\EndProcedure
\end{algorithmic}
\end{algorithm}

\begin{theorem}\label{thm:f_generalized_matrix_row}
Given matrices $A\in \mathbb{R}^{n\times m}$ and $C\in \mathbb{R}^{n\times k}$, there exists an algorithm which takes $O(\nnz(A) \log n ) + (m+n)\poly(k,1/\epsilon)$ time and outputs a diagonal matrix $D \in \mathbb{R}^{n\times n}$ with $d=O(k/\epsilon)$ nonzeros (or equivalently a matrix $R$ that contains $d=O(k/\epsilon)$ rescaled rows of $A$) and a matrix $U\in \mathbb{R}^{k \times d}$ such that
\begin{align*}
\| C U D A - A \|_F^2 \leq (1+\epsilon) \min_{X\in \mathbb{R}^{k\times m} }\| C X - A \|_F^2
\end{align*}
holds with probability $.99$.
\end{theorem}
\begin{proof}
This follows by combining Lemma~\ref{lem:f_generalized_A_minus_ZZARR} and \ref{lem:f_generalized_A_minus_CUR}.
Let $U,R$ denote the output of procedure \textsc{GeneralizedMatrixRowSubsetSelection},
\begin{align*}
\| A - CUR \|_F^2 \leq & ~(1+\epsilon) \| A - Z_2Z_2^\top A R^\dagger R\|_F^2 \\
 \leq & ~ (1+\epsilon) (1+60\epsilon) \| A - \Pi_{C,k}^F(A) \|_F^2\\
 \leq & ~ (1+130\epsilon) \| A - \Pi_{C,k}^F(A) \|_F^2.
\end{align*}
Because $R$ is a subset of rows of $A$ and $R$ has size $O(k/\epsilon)\times m$, there must exist a diagonal matrix $D\in \mathbb{R}^{n\times n}$ with $O(k/\epsilon)$ nonzeros such that $R = D A$. This completes the proof.
\end{proof}

\begin{corollary}[A slightly different version of Theorem~\ref{thm:f_generalized_matrix_row}, faster running time, and small input matrix]\label{cor:f_generalized_matrix_row}
Given matrices $A\in \mathbb{R}^{n\times m}$ and $C\in \mathbb{R}^{n\times k}$, if $\min(m,n)=\poly(k,1/\epsilon)$, then there exists an algorithm which takes $O(\nnz(A)) + (m+n)\poly(k,1/\epsilon)$ time and outputs a diagonal matrix $D \in \mathbb{R}^{n\times n}$ with $d=O(k/\epsilon)$ nonzeros (or equivalently a matrix $R$ that contains $d=O(k/\epsilon)$ rescaled rows of $A$) and a matrix $U\in \mathbb{R}^{k \times d}$ such that
\begin{align*}
\| C U D A - A \|_F^2 \leq (1+\epsilon) \min_{X\in \mathbb{R}^{k\times m} }\| C X - A \|_F^2
\end{align*}
holds with probability $.99$.
\end{corollary}
\begin{proof}
The $\log n$ factor comes from the adaptive sampling where we need to choose a Gaussian matrix with $O(\log n)$ rows and compute $SA$. If $A$ has $\poly(k,1/\epsilon)$ columns, it is sufficient to choose $S$ to be a CountSketch matrix with $\poly(k,1/\epsilon)$ rows. Then, we do not need a $\log n$ factor in the running time. If $S$ has $\poly(k,1/\epsilon)$ rows, then we no longer need the matrix $S$.
\end{proof}

\begin{claim}\label{cla:f_generalized_approx_subspace}
Given matrices $A\in \mathbb{R}^{m\times n}$ and $C\in \mathbb{R}^{m\times c}$, let $Y \in \mathbb{R}^{m\times c},\Phi\in \mathbb{R}^{c\times c}$ and $\Delta \in \mathbb{R}^{c\times k}$ denote the output of procedure $\textsc{ApproxSubspaceSVD}(A,C,k,\epsilon)$. Then with probability $.99$, we have,
\begin{align*}
\| A - Y \Delta \Delta^\top Y^\top A \|_F^2 \leq (1+30\epsilon)  \| A - \Pi_{C,k}^F (A) \|_F^2.
\end{align*}
\end{claim}
\begin{proof}
This follows by Lemma 3.12 in \cite{bw14}.
\end{proof}

\begin{lemma}\label{lem:f_generalized_A_minus_ZZARR}
The matrices $R$ and $Z_2$ in procedure \textsc{GeneralizedMatrixRowSubsetSelection} (Algorithm~\ref{alg:f_generalized_matrix_row}) satisfy with probability at least $0.17-2/n$,
\begin{align*}
\| A  - Z_2 Z_2^\top A R^\dagger R \|_F^2 \leq \| A - \Pi_{C,k}^F(A) \|_F^2 + 60\epsilon \| A - \Pi_{C,k}^F(A) \|_F^2.
\end{align*}
\end{lemma}

\begin{proof}
We can show,
\begin{align*}
 & ~\| A - Z_2 Z_2^\top A \|_F^2 + \frac{30\epsilon}{4820} \| A - A R_1^\dagger R_1 \|_F^2  \\
=& ~ \| A -  BB^\dagger A\|_F^2 + \frac{30\epsilon}{4820} \| A - A R_1^\dagger R_1 \|_F^2 \\
\leq & ~ \| A -  BB^\dagger A\|_F^2 + 30\epsilon \| A - A_k \|_F^2 \\
\leq & ~ \| A - Y \Delta \Delta^\top Y A \|_F^2 + 30\epsilon \| A - \Pi_{C,k}^F(A) \|_F^2 \\
\leq & ~ (1+30\epsilon)\| A - \Pi_{C,k}^F(A) \|_F^2 + 30\epsilon \| A - \Pi_{C,k}^F(A) \|_F^2,
\end{align*}
where the first step follows by the fact that $Z_2 Z_2^\top = Z_2 D D^{-1} Z_2^\top = (Z_2 D) (Z_2 D)^\dagger =  BB^\dagger$, the second step follows by $\| A - A R_1^\dagger R_1\|_F^2 \leq 4820 \| A - A_k\|_F^2$, the third step follows by $B=Y \Delta$ and $B^\dagger = (Y \Delta)^\dagger = \Delta^\dagger Y^\dagger = \Delta^\top Y^\top$, and the last step follows by Claim~\ref{cla:f_generalized_approx_subspace}.
\end{proof}

\begin{lemma}\label{lem:f_generalized_A_minus_CUR}
The matrices $C,U$ and $R$ in procedure \textsc{GeneralizedMatrixRowSubsetSelection} (Algorithm~\ref{alg:f_generalized_matrix_row}) satisfy that
\begin{align*}
\| A- CUR \|_F^2 \leq (1+\epsilon) \| A - Z_2 Z_2^\top A R^\dagger R \|_F^2
\end{align*}
with probability at least $.99$.
\end{lemma}
\begin{proof}
Let $U_R,\Sigma_R,V_R$ denote the SVD of $R$. Then $V_R V_R^\top = R^\dagger R$.

We define $Y^*$ to be the optimal solution of
\begin{align*}
\min_{X\in \mathbb{R}^{k\times r} } \| W A V_R V_R^\top - WC \Phi^{-1} \Delta D^{-1} Y R \|_F^2.
\end{align*}
We define $\wh{X}^*$ to be $Y^* R \in \mathbb{R}^{k \times n}$, which is also equivalent to defining $\wh{X}^*$ to be the optimal solution of
\begin{align*}
\min_{X\in \mathbb{R}^{k\times n} } \| W A V_R V_R^\top - W C \Phi^{-1} \Delta D^{-1} X \|_F^2.
\end{align*}
Furthermore, it implies $\wh{X}^* = (WC \Phi^{-1} \Delta D^{-1})^\dagger WA V_R V_R^\dagger$.

We also define $X^*$ to be the optimal solution of
\begin{align*}
\min_{X\in \mathbb{R}^{k\times n}} \| A V_R V_R^\dagger - C \Phi^{-1} \Delta D^{-1} X \|_F^2,
\end{align*}
which implies that,
\begin{align*}
X^* = (C\Phi^{-1} \Delta D^{-1})^\dagger A V_R V_R^\top = Z_2^\top A V_R V_R^\top.
\end{align*}
Now, we start to prove an upper bound on $\| A - CUR \|_F^2$,
\begin{align}\label{eq:f_generalized_A_minus_CUR_eq1}
\| A - CUR \|_F^2 = & ~ \| A - C \Phi^{-1} \Delta D^{-1} Y^* R\|_F^2 & \text{~by~definition~of~}U \notag \\
= & ~ \| A - C \Phi^{-1} \Delta D^{-1} \wh{X}^*\|_F^2 & \text{~by~}\wh{X}^* = Y^*R \notag \\
= & ~ \| A V_R V_R^\top - C\Phi^{-1} \Delta D^{-1} \wh{X}^* + A - A V_R V_R^\top \|_F^2 \notag \\
= & ~ \underbrace{\| A V_R V_R^\top - C\Phi^{-1} \Delta D^{-1} \wh{X}^* \|_F^2 }_{\alpha} + \underbrace{\|A - A V_R V_R^\top \|_F^2}_{\beta},
\end{align}
where the last step follows by $\wh{X}^*=MV_R^\top$, $A-AV_RV_R^\top = A(I-V_RV_R^\top)$ and the Pythagorean theorem.
We show how to upper bound the term $\alpha$,
\begin{align}\label{eq:f_generalized_A_minus_CUR_eq2}
\alpha \leq & ~ (1+\epsilon) \| AV_R V_R^\top -C \Phi^{-1} \Delta D^{-1} X^*\|_F^2  &\text{~by~Lemma~\ref{lem:bw14_lemma3_16}} \notag \\
= & ~ \epsilon \| AV_R V_R^\top -C \Phi^{-1} \Delta D^{-1} X^*\|_F^2 + \| AV_R V_R^\top -C \Phi^{-1} \Delta D^{-1} X^*\|_F^2 \notag \\
= & ~\epsilon \| AV_R V_R^\top -C \Phi^{-1} \Delta D^{-1} X^*\|_F^2 + \| AV_R V_R^\top -C \Phi^{-1} \Delta D^{-1} (Z_2^\top A R^\dagger R)\|_F^2.
\end{align}
By the Pythagorean theorem and the definition of $Z_2$ (which means $Z_2 = C \Phi^{-1}\Delta D^{-1}$), we have,
\begin{align}\label{eq:f_generalized_A_minus_CUR_eq3}
  & ~ \| AV_R V_R^\top -C \Phi^{-1} \Delta D^{-1} (Z_2^\top A R^\dagger R)\|_F^2 + \beta \notag \\
= & ~ \| AV_R V_R^\top -C \Phi^{-1} \Delta D^{-1} (Z_2^\top A R^\dagger R)\|_F^2 + \| A - A V_R V_R^\top \|_F^2 \notag \\
= & ~ \| A - C \Phi^{-1} \Delta D^{-1} (Z_2^\top A R^\dagger R) \|_F^2 \notag  \\
= & ~ \| A - Z_2 Z_2^\top A R^\dagger R \|_F^2.
\end{align}
Combining Equations (\ref{eq:f_generalized_A_minus_CUR_eq1}), (\ref{eq:f_generalized_A_minus_CUR_eq2}) and (\ref{eq:f_generalized_A_minus_CUR_eq3}) together, we obtain,
\begin{align*}
\| A - CUR \|_F^2 \leq \epsilon \| A V_RV_R^\top - C\Phi^{-1} \Delta D^{-1} X^*\|_F^2 + \| A -  Z_2 Z_2^\top A R^\dagger R \|_F^2.
\end{align*}
We want to show $\| A V_RV_R^\top - C\Phi^{-1} \Delta D^{-1} X^*\|_F^2\leq \| A -Z_2 Z_2^\top A R^\dagger R\|_F^2 $,
\begin{align*}
& ~\| A V_RV_R^\top - C\Phi^{-1} \Delta D^{-1} X^*\|_F^2\\
= & ~ \| A V_RV_R^\top - C\Phi^{-1} \Delta D^{-1} Z_2^\top A V_R V_R^\top \|_F^2 & \text{~by~} X^* = Z_2^\top A V_R V_R^\top\\
\leq & ~ \| A  - C\Phi^{-1} \Delta D^{-1} Z_2^\top A \|_F^2 &\text{~by~properties~of~projections}\\
\leq & ~ \| A  - C\Phi^{-1} \Delta D^{-1} Z_2^\top A R^\dagger R\|_F^2 &\text{~by~properties~of~projections} \\
= & ~ \| A - Z_2 Z_2^\top A R^\dagger R\|_F^2. & \text{~by~} Z_2 = C \Phi^{-1} \Delta D^{-1}
\end{align*}
This completes the proof.
\end{proof}

\begin{lemma}[\cite{cw13}]\label{lem:bw14_lemma3_16}
Let $A\in \mathbb{R}^{n\times d}$ have rank $\rho$ and $B\in \mathbb{R}^{n \times r }$. Let $W\in \mathbb{R}^{r\times n} $ be a randomly chosen sparse subspace embedding with $r= \Omega(\rho^2 \epsilon^{-2})$. Let $\wh{X}^* = \underset{X\in \mathbb{R}^{d\times r}}{\arg\min} \| W  A X -  W B \|_F^2$ and let $X^*=\underset{X\in \mathbb{R}^{d\times r}}{\arg\min} \| A X - B \|_F^2$. Then with probability at least $.99$,
\begin{align*}
\| A \wt{X}^* - B \|_F^2 \leq (1+\epsilon) \| AX^* - B \|_F^2.
\end{align*}
\end{lemma}

\subsection{Column, row, and tube subset selection, $(1+\epsilon)$-approximation}\label{sec:f_columns_rows_tubes_subset_selection}

\begin{algorithm}[h]\caption{Frobenius Norm Tensor Column, Row and Tube Subset Selection, Polynomial Time}\label{alg:f_fast_curt_without_u}
\begin{algorithmic}[1]
\Procedure{\textsc{FCRTSelection}}{$A,n,k,\epsilon$} \Comment{Theorem \ref{thm:f_fast_curt_without_u}}
\State $s_1 \leftarrow s_2 \leftarrow O(k/\epsilon)$.
\State Choose a Gaussian matrix $S_1$ with $s_1$ columns. \Comment{Definition~\ref{def:fast_gaussian_transform}}
\State Choose a Gaussian matrix $S_2$ with $s_2$ columns. \Comment{Definition~\ref{def:fast_gaussian_transform}}
\State Form matrix $Z_3'$ by setting the $(i,j)$-th row to be the vectorization of $(A_1 S_1)_i \otimes (A_2 S_2)_j$.
\State $D_3\leftarrow$\textsc{GeneralizedMatrixRowSubsetSelection}($A_3^\top$, $(Z_3')^\top$,$n^2$,$n$,$s_1 s_2$,$\epsilon$). \Comment{Algorithm \ref{alg:f_generalized_matrix_row}}
\State Let $d_3$ denote the number of nonzero entries in $D_3$. \Comment{$d_3= O(s_1 s_2/\epsilon)$}
\State Form matrix $Z_2'$ by setting the $(i,j)$-th row to be the vectorization of $( A_1 S_1)_i \otimes (A_3 S_3')_j$.
\State $D_2\leftarrow$\textsc{GeneralizedMatrixRowSubsetSelection}($A_2^\top$, $(Z_2')^\top$,$n^2$,$n$,$s_1 d_3$,$\epsilon$).
\State Let $d_2$ denote the number of nonzero entries in $D_2$. \Comment{$d_2 = O(s_1 d_3/\epsilon)$}
\State Form matrix $Z_1'$ by setting the $(i,j)$-th row to be the vectorization of  $(A_2 D_2)_i \otimes ( A_3 D_3)_j$.
\State $D_1\leftarrow$\textsc{GeneralizedMatrixRowSubsetSelection}($A_1^\top$, $(Z_1')^\top$,$n^2$,$n$,$d_2 d_3$,$\epsilon$).
\State Let $d_1$ denote the number of nonzero entries in $D_1$. \Comment{$d_1 = O(d_2 d_3/\epsilon)$}
\State $C\leftarrow A_1 D_1$, $R\leftarrow A_2 D_2$ and $T\leftarrow A_3 D_3$.
\State \Return $C$, $R$ and $T$.
\EndProcedure
\end{algorithmic}
\end{algorithm}

We provide two bicriteria CURT results in this Section. We first present a warm-up result. That result (Theorem~\ref{thm:f_fast_curt_without_u}) does not output tensor $U$ and only guarantees that there is a $\rank$-$\poly(k/\epsilon)$ tensor $U$. Then we show the second result (Theorem~\ref{thm:f_fast_curt_with_u_but_bicriteria}), our second result is able to output tensor $U$. The $U$ has rank $\poly(k/\epsilon)$, but not $k$.

\begin{theorem}\label{thm:f_fast_curt_without_u}
Given a 3rd order tensor $A\in \mathbb{R}^{n\times n\times n}$, for any $k\geq 1$, there exists an algorithm which takes $O(\nnz(A))+ n \poly(k,1/\epsilon)$ 
 time and outputs three matrices: $C\in \mathbb{R}^{n\times c}$, a subset of columns of $A$, $R\in \mathbb{R}^{n\times r}$ a subset of rows of $A$, and $T\in \mathbb{R}^{n\times t}$, a subset of tubes of $A$ where $c=r=t= \poly(k,1/\epsilon)$, and  there exists a tensor $U\in \mathbb{R}^{c\times r\times t}$ such that
\begin{align*}
 \| ( ( ( U \cdot T^\top)^\top \cdot R^\top )^\top \cdot C^\top )^\top  - A \|_F^2  \leq (1+\epsilon) \underset{\rank-k~A_k}{\min} \| A_k - A \|_F^2,
\end{align*}
or equivalently,
\begin{align*}
\left\| \sum_{i=1}^c \sum_{j=1}^r \sum_{l=1}^t U_{i,j,l} \cdot C_i \otimes R_j \otimes T_l -A \right\|_F^2 \leq (1+\epsilon) \underset{\rank-k~A_k}{\min} \| A_k - A \|_F^2
\end{align*}
holds with probability $9/10$.
\end{theorem}
\begin{proof}
We mainly analyze Algorithm~\ref{alg:f_fast_curt_without_u} and it is easy to extend to Algorithm~\ref{alg:f_fast_curt_without_u_optimal_time}.

We fix $V^*\in \mathbb{R}^{n\times k}$ and $W^*\in \mathbb{R}^{n\times k}$. We define $Z_1\in \mathbb{R}^{k\times n^2}$ where the $i$-th row of $Z_1$ is the vector $V_i \otimes W_i$. Choose sketching (Gaussian) matrix $S_1 \in \mathbb{R}^{n^2 \times s_1}$ (Definition~\ref{def:fast_gaussian_transform}), and let $\wh{U} = A_1 S_1 (Z_1 S_1)^\dagger \in \mathbb{R}^{n\times k}$. Following a similar argument as in the previous theorem, we have
\begin{align*}
\| \wh{U} Z_1 - A_1 \|_F^2 \leq (1+\epsilon) \OPT.
\end{align*}
We fix $\wh{U}$ and $W^*$. We define $Z_2\in \mathbb{R}^{k\times n^2}$ where the $i$-th row of $Z_2$ is the vector $\wh{U}_i \otimes W^*_i$. Choose sketching (Gaussian) matrix $S_2\in \mathbb{R}^{n^2\times s_2}$ (Definition~\ref{def:fast_gaussian_transform}), and let $\wh{V} = A_2 S_2 (Z_2 S_2)^\dagger \in \mathbb{R}^{n\times k}$. Following a similar argument as in the previous theorem, we have
\begin{align*}
\| \wh{V} Z_2 - A_2 \|_F^2 \leq (1+\epsilon)^2 \OPT.
\end{align*}

We fix $\wh{U}$ and $\wh{V}$. Note that $\wh{U}=A_1 S_1 (Z_1 S_1)^\dagger$ and $\wh{V}=A_2 S_2 (Z_2 S_2)^\dagger$. We define $Z_3\in\mathbb{R}^{k\times n^2}$ such that the $i$-th row of $Z_3$ is the vector $\wh{U}_i \otimes \wh{V}_i$. Let $z_3 = s_1 \cdot s_2$. We define $Z'_3\in \mathbb{R}^{z_3\times n^2}$ such that, $\forall i\in [s_1], \forall j\in [s_2]$, the $i+(j-1) s_1$-th row of $Z'_3$ is the vector $ (A_1 S_1)_i \otimes (A_2 S_2)_j$. We consider the following objective function,
\begin{align*}
\min_{W \in \mathbb{R}^{n\times k}, X\in \mathbb{R}^{k \times z_3} }\| W X Z_3' - A_3 \|_F^2 \leq \min_{W \in \mathbb{R}^{n\times k} } \| W Z_3 - A_3 \|_F^2\leq (1+\epsilon)^2 \OPT.
\end{align*}

Using Theorem~\ref{thm:f_generalized_matrix_row}, we can find a diagonal matrix $D_3\in \mathbb{R}^{n^2\times n^2}$ with $d_3 = O(z_3/\epsilon) = O(k^2/\epsilon^3)$ nonzero entries such that
\begin{align*}
\min_{X\in \mathbb{R}^{d_3\times z_3 }} \| A_3 D_3 X  Z_3' - A_3\|_F^2  \leq (1+\epsilon)^3 \OPT.
\end{align*}
In the following, we abuse notation and let $A_3 D_3 \in \mathbb{R}^{n\times d_3}$ by deleting zero columns.
 Let $W'$ denote $A_3 D_3 \in \mathbb{R}^{n\times d_3}$. Then,
\begin{align*}
\min_{X \in \mathbb{R}^{d_3\times z_3} } \| W' X Z_3' - A_3 \|_F^2 \leq (1+\epsilon)^3 \OPT.
\end{align*}

We fix $\wh{U}$ and $W'$. Let $z_2 = s_1 \cdot d_3 $. We define $Z'_2 \in \mathbb{R}^{z_2\times n^2}$ such that, $\forall i\in [s_1], \forall j \in [d_3]$,  the $i+(j-1)s_1$-th row of $Z'_2$ is the vector $(A_1 S_1)_i\otimes (A_3 D_3)_j$.

Using Theorem~\ref{thm:f_generalized_matrix_row}, we can find a diagonal matrix $D_2 \in \mathbb{R}^{n^2 \times n^2}$ with $d_2= O(z_2/\epsilon) = O(s_1 d_3/\epsilon) = O(k^3 /\epsilon^5)$ nonzero entries such that
\begin{align*}
\min_{X\in \mathbb{R}^{d_2\times z_2 }} \| A_2 D_2 X  Z_2' - A_2 \|_F^2  \leq (1+\epsilon)^4 \OPT.
\end{align*}
Let $V'$ denote $A_2 D_2$. Then,
\begin{align*}
\min_{X\in \mathbb{R}^{d_2\times z_2 }} \| V' X Z'_2 - A_2 \|_F^2 \leq (1+\epsilon)^4 \OPT.
\end{align*}

We fix $V'$ and $W'$. Let $z_1 = d_2 \cdot d_3 $. We define $Z'_1 \in \mathbb{R}^{z_1\times n^2}$ such that, $\forall i\in [d_2], \forall j \in [d_3]$,  the $i+(j-1)s_1$-th row of $Z'_1$ is the vector $(A_2 D_2)_i\otimes (A_3 D_3)_j$.

Using Theorem~\ref{thm:f_generalized_matrix_row}, we can find a diagonal matrix $D_1 \in \mathbb{R}^{n^2 \times n^2}$ with $d_1= O(z_1/\epsilon) = O(d_2 d_3/\epsilon) = O(k^5/\epsilon^9)$ nonzero entries such that
\begin{align*}
\min_{X\in \mathbb{R}^{d_1\times z_1 }} \| A_1 D_1 X  Z_1' - A_1 \|_F^2  \leq (1+\epsilon)^5 \OPT.
\end{align*}
Let $U'$ denote $A_1 D_1$. Then,
\begin{align*}
\min_{X\in \mathbb{R}^{d_1\times z_1 }} \| U' X Z'_1 - A_1 \|_F^2 \leq (1+\epsilon)^5 \OPT.
\end{align*}

Putting $U',V',W'$ all together, we complete the proof.

All the above analysis gives the running time $O(\nnz(A)) \log n + n^2 \poly(\log n, k, 1/\epsilon)$. To improve the running time, we need to use Algorithm~\ref{alg:f_fast_curt_without_u_optimal_time}, the similar analysis will go through, the running time will be improved to $O(\nnz(A) + n\poly(k,1/\epsilon))$, but the sample complexity of $c,r,k$ will be slightly worse ($\poly \log $ factors).
\end{proof}

\begin{algorithm}[h]\caption{Frobenius Norm Tensor Column, Row and Tube Subset Selection, Input Sparsity Time}\label{alg:f_fast_curt_without_u_optimal_time}
\begin{algorithmic}[1]
\Procedure{\textsc{FCRTSelection}}{$A,n,k,\epsilon$} \Comment{Theorem \ref{thm:f_fast_curt_without_u}}
\State $s_1 \leftarrow s_2 \leftarrow O(k/\epsilon)$.
\State $\epsilon_0 \leftarrow 0.001$.
\State Choose a Gaussian matrix $S_1$ with $s_1$ columns. \Comment{Definition~\ref{def:fast_gaussian_transform}}
\State Choose a Gaussian matrix $S_2$ with $s_2$ columns. \Comment{Definition~\ref{def:fast_gaussian_transform}}
\State Form matrix $B_1$ by setting $(i,j)$-th column to be $(A_1S_1)_i$.
\State Form matrix $B_2$ by setting $(i,j)$-th column to be $(A_2S_2)_j$. \Comment{$Z_3'=B_1^\top \odot B_2^\top$}
\State $d_3\leftarrow O(s_1 s_2 \log(s_1 s_2) + (s_1 s_2/\epsilon))$.
\State $D_3 \leftarrow$\textsc{FastTensorLeverageScoreGeneralOrder}($B_1^\top,B_2^\top,n,n,s_1s_2,\epsilon_0,d_1$). \Comment{Algorithm~\ref{alg:f_fast_tensor_leverage_score_general_order}}
\State Form matrix $B_1$ by setting $(i,j)$-th column to be $(A_1S_1)_i$.
\State Form matrix $B_3$ by setting $(i,j)$-th column to be $(A_3D_3)_j$. \Comment{$Z_2'=B_1^\top \odot B_3^\top$}
\State $d_2\leftarrow O(s_1 d_3 \log(s_1 d_3) + (s_1 d_3/\epsilon))$.
\State $D_2 \leftarrow$\textsc{FastTensorLeverageScoreGeneralOrder}($B_1^\top,B_3^\top,n,n,s_1d_3,\epsilon_0,d_2$).

\State Form matrix $B_2$ by setting $(i,j)$-th column to be $(A_2D_2)_i$.
\State Form matrix $B_3$ by setting $(i,j)$-th column to be $(A_3D_3)_j$. \Comment{$Z_1'=B_2^\top \odot B_3^\top$}
\State $d_1\leftarrow O(d_2 d_3 \log(d_2 d_3) + (d_2 d_3/\epsilon))$.
\State $D_1 \leftarrow$\textsc{FastTensorLeverageScoreGeneralOrder}($B_2^\top,B_3^\top,n,n,d_2d_3,\epsilon_0,d_1$).
\State $C\leftarrow A_1 D_1$, $R\leftarrow A_2 D_2$ and $T\leftarrow A_3 D_3$.
\State \Return $C$, $R$ and $T$.
\EndProcedure
\end{algorithmic}
\end{algorithm}

\begin{theorem}\label{thm:f_fast_curt_with_u_but_bicriteria}
Given a 3rd order tensor $A\in \mathbb{R}^{n\times n\times n}$, for any $k\geq 1$, there exists an algorithm which takes $O(\nnz(A) + n\poly(k,1/\epsilon))$ time and outputs three matrices: $C\in \mathbb{R}^{n\times c}$, a subset of columns of $A$, $R\in \mathbb{R}^{n\times r}$ a subset of rows of $A$, and $T\in \mathbb{R}^{n\times t}$, a subset of tubes of $A$, together with a tensor $U\in \mathbb{R}^{c\times r\times t}$ with $\rank(U)=k'$ where $c=r=t=\poly(k,1/\epsilon)$ and $k'=\poly(k,1/\epsilon)$ such that
\begin{align*}
 \| U(C,R,T)  - A \|_F^2  \leq (1+\epsilon) \underset{\rank-k~A_k}{\min} \| A_k - A \|_F^2,
\end{align*}
or equivalently,
\begin{align*}
\left\| \sum_{i=1}^c \sum_{j=1}^r \sum_{l=1}^t U_{i,j,l} \cdot C_i \otimes R_j \otimes T_l -A \right\|_F^2 \leq (1+\epsilon) \underset{\rank-k~A_k}{\min} \| A_k - A \|_F^2
\end{align*}
holds with probability $9/10$.
\end{theorem}
\begin{proof}
The proof follows by combining Theorem~\ref{thm:bicriteria} and Theorem~\ref{thm:intro_f_curt} directly.
\end{proof}

\subsection{CURT decomposition, $(1+\epsilon)$-approximation}\label{sec:f_curt}

\subsubsection{Properties of leverage score sampling and BSS sampling}

Notice that, the \textsc{BSS} algorithm is a deterministic procedure developed in \cite{bss12} for selecting rows from a matrix $A\in \mathbb{R}^{n\times d}$ (with $\| A\|_2\leq 1$ and $\| A \|_F^2\leq k$) using a selection matrix $S$ so that
\begin{align*}
\| A^\top S^\top S A - A^\top A \|_2 \leq \epsilon.
\end{align*}
The algorithm runs in $\poly(n,d,1/\epsilon)$ time.  Using the ideas from \cite{bw14} and \cite{cemmp15}, we are able to reduce the number of nonzero entries from $O(\epsilon^{-2} k\log k)$ to $O( \epsilon^{-2}k )$, and also improve the running time to input sparsity.

\begin{lemma}[Leverage score preserves subspace embedding - Theorem 2.11 in~\cite{w14}]\label{lem:f_leverage_score_klogkeps2_subspace_embedding}
Given a $\rank$-$k$ matrix $A\in \mathbb{R}^{n\times d}$, via leverage score sampling, we can obtain a diagonal matrix $D$ with $m$ nonzero entries such that, letting $B = DA$, if $m=O(\epsilon^{-2} k\log k)$, then, with probability at least $0.999$, for all $x\in \mathbb{R}^d$,
\begin{align*}
(1-\epsilon) \| A x\|_2  \leq \| B x\|_2 \leq (1+\epsilon) \| A x\|_2
\end{align*}
\end{lemma}

\begin{lemma}\label{lem:f_leverage_score_A_to_B}
Given a $\rank$-$k$ matrix $A\in \mathbb{R}^{n\times d}$, there exists an algorithm that runs in $O(\nnz(A) + n \poly( k,1/\epsilon))$ time and outputs a matrix $B$ containing $O(\epsilon^{-2} k\log k )$ re-weighted rows of $A$, such that
with probability at least $0.999$, for all $x\in \mathbb{R}^d$,
\begin{align*}
(1-\epsilon) \| A x\|_2  \leq \| B x\|_2 \leq (1+\epsilon) \| A x\|_2
\end{align*}
\end{lemma}
\begin{proof}
We choose a sparse embedding matrix (Definition~\ref{def:count_sketch_transform}) $\Pi\in\mathbb{R}^{d\times s}$ with $s=\poly(k/\varepsilon)$. With probability at least $0.999$, $\Pi^\top$ is a subspace embedding of $A^\top$. Thus, $\rank(A\Pi)=\rank(A)$. Also, the leverage scores of $A\Pi$ are the same as those of $A$. Thus, we can compute the leverage scores of $A\Pi$. The running time of computing $A\Pi$ is $O(\nnz(A))$. Thus the total running time is $O(\nnz(A) + n \poly( k,1/\epsilon))$.
\end{proof}

\begin{lemma}\label{lem:f_leverage_score_B_to_BPi}
Let $B$ denote a matrix which contains $O(\epsilon^{-2} k\log k)$ rows of $A\in \mathbb{R}^{n\times d}$. Choosing $\Pi$ to be a sparse subspace embedding matrix of size $d\times O( \epsilon^{-6} (k\log k)^2 )$, with probability at least $0.999$,
\begin{align*}
\| B \Pi \Pi^\top B^\top - B B^\top \|_2 \leq \epsilon \| B \|_2^2.
\end{align*}
\end{lemma}
Combining Lemma~\ref{lem:f_leverage_score_A_to_B}, \ref{lem:f_leverage_score_B_to_BPi} and the \textsc{BSS} algorithm, we obtain:
\begin{lemma}
Given a $\rank$-$k$ matrix $A\in \mathbb{R}^{n\times d}$,
there exists an algorithm that runs in $O(\nnz(A)  + n \poly(k,1/\epsilon))$ time and outputs a sampling and rescaling diagonal matrix $S$ that selects $O(\epsilon^{-2} k)$ re-weighted rows of $A$, such that, with probability at least $0.999$,
\begin{align*}
\| A^\top S^\top S A - A^\top A \|_2 \leq \epsilon \| A\|_2^2.
\end{align*}
or equivalently, for all $x\in \mathbb{R}^d$,
\begin{align*}
(1-\epsilon) \| A x\|_2 \leq \|SA x\|_2 \leq (1+\epsilon) \| A x\|_2.
\end{align*}
\end{lemma}
\begin{proof}
Using Lemma~\ref{lem:f_leverage_score_A_to_B}, we can obtain $B$. Then we apply a sparse subspace embedding matrix $\Pi$ on the right of $B$. At the end, we run the \textsc{BSS} algorithm on $B\Pi$ and we are able to output $O(\epsilon^{-2} k)$ re-weighted rows of $B\Pi$. Using these rows, we are able to determine $O(\epsilon^{-2} k)$ re-weighted rows of $A$.
\end{proof}

\subsubsection{Row sampling for linear regression}


\begin{theorem}[Theorem 5 in \cite{cnw15}]\label{thm:theorem5_cnw15}
We are given $A\in \mathbb{R}^{n\times d}$ with $\| A\|_2^2 \leq 1$ and $\| A\|_F^2 \leq k$, and an $\epsilon\in (0,1)$. There exists a diagonal matrix $S$ with $O(k/\epsilon^2)$ nonzero entries such that
\begin{align*}
\| (SA)^\top SA - A^\top A \|_2 \leq \epsilon.
\end{align*}
\end{theorem}

\begin{corollary}\label{cor:leverage_score_size_single_regression}
Given a rank-$k$ matrix $A\in \mathbb{R}^{n\times d}$, vector $b\in \mathbb{R}^n$, and parameter $\epsilon>0$, let $U\in \mathbb{R}^{n\times (k+1)}$ denote an orthonormal basis of $[A,b]$.  Let $S\in \mathbb{R}^{n\times n}$ denote a sampling and rescaling diagonal matrix according to Leverage score sampling and sparse BSS sampling of $U$ with $m$ nonzero entries. If $m=O(k)$, then $S$ is a $(1\pm 1/2)$ subspace embedding for $U$; if $m=O(k/\epsilon)$, then $S$ satisfies $\sqrt{\epsilon}$-operator norm approximate matrix product for $U$.
\end{corollary}
\begin{proof}
This follows by Lemma~\ref{lem:f_leverage_score_klogkeps2_subspace_embedding}, Lemma~\ref{lem:f_leverage_score_B_to_BPi} and Theorem~\ref{thm:theorem5_cnw15}.
\end{proof}

\begin{lemma}[\cite{nw14}]
Given $A\in \mathbb{R}^{n\times d}$ and $b\in \mathbb{R}^{n}$, let $S\in \mathbb{R}^{n\times n}$ denote a sampling and rescaling diagonal matrix. Let $x^*$ denote $\arg\min_{x}\|Ax-b\|_2^2$ and $x'$ denote $\arg\min_{x}\|S A x - Sb\|_2^2$. If $S$ is a $(1\pm 1/2)$ subspace embedding for the column span of $A$, and $\epsilon'$ (=$\sqrt{\epsilon}$)-operator norm approximate matrix product for $U$ adjoined with $b-Ax^*$, then, with probability at least $.999$,
\begin{align*}
\| A x' - b \|_2^2 \leq (1+\epsilon ) \| A x^* - b \|_2^2.
\end{align*}
\end{lemma}

\begin{proof}
We define $\OPT= \underset{x}{\min} \| A x - b\|_2$.
We define $x' = \underset{x}{\arg\min} \| SA x - S b \|_2^2$ and $x^* = \underset{x}{\arg\min} \| A x - b \|_2^2$. Let $w = b-Ax^*$. Let $U$ denote an orthonormal basis of $A$. We can write $Ax'-Ax^*=U \beta$. Then, we have,
\begin{align*}
\|Ax'-b\|_2^2 = & ~ \| Ax' - A x^* + AA^\dagger b - b\|_2^2 & \text{~by~}x^*=A^\dagger b\\
= & ~ \| U \beta + (UU^\top - I) b \|_2^2  \\
= & ~  \| Ax^*-Ax' \|_2^2 + \| Ax^*-b \|_2^2 & \text{~by~Pythagorean Theorem} \\
= & ~ \| U \beta \|_2^2 + \OPT^2 \\
= & ~ \| \beta \|_2^2 + \OPT^2.
\end{align*}
If S is a $(1\pm1/2)$ subspace embedding for $U$, then we can show

 \begin{align*}
 & \| \beta\|_2 - \| U^\top S^\top S U \beta \|_2 \\
\leq & ~\|  \beta- U^\top S^\top S U \beta \|_2 & \text{~by~triangle~inequality}  \\
= & ~\| ( I - U^\top S^\top S U ) \beta \|_2 \\
\leq & ~ \| I - U^\top S^\top S U \|_2 \cdot \| \beta \|_2 \\
\leq & ~ \frac{1}{2} \|\beta \|_2.
 \end{align*}
 Thus, we obtain
\begin{align*}
\|U^\top S^\top S U \beta\|_2 \geq  \|\beta\|_2/2.
\end{align*}
Next, we can show
\begin{align*}
\| U^\top S^\top S U \beta \|_2 = & ~ \| U^\top S^\top S (Ax'-Ax^*) \|_2^2 \\
= & ~ \| U^\top S^\top S (A (SA)^\dagger Sb -Ax^*) \|_2  & \text{~by~}x'= (SA)^\dagger Sb\\
= & ~ \| U^\top S^\top S (b -Ax^*) \|_2  & \text{~by~} SA (SA)^\dagger = I\\
= & ~ \| U^\top S^\top S w\|_2. & \text{~by~} w= b- Ax^*
\end{align*}
We define $U' = \begin{bmatrix} U & w/\| w\|_2 \end{bmatrix}$. We define $X$ and $y$ to satisfy $U=U'X$ and $w=U'y$. Then, we have
\begin{align*}
& ~ \|U^\top S^\top Sw \|_2 \\
=  & ~\|U^\top S^\top Sw - U^\top w\|_2  & \text{~by~}U^\top w= 0\\
= & ~ \| X^\top U'^\top S^\top S U' y - X^\top U'^\top U' y \|_2 \\
= & ~ \| X^\top(U'^\top S^\top S U' - I) y \|_2 \\
\leq & ~ \| X \|_2 \cdot \| U'^\top S^\top S U' - I \|_2 \cdot \| y \|_2 \\
\leq & ~  \epsilon' \|X\|_2 \|y\|_2 \\
= & ~ \epsilon' \| U\|_2 \|w \|_2  \\
= & ~\epsilon' \OPT, & \text{~by~} \| U \|_2=1 \text{~and~} \|w\|_2 =\OPT
\end{align*}
where the fifth inequality follows since $S$ satisfies $\epsilon'$-operator norm approximate matrix product for the column span of $U$ adjoined with $w$.

Putting it all together, we have
\begin{align*}
\| Ax'-b\|_2^2 = & ~ \| A x^* -b \|_2^2 + \| Ax^* - A x'\|_2^2   \\
= & ~  \OPT^2 + \| \beta \|_2^2 \\
\leq & ~ \OPT^2 + 4 \|U^\top S^\top S w \|_2^2 \\
\leq & ~ \OPT^2 + 4 (\epsilon' \OPT)^2 \\
\leq & ~ (1+\epsilon) \OPT^2. & \text{~by~} \epsilon'= \frac{1}{2} \sqrt{\epsilon}.
\end{align*}

Finally, note that $S$ satisfies $\epsilon'$-operator norm approximate matrix product for $U$ adjoined with $w$ if it is a $(1\pm\epsilon')$-subspace embedding for $U$ adjoined with $w$, which holds using BSS sampling by Theorem 5 of \cite{cnw15} with $O(d/\epsilon)$ samples.
\end{proof}

\subsubsection{Leverage scores for multiple regression}\label{sec:f_leverage_score_multiple_regression}

\begin{lemma}[see, e.g., Lemma 32 in \cite{cw13} among other places]\label{lem:lemma_32_in_cw13}
Given matrix $A\in \mathbb{R}^{n\times d}$ with orthonormal columns, and parameter $\epsilon>0$, if $S\in \mathbb{R}^{n\times n}$ is a sampling and rescaling diagonal matrix according to the leverage scores of $A$ where the number of nonzero entries is $t=O(1/\epsilon^2)$, then, for any $B \in \mathbb{R}^{n \times m}$, we have
\begin{align*}
 \| A^\top S^\top S B - A^\top B \|_F^2 < \epsilon^2 \| A \|_F^2 \| B \|_F^2 ,
\end{align*}
holds with probability at least $0.9999$.
\end{lemma}

\begin{corollary}\label{cor:leverage_score_size_multiple_regression}
Given matrix $A\in \mathbb{R}^{n\times d}$ with orthonormal columns, and parameter $\epsilon>0$, if $S\in \mathbb{R}^{n\times n}$ is a sampling and rescaling diagonal matrix according to the leverage scores of $A$ with $m$ nonzero entries, then if $m=O(d \log d)$, then $S$ is a $(1\pm 1/2)$ subspace embedding for $A$. If $m=O(d/\epsilon)$, then $S$ satisfies $\sqrt{\epsilon/d}$-Frobenius norm approximate matrix product for $A$.
\end{corollary}
\begin{proof}
This follows by Lemma~\ref{lem:f_leverage_score_klogkeps2_subspace_embedding} and Lemma~\ref{lem:lemma_32_in_cw13}.
\end{proof}

\begin{lemma}[\cite{nw14}]\label{lem:nw14_multiple_regression}
Given $A\in \mathbb{R}^{n\times d}$ and $B\in \mathbb{R}^{n\times m}$, let $S\in \mathbb{R}^{n\times n}$ denote a sampling and rescaling matrix according to $A$. Let $X^*$ denote $\arg\min_{X}\| A X - B\|_F^2$ and $X'$ denote $\arg\min_{X} \| S A X - S B \|_F^2$. Let $U$ denote an orthonormal basis for $A$. If $S$ is a $(1\pm 1/2)$ subspace embedding for $U$, and satisfies $\epsilon'$(=$\sqrt{\epsilon/d}$)-Frobenius norm approximate matrix product for $U$, then, we have that
\begin{align*}
\| A X' - B \|_F^2 \leq  (1+\epsilon) \| A X^* - B \|_F^2
\end{align*}
holds with probability at least $0.999$.
\end{lemma}
\begin{proof}

We define $\OPT= \min_{X}\| A X - B\|_F$.
Let $A = U \Sigma V^\top$ denote the SVD of $A$. Since $A$ has rank $k$, $U$ and $V$ have $k$ columns. We can write $ A (X'-X^*)=U \beta$. Then, we have
\begin{align}\label{eq:f_leverage_score_multiple_regression_AX_minus_B}
\| A X'-B \|_F^2 = & ~ \| A X'- AX^* + AA^\dagger B - B \|_F^2  &\text{~by~}X^* = A^\dagger B \notag \\
= & ~ \| U \beta + (UU^\top - I) B \|_F^2  \notag \\
= & ~ \| AX^* - AX'\|_F^2 + \|AX^* - B\|_F^2 & \text{~by~Pythagorean Theorem} \notag \\
= & ~ \| U \beta \|_F^2 + \OPT^2 \notag \\
= & ~ \| \beta \|_F^2 + \OPT^2.
\end{align}

If $S$ is a $(1\pm 1/2)$ subspace embedding for $U$, then we can show,
\begin{align*}
& ~\| \beta\|_F - \| U^\top S^\top S SU \beta \|_F\\
\leq & ~ \| \beta - U^\top S^\top S U \beta \|_F &\text{~by~triangle~inequality} \\
= & ~ \| (I - U^\top S^\top S U) \beta \|_F \\
\leq & ~  \| (I - U^\top S^\top S U) \|_2 \cdot \| \beta \|_F & \text{~by~} \|A B \|_F \leq \|A\|_2 \|B\|_F \\
\leq & ~\frac{1}{2} \| \beta \|_F. & \text{~by~}\| (I - U^\top S^\top S U) \|_2\leq 1/2
\end{align*}
Thus, we obtain
\begin{align}\label{eq:f_leverage_score_multiple_regression_USSUbeta_geq_beta}
\| U^\top S^\top S U \beta \|_F \geq \| \beta \|_F/2.
\end{align}
Next, we can show
\begin{align*}
\| U^\top S^\top S U \beta \|_F = & ~ \| U^\top S^\top S (AX'- AX^*) \|_F \\
= & ~ \|U^\top S^\top S (  A(SA)^\dagger Sb - AX^* ) \|_F & \text{~by~} X'=(SA)^\dagger SB \\
= & ~ \| U^\top S^\top S (B - A X^* ) \|_F. & \text{~by~} SA (SA)^\dagger = I
\end{align*}

Then, we can show
\begin{align}\label{eq:f_leverage_score_multiple_regression_USSB_minus_AX}
 \|U^\top S^\top S (B-AX^*) \|_F \leq & ~ \epsilon' \| U^\top\|_F \| B- AX^* \|_F  & \text{~by~Lemma~\ref{lem:lemma_32_in_cw13}} \notag \\
= & ~\epsilon' \sqrt{d} \OPT. & \text{~by~} \| U \|_F=\sqrt{d} \text{~and~} \|B-AX^*\|_F =\OPT
\end{align}
Putting it all together, we have
\begin{align*}
\| AX'-B\|_F^2 = & ~ \| A X^* -B \|_F^2 + \| AX^* - A X'\|_F^2   \\
= & ~  \OPT^2 + \| \beta \|_F^2 & \text{~by~Equation~\eqref{eq:f_leverage_score_multiple_regression_AX_minus_B}}\\
\leq & ~ \OPT^2 + 4 \|U^\top S^\top S w \|_F^2 & \text{~by~Equation~\eqref{eq:f_leverage_score_multiple_regression_USSUbeta_geq_beta}}\\
\leq & ~ \OPT^2 + 4 (\epsilon' \sqrt{d} \OPT)^2 &\text{~by~Equation~\eqref{eq:f_leverage_score_multiple_regression_USSB_minus_AX}}\\
\leq & ~ (1+\epsilon) \OPT^2. & \text{~by~} \epsilon'= \frac{1}{2} \sqrt{\epsilon/d}
\end{align*}

\end{proof}

\subsubsection{Sampling columns according to leverage scores implicitly, improving polynomial running time to nearly linear running time}
This section explains an algorithm that is able to sample from the leverage scores from the $\odot$ product of two matrices $U,V$ without explicitly writing down $U\odot V$. To build this algorithm we combine \textsc{TensorSketch}, some ideas from \cite{dmmw12} and some ideas from \cite{ako11,mw10}. Finally, we are able to improve the running time of sampling columns according to leverage scores from $\Omega(n^2)$ to $\wt{O}(n)$. Given two matrices $U, V\in \mathbb{R}^{k\times n}$, we define $A\in \mathbb{R}^{k \times n_1 n_2}$ to be the matrix where the $i$-th row of $A$ is the vectorization of $U^i \otimes V^i$, $\forall i \in [k]$. Na\"{i}vely, in order to sample $O(\poly(k,1/\epsilon))$ rows from $A^\top$ according to leverage scores, we need to write down $n^2$ leverage scores. This approach will take at least $\Omega(n^2)$ running time. In the rest of this section, we will explain how to do it in $O(n \cdot \poly(\log n,k,1/\epsilon) )$ time.
In Section~\ref{sec:f_fast_tensor_leverage_score_general_order}, we will explain how to extend this idea from $3$rd order tensors to general $q$-th order tensors and remove the $\poly(\log n)$ factor from running time, i.e., obtain $O(n \cdot \poly(k,1/\epsilon) )$ time.

\begin{algorithm}[!h]\caption{Fast Tensor Leverage Score Sampling}\label{alg:f_fast_tensor_leverage_score}
\begin{algorithmic}[1]
\Procedure{\textsc{FastTensorLeverageScore}}{$U,V,n_1,n_2,k,\epsilon,R_{\text{samples}}$} \Comment{Lemma \ref{lem:f_fast_tensor_leverage_score}}
\State $s_1 \leftarrow \poly(k,1/\epsilon)$.
\State $g_1 \leftarrow g_2 \leftarrow g_3 \leftarrow O(\epsilon^{-2} \log(n_1 n_2))$.
\State Choose $\Pi\in \mathbb{R}^{n_1 n_2 \times s_1}$ to be a \textsc{TensorSketch}. \Comment{Definition~\ref{def:tensor_sketch}}
\State Compute $R^{-1} \in \mathbb{R}^{k\times k}$ by using $(U \odot V) \Pi$. \Comment{$U\in \mathbb{R}^{k\times n_1}$, $V\in \mathbb{R}^{k\times n_2}$}
\State Choose $G_1\in \mathbb{R}^{g_1\times k}$ to be a Gaussian sketching matrix.
\For {$i=1\to g_1$}
	\State $w\leftarrow ( G^i R^{-1})^\top $ \Comment{$G^i$ denotes the $i$-th row of $G$}
	\For{$j=1\to [n_1]$} \Comment{Form matrix $U'^i\in \mathbb{R}^{k\times n_1}$}
		\State $U'^i_j \leftarrow w \circ U_j, \forall j\in[n_1]$. \Comment{$U_j$ denotes the $j$-th column of $U\in \mathbb{R}^{k\times n_1}$}
	\EndFor
\EndFor
\State Choose $G_{2,i}\in \mathbb{R}^{g_2 \times n_1}$ to be a Gaussian sketching matrix.
\For {$i=1 \to g_1$}
	\State $\alpha_i \leftarrow \| ( G_{2,i} U'^{i\top} ) V \|_F^2 $.
	\State Choose $G_{3,i}\in \mathbb{R}^{g_3 \times n_1}$ to be a Gaussian sketching matrix.
	\For{$j_2=1 \to n_2$}
		\State $\beta_{i,j} \leftarrow \|G_{3,i} (U'^{i\top}) V_{j_2} \|_2^2$.
	\EndFor
\EndFor
\State ${\cal S} \leftarrow \emptyset$.
\For {$r=1\to R_{\text{samples}}$}
	\State Sample $i$ from $[g_1]$ with probability $\alpha_i / \sum_{i'=1}^{g_1} \alpha_{i'}$.
	\State Sample $j_2$ from $[n_2]$ with probability $\beta_{i,j_2}/\sum_{j'_2=1}^{n_2} \beta_{i,j'_2}$.
	\For{$j_1=1\to n_1$}
		\State $\gamma_{j_1} \leftarrow (  ( U'^{i\top})^{j_1} V_{j_2} )^2 $.
	\EndFor
	\State Sample $j_1$ from $[n_1]$ with probability $\gamma_{j_1}/\sum_{j'_1=1}^{n_1} \gamma_{j'_1}$.
	\State ${\cal S}\leftarrow {\cal S} \cup (j_1,j_2)$.
\EndFor
\State Convert ${\cal S}$ into a diagonal matrix $D$ with at most $R_{\text{samples}}$ nonzero entries.
\State \Return $D$. \Comment{Diagonal matrix $D\in \mathbb{R}^{n_1n_2 \times n_1n_2}$}
\EndProcedure
\end{algorithmic}
\end{algorithm}

\begin{lemma}\label{lem:f_fast_tensor_leverage_score}
Given two matrices $U\in \mathbb{R}^{k\times n_1}$ and $V\in \mathbb{R}^{k\times n_2}$, there exists an algorithm that takes $O( (n_1 +n_2) \cdot \poly(\log(n_1 n_2), k) \cdot R_{\mathrm{samples}})$ time and samples $R_{\mathrm{samples}}$ columns of $U \odot V \in \mathbb{R}^{k \times n_1 n_2}$ according to the leverage scores of $R^{-1} (U \odot V)$, where $R$ is the $R$ of a QR factorization.
\end{lemma}
\begin{proof}
We choose $\Pi\in \mathbb{R}^{n_1 n_2 \times s_1}$ to be a \textsc{TensorSketch}. Then, according to Section~\ref{sec:def_tensor_sketch}, we can compute $R^{-1}$ in $n \cdot \poly(\log n, k, 1/\epsilon)$ time, where $R$ is the $R$ in a QR-factorization. We want to sample columns from $U \odot V$ according to the square of the $\ell_2$-norms of each column of $R^{-1} (U \odot V)$. However, explicitly writing down the matrix $R^{-1} (U \odot V)$ takes $kn_1n_2$ time, and the number of columns is already $n_1 n_2$. The goal is to sample columns from $R^{-1} (U \odot V)$ without explicitly computing the square of the $\ell_2$-norm of each column.

The first simple observation is that the following two sampling procedures are equivalent in terms of the column samples of a matrix that they take.
(1) We sample a single entry from the matrix $R^{-1} (U \odot V)$ proportional to its squared value. (2) We sample a column from the matrix $R^{-1} (U \odot V)$ proportional to its squared $\ell_2$-norm. Let the $(i,j_1,j_2)$-th entry denote the entry in the $i$-th row and the $(j_1-1)n_2 +j_2$-th column. We can show, for a particular column $(j_1-1)n_2 +j_2$,
\begin{align}\label{eq:f_sampling_entry_is_sampling_column}
 & ~\Pr[\text{sample~an~entry~from~the~}(j_1-1)n_2 +j_2\text{~th~column~of~a~matrix}]\notag \\
= & ~ \sum_{i=1}^k \Pr[\text{sample~the}~(i,j_1,j_2)\text{-th~entry~of~matrix}] \notag \\
 = & ~ \sum_{i=1}^k \frac{ | (R^{-1}(U\odot V))_{i,(j_1-1)n_2 +j_2} |^2 }{\| R^{-1} (U \odot V) \|_F^2} \notag \\
 = & ~ \frac{ \| (R^{-1}(U\odot V))_{(j_1-1)n_2 +j_2} \|^2 }{\| R^{-1} (U \odot V) \|_F^2} \notag \\
 = & ~ \Pr[\text{sample~the~}(j_1-1)n_2 +j_2\text{~th~column~of~matrix}].
\end{align}
Thus, it is sufficient to show how to sample a single entry from matrix $R^{-1}(U\odot V)$ proportional to its squared value without writing down all of the entries of a $k\times n_1 n_2$ matrix.

We choose a Gaussian matrix $G_1\in \mathbb{R}^{g_1\times k}$ with $g_1 = O(\epsilon^{-2} \log(n_1 n_2))$. By Claim~\ref{cla:f_G1_JL} we can reduce the length of each column vector of matrix $R^{-1} U \odot V$ from $k$ to $g_1$ while preserving the squared $\ell_2$-norm of all columns simultaneously. Thus, we obtain a new matrix $GR^{-1} (U \odot V) \in \mathbb{R}^{g_1 \times n_1 n_2}$, and sampling from this new matrix is equivalent to sampling from the original matrix $R^{-1} (U \odot V)$.

 In the following paragraphs, we explain a sampling procedure (also described in Procedure \textsc{FastTensorLeverageScore} in Algorithm~\ref{alg:f_fast_tensor_leverage_score}) which contains three sampling steps. The first step is sampling $i$ from $[g_1]$, the second step is sampling $j_2$ from $[n_2]$, and the last step is sampling $j_1$ from $[n_1]$.

For each $j_1\in[n_1]$, let $U_{j_1}$ denote the $j_1$-th column of $U$. For each $i\in [g_1]$, let $G_1^i$ denote the $i$-th row of matrix $G_1\in \mathbb{R}^{g_1 \times k}$, let $U'^i\in \mathbb{R}^{k\times n_1}$ denote a matrix where the $j_1$-th column is $ (G^{i} R^{-1})^\top \circ U_{j_1} \in \mathbb{R}^k,\forall j\in [n_1]$. Then, using Claim~\ref{cla:f_GiRUV_is_UiV}, we have that $(G^{i} R^{-1}) \cdot (U \odot V) \in \mathbb{R}^{n_1 n_2}$ is a row vector where the entry in the $(j_1-1)n_2 + j_2$-th coordinate is the entry in the $j_1$-th row and $j_2$-th column of matrix $(U'^{i\top} V)\in \mathbb{R}^{n_1 \times n_2}$. Further, the squared $\ell_2$-norm of vector $ (G^{i} R^{-1}) \cdot (U \odot V)$ is equal to the squared Frobenius norm of matrix $(U'^{i\top} V)$. Thus, sampling $i$ proportional to the squared $\ell_2$-norm of vector $ (G^{i} R^{-1}) \cdot (U \odot V)$ is equivalent to sampling $i$ proportional to the squared Frobenius norm of matrix $(U'^{i\top} V)$. Na\"{i}vely, computing the Frobenius norm of an $n_1 \times n_2$ matrix requires $O(n_1 n_2)$ time. However, we can choose a Gaussian matrix $G_{2,i}\in \mathbb{R}^{g_2 \times n_1}$ to sample according to the value $\| (G_{2,i} U'^{i\top}) V \|_F^2$, which can be computed in $O((n_1+n_2)g_2k)$ time. By claim~\ref{cla:f_G2_JL}, $\| (G_{2,i} U'^{i\top}) V \|_F^2 \approx \| ( U'^{i\top}) V \|_F^2$ with high probability. So far, we have finished the first step of the sampling procedure.

For the second step of the sampling procedure, we need to sample $j_2$ from $[n_2]$. To do that, we need to compute the squared $\ell_2$-norm of each column of $U'^{i\top} V\in \mathbb{R}^{n_1 \times n_2}$. This can be done by choosing another Gaussian matrix $G_{3,i}\in \mathbb{R}^{g_3 \times n_1}$. For all $j_2\in [n_2]$, by Claim~\ref{cla:f_G3_JL}, we have $\|G_{3,i} U'^{i\top } V_{j_2} \|_2^2 \approx \|U'^{i\top } V_{j_2} \|_2^2$. Also, for $j_2\in [n_2]$, $\|G_{3,i} U'^{i\top } V_{j_2} \|_2^2$ can be computed in nearly linear in $n_1+n_2$ time.

For the third step of the sampling procedure, we need to sample $j_1$ from $[n_1]$. Since we already have $i$ and $j_2$ from the previous two steps, we can directly compute $| (U'^{i\top})^{j_1} V_{j_2} |^2$, for all $j_1$. This only takes $O(n_1 k)$ time.

Overall, the running time is $O( (n_1 +n_2) \cdot \poly(\log (n_1 n_2), k, 1/\epsilon) )$. Because our estimates are accurate enough, our sampling probabilities are also good approximations to the leverage score sampling probabilities. Putting it all together, we complete the proof.
 \end{proof}

\begin{claim}\label{cla:f_G1_JL}
Given matrix $R^{-1} (U \odot V) \in \mathbb{R}^{k\times n_1n_2}$, let $G_1 \in \mathbb{R}^{g_1 \times k}$ denote a Gaussian matrix with $g_1 = (\epsilon^{-2} \log(n_1 n_2))$. Then with probability at least $1-1/\poly(n_1 n_2)$, we have: for all $j\in [n_1n_2]$,
\begin{align*}
(1-\epsilon) \| R^{-1} (U \odot V)_j \|_2^2 \leq \| G_1 R^{-1} (U \odot V)_j \|_2^2 \leq (1+\epsilon) \| R^{-1} (U \odot V)_j \|_2^2.
\end{align*}
\end{claim}
\begin{proof}
This follows by the Johnson-Lindenstrauss Lemma.
\end{proof}

\begin{claim}
For a fixed $i\in [g_1]$, let $G_{2,i} \in \mathbb{R}^{g_2\times n_1}$ denote a Gaussian matrix with $g_2 = O(\epsilon^{-2}\log(n_1 n_2))$. Then with probability at least $1-1/\poly(n_1n_2)$, we have: for all $j_2 \in [n_2]$,
\begin{align*}
(1-\epsilon)  \| U'^{i\top} V_{j_2}\|_2^2 \leq \| (G_{2,i} U'^{i\top}) V_{j_2} \|_2 \leq (1+\epsilon) \| U'^{i\top} V_{j_2}\|_2^2.
\end{align*}
\end{claim}
By taking the union bound over all $i\in [g_1]$, we obtain a stronger claim,
\begin{claim}\label{cla:f_G2_JL}
With probability at least $1-1/\poly(n_1n_2)$, we have : for all $i\in [g_1]$, for all $j_2 \in [n_2]$,
\begin{align*}
(1-\epsilon)  \| U'^{i\top} V_{j_2}\|_2^2 \leq \| (G_{2,i} U'^{i\top}) V_{j_2} \|_2 \leq (1+\epsilon) \| U'^{i\top} V_{j_2}\|_2^2.
\end{align*}
\end{claim}
Similarly, if we choose $G_{3,i}$ to be a Gaussian matrix, we can obtain the same result as for $G_{2,i}$:
\begin{claim}\label{cla:f_G3_JL}
With probability at least $1-1/\poly(n_1n_2)$, we have : for all $i\in [g_1]$, for all $j_2 \in [n_2]$,
\begin{align*}
(1-\epsilon)  \| U'^{i\top} V_{j_2}\|_2^2 \leq \| (G_{3,i} U'^{i\top}) V_{j_2} \|_2 \leq (1+\epsilon) \| U'^{i\top} V_{j_2}\|_2^2.
\end{align*}
\end{claim}

\begin{claim}\label{cla:f_GiRUV_is_UiV}
For any $i\in [g_1]$, $j_1\in [n_1]$, $j_2\in [n_2]$, let $G_1^i$ denote the $i$-th row of matrix $G_1\in \mathbb{R}^{g_1 \times k}$. Let $(U \odot V)_{(j_1-1)n_2 +j_2}$ denote the $(j_1-1)n_2 + j_2$-th column of matrix $\mathbb{R}^{k\times n_1 n_2}$. Let $(U'^{i\top})^{j_1}$ denote the $j_1$-th row of matrix $(U'^{i\top})\in \mathbb{R}^{n_1 \times k}$. Let $V_{j_2}$ denote the $j_2$-th column of matrix $V\in \mathbb{R}^{k\times n_2}$. Then, we have
\begin{align*}
G_1^i R^{-1} (U \odot V)_{(j_1-1)n_2 +j_2} = (U'^{i\top})^{j_1} V_{j_2}.
\end{align*}
\end{claim}
\begin{proof}
This follows by,
\begin{align*}
 G_1^i R^{-1} (U \odot V)_{(j_1-1)n_2 +j_2} = G_1^i R^{-1} (U_{j_1} \circ V_{j_2}) = (G_1^i R^{-1} \circ (U_{j_1})^\top ) V_{j_2} =  (U'^{i\top})^{j_1} V_{j_2}.
\end{align*}
\end{proof}

\begin{lemma}\label{lem:fast_tensor_leverage_score_multiple_regression}
Given $A\in \mathbb{R}^{n\times n^2}$, $V,W\in \mathbb{R}^{k\times n}$, for any $\epsilon>0$, there exists an algorithm that runs in $O(n \cdot \poly(k,1/\epsilon))$ time and outputs a diagonal matrix $D\in \mathbb{R}^{n^2 \times n^2}$ with $m=O(k\log k + k/\epsilon)$ nonzero entries such that,
\begin{align*}
\| \wh{U} (V\odot W) - A\|_F^2 \leq (1+\epsilon) \min_{U\in \mathbb{R}^{n\times k}}\| U (V\odot W) - A\|_F^2,
\end{align*}
holds with probability at least $0.999$, where $\wh{U}$ denotes the optimal solution to $\min_U \| U (V \odot W  ) D - A D \|_F^2$.
\end{lemma}
\begin{proof}
This follows by combining Theorem~\ref{thm:f_fast_tensor_leverage_score_general_order}, Corollary~\ref{cor:leverage_score_size_multiple_regression}, and Lemma~\ref{lem:nw14_multiple_regression}.
\end{proof}

\begin{remark}
Replacing Theorem~\ref{thm:f_fast_tensor_leverage_score_general_order} (Algorithm~\ref{alg:f_fast_tensor_leverage_score_general_order}) by Lemma~\ref{lem:f_fast_tensor_leverage_score} (Algorithm~\ref{alg:f_fast_tensor_leverage_score}), we can obtain a slightly different version of Lemma~\ref{lem:fast_tensor_leverage_score_multiple_regression} with $n\poly(\log n,k,1/\epsilon)$ running time, where the dependence on $k$ is better.
\end{remark}

\subsubsection{Input sparsity time algorithm}

\begin{algorithm}[h]\caption{Frobenius Norm CURT Decomposition Algorithm, Input Sparsity Time and Nearly Optimal Number of Samples}\label{alg:f_curt_algorithm_input_sparsity}
\begin{algorithmic}[1]
\Procedure{\textsc{FCURTInputSparsity}}{$A,U_B,V_B,W_B,n,k,\epsilon$} \Comment{Theorem \ref{thm:f_curt_algorithm_input_sparsity}}
\State $d_1\leftarrow d_2 \leftarrow d_3 \leftarrow O(k\log k + k/\epsilon)$.
\State $\epsilon_0\leftarrow 0.01$.
\State Form $B_1 = V_B^\top \odot W_B^\top \in \mathbb{R}^{k\times n^2}$.
\State $D_1\leftarrow$\textsc{FastTensorLeverageScoreGeneralOrder}($V_B^\top, W_B^\top, n,n,k,\epsilon_0,d_1$). \Comment{Algorithm~\ref{alg:f_fast_tensor_leverage_score_general_order}}
\State Form $\wh{U} = A_1 D_1 (B_1 D_1)^\dagger \in \mathbb{R}^{n\times k}$.
\State Form $B_2 = \wh{U}^\top \odot W_B^\top \in \mathbb{R}^{k\times n^2}$.
\State $D_2\leftarrow$\textsc{FastTensorLeverageScoreGeneralOrder}($\wh{U}^\top, W_B^\top,n,n,k,\epsilon_0,d_2$).
\State Form $\wh{V} = A_2 D_2 (B_2 D_2)^\dagger \in \mathbb{R}^{n\times k}$.
\State Form $B_3 = \wh{U}^\top \odot \wh{V}^\top \in \mathbb{R}^{k\times n^2}$.
\State $D_3\leftarrow$\textsc{FastTensorLeverageScoreGeneralOrder}($\wh{U}^\top, \wh{V}^\top,n,n,k,\epsilon_0,d_3$).
\State $C\leftarrow A_1 D_1$, $R\leftarrow A_2 D_2$, $T\leftarrow A_3 D_3$.
\State $U\leftarrow \sum_{i=1}^k ( (B_1 D_1)^\dagger )_i \otimes ( (B_2 D_2)^\dagger )_i \otimes ( (B_3 D_3)^\dagger )_i$.
\State \Return $C$, $R$, $T$ and $U$.
\EndProcedure
\end{algorithmic}
\end{algorithm}

\begin{theorem}\label{thm:f_curt_algorithm_input_sparsity}
Given a $3$rd order tensor $A\in \mathbb{R}^{n\times n \times n}$, let $k\geq 1$, and let $U_B,V_B,W_B\in \mathbb{R}^{n\times k}$ denote a rank-$k$, $\alpha$-approximation to $A$. Then there exists an algorithm which takes $O(\nnz(A) +  n \poly( k, 1/\epsilon) )$ time and outputs three matrices $C\in \mathbb{R}^{n\times c}$ with columns from $A$, $R\in \mathbb{R}^{n\times r}$ with rows from $A$, $T\in \mathbb{R}^{n\times t}$ with tubes from $A$, and a tensor $U\in \mathbb{R}^{c\times r\times t}$ with $\rank(U)=k$ such that $c=r=t=O(k\log k +k/\epsilon)$, and
\begin{align*}
\left\| \sum_{i=1}^c \sum_{j=1}^r \sum_{l=1}^t U_{i,j,l} \cdot C_i \otimes R_j \otimes T_l - A \right\|_F^2 \leq(1+\epsilon) \alpha \min_{\rank-k~A'} \| A' - A\|_F^2
\end{align*}
holds with probability $9/10$.
\end{theorem}

\begin{proof}
We define
\begin{align*}
\OPT : = \min_{\rank-k~A'} \| A' - A\|_F^2.
\end{align*}
We already have three matrices $U_B\in \mathbb{R}^{n\times k}$, $V_B\in \mathbb{R}^{n\times k}$ and $W_B\in \mathbb{R}^{n\times k}$ and these three matrices provide a $\rank$-$k$, $\alpha$-approximation to $A$, i.e.,
\begin{align}\label{eq:f_cur_UBVBWB_minus_A}
\left\| \sum_{i=1}^k ( U_B )_i \otimes (V_B)_i \otimes (W_B)_i - A \right\|_F^2 \leq \alpha \OPT.
\end{align}
Let $B_1 = V_B^\top \odot W_B^\top \in \mathbb{R}^{k\times n^2}$ denote the matrix where the $i$-th row is the vectorization of $(V_B)_i \otimes (W_B)_i$. Let $D_1 \in \mathbb{R}^{n^2 \times n^2}$ be a sampling and rescaling matrix corresponding to sampling by the leverage scores of $B_1^\top$; there are $d_1$ nonzero entries on the diagonal of $D_1$. Let $A_i\in \mathbb{R}^{n\times n^2}$ denote the matrix obtained by flattening $A$ along the $i$-th direction, for each $i\in [3]$.

Define $U^*\in \mathbb{R}^{n\times k}$ to be the optimal solution to $\underset{U\in \mathbb{R}^{n\times k} }{\min} \| U B_1 - A_1\|_F^2$, $\wh{U} = A_1 D_1 (B_1 D_1)^\dagger \in \mathbb{R}^{n\times k}$, and $V_0 \in \mathbb{R}^{n\times k}$ to be the optimal solution to $\underset{V\in \mathbb{R}^{n\times k} }{\min} \| V \cdot  (\wh{U}^\top \odot W_B^\top) - A_2 \|_F^2 $.
Due to Lemma~\ref{lem:fast_tensor_leverage_score_multiple_regression}, if $d_1=O(k\log k+k/\epsilon)$ then with constant probability, we have
\begin{align}\label{eq:f_cur_Uwh_B1_minus_A1}
\| \wh{U} B_1 - A _1 \|_F^2 \leq  \alpha_{D_1} \| U^* B_1 - A_1 \|_F^2.
\end{align}

Recall that $( \wh{U}^\top \odot W_B^\top) \in \mathbb{R}^{k\times n^2}$ denotes the matrix where the $i$-th row is the vectorization of $\wh{U}_i \otimes (W_B)_i$, $\forall i\in [k]$. Now, we can show,
\begin{align}\label{eq:f_cur_V0B2_minus_A2}
\| V_0 \cdot ( \wh{U}^\top \odot W_B^\top ) - A_2 \|_F^2 \leq & ~ \| \wh{U} B_1 - A_1 \|_F^2 & \text{~by~} V_0 = \underset{V\in \mathbb{R}^{n\times k}}{\arg\min} \| V \cdot ( \wh{U}^\top \odot W_B^\top ) - A_2  \|_F^2 \notag \\
\leq & ~ \alpha_{D_1} \| U^* B_1 - A_1 \|_F^2 & \text{~by~Equation~\eqref{eq:f_cur_Uwh_B1_minus_A1}} \notag \\
\leq & ~ \alpha_{D_1} \| U_B B_1 - A_1 \|_F^2 & \text{~by~} U^* = \underset{U\in \mathbb{R}^{n\times k} }{\arg\min} \| U B_1 - A_1 \|_F^2 \notag \\
\leq & ~ \alpha_{D_1} \alpha \OPT. & \text{~by~Equation~\eqref{eq:f_cur_UBVBWB_minus_A}}
\end{align}

We define $B_2= \wh{U}^\top \odot W_B^\top$. Let $D_2\in \mathbb{R}^{n^2 \times n^2}$ be a sampling and rescaling matrix corresponding to the leverage scores of $B_2^\top$. Suppose there are $d_2$ nonzero entries on the diagonal of $D_2$.

Define $V^*\in \mathbb{R}^{n\times k}$ to be the optimal solution to $\min_{V\in \mathbb{R}^{n\times k}} \| V B_2 - A_2 \|_F^2$, $\wh{V}= A_2 D_2 (B_2 D_2)^\dagger \in \mathbb{R}^{n\times k}$, $W_0\in \mathbb{R}^{n\times k}$ to be the optimal solution to $\underset{W\in \mathbb{R}^{n\times k}}{\min} \| W\cdot ( \wh{U}^\top \odot \wh{V}^\top ) - A_3 \|_F^2$, and $V'$ to be the optimal solution to $\underset{V\in \mathbb{R}^{n\times k}}{\min} \|V B_2 D_2 - A_2 D_2\|_F^2$.

Due to Lemma~\ref{lem:fast_tensor_leverage_score_multiple_regression}, with constant probability, we have
\begin{align}\label{eq:f_cur_Vwh_B2_minus_A2}
\| \wh{V} B_2 - A_2 \|_F^2 \leq  \alpha_{D_2} \| V^* B_2 - A_2 \|_F^2.
\end{align}

Recall that $(\wh{U}^\top \odot \wh{V}^\top) \in \mathbb{R}^{k\times n^2}$ denotes the matrix where the $i$-th row is the vectorization of $\wh{U}_i \otimes \wh{V}_i$, $\forall i\in [k]$. Now, we can show,
\begin{align}\label{eq:f_cur_W0B3_minus_A3}
\| W_0 \cdot (\wh{U}^\top \odot \wh{V}^\top ) - A_3 \|_F^2 \leq & ~ \| \wh{V} B_2 - A_2 \|_F^2 & \text{~by~} W_0 = \underset{W\in \mathbb{R}^{n\times k} }{\arg\min} \| W \cdot ( \wh{U}^\top \odot \wh{V}^\top ) - A_3 \|_F^2 \notag \\
\leq & ~ \alpha_{D_2}  \| V^* B_2 - A_2 \|_F^2 & \text{~by~Equation~\eqref{eq:f_cur_Vwh_B2_minus_A2}} \notag \\
\leq & ~ \alpha_{D_2} \| V_0 B_2 - A_2 \|_F^2 & \text{~by~} V^* =\underset{V\in \mathbb{R}^{n\times k}}{\arg\min} \| V B_2 - A_2 \|_F^2 \notag \\
\leq & ~ \alpha_{D_2} \alpha_{D_1} \alpha \OPT. & \text{~by~Equation~\eqref{eq:f_cur_V0B2_minus_A2}}
\end{align}

We define $B_3= \wh{U}^\top \odot \wh{V}^\top$. Let $D_3\in \mathbb{R}^{n^2 \times n^2}$ denote a sampling and rescaling matrix corresponding to sampling by the leverage scores of $B_3^\top$. Suppose there are $d_3$ nonzero entries on the diagonal of $D_3$.

Define $W^*\in \mathbb{R}^{n\times k}$ to be the optimal solution to $\min_{W\in \mathbb{R}^{n\times k}} \| W B_3 - A_3 \|_F^2$, $\wh{W}= A_3 D_3 (B_3 D_3)^\dagger \in \mathbb{R}^{n\times k}$,
and $W'$ to be the optimal solution to $\underset{W\in \mathbb{R}^{n\times k}}{\min} \|W B_3 D_3 - A_3 D_3\|_F^2$.

Due to Lemma~\ref{lem:fast_tensor_leverage_score_multiple_regression} with constant probability, we have
\begin{align}\label{eq:f_cur_Wwh_B3_minus_A3}
\| \wh{W} B_3 - A_3 \|_F^2 \leq  \alpha_{D_3} \| W^* B_3 - A_3 \|_F^2.
\end{align}

Now we can show,
\begin{align*}
\| \wh{W} B_3 - A_3 \|_F^2 \leq & ~ \alpha_{D_3} \| W^* B_3 - A_3 \|_F^2, & \text{~by~Equation~\eqref{eq:f_cur_Wwh_B3_minus_A3}} \\
\leq & ~ \alpha_{D_3} \| W_0 B_3 - A_3 \|_F^2, & \text{~by~}W^* = \underset{W\in \mathbb{R}^{n\times k} }{\arg\min} \| W B_3 - A_3 \|_F^2 \\
\leq & ~  \alpha_{D_3} \alpha_{D_2} \alpha_{D_1} \alpha \OPT. & \text{~by~Equation~\eqref{eq:f_cur_W0B3_minus_A3}}
\end{align*}
This implies,
\begin{align*}
\left\| \sum_{i=1}^k \wh{U}_i \otimes \wh{V}_i \otimes \wh{W}_i - A \right\|_F^2 \leq O(1) \alpha \OPT^2.
\end{align*}
where $\wh{U} = A_1 D_1 (B_1 D_1)^\dagger$, $\wh{V} = A_2D_2 (B_2 D_2)^\dagger$, $\wh{W}=A_3D_3 (B_3 D_3)^\dagger$.

By Lemma~\ref{lem:fast_tensor_leverage_score_multiple_regression}, we need to set $d_1=d_2=d_3=O(k\log k+k/\epsilon)$. Note that $B_1= (V_B^\top \odot W_B^\top)$. Thus $D_1$ can be found in $n\cdot \poly(k,1/\epsilon)$ time. Because $D_1$ has a small number of nonzero entries on the diagonal, we can compute $B_1 D_1$ quickly without explicitly writing down $B_1$. Also $A_1 D_1$ can be computed in $\nnz(A)$ time. Using $(A_1D_1)$ and $(B_1D_1)$, we can compute $\wh{U}$ in $n\poly(k,1/\epsilon)$ time. In a similar way, we can compute $B_2$, $D_2$, $B_3$, and $D_3$. Since tensor $U$ is constructed based on three $\poly(k,1/\epsilon)$ size matrices, $(B_1D_1)^\dagger$, $(B_2D_2)^\dagger$, and $(B_3 D_3)^\dagger$, the overall running time is $O(\nnz(A) + n \poly(k,1/\epsilon) )$
\end{proof}

\subsubsection{Optimal sample complexity algorithm}

\begin{algorithm}[h]\caption{Frobenius Norm CURT Decomposition Algorithm, Optimal Sample Complexity}\label{alg:f_curt_algorithm_optimal_samples}
\begin{algorithmic}[1]
\Procedure{\textsc{FCURTOptimalSamples}}{$A,U_B,V_B,W_B,n,k$} \Comment{Theorem \ref{thm:f_curt_algorithm_optimal_samples}}
\State $d_1\leftarrow d_2 \leftarrow d_3 \leftarrow O( k/\epsilon)$.
\State Form $B_1 = V_B^\top\odot W_B^\top \in \mathbb{R}^{k\times n^2}$.
\State $D_1\leftarrow$\textsc{GeneralizedMatrixRowSubsetSelection}($A_1^\top,B_1^\top, n^2,n,k,\epsilon$). \Comment{Algorithm~\ref{alg:f_generalized_matrix_row}}
\State Let $d_1$ denote the number of nonzero entries in $D_1$. \Comment{$d_1=O(k/\epsilon)$}
\State Form $\wh{U} = A_1 D_1 (B_1 D_1)^\dagger \in \mathbb{R}^{n\times k}$.
\State Form $B_2 = \wh{U}^\top \odot W_B^\top \in \mathbb{R}^{k\times n^2}$.
\State $D_2\leftarrow$\textsc{GeneralizedMatrixRowSubsetSelection}($A_2^\top, B_2^\top,n^2,n,k,\epsilon$). \Comment{Algorithm~\ref{alg:f_generalized_matrix_row}}
\State Let $d_2$ denote the number of nonzero entries in $D_2$. \Comment{$d_2=O(k/\epsilon)$}
\State Form $\wh{V} = A_2 D_2 (B_2 D_2)^\dagger \in \mathbb{R}^{n\times k}$.
\State Form $B_3 = \wh{U}^\top \odot \wh{V}^\top \in \mathbb{R}^{k\times n^2}$.
\State $D_3\leftarrow$\textsc{GeneralizedMatrixRowSubsetSelection}($A_3^\top, B_3^\top,n^2,n,k,\epsilon$). \Comment{Algorithm~\ref{alg:f_generalized_matrix_row}}
\State $d_3$ denote the number of nonzero entries in $D_3$. \Comment{$d_3=O(k/\epsilon)$}
\State $C\leftarrow A_1 D_1$, $R\leftarrow A_2 D_2$, $T\leftarrow A_3 D_3$.
\State $U\leftarrow \sum_{i=1}^k ( (B_1 D_1)^\dagger )_i \otimes ( (B_2 D_2)^\dagger )_i \otimes ( (B_3 D_3)^\dagger )_i$.
\State \Return $C$, $R$, $T$ and $U$.
\EndProcedure
\end{algorithmic}
\end{algorithm}

\begin{theorem}\label{thm:f_curt_algorithm_optimal_samples}
Given a $3$rd order tensor $A\in \mathbb{R}^{n\times n \times n}$, let $k\geq 1$, and let $U_B,V_B,W_B\in \mathbb{R}^{n\times k}$ denote a rank-$k$, $\alpha$-approximation to $A$. Then there exists an algorithm which takes $O(\nnz(A) \log n +  n^2 \poly(\log n, k, 1/\epsilon) )$ time and outputs three matrices: $C\in \mathbb{R}^{n\times c}$ with columns from $A$, $R\in \mathbb{R}^{n\times r}$ with rows from $A$, $T\in \mathbb{R}^{n\times t}$ with tubes from $A$, and a tensor $U\in \mathbb{R}^{c\times r\times t}$ with $\rank(U)=k$ such that $c=r=t=O(k/\epsilon)$, and
\begin{align*}
\left\| \sum_{i=1}^c \sum_{j=1}^r \sum_{l=1}^t U_{i,j,l} \cdot C_i \otimes R_j \otimes T_l - A \right\|_F^2 \leq(1+\epsilon) \alpha \min_{\rank-k~A'} \| A' - A\|_F^2
\end{align*}
holds with probability $9/10$.
\end{theorem}
\begin{proof}
The proof is almost the same as the proof of Theorem~\ref{thm:f_curt_algorithm_input_sparsity}. The only difference is that instead of using Theorem~\ref{lem:fast_tensor_leverage_score_multiple_regression}, we use Theorem~\ref{thm:f_generalized_matrix_row}.
\end{proof}

\subsection{Face-based selection and decomposition}
Previously we provided column-based tensor CURT algorithms, which are algorithms that can select a subset of columns from each of the three dimensions. Here we provide two face-based tensor CURT decomposition algorithms. The first algorithm runs in polynomial time and is a bicriteria algorithm (the number of samples is $\poly(k/\epsilon)$). The second algorithm needs to start with a rank-$k$ $(1+O(\epsilon))$-approximate solution, which we then show how to combine with our previous algorithm. Both of our algorithms are able to select a subset of column-row faces, a subset of row-tube faces and a subset of column-tube faces. The second algorithm is able to output $U$, but the first algorithm is not.

\subsubsection{Column-row, column-tube, row-tube face subset selection}

\begin{algorithm}[h]\caption{Frobenius Norm Tensor Column-row, Row-tube and Tube-column Face Subset Selection}\label{alg:f_fast_face_curt_without_u}
\begin{algorithmic}[1]
\Procedure{\textsc{FFaceCRTSelection}}{$A,n,k,\epsilon$} \Comment{Theorem \ref{thm:f_fast_face_curt_without_u}}
\State $s_1 \leftarrow s_2 \leftarrow O(k/\epsilon)$.
\State Choose a Gaussian matrix $S_1$ with $s_1$ columns. \Comment{Definition~\ref{def:fast_gaussian_transform}}
\State Choose a Gaussian matrix $S_2$ with $s_2$ columns. \Comment{Definition~\ref{def:fast_gaussian_transform}}
\State Form matrix $V_3$ by setting the $(i,j)$-th column to be $ (A_2 S_2)_j$.
\State $D_3\leftarrow$\textsc{GeneralizedMatrixRowSubsetSelection}($A_2$,$V_3$,$n$,$n^2$,$s_1 s_2$,$\epsilon$). \Comment{Algorithm \ref{alg:f_generalized_matrix_row}}
\State Let $d_3$ denote the number of nonzero entries in $D_3$. \Comment{$d_3= O(s_1 s_2/\epsilon)$}
\State Form matrix $U_2$ by setting the $(i,j)$-th column to be $( A_1 S_1)_i$.
\State $D_2\leftarrow$\textsc{GeneralizedMatrixRowSubsetSelection}($A_1$,$U_2$,$n$,$n^2$,$s_1 s_2$,$\epsilon$).
\State Let $d_2$ denote the number of nonzero entries in $D_2$.  \Comment{$d_2= O(s_1 s_2/\epsilon)$}
\State Form matrix $W_1$ by setting the $(i,j)$-th column to be  $(A(I,D_3,I)_3)_j$.
\State $D_1\leftarrow$\textsc{GeneralizedMatrixRowSubsetSelection}($A_3$,$W_1$,$n$,$n^2$,$s_1s_2$,$\epsilon$). 
\State Let $d_1$ denote the number of nonzero entries in $D_1$.  \Comment{$d_1= O(s_1 s_2/\epsilon)$} 
\State $T\leftarrow A(I,I,D_1)$, $C\leftarrow A(D_2,I,I)$, and $R\leftarrow A(I,D_3,I)$.
\State \Return $C$, $R$ and $T$.
\EndProcedure
\end{algorithmic}
\end{algorithm}

\begin{theorem}\label{thm:f_fast_face_curt_without_u}
Given a $3$rd order tensor $A\in \mathbb{R}^{n\times n\times n}$, for any $k\geq 1$, there exists an algorithm which takes $O(\nnz(A)) \log n + n^2 \poly(\log n, k,1/\epsilon)$ time and outputs three tensors : a subset $C\in \mathbb{R}^{c\times n\times n}$ of row-tube faces of $A$, a subset $R\in \mathbb{R}^{n\times r \times n}$ of column-tube faces of $A$, and a subset $T\in \mathbb{R}^{n\times n \times t}$ of column-row faces of $A$, where $c=r=t=\poly(k,1/\epsilon)$, and for which there exists a tensor $U\in \mathbb{R}^{tn \times cn \times rn}$ for which
\begin{align*}
\| U(T_1,C_2,R_3)-A \|_F^2 \leq (1+\epsilon) \min_{\rank-k~A'} \| A' - A \|_F^2,
\end{align*}
or equivalently,
\begin{align*}
\left\| \sum_{i=1}^{tn} \sum_{j=1}^{cn} \sum_{l=1}^{rn}  U_{i,j,l} \cdot (T_1)_i \otimes (C_2)_j \otimes (R_3)_l  - A \right\|_F^2 \leq(1+\epsilon)  \min_{\rank-k~A'} \| A' - A\|_F^2.
\end{align*}
\end{theorem}

\begin{proof}
We fix $V^*\in \mathbb{R}^{n\times k}$ and $W^*\in \mathbb{R}^{n\times k}$. We define $Z_1\in \mathbb{R}^{k\times n^2}$ where the $i$-th row of $Z_1$ is the vector $V_i \otimes W_i$. Choose a sketching (Gaussian) matrix $S_1 \in \mathbb{R}^{n^2 \times s_1}$ (Definition~\ref{def:fast_gaussian_transform}), and let $\wh{U} = A_1 S_1 (Z_1 S_1)^\dagger \in \mathbb{R}^{n\times k}$. Following a similar argument as in the previous theorem, we have
\begin{align*}
\| \wh{U} Z_1 - A_1 \|_F^2 \leq (1+\epsilon) \OPT.
\end{align*}
We fix $\wh{U}$ and $W^*$. We define $Z_2\in \mathbb{R}^{k\times n^2}$ where the $i$-th row of $Z_2$ is the vector $\wh{U}_i \otimes W^*_i$. Choose a sketching (Gaussian) matrix $S_2\in \mathbb{R}^{n^2\times s_2}$ (Definition~\ref{def:fast_gaussian_transform}), and let $\wh{V} = A_2 S_2 (Z_2 S_2)^\dagger \in \mathbb{R}^{n\times k}$. Following a similar argument as in the previous theorem, we have
\begin{align*}
\| \wh{V} Z_2 - A_2 \|_F^2 \leq (1+\epsilon)^2 \OPT.
\end{align*}

We fix $\wh{U}$ and $\wh{V}$. Note that $\wh{U}=A_1 S_1 (Z_1 S_1)^\dagger$ and $\wh{V}=A_2 S_2 (Z_2 S_2)^\dagger$. We define $Z_3\in\mathbb{R}^{k\times n^2}$ such that the $i$-th row of $Z_3$ is the vector $\wh{U}_i \otimes \wh{V}_i$. Let $z_3 = s_1 \cdot s_2$. We define $Z'_3\in \mathbb{R}^{z_3\times n^2}$ such that, $\forall i\in [s_1], \forall j\in [s_2]$, the $i+(j-1) s_1$-th row of $Z'_3$ is the vector $ (A_1 S_1)_i \otimes (A_2 S_2)_j$.

We define $U_3\in \mathbb{R}^{n\times z_3}$ to be the matrix where the $i+(j-1) s_1$-th column is $(A_1 S_1)_i$ and $V_3 \in \mathbb{R}^{n\times z_3}$ to be the matrix where the $i+(j-1) s_1$-th column is $(A_2 S_2)_j$. Then $Z_3' = (U_3^\top \odot V_3^\top)$.

 We first have,
\begin{align*}
\min_{W \in \mathbb{R}^{n\times k}, X\in \mathbb{R}^{k \times z_3} }\| W X Z_3' - A_3 \|_F^2 \leq \min_{W \in \mathbb{R}^{n\times k} } \| W Z_3 - A_3 \|_F^2\leq (1+\epsilon)^2 \OPT.
\end{align*}

Now consider the following objective function,
\begin{align*}
\min_{W\in \mathbb{R}^{n\times z_3}} \| V_3 \cdot( W^\top \odot U_3^\top ) - A_2 \|_F^2.
\end{align*}
Let $D_3$ denote a sampling and rescaling diagonal matrix according to $V_1\in \mathbb{R}^{n\times z_3}$, let $d_3$ denote the number of nonzero entries of $D_3$. Then we have
\begin{align*}
& ~\min_{W\in \mathbb{R}^{n\times z_3}} \| D_3 V_3 \cdot( W^\top \odot U_3^\top ) - D_3 A_2 \|_F^2 \\
= & ~ \min_{W\in \mathbb{R}^{n\times z_3}} \|  U_3 \otimes (D_3 V_3) \otimes W  - A(I,D_3,I) \|_F^2 \\
= & ~ \min_{W\in \mathbb{R}^{n\times z_3}} \| W\cdot ( U_3^\top \odot (D_3 V_3)^\top )  - (A(I,D_3,I))_3 \|_F^2,
\end{align*}
where the first equality follows by retensorizing the objective function, and the second equality follows by flattening the tensor along the third dimension.

Let $\ov{Z}_3$ denote $( U_3^\top \odot (D_3 V_3)^\top ) \in \mathbb{R}^{z_3 \times nd_3}$ and $W'=(A(I,D_3,I))_3  \in \mathbb{R}^{n \times nd_3}$.
Using Theorem~\ref{thm:f_generalized_matrix_row}, we can find a diagonal matrix $D_3\in \mathbb{R}^{n^2\times n^2}$ with $d_3 = O(z_3/\epsilon) = O(k^2/\epsilon^3)$ nonzero entries such that
\begin{align*}
\| U_3 \otimes V_3 \otimes (W'Z_3^\dagger) - A \|_F^2 \leq (1+\epsilon)^3 \OPT.
\end{align*}

We define $U_2=U_3\in \mathbb{R}^{n\times z_2}$ with $z_2=z_3$. We define $W_2 = W' \ov{Z}_3^\dagger\in \mathbb{R}^{n\times z_2}$ with $z_2=z_3$.
We consider,
\begin{align*}
\min_{V\in \mathbb{R}^{n\times z_2}} \| U_2 \cdot (V^\top \odot W_2^\top ) -  A_1\|_F^2.
\end{align*}
Let $D_2$ denote a sampling and rescaling matrix according to $U_2$, and let $d_2$ denote the number of nonzero entries of $D_2$. Then, we have
\begin{align*}
 & ~\min_{V\in \mathbb{R}^{n\times z_2}} \| D_2 U_2 \cdot (V^\top \odot W_2^\top) - D_2 A_1\|_F^2\\
 = & ~ \min_{V\in \mathbb{R}^{n\times z_2}} \| D_2 U_2 \otimes V \otimes W_2  - A(D_2,I,I)\|_F^2\\
 = & ~ \min_{V\in \mathbb{R}^{n\times z_2}} \| V\cdot (  W_2^\top \odot  (D_2 U_2)^\top ) - (A(D_2,I,I))_2\|_F^2,
\end{align*}
where the first equality follows by retensorizing the objective function, and the second equality follows by flattening the tensor along the second dimension.

Let $\ov{Z}_2$ denote $(  W_2^\top \odot  (D_2 U_2)^\top ) \in \mathbb{R}^{z_2 \times nd_2}$ and $V' = (A(D_2,I,I))_2 \in \mathbb{R}^{n\times nd_2}$. Using Theorem~\ref{thm:f_generalized_matrix_row}, we can find a diagonal matrix $D_2 \in \mathbb{R}^{n^2 \times n^2}$ with $d_2= O(z_2/\epsilon)$
nonzero entries such that
\begin{align*}
\| U_2 \otimes (V' \ov{Z}_2^\dagger) \otimes W_2 -A \|_F^2 \leq (1+\epsilon)^4 \OPT.
\end{align*}

We define $W_1=W_2\in \mathbb{R}^{n\times z_1}$ with $z_1=z_2$, and define $V_1=(V' \ov{Z}_2^\dagger) \in \mathbb{R}^{n\times z_1} $ with $z_1 =z_2$.

Let $D_1$ denote a sampling and rescaling matrix according to $W_1$, and let $d_1$ denote the number of nonzero entries of $D_1$. Then we have
\begin{align*}
& ~\min_{U\in \mathbb{R}^{n\times z_1}} \| D_1 W_1 \cdot (U^\top \odot V_1^\top) - D_1 A_3 \|_F^2\\
= & ~ \min_{U\in \mathbb{R}^{n\times z_1}} \| U \otimes V_1 \otimes  (D_1 W_1) -  A(I,I,D_1) \|_F^2\\
= & ~ \min_{U\in \mathbb{R}^{n\times z_1}} \| U \cdot( V_1^\top \odot  (D_1 W_1)^\top) -  A(I,I,D_1)_1 \|_F^2
\end{align*}
where the first equality follows by unflattening the objective function, and second equality follows by flattening the tensor along the first dimension.

Let $\ov{Z}_1$ denote $( V_1^\top \odot  (D_1 W_1)^\top)\in \mathbb{R}^{z_1\times nd_1}$, and $U'=A(I,I,D_1)_1\in \mathbb{R}^{n\times nd_1}$.
Using Theorem~\ref{thm:f_generalized_matrix_row}, we can find a diagonal matrix $D_1 \in \mathbb{R}^{n^2 \times n^2}$ with $d_1= O(z_1/\epsilon)$
nonzero entries such that
\begin{align*}
\| (U' \ov{Z}_1^\dagger) \otimes (V_1) \otimes W_1 -A \|_F^2 \leq (1+\epsilon)^5 \OPT,
\end{align*}
which means,
\begin{align*}
\| (U' \ov{Z}_1^\dagger) \otimes (V'\ov{Z}_2^\dagger) \otimes (W'\ov{Z}_3^\dagger) -A \|_F^2 \leq (1+\epsilon)^5 \OPT.
\end{align*}

Putting $U',V',W'$ together completes the proof.
\end{proof}

\begin{corollary}\label{cor:f_fast_face_curt_without_u}
Given a $3$rd order tensor $A\in \mathbb{R}^{n\times n\times n}$, for any $k\geq 1$, there exists an algorithm which takes $O(\nnz(A)) + n^2 \poly( k,1/\epsilon)$ time and outputs three tensors : a subset $C\in \mathbb{R}^{c\times n\times n}$ of row-tube faces of $A$, a subset $R\in \mathbb{R}^{n\times r \times n}$ of column-tube faces of $A$, and a subset $T\in \mathbb{R}^{n\times n \times t}$ of column-row faces of $A$, where $c=r=t=\poly(k,1/\epsilon)$, so that there exists a tensor $U\in \mathbb{R}^{tn \times cn \times rn}$ for which
\begin{align*}
\| U(T_1,C_2,R_3)-A \|_F^2 \leq (1+\epsilon) \min_{\rank-k~A'} \| A' - A \|_F^2,
\end{align*}
or equivalently,
\begin{align*}
\left\| \sum_{i=1}^{tn} \sum_{j=1}^{cn} \sum_{l=1}^{rn}  U_{i,j,l} \cdot (T_1)_i \otimes (C_2)_j \otimes (R_3)_l  - A \right\|_F^2 \leq(1+\epsilon)  \min_{\rank-k~A'} \| A' - A\|_F^2
\end{align*}
\end{corollary}
\begin{proof}
If we allow a $\poly(k/\varepsilon)$ factor increase in running time and a $\poly(k/\varepsilon)$ factor increase in the number of faces selected, then instead of using generalized row subset selection, which has running time depending on $\log n$, we can use the technique in Section~\ref{sec:f_matrix_cur} to avoid the $\log n$ factor.
\end{proof}

\subsubsection{CURT decomposition}

\begin{algorithm}[h]\caption{Frobenius Norm (Face-based) CURT Decomposition Algorithm, Optimal Sample Complexity}\label{alg:f_curt_algorithm_optimal_samples}
\begin{algorithmic}[1]
\Procedure{\textsc{FFaceCURTDecomposition}}{$A,U_B,V_B,W_B,n,k$} \Comment{Theorem \ref{thm:f_face_curt_algorithm_optimal_samples}}
\State $D_1\leftarrow$\textsc{GeneralizedMatrixRowSubsetSelection}($A_3,W_B,n,n^2,k,\epsilon$). \Comment{Algorithm~\ref{alg:f_generalized_matrix_row}, the number of nonzero entries is $d_1=O(k/\epsilon)$}
\State Form $Z_1 = V_B^\top \odot (D_1 W_B)^\top$.
\State Form $\wh{U} = (A(I,I,D_1))_1 Z_1^\dagger \in \mathbb{R}^{n\times k}$.
\State $D_2\leftarrow$\textsc{GeneralizedMatrixRowSubsetSelection}($A_1,\wh{U},n,n^2,k,\epsilon$). \Comment{The number of nonzero entries is $d_2=O(k/\epsilon)$}
\State Form $Z_2 = (  W_B^\top \odot (D_2 \wh{U}) )$.
\State Form $\wh{V} = (A(D_2,I,I))_2  Z_2^\dagger \in \mathbb{R}^{n\times k}$.
\State $D_3\leftarrow$\textsc{GeneralizedMatrixRowSubsetSelection}($A_2,\wh{V},n,n^2,k,\epsilon$). \Comment{The number of nonzero entries is $d_3=O(k/\epsilon)$}
\State Form $Z_3 = \wh{U}^\top \odot (D_3 \wh{V})^\top$.
\State Form $\wh{W} = (A(I,D_3,I))_3 (Z_3)^\dagger \in \mathbb{R}^{n\times k}$.
\State $T\leftarrow A(I,I,D_1)$, $C \leftarrow A(D_2,I,I)$, $R\leftarrow A(I,D_3,I)$.
\State $U\leftarrow \sum_{i=1}^k ( (Z_1)^\dagger )_i \otimes ( (Z_2)^\dagger )_i \otimes ( (Z_3)^\dagger )_i$.
\State \Return $C$, $R$, $T$ and $U$.
\EndProcedure
\end{algorithmic}
\end{algorithm}

\begin{theorem}\label{thm:f_face_curt_algorithm_optimal_samples}
Given a $3$rd order tensor $A\in \mathbb{R}^{n\times n \times n}$, let $k\geq 1$, and let $U_B,V_B,W_B\in \mathbb{R}^{n\times k}$ denote a rank-$k$, $\alpha$-approximation to $A$. Then there exists an algorithm which takes $O(\nnz(A)) \log n +  n^2 \poly(\log n, k, 1/\epsilon)$ time and outputs three tensors: $C\in \mathbb{R}^{c\times n \times n}$ with row-tube faces from $A$, $R\in \mathbb{R}^{n\times r \times n}$ with colum-tube faces from $A$, $T\in \mathbb{R}^{n\times n\times t}$ with column-row faces from $A$, and a (factorization of a) tensor $U\in \mathbb{R}^{tn\times cn\times rn}$ with $\rank(U)=k$ for which $c=r=t=O(k/\epsilon)$ and
\begin{align*}
\| U(T_1,C_2,R_3) - A \|_F^2 \leq (1+\epsilon) \alpha \min_{\rank-k~A'} \| A' - A \|_F^2,
\end{align*}
or equivalently,
\begin{align*}
\left\| \sum_{i=1}^{tn} \sum_{j=1}^{cn} \sum_{l=1}^{rn}  U_{i,j,l} \cdot (T_1)_i \otimes (C_2)_j \otimes (R_3)_l  - A \right\|_F^2 \leq(1+\epsilon) \alpha \min_{\rank-k~A'} \| A' - A\|_F^2
\end{align*}
holds with probability $9/10$.
\end{theorem}
\begin{proof}
We already have three matrices $U_B\in \mathbb{R}^{n\times k}$, $V_B\in \mathbb{R}^{n\times k}$ and $W_B\in \mathbb{R}^{n\times k}$ and these three matrices provide a $\rank$-$k$, $\alpha$-approximation to $A$, i.e.,
\begin{align*}
\| U_B \otimes V_B \otimes W_B - A \|_F^2 \leq \alpha \underbrace{ \min_{\rank-k~A'} \| A' - A \|_F^2 }_{\OPT}.
\end{align*}

We can consider the following problem,
\begin{align*}
\min_{U \in \mathbb{R}^{n\times k} } \| W_B \cdot (U^\top \odot V_B^{\top}) - A_3 \|_F^2.
\end{align*}
Let $D_1$ denote a sampling and rescaling diagonal matrix according to $W_B$, and let $d_1$ denote the number of nonzero entries of $D_1$. Then we have
\begin{align*}
 & \min_{U \in \mathbb{R}^{n\times k} } \| (D_1 W_B) \cdot  (U^\top \odot V_B^{\top} ) - D_1 A_3 \|_F^2 \\
= & \min_{U \in \mathbb{R}^{n\times k} } \| U \otimes V_B \otimes D_1 W_B - A(I,I,D_1) \|_F^2 \\
= & \min_{U \in \mathbb{R}^{n\times k} } \| U \cdot ( V_B^\top \odot (D_1 W_B)^\top) - ( A(I,I,D_1) )_1 \|_F^2,
\end{align*}
where the first equality follows by retensorizing the objective function, and the second equality follows by flattening the tensor along the first dimension. Let $Z_1$ denote $V_B^\top \odot (D_1 W_B)^\top \in \mathbb{R}^{k \times n d_1}$, and define $\wh{U} = (A(I,I,D_1))_1 Z_1^\dagger \in \mathbb{R}^{n\times k}$. Then we have
\begin{align*}
\| \wh{U} \otimes V_B \otimes W_B -A \|_F^2 \leq (1+\epsilon) \alpha \OPT.
\end{align*}
In the second step, we fix $\wh{U}$ and $W_B$, and consider the following objective function,
\begin{align*}
\min_{V\in \mathbb{R}^{n\times k}} \| \wh{U} \cdot (V^\top \odot W_B) - A_1 \|_F^2.
\end{align*}
Let $D_2$ denote a sampling and rescaling matrix according to $\wh{U}$, and let $d_2$ denote the number of nonzero entries of $D_2$. Then we have,
\begin{align*}
 & ~\min_{V\in \mathbb{R}^{n\times k}} \| (D_2 \wh{U}) \cdot (V^\top \odot W_B^\top) - D_2 A_1 \|_F^2 \\
= & ~ \min_{V\in \mathbb{R}^{n\times k}} \| (D_2 \wh{U}) \otimes V  \otimes W_B - A(D_2,I,I) \|_F^2 \\
= & ~ \min_{V\in \mathbb{R}^{n\times k}} \| V \cdot  (W_B^\top \odot (D_2 \wh{U})^\top   )  - (A(D_2,I,I))_2 \|_F^2,
\end{align*}
where the first equality follows by unflattening the objective function, and the second equality follows by flattening the tensor along the second dimension. Let $Z_2$ denote $( W_B^\top \odot (D_2 \wh{U})^\top )\in \mathbb{R}^{k \times nd_2}$, and define $\wh{V} =(A(D_2,I,I))_2 (Z_2)^\dagger \in \mathbb{R}^{n\times k} $. Then we have,
\begin{align*}
\| \wh{U} \otimes \wh{V} \otimes W_B - A \|_F^2 \leq (1+\epsilon)^2 \alpha \OPT.
\end{align*}
In the third step, we fix $\wh{U}$ and $\wh{V}$, and consider the following objective function,
\begin{align*}
\min_{W\in \mathbb{R}^{n\times k}} \| \wh{V} \cdot ( W \odot \wh{U} ) - A_2 \|_F^2.
\end{align*}
Let $D_3$ denote a sampling and rescaling matrix according to $\wh{V}$, and let $d_3$ denote the number of nonzero entries of $D_3$. Then we have,
\begin{align*}
& ~ \min_{W\in \mathbb{R}^{n\times k} } \| (D_3 \wh{V}) \cdot (W^\top \odot \wh{U}^\top) - D_3 A_2 \|_F^2 \\
= & ~ \min_{W\in \mathbb{R}^{n\times k} } \| \wh{U} \otimes (D_3 \wh{V}) \otimes W  -  A(I,D_3,I) \|_F^2 \\
= & ~ \min_{W\in \mathbb{R}^{n\times k} } \| W \cdot (  \wh{U}^\top \odot (D_3 \wh{V})^\top )  -  (A(I,D_3,I))_3 \|_F^2,
\end{align*}
where the first equality follows by retensorizing the objective function, and the second equality follows by flattening the tensor along the third dimension. Let $Z_3$ denote $(  \wh{U}^\top \odot (D_3 \wh{V})^\top ) \in \mathbb{R}^{k\times nd_3}$, and define $\wh{W} = (A(I,D_3,I))_3  (Z_3)^\dagger$. Putting it all together, we have,
\begin{align*}
\| \wh{U} \otimes \wh{V} \otimes \wh{W} - A\|_F^2 \leq (1+\epsilon)^3 \alpha \OPT.
\end{align*}
This implies
\begin{align*}
\| (A(I,I,D_1))_1 Z_1^\dagger \otimes (A(D_2,I,I))_2 Z_2^\dagger \otimes (A(I,D_3,I))_3 Z_3^\dagger - A\|_F^2 \leq (1+\epsilon)^3 \alpha \OPT.
\end{align*}
\end{proof}

\subsection{Solving small problems}\label{sec:f_solving_small_problems}

\begin{theorem}\label{thm:f_solving_small_problems}
Let $\max_{i} \{t_i, d_i\} \leq n$. Given a $t_1 \times t_2 \times t_3$ tensor $A$ and three matrices: a $t_1 \times d_1$ matrix $T_1$, a $t_2 \times d_2$ matrix $T_2$, and a $t_3 \times d_3$ matrix $T_3$, if for any $\delta > 0$ there exists a solution to
\begin{align*}
\min_{X_1,X_2,X_3} \left\| \sum_{i=1}^k (T_1 X_1)_i \otimes (T_2 X_2)_i \otimes (T_3 X_3)_i - A \right\|_F^2 := \OPT,
\end{align*}
and each entry of $X_i$ can be expressed using $O(n^\delta)$ bits, then there exists an algorithm that takes $ n^{O(\delta)} \cdot 2^{ O( d_1 k+d_2 k+d_3 k)}$ time and outputs three matrices: $\wh{X}_1$, $\wh{X}_2$, and $\wh{X}_3$ such that $\| (T_1 \wh{X}_1)\otimes (T_2 \wh{X}_2) \otimes (T_3\wh{X}_3) - A\|_F^2 =\OPT$.
\end{theorem}
\begin{proof}
 For each $i\in [3]$, we can create $t_i\times d_i$ variables to represent matrix $X_i$. Let $x$ denote this list of variables. Let $B$ denote tensor $\sum_{i=1}^k (T_1 X_1)_i \otimes (T_2X_2)_i\otimes (T_3X_3)_i$ and let $B_{i,j,l}(x)$ denote an entry of tensor $B$ (which can be thought of as a polynomial written in terms of $x$). Then we can write the following objective function,
\begin{align*}
\min_{x} \sum_{i=1}^{t_1} \sum_{j=1}^{t_2} \sum_{l=1}^{t_3} ( B_{i,j,l}(x) -A_{i,j,l} )^2.
\end{align*}
We slightly modify the above objective function to obtain a new objective function,
\begin{align*}
\min_{x,\sigma} & ~ \sum_{i=1}^{t_1} \sum_{j=1}^{t_2} \sum_{l=1}^{t_3}  (B_{i,j,l}(x) -A_{i,j,l})^2 , \\
\text{s.t.} & ~ \| x \|_2^2  \leq 2^{O(n^\delta)},
\end{align*}
where the last constraint is unharmful, because there exists a solution that can be written using $O(n^\delta)$ bits. Note that the number of inequality constraints in the above system is $O(1)$, the degree is $O(1)$, and the number of variables is $v=(d_1k+d_2k+d_3k)$. Thus by Theorem~\ref{thm:minimum_positive}, the minimum nonzero cost is at least
\begin{align*}
(2^{O(n^\delta)} )^{-2^{{O} ( v ) }}.
\end{align*}
It is clear that the upper bound on the cost is at most $2^{O(n^\delta)}$. Thus the number of binary search steps is at most $\log (2^{O(n^\delta)} ) 2^{{O}(v)}$. In each step of the binary search, we need to choose a cost $C$ between the lower bound and the upper bound, and write down the polynomial system,
\begin{align*}
 & ~ \sum_{i=1}^{t_1} \sum_{j=1}^{t_2} \sum_{l=1}^{t_3}  (B_{i,j,l}(x) -A_{i,j,l})^2 \leq C, \\
& ~ \| x \|_2^2  \leq 2^{O(n^\delta)}.
\end{align*}
Using Theorem~\ref{thm:decision_solver}, we can determine if there exists a solution to the above polynomial system. Since the number of variables is $v$, and the degree is $O(1)$, the number of inequality constraints is $O(1)$. Thus, the running time is
\begin{align*}
\poly( \text{bitsize}) \cdot (\#\constraints \cdot \degree)^{\#\variables} = n^{O(\delta)} 2^{O(v)}.
\end{align*}
\end{proof}

\subsection{Extension to general $q$-th order tensors}\label{sec:f_general_order}
This section provides the details for our extensions from $3$rd order tensors to general $q$-th order tensors. In most practical applications, the order $q$ is a constant. Thus, to simplify the analysis, we use $O_q(\cdot)$ to hide dependencies on $q$.

\subsubsection{Fast sampling of columns according to leverage scores, implicitly}\label{sec:f_fast_tensor_leverage_score_general_order}

This section explains an algorithm that is able to sample from the leverage scores from the $\odot$ product of $q$ matrices $U_1,U_2,\cdots,U_q$ without explicitly writing down $U_1 \odot U_2 \odot \cdots U_q$. To build this algorithm we combine \textsc{TensorSketch}, some ideas from \cite{dmmw12}, and some techniques from \cite{ako11,mw10}. Finally, we improve the running time for sampling columns according to the leverage scores from $\poly(n)$ to $\wt{O}(n)$. Given $q$ matrices $U_1, U_2,\cdots, U_q$, with each such matrix $U_i$ having size $k \times n_i$, we define $A\in \mathbb{R}^{k\times \prod_{i=1}^q n_i}$ to be the matrix where the $i$-th row of $A$ is the vectorization of $U_1^i \otimes U_2^i \otimes \cdots \otimes U_q^i$, $\forall i \in [k]$. Na\"{i}vely, in order to sample $\poly(k,1/\epsilon)$ rows from $A$ according to the leverage scores, we need to write down $\prod_{i=1}^q n_i$ leverage scores. This approach will take at least $\prod_{i=1}^q n_i$ running time. In the remainder of this section, we will explain how to do it in $O_q(n \cdot \poly(k,1/\epsilon))$ time for any constant $p$, and $\max_{i\in[q]} n_i \leq n$.

\begin{algorithm}[!t]\caption{Fast Tensor Leverage Score Sampling, for General $q$-th Order }\label{alg:f_fast_tensor_leverage_score_general_order}
\begin{algorithmic}[1]
\Procedure{\textsc{FastTensorLeverageScoreGeneralOrder}}{$\{U_i\}_{i\in [q]},\{n_i\}_{i\in[q]},k,\epsilon,R_{\text{samples}}$} \Comment{Theorem \ref{thm:f_fast_tensor_leverage_score_general_order}}
\State $s_1 \leftarrow \poly(k,1/\epsilon)$.
\State Choose $\Pi_0 , \Pi_1 \in \mathbb{R}^{n_1 n_2\cdots n_q \times s_1}$ to each be a \textsc{TensorSketch}. \Comment{Definition~\ref{def:tensor_sketch}}
\State Compute $R^{-1} \in \mathbb{R}^{k\times k}$ by using $(U_1 \odot U_2 \odot \cdots \odot U_q) \Pi_0$. \Comment{$U_i \in \mathbb{R}^{k\times n_i},\forall i\in [q]$}
\State $V_0 \leftarrow  R^{-1}$, $n_0\leftarrow k$.
\For{$i=1\to [n_0]$}
	\State $\alpha_i \leftarrow \| (V_0)^i ( (U_1 \odot U_2 \odot \cdots \odot U_q) \Pi_1 ) \|_2^2$.
\EndFor
\For{$r=1\to R_{\text{samples}}$}
	\State Sample $\wh{j}_0$ from $[n_0]$ with probability $\alpha_i/\sum_{i'=1}^{n_0} \alpha_{i'}$.
	\For{$l=1 \to q-1$}
		\State $s_{l+1} \leftarrow O_q(\poly(k,1/\epsilon))$.
		\State Choose $\Pi_{l+1} \in \mathbb{R}^{ n_{l+1} \cdots n_q \times s_{l+1}}$ to be a \textsc{TensorSketch}.
		\For {$j_l = 1 \to [n_l]$} \Comment{Form $V_l\in \mathbb{R}^{n_l \times k}$}
			\State $(V_{l})^{j_l} \leftarrow (V_{l-1})^{\wh{j}_{l-1}} \circ (U_l)_{j_l}^\top $.
		\EndFor
		\For{$i=1\to n_q$}
			\State $\beta_i \leftarrow \| (V_l)^{i} ( (U_{l+1}  \odot \cdots \odot U_q) \Pi_{l+1} ) \|_2^2 $.
		\EndFor
		\State Sample $\wh{j}_l$ from $[n_l]$ with probability $ \beta_i /\sum_{i'=1}^{n_l} \beta_{i'}$.
	\EndFor
	\For{$i=1\to n_q$}
		\State $\beta_i \leftarrow | (V_{q-1})^{\wh{j}_{q-1}}  ( U_q)_i |^2 $.
	\EndFor
	\State Sample $\wh{j}_q$ from $[n_q]$ with probability $ \beta_i /\sum_{i'=1}^{n_q} \beta_{i'}$.
	\State ${\cal S} \leftarrow {\cal S} \cup (\wh{j}_1,\cdots, \wh{j}_q)$.
\EndFor
\State Convert ${\cal S}$ into a diagonal matrix $D$ with at most $R_{\text{samples}}$ nonzero entries.
\State \Return $D$. \Comment{Diagonal matrix $D\in \mathbb{R}^{n_1n_2 \cdots n_q \times n_1n_2 \cdots n_q}$}
\EndProcedure
\end{algorithmic}
\end{algorithm}

\begin{theorem}\label{thm:f_fast_tensor_leverage_score_general_order}
Given $q$ matrices $U_1\in \mathbb{R}^{k\times n_1}$, $U_2 \in \mathbb{R}^{k\times n_2}$, $\cdots$, $U_q \in \mathbb{R}^{k\times n_q}$, let $\max_{i} n_i \leq n$. There exists an algorithm that takes $O_q( n \cdot \poly(k,1/\epsilon) \cdot R_{\mathrm{samples}})$ time and samples $R_{\mathrm{samples}}$ columns of $U_1 \odot U_2 \odot \cdots \odot U_q \in \mathbb{R}^{k \times \prod_{i=1}^q n_i}$ according to the leverage scores of $U_1 \odot U_2 \odot \cdots \odot U_q$.
\end{theorem}
\begin{proof}
Let $\max_{i} n_i \leq n$. First, choosing $\Pi_0$ to be a \textsc{TensorSketch}, we can compute $R^{-1}$ in $O_q(n \poly( k,1/\epsilon))$ time, where $R$ is the $R$ in a QR-factorization. We want to sample columns from $U_1 \odot U_2 \odot \cdots \odot U_q$ according to the square of the $\ell_2$-norm of each column of $R^{-1}(U_1 \odot U_2 \odot \cdots U_q )$. The issue is the number of columns of this matrix is already $\prod_{i=1}^q n_i$. The goal is to sample columns from $R^{-1}(U_1 \odot U_2 \odot \cdots U_q )$ without explicitly computing the square of the $\ell_2$-norm of each column.

Similarly as in the proof of Lemma~\ref{lem:f_fast_tensor_leverage_score}, we have the observation that the following two sampling procedures are equivalent in terms of sampling a column of a matrix:
(1) We sample a single entry from matrix $R^{-1} (U_1\odot U_2 \odot \cdots \odot U_q)$ proportional to its squared value, (2) We sample a column from matrix $R^{-1} (U_1 \odot U_2 \odot \cdots \odot U_q)$ proportional to its squared $\ell_2$-norm. Let the $(i,j_1,j_2,\cdots,j_q)$-th entry denote the entry in the $i$-th row and the $j$-th column, where
\begin{align*}
j = \sum_{l=1}^{q-1} (j_l - 1) \prod_{t=l+1}^q n_t + j_q.
\end{align*}
Similarly to Equation~\eqref{eq:f_sampling_entry_is_sampling_column}, we can show, for a particular column $j$,
\begin{align*}
\Pr[\text{we~sample~an~entry~from~the~}j\text{-th~column~of~matrix}] = \Pr[\text{we~sample~the~}j\text{-th~column~of~a~matrix}].
\end{align*}
Thus, it is sufficient to show how to sample a single entry from matrix $R^{-1} (U_1 \odot U_2 \odot \cdots \odot U_q)$ proportional to its squared value without writing down all the entries of the $k\times \prod_{i=1}^q n_i$ matrix.


Let $V_0$ denote $R^{-1}$. Let $n_0$ denote the number of rows of $V_0$.

In the next few paragraphs, we describe a sampling procedure (procedure \textsc{FastTensorLeverageScoreGeneralOrder} in Algorithm~\ref{alg:f_fast_tensor_leverage_score_general_order}) which first samples $\wh{j}_0$ from $[n_0]$, then samples $\wh{j}_1$ from $[n_1]$, $\cdots$, and at the end samples $\wh{j}_q$ from $[n_q]$.

In the first step, we want to sample $\wh{j}_0$ from $[n_0]$ proportional to the squared $\ell_2$-norm of that row. To do this efficiently, we choose $\Pi_1\in \mathbb{R}^{\prod_{i=1}^q n_i \times s_1}$ to be a \textsc{TensorSketch} to sketch on the right of $V_0 ( U_1 \odot U_2 \odot \cdots \odot U_q )$. By Section~\ref{sec:def_tensor_sketch}, as long as $s_1 = O_q(\poly(k,1/\epsilon))$, then $\Pi_1$ is a $(1\pm\epsilon)$-subspace embedding matrix. Thus with probability $1-1/\Omega(q)$, for all $i\in [n_0]$,
\begin{align*}
\| (V_0)^i ( ( U_1 \odot U_2 \odot \cdots \odot U_q ) \Pi_1)\|_2^2 = (1\pm\epsilon) \| (V_0)^i ( ( U_1 \odot U_2 \odot \cdots \odot U_q ))\|_2^2,
\end{align*}
which means we can sample $\wh{j}_0$ from $[n_0]$ in $O_q( n\poly(k,1/\epsilon) )$ time.

In the second step, we have already obtained $\wh{j}_0$. Using that row of $V_0$ with $U_1$, we can form a new matrix $V_1\in \mathbb{R}^{n_1\times k}$ in the following sense,
\begin{align*}
(V_1)^{i} =  (V_0)^{\wh{j}_0} \circ (U_1)_{i}^\top  , \forall i \in [n_1],
\end{align*}
where $(V_1)^i$ denotes the $i$-th row of matrix $V_1$, $(V_0)^{\wh{j}_0}$ denotes the $\wh{j}_0$-th row of $V_0$ and $(U_1)_i$ is the $i$-th column of $U_1$. Another important observation is, the entry in the $(j_1,j_2,\cdots,j_q)$-th coordinate of vector $(V_0)^{\wh{j}_0} ( U_1 \odot U_2 \odot \cdots \odot U_q )$ is the same as the entry in the $j_1$-th row and $(j_2,\cdots,j_q)$-th column of matrix $V_1 (U_2\odot U_3 \odot \cdots \odot U_q)$. Thus, sampling $j_1$ is equivalent to sampling $j_1$ from the new matrix $V_1 (U_2\odot U_3 \odot \cdots \odot U_q)$ proportional to the squared $\ell_2$-norm of that row. We still have the computational issue that the length of the row vector is very long. To deal with this, we can choose $\Pi_2\in \mathbb{R}^{\prod_{i=2}^q n_i \times s_2}$ to be a \textsc{TensorSketch} to multiply on the right of $V_1 (U_2\odot U_3 \odot \cdots \odot U_q)$.

By Section~\ref{sec:def_tensor_sketch}, as long as $s_2 = O_q(\poly(k,1/\epsilon))$, then $\Pi_2$ is a $(1\pm\epsilon)$-subspace embedding matrix. Thus with probability $1-1/\Omega(q)$, for all $i\in [n_1]$,
\begin{align*}
\| (V_1)^i ( (  U_2 \odot \cdots \odot U_q ) \Pi_2)\|_2^2 = (1\pm\epsilon) \| (V_1)^i ( (   U_2 \odot \cdots \odot U_q ))\|_2^2,
\end{align*}
which means we can sample $\wh{j}_1$ from $[n_1]$ in $O_q( n \poly ( k,1/\epsilon ) )$ time.

We repeat the above procedure until we obtain each of $\wh{j}_0, \wh{j}_1, \cdots, \wh{j}_{q}$. Note that the last one, $\wh{j}_q$, is easier, since the length of the vector is already small enough, and so we do not need to use \textsc{TensorSketch} for it.

By Section~\ref{sec:def_tensor_sketch}, the time for multiplying by \textsc{TensorSketch} is $O_q (n \poly ( k, 1 / \epsilon ) )$. Setting $\epsilon$ to be a small constant, and taking a union bound over $O(q)$ events completes the proof.
\end{proof}

\begin{lemma}\label{lem:tensor_leverage_score_multiple_regression_general_order}
Given $A\in \mathbb{R}^{n_0\times \prod_{i=1}^q n_i}$, $U_1,U_2,\cdots, U_q\in \mathbb{R}^{k\times n}$, for any $\epsilon>0$, there exists an algorithm that runs in $O(n \cdot \poly(k,1/\epsilon))$ time and outputs a diagonal matrix $D\in \mathbb{R}^{ \prod_{i=1}^q n_i \times \prod_{i=1}^q n_i}$ with $m=O(k\log k + k/\epsilon)$ nonzero entries such that,
\begin{align*}
\| \wh{U} (U_1\odot U_2 \odot \cdots \odot U_q) - A\|_F^2 \leq (1+\epsilon) \min_{U\in \mathbb{R}^{n\times k}}\| U (U_1\odot U_2 \odot \cdots \odot U_q) - A\|_F^2,
\end{align*}
holds with probability at least $0.999$, where $\wh{U}$ denotes the optimal solution of
\begin{align*}
\min_{U\in \mathbb{R}^{n_0 \times k}} \| U (U_1\odot U_2 \odot \cdots \odot U_q) D - A D \|_F^2.
\end{align*}
\end{lemma}
\begin{proof}
This follows by combining Theorem~\ref{thm:f_fast_tensor_leverage_score_general_order}, Corollary~\ref{cor:leverage_score_size_multiple_regression}, and Lemma~\ref{lem:nw14_multiple_regression}.
\end{proof}

\subsubsection{General iterative existential proof}\label{sec:general_iterative_existential_proof}

\begin{algorithm}[!]\caption{General $q$-th Order Iterative Existential Proof}\label{alg:general_iterative_existential_proof}
\begin{algorithmic}[1]
\Procedure{\textsc{GeneralIterativeExistentialProof}}{$A,n,k,q,\epsilon$} \Comment{Section~\ref{sec:general_iterative_existential_proof}}
\State Fix $U_1^*, U_2^*, \cdots, U_q^* \in \mathbb{R}^{n\times k}$.
\For{$i=1 \to q$}
	\State Choose sketching matrix $S_i \in \mathbb{R}^{n^{q-1} \times s_i}$ with $s_i = O_q(k/\epsilon)$.
	\State Define $Z_i\in \mathbb{R}^{k\times n^{q-1}}$ to be  $\underset{j<i}{ \odot }   \wh{U}_j^\top \odot  \underset{j'>i}{\odot} U^{*\top}_{j'}$.
	\State Let $A_i$ denote the matrix obtained by flattening tensor $A$ along the $i$-th dimension.
	\State Define $\wh{U}_i$ to be $A_i S_i (Z_i S_i)^\dagger$.
\EndFor
\State \Return $\wh{U}_1, \wh{U}_2, \cdots, \wh{U}_q$.
\EndProcedure
\end{algorithmic}
\end{algorithm}
Given a $q$-th order tensor $A\in \mathbb{R}^{n\times n\times \cdots \times n}$, we fix $U_1^*, U_2^*, \cdots, U_q^* \in \mathbb{R}^{n\times k}$ to be the best rank-$k$ solution (if it does not exist, then we replace it by a good approximation, as discussed). We define $\OPT = \| U_1^* \otimes U_2^* \otimes \cdots \otimes U_q^* - A \|_F^2$. Our iterative proof works as follows. We first obtain the objective function,
\begin{align*}
\min_{U_1 \in \mathbb{R}^{n\times k}} \| U_1 \cdot Z_1 - A_1 \|_F^2 \leq \OPT,
\end{align*}
where $A_1$ is a matrix obtained by flattening tensor $A$ along the first dimension, $Z_1 = (U_2^{*\top} \odot U_3^{*\top} \odot \cdots \odot U_q^{*\top} )$ denotes a $k\times n^{q-1}$ matrix. Choosing $S_1\in \mathbb{R}^{n^{q-1} \times s_1}$ to be a Gaussian sketching matrix with $s_1 = O(k/\epsilon)$, we obtain a smaller problem,
\begin{align*}
\min_{U_1 \in \mathbb{R}^{n\times k}} \| U_1 \cdot Z_1 S_1 - A_1 S_1 \|_F^2.
\end{align*}
We define $\wh{U}_1$ to be $A_1 S_1 (Z_1 S_1)^\dagger \in \mathbb{R}^{n\times k}$, which gives,
\begin{align*}
\| \wh{U}_1 \cdot Z_1 - A_1 \|_F^2 \leq (1+\epsilon) \OPT.
\end{align*}
After retensorizing the above, we have, 
\begin{align*}
\| \wh{U}_1 \otimes U_2^*\otimes \cdots \otimes  U_q^* - A \|_F^2 \leq (1+\epsilon) \OPT.
\end{align*}
In the second round, we fix $\wh{U}_1$, $U_3^*$, $\cdots$, $U_q^* \in \mathbb{R}^{n\times k}$, and choose $S_2 \in \mathbb{R}^{n^{q-1} \times s_2}$ to be a Gaussian sketching matrix with $s_2 = O(k/\epsilon)$. We define $Z_2 \in \mathbb{R}^{k\times n^{q-1}}$ to be $(\wh{U}_1^\top \odot U_3^{*\top} \odot \cdots \odot U_q^{*\top} )$. We define $\wh{U}_2$ to be $A_2 S_2 (Z_2 S_2)^\dagger \in \mathbb{R}^{n\times k}$. Then, we have
\begin{align*}
\| \wh{U}_1 \otimes \wh{U}_2\otimes U_3^*\otimes  \cdots \otimes U_q^* - A \|_F^2 \leq (1+\epsilon)^2 \OPT.
\end{align*}
We repeat the above process, where in the $i$-th round we fix $\wh{U}_1, \cdots, \wh{U}_{i-1}$, $U_{i+1}^*$, $\cdots$, $U_q^*\in \mathbb{R}^{n\times k}$, and choose $S_i\in \mathbb{R}^{n^{q-1} \times s_i}$ to be a Gaussian sketching matrix with $s_i = O(k/\epsilon)$. We define $Z_i \in \mathbb{R}^{k\times n^{q-1}}$ to be $(\wh{U}_1^\top \odot \cdots \odot \wh{U}_{i-1}^\top \odot U_{i+1}^{*\top} \odot \cdots \odot U_q^{*\top} )$. We define $\wh{U}_i$ to be $A_i S_i (Z_i S_i)^\dagger \in \mathbb{R}^{n\times k}$. Then, we have
\begin{align*}
\| \wh{U}_1 \otimes \cdots \otimes \wh{U}_{i-1} \otimes \wh{U}_i \otimes U_{i+1}^* \otimes \cdots \otimes U_q^* - A \|_F^2 \leq (1+\epsilon)^2 \OPT.
\end{align*}
At the end of the $q$-th round, we have
\begin{align*}
\| \wh{U}_1 \otimes \cdots \otimes \wh{U}_{q} - A \|_F^2 \leq (1+\epsilon)^q \OPT.
\end{align*}

Replacing $\epsilon = \epsilon'/ (2q)$, we obtain
\begin{align*}
\| \wh{U}_1 \otimes \cdots \otimes \wh{U}_{q} - A \|_F^2 \leq (1+\epsilon') \OPT.
\end{align*}
where for all $i\in[q]$, $s_i = O( k q /\epsilon' ) = O_q(k/\epsilon')$ .

\subsubsection{General input sparsity reduction}\label{sec:general_input_sparsity_reduction}
This section shows how to extend the input sparsity reduction from third order tensors to general $q$-th order tensors. Given a tensor $A\in \mathbb{R}^{n\times n \times \cdots \times n}$ and $q$ matrices, for each $i\in [q]$, matrix $V_i$ has size $V_i \in \mathbb{R}^{n\times b_i}$, with $b_i \leq \poly(k,1/\epsilon)$. We choose a batch of sparse embedding matrices $T_i\in \mathbb{R}^{t_i \times n}$. Define $\wh{V}_i = T_i V_i$, and $C=A(T_1,T_2,\cdots,T_q)$. Thus we have with probability $99/100$, for any $\alpha\geq 0$, for all $\{ X_i,X_i'\in \mathbb{R}^{b_i \times k} \}_{i\in [q]}$,
if
\begin{align*}
\| \wh{V}_1 X_1' \otimes \wh{V}_2 X_2' \otimes \cdots \otimes \wh{V}_q X_q' - C \|_F^2 \leq \alpha  \| \wh{V}_1 X_1 \otimes \wh{V}_2 X_2 \otimes \cdots \otimes \wh{V}_q X_q -C\|_F^2,
\end{align*}
then
\begin{align*}
\| V_1 X_1' \otimes V_2 X_2' \otimes \cdots \otimes V_q X_q' -A\|_F^2 \leq (1+\epsilon)\alpha  \| V_1 X_1 \otimes V_2 X_2 \otimes \cdots \otimes V_q X_q -A\|_F^2,
\end{align*}
where $t_i = O_q( \poly(b_i,1/\epsilon) )$.

\begin{algorithm}[!]\caption{General $q$-th Order Input Sparsity Reduction}\label{alg:general_input_sparsity_reduction}
\begin{algorithmic}[1]
\Procedure{\textsc{GeneralInputSparsityReduction}}{$A,\{V_i\}_{i\in [q]},n,k,q,\epsilon$} \Comment{Section~\ref{sec:general_input_sparsity_reduction}}
\For{$i=1 \to q$}
	\State Choose sketching matrix $T_i\in \mathbb{R}^{t_i \times n}$ with $t_i = \poly(k,q,1/\epsilon)$.
	\State $\wh{V}_i \leftarrow T_i V_i$.
\EndFor
\State $C\leftarrow A(T_1, T_2,\cdots, T_q)$.
\State \Return $\{ \wh{V}_i\}_{i\in [q]}, C$.
\EndProcedure
\end{algorithmic}
\end{algorithm}

\subsubsection{Bicriteria algorithm}\label{sec:general_bicriteria_algorithm}
This section explains how to extend the bicriteria algorithm from third order tensors (Section~\ref{sec:f_bicriteria_algorithm}) to general $q$-th order tensors. Given any $q$-th order tensor $A\in \mathbb{R}^{n\times n\times \cdots \times n}$, we can output a $\rank$-$r$ tensor (or equivalently $q$ matrices $U_1, U_2, \cdots, U_q \in \mathbb{R}^{n\times r}$) such that,
\begin{align*}
\| U_1 \otimes U_2 \otimes \cdots \otimes U_q - A \|_F^2 \leq (1+\epsilon) \OPT,
\end{align*}
where $r= O_q( (k  /\epsilon)^{q-1} ) $ and the algorithm takes $O_q( \nnz(A) + n \cdot \poly(k,1/\epsilon))$.

\begin{algorithm}[!]\caption{General $q$-th Order Bicriteria Algorithm}\label{alg:general_bicriteria_algorithm}
\begin{algorithmic}[1]
\Procedure{\textsc{GeneralBicriteriaAlgorithm}}{$A,n,k,q,\epsilon$} \Comment{Section~\ref{sec:general_bicriteria_algorithm}}
\For{$i=2 \to q$}
	\State Choose sketching matrix $S_i\in \mathbb{R}^{ n^{q-1} \times s_i}$ with $s_i = O(kq/\epsilon)$.
	\State Choose sketching matrix $T_i\in \mathbb{R}^{t_i \times n}$ with $t_i = \poly(k,q,1/\epsilon)$.
	\State Form matrix $\wh{U}_i$ by setting $(j_2,j_3,\cdots,j_q)$-th column to be $(A_i S_i)_{j_i}$.
\EndFor
\State Solve $ \min_{U_1}\| U_1 B -  ( A(I, T_2,\cdots, T_q) )_1 \|_F^2 $.
\State \Return $\{ \wh{U}_i\}_{i\in [q]}$.
\EndProcedure
\end{algorithmic}
\end{algorithm}

\subsubsection{CURT decomposition}\label{sec:general_CURT_decomposition}
This section extends the tensor CURT algorithm from $3$rd order tensors (Section~\ref{sec:f_curt}) to general $q$-th order tensors. Given a $q$-th order tensor $A\in \mathbb{R}^{n\times n\times \cdots \times n}$ and a batch of matrices $U_1, U_2, \cdots, U_q \in \mathbb{R}^{n\times k}$, we iteratively apply the proof in Theorem~\ref{thm:f_curt_algorithm_input_sparsity} (or Theorem~\ref{thm:f_curt_algorithm_optimal_samples}) $q$ times. Then for each $i\in [q]$, we are able to select $d_i$ columns from the $i$-th dimension of tensor $A$ (let $C_i$ denote those columns) and also find a tensor $U\in \mathbb{R}^{d_1 \times d_2 \times \cdots \times d_q}$ such that,
\begin{align*}
\| U(C_1, C_2, \cdots, C_q) - A\|_F^2 \leq (1+\epsilon) \| U_1 \otimes U_2 \otimes \cdots \otimes U_q - A \|_F^2,
\end{align*}
where either $d_i = O_q(k\log k+k/\epsilon)$ (similar to Theorem~\ref{thm:f_curt_algorithm_input_sparsity}) or $d_i = O_q(k/\epsilon)$ (similar to Theorem~\ref{thm:f_curt_algorithm_optimal_samples}).

\begin{algorithm}[!]\caption{General $q$-th Order CURT Decomposition}\label{alg:general_CURT_decomposition}
\begin{algorithmic}[1]
\Procedure{\textsc{GeneralCURTDecomposition}}{$A,\{U_i\}_{i\in [q]},n,k,q,\epsilon$} \Comment{Section~\ref{sec:general_CURT_decomposition}}
\For{$i=1 \to q$}
	\State Form $B_i = \underset{j<i}{\odot} \wh{U}_j^\top \odot \underset{j>i}{\odot} U_j^\top \in \mathbb{R}^{k\times n^{q-1}}$.
	\If { fast = true} \Comment{Optimal running time}
		\State $\epsilon_0\leftarrow 0.01$.
		\State $d_i \leftarrow O_q(k\log k + k/\epsilon)$.
		\State $D_i\leftarrow$ \textsc{FastTensorLeverageScoreGeneralOrder} {\small $( \{ \wh{U}_j \}_{j<i}, \{ U_j \}_{j>i},n,k,\epsilon_0,d_i)$. } \Comment{Algorithm~\ref{alg:f_fast_tensor_leverage_score_general_order}}
	\Else \Comment{Optimal sample complexity}
		\State $\epsilon_0 \leftarrow O_q(\epsilon)$.
		\State $D_i\leftarrow$ \textsc{GeneralizedMatrixRowSubsetSelection} $(A_i^\top, B_i^\top, n^{q-1}, n, k, \epsilon_0)$. \Comment{Algorithm~\ref{sec:f_generalized_matrix_row},  $d_i = O_q(k/\epsilon)$.}
	\EndIf
	\State $\wh{U}_i \leftarrow A_i D_i (B_i D_i)^\dagger$.
	\State $C_i \leftarrow A_i D_i$.
\EndFor
\State $U \leftarrow  (B_1 D_1)^\dagger \otimes (B_2 D_2)^\dagger \otimes \cdots \otimes (B_q D_q)^\dagger$.
\State \Return $\{ C_i\}_{i\in [q]}$, $U$.
\EndProcedure
\end{algorithmic}
\end{algorithm}


\subsection{Matrix CUR decomposition}\label{sec:f_matrix_cur}
There is a long line of research on matrix CUR decomposition under operator, Frobenius or recently, entry-wise $\ell_1$ norm \cite{dmm08, bmd09, dr10, bdm11, bw14, swz17}. We provide the first algorithm that runs in $\nnz(A)$ time, which improves the previous best matrix CUR decomposition algorithm under Frobenius norm \cite{bw14}.

\subsubsection{Algorithm}
\begin{algorithm}[!]\caption{Optimal Matrix CUR Decomposition Algorithm}
\begin{algorithmic}[1]
\Procedure{\textsc{OptimalMatrixCUR}}{$A,n,k,\epsilon$} \Comment{Theorem~\ref{thm:f_matrix_cur_algorithm}}
\State $\varepsilon'\leftarrow0.1\varepsilon$. $\varepsilon''\leftarrow 0.001\varepsilon'$.
\State $\wh{U}\leftarrow$\textsc{SparseSVD}($A,k,\epsilon'$). \Comment{$\wh{U}\in \mathbb{R}^{n\times k}$}
\State Choose $S_1\in \mathbb{R}^{n\times n}$ to be a sampling and rescaling diagonal matrix according to the leverage scores of $\wh{U}$ with $s_1 = O(\epsilon^{-2} k\log k)$ nonzero entries.
\State $R,Y\leftarrow$\textsc{GeneralizedMatrixRowSubsetSelection}$(S_1A,S_1\wh{U},s_1,n,k,\epsilon'')$. \Comment{Algorithm~\ref{alg:f_generalized_matrix_row}, $R\in \mathbb{R}^{r \times n}, Y\in \mathbb{R}^{k\times r}$ and $r=O(k/\epsilon)$}
\State $\wh{V}\leftarrow YR \in \mathbb{R}^{k\times n}$.
\State Choose $S_2^\top \in \mathbb{R}^{n\times n}$ to be a sampling and rescaling diagonal matrix according to the leverage scores of $\wh{V}^\top \in \mathbb{R}^{n\times k}$ with $s_2=O(\epsilon^{-2} k\log k)$ nonzero entries.
\State  $C^\top ,Z^\top \leftarrow$  \textsc{GeneralizedMatrixRowSubsetSelection} $( (AS_2)^\top, (\wh{V} S_2)^\top,s_2,n,k,\epsilon'')$. \Comment{Algorithm~\ref{alg:f_generalized_matrix_row}, $C\in \mathbb{R}^{n\times c}, Z\in \mathbb{R}^{c\times k}$, and $c=O(k/\epsilon)$}
\State $U \leftarrow ZY$. \Comment{$U\in \mathbb{R}^{c\times r}$ and $\rank(U) = k$}
\State \Return $C,U,R$.
\EndProcedure
\end{algorithmic}
\end{algorithm}

\begin{theorem}\label{thm:f_matrix_cur_algorithm}
Given matrix $A\in \mathbb{R}^{n\times n}$, for any $k\geq 1$ and $\epsilon \in (0,1)$, there exists an algorithm that takes $O(\nnz(A) + n\poly(k,1/\epsilon) )$ time and outputs three matrices $C\in \mathbb{R}^{n\times c}$ with $c$ columns from $A$, $R\in \mathbb{R}^{r\times n}$ with $r$ rows from $A$, and $U\in \mathbb{R}^{c \times r}$ with $\rank(U)=k$ such that $r=c=O(k/\epsilon)$ and,
\begin{align*}
\| CUR - A \|_F^2 \leq (1+\epsilon) \min_{\rank-k~A_k} \| A_k - A \|_F^2,
\end{align*}
holds with probability at least $9/10$.
\end{theorem}
\begin{proof}
We define
\begin{align*}
\OPT = \min_{\rank-k~A_k} \| A_k - A \|_F^2.
\end{align*}
We first compute $\wh{U}\in \mathbb{R}^{n\times k}$ by using the result of \cite{cw13}, so that $\wh{U}$ satisfies:
\begin{align}\label{eq:f_whU_is_good}
\min_{X \in \mathbb{R}^{k\times n}} \| \wh{U} X - A \|_F^2 \leq (1+\epsilon) \OPT.
\end{align}
This step can be done in $O(\nnz(A) + n\poly(k,1/\epsilon))$ time.

We choose $S_1\in \mathbb{R}^{n\times n}$ to be a sampling and rescaling diagonal matrix according to the leverage scores of $\wh{U}$, where here $s_1 = O(\epsilon^{-2} k\log k)$ is the number of samples. This step also can be done in $O( n\poly(k,1/\epsilon))$ time.

We run \textsc{GeneralizedMatrixRowSubsetSelection}(Algorithm~\ref{alg:f_generalized_matrix_row}) on matrices $S_1 A$ and $S_1\wh{U}$. Then we obtain two new matrices $R$ and $Y$, where $R$ contains $r=O(k/\epsilon)$ rows of $S_1A$ and $Y$ has size $k\times r$. According to Theorem~\ref{thm:f_generalized_matrix_row} and Corollary~\ref{cor:f_generalized_matrix_row}, this step takes $n\poly(k,1/\epsilon)$ time.

We construct $\wh{V}=YR$, and choose $S_2^\top$ to be another sampling and rescaling diagonal matrix according to the leverage scores of $\wh{V}^\top$ with $s_2=O(\epsilon^{-2} k \log k)$ nonzero entries. As in the case of constructing $S_1$, this step can be done in $O( n\poly(k,1/\epsilon))$ time.

We run \textsc{GeneralizedMatrixRowSubsetSelection}(Algorithm~\ref{alg:f_generalized_matrix_row}) on matrices $(A S_2)^\top$ and $(\wh{V}S_2)^\top$. Then we can obtain two new matrices $C^\top$ and $Y^\top$, where $C^\top$ contains $c=O(k/\epsilon)$ rows of $(AS_2)^\top$ and $Z^\top$ has size $k\times c$. According to Theorem~\ref{thm:f_generalized_matrix_row} and Corollary~\ref{cor:f_generalized_matrix_row}, this step takes $n\poly(k,1/\epsilon)$ time.

Thus, overall the running time is $O(\nnz(A) + n \poly(k,1/\epsilon))$.

\paragraph{Correctness.}
Let
\begin{align*}
X^*=\arg\min_{X\in\mathbb{R}^{n\times k}}\|X\wh{V}-A\|_F^2.
\end{align*}
According to Corollary~\ref{cor:f_generalized_matrix_row},
\begin{align*}
\|CZ\wh{V}S_2-AS_2\|_F^2\leq (1+\varepsilon'')\min_{X\in\mathbb{R}^{n\times k}} \|X\wh{V}S_2-AS_2\|_F^2\leq (1+\varepsilon'')\|X^*\wh{V}S_2-AS_2\|_F^2.
\end{align*}
According to Theorem~\ref{thm:f_regression_cost_preserve_by_leverage_score}, $\varepsilon''=0.001\varepsilon'$,
\begin{align}\label{eq:f_CZ_is_good}
\|CZ\wh{V}-A\|_F^2\leq (1+\varepsilon')\|X^*\wh{V}-A\|_F^2.
\end{align}

Let
\begin{align*}
\wt{X}=\arg\min_{X\in\mathbb{R}^{k\times n}}\|\wh{U}X-A\|_F^2.
\end{align*}
According to Corollary~\ref{cor:f_generalized_matrix_row},
\begin{align*}
\|S_1\wh{U}YR-S_1A\|_F^2\leq (1+\varepsilon'')\min_{X\in\mathbb{R}^{k\times n}} \|S_1\wh{U}X-S_1A\|_F^2\leq (1+\varepsilon'')\|S_1\wh{U}\wt{X}-S_1A\|_F^2.
\end{align*}
According to Theorem~\ref{thm:f_regression_cost_preserve_by_leverage_score}, since $\varepsilon''=0.001\varepsilon'$,
\begin{align}\label{eq:f_YR_is_good}
\|\wh{U}YR-A\|_F^2\leq (1+\varepsilon')\|\wh{U}\wt{X}-A\|_F^2.
\end{align}
Then, we can conclude
\begin{align*}
 \|CUR-A\|_F^2
=& ~ \|CZYR-A\|_F^2\\
=& ~ \|CZ\wh{V}-A\|_F^2\\
\leq & ~ (1+\varepsilon') \min_{X\in\mathbb{R}^{n\times k}} \|X\wh{V}-A\|_F^2\\
\leq & ~ (1+\varepsilon') \|\wh{U}\wh{V}-A\|_F^2\\
\leq & ~ (1+\varepsilon')^2 \min_{X\in\mathbb{R}^{k\times n}}\|\wh{U}X-A\|_F^2\\
\leq & ~ (1+\varepsilon')^3 \OPT\\
\leq & ~ (1+\eps) \OPT.
\end{align*}
The first equality follows since $U=ZY$. The second equality follows since $YR=\wh{V}$. The first inequality follows by Equation~\eqref{eq:f_CZ_is_good}. The third inequality follows by Equation~\eqref{eq:f_YR_is_good}. The fourth inequality follows by Equation~\eqref{eq:f_whU_is_good}. The last inequality follows since $\eps'=0.1\eps$.

Notice that $C$ has $O(k/\varepsilon)$ reweighted columns of $AS_2$, and $AS_2$ is a subset of reweighted columns of $A$, so $C$ has $O(k/\varepsilon)$ reweighted columns of $A$. Similarly, we can prove that $R$ has $O(k/\varepsilon)$  reweighted rows of $A$. Thus, $CUR$ is a CUR decomposition of $A$.

\end{proof}

\subsubsection{Stronger property achieved by leverage scores}
\begin{claim}\label{cla:f_sampling_matrix_markov_bound}
Given matrix $A\in \mathbb{R}^{n\times m}$, for any distribution $p=(p_1,p_2,\cdots,p_n)$ define random variable $X$ such that $X= \| A_i \|_2^2 /p_i$ with probability $p_i$, where $A_i$ is the $i$-th row of matrix $A$. Then take $m$ independent samples $X^1, X^2, \cdots, X^m$, and let $Y = \frac{1}{m} \sum_{j=1}^m X^j$. We have
\begin{align*}
\Pr[Y \leq 100 \| A\|_F^2 ] \geq .99.
\end{align*}
\end{claim}
\begin{proof}
We can compute the expectation of $X^j$, for any $j\in [m]$,
\begin{align*}
\E [X^j ] = \sum_{i=1}^n \frac{\| A_i \|_2^2 }{p_i} \cdot p_i = \| A\|_F^2.
\end{align*}
Then $\E[Y] = \frac{1}{m} \sum_{j=1}^m \E[X^j] = \| A\|_F^2$. Using Markov's inequality, we have
\begin{align*}
\Pr[Y \geq \| A\|_F^2 ] \leq .01.
\end{align*}
\end{proof}

\begin{theorem}[The leverage score case of Theorem 39 in~\cite{cw13}]\label{thm:f_theorem_39_in_cw13}
Let $A\in \mathbb{R}^{n\times k}$, $B\in\mathbb{R}^{n\times d}$. Let $S\in\mathbb{R}^{n\times n}$ denote a sampling and rescaling diagonal matrix according to the leverage scores of $A$. If the event occurs that $S$ satisfies $(\varepsilon/\sqrt{k})$-Frobenius norm approximate matrix product for $A$, and also $S$ is a $(1+\epsilon)$-subspace embedding for $A$, then let $X^*$ be the optimal solution of $\min_{X}\|AX-B\|_F^2,$ and $\wt{B}\equiv AX^*-B$. Then, for all $X\in\mathbb{R}^{k\times d},$
\begin{align*}
(1-2\varepsilon)\|AX-B\|_F^2\leq \|S(AX-B)\|_F^2+\|\wt{B}\|_F^2-\|S\wt{B}\|_F^2\leq (1+2\varepsilon)\|AX-B\|_F^2.
\end{align*}
Furthermore, if $S$ has $m=O(\epsilon^{-2} k\log(k) )$ nonzero entries, the above event happens with probability at least $0.99$.
\end{theorem}
Note that Theorem 39 in~\cite{cw13} is stated in a way that holds for general sketching matrices. However, we are only interested in the case when $S$ is a sampling and rescaling diagonal matrix according to the leverage scores. For completeness, we provide the full proof of the leverage score case with certain parameters.
\begin{proof} Suppose $S$ is a sampling and rescaling diagonal matrix according to the leverage scores of $A$, and it has $m=O( \epsilon^{-2} k\log k )$ nonzero entries. Then, according to Lemma~\ref{lem:f_leverage_score_klogkeps2_subspace_embedding}, $S$ is a $(1+\varepsilon)$-subspace embedding for $A$ with probability at least $0.999$, and according to Lemma~\ref{lem:lemma_32_in_cw13}, $S$ satisfies $(\epsilon/\sqrt{k})$-Frobenius norm approximate matrix product for $A$ with probability at least $0.999$.

Let $U\in\mathbb{R}^{n\times k}$ denote an orthonormal basis of the column span of $A$. Then the leverage scores of $U$ are the same as the leverage scores of $A$. Furthermore, for any $X\in\mathbb{R}^{k\times d}$, there is a matrix $Y$ such that $AX=UY$, and vice versa, so we can now assume $A$ has $k$ orthonormal columns.

Then,
\begin{align}
& ~ \|S(AX-B)\|_F^2-\|S\wt{B}\|_F^2\notag\\
=& ~ \|SA(X-X^*)+S(AX^*-B)\|_F^2-\|S\wt{B}\|_F^2\notag\\
=& ~ \|SA(X-X^*)\|_F^2+\|S(AX^*-B)\|_F^2+2\tr\left((X-X^*)^\top A^\top S^\top S (AX^*-B)\right)-\|S\wt{B}\|_F^2\notag\\
=& ~ \underbrace{ \|SA(X-X^*)\|_F^2+2\tr\left((X-X^*)^\top A^\top S^\top S \wt{B}\right) }_{\alpha}. \label{eq:f_trace_frobenius}
\end{align}
The second equality follows using $\|C+D\|_F^2=\|C\|_F^2+\|D\|_F^2+2\tr(C^\top D)$. The third equality follows from $\wt{B}=AX^*-B$.
Now, let us first upper bound the term $\alpha$ in Equation~(\ref{eq:f_trace_frobenius}):
\begin{align*}
& ~ \|SA(X-X^*)\|_F^2+2\tr\left((X-X^*)^\top A^\top S^\top S \wt{B}\right)\\
\leq & ~ (1+\varepsilon)\|A(X-X^*)\|_F^2+2\|X-X^*\|_F\|A^\top S^\top S \wt{B}\|_F\\
\leq & ~ (1+\varepsilon)\|A(X-X^*)\|_F^2+2(\varepsilon/\sqrt{k})\cdot\|X-X^*\|_F\|A\|_F \|\wt{B}\|_F\\
\leq & ~ (1+\varepsilon)\|A(X-X^*)\|_F^2+2\varepsilon\|A(X-X^*)\|_F\|\wt{B}\|_F.
\end{align*}
The first inequality follows since $S$ is a $(1+\varepsilon)$ subspace embedding of $A$, and $\tr(C^\top D)\leq \|C\|_F\|D\|_F$. The second inequality follows since $S$ satisfies $(\varepsilon/\sqrt{k})$-Frobenius norm approximate matrix product for $A$. The last inequality follows using that $\|A\|_F\leq \sqrt{k}$ since $A$ only has $k$ orthonormal columns.
Now, let us lower bound the term $\alpha$ in Equation~(\ref{eq:f_trace_frobenius}):
\begin{align*}
& ~ \|SA(X-X^*)\|_F^2+2\tr\left((X-X^*)^\top A^\top S^\top S \wt{B}\right)\\
\geq & ~ (1-\varepsilon)\|A(X-X^*)\|_F^2-2\|X-X^*\|_F\|A^\top S^\top S \wt{B}\|_F\\
\geq & ~ (1-\varepsilon)\|A(X-X^*)\|_F^2-2(\varepsilon/\sqrt{k})\cdot\|X-X^*\|_F\|A\|_F \|\wt{B}\|_F\\
\geq & ~ (1-\varepsilon)\|A(X-X^*)\|_F^2-2\varepsilon\|A(X-X^*)\|_F\|\wt{B}\|_F.
\end{align*}
The first inequality follows since $S$ is a $(1+\varepsilon)$ subspace embedding of $A$, and $\tr(C^\top D)\geq -\|C\|_F\|D\|_F$. The second inequality follows since $S$ satisfies $(\varepsilon/\sqrt{k})$-Frobenius norm approximate matrix product for $A$. The last inequality follows using that $\|A\|_F\leq \sqrt{k}$ since $A$ only has $k$ orthonormal columns.

Therefore,
\begin{align}\label{eq:f_leverage_lower_bound}
(1-\varepsilon)\|A(X-X^*)\|_F^2-2\varepsilon\|A(X-X^*)\|_F\|\wt{B}\|_F \leq \|S(AX-B)\|_F^2-\|S\wt{B}\|_F^2, 
\end{align}
and
\begin{align}\label{eq:f_leverage_upper_bound}
(1+\varepsilon)\|A(X-X^*)\|_F^2+2\varepsilon\|A(X-X^*)\|_F\|\wt{B}\|_F \geq \|S(AX-B)\|_F^2-\|S\wt{B}\|_F^2.
\end{align}
Notice that $\wt{B}=AX^*-B=AA^\dagger B-B=(AA^\dagger-I)B,$ so according to the Pythagorean theorem, we have
\begin{align*}
\|AX-B\|_F^2=\|A(X-X^*)\|_F^2+\|\wt{B}\|_F^2,
\end{align*}
which means that
\begin{align}\label{eq:f_leverage_AX_minus_AXstar}
\|A(X-X^*)\|_F^2=\|AX-B\|_F^2-\|\wt{B}\|_F^2.
\end{align}
Using Equation~\eqref{eq:f_leverage_AX_minus_AXstar}, we can rewrite and lower bound the $\LHS$ of Equation~(\ref{eq:f_leverage_lower_bound}),
\begin{align}\label{eq:f_leverage_lower_bound_new}
& ~ (1-\varepsilon)\|A(X-X^*)\|_F^2-2\varepsilon\|A(X-X^*)\|_F\|\wt{B}\|_F \notag \\
= & ~ \|A(X-X^*)\|_F^2 - \varepsilon \left( \|A(X-X^*)\|_F^2+ 2 \|A(X-X^*)\|_F\|\wt{B}\|_F \right) \notag \\
=& ~ \|AX-B\|_F^2-\|\wt{B}\|_F^2-\varepsilon\left(\|A(X-X^*)\|_F^2+2\|A(X-X^*)\|_F\|\wt{B}\|_F\right) \notag \\
\geq & ~ \|AX-B\|_F^2-\|\wt{B}\|_F^2-\varepsilon\left(\|A(X-X^*)\|_F+\|\wt{B}\|_F\right)^2 \notag \\
\geq & ~ \|AX-B\|_F^2-\|\wt{B}\|_F^2-2\varepsilon\left(\|A(X-X^*)\|_F^2+\|\wt{B}\|_F^2\right) \notag \\
=& ~ (1-2\varepsilon)\|AX-B\|_F^2-\|\wt{B}\|_F^2.
\end{align}
The second step follows by Equation~\eqref{eq:f_leverage_AX_minus_AXstar}. The first inequality follows using $a^2+2ab<(a+b)^2$. The second inequality follows using $(a+b)^2\leq2(a^2+b^2).$ The last equality follows using $\|A(X-X^*)\|_F^2+\|\wt{B}\|_F^2=\|AX-B\|_F^2$.
Similarly, using Equation~\eqref{eq:f_leverage_AX_minus_AXstar}, we can rewrite and upper bound the $\LHS$ of Equation~(\ref{eq:f_leverage_upper_bound})
\begin{align}\label{eq:f_leverage_upper_bound_new}
(1+\varepsilon)\|A(X-X^*)\|_F^2+2\varepsilon\|A(X-X^*)\|_F\|\wt{B}\|_F \leq  (1+2\varepsilon)\|AX-B\|_F^2-\|\wt{B}\|_F^2.
\end{align}
Combining Equations~\eqref{eq:f_leverage_lower_bound},\eqref{eq:f_leverage_lower_bound_new},\eqref{eq:f_leverage_upper_bound},\eqref{eq:f_leverage_upper_bound_new}, we conclude that
\begin{align*}
(1-2\varepsilon)\|AX-B\|_F^2-\|\wt{B}\|_F^2\leq \|S(AX-B)\|_F^2-\|S\wt{B}\|_F^2\leq (1+2\varepsilon)\|AX-B\|_F^2-\|\wt{B}\|_F^2.
\end{align*}

\end{proof}

\begin{theorem}\label{thm:f_regression_cost_preserving_sketch}
Let $A\in \mathbb{R}^{n\times k}$, $B\in\mathbb{R}^{n\times d}$, and $\frac{1}{2}>\varepsilon>0$. Let $X^*$ be the optimal solution to $\min_{X}\|AX-B\|_F^2,$ and $\wt{B}\equiv AX^*-B$. Let $S\in\mathbb{R}^{n\times n}$ denote a sketching matrix which satisfies the following:
\begin{enumerate}
\item $\|S\wt{B}\|_F^2\leq 100\cdot \|\wt{B}\|_F^2,$
\item for all $X\in\mathbb{R}^{k\times d},$
\begin{align*}
(1-\varepsilon)\|AX-B\|_F^2\leq \|S(AX-B)\|_F^2+\|\wt{B}\|_F^2-\|S\wt{B}\|_F^2\leq (1+\varepsilon)\|AX-B\|_F^2.
\end{align*}
\end{enumerate}
Then, for all $X_1,X_2\in\mathbb{R}^{k\times d}$ satisfying
\begin{align*}
\|SAX_1-SB\|_F^2\leq \left(1+\frac{\varepsilon}{100}\right)\cdot\|SAX_2-SB\|_F^2,
\end{align*}
we have
\begin{align*}
\|AX_1-B\|_F^2\leq (1+5\varepsilon)\cdot\|AX_2-B\|_F^2.
\end{align*}
\end{theorem}
\begin{proof}
Let $A,B,S,\varepsilon$ be the same as in the statement of the theorem, and suppose $S$ satisfies those two conditions. Let $X_1,X_2\in\mathbb{R}^{k\times d}$ satisfy
\begin{align*}
\|SAX_1-SB\|_F^2\leq \left(1+\frac{\varepsilon}{100}\right)\|SAX_2-SB\|_F^2.
\end{align*}
We have
\begin{align*}
& ~ \|AX_1-B\|_F^2\\
\leq & ~ \frac{1}{1-\varepsilon} \left(\|S(AX_1-B)\|_F^2+\|\wt{B}\|_F^2-\|S\wt{B}\|_F^2\right)\\
\leq & ~ \frac{1}{1-\varepsilon} \left(\left(1+\frac{\varepsilon}{100}\right)\cdot\|S(AX_2-B)\|_F^2+\|\wt{B}\|_F^2-\|S\wt{B}\|_F^2\right)\\
= & ~ \frac{1}{1-\varepsilon} \left(\left(1+\frac{\varepsilon}{100}\right)\cdot\left(\|S(AX_2-B)\|_F^2+\|\wt{B}\|_F^2-\|S\wt{B}\|_F^2\right)-\frac{\varepsilon}{100}\cdot\left(\|\wt{B}\|_F^2-\|S\wt{B}\|_F^2\right)\right)\\
\leq & ~ \frac{1}{1-\varepsilon}\cdot \left(1+\frac{\varepsilon}{100}\right) \cdot\|AX_2-B\|_F^2-\frac{1}{1-\varepsilon}\cdot\frac{\varepsilon}{100}\cdot\left(\|\wt{B}\|_F^2-\|S\wt{B}\|_F^2\right)\\
\leq & ~ (1+3\varepsilon)\|AX_2-B\|_F^2+\frac{1}{1-\varepsilon}\cdot\frac{\varepsilon}{100} \|S\wt{B}\|_F^2\\
\leq & ~ (1+3\varepsilon)\|AX_2-B\|_F^2+2\varepsilon \|\wt{B}\|_F^2\\
\leq & ~ (1+5\varepsilon)\|AX_2-B\|_F^2.
\end{align*}
The first inequality follows since $S$ satisfies the second condition. The second inequality follows by the relationship between $X_1$ and $X_2$. The third inequality follows since $S$ satisfies the second condition. The fifth inequality follows using that $\varepsilon<\frac{1}{2}$ and that $S$ satisfies the first condition. The last inequality follows using that $\|\wt{B}\|_F^2=\|AX^*-B\|_F^2\leq \|AX_2-B\|_F^2.$
\end{proof}

\begin{theorem}\label{thm:f_regression_cost_preserve_by_leverage_score}
Let $A\in \mathbb{R}^{n\times k}$, $B\in\mathbb{R}^{n\times d}$, and $\frac{1}{2}>\varepsilon>0$. Let $S\in\mathbb{R}^{n\times n}$ denote a sampling and rescaling diagonal matrix according to the leverage scores of $A$. If $S$ has at least $m=O(k\log(k)/\varepsilon^2)$ nonzero entries, then with probability at least $0.98$, for all $X_1,X_2\in\mathbb{R}^{k\times d}$ satisfying
\begin{align*}
\|SAX_1-SB\|_F^2\leq (1+\frac{\varepsilon}{500})\cdot\|SAX_2-SB\|_F^2,
\end{align*}
we have
\begin{align*}
\|AX_1-B\|_F^2\leq (1+\varepsilon)\cdot\|AX_2-B\|_F^2.
\end{align*}
\end{theorem}
\begin{proof}
The proof directly follows by Claim~\ref{cla:f_sampling_matrix_markov_bound}, Theorem~\ref{thm:f_theorem_39_in_cw13} and Theorem~\ref{thm:f_regression_cost_preserving_sketch}. Because of Claim~\ref{cla:f_sampling_matrix_markov_bound}, $S$ satisfies the first condition in the statement of Theorem~\ref{thm:f_regression_cost_preserving_sketch} with probability at least $0.99$. According to Theorem~\ref{thm:f_theorem_39_in_cw13}, $S$ satisfies the second condition in the statement of Theorem~\ref{thm:f_regression_cost_preserving_sketch} with probability at least $0.99$. Thus, with probability $0.98$, by Theorem~\ref{thm:f_regression_cost_preserving_sketch}, we complete the proof.
\end{proof}
\newpage

\section{Entry-wise $\ell_1$ Norm for Arbitrary Tensors}\label{sec:l1}
In this section, we provide several different algorithms for tensor $\ell_1$-low rank approximation. Section~\ref{sec:l1_facts} provides some useful facts and definitions. Section~\ref{sec:l1_existence_results} presents several existence results. Section~\ref{sec:l1_polyk_size_reduction} describes a tool that is able to reduce the size of the objective function from $\poly(n)$ to $\poly(k)$. Section~\ref{sec:l1_solving_small_problems} discusses the case when the problem size is small. Section~\ref{sec:l1_bicriteria_algorithm} provides several bicriteria algorithms. Section~\ref{sec:l1_algorithm} summarizes a batch of algorithms. Section~\ref{sec:l1_curt} provides an algorithm for $\ell_1$ norm CURT decomposition.

Notice that if the $\rank-k$ solution does not exist, then every bicriteria algorithm in Section~\ref{sec:l1_bicriteria_algorithm} can be stated in a form similar to Theorem~\ref{thm:bicriteria}, and every algorithm which can output a $\rank-k$ solution in Section~\ref{sec:l1_algorithm} can be stated in a form similar to Theorem~\ref{thm:smallk}. See Section~\ref{sec:intro} for more details.

\subsection{Facts}\label{sec:l1_facts}

We present a method that is able to reduce the entry-wise $\ell_1$-norm objective function to the Frobenius norm objective function.
\begin{fact}\label{fac:l1_relax_to_frobenius_norm}
Given a $3$rd order tensor $C\in \mathbb{R}^{c_1 \times c_2 \times c_3}$, three matrices $V_1\in \mathbb{R}^{c_1 \times b_1}$, $V_2\in \mathbb{R}^{c_2\times b_2}$, $V_3\in \mathbb{R}^{c_3\times b_3}$, for any $k \in [1,\min_{i} b_i]$, if $X'_1\in \mathbb{R}^{b_1\times k},X'_2\in \mathbb{R}^{b_2\times k},X'_3\in \mathbb{R}^{b_3\times k}$ satisfies that,
\begin{align*}
\| (V_1 X'_1) \otimes (V_2 X'_2) \otimes (V_3 X'_3) - C \|_F \leq \alpha \min_{X_1,X_2,X_3} \| (V_1 X_1) \otimes (V_2 X_2) \otimes (V_3 X_3) - C \|_F,
\end{align*}
then
\begin{align*}
\| (V_1 X'_1) \otimes (V_2 X'_2) \otimes (V_3 X'_3) - C \|_1 \leq \alpha \sqrt{c_1 c_2 c_3} \min_{X_1,X_2,X_3} \| (V_1 X_1) \otimes (V_2 X_2) \otimes (V_3 X_3) - C \|_1.
\end{align*}
\end{fact}

We extend Lemma C.15 in \cite{swz17} to tensors:
\begin{fact}
Given tensor $A\in \mathbb{R}^{n\times n\times n}$, let $\OPT=\underset{\rank-k~A_k}{\min} \| A - A_k \|_1$. For any $r\geq k$, if $\rank$-$r$ tensor $B\in \mathbb{R}^{n\times n\times n}$ is an $f$-approximation to $A$, i.e.,
\begin{align*}
\| B -  A \|_1 \leq f \cdot \OPT,
\end{align*}
and $U,V,W \in \mathbb{R}^{n\times k}$ is a $g$-approximation to $B$, i.e.,
\begin{align*}
\| U \otimes V \otimes W - B \|_1 \leq g \cdot \underset{\rank-k~B_k}{\min} \| B_k - B \|_1,
\end{align*}
then,
\begin{align*}
\| U \otimes V \otimes W - A \|_1 \lesssim g f \cdot \OPT.
\end{align*}
\end{fact}
\begin{proof}
We define $\wt{U}, \wt{V}, \wt{W} \in \mathbb{R}^{n\times k}$ to be three matrices, such that
\begin{align*}
\| \wt{U} \otimes \wt{V} \otimes \wt{W} - B \|_1 \leq g \underset{\rank-k~B_k}{\min} \| B_k - B \|_1,
\end{align*}
and also define,
\begin{align*}
\wh{U}, \wh{V}, \wh{W} = \underset{U,V,W\in \mathbb{R}^{n\times k} }{\arg\min} \| U \otimes V \otimes W - B \|_1 \text{~and~} U^*,V^*,W^* = \underset{U,V,W\in \mathbb{R}^{n\times k} }{\arg\min} \| U \otimes V \otimes W - A \|_1.
\end{align*}
It is obvious that,
\begin{align}\label{eq:l1_whUVW_starUVW_B}
\| \wh{U} \otimes \wh{V} \otimes \wh{W} - B \|_1 \leq \| U^* \otimes V^* \otimes W^* - B \|_1.
\end{align}
Then,
\begin{align*}
 & ~ \| \wt{U} \otimes \wt{V} \otimes \wt{W} - A \|_1 \\
\leq & ~ \| \wt{U} \otimes \wt{V} \otimes \wt{W} - B \|_1 + \| B - A \|_1 & \text{~by~the~triangle~inequality} \\
\leq & ~ g \| \wh{U} \otimes \wh{V} \otimes \wh{W} - B \|_1 + \| B - A \|_1 &\text{~by~definition} \\
\leq & ~ g \| U^* \otimes V^* \otimes W^* - B \|_1 + \| B - A \|_1 & \text{~by~Equation~\eqref{eq:l1_whUVW_starUVW_B}}  \\
\leq & ~ g \| U^* \otimes V^* \otimes W^* - A \|_1 + g \| B-A\|_1 + \| B-A\|_1 & \text{~by~the~triangle~inequality} \\
= & ~ g \OPT + (g+1) \| B - A \|_1 & \text{~by~definition~of~}\OPT\\
\leq & ~ g \OPT + (g+1 ) f \cdot \OPT & \text{~since~$B$~is~an~$f$-approximation~to~}A\\
\lesssim & ~ g f \OPT.
\end{align*}
This completes the proof.
\end{proof}

Using the above fact, we are able to optimize our approximation ratio.

\subsection{Existence results}\label{sec:l1_existence_results}

\begin{definition}[$\ell_1$ multiple regression cost preserving sketch - Definition D.5 in~\cite{swz17}]\label{def:l1_multiple_regression_cost_preserving_sketch}
Given matrices $U\in\mathbb{R}^{n\times r},A\in\mathbb{R}^{n\times d}$, let $S\in\mathbb{R}^{m\times n}$. If $\forall \beta\geq 1,\wh{V}\in\mathbb{R}^{r\times d}$ which satisfy
\begin{align*}
\|SU\wh{V}-SA\|_1\leq \beta\cdot\min_{V\in\mathbb{R}^{r\times d}}\|SUV-SA\|_1,
\end{align*}
it holds that
\begin{align*}
\|U\wh{V}-A\|_1\leq \beta\cdot c\cdot\min_{V\in\mathbb{R}^{r\times d}}\|UV-A\|_1,
\end{align*}
then $S$ provides a $c$-$\ell_1$-multiple-regression-cost-preserving-sketch for $(U,A)$.
\end{definition}

\begin{theorem}\label{thm:l1_existence_results}
Given a 3rd order tensor $A\in \mathbb{R}^{n\times n \times n}$, for any $k\geq 1$,
there exist three matrices $S_1\in \mathbb{R}^{n^2\times s_1}$, $S_2\in \mathbb{R}^{n^2\times s_2}$, $S_3 \in \mathbb{R}^{n^2 \times s_3}$ such that
\begin{align*}
\min_{  X_1, X_2 , X_3 } \left\| \sum_{i=1}^k (A_1S_1 X_1)_i \otimes (A_2S_2X_2)_i \otimes (A_3S_3X_3)_i -A \right\|_1 \leq \alpha \underset{\rank-k~A_k \in \mathbb{R}^{n\times n \times n}}{\min} \| A_k -A \|_1,
\end{align*}
holds with probability $99/100$.

$\mathrm{(\RN{1})}$. Using a dense Cauchy transform,\\
$s_1=s_2=s_3=\wt{O}(k)$, $\alpha = \wt{O}(k^{1.5}) \log^3 n$. 

$\mathrm{(\RN{2})}$. Using a sparse Cauchy transform,\\
$s_1=s_2=s_3=\wt{O}(k^5)$, $\alpha = \wt{O}(k^{13.5}) \log^3 n$. 

$\mathrm{(\RN{3})}$. Guessing Lewis weights,\\
$s_1=s_2=s_3=\wt{O}(k)$, $\alpha = \wt{O}(k^{1.5})$. 
\end{theorem}
\begin{proof}
We use $\OPT$ to denote
\begin{align*}
\OPT : = \underset{\rank-k~A_k \in \mathbb{R}^{n\times n \times n} }{\min} \| A_k - A \|_1.
\end{align*}

Given a tensor $A\in \mathbb{R}^{n_1\times n_2 \times n_3}$, we define three matrices $ A_1 \in \mathbb{R}^{n_1 \times n_2 n_3}, A_2 \in \mathbb{R}^{n_2 \times n_3 n_1}, A_3 \in \mathbb{R}^{n_3 \times n_1 n_2}$ such that, for any $i\in [n_1], j \in [n_2], l \in [n_3]$,
\begin{align*}
A_{i,j,l} = ( A_{1})_{i, (j-1) \cdot n_3 + l} = ( A_{2} )_{ j, (l-1) \cdot n_1 + i } = ( {A}_3)_{l, (i-1) \cdot n_2 + j }.
\end{align*}

We fix $V^* \in \mathbb{R}^{n\times k}$ and $W^* \in \mathbb{R}^{n\times k}$, and use $V_1^*, V_2^*, \cdots, V_k^*$ to denote the columns of $V^*$ and $W_1^*, W_2^*, \cdots, W_k^*$ to denote the columns of $W^*$.

We consider the following optimization problem,
\begin{align*}
\min_{U_1, \cdots, U_k \in \mathbb{R}^n } \left\| \sum_{i=1}^k U_i \otimes V_i^* \otimes W_i^* - A \right\|_1,
\end{align*}
which is equivalent to
\begin{align*}
\min_{U_1, \cdots, U_k \in \mathbb{R}^n } \left\|
\begin{bmatrix}
U_1 & U_2 & \cdots & U_k
\end{bmatrix}
\begin{bmatrix}
 V_1^* \otimes W_1^*  \\
 V_2^* \otimes W_2^*  \\
\cdots \\
 V_k^* \otimes W_k^*
\end{bmatrix}
- A \right\|_1.
\end{align*}

We use matrix $Z_1$ to denote $ V^{*\top} \odot W^{*\top} \in \mathbb{R}^{k\times n^2}$ and matrix $U$ to denote $\begin{bmatrix} U_1 & U_2 & \cdots & U_k \end{bmatrix}$. Then we can obtain the following equivalent objective function,
\begin{align*}
\min_{U \in \mathbb{R}^{n\times k} } \| U Z_1  - A_1 \|_1.
\end{align*}

Choose an $\ell_1$ multiple regression cost preserving sketch $S_1 \in \mathbb{R}^{n^2 \times s_1}$ for $(Z_1^\top,A_1^\top)$. We can obtain the optimization problem,
\begin{align*}
\min_{U \in \mathbb{R}^{n\times k} } \| U Z_1 S_1 - A_1 S_1 \|_1 = \min_{U\in \mathbb{R}^{n\times k} } \sum_{i=1}^n \| U^i Z_1 S_1 - (A_1 S_1)^i \|_1,
\end{align*}
where $U^i$ denotes the $i$-th row of matrix $U\in \mathbb{R}^{n\times k}$ and $(A_1 S_1)^i$ denotes the $i$-th row of matrix $A_1 S_1$. Instead of solving it under the $\ell_1$-norm, we consider the $\ell_2$-norm relaxation,
\begin{align*}
\underset{U \in \mathbb{R}^{n\times k} }{\min} \| U Z_1 S_1 - A_1 S_1 \|_F^2 = \underset{U\in \mathbb{R}^{n\times k} }{\min} \sum_{i=1}^n \| U^i Z_1 S_1 - (A_1 S_1)^i \|_2^2.
\end{align*}
Let $ \wh{U} \in \mathbb{R}^{n\times k}$ denote the optimal solution of the above optimization problem. Then, $\wh{U} = A_1 S_1 (Z_1 S_1)^\dagger$. We plug $\wh{U}$ into the objective function under the $\ell_1$-norm. According to Claim~\ref{cla:ell2_relax_ell1_regression}, we have,
\begin{align*}
 \| \wh{U} Z_1 S_1 - A_1 S_1 \|_1= \sum_{i=1}^n \| \wh{U}^i Z_1 S_1 - (A_1 S_1)^i \|_1\leq \sqrt{s_1}\min_{U \in \mathbb{R}^{n\times k} } \| U Z_1 S_1 - A_1 S_1 \|_1.
\end{align*}

Since $S_1 \in \mathbb{R}^{n^2 \times s_1}$ satisfies Definition~\ref{def:l1_multiple_regression_cost_preserving_sketch}, we have
\begin{align*}
\| \wh{U} Z_1  - A_1  \|_1 \leq \alpha \underset{U\in \mathbb{R}^{n\times k}}{\min} \| U Z_1 - A_1 \|_1 = \alpha \OPT,
\end{align*}
where $\alpha=\sqrt{s_1}\beta$ and $\beta$ (see Definition~\ref{def:l1_multiple_regression_cost_preserving_sketch}) is a parameter which depends on which kind of sketching matrix we actually choose.
It implies
\begin{align*}
\|  \wh{U} \otimes V^* \otimes W^* - A \|_1 \leq \alpha \OPT.
\end{align*}

As a second step, we fix $\wh{U} \in \mathbb{R}^{n\times k}$ and $W^* \in \mathbb{R}^{n\times k}$, and convert tensor $A$ into matrix $A_2$. Let matrix $Z_2$ denote $\wh{U}^\top \odot W^{*\top} $. We consider the following objective function,
\begin{align*}
\min_{V \in \mathbb{R}^{n\times k} } \| V Z_2 -A_2  \|_1,
\end{align*}
and the optimal cost of it is at most $ \alpha \OPT$.

Choose an $\ell_1$ multiple regression cost preserving sketch $S_2 \in \mathbb{R}^{n^2 \times s_2}$ for $(Z_2^\top,A_2^\top)$, and sketch on the right of the objective function to obtain this new objective function,
\begin{align*}
\underset{V\in \mathbb{R}^{n\times k} }{\min} \| V Z_2 S_2 - A_2 S_2 \|_1 = \min_{U\in \mathbb{R}^{n\times k}} \sum_{i=1}^n \| V^i Z_2 S_2 - (A_2 S_2)^i \|_1,
\end{align*}
where $V^i$ denotes the $i$-th row of matrix $V$ and $(A_2 S_2)^i$ denotes the $i$-th row of matrix $A_2 S_2$. Instead of solving this under the $\ell_1$-norm, we consider the $\ell_2$-norm relaxation,
\begin{align*}
\underset{U \in \mathbb{R}^{n\times k} }{\min} \| V Z_2 S_2 - A_2 S_2 \|_F^2 = \underset{V\in \mathbb{R}^{n\times k} }{\min} \| V^i (Z_2 S_2) - (A_2 S_2)^i  \|_2^2.
\end{align*}

Let $\wh{V} \in \mathbb{R}^{n\times k}$ denote the optimal solution of the above problem. Then $\wh{V} = A_2 S_2 (Z_2 S_2)^\dagger$. By properties of the sketching matrix $S_2 \in \mathbb{R}^{n^2 \times s_2}$, we have,
\begin{align*}
\| \wh{V} Z_2 - A_2 \|_1 \leq \alpha \underset{V\in \mathbb{R}^{n\times k} }{\min} \| V Z_2  - A_2 \|_1 \leq  \alpha^2 \OPT,
\end{align*}
which implies
\begin{align*}
\|  \wh{U} \otimes \wh{V} \otimes W^* - A \|_1 \leq \alpha^2 \OPT.
\end{align*}

As a third step, we fix the matrices $\wh{U} \in \mathbb{R}^{n\times k}$ and $\wh{V}\in \mathbb{R}^{n \times k}$. We can convert tensor $A\in \mathbb{R}^{n\times n \times n}$ into matrix $A_3 \in \mathbb{R}^{n^2 \times n}$. Let matrix $Z_3$ denote $ \wh{U}^\top \odot \wh{V}^\top \in \mathbb{R}^{k\times n^2}$. We consider the following objective function,
\begin{align*}
\underset{W\in \mathbb{R}^{n\times k} }{\min} \| W Z_3 - A_3 \|_1,
\end{align*}
and the optimal cost of it is at most $\alpha^2 \OPT$.

Choose an $\ell_1$ multiple regression cost preserving sketch $S_3 \in \mathbb{R}^{n^2 \times s_3}$ for $(Z_3^\top,A_3^\top)$ and sketch on the right of the objective function to obtain the new objective function,
\begin{align*}
\underset{ W \in \mathbb{R}^{n\times k} }{ \min } \| W Z_3 S_3 - A_3 S_3 \|_1.
\end{align*}
Let $\wh{W} \in \mathbb{R}^{n\times k}$ denote the optimal solution of the above problem. Then $\wh{W} = A_3 S_3 (Z_3 S_3)^\dagger$. By properties of sketching matrix $S_3\in \mathbb{R}^{n^2 \times s_3}$, we have,
\begin{align*}
\| \wh{W} Z_3 - A_3 \|_1 \leq \alpha \underset{W\in \mathbb{R}^{n\times k} }{\min} \| W Z_3 - A_3 \|_1 \leq \alpha^3 \OPT.
\end{align*}
Thus, we obtain,
\begin{align*}
\min_{X_1 \in \mathbb{R}^{s_1\times k}, X_2 \in \mathbb{R}^{s_2 \times k}, X_3 \in \mathbb{R}^{s_3 \times k}} \left\| \sum_{i=1}^k (A_1S_1 X_1)_i \otimes (A_2 S_2 X_2)_i \otimes (A_3 S_3 X_3)_i - A \right\|_1 \leq \alpha^3 \OPT.
\end{align*}

{\bf Proof of (\RN{1})}
By Theorem C.1 in \cite{swz17}, we can use dense Cauchy transforms for $S_1, S_2, S_3$, and then $s_1=s_2=s_3 = O(k\log k)$ and $\alpha = O(\sqrt{k\log k} \log n)$.

{\bf Proof of (\RN{2})}
By Theorem C.1 in \cite{swz17}, we can use sparse Cauchy transforms for $S_1, S_2, S_3$, and then $s_1=s_2=s_3 = O(k^5\log^5 k)$ and $\alpha = O(k^{4.5} \log^{4.5} k \log n)$.

{\bf Proof of (\RN{3})}
By Theorem C.1 in \cite{swz17}, we can sample by Lewis weights. Then $S_1, S_2, S_3 \in \mathbb{R}^{n^2 \times n^2}$ are diagonal matrices, and each of them has $O(k\log k)$ nonzero rows. This gives $\alpha = O(\sqrt{k\log k} )$.

\end{proof}
\subsection{Polynomial in $k$ size reduction}\label{sec:l1_polyk_size_reduction}

\begin{definition}[Definition D.1 in~\cite{swz17}]\label{def:l1_no_dilation_bound}
Given a matrix $M\in\mathbb{R}^{n\times d},$ if matrix $S\in\mathbb{R}^{m\times n}$ satisfies
\begin{align*}
\|SM\|_1\leq \beta\|M\|_1,
\end{align*}
then $S$ has at most $\beta$ dilation on $M$.
\end{definition}

\begin{definition}[Definition D.2 in~\cite{swz17}]\label{def:l1_no_contraction_bound}
Given a matrix $U\in\mathbb{R}^{n\times k},$ if matrix $S\in\mathbb{R}^{m\times n}$ satisfies
\begin{align*}
\forall x\in\mathbb{R}^k, \|SUx\|_1\geq \frac{1}{\beta}\|Ux\|_1,
\end{align*}
then $S$ has at most $\beta$ contraction on $U$.
\end{definition}

\begin{theorem}\label{thm:l1_sketch_one_side}
Given a tensor $A\in\mathbb{R}^{n_1\times n_2\times n_3}$ and three matrices $V_1\in\mathbb{R}^{n_1\times b_1},V_2\in\mathbb{R}^{n_2\times b_2},V_3\in\mathbb{R}^{n_3\times b_3},$ let $X_1^*\in\mathbb{R}^{b_1\times k},X_2^*\in\mathbb{R}^{b_2\times k},X_3^*\in\mathbb{R}^{b_3\times k}$ satisfies
\begin{align*}
X_1^*,X_2^*,X_3^*=\underset{X_1\in\mathbb{R}^{b_1\times k},X_2\in\mathbb{R}^{b_2\times k},X_3\in\mathbb{R}^{b_3\times k}}{\arg\min}\|V_1X_1\otimes V_2X_2 \otimes V_3X_3 - A\|_1.
\end{align*}
Let $S\in\mathbb{R}^{m\times n}$ have at most $\beta_1\geq 1$ dilation on $V_1X_1^*\cdot ((V_2X_2^*)^\top \odot (V_3X_3^*)^\top)-A_1$ and $S$ have at most $\beta_2\geq 1$ contraction on $V_1$. If $\wh{X}_1\in\mathbb{R}^{b_1\times k},\wh{X}_2\in\mathbb{R}^{b_2\times k},\wh{X}_3\in\mathbb{R}^{b_3\times k}$ satisfies
\begin{align*}
\|SV_1\wh{X}_1\otimes V_2\wh{X}_2\otimes V_3\wh{X}_3 - SA \|_1 \leq \beta \underset{X_1\in\mathbb{R}^{b_1\times k},X_2\in\mathbb{R}^{b_2\times k},X_3\in\mathbb{R}^{b_3\times k}}{\min}\|SV_1X_1\otimes V_2X_2\otimes V_3X_3 - SA\|_1,
\end{align*}
where $\beta\geq 1$,
then
\begin{align*}
\|V_1\wh{X}_1\otimes V_2\wh{X}_2\otimes V_3\wh{X}_3 - A \|_1 \lesssim \beta_1\beta_2\beta \min_{X_1,X_2,X_3}\|V_1X_1\otimes V_2X_2\otimes V_3X_3 - A\|_1.
\end{align*}
\end{theorem}
The proof idea is similar to~\cite{swz17}.
\begin{proof}
Let $A,V_1,V_2,V_3,S,X_1^*,X_2^*,X_3^*,\beta_1,\beta_2$ be the same as stated in the theorem. Let
$\wh{X}_1\in\mathbb{R}^{b_1\times k},\wh{X}_2\in\mathbb{R}^{b_2\times k},\wh{X}_3\in\mathbb{R}^{b_3\times k}$ satisfy
\begin{align*}
\|SV_1\wh{X}_1\otimes V_2\wh{X}_2\otimes V_3\wh{X}_3 - SA \|_1 \leq \beta \underset{X_1\in\mathbb{R}^{b_1\times k},X_2\in\mathbb{R}^{b_2\times k},X_3\in\mathbb{R}^{b_3\times k}}{\min}\|SV_1X_1\otimes V_2X_2\otimes V_3X_3 - SA\|_1.
\end{align*}
We have,
\begin{align}\label{eq:l1_one_of_the_no_contraction_bound}
&\|SV_1\wh{X}_1\otimes V_2\wh{X}_2\otimes V_3\wh{X}_3 - SA \|_1\notag\\
\geq~ & \|SV_1\wh{X}_1\otimes V_2\wh{X}_2\otimes V_3\wh{X}_3 - SV_1X^*_1\otimes V_2X^*_2\otimes V_3X^*_3 \|_1-\|SV_1X^*_1\otimes V_2X^*_2\otimes V_3X^*_3-SA\|_1\notag\\
\geq~ & \frac{1}{\beta_2}\|V_1\wh{X}_1\otimes V_2\wh{X}_2\otimes V_3\wh{X}_3 - V_1X^*_1\otimes V_2X^*_2\otimes V_3X^*_3 \|_1-\beta_1\|V_1X^*_1\otimes V_2X^*_2\otimes V_3X^*_3-A\|_1\notag\\
\geq~ & \frac{1}{\beta_2}\|V_1\wh{X}_1\otimes V_2\wh{X}_2\otimes V_3\wh{X}_3 - A \|_1-\frac{1}{\beta_2}\|V_1X^*_1\otimes V_2X^*_2\otimes V_3X^*_3-A\|_1\notag\\
& -\beta_1\|V_1X^*_1\otimes V_2X^*_2\otimes V_3X^*_3-A\|_1\notag\\
=~ & \frac{1}{\beta_2}\|V_1\wh{X}_1\otimes V_2\wh{X}_2\otimes V_3\wh{X}_3 - A \|_1-(\frac{1}{\beta_2}+\beta_1)\|V_1X^*_1\otimes V_2X^*_2\otimes V_3X^*_3-A\|_1.
\end{align}
The first and the third inequality follow by the triangle inequalities. The second inequality follows using that
\begin{align*}
&\|SV_1\wh{X}_1\otimes V_2\wh{X}_2\otimes V_3\wh{X}_3 - SV_1X^*_1\otimes V_2X^*_2\otimes V_3X^*_3 \|_1\\
=~&\left\|SV_1(\wh{X}_1-X^*_1)\cdot\left((V_2(\wh{X}_2-X^*_2))^\top\odot(V_3(\wh{X}_3-X^*_3))^\top\right) \right\|_1\\
\geq~ & \frac{1}{\beta_2}\left\|V_1(\wh{X}_1-X^*_1)\cdot\left((V_2(\wh{X}_2-X^*_2))^\top\odot(V_3(\wh{X}_3-X^*_3))^\top\right) \right\|_1\\
\geq~ & \frac{1}{\beta_2} \|V_1\wh{X}_1\otimes V_2\wh{X}_2\otimes V_3\wh{X}_3 - V_1X^*_1\otimes V_2X^*_2\otimes V_3X^*_3 \|_1,
\end{align*}
and
\begin{align}
&\|SV_1X^*_1\otimes V_2X^*_2\otimes V_3X^*_3-SA\|_1\notag\\
=~&\|S(V_1X_1^*\cdot ((V_2X_2^*)^\top \odot (V_3X_3^*)^\top)-A_1)\|_1\notag\\
\leq~& \|V_1X_1^*\cdot ((V_2X_2^*)^\top \odot (V_3X_3^*)^\top)-A_1\|_1\notag\\
=~& \beta_1\|V_1X^*_1\otimes V_2X^*_2\otimes V_3X^*_3-A\|_1.\label{eq:l1_some_no_dilation}
\end{align}
Then, we have
\begin{align*}
&\|V_1\wh{X}_1\otimes V_2\wh{X}_2\otimes V_3\wh{X}_3 - A \|_1\\
\leq~ & \beta_2 \|SV_1\wh{X}_1\otimes V_2\wh{X}_2\otimes V_3\wh{X}_3 - SA \|_1+(1+\beta_1\beta_2) \|V_1X^*_1\otimes V_2X^*_2\otimes V_3X^*_3-A\|_1\\
\leq~ & \beta_2\beta \|SV_1X^*_1\otimes V_2X^*_2\otimes V_3X^*_3 - SA \|_1+(1+\beta_1\beta_2) \|V_1X^*_1\otimes V_2X^*_2\otimes V_3X^*_3-A\|_1\\
\leq~ & \beta_1\beta_2\beta \|V_1X^*_1\otimes V_2X^*_2\otimes V_3X^*_3 - A \|_1+(1+\beta_1\beta_2) \|V_1X^*_1\otimes V_2X^*_2\otimes V_3X^*_3-A\|_1\\
\leq~ & \beta(1+2\beta_1\beta_2)\|V_1X^*_1\otimes V_2X^*_2\otimes V_3X^*_3 - A \|_1.
\end{align*}
The first inequality follows by Equation~\eqref{eq:l1_one_of_the_no_contraction_bound}. The second inequality follows by
\begin{align*}
\|SV_1\wh{X}_1\otimes V_2\wh{X}_2\otimes V_3\wh{X}_3 - SA \|_1 \leq \beta \underset{X_1,X_2,X_3}{\min}\|SV_1X_1\otimes V_2X_2\otimes V_3X_3 - SA\|_1.
\end{align*}
 The third inequality follows by Equation~\eqref{eq:l1_some_no_dilation}. The final inequality follows using that $\beta\geq 1$.
\end{proof}

\begin{algorithm}[h]\caption{Reducing the Size of the Objective Function to $\poly(k)$}\label{alg:l1_polyk_size_reduction}
\begin{algorithmic}[1]
\Procedure{\textsc{L1PolyKSizeReduction}}{$A,V_1,V_2,V_3,n,b_1,b_2,b_3,k$} \Comment{Lemma \ref{lem:l1_polyk_size_reduction}}
\For{$i=1\to 3$}
	\State $c_i \leftarrow \wt{O}(b_i)$.
	\State Choose sampling and rescaling matrices $T_i\in \mathbb{R}^{c_i \times n}$ according to the Lewis weights of $V_i$.
	\State $\wh{V}_i \leftarrow T_i V_i \in \mathbb{R}^{c_i \times b_i}$.
\EndFor
\State $C\leftarrow A(T_1,T_2,T_3) \in \mathbb{R}^{c_1\times c_2 \times c_3}$.
\State \Return $\wh{V}_1$, $\wh{V}_2$, $\wh{V}_3$ and $C$.
\EndProcedure
\end{algorithmic}
\end{algorithm}

\begin{lemma}\label{lem:l1_polyk_size_reduction}
Let $\min(b_1,b_2,b_3)\geq k$. Given three matrices $V_1\in \mathbb{R}^{n\times b_1}$, $V_2 \in \mathbb{R}^{n\times b_2}$, and $V_3 \in \mathbb{R}^{n\times b_3}$, there exists an algorithm that takes $O(\nnz(A)) + n \poly(b_1,b_2,b_3)$ time and outputs a tensor $C\in \mathbb{R}^{c_1\times c_2\times c_3}$ and three matrices $\wh{V}_1\in \mathbb{R}^{c_1\times b_1}$, $\wh{V}_2 \in \mathbb{R}^{c_2\times b_2}$ and $\wh{V}_3 \in \mathbb{R}^{c_3 \times b_3}$ with $c_1=c_2=c_3=\poly(b_1,b_2,b_3)$, such that with probability $0.99$, for any $\alpha\geq 1$, if $X'_1,X'_2,X'_3$ satisfy that,
\begin{align*}
\left\| \sum_{i=1}^k (\wh{V}_1 X_1')_i \otimes (\wh{V}_2 X_2')_i \otimes (\wh{V}_3 X_3')_i - C \right\|_1 \leq \alpha \underset{X_1, X_2, X_3}{\min} \left\| \sum_{i=1}^k (\wh{V}_1 X_1)_i \otimes (\wh{V}_2 X_2)_i \otimes (\wh{V}_3 X_3)_i - C \right\|_1,
\end{align*}
then,
\begin{align*}
\left\| \sum_{i=1}^k (V_1 X_1')_i \otimes ( V_2 X_2')_i \otimes (V_3 X_3')_i - A \right\|_1 \lesssim \alpha \min_{X_1, X_2, X_3}\left\| \sum_{i=1}^k ({V}_1 X_1)_i \otimes ({V}_2 X_2)_i \otimes ({V}_3 X_3)_i - A \right\|_1.
\end{align*}
\end{lemma}
\begin{proof}
For simplicity, we define $\OPT$ to be
\begin{align*}
\min_{X_1, X_2, X_3}\left\| \sum_{i=1}^k ({V}_1 X_1)_i \otimes ({V}_2 X_2)_i \otimes ({V}_3 X_3)_i - A \right\|_1.
\end{align*}
Let $T_1\in \mathbb{R}^{c_1\times n}$ sample according to the Lewis weights of $V_1\in \mathbb{R}^{n\times b_1}$, where $c_1=\wt{O}(b_1)$. Let $T_2\in \mathbb{R}^{c_2\times n}$ sample according to the Lewis weights of $V_2\in \mathbb{R}^{n\times b_2}$, where $c_2=\wt{O}(b_2)$. Let $T_3\in \mathbb{R}^{c_3\times n}$ sample according to the Lewis weights of $V_3\in \mathbb{R}^{n\times b_3}$, where $c_3=\wt{O}(b_3)$.

For any $\alpha\geq 1,$ let $X'_1\in\mathbb{R}^{b_1\times k},X'_2\in\mathbb{R}^{b_2\times k},X'_3\in\mathbb{R}^{b_3\times k}$ satisfy
\begin{align*}
&\|T_1V_1X'_1 \otimes T_2V_2X'_2 \otimes T_3V_3X'_3 - A(T_1,T_2,T_3)\|_1 \\
\leq~& \alpha \min_{X_1\in\mathbb{R}^{b_1\times k},X_2\in\mathbb{R}^{b_2\times k},X_3\in\mathbb{R}^{b_3\times k}}\|T_1V_1X_1 \otimes T_2V_2X_2 \otimes T_3V_3X_3 - A(T_1,T_2,T_3)\|_1.
\end{align*}
First, we regard $T_1$ as the sketching matrix for the remainder. Then by Lemma D.11 in \cite{swz17} and Theorem~\ref{thm:l1_sketch_one_side}, we have
\begin{align*}
&\|V_1X'_1 \otimes T_2V_2X'_2 \otimes T_3V_3X'_3 - A(I,T_2,T_3)\|_1 \\
\lesssim~& \alpha \min_{X_1\in\mathbb{R}^{b_1\times k},X_2\in\mathbb{R}^{b_2\times k},X_3\in\mathbb{R}^{b_3\times k}}\|V_1X_1 \otimes T_2V_2X_2 \otimes T_3V_3X_3 - A(I,T_2,T_3)\|_1.
\end{align*}
Second, we regard $T_2$ as a sketching matrix for $V_1X_1 \otimes V_2X_2 \otimes T_3V_3X_3 - A(I,I,T_3)$. Then by Lemma D.11 in \cite{swz17} and Theorem~\ref{thm:l1_sketch_one_side}, we have
\begin{align*}
&\|V_1X'_1 \otimes V_2X'_2 \otimes T_3V_3X'_3 - A(I,I,T_3)\|_1 \\
\lesssim~& \alpha \min_{X_1\in\mathbb{R}^{b_1\times k},X_2\in\mathbb{R}^{b_2\times k},X_3\in\mathbb{R}^{b_3\times k}}\|V_1X_1 \otimes V_2X_2 \otimes T_3V_3X_3 - A(I,I,T_3)\|_1.
\end{align*}
Third, we regard $T_3$ as a sketching matrix for $V_1X_1 \otimes V_2X_2 \otimes V_3X_3 - A$. Then by Lemma D.11 in \cite{swz17} and Theorem~\ref{thm:l1_sketch_one_side}, we have
\begin{align*}
\left\| \sum_{i=1}^k (V_1 X_1')_i \otimes ( V_2 X_2')_i \otimes (V_3 X_3')_i - A \right\|_1 \lesssim \alpha \min_{X_1, X_2, X_3}\left\| \sum_{i=1}^k ({V}_1 X_1)_i \otimes ({V}_2 X_2)_i \otimes ({V}_3 X_3)_i - A \right\|_1.
\end{align*}
\end{proof}

\begin{lemma}\label{lem:l1_polyk_size_reduction_for regression}
Given tensor $A\in \mathbb{R}^{n_1\times n_2\times n_3}$, and two matrices $U\in \mathbb{R}^{n_1\times s}, V\in \mathbb{R}^{n_2\times s}$ with $\rank(U)=r$, let $T\in\mathbb{R}^{t\times n_1}$ be a sampling/rescaling matrix according to the Lewis weights of $U$ with $t=\wt{O}(r)$. Then with probability at least $0.99$, for all $X'\in\mathbb{R}^{n_3\times s},\alpha\geq 1$ which satisfy
\begin{align*}
\|T_1U\otimes V\otimes X'-T_1A\|_1\leq \alpha\cdot\min_{X\in\mathbb{R}^{n_3\times s}}\|T_1U\otimes V\otimes X-T_1A\|_1,
\end{align*}
it holds that
\begin{align*}
\|U\otimes V\otimes X'-A\|_1\lesssim \alpha\cdot\min_{X\in\mathbb{R}^{n_3\times s}}\|U\otimes V\otimes X-A\|_1.
\end{align*}
\end{lemma}
The proof is similar to the proof of Lemma~\ref{lem:l1_polyk_size_reduction}.
\begin{proof}
Let $X^*=\underset{X\in\mathbb{R}^{n_3\times s}}{\arg\min}\|U\otimes V\otimes X-A\|_1.$ Then according to Lemma D.11 in~\cite{swz17}, $T$ has at most constant dilation (Definition~\ref{def:l1_no_dilation_bound}) on $U\cdot(V^\top \odot (X^*)^\top) - A_1$, and has at most constant contraction (Definition~\ref{def:l1_no_contraction_bound}) on $U$. We first look at
\begin{align*}
&\|TU\otimes V\otimes X' - TA\|_1\\
= ~& \|TU\cdot(V^\top \odot (X')^\top) - TA_1\|_1\\
\geq~ & \|TU\cdot( (V^\top \odot (X')^\top) - (V^\top \odot (X^*)^\top)) \|_1-\|TU\cdot(V^\top \odot (X^*)^\top) - TA_1\|_1\\
\geq~ & \frac{1}{\beta_2}\|U\cdot( (V^\top \odot (X')^\top)-A_1\|_1-(\frac{1}{\beta_2}+\beta_1)\|U\cdot(V^\top \odot (X^*)^\top) - A_1\|_1,
\end{align*}
where $\beta_1\geq 1,\beta_2\geq 1$ are two constants. Then we have:
\begin{align*}
&\|U\otimes V\otimes X' - A\|_1\\
\leq~& \beta_2 \|TU\cdot(V^\top \odot (X')^\top) - TA_1\|_1 + (1+\beta_1\beta_2)\|U\cdot(V^\top \odot (X^*)^\top) - A_1\|_1\\
\leq~& \alpha\beta_2\|TU\cdot(V^\top \odot (X^*)^\top) - TA_1\|_1 + (1+\beta_1\beta_2)\|U\cdot(V^\top \odot (X^*)^\top) - A_1\|_1\\
\leq~& \alpha\beta_1\beta_2\|U\cdot(V^\top \odot (X^*)^\top) - A_1\|_1 + (1+\beta_1\beta_2)\|U\cdot(V^\top \odot (X^*)^\top) - A_1\|_1\\
\lesssim~& \alpha\|U\otimes V\otimes X^*-A\|_1.
\end{align*}
\end{proof}

\begin{corollary}\label{cor:l1_polyk_size_reduction}
Given tensor $A\in \mathbb{R}^{n\times n\times n}$, and two matrices $U\in \mathbb{R}^{n\times s}, V\in \mathbb{R}^{n\times s}$ with $\rank(U)=r_1,\rank(V)=r_2$, let $T_1\in\mathbb{R}^{t_1\times n}$ be a sampling/rescaling matrix according to the Lewis weights of $U$, and let $T_2\in\mathbb{R}^{t_2\times n}$ be a sampling/rescaling matrix according to the Lewis weights of $V$ with $t_1=\wt{O}(r_1),t_2=\wt{O}(r_2)$. Then with probability at least $0.99$, for all $X'\in\mathbb{R}^{n\times s},\alpha\geq 1$ which satisfy
\begin{align*}
\|T_1U\otimes T_2V\otimes X'-A(T_1,T_2,I)\|_1\leq \alpha\cdot\min_{X\in\mathbb{R}^{n\times s}}\|T_1U\otimes T_2V\otimes X-A(T_1,T_2,I)\|_1,
\end{align*}
it holds that
\begin{align*}
\|U\otimes V\otimes X'-A\|_1\lesssim \alpha\cdot\min_{X\in\mathbb{R}^{n\times s}}\|U\otimes V\otimes X-A\|_1.
\end{align*}
\end{corollary}
\begin{proof}
We apply Lemma~\ref{lem:l1_polyk_size_reduction_for regression} twice: if
\begin{align*}
\|T_1U\otimes T_2V\otimes X'-A(T_1,T_2,I)\|_1\leq \alpha\cdot\min_{X\in\mathbb{R}^{n\times s}}\|T_1U\otimes T_2V\otimes X-A(T_1,T_2,I)\|_1,
\end{align*}
then
\begin{align*}
\|U\otimes T_2V\otimes X'-A(I,T_2,I)\|_1\lesssim \alpha\cdot\min_{X\in\mathbb{R}^{n\times s}}\|U\otimes T_2V\otimes X-A(I,T_2,I)\|_1.
\end{align*}
Then, we have
\begin{align*}
\|U\otimes V\otimes X'-A\|_1\lesssim \alpha\cdot\min_{X\in\mathbb{R}^{n\times s}}\|U\otimes V\otimes X-A\|_1.
\end{align*}
\end{proof}

\subsection{Solving small problems}\label{sec:l1_solving_small_problems}

\begin{theorem}\label{thm:l1_solving_small_problems}
Let $\max_{i} \{t_i, d_i\} \leq n$. Given a $t_1 \times t_2 \times t_3$ tensor $A$ and three matrices: a $t_1 \times d_1$ matrix $T_1$, a $t_2 \times d_2$ matrix $T_2$, and a $t_3 \times d_3$ matrix $T_3$, if for $\delta > 0$ there exists a solution to
\begin{align*}
\min_{X_1,X_2,X_3} \left\| \sum_{i=1}^k (T_1 X_1)_i \otimes (T_2 X_2)_i \otimes (T_3 X_3)_i - A \right\|_1 := \OPT,
\end{align*}
such that each entry of $X_i$ can be expressed using $O(n^\delta)$ bits, then there exists an algorithm that takes $ n^{O(\delta)} \cdot 2^{ O( d_1 k+d_2 k+d_3 k)}$ time and outputs three matrices: $\wh{X}_1$, $\wh{X}_2$, and $\wh{X}_3$ such that $\| (T_1 \wh{X}_1)\otimes (T_2 \wh{X}_2) \otimes (T_3\wh{X}_3) - A\|_1 =\OPT$.
\end{theorem}
\begin{proof}
For each $i\in [3]$, we can create $t_i\times d_i$ variables to represent matrix $X_i$. Let $x$ denote the list of these variables. Let $B$ denote tensor $\sum_{i=1}^k (T_1 X_1)_i \otimes (T_2X_2)_i\otimes (T_3X_3)_i$. Then we can write the following objective function,
\begin{align*}
\min_{x} \sum_{i=1}^{t_1} \sum_{j=1}^{t_2} \sum_{l=1}^{t_3} |B_{i,j,l}(x) -A_{i,j,l}|.
\end{align*}
To remove the $|\cdot|$, we create $t_1t_2t_3$ extra variables $\sigma_{i,j,l}$. Then we obtain the objective function:
\begin{align*}
\min_{x,\sigma} & ~ \sum_{i=1}^{t_1} \sum_{j=1}^{t_2} \sum_{l=1}^{t_3} \sigma_{i,j,l} (B_{i,j,l}(x) -A_{i,j,l}) \\
\text{s.t.} & ~ \sigma_{i,j,l}^2 =1, \\
& ~ \sigma_{i,j,l} (B_{i,j,l}(x) -A_{i,j,l}) \geq 0, \\
& ~ \| x \|_2^2 + \| \sigma \|_2^2 \leq 2^{O(n^\delta)}
\end{align*}
where the last constraint is unharmful, because there exists a solution that can be written using $O(n^\delta)$ bits. Note that the number of inequality constraints in the above system is $O(t_1t_2t_3)$, the degree is $O(1)$, and the number of variables is $v=(t_1t_2t_3+d_1k+d_2k+d_3k)$. Thus by Theorem~\ref{thm:minimum_positive}, we know that the minimum nonzero cost is at least
\begin{align*}
(2^{O(n^\delta)} )^{-2^{\wt{O} ( v ) }}.
\end{align*}
It is immediate that the upper bound on cost is at most $2^{O(n^\delta)}$, and thus the number of binary search steps is at most $\log(2^{O(n^\delta)}) 2^{\wt{O}(v)}$. In each step of the binary search, we need to choose a cost $C$ between the lower bound and the upper bound, and write down the polynomial system,
\begin{align*}
 & ~ \sum_{i=1}^{t_1} \sum_{j=1}^{t_2} \sum_{l=1}^{t_3} \sigma_{i,j,l} (B_{i,j,l}(x) -A_{i,j,l}) \leq C, \\
 & ~ \sigma_{i,j,l}^2 =1, \\
& ~ \sigma_{i,j,l} (B_{i,j,l}(x) -A_{i,j,l}) \geq 0, \\
& ~ \| x \|_2^2 + \| \sigma \|_2^2 \leq 2^{O(n^\delta)}.
\end{align*}
Using Theorem~\ref{thm:decision_solver}, we can determine if there exists a solution to the above polynomial system. Since the number of variables is $v$, and the degree is $O(1)$, the number of inequality constraints is $t_1 t_2 t_2$. Thus, the running time is
\begin{align*}
\poly(\text{bitsize}) \cdot (\#\constraints \cdot \degree)^{\#\variables} = n^{O(\delta)} 2^{\wt{O}(v)}
\end{align*}
\end{proof}

\subsection{Bicriteria algorithms}\label{sec:l1_bicriteria_algorithm}

We present several bicriteria algorithms with different tradeoffs. We first present an algorithm that runs in nearly linear time and outputs a solution with rank $\wt{O}(k^3)$ in Theorem~\ref{thm:l1_bicriteria_algorithm_rank_k3_nearly_input_sparsity_time}. Then we show an algorithm that runs in $\nnz(A)$ time but outputs a solution with rank $\poly(k)$ in Theorem~\ref{thm:l1_bicriteria_algorithm_rank_k15_input_sparsity_time}. Then we explain an idea which is able to decrease the cubic rank to quadratic rank, and thus we can obtain Theorem~\ref{thm:l1_bicriteria_algorithm_rank_k2_nearly_input_sparsity_time} and Theorem~\ref{thm:l1_bicriteria_algorithm_rank_k10_input_sparsity_time}.
\subsubsection{Input sparsity time}

\begin{algorithm}[h]\caption{$\ell_1$-Low Rank Approximation, Bicriteria Algorithm, $\rank$-$\wt{O}(k^3)$, Nearly Input Sparsity Time}\label{alg:l1_bicriteria_algorithm_rank_k3_nearly_input_sparsity_time}
\begin{algorithmic}[1]
\Procedure{\textsc{L1BicriteriaAlgorithm}}{$A,n,k$} \Comment{Theorem \ref{thm:l1_bicriteria_algorithm_rank_k3_nearly_input_sparsity_time}}
\State $s_1\leftarrow s_2 \leftarrow s_3 \leftarrow \wt{O}(k)$.
\State For each $i\in [3]$, choose $S_i\in \mathbb{R}^{n^2 \times s_i}$ to be a dense Cauchy transform. \Comment{Part (\RN{1}) of Theorem~\ref{sec:l1_existence_results}}
\State Compute $A_1 \cdot S_1$, $A_2 \cdot S_2$, $A_3 \cdot S_3$.
\State $Y_1, Y_2, Y_3, C\leftarrow$\textsc{L1PolyKSizeReduction}($A,A_1S_1, A_2S_2,A_3S_3,n,s_1,s_2,s_3,k$) \Comment{Algorithm~\ref{alg:l1_polyk_size_reduction}}
\State Form objective function \begin{align*}\min_{X\in \mathbb{R}^{s_1 \times s_2 \times s_3}} \left\| \sum_{i=1}^{s_1} \sum_{j=1}^{s_2} \sum_{l=1}^{s_3} X_{i,j,l} (Y_1)_i \otimes (Y_2)_j \otimes (Y_3)_l -C \right\|_1.\end{align*}
\State Run $\ell_1$-regression solver to find $X$.
\State \Return $A_1S_1$, $A_2S_2$, $A_3S_3$ and $X$.
\EndProcedure
\end{algorithmic}
\end{algorithm}

\begin{theorem}\label{thm:l1_bicriteria_algorithm_rank_k3_nearly_input_sparsity_time}
Given a $3$rd order tensor $A\in \mathbb{R}^{n\times n \times n}$, for any $k\geq 1$, $\epsilon \in (0,1)$, let $r=\wt{O}(k^3)$. There exists an algorithm which takes $\nnz(A)\cdot \wt{O}(k) + O(n) \poly(k) + \poly(k)$ time and outputs three matrices $U,V,W\in \mathbb{R}^{n\times r}$ such that
\begin{align*}
\left\| \sum_{i=1}^r U_i \otimes V_i \otimes W_i - A \right\|_1 \leq \wt{O}(k^{3/2}) \log^3 n \min_{\rank-k~A_k} \| A_k - A\|_1
\end{align*}
holds with probability $9/10$.
\end{theorem}

\begin{proof}
We first choose three dense Cauchy transforms $S_i\in \mathbb{R}^{n^2\times s_i}$. According to Section~\ref{sec:def_cauchy_pstable}, for each $i\in[3]$, $A_i S_i$ can be computed in $\nnz(A) \cdot \wt{O}(k)$ time. Then we apply Lemma~\ref{lem:l1_polyk_size_reduction} (Algorithm \ref{alg:l1_polyk_size_reduction}). We obtain three matrices $Y_1,Y_2,Y_3$ and a tensor $C$. Note that for each $i\in [3]$, $Y_i$ can be computed in $ n \poly(k)$ time. Because $C= A(T_1,T_2,T_3)$ and $T_1,T_2,T_3 \in \mathbb{R}^{n\times \wt{O}(k)}$ are three sampling and rescaling matrices, $C$ can be computed in $\nnz(A) + \wt{O}(k^3)$ time. At the end, we just need to run an $\ell_1$-regression solver to find the solution to the problem,
\begin{align*}
\min_{X\in \mathbb{R}^{s_1 \times s_2 \times s_3}} \left\| \sum_{i=1}^{s_1} \sum_{j=1}^{s_2} \sum_{l=1}^{s_3} X_{i,j,l} (Y_1)_i \otimes (Y_2)_j \otimes (Y_3)_j  \right\|_1,
\end{align*}
where $(Y_1)_i$ denotes the $i$-th column of matrix $Y_1$. Since the size of the above problem is only $\poly(k)$, this can be solved in $\poly(k)$ time.
\end{proof}

\begin{algorithm}[h]\caption{$\ell_1$-Low Rank Approximation, Bicriteria Algorithm, $\rank$-$\poly(k)$, Input Sparsity Time}\label{alg:l1_bicriteria_algorithm_rank_k15_input_sparsity_time}
\begin{algorithmic}[1]
\Procedure{\textsc{L1BicriteriaAlgorithm}}{$A,n,k$} \Comment{Theorem \ref{thm:l1_bicriteria_algorithm_rank_k15_input_sparsity_time}}
\State $s_1\leftarrow s_2 \leftarrow s_3 \leftarrow \wt{O}(k^5)$.
\State For each $i\in [3]$, choose $S_i\in \mathbb{R}^{n^2 \times s_i}$ to be a sparse Cauchy transform. \Comment{Part (\RN{2}) of Theorem~\ref{thm:l1_existence_results}}
\State Compute $A_1 \cdot S_1$, $A_2 \cdot S_2$, $A_3 \cdot S_3$.
\State $Y_1, Y_2, Y_3, C\leftarrow$\textsc{L1PolyKSizeReduction}($A,A_1S_1, A_2S_2,A_3,S_3,n,s_1,s_2,s_3,k$) \Comment{Algorithm~\ref{alg:l1_polyk_size_reduction}}
\State Form objective function \begin{align*}\min_{X\in \mathbb{R}^{s_1 \times s_2 \times s_3}} \left\| \sum_{i=1}^{s_1} \sum_{j=1}^{s_2} \sum_{l=1}^{s_3} X_{i,j,l} (Y_1)_i \otimes (Y_2)_j \otimes (Y_3)_l -C \right\|_1.\end{align*}
\State Run $\ell_1$-regression solver to find $X$.
\State \Return $A_1S_1$, $A_2S_2$, $A_3S_3$ and $X$.
\EndProcedure
\end{algorithmic}
\end{algorithm}

\begin{theorem}\label{thm:l1_bicriteria_algorithm_rank_k15_input_sparsity_time}
Given a $3$rd order tensor $A\in \mathbb{R}^{n\times n \times n}$, for any $k\geq 1$, $\epsilon \in (0,1)$, let $r=\wt{O}(k^{15})$. There exists an algorithm that takes $\nnz(A) + O(n) \poly(k) + \poly(k)$ time and outputs three matrices $U,V,W\in \mathbb{R}^{n\times r}$ such that
\begin{align*}
\left\| \sum_{i=1}^r U_i \otimes V_i \otimes W_i - A \right\|_1 \leq \poly( k, \log n ) \min_{\rank-k~A_k} \| A_k - A\|_1
\end{align*}
holds with probability $9/10$.
\end{theorem}

\begin{proof}
We first choose three dense Cauchy transforms $S_i\in \mathbb{R}^{n^2\times s_i}$. According to Section~\ref{sec:def_cauchy_pstable}, for each $i\in[3]$, $A_i S_i$ can be computed in $O(\nnz(A))$ time. Then we apply Lemma~\ref{lem:l1_polyk_size_reduction} (Algorithm \ref{alg:l1_polyk_size_reduction}), and can obtain three matrices $Y_1,Y_2,Y_3$ and a tensor $C$. Note that for each $i\in [3]$, $Y_i$ can be computed in $O(n) \poly(k)$ time. Because $C= A(T_1,T_2,T_3)$ and $T_1,T_2,T_3 \in \mathbb{R}^{n\times \wt{O}(k)}$ are three sampling and rescaling matrices, $C$ can be computed in $\nnz(A) + \wt{O}(k^3)$ time. At the end, we just need to run an $\ell_1$-regression solver to find the solution to the problem,
\begin{align*}
\min_{X\in \mathbb{R}^{s_1 \times s_2 \times s_3}} \left\| \sum_{i=1}^{s_1} \sum_{j=1}^{s_2} \sum_{l=1}^{s_3} X_{i,j,l} (Y_1)_i \otimes (Y_2)_j \otimes (Y_3)_l - C  \right\|_1,
\end{align*}
where $(Y_1)_i$ denotes the $i$-th column of matrix $Y_1$. Since the size of the above problem is only $\poly(k)$, it can be solved in $\poly(k)$ time.
\end{proof}

\subsubsection{Improving cubic rank to quadratic rank}

\begin{algorithm}[h]\caption{$\ell_1$-Low Rank Approximation, Bicriteria Algorithm, $\rank$-$\wt{O}(k^2)$, Nearly Input Sparsity Time}\label{alg:l1_bicriteria_algorithm_rank_k2_nearly_input_sparsity_time}
\begin{algorithmic}[1]
\Procedure{\textsc{L1BicriteriaAlgorithm}}{$A,n,k$} \Comment{Theorem \ref{thm:l1_bicriteria_algorithm_rank_k2_nearly_input_sparsity_time}}
\State $s_1\leftarrow s_2 \leftarrow s_3 \leftarrow \wt{O}(k)$.
\State For each $i\in [3]$, choose $S_i\in \mathbb{R}^{n^2 \times s_i}$ to be a dense Cauchy transform. \Comment{Part (\RN{1}) of Theorem~\ref{sec:l1_existence_results}}
\State Compute $A_1 \cdot S_1$, $A_2 \cdot S_2$.
\State For each $i\in [2]$, choose $T_i$ to be a sampling and rescaling diagonal matrix according to the Lewis weights of $A_i S_i$, with $t_i=\wt{O}(k)$ nonzero entries.
\State $C\leftarrow A(T_1,T_2,I)$.
\State $B^{i+(j-1)s_1} \leftarrow \vect( (T_1A_1 S_1)_i \otimes (T_2A_2S_2)_j ), \forall i\in[s_1],j\in[s_2]$.
\State Form objective function $\min_{W}\| W B - C_3\|_1$
\State Run $\ell_1$-regression solver to find $\wh{W}$.
\State Construct $\wh{U}$ by using $A_1 S_1$ according to Equation~\eqref{eq:l1_bicriteria_rank_k2_whU}.
\State Construct $\wh{V}$ by using $A_2 S_2$ according to Equation~\eqref{eq:l1_bicriteria_rank_k2_whV}.
\State \Return $\wh{U},\wh{V},\wh{W}$.
\EndProcedure
\end{algorithmic}
\end{algorithm}

\begin{theorem}\label{thm:l1_bicriteria_algorithm_rank_k2_nearly_input_sparsity_time}
Given a $3$rd order tensor $A\in \mathbb{R}^{n\times n \times n}$, for any $k\geq 1$, $\epsilon \in (0,1)$, let $r=\wt{O}(k^2)$. There exists an algorithm which takes $\nnz(A)\cdot \wt{O}(k) + O(n) \poly(k) + \poly(k)$ time and outputs three matrices $U,V,W\in \mathbb{R}^{n\times r}$ such that
\begin{align*}
\left\| \sum_{i=1}^r U_i \otimes V_i \otimes W_i - A \right\|_1 \leq \wt{O}(k^{3/2}) \log^3 n \min_{\rank-k~A_k} \| A_k - A\|_1
\end{align*}
holds with probability $9/10$.
\end{theorem}
\begin{proof}
Let $\OPT=\underset{A_k\in\mathbb{R}^{n\times n \times n}}{\min} \| A_k -A \|_1.$
We first choose three dense Cauchy transforms $S_i \in \mathbb{R}^{n^2 \times s_i}$, $\forall i\in [3]$. According to Section~\ref{sec:def_cauchy_pstable}, for each $i\in [3]$, $A_iS_i$ can be computed in $\nnz(A) \cdot \wt{O}(k)$ time. Then we choose $T_i$ to be a sampling and rescaling diagonal matrix according to the Lewis weights of $A_iS_i$, $\forall i\in [2]$.

According to Theorem~\ref{thm:l1_existence_results}, we have
\begin{align*}
\min_{X_1\in\mathbb{R}^{s_1\times k},X_2\in\mathbb{R}^{s_2\times k},X_3\in\mathbb{R}^{s_3\times k}}\left\| \sum_{l=1}^k (A_1 S_1 X_1)_l \otimes (A_2S_2 X_2)_l \otimes (A_3S_3 X_3)_l - A \right\|_1 \leq \wt{O}(k^{1.5})\log^3 n \OPT
\end{align*}
Now we fix an $l$ and we have:
\begin{align*}
&(A_1 S_1 X_1)_l \otimes (A_2S_2 X_2)_l \otimes (A_3S_3 X_3)_l\\
=&\left(\sum_{i=1}^{s_1} (A_1S_1)_i (X_1)_{i,l}\right)\otimes \left(\sum_{j=1}^{s_2} (A_2S_2)_j (X_2)_{j,l}\right)\otimes (A_3S_3 X_3)_l\\
=&\sum_{i=1}^{s_1}\sum_{j=1}^{s_2} (A_1S_1)_i \otimes (A_2S_2)_j \otimes (A_3S_3 X_3)_l(X_1)_{i,l}(X_2)_{j,l}
\end{align*}

Thus, we have
\begin{align}\label{eq:l1_A1S1otimesA2S2}
\min_{X_1,X_2,X_3}\left\| \sum_{i=1}^{s_1}\sum_{j=1}^{s_2} (A_1S_1)_i \otimes (A_2S_2)_j \otimes \left( \sum_{l=1}^k(A_3S_3 X_3)_l(X_1)_{i,l}(X_2)_{j,l}\right) - A \right\|_1 \leq \wt{O}(k^{1.5})\log^3 n \OPT.
\end{align}

We create matrix $\wh{U}\in \mathbb{R}^{n\times s_1 s_2}$ by copying matrix $A_1S_1$ $s_2$ times, i.e.,
\begin{align}\label{eq:l1_bicriteria_rank_k2_whU}
\wh{U} = \begin{bmatrix} A_1 S_1 & A_1 S_1 & \cdots & A_1 S_1 \end{bmatrix}.
\end{align}
We create matrix $\wh{V}\in \mathbb{R}^{n\times s_1 s_2}$ by copying the $i$-th column of $A_2 S_2$ a total of $s_1$ times into the columns $(i-1)s_1, \cdots, is_1$ of $\wh{V}$, for each $i\in [s_2]$, i.e.,
\begin{align}\label{eq:l1_bicriteria_rank_k2_whV}
\wh{V} =
\begin{bmatrix}
(A_2S_2)_1 & \cdots & (A_2 S_2)_1 &  (A_2S_2)_2 & \cdots & (A_2 S_2)_2 & \cdots (A_2S_2)_{s_2} & \cdots & (A_2 S_2)_{s_2}.
\end{bmatrix}
\end{align}
According to Equation~\eqref{eq:l1_A1S1otimesA2S2}, we have:
\begin{align*}
\underset{W\in \mathbb{R}^{n\times s_1 s_2} }{\min} \| \wh{U}\otimes \wh{V}\otimes W - A \|_1\leq \wt{O}(k^{1.5})\log^3 n\cdot\OPT.
\end{align*}
Let
\begin{align*}
\wh{W}=\underset{W\in \mathbb{R}^{n\times s_1 s_2}}{\arg\min} \| T_1\wh{U}\otimes T_2\wh{V}\otimes W - A(T_1,T_2,I) \|_1.
\end{align*}
Due to Corollary~\ref{cor:l1_polyk_size_reduction}, we have
\begin{align*}
\|\wh{U}\otimes\wh{V}\otimes \wh{W}-A\|_1\leq \wt{O}(k^{1.5})\log^3 n\cdot\OPT.
\end{align*}
Putting it all together, we have that $\wh{U},\wh{V},\wh{W}$ gives a rank-$\wt{O}(k^2)$ bicriteria algorithm to the original problem.
\end{proof}

\begin{algorithm}[h]\caption{$\ell_1$-Low Rank Approximation, Bicriteria Algorithm, $\rank$-$\poly(k)$, Input Sparsity Time}\label{alg:l1_bicriteria_algorithm_rank_k10_input_sparsity_time}
\begin{algorithmic}[1]
\Procedure{\textsc{L1BicriteriaAlgorithm}}{$A,n,k$} \Comment{Theorem \ref{thm:l1_bicriteria_algorithm_rank_k10_input_sparsity_time}}
\State $s_1\leftarrow s_2 \leftarrow s_3 \leftarrow \wt{O}(k^5)$.
\State For each $i\in [3]$, choose $S_i\in \mathbb{R}^{n^2 \times s_i}$ to be a sparse Cauchy transform. \Comment{Part (\RN{2}) of Theorem~\ref{sec:l1_existence_results}}
\State Compute $A_1 \cdot S_1$, $A_2 \cdot S_2$.
\State For each $i\in [2]$, choose $T_i$ to be a sampling and rescaling diagonal matrix according to the Lewis weights of $A_i S_i$, with $t_i=\wt{O}(k)$ nonzero entries.
\State $C\leftarrow A(T_1,T_2,I)$.
\State $B^{i+(j-1)s_1} \leftarrow \vect( (T_1A_1 S_1)_i \otimes (T_2A_2S_2)_j ), \forall i\in[s_1],j\in[s_2]$.
\State Form objective function $\min_{W}\| W B - C_3\|_1$.
\State Run $\ell_1$-regression solver to find $\wh{W}$.
\State Construct $\wh{U}$ by using $A_1 S_1$ according to Equation~\eqref{eq:l1_bicriteria_rank_k2_whU}.
\State Construct $\wh{V}$ by using $A_2 S_2$ according to Equation~\eqref{eq:l1_bicriteria_rank_k2_whV}.
\State \Return $\wh{U},\wh{V},\wh{W}$.
\EndProcedure
\end{algorithmic}
\end{algorithm}

\begin{theorem}\label{thm:l1_bicriteria_algorithm_rank_k10_input_sparsity_time}
Given a $3$rd order tensor $A\in \mathbb{R}^{n\times n \times n}$, for any $k\geq 1$, $\epsilon \in (0,1)$, let $r=\wt{O}(k^{10})$. There exists an algorithm which takes $\nnz(A) + O(n) \poly(k) + \poly(k)$ time and outputs three matrices $U,V,W\in \mathbb{R}^{n\times r}$ such that
\begin{align*}
\left\| \sum_{i=1}^r U_i \otimes V_i \otimes W_i - A \right\|_1 \leq \poly( k, \log n ) \min_{\rank-k~A_k} \| A_k - A\|_1
\end{align*}
holds with probability $9/10$.
\end{theorem}
\begin{proof}
The proof is similar to the proof of Theorem~\ref{thm:l1_bicriteria_algorithm_rank_k2_nearly_input_sparsity_time}. The only difference is that instead of choosing dense Cauchy matrices $S_1,S_2$, we choose sparse Cauchy matrices.
\end{proof}

Notice that if we firstly apply a sparse Cauchy transform, we can reduce the rank of the matrix to $\poly(k)$. Then we apply a dense Cauchy transform and can further reduce the dimension while only incurring another $\poly(k)$ factor in the approximation ratio. By combining a sparse Cauchy transform and a dense Cauchy transform, we can improve the running time from $\nnz(A) \cdot \wt{O}(k)$ to $\nnz(A)$.

\begin{corollary}\label{cor:l1_bicriteria_algorithm_rank_k2_input_sparsity_time}
Given a $3$rd order tensor $A\in \mathbb{R}^{n\times n \times n}$, for any $k\geq 1$, $\epsilon \in (0,1)$, let $r=\wt{O}(k^{2})$. There exists an algorithm which takes $\nnz(A) + O(n) \poly(k) + \poly(k)$ time and outputs three matrices $U,V,W\in \mathbb{R}^{n\times r}$ such that
\begin{align*}
\left\| \sum_{i=1}^r U_i \otimes V_i \otimes W_i - A \right\|_1 \leq \poly( k, \log n ) \min_{\rank-k~A_k} \| A_k - A\|_1
\end{align*}
holds with probability $9/10$.
\end{corollary}

\begin{algorithm}[h]\caption{$\ell_1$-Low Rank Approximation, Bicriteria Algorithm, $\rank$-$\wt{O}(k^2)$, Input Sparsity Time}\label{alg:l1_bicriteria_algorithm_rank_k2_real_input_sparsity_time}
\begin{algorithmic}[1]
\Procedure{\textsc{L1BicriteriaAlgorithm}}{$A,n,k$} \Comment{Corollary \ref{cor:l1_bicriteria_algorithm_rank_k2_input_sparsity_time}}
\State $s_1\leftarrow s_2 \leftarrow s_3 \leftarrow \wt{O}(k)$.
\State For each $i\in [3]$, choose $S_i\in \mathbb{R}^{n^2 \times s_i}$ to be the composition of a sparse Cauchy transform and a dense Cauchy transform. \Comment{Part (\RN{1},\RN{2}) of Theorem~\ref{sec:l1_existence_results}}
\State Compute $A_1 \cdot S_1$, $A_2 \cdot S_2$.
\State For each $i\in [2]$, choose $T_i$ to be a sampling and rescaling diagonal matrix according to the Lewis weights of $A_i S_i$, with $t_i=\wt{O}(k)$ nonzero entries.
\State $C\leftarrow A(T_1,T_2,I)$.
\State $B^{i+(j-1)s_1} \leftarrow \vect( (T_1A_1 S_1)_i \otimes (T_2A_2S_2)_j ), \forall i\in[s_1],j\in[s_2]$.
\State Form objective function $\min_{W}\| W B - C_3\|_1$.
\State Run $\ell_1$-regression solver to find $\wh{W}$.
\State Construct $\wh{U}$ by using $A_1 S_1$ according to Equation~\eqref{eq:l1_bicriteria_rank_k2_whU}.
\State Construct $\wh{V}$ by using $A_2 S_2$ according to Equation~\eqref{eq:l1_bicriteria_rank_k2_whV}.
\State \Return $\wh{U},\wh{V},\wh{W}$.
\EndProcedure
\end{algorithmic}
\end{algorithm}

\subsection{Algorithms}\label{sec:l1_algorithm}
In this section, we show two different algorithms by using different kind of sketches. One is shown in Theorem~\ref{thm:l1_input_sparsity_time} which gives a fast running time. Another one is shown in Theorem~\ref{thm:l1_best_approximation_ratio} which gives the best approximation ratio.
\subsubsection{Input sparsity time algorithm}
\begin{algorithm}[h]\caption{$\ell_1$-Low Rank Approximation, Input sparsity Time Algorithm}\label{alg:l1_input_sparsity_time}
\begin{algorithmic}[1]
\Procedure{\textsc{L1TensorLowRankApproxInputSparsity}}{$A,n,k$} \Comment{Theorem~\ref{thm:l1_input_sparsity_time}}
\State $s_1 \leftarrow s_2 \leftarrow s_3 \leftarrow \wt{O}(k^5)$.
\State Choose $S_i \in \mathbb{R}^{n^2 \times s_i}$ to be a dense Cauchy transform, $\forall i\in [3]$. \Comment{Part (\RN{1}) of Theorem \ref{thm:l1_existence_results}}
\State Compute $A_1 \cdot S_1$, $A_2 \cdot S_2$, and $A_3 \cdot S_3$.
\State $Y_1,Y_2,Y_3,C \leftarrow$\textsc{L1PolyKSizeReduction}($A,A_1S_1,A_2S_2,A_3S_3,n,s_1,s_2,s_3,k$). \Comment{Algorithm~\ref{alg:l1_polyk_size_reduction}}
\State Create variables $s_1 \times k+s_2\times k +s_3\times k$ variables for each entry of $X_1$, $X_2$, $X_3$.
\State Form objective function $\|  (Y_1X_1) \otimes (Y_2X_2) \otimes (Y_3X_3) - C \|_F^2 $.
\State Run polynomial system verifier.
\State \Return $A_1S_1 X_1,A_2S_2 X_2,A_3S_3 X_3$.
\EndProcedure
\end{algorithmic}
\end{algorithm}

\begin{theorem}\label{thm:l1_input_sparsity_time}
Given a $3$rd tensor $A\in \mathbb{R}^{n\times n\times n}$, for any $k\geq 1$, there exists an algorithm that takes $\nnz(A) \cdot \wt{O}(k) + O(n) \poly(k) + 2^{\wt{O}(k^2)}$ time and outputs three matrices $U,V,W\in \mathbb{R}^{n\times k}$ such that,
\begin{align*}
\left\| U \otimes V \otimes W - A \right\|_1 \leq \poly(k,\log n) \min_{\rank-k~A'}\|A' - A \|_1.
\end{align*}
holds with probability at least $9/10$.
\end{theorem}
\begin{proof}
First, we apply part (\RN{2}) of Theorem~\ref{thm:l1_existence_results}. Then $A_i S_i$ can be computed in $O(\nnz(A))$ time. Second, we use Lemma~\ref{lem:l1_polyk_size_reduction} to reduce the size of the objective function from $O(n^3)$ to $\poly(k)$ in $n \poly(k)$ time by only losing a constant factor in approximation ratio. Third, we use Claim~\ref{cla:tensor_frobenius_relax_ell1_lowrank} to relax the objective function from entry-wise $\ell_1$-norm to Frobenius norm, and this step causes us to lose some other $\poly(k)$ factors in approximation ratio. As a last step, we use Theorem~\ref{thm:f_solving_small_problems} to solve the Frobenius norm objective function.
\end{proof}

Notice again that if we first apply a sparse Cauchy transform, we can reduce the rank of the matrix to $\poly(k)$. Then as before we can apply a dense Cauchy transform to further reduce the dimension while only incurring another $\poly(k)$ factor in the approximation ratio. By combining a sparse Cauchy transform and a dense Cauchy transform, we can improve the running time from $\nnz(A) \cdot \wt{O}(k)$ to $\nnz(A)$, while losing some additional $\poly(k)$ factors in approximation ratio.

 \begin{corollary}\label{cor:l1_input_sparsity_time}
Given a $3$rd tensor $A\in \mathbb{R}^{n\times n\times n}$, for any $k\geq 1$, there exists an algorithm that takes $\nnz(A) + O(n) \poly(k) + 2^{\wt{O}(k^2)}$ time and outputs three matrices $U,V,W\in \mathbb{R}^{n\times k}$ such that,
\begin{align*}
\left\| U \otimes V \otimes W - A \right\|_1 \leq \poly(k,\log n) \min_{\rank-k~A'}\|A' - A \|_1.
\end{align*}
holds with probability at least $9/10$.
 \end{corollary}

\subsubsection{$\wt{O}(k^{3/2})$-approximation algorithm}

\begin{algorithm}[h]\caption{$\ell_1$-Low Rank Approximation Algorithm, $\wt{O}(k^{3/2})$-approximation}
\begin{algorithmic}[1]
\Procedure{\textsc{L1TensorLowRankApproxK}}{$A,n,k$} \Comment{Theorem~\ref{thm:l1_best_approximation_ratio}}
\State $s_1 \leftarrow s_2 \leftarrow s_3 \leftarrow \wt{O}(k)$.
\State Guess diagonal matrices $S_i \in \mathbb{R}^{n^2 \times s_i}$ with $s_i$ nonzero entries, $\forall i\in [3]$. \Comment{Part (\RN{3}) of Theorem \ref{thm:l1_existence_results}}
\State Compute $A_1 \cdot S_1$, $A_2 \cdot S_2$, and $A_3 \cdot S_3$.

\State $Y_1,Y_2,Y_3,C \leftarrow$\textsc{L1PolyKSizeReduction}($A,A_1S_1,A_2S_2,A_3S_3,n,s_1,s_2,s_3,k$). \Comment{Algorithm~\ref{alg:l1_polyk_size_reduction}}
\State Create $s_1 \times k+s_2\times k +s_3\times k$ variables for each entry of $X_1$, $X_2$, $X_3$.
\State Form objective function $\|  (Y_1X_1) \otimes (Y_2X_2) \otimes (Y_3X_3)  - C \|_1 $.
\State Run polynomial system verifier.
\State \Return $U,V,W$.
\EndProcedure
\end{algorithmic}
\end{algorithm}

\begin{theorem}\label{thm:l1_best_approximation_ratio}
Given a $3$rd order tensor $A\in \mathbb{R}^{n\times n\times n}$, for any $k\geq 1$, there exists an algorithm that takes $n^{\wt{O}(k)} 2^{\wt{O}(k^3)}$ time and output three matrices $U,V,W\in \mathbb{R}^{n\times k}$ such that,
\begin{align*}
\| U \otimes V \otimes W - A \|_1 \leq \wt{O}(k^{3/2}) \min_{\rank-k~A'}\|A' - A \|_1.
\end{align*}
holds with probability at least $9/10$.
\end{theorem}
\begin{proof}

First, we apply part (\RN{3}) of Theorem~\ref{thm:l1_existence_results}. Then, guessing $S_i$ requires $n^{\wt{O}(k)}$ time. Second, we use Lemma~\ref{lem:l1_polyk_size_reduction} to reduce the size of the objective from $O(n^3)$ to $\poly(k)$ in polynomial time while only losing a constant factor in approximation ratio. Third, we use Theorem~\ref{thm:l1_solving_small_problems} to solve the entry-wise $\ell_1$-norm objective function directly.
\end{proof}



\subsection{CURT decomposition}\label{sec:l1_curt}

\begin{algorithm}[h]\caption{$\ell_1$-CURT Decomposition Algorithm}\label{alg:l1_curt_algorithm}
\begin{algorithmic}[1]
\Procedure{\textsc{L1CURT}}{$A,U_B,V_B,W_B,n,k$} \Comment{Theorem \ref{thm:l1_curt_algorithm}}
\State Form $B_1 = V_B^\top \odot W_B^\top \in \mathbb{R}^{k\times n^2}$.
\State Let $D_1^\top \in \mathbb{R}^{n^2 \times n^2}$ be the sampling and rescaling diagonal matrix corresponding to the Lewis weights of $B_1^\top$, and let $D_1$ have $d_1=O(k\log k)$ nonzero entries.
\State Form $\wh{U} = A_1 D_1 (B_1 D_1)^\dagger \in \mathbb{R}^{n\times k}$.
\State Form $B_2 = \wh{U}^\top \odot W_B^\top \in \mathbb{R}^{k\times n^2}$.
\State Let $D_2^\top \in \mathbb{R}^{n^2 \times n^2}$ be the sampling and rescaling diagonal matrix corresponding to the Lewis weights of $B_2^\top$, and let $D_2$ have $d_2=O(k\log k)$ nonzero entries.
\State Form $\wh{V} = A_2 D_2 (B_2 D_2)^\dagger \in \mathbb{R}^{n\times k}$.
\State Form $B_3 = \wh{U}^\top \odot \wh{V}^\top \in \mathbb{R}^{k\times n^2}$.
\State Let $D_3^\top \in \mathbb{R}^{n^2 \times n^2}$ be the sampling and rescaling diagonal matrix corresponding to the Lewis weights of $B_3^\top$, and let $D_3$ have $d_3=O(k\log k)$ nonzero entries.
\State $C\leftarrow A_1 D_1$, $R\leftarrow A_2 D_2$, $T\leftarrow A_3 D_3$.
\State $U\leftarrow \sum_{i=1}^k ( (B_1 D_1)^\dagger )_i \otimes ( (B_2 D_2)^\dagger )_i \otimes ( (B_3 D_3)^\dagger )_i$.
\State \Return $C$, $R$, $T$ and $U$.
\EndProcedure
\end{algorithmic}
\end{algorithm}

\begin{theorem}\label{thm:l1_curt_algorithm}
Given a $3$rd order tensor $A\in \mathbb{R}^{n\times n \times n}$, let $k\geq 1$, let $U_B,V_B,W_B\in \mathbb{R}^{n\times k}$ denote a rank-$k$, $\alpha$-approximation to $A$. Then there exists an algorithm which takes $O(\nnz(A)) + O(n^2) \poly(k)$ time and outputs three matrices: $C\in \mathbb{R}^{n\times c}$ with columns from $A$, $R\in \mathbb{R}^{n\times r}$ with rows from $A$, $T\in \mathbb{R}^{n\times t}$ with tubes from $A$, and a tensor $U\in \mathbb{R}^{c\times r\times t}$ with $\rank(U)=k$ such that $c=r=t=O(k\log k)$, and
\begin{align*}
\left\| \sum_{i=1}^c \sum_{j=1}^r \sum_{l=1}^t U_{i,j,l} \cdot C_i \otimes R_j \otimes T_l - A \right\|_1 \leq \wt{O}(k^{1.5}) \alpha \min_{\rank-k~A'} \| A' - A\|_1
\end{align*}
holds with probability $9/10$.
\end{theorem}

\begin{proof}
We define
\begin{align*}
\OPT : = \min_{\rank-k~A'} \| A' - A\|_1.
\end{align*}
We already have three matrices $U_B\in \mathbb{R}^{n\times k}$, $V_B\in \mathbb{R}^{n\times k}$ and $W_B\in \mathbb{R}^{n\times k}$ and these three matrices provide a $\rank$-$k$, $\alpha$ approximation to $A$, i.e.,
\begin{align}\label{eq:l1_cur_UBVBWB_minus_A}
\left\| \sum_{i=1}^k ( U_B )_i \otimes (V_B)_i \otimes (W_B)_i - A \right\|_1 \leq \alpha \OPT
\end{align}
Let $B_1 = V_B^\top \odot W_B^\top \in \mathbb{R}^{k\times n^2}$ denote the matrix where the $i$-th row is the vectorization of $(V_B)_i \otimes (W_B)_i$. By Section B.3, we can compute $D_1 \in \mathbb{R}^{n^2 \times n^2}$ which is a sampling and rescaling matrix corresponding to the Lewis weights of $B_1^\top$ in $O(n^2\poly(k))$ time, and there are $d_1 = O(k\log k)$ nonzero entries on the diagonal of $D_1$. Let $A_i\in \mathbb{R}^{n\times n^2}$ denote the matrix obtained by flattening $A$ along the $i$-th direction, for each $i\in [3]$.

Define $U^*\in \mathbb{R}^{n\times k}$ to be the optimal solution to $\underset{U\in \mathbb{R}^{n\times k} }{\min} \| U B_1 - A_1\|_1$, $\wh{U} = A_1 D_1 (B_1 D_1)^\dagger \in \mathbb{R}^{n\times k}$, $V_0 \in \mathbb{R}^{n\times k}$ to be the optimal solution to $\underset{V\in \mathbb{R}^{n\times k} }{\min} \| V \cdot  (\wh{U}^\top \odot W_B^\top) - A_2 \|_1 $, and $U'$ to be the optimal solution to $\underset{U\in \mathbb{R}^{n\times k}}{\min} \| U B_1 D_1 - A_1 D_1 \|_1$.

By Claim~\ref{cla:ell2_relax_ell1_regression}, we have
\begin{align*}
\|\wh{U} B_1 D_1 - A_1 D_1 \|_1 \leq \sqrt{d_1} \| U' B_1 D_1 - A_1 D_1\|_1
\end{align*}
Due to Lemma D.11 and Lemma D.8 (in \cite{swz17}) with constant probability, we have
\begin{align}\label{eq:l1_cur_Uwh_B1_minus_A1}
\| \wh{U} B_1 - A _1 \|_1 \leq \sqrt{d_1} \alpha_{D_1} \| U^* B_1 - A_1 \|_1,
\end{align}
where $\alpha_{D_1} = O(1)$.

Recall that $( \wh{U}^\top \odot W_B^\top) \in \mathbb{R}^{k\times n^2}$ denotes the matrix where the $i$-th row is the vectorization of $\wh{U}_i \otimes (W_B)_i$, $\forall i\in [k]$. Now, we can show,
\begin{align}\label{eq:l1_cur_V0B2_minus_A2}
\| V_0 \cdot ( \wh{U}^\top \odot W_B^\top) - A_2 \|_1 \leq & ~ \| \wh{U} B_1 - A_1 \|_1 & \text{~by~} V_0 = \underset{V\in \mathbb{R}^{n\times k}}{\arg\min} \| V \cdot ( \wh{U}^\top \odot W_B^\top) - A_2  \|_1 \notag \\
\lesssim & ~ \sqrt{d_1}\| U^* B_1 - A_1 \|_1 & \text{~by~Equation~\eqref{eq:l1_cur_Uwh_B1_minus_A1}} \notag \\
\leq & ~ \sqrt{d_1} \| U_B B_1 - A_1 \|_1 & \text{~by~} U^* = \underset{U\in \mathbb{R}^{n\times k} }{\arg\min} \| U B_1 - A_1 \|_1 \notag \\
\leq & ~ O(\sqrt{d_1}) \alpha \OPT & \text{~by~Equation~\eqref{eq:l1_cur_UBVBWB_minus_A}}
\end{align}

We define $B_2= \wh{U}^\top \odot W_B^\top$. We can compute $D_2\in \mathbb{R}^{n^2 \times n^2}$ which is a sampling and rescaling matrix corresponding to the Lewis weights of $B_2^\top$ in $O(n^2 \poly(k))$ time, and there are $d_2 = O(k\log k)$ nonzero entries on the diagonal of $D_2$.

Define $V^*\in \mathbb{R}^{n\times k}$ to be the optimal solution of $\min_{V\in \mathbb{R}^{n\times k}} \| V B_2 - A_2 \|_1$, $\wh{V}= A_2 D_2 (B_2 D_2)^\dagger \in \mathbb{R}^{n\times k}$, $W_0\in \mathbb{R}^{n\times k}$ to be the optimal solution of $\underset{W\in \mathbb{R}^{n\times k}}{\min} \| W\cdot (\wh{U}^\top \odot \wh{V}^\top) - A_3 \|_1$, and $V'$ to be the optimal solution of $\underset{V\in \mathbb{R}^{n\times k}}{\min} \|V B_2 D_2 - A_2 D_2\|_1$.

By Claim~\ref{cla:ell2_relax_ell1_regression}, we have
\begin{align*}
\| \wh{V} B_2 D_2 - A_2 D_2 \|_1 \leq \sqrt{d_2} \|V' B_2 D_2 - A_2 D_2 \|_1.
\end{align*}
Due to Lemma D.11 and Lemma D.8(in \cite{swz17}) with constant probability, we have
\begin{align}\label{eq:l1_cur_Vwh_B2_minus_A2}
\| \wh{V} B_2 - A_2 \|_1 \leq \sqrt{d_2} \alpha_{D_2} \| V^* B_2 - A_2 \|_1,
\end{align}
where $\alpha_{D_2} = O(1)$.

Recall that $(\wh{U}^\top \odot \wh{V}^\top) \in \mathbb{R}^{k\times n^2}$ denotes the matrix for which the $i$-th row is the vectorization of $\wh{U}_i \otimes \wh{V}_i$, $\forall i\in [k]$. Now, we can show,
\begin{align}\label{eq:l1_cur_W0B3_minus_A3}
\| W_0 \cdot (\wh{U}^\top \odot \wh{V}^\top ) - A_3 \|_1 \leq & ~ \| \wh{V} B_2 - A_2 \|_1 & \text{~by~} W_0 = \underset{W\in \mathbb{R}^{n\times k} }{\arg\min} \| W \cdot ( \wh{U}^\top \odot \wh{V}^\top ) - A_3 \|_1 \notag \\
\lesssim & ~ \sqrt{d_2} \| V^* B_2 - A_2 \|_1 & \text{~by~Equation~\eqref{eq:l1_cur_Vwh_B2_minus_A2}} \notag \\
\leq & ~ \sqrt{d_2} \| V_0 B_2 - A_2 \|_1 & \text{~by~} V^* =\underset{V\in \mathbb{R}^{n\times k}}{\arg\min} \| V B_2 - A_2 \|_1 \notag \\
\leq & ~ O(\sqrt{d_1 d_2}) \alpha \OPT & \text{~by~Equation~\eqref{eq:l1_cur_V0B2_minus_A2}}
\end{align}

We define $B_3= \wh{U}^\top \odot \wh{V}^\top$. We can compute $D_3\in \mathbb{R}^{n^2 \times n^2}$ which is a sampling and rescaling matrix corresponding to the Lewis weights of $B_3^\top$ in $O(n^2 \poly(k))$ time, and there are $d_3 = O(k\log k)$ nonzero entries on the diagonal of $D_3$.

Define $W^*\in \mathbb{R}^{n\times k}$ to be the optimal solution to $\min_{W\in \mathbb{R}^{n\times k}} \| W B_3 - A_3 \|_1$, $\wh{W}= A_3 D_3 (B_3 D_3)^\dagger \in \mathbb{R}^{n\times k}$,
and $W'$ to be the optimal solution to $\underset{W\in \mathbb{R}^{n\times k}}{\min} \|W B_3 D_3 - A_3 D_3\|_1$.

By Claim~\ref{cla:ell2_relax_ell1_regression}, we have
\begin{align*}
\| \wh{W} B_3 D_3 - A_3 D_3 \|_1 \leq \sqrt{d_3} \|W' B_3 D_3 - A_3 D_3 \|_1.
\end{align*}
Due to Lemma D.11 and Lemma D.8(in \cite{swz17}) with constant probability, we have
\begin{align}\label{eq:l1_cur_Wwh_B3_minus_A3}
\| \wh{W} B_3 - A_3 \|_1 \leq \sqrt{d_3} \alpha_{D_3} \| W^* B_3 - A_3 \|_1,
\end{align}
where $\alpha_{D_3} = O(1)$. Now we can show,
\begin{align*}
\| \wh{W} B_3 - A_3 \|_1 \lesssim & ~ \sqrt{d_3} \| W^* B_3 - A_3 \|_1, & \text{~by~Equation~\eqref{eq:l1_cur_Wwh_B3_minus_A3}} \\
\leq & ~  \sqrt{d_3} \| W_0 B_3 - A_3 \|_1, & \text{~by~}W^* = \underset{W\in \mathbb{R}^{n\times k} }{\arg\min} \| W B_3 - A_3 \|_1 \\
\leq & ~ O(\sqrt{d_1d_2d_3}) \alpha \OPT & \text{~by~Equation~\eqref{eq:l1_cur_W0B3_minus_A3}}
\end{align*}
Thus, it implies,
\begin{align*}
\left\| \sum_{i=1}^k \wh{U}_i \otimes \wh{V}_i \otimes \wh{W}_i - A \right\|_1 \leq \poly(k,\log n) \OPT.
\end{align*}
where $\wh{U} = A_1 D_1 (B_1 D_1)^\dagger$, $\wh{V} = A_2D_2 (B_2 D_2)^\dagger$, $\wh{W}=A_3D_3 (B_3 D_3)^\dagger$.

\end{proof}

\begin{algorithm}[h]\caption{$\ell_1$-CURT decomposition algorithm}\label{alg:l1_curt_algorithm_full}
\begin{algorithmic}[1]
\Procedure{\textsc{L1CURT}$^+$}{$A,n,k$} \Comment{Theorem \ref{thm:l1_curt_algorithm_full}}
\State $U_B, V_B, W_B \leftarrow $\textsc{L1LowRankApproximation}($A,n,k$). \Comment{Corollary~\ref{cor:l1_input_sparsity_time}}
\State $C,R,T,U\leftarrow$ \textsc{L1CURT}($A,U_B,V_B,W_B,n,k$). \Comment{Algorithm~\ref{alg:l1_curt_algorithm}}
\State \Return $C$, $R$, $T$ and $U$.
\EndProcedure
\end{algorithmic}
\end{algorithm}

\begin{theorem}\label{thm:l1_curt_algorithm_full}
Given a $3$rd order tensor $A\in \mathbb{R}^{n\times n\times n}$, for any $k\geq 1$, there exists an algorithm which takes $O(\nnz(A)) + O(n^2)\poly(k) + 2^{\wt{O}(k^2)}$ time and outputs three matrices $C\in \mathbb{R}^{n\times c}$ with columns from $A$, $R\in \mathbb{R}^{n\times r}$ with rows from $A$, $T\in \mathbb{R}^{n\times t}$ with tubes from $A$, and a tensor $U\in \mathbb{R}^{c\times r \times t}$ with $\rank(U)=k$ such that $c=r=t=O(k\log k)$, and
\begin{align*}
\left\| \sum_{i=1}^c \sum_{j=1}^r \sum_{l=1}^t U_{i,j,l} \cdot C_i \otimes R_j \otimes T_l - A \right\|_1 \leq \poly(k,\log n) \underset{\rank-k~A'}{\min} \| A' - A\|_1,
\end{align*}
holds with probability $9/10$.
\end{theorem}
\begin{proof}
This follows by combining Corollary~\ref{cor:l1_input_sparsity_time} and Theorem~\ref{thm:l1_curt_algorithm}.
\end{proof}


\newpage
\section{Entry-wise $\ell_p$ Norm for Arbitrary Tensors, $1<p<2$}\label{sec:lp}
There is a long line of research dealing with $\ell_p$ norm-related problems \cite{ddhkm09,mm13,cdmmmw13,cp15,bcky16,ycrm16,bbcky17}.

In this section, we provide several different algorithms for tensor $\ell_p$-low rank approximation. Section~\ref{sec:lp_matrix_existence_results} formally states the $\ell_p$ version of Theorem C.1 in \cite{swz17}. Section~\ref{sec:lp_existence_results} presents several existence results. Section~\ref{sec:lp_polyk_size_reduction} describes a tool that is able to reduce the size of the objective function from $\poly(n)$ to $\poly(k)$. Section~\ref{sec:lp_solving_small_problems} discusses the case when the problem size is small. Section~\ref{sec:lp_bicriteria_algorithm} provides several bicriteria algorithms.
Section~\ref{sec:lp_algorithm} summarizes a batch of algorithms. Section~\ref{sec:lp_curt} provides an algorithm for $\ell_p$ norm CURT decomposition.

Notice that if the $\rank$-$k$ solution does not exist, then every bicriteria algorithm in Section~\ref{sec:lp_bicriteria_algorithm} can be stated in the form as Theorem~\ref{thm:bicriteria}, and every algorithm which can output a $\rank$-$k$ solution in Section~\ref{sec:lp_algorithm} can be stated in the form as Theorem~\ref{thm:smallk}. See Section~\ref{sec:intro} for more details.

\subsection{Existence results for matrix case}\label{sec:lp_matrix_existence_results}

\begin{theorem}[\cite{swz17}]\label{thm:lp_matrix_exist}
Let $1\leq p<2$. Given $V\in\mathbb{R}^{k\times n},A\in\mathbb{R}^{d\times n}$. Let $S\in\mathbb{R}^{n\times s}$ be a proper random sketching matrix. Let
\begin{align*}
\wh{U}=\arg\min_{U\in\mathbb{R}^{d\times k}} \|UVS-AS\|_F^2,
\end{align*}
i.e.,
\begin{align*}
\wh{U}=AS(VS)^{\dagger}.
\end{align*}
Then with probability at least $0.999,$
\begin{align*}
\|\wh{U}V-A\|_p^p\leq \alpha\cdot \min_{U\in\mathbb{R}^{d\times k}} \|UV-A\|_p^p.
\end{align*}

$\mathrm{(\RN{1})}$. $S$ denotes a dense $p$-stable transform,\\
$s=\wt{O}(k)$, $\alpha = \wt{O}(k^{1-p/2}) \log d$. 

$\mathrm{(\RN{2})}$. $S$ denotes a sparse $p$-stable transform,\\
$s=\wt{O}(k^5)$, $\alpha = \wt{O}(k^{5-5p/2+2/p}) \log d$. 

$\mathrm{(\RN{3})}$. $S^\top$ denotes a sampling/rescaling matrix according to the $\ell_p$ Lewis weights of $V^\top$,\\
$s=\wt{O}(k)$, $\alpha = \wt{O}(k^{1-p/2})$. 
\end{theorem}
We give the proof for completeness.

\begin{proof}
Let $S\in\mathbb{R}^{n\times s}$ be a sketching matrix which satisfies the property $(*)$: $\forall c\geq 1,\wt{U}\in\mathbb{R}^{d\times k}$ which satisfy
\begin{align*}
\|\wt{U}VS-AS\|_p^p\leq c\cdot\min_{U\in\mathbb{R}^{d\times k}} \|UVS-AS\|_p^p,
\end{align*}
we have
\begin{align*}
\|\wt{U}V-A\|_p^p \leq c  \beta_S\cdot \min_{U\in\mathbb{R}^{d\times k}} \|UV-A\|_p^p,
\end{align*}
where $\beta_S\geq 1$ only depends on the sketching matrix $S$. Let
\begin{align*}
\forall i\in[d],(\wh{U}^i)^\top=\arg\min_{x\in\mathbb{R}^{k}} \|x^\top VS-A^iS\|_2^2,
\end{align*}
i.e.,
\begin{align*}
\wh{U}=AS(VS)^{\dagger}.
\end{align*}
Let
\begin{align*}
\wt{U}=\arg\min_{U\in\mathbb{R}^{d\times k}} \|UVS-AS\|_p^p.
\end{align*}
Then, we have:
\begin{align*}
& ~\|\wh{U}VS-AS\|_p^p\\
=& ~ \sum_{i=1}^d \|\wh{U}^i VS-A^iS\|_p^p\\
\leq & ~ \sum_{i=1}^d (s^{1/p-1/2}\|\wh{U}^i VS-A^iS\|_2)^p\\
\leq & ~ \sum_{i=1}^d (s^{1/p-1/2}\|\wt{U}^i VS-A^iS\|_2)^p\\
\leq & ~ \sum_{i=1}^d (s^{1/p-1/2}\|\wt{U}^i VS-A^iS\|_p)^p\\
\leq & ~ s^{1-p/2}\|\wt{U} VS-AS\|_p^p.
\end{align*}
The first inequality follows using $\forall x\in\mathbb{R}^s,\|x\|_p\leq s^{1/p-1/2}\|x\|_2$ since $p<2$. The third inequality follows using $\forall x\in\mathbb{R}^s,\|x\|_2\leq \|x\|_p$ since $p<2$. Thus, according to the property $(*)$ of $S$,
\begin{align*}
\|\wh{U}V-A\|_p^p \leq  s^{1-p/2}\beta_S \min_{U\in\mathbb{R}^{d\times k}} \|UV-A\|_p^p.
\end{align*}
Due to Lemma E.8 and Lemma E.11 of~\cite{swz17}, we have:

for (\RN{1}), $s=\wt{O}(k),\beta_S=O(\log d),\alpha=s^{1-p/2}\beta_S=\wt{O}(k^{1-p/2})\log d$,

for (\RN{2}), $s=\wt{O}(k^5),\beta_S=\wt{O}(k^{2/p}\log d),\alpha=s^{1-p/2}\beta_S=\wt{O}(k^{5-5p/2+2/p}) \log d$,

for (\RN{3}), $s=\wt{O}(k),\beta_S=O(1),\alpha=s^{1-p/2}\beta_S=\wt{O}(k^{1-p/2})$.
\end{proof}

\subsection{Existence results}\label{sec:lp_existence_results}

\begin{theorem}\label{thm:lp_existence_results}
Given a 3rd order tensor $A\in \mathbb{R}^{n\times n \times n}$, for any $k\geq 1$,
there exist three matrices $S_1\in \mathbb{R}^{n^2\times s_1}$, $S_2\in \mathbb{R}^{n^2\times s_2}$, $S_3 \in \mathbb{R}^{n^2 \times s_3}$ such that
\begin{align*}
\min_{  X_1, X_2 , X_3 } \left\| \sum_{i=1}^k (A_1S_1 X_1)_i \otimes (A_2S_2X_2)_i \otimes (A_3S_3X_3)_i -A \right\|_p^p \leq \alpha \underset{\rank-k~A_k \in \mathbb{R}^{n\times n \times n}}{\min} \| A_k -A \|_p^p,
\end{align*}
holds with probability $99/100$.

$\mathrm{(\RN{1})}$. Using a dense $p$-stable transform,\\
$s_1=s_2=s_3=\wt{O}(k)$, $\alpha = \wt{O}(k^{3-1.5p}) \log^3 n$. 

$\mathrm{(\RN{2})}$. Using a sparse $p$-stable transform,\\
$s_1=s_2=s_3=\wt{O}(k^5)$, $\alpha = \wt{O}(k^{15-7.5p+6/p}) \log^3 n$. 

$\mathrm{(\RN{3})}$. Guessing Lewis weights,\\
$s_1=s_2=s_3=\wt{O}(k)$, $\alpha = \wt{O}(k^{3-1.5p})$. 
\end{theorem}

\begin{proof}
We use $\OPT$ to denote
\begin{align*}
\OPT : = \underset{\rank-k~A_k \in \mathbb{R}^{n\times n \times n} }{\min} \| A_k - A \|_p^p.
\end{align*}

Given a tensor $A\in \mathbb{R}^{n_1\times n_2 \times n_3}$, we define three matrices $ A_1 \in \mathbb{R}^{n_1 \times n_2 n_3}, A_2 \in \mathbb{R}^{n_2 \times n_3 n_1}, A_3 \in \mathbb{R}^{n_3 \times n_1 n_2}$ such that, for any $i\in [n_1], j \in [n_2], l \in [n_3]$
\begin{align*}
A_{i,j,l} = ( A_{1})_{i, (j-1) \cdot n_3 + l} = ( A_{2} )_{ j, (l-1) \cdot n_1 + i } = ( {A}_3)_{l, (i-1) \cdot n_2 + j } .
\end{align*}

We fix $V^* \in \mathbb{R}^{n\times k}$ and $W^* \in \mathbb{R}^{n\times k}$, and use $V_1^*, V_2^*, \cdots, V_k^*$ to denote the columns of $V^*$ and $W_1^*, W_2^*, \cdots, W_k^*$ to denote the columns of $W^*$.

We consider the following optimization problem,
\begin{align*}
\min_{U_1, \cdots, U_k \in \mathbb{R}^n } \left\| \sum_{i=1}^k U_i \otimes V_i^* \otimes W_i^* - A \right\|_p^p,
\end{align*}
which is equivalent to
\begin{align*}
\min_{U_1, \cdots, U_k \in \mathbb{R}^n } \left\|
\begin{bmatrix}
U_1 & U_2 & \cdots & U_k
\end{bmatrix}
\begin{bmatrix}
 V_1^* \otimes W_1^*  \\
 V_2^* \otimes W_2^*  \\
\cdots \\
 V_k^* \otimes W_k^*
\end{bmatrix}
- A \right\|_p^p.
\end{align*}

We use matrix $Z_1$ to denote $ V^{*\top} \odot W^{*\top} \in \mathbb{R}^{k\times n^2}$ and matrix $U$ to denote $\begin{bmatrix} U_1 & U_2 & \cdots & U_k \end{bmatrix}$. Then we can obtain the following equivalent objective function,
\begin{align*}
\min_{U \in \mathbb{R}^{n\times k} } \| U Z_1  - A_1 \|_p^p.
\end{align*}

Choose a sketching matrix (a dense $p$-stable, a sparse $p$-stable or an $\ell_p$ Lewis weight sampling/rescaling matrix to $Z_1$) $S_1 \in \mathbb{R}^{n^2 \times s_1}$. We can obtain the optimization problem,
\begin{align*}
\min_{U \in \mathbb{R}^{n\times k} } \| U Z_1 S_1 - A_1 S_1 \|_p^p = \min_{U\in \mathbb{R}^{n\times k} } \sum_{i=1}^n \| U^i Z_1 S_1 - (A_1 S_1)^i \|_p^p,
\end{align*}
where $U^i$ denotes the $i$-th row of matrix $U\in \mathbb{R}^{n\times k}$ and $(A_1 S_1)^i$ denotes the $i$-th row of matrix $A_1 S_1$. Instead of solving it under the $\ell_p$-norm, we consider the $\ell_2$-norm relaxation,
\begin{align*}
\underset{U \in \mathbb{R}^{n\times k} }{\min} \| U Z_1 S_1 - A_1 S_1 \|_F^2 = \underset{U\in \mathbb{R}^{n\times k} }{\min} \sum_{i=1}^n \| U^i Z_1 S_1 - (A_1 S_1)^i \|_2^2.
\end{align*}
Let $ \wh{U} \in \mathbb{R}^{n\times k}$ denote the optimal solution of the above optimization problem. Then, $\wh{U} = A_1 S_1 (Z_1 S_1)^\dagger$. We plug $\wh{U}$ into the objective function under the $\ell_p$-norm. By choosing $s_1$ and by the properties of sketching matrices (a dense $p$-stable, a sparse $p$-stable or an $\ell_p$ Lewis weight sampling/rescaling matrix to $Z_1$) $S_1 \in \mathbb{R}^{n^2 \times s_1}$, we have
\begin{align*}
\| \wh{U} Z_1  - A_1  \|_p^p \leq \alpha \underset{U\in \mathbb{R}^{n\times k}}{\min} \| U Z_1 - A_1 \|_p^p = \alpha \OPT.
\end{align*}
This implies
\begin{align*}
\|  \wh{U} \otimes V^* \otimes W^* - A \|_p^p \leq \alpha \OPT.
\end{align*}

As a second step, we fix $\wh{U} \in \mathbb{R}^{n\times k}$ and $W^* \in \mathbb{R}^{n\times k}$, and convert tensor $A$ into matrix $A_2$. Let matrix $Z_2$ denote $\wh{U}^\top \odot W^{*\top} $. We consider the following objective function,
\begin{align*}
\min_{V \in \mathbb{R}^{n\times k} } \| V Z_2 -A_2  \|_p^p,
\end{align*}
and the optimal cost of it is at most $ \alpha \OPT$.

We choose a sketching matrix (a dense $p$-stable, a sparse $p$-stable or an $\ell_p$ Lewis weight sampling/rescaling matrix to $Z_2$) $S_2 \in \mathbb{R}^{n^2 \times s_2}$ and sketch on the right of the objective function to obtain the new objective function,
\begin{align*}
\underset{V\in \mathbb{R}^{n\times k} }{\min} \| V Z_2 S_2 - A_2 S_2 \|_p^p = \min_{V\in \mathbb{R}^{n\times k}} \sum_{i=1}^n \| V^i Z_2 S_2 - (A_2 S_2)^i \|_p^p,
\end{align*}
where $V^i$ denotes the $i$-th row of matrix $V$ and $(A_2 S_2)^i$ denotes the $i$-th row of matrix $A_2 S_2$. Instead of solving this under the $\ell_p$-norm, we consider the $\ell_2$-norm relaxation,
\begin{align*}
\underset{V \in \mathbb{R}^{n\times k} }{\min} \| V Z_2 S_2 - A_2 S_2 \|_F^2 = \underset{V\in \mathbb{R}^{n\times k} }{\min}\sum_{i=1}^n \| V^i (Z_2 S_2) - (A_2 S_2)^i  \|_2^2.
\end{align*}

Let $\wh{V} \in \mathbb{R}^{n\times k}$ denote the optimal solution of the above problem. Then $\wh{V} = A_2 S_2 (Z_2 S_2)^\dagger$. By properties of sketching matrix $S_2 \in \mathbb{R}^{n^2 \times s_2}$, we have,
\begin{align*}
\| \wh{V} Z_2 - A_2 \|_p^p \leq \alpha \underset{V\in \mathbb{R}^{n\times k} }{\min} \| V Z_2  - A_2 \|_p^p \leq  \alpha^2 \OPT,
\end{align*}
which implies
\begin{align*}
\|  \wh{U} \otimes \wh{V} \otimes W^* - A \|_p^p \leq \alpha^2 \OPT,
\end{align*}

As a third step, we fix the matrices $\wh{U} \in \mathbb{R}^{n\times k}$ and $\wh{V}\in \mathbb{R}^{n \times k}$. We can convert tensor $A\in \mathbb{R}^{n\times n \times n}$ into matrix $A_3 \in \mathbb{R}^{n^2 \times n}$. Let matrix $Z_3$ denote $ \wh{U}^\top \odot \wh{V}^\top \in \mathbb{R}^{k\times n^2}$. We consider the following objective function,
\begin{align*}
\underset{W\in \mathbb{R}^{n\times k} }{\min} \| W Z_3 - A_3 \|_p^p,
\end{align*}
and the optimal cost of it is at most $\alpha^2 \OPT$.

We choose sketching matrix (a dense $p$-stable, a sparse $p$-stable or an $\ell_p$ Lewis weight sampling/rescaling matrix to $Z_3$) $S_3\in \mathbb{R}^{n^2 \times s_3}$ and sketch on the right of the objective function to obtain the new objective function,
\begin{align*}
\underset{ W \in \mathbb{R}^{n\times k} }{ \min } \| W Z_3 S_3 - A_3 S_3 \|_p^p.
\end{align*}
Instead of solving this under the $\ell_p$-norm, we consider the $\ell_2$-norm relaxation,
\begin{align*}
\underset{W \in \mathbb{R}^{n\times k} }{\min} \| W Z_3 S_3 - A_3 S_3 \|_F^2 = \underset{W\in \mathbb{R}^{n\times k} }{\min}\sum_{i=1}^n \| W^i (Z_3 S_3) - (A_3 S_3)^i  \|_2^2.
\end{align*}
Let $\wh{W} \in \mathbb{R}^{n\times k}$ denote the optimal solution of the above problem. Then $\wh{W} = A_3 S_3 (Z_3 S_3)^\dagger$. By properties of sketching matrix $S_3\in \mathbb{R}^{n^2 \times s_3}$, we have,
\begin{align*}
\| \wh{W} Z_3 - A_3 \|_p^p \leq \alpha \underset{W\in \mathbb{R}^{n\times k} }{\min} \| W Z_3 - A_3 \|_p^p \leq \alpha^3 \OPT.
\end{align*}
Thus, we obtain,
\begin{align*}
\min_{X_1 \in \mathbb{R}^{s_1\times k}, X_2 \in \mathbb{R}^{s_2 \times k}, X_3 \in \mathbb{R}^{s_3 \times k}} \left\| \sum_{i=1}^k (A_1S_1 X_1)_i \otimes (A_2 S_2 X_2)_i \otimes (A_3 S_3 X_3)_i - A \right\|_p^p \leq \alpha^3 \OPT.
\end{align*}

According to Theorem~\ref{thm:lp_matrix_exist}, we let $s=s_1=s_2=s_3$ and take the corresponding $\alpha$. We can directly get the results for (\RN{1}), (\RN{2}) and (\RN{3}).
\end{proof}

\subsection{Polynomial in $k$ size reduction}\label{sec:lp_polyk_size_reduction}
\begin{definition}[Definition E.1 in~\cite{swz17}]
Given a matrix $M\in\mathbb{R}^{n\times d},$ if matrix $S\in\mathbb{R}^{m\times n}$ satisfies
\begin{align*}
\|SM\|_p^p\leq \beta\|M\|_p^p,
\end{align*}
then $S$ has at most $\beta$ dilation on $M$ in the $\ell_p$ case.
\end{definition}

\begin{definition}[Definition E.2 in~\cite{swz17}]
Given a matrix $U\in\mathbb{R}^{n\times k},$ if matrix $S\in\mathbb{R}^{m\times n}$ satisfies
\begin{align*}
\forall x\in\mathbb{R}^k, \|SUx\|_p^p\geq \frac{1}{\beta}\|Ux\|_p^p,
\end{align*}
then $S$ has at most $\beta$ contraction on $U$ in the $\ell_p$ case.
\end{definition}

\begin{theorem}\label{thm:lp_sketch_one_side}
Given a tensor $A\in\mathbb{R}^{n_1\times n_2\times n_3}$ and three matrices $V_1\in\mathbb{R}^{n_1\times b_1},V_2\in\mathbb{R}^{n_2\times b_2},V_3\in\mathbb{R}^{n_3\times b_3},$ let $X_1^*\in\mathbb{R}^{b_1\times k},X_2^*\in\mathbb{R}^{b_2\times k},X_3^*\in\mathbb{R}^{b_3\times k}$ satisfy
\begin{align*}
X_1^*,X_2^*,X_3^*=\underset{X_1\in\mathbb{R}^{b_1\times k},X_2\in\mathbb{R}^{b_2\times k},X_3\in\mathbb{R}^{b_3\times k}}{\arg\min}\|V_1X_1\otimes V_2X_2 \otimes V_3X_3 - A\|_p^p.
\end{align*}
Let $S\in\mathbb{R}^{m\times n}$ have at most $\beta_1\geq 1$ dilation on $V_1X_1^*\cdot ((V_2X_2^*)^\top \odot (V_3X_3^*)^\top)-A_1$ and $S$ have at most $\beta_2\geq 1$ contraction on $V_1$ in the $\ell_p$ case. If $\wh{X}_1\in\mathbb{R}^{b_1\times k},\wh{X}_2\in\mathbb{R}^{b_2\times k},\wh{X}_3\in\mathbb{R}^{b_3\times k}$ satisfy
\begin{align*}
\|SV_1\wh{X}_1\otimes V_2\wh{X}_2\otimes V_3\wh{X}_3 - SA \|_p^p \leq \beta \underset{X_1\in\mathbb{R}^{b_1\times k},X_2\in\mathbb{R}^{b_2\times k},X_3\in\mathbb{R}^{b_3\times k}}{\min}\|SV_1X_1\otimes V_2X_2\otimes V_3X_3 - SA\|_p^p,
\end{align*}
where $\beta\geq 1$,
then
\begin{align*}
\|V_1\wh{X}_1\otimes V_2\wh{X}_2\otimes V_3\wh{X}_3 - A \|_p^p \lesssim \beta_1\beta_2\beta \min_{X_1,X_2,X_3}\|V_1X_1\otimes V_2X_2\otimes V_3X_3 - A\|_p^p.
\end{align*}
\end{theorem}
The proof is essentially the same as the proof of Theorem~\ref{thm:l1_sketch_one_side}:
\begin{proof}
Let $A,V_1,V_2,V_3,S,X_1^*,X_2^*,X_3^*,\beta_1,\beta_2$ be as stated in the theorem. Let
$\wh{X}_1\in\mathbb{R}^{b_1\times k},\wh{X}_2\in\mathbb{R}^{b_2\times k},\wh{X}_3\in\mathbb{R}^{b_3\times k}$ satisfy
\begin{align*}
\|SV_1\wh{X}_1\otimes V_2\wh{X}_2\otimes V_3\wh{X}_3 - SA \|_p^p \leq \beta \underset{X_1\in\mathbb{R}^{b_1\times k},X_2\in\mathbb{R}^{b_2\times k},X_3\in\mathbb{R}^{b_3\times k}}{\min}\|SV_1X_1\otimes V_2X_2\otimes V_3X_3 - SA\|_p^p.
\end{align*}
Similar to the proof of Theorem~\ref{thm:l1_sketch_one_side}, we have,
\begin{align}
&\|SV_1\wh{X}_1\otimes V_2\wh{X}_2\otimes V_3\wh{X}_3 - SA \|_p^p\notag\\
=~ & 2^{2-2p}\frac{1}{\beta_2}\|V_1\wh{X}_1\otimes V_2\wh{X}_2\otimes V_3\wh{X}_3 - A \|_p^p-(2^{1-p}\frac{1}{\beta_2}+\beta_1)\|V_1X^*_1\otimes V_2X^*_2\otimes V_3X^*_3-A\|_p^p\notag
\end{align}
The only difference from the proof of Theorem~\ref{thm:l1_sketch_one_side} is that instead of using triangle inequality, we actually use $\|x+y\|_p^p\leq 2^{p-1}\|x\|_p^p+\|y\|_p^p.$
Then, we have
\begin{align*}
&\|V_1\wh{X}_1\otimes V_2\wh{X}_2\otimes V_3\wh{X}_3 - A \|_p^p\\
\leq~ & 2^{2p-2}\beta_2 \|SV_1\wh{X}_1\otimes V_2\wh{X}_2\otimes V_3\wh{X}_3 - SA \|_p^p+(2^{p-1}+2^{2p-2}\beta_1\beta_2) \|V_1X^*_1\otimes V_2X^*_2\otimes V_3X^*_3-A\|_p^p\\
\leq~ & 2^{2p-2}\beta_2\beta \|SV_1X^*_1\otimes V_2X^*_2\otimes V_3X^*_3 - SA \|_p^p+(2^{p-1}+2^{2p-2}\beta_1\beta_2) \|V_1X^*_1\otimes V_2X^*_2\otimes V_3X^*_3-A\|_p^p\\
\leq~ & 2^{2p-2}\beta_1\beta_2\beta \|V_1X^*_1\otimes V_2X^*_2\otimes V_3X^*_3 - A \|_p^p+(2^{p-1}+2^{2p-2}\beta_1\beta_2) \|V_1X^*_1\otimes V_2X^*_2\otimes V_3X^*_3-A\|_p^p\\
\leq~ & 2^{p-1}\beta(1+2\beta_1\beta_2)\|V_1X^*_1\otimes V_2X^*_2\otimes V_3X^*_3 - A \|_p^p.
\end{align*}
\end{proof}

\begin{lemma}\label{lem:lp_polyk_size_reduction}
Let $\min(b_1,b_2,b_3)\geq k$. Given three matrices $V_1\in \mathbb{R}^{n\times b_1}$, $V_2 \in \mathbb{R}^{n\times b_2}$, and $V_3 \in \mathbb{R}^{n\times b_3}$, there exists an algorithm which takes $O(\nnz(A)) + n \poly(b_1,b_2,b_3)$ time and outputs a tensor $C\in \mathbb{R}^{c_1\times c_2\times c_3}$ and three matrices $\wh{V}_1\in \mathbb{R}^{c_1\times b_1}$, $\wh{V}_2 \in \mathbb{R}^{c_2\times b_2}$ and $\wh{V}_3 \in \mathbb{R}^{c_3 \times b_3}$ with $c_1=c_2=c_3=\poly(b_1,b_2,b_3)$, such that with probability $0.99$, for any $\alpha\geq 1$, if $X'_1,X'_2,X'_3$ satisfy that,
\begin{align*}
\left\| \sum_{i=1}^k (\wh{V}_1 X_1')_i \otimes (\wh{V}_2 X_2')_i \otimes (\wh{V}_3 X_3')_i - C \right\|_p^p \leq \alpha \underset{X_1, X_2, X_3}{\min} \left\| \sum_{i=1}^k (\wh{V}_1 X_1)_i \otimes (\wh{V}_2 X_2)_i \otimes (\wh{V}_3 X_3)_i - C \right\|_p^p,
\end{align*}
then,
\begin{align*}
\left\| \sum_{i=1}^k (V_1 X_1')_i \otimes ( V_2 X_2')_i \otimes (V_3 X_3')_i - A \right\|_p^p \lesssim \alpha \min_{X_1, X_2, X_3}\left\| \sum_{i=1}^k ({V}_1 X_1)_i \otimes ({V}_2 X_2)_i \otimes ({V}_3 X_3)_i - A \right\|_p^p.
\end{align*}
\end{lemma}

\begin{proof}
For simplicity, we define $\OPT$ to be
\begin{align*}
\min_{X_1, X_2, X_3}\left\| \sum_{i=1}^k ({V}_1 X_1)_i \otimes ({V}_2 X_2)_i \otimes ({V}_3 X_3)_i - A \right\|_p^p.
\end{align*}
Let $T_1\in \mathbb{R}^{c_1\times n}$ correspond to sampling according to the $\ell_p$ Lewis weights of $V_1\in \mathbb{R}^{n\times b_1}$, where $c_1=\wt{b_1}$. Let $T_2\in \mathbb{R}^{c_2\times n}$ be sampling according to the $\ell_p$ Lewis weights of $V_2\in \mathbb{R}^{n\times b_2}$, where $c_2=\wt{b_2}$. Let $T_3\in \mathbb{R}^{c_3\times n}$ be sampling according to the $\ell_p$ Lewis weights of $V_3\in \mathbb{R}^{n\times b_3}$, where $c_3=\wt{b_3}$.

For any $\alpha\geq 1,$ let $X'_1\in\mathbb{R}^{b_1\times k},X'_2\in\mathbb{R}^{b_2\times k},X'_3\in\mathbb{R}^{b_3\times k}$ satisfy
\begin{align*}
&\|T_1V_1X'_1 \otimes T_2V_2X'_2 \otimes T_3V_3X'_3 - A(T_1,T_2,T_3)\|_p^p \\
\leq~& \alpha \min_{X_1\in\mathbb{R}^{b_1\times k},X_2\in\mathbb{R}^{b_2\times k},X_3\in\mathbb{R}^{b_3\times k}}\|T_1V_1X_1 \otimes T_2V_2X_2 \otimes T_3V_3X_3 - A(T_1,T_2,T_3)\|_p^p.
\end{align*}
First, we regard $T_1$ as the sketching matrix for the remainder. Then by Lemma D.11 in \cite{swz17} and Theorem~\ref{thm:l1_sketch_one_side}, we have
\begin{align*}
&\|V_1X'_1 \otimes T_2V_2X'_2 \otimes T_3V_3X'_3 - A(I,T_2,T_3)\|_p^p \\
\lesssim~& \alpha \min_{X_1\in\mathbb{R}^{b_1\times k},X_2\in\mathbb{R}^{b_2\times k},X_3\in\mathbb{R}^{b_3\times k}}\|V_1X_1 \otimes T_2V_2X_2 \otimes T_3V_3X_3 - A(I,T_2,T_3)\|_p^p.
\end{align*}
Second, we regard $T_2$ as the sketching matrix for $V_1X_1 \otimes V_2X_2 \otimes T_3V_3X_3 - A(I,I,T_3)$. Then by Lemma D.11 in \cite{swz17} and Theorem~\ref{thm:l1_sketch_one_side}, we have
\begin{align*}
&\|V_1X'_1 \otimes V_2X'_2 \otimes T_3V_3X'_3 - A(I,I,T_3)\|_p^p \\
\lesssim~& \alpha \min_{X_1\in\mathbb{R}^{b_1\times k},X_2\in\mathbb{R}^{b_2\times k},X_3\in\mathbb{R}^{b_3\times k}}\|V_1X_1 \otimes V_2X_2 \otimes T_3V_3X_3 - A(I,I,T_3)\|_p^p.
\end{align*}
Third, we regard $T_3$ as the sketching matrix for $V_1X_1 \otimes V_2X_2 \otimes V_3X_3 - A$. Then by Lemma D.11 in \cite{swz17} and Theorem~\ref{thm:l1_sketch_one_side}, we have
\begin{align*}
\left\| \sum_{i=1}^k (V_1 X_1')_i \otimes ( V_2 X_2')_i \otimes (V_3 X_3')_i - A \right\|_p^p \lesssim \alpha \min_{X_1, X_2, X_3}\left\| \sum_{i=1}^k ({V}_1 X_1)_i \otimes ({V}_2 X_2)_i \otimes ({V}_3 X_3)_i - A \right\|_p^p.
\end{align*}
\end{proof}

\subsection{Solving small problems}\label{sec:lp_solving_small_problems}
Combining Section~B.5 in \cite{swz17} and the proof of Theorem~\ref{sec:l1_solving_small_problems}, for any $p=a/b$ with $a,b$ are integers, we can obtain the $\ell_p$ version of Theorem~\ref{sec:l1_solving_small_problems}.

\subsection{Bicriteria algorithm}\label{sec:lp_bicriteria_algorithm}

We present several bicriteria algorithms with different tradeoffs. We first present an algorithm that runs in nearly linear time and outputs a solution with rank $\wt{O}(k^3)$ in Theorem~\ref{thm:lp_bicriteria_algorithm_rank_k3_nearly_input_sparsity_time}. Then we show an algorithm that runs in $\nnz(A)$ time but outputs a solution with rank $\poly(k)$ in Theorem~\ref{thm:lp_bicriteria_algorithm_rank_k15_input_sparsity_time}. Then we explain an idea which is able to decrease the cubic rank to quadratic, and thus we can obtain Theorem~\ref{thm:lp_bicriteria_algorithm_rank_k2_nearly_input_sparsity_time}.


\begin{theorem}\label{thm:lp_bicriteria_algorithm_rank_k3_nearly_input_sparsity_time}
Given a $3$rd order tensor $A\in \mathbb{R}^{n\times n \times n}$, for any $k\geq 1$, let $r=\wt{O}(k^3)$. There exists an algorithm which takes $\nnz(A)\cdot \wt{O}(k) + n \poly(k) + \poly(k)$ time and outputs three matrices $U,V,W\in \mathbb{R}^{n\times r}$ such that
\begin{align*}
\left\| \sum_{i=1}^r U_i \otimes V_i \otimes W_i - A \right\|_p^p \leq \wt{O}(k^{3-p/2}) \log^3 n \min_{\rank-k~A_k} \| A_k - A\|_p^p
\end{align*}
holds with probability $9/10$.
\end{theorem}

\begin{proof}
We first choose three dense Cauchy transforms $S_i\in \mathbb{R}^{n^2\times s_i}$. According to Section~\ref{sec:def_cauchy_pstable}, for each $i\in[3]$, $A_i S_i$ can be computed in $\nnz(A) \cdot \wt{O}(k)$ time. Then we apply Lemma~\ref{lem:lp_polyk_size_reduction}. We obtain three matrices $Y_1=T_1A_1S_1,Y_2=T_2A_2S_2,Y_3=T_3A_3S_3$ and a tensor $C=A(T_1,T_2,T_3)$. Note that for each $i\in [3]$, $Y_i$ can be computed in $n \poly(k)$ time. Because $C= A(T_1,T_2,T_3)$ and $T_1,T_2,T_3 \in \mathbb{R}^{n\times \wt{O}(k)}$ are three sampling and rescaling matrices, $C$ can be computed in $\nnz(A) + \wt{O}(k^3)$ time. At the end, we just need to run an $\ell_p$-regression solver to find the solution for the problem:
\begin{align*}
\min_{X\in \mathbb{R}^{s_1 \times s_2 \times s_3}} \left\| \sum_{i=1}^{s_1} \sum_{j=1}^{s_2} \sum_{l=1}^{s_3} X_{i,j,l} (Y_1)_i \otimes (Y_2)_j \otimes (Y_3)_j  \right\|_p^p,
\end{align*}
where $(Y_1)_i$ denotes the $i$-th column of matrix $Y_1$. Since the size of the above problem is only $\poly(k)$, this can be solved in $\poly(k)$ time.
\end{proof}

\begin{theorem}\label{thm:lp_bicriteria_algorithm_rank_k15_input_sparsity_time}
Given a $3$rd order tensor $A\in \mathbb{R}^{n\times n \times n}$, for any $k\geq 1$, let $r=\wt{O}(k^{15})$. There exists an algorithm that takes $\nnz(A) + n \poly(k) + \poly(k)$ time and outputs three matrices $U,V,W\in \mathbb{R}^{n\times r}$ such that
\begin{align*}
\left\| \sum_{i=1}^r U_i \otimes V_i \otimes W_i - A \right\|_p^p \leq \poly( k, \log n ) \min_{\rank-k~A_k} \| A_k - A\|_p^p
\end{align*}
holds with probability $9/10$.
\end{theorem}

\begin{proof}
We first choose three sparse $p$-stable transforms $S_i\in \mathbb{R}^{n^2\times s_i}$. According to Section~\ref{sec:def_cauchy_pstable}, for each $i\in[3]$, $A_i S_i$ can be computed in $O(\nnz(A))$ time. Then we apply Lemma~\ref{lem:lp_polyk_size_reduction}, and can obtain three matrices $Y_1=T_1A_1S_1,Y_2=T_2A_2S_2,Y_3=T_3A_3S_3$ and a tensor $C=A(T_1,T_2,T_3)$. Note that for each $i\in [3]$, $Y_i$ can be computed in $n \poly(k)$ time. Because $C= A(T_1,T_2,T_3)$ and $T_1,T_2,T_3 \in \mathbb{R}^{n\times \wt{O}(k)}$ are three sampling and rescaling matrices, $C$ can be computed in $\nnz(A) + \wt{O}(k^3)$ time. At the end, we just need to run an $\ell_p$-regression solver to find the solution to the problem,
\begin{align*}
\min_{X\in \mathbb{R}^{s_1 \times s_2 \times s_3}} \left\| \sum_{i=1}^{s_1} \sum_{j=1}^{s_2} \sum_{l=1}^{s_3} X_{i,j,l} (Y_1)_i \otimes (Y_2)_j \otimes (Y_3)_l - C  \right\|_p^p,
\end{align*}
where $(Y_1)_i$ denotes the $i$-th column of matrix $Y_1$. Since the size of the above problem is only $\poly(k)$, it can be solved in $\poly(k)$ time.
\end{proof}

\begin{theorem}\label{thm:lp_bicriteria_algorithm_rank_k2_nearly_input_sparsity_time}
Given a $3$rd order tensor $A\in \mathbb{R}^{n\times n \times n}$, for any $k\geq 1$, $\epsilon \in (0,1)$, let $r=\wt{O}(k^2)$. There exists an algorithm which takes $\nnz(A)\cdot \wt{O}(k) + n \poly(k) + \poly(k)$ time and outputs three matrices $U,V,W\in \mathbb{R}^{n\times r}$ such that
\begin{align*}
\left\| \sum_{i=1}^r U_i \otimes V_i \otimes W_i - A \right\|_p^p \leq \wt{O}(k^{3-1.5p}) \log^3 n \min_{\rank-k~A_k} \| A_k - A\|_p^p
\end{align*}
holds with probability $9/10$.
\end{theorem}
\begin{proof}
The proof is similar to Theorem~\ref{thm:l1_bicriteria_algorithm_rank_k2_nearly_input_sparsity_time}.
\end{proof}

\begin{algorithm}[h]\caption{$\ell_p$-Low Rank Approximation, Bicriteria Algorithm, $\rank$-$\wt{O}(k^2)$, Input Sparsity Time}\label{alg:lp_bicriteria_algorithm_rank_k2_real_input_sparsity_time}
\begin{algorithmic}[1]
\Procedure{\textsc{LpBicriteriaAlgorithm}}{$A,n,k$} \Comment{Corollary \ref{cor:lp_bicriteria_algorithm_rank_k2_input_sparsity_time}}
\State $s_1\leftarrow s_2 \leftarrow s_3 \leftarrow \wt{O}(k)$.
\State For each $i\in [3]$, choose $S_i\in \mathbb{R}^{n^2 \times s_i}$ to be the composition of a sparse $p$-stable transform and a dense $p$-stable transform. \Comment{Part (\RN{1},\RN{2}) of Theorem~\ref{sec:lp_existence_results}}
\State Compute $A_1 \cdot S_1$, $A_2 \cdot S_2$.
\State For each $i\in [2]$, choose $T_i$ to be a sampling and rescaling diagonal matrix according to the Lewis weights of $A_i S_i$, with $t_i=\wt{O}(k)$ nonzero entries.
\State $C\leftarrow A(T_1,T_2,I)$.
\State $B^{i+(j-1)s_1} \leftarrow \vect( (T_1A_1 S_1)_i \otimes (T_2A_2S_2)_j ), \forall i\in[s_1],j\in[s_2]$.
\State Form objective function $\min_{W}\| W B - C_3\|_1$.
\State Run $\ell_p$-regression solver to find $\wh{W}$.
\State Construct $\wh{U}$ by copying $(A_1 S_1)_i$ to the $(i,j)$-th column of $\wh{U}$.
\State Construct $\wh{V}$ by copying $(A_2 S_2)_j$ to the $(i,j)$-th column of $\wh{V}$ .
\State \Return $\wh{U},\wh{V},\wh{W}$.
\EndProcedure
\end{algorithmic}
\end{algorithm}

As for $\ell_1$, notice that if we first apply a sparse Cauchy transform, we can reduce the rank of the matrix to $\poly(k)$. Theyn we can apply a dense Cauchy transform and further reduce the dimension, while only incurring another $\poly(k)$ factor in the approximation ratio. By combining sparse $p$-stable and dense $p$-stable transforms, we can improve the running time from $\nnz(A) \cdot \wt{O}(k)$ to be $\nnz(A)$ by losing some additional $\poly(k)$ factors in the approximation ratio.

\begin{corollary}\label{cor:lp_bicriteria_algorithm_rank_k2_input_sparsity_time}
Given a $3$rd order tensor $A\in \mathbb{R}^{n\times n \times n}$, for any $k\geq 1$, $\epsilon \in (0,1)$, let $r=\wt{O}(k^2)$. There exists an algorithm which takes $\nnz(A)  + n \poly(k) + \poly(k)$ time and outputs three matrices $U,V,W\in \mathbb{R}^{n\times r}$ such that
\begin{align*}
\left\| \sum_{i=1}^r U_i \otimes V_i \otimes W_i - A \right\|_p^p \leq \poly(k,\log n) \min_{\rank-k~A_k} \| A_k - A\|_p^p
\end{align*}
holds with probability $9/10$.
\end{corollary}

\subsection{Algorithms}\label{sec:lp_algorithm}
In this section, we show two different algorithms by using different kind of sketches. One is shown in Theorem~\ref{thm:lp_input_sparsity_time} which gives a fast running time. Another one is shown in Theorem~\ref{thm:lp_best_approximation_ratio} which gives the best approximation ratio.
\begin{theorem}\label{thm:lp_input_sparsity_time}
Given a $3$rd tensor $A\in \mathbb{R}^{n\times n\times n}$, for any $k\geq 1$, there exists an algorithm which takes $O(\nnz(A))+ n \poly(k) + 2^{\wt{O}(k^2)}$ time and outputs three matrices $U,V,W\in \mathbb{R}^{n\times k}$ such that,
\begin{align*}
\left\| U \otimes V \otimes W - A \right\|_p^p \leq \poly(k,\log n) \min_{\rank-k~A'}\|A' - A \|_p^p.
\end{align*}
holds with probability at least $9/10$.
\end{theorem}
\begin{proof}
First, we apply part (\RN{2}) of Theorem~\ref{thm:lp_existence_results}. Then $A_i S_i$ can be computed in $O(\nnz(A))$ time. Second, we use Lemma~\ref{lem:lp_polyk_size_reduction} to reduce the size of the objective function from $O(n^3)$ to $\poly(k)$ in $n \poly(k)$ time by only losing a constant factor in approximation ratio. Third, we use Claim~\ref{cla:tensor_frobenius_relax_ell1_lowrank} to relax the objective function from entry-wise $\ell_p$-norm to Frobenius norm, and this step causes us to lose some other $\poly(k)$ factors in approximation ratio. As a last step, we use Theorem~\ref{thm:f_solving_small_problems} to solve the Frobenius norm objective function.
\end{proof}

\begin{theorem}\label{thm:lp_best_approximation_ratio}
Given a $3$rd order tensor $A\in \mathbb{R}^{n\times n\times n}$, for any $k\geq 1$, there exists an algorithm that takes $n^{\wt{O}(k)} 2^{\wt{O}(k^3)}$ time and output three matrices $U,V,W\in \mathbb{R}^{n\times k}$ such that,
\begin{align*}
\| U \otimes V \otimes W - A \|_p^p \leq \wt{O}(k^{3-1.5p}) \min_{\rank-k~A'}\|A' - A \|_p^p.
\end{align*}
holds with probability at least $9/10$.
\end{theorem}
\begin{proof}

First, we apply part (\RN{3}) of Theorem~\ref{thm:lp_existence_results}. Then, guessing $S_i$ requires $n^{\wt{O}(k)}$ time. Second, we use Lemma~\ref{lem:lp_polyk_size_reduction} to reduce the size of the objective from $O(n^3)$ to $\poly(k)$ in polynomial time while only losing a constant factor in approximation ratio. Third, we solve the small optimization problem.
\end{proof}

\subsection{CURT decomposition}\label{sec:lp_curt}
\begin{theorem}\label{thm:lp_curt_algorithm}
Given a $3$rd order tensor $A\in \mathbb{R}^{n\times n \times n}$, let $k\geq 1$, and let $U_B,V_B,W_B\in \mathbb{R}^{n\times k}$ denote a rank-$k$, $\alpha$-approximation to $A$. Then there exists an algorithm which takes $O(\nnz(A)) + O(n^2) \poly(k)$ time and outputs three matrices $C\in \mathbb{R}^{n\times c}$ with columns from $A$, $R\in \mathbb{R}^{n\times r}$ with rows from $A$, $T\in \mathbb{R}^{n\times t}$ with tubes from $A$, and a tensor $U\in \mathbb{R}^{c\times r\times t}$ with $\rank(U)=k$ such that $c=r=t=O(k\log k\log \log k)$, and
\begin{align*}
\left\| \sum_{i=1}^c \sum_{j=1}^r \sum_{l=1}^t U_{i,j,l} \cdot C_i \otimes R_j \otimes T_l - A \right\|_p^p \leq \wt{O}(k^{3-1.5p}) \alpha \min_{\rank-k~A'} \| A' - A\|_p^p
\end{align*}
holds with probability $9/10$.
\end{theorem}

\begin{proof}
We define
\begin{align*}
\OPT : = \min_{\rank-k~A'} \| A' - A\|_p^p.
\end{align*}
We already have three matrices $U_B\in \mathbb{R}^{n\times k}$, $V_B\in \mathbb{R}^{n\times k}$ and $W_B\in \mathbb{R}^{n\times k}$ and these three matrices provide a $\rank$-$k$, $\alpha$ approximation to $A$, i.e.,
\begin{align}\label{eq:lp_cur_UBVBWB_minus_A}
\left\| \sum_{i=1}^k ( U_B )_i \otimes (V_B)_i \otimes (W_B)_i - A \right\|_p^p \leq \alpha \OPT.
\end{align}
Let $B_1 = V_B^\top \odot W_B^\top \in \mathbb{R}^{k\times n^2}$ denote the matrix where the $i$-th row is the vectorization of $(V_B)_i \otimes (W_B)_i$. By Section B.3 in~\cite{swz17}, we can compute $D_1 \in \mathbb{R}^{n^2 \times n^2}$ which is a sampling and rescaling matrix corresponding to the Lewis weights of $B_1^\top$ in $O(n^2\poly(k))$ time, and there are $d_1 = O(k\log k\log\log k)$ nonzero entries on the diagonal of $D_1$. Let $A_i\in \mathbb{R}^{n\times n^2}$ denote the matrix obtained by flattening $A$ along the $i$-th direction, for each $i\in [3]$.

Define $U^*\in \mathbb{R}^{n\times k}$ to be the optimal solution to $\underset{U\in \mathbb{R}^{n\times k} }{\min} \| U B_1 - A_1\|_p^p$, $\wh{U} = A_1 D_1 (B_1 D_1)^\dagger \in \mathbb{R}^{n\times k}$, $V_0 \in \mathbb{R}^{n\times k}$ to be the optimal solution to $\underset{V\in \mathbb{R}^{n\times k} }{\min} \| V \cdot  (\wh{U}^\top \odot W_B^\top) - A_2 \|_p^p $, and $U'$ to be the optimal solution to $\underset{U\in \mathbb{R}^{n\times k}}{\min} \| U B_1 D_1 - A_1 D_1 \|_p^p$.

By Claim~\ref{cla:ell2_relax_ell1_regression}, we have
\begin{align*}
\|\wh{U} B_1 D_1 - A_1 D_1 \|_p^p \leq d_1^{1-p/2} \| U' B_1 D_1 - A_1 D_1\|_p^p.
\end{align*}
Due to Lemma E.11 and Lemma E.8 in \cite{swz17}, with constant probability, we have
\begin{align}\label{eq:lp_cur_Uwh_B1_minus_A1}
\| \wh{U} B_1 - A _1 \|_p^p \leq d_1^{1-p/2} \alpha_{D_1} \| U^* B_1 - A_1 \|_p^p,
\end{align}
where $\alpha_{D_1} = O(1)$.

Recall that $( \wh{U}^\top \odot W_B^\top) \in \mathbb{R}^{k\times n^2}$ denotes the matrix where the $i$-th row is the vectorization of $\wh{U}_i \otimes (W_B)_i$, $\forall i\in [k]$. Now, we can show,
\begin{align}\label{eq:lp_cur_V0B2_minus_A2}
\| V_0 \cdot ( \wh{U}^\top \odot W_B^\top) - A_2 \|_p^p \leq & ~ \| \wh{U} B_1 - A_1 \|_p^p & \text{~by~} V_0 = \underset{V\in \mathbb{R}^{n\times k}}{\arg\min} \| V \cdot ( \wh{U}^\top \odot W_B^\top) - A_2  \|_p^p \notag \\
\lesssim & ~ d_1^{1-p/2}\| U^* B_1 - A_1 \|_p^p & \text{~by~Equation~\eqref{eq:lp_cur_Uwh_B1_minus_A1}} \notag \\
\leq & ~ d_1^{1-p/2} \| U_B B_1 - A_1 \|_p^p & \text{~by~} U^* = \underset{U\in \mathbb{R}^{n\times k} }{\arg\min} \| U B_1 - A_1 \|_p^p \notag \\
\leq & ~ O(d_1^{1-p/2}) \alpha \OPT. & \text{~by~Equation~\eqref{eq:lp_cur_UBVBWB_minus_A}}
\end{align}

We define $B_2= \wh{U}^\top \odot W_B^\top$. We can compute $D_2\in \mathbb{R}^{n^2 \times n^2}$ which is a sampling and rescaling matrix corresponding to the $\ell_p$ Lewis weights of $B_2^\top$ in $O(n^2 \poly(k))$ time, and there are $d_2 = O(k\log k\log \log k)$ nonzero entries on the diagonal of $D_2$.

Define $V^*\in \mathbb{R}^{n\times k}$ to be the optimal solution of $\min_{V\in \mathbb{R}^{n\times k}} \| V B_2 - A_2 \|_p^p$, $\wh{V}= A_2 D_2 (B_2 D_2)^\dagger \in \mathbb{R}^{n\times k}$, $W_0\in \mathbb{R}^{n\times k}$ to be the optimal solution of $\underset{W\in \mathbb{R}^{n\times k}}{\min} \| W\cdot (\wh{U}^\top \odot \wh{V}^\top) - A_3 \|_p^p$, and $V'$ to be the optimal solution of $\underset{V\in \mathbb{R}^{n\times k}}{\min} \|V B_2 D_2 - A_2 D_2\|_p^p$.

By Claim~\ref{cla:ell2_relax_ell1_regression}, we have
\begin{align*}
\| \wh{V} B_2 D_2 - A_2 D_2 \|_p^p \leq d_2^{1-p/2} \|V' B_2 D_2 - A_2 D_2 \|_p^p.
\end{align*}
Due to Lemma E.11 and Lemma E.8 in \cite{swz17}, with constant probability, we have
\begin{align}\label{eq:lp_cur_Vwh_B2_minus_A2}
\| \wh{V} B_2 - A_2 \|_p^p \leq d_2^{1-p/2} \alpha_{D_2} \| V^* B_2 - A_2 \|_p^p,
\end{align}
where $\alpha_{D_2} = O(1)$.

Recall that $(\wh{U}^\top \odot \wh{V}^\top) \in \mathbb{R}^{k\times n^2}$ denotes the matrix for which the $i$-th row is the vectorization of $\wh{U}_i \otimes \wh{V}_i$, $\forall i\in [k]$. Now, we can show,
\begin{align}\label{eq:lp_cur_W0B3_minus_A3}
&\| W_0 \cdot (\wh{U}^\top \odot \wh{V}^\top ) - A_3 \|_p^p \notag\\
\leq & ~ \| \wh{V} B_2 - A_2 \|_p^p & \text{~by~} W_0 = \underset{W\in \mathbb{R}^{n\times k} }{\arg\min} \| W \cdot ( \wh{U}^\top \odot \wh{V}^\top ) - A_3 \|_p^p \notag \\
\lesssim & ~ d_2^{1-p/2} \| V^* B_2 - A_2 \|_p^p & \text{~by~Equation~\eqref{eq:lp_cur_Vwh_B2_minus_A2}} \notag \\
\leq & ~ d_2^{1-p/2} \| V_0 B_2 - A_2 \|_p^p & \text{~by~} V^* =\underset{V\in \mathbb{R}^{n\times k}}{\arg\min} \| V B_2 - A_2 \|_p^p \notag \\
\leq & ~ O((d_1d_2)^{1-p/2}) \alpha \OPT. & \text{~by~Equation~\eqref{eq:lp_cur_V0B2_minus_A2}}
\end{align}

We define $B_3= \wh{U}^\top \odot \wh{V}^\top$. We can compute $D_3\in \mathbb{R}^{n^2 \times n^2}$ which is a sampling and rescaling matrix corresponding to the $\ell_p$ Lewis weights of $B_3^\top$ in $O(n^2 \poly(k))$ time, and there are $d_3 = O(k\log k\log \log k)$ nonzero entries on the diagonal of $D_3$.

Define $W^*\in \mathbb{R}^{n\times k}$ to be the optimal solution to $\min_{W\in \mathbb{R}^{n\times k}} \| W B_3 - A_3 \|_p^p$, $\wh{W}= A_3 D_3 (B_3 D_3)^\dagger \in \mathbb{R}^{n\times k}$,
and $W'$ to be the optimal solution to $\underset{W\in \mathbb{R}^{n\times k}}{\min} \|W B_3 D_3 - A_3 D_3\|_p^p$.

By Claim~\ref{cla:ell2_relax_ell1_regression}, we have
\begin{align*}
\| \wh{W} B_3 D_3 - A_3 D_3 \|_p^p \leq d_3^{1-p/2} \|W' B_3 D_3 - A_3 D_3 \|_p^p.
\end{align*}
Due to Lemma E.11 and Lemma E.8 in \cite{swz17}, with constant probability, we have
\begin{align}\label{eq:lp_cur_Wwh_B3_minus_A3}
\| \wh{W} B_3 - A_3 \|_p^p \leq d_3^{1-p/2} \alpha_{D_3} \| W^* B_3 - A_3 \|_p^p,
\end{align}
where $\alpha_{D_3} = O(1)$. Now we can show,
\begin{align*}
\| \wh{W} B_3 - A_3 \|_p^p \lesssim & ~ d_3^{1-p/2} \| W^* B_3 - A_3 \|_p^p, & \text{~by~Equation~\eqref{eq:lp_cur_Wwh_B3_minus_A3}} \\
\leq & ~  d_3^{1-p/2} \| W_0 B_3 - A_3 \|_p^p, & \text{~by~}W^* = \underset{W\in \mathbb{R}^{n\times k} }{\arg\min} \| W B_3 - A_3 \|_p^p \\
\leq & ~ O((d_1d_2d_3)^{1-p/2}) \alpha \OPT. & \text{~by~Equation~\eqref{eq:lp_cur_W0B3_minus_A3}}
\end{align*}
Thus, it implies,
\begin{align*}
\left\| \sum_{i=1}^k \wh{U}_i \otimes \wh{V}_i \otimes \wh{W}_i - A \right\|_p^p \leq \poly(k,\log n) \OPT.
\end{align*}
where $\wh{U} = A_1 D_1 (B_1 D_1)^\dagger$, $\wh{V} = A_2D_2 (B_2 D_2)^\dagger$, $\wh{W}=A_3D_3 (B_3 D_3)^\dagger$.

\end{proof}

\newpage
\section{Robust Subspace Approximation (Asymmetric Norms for Arbitrary Tensors)}\label{sec:lvu}
Recently, \cite{cw15soda} and \cite{cw15focs} study the linear regression problem and low-rank approximation problem under M-Estimator loss functions. In this section, we extend the matrix version of the low rank approximation problem to tensors, i.e., in particular focusing on tensor low-rank approximation under M-Estimator norms. Note that M-Estimators are very different from Frobenius norm and Entry-wise $\ell_1$ norm, which are symmetric norms. Namely, flattening the tensor objective function along any of the dimensions does not change the cost if the norm is Frobenius or Entry-wise $\ell_1$-norm. However, for M-Estimator norms, we cannot flatten the tensor along all three dimensions. This property makes the tensor low-rank approximation problem under M-Estimator norms more difficult. This section can be split into two independent parts. Section~\ref{sec:lv_l122} studies the $\ell_1$-$\ell_2$-$\ell_2$ norm setting, and Section~\ref{sec:lu_l112} studies the $\ell_1$-$\ell_1$-$\ell_2$ norm setting.
%

\subsection{Preliminaries}\label{sec:lvu_preliminaries}

\begin{definition}[Nice functions for $M$-Estimators, ${\cal M}_2$, ${\cal L}_p$, \cite{cw15focs}]
We say an $M$-Estimator is {\bf nice} if $M(x) = M(-x)$, $M(0)=0$, $M$ is non-decreasing in $|x|$, there is a constant $C_M>0$ and a constant $p\geq 1$ so that for all $a,b\in \mathbb{R}_{>0}$ with $a\geq b$, we have
\begin{align*}
C_m \frac{|a|}{|b|} \leq \frac{ M(a) }{ M(b) } \leq (\frac{a}{b})^p,
\end{align*}
and also that $M(x)^{\frac{1}{p}}$ is subadditive, that is, $M(x+y)^\frac{1}{p} \leq M(x)^\frac{1}{p} + M(y)^{\frac{1}{p}}$.

Let ${\cal M}_2$ denote the set of such nice $M$-estimators, for $p=2$. Let ${\cal L}_p$ denote $M$-Estimators with $M(x) = |x|^p$ and $p\in [1,2)$.
\end{definition}

\subsection{$\ell_1$-Frobenius (a.k.a $\ell_1$-$\ell_2$-$\ell_2$) norm}\label{sec:lv_l122}
Section~\ref{sec:lv_l122_definitions} presents basic definitions and facts for the $\ell_1$-$\ell_2$-$\ell_2$ norm setting. Section~\ref{sec:lv_l122_sampling_rescaling} introduces some useful tools. Section~\ref{sec:lv_l122_dilation_contraction} presents the ``no dilation'' and ``no contraction'' bounds, which are the key ideas for reducing the problem to a ``generalized'' Frobenius norm low rank approximation problem. Finally, we provide our algorithms in Section~\ref{sec:lv_l122_algorithms}.

\subsubsection{Definitions}\label{sec:lv_l122_definitions}
We first give the definition for the $v$-norm of a tensor, and then give the definition of the $v$-norm for a matrix and a weighted version of the $v$-norm for a matrix.
\begin{definition}[Tensor $v$-norm]
For an $n\times n\times n$ tensor $A$, we define the $v$-norm of $A$, denoted $\| A \|_v$, to be
\begin{align*}
\left( \sum_{i=1}^n M( \| A_{i,*,*} \|_F ) \right) ^{1/p},
\end{align*}
where $A_{i,*,*}$ is the $i$-th face of $A$ (along the $1$st direction), and $p$ is a parameter associated with the function $M()$, which defines a nice $M$-Estimator.
\end{definition}

\begin{definition}[Matrix $v$-norm]
For an $n\times d$ matrix $A$, we define the $v$-norm of $A$, denoted $\| A \|_v$, to be
\begin{align*}
\sum_{i=1}^n M( \| A_{i,*} \|_2 )^{1/p},
\end{align*}
where $A_{i,*}$ is the $i$-th row of $A$, and $p$ is a parameter associated with the function $M()$, which defines a nice $M$-Estimator.
\end{definition}

\begin{definition}
Given matrix $A\in \mathbb{R}^{n\times d}$, let $A_{i,*}$ denote the $i$-th row of $A$. Let $T_S \subset [n]$ denote the indices $i$ such that $e_i$ is chosen for $S$. Using a probability vector $q$ and a sampling and rescaling matrix $S\in \mathbb{R}^{n\times n}$ from $q$, we will estimate $\| A \|_v$ using $S$ and a re-weighted version, $\|S \cdot \|_{v,w'}$ of $\| \cdot \|_{v}$, with
\begin{align*}
\| SA \|_{v,w'} = \left( \sum_{i\in T_S} w_i' M(\| A_{i,*} \|_2 ) \right)^{1/p},
\end{align*}
where $w_i'= w_i/q_i$. Since $w'$ is generally understood, we will usually just write $\| SA \|_v$. We  will also need an ``entrywise row-weighted'' version :
\begin{align*}
|||SA||| = \left( \sum_{i\in T_S} \frac{w_i}{q_i} \| A_{i,*} \|_M^p \right)^{1/p} = \left( \sum_{i\in T_S, j \in [d]} \frac{w_i}{q_i} M(A_{i,j}) \right)^{1/p},
\end{align*}
where $A_{i,j}$ denotes the entry in the $i$-th row and $j$-th column of $A$.
\end{definition}

\begin{fact}
For $p=1$, for any two matrices $A$ and $B$, we have $\| A  + B\|_v \leq \| A \|_v + \|B\|_v$. For any two tensors $A$ and $B$, we have $\| A + B \|_v \leq \| A \|_v + \| B\|_v$.
\end{fact}

\subsubsection{Sampling and rescaling sketches}\label{sec:lv_l122_sampling_rescaling}

Note that Lemmas 42 and 44 in \cite{cw15focs} are stronger than stated. In particular, we do not need to assume $X$ is a square matrix. For any $m \geq z$, if $X\in \mathbb{R}^{d\times m}$, then we have the same result.
\begin{lemma}[Lemma 42 in \cite{cw15focs}]
Let $\rho>0$ and integer $z>0$. For sampling matrix $S$, suppose for a given $y\in \mathbb{R}^d$ with failure probability $\delta$ it holds that $\| SA y\|_M = (1\pm 1/10) \| A y \|_M$. There is $K_1 = O(z^2 / C_M)$ so that with failure probability $\delta (K_{\cal N} / C_{M})^{(1+p)d}$, for a constant $K_{\cal N}$, any rank-$z$ matrix $X\in \mathbb{R}^{d\times m}$ has the property that if $\| A X\|_v \geq K_1 \rho$, then $\| S A X \|_v\geq \rho$, and that if $\| A X\|_v \leq \rho /K_1$, then $\| SAX\|_v \leq \rho$.
\end{lemma}

\begin{lemma}[Lemma 44 in \cite{cw15focs}]
Let $\delta, \rho>0$ and integer $z>0$. Given matrix $A\in \mathbb{R}^{n\times d}$, there exists a sampling and rescaling matrix $S\in \mathbb{R}^{n\times n}$ with $r=O(\gamma(A,M,w) \epsilon^{-2} dz^2 \log(z/\epsilon) \log(1/\delta) )$ nonzero entries such that, with probability at least $1-\delta$, for any $\rank$-$z$ matrix $X\in \mathbb{R}^{d\times m}$, we have either
\begin{align*}
 \| SA X \|_v \geq \rho,
\end{align*}
or
\begin{align*}
(1-\epsilon)\| AX \|_v -\epsilon \rho \leq \| SAX \|_v \leq  (1+\epsilon) \| A X \|_v + \epsilon \rho  .
\end{align*}
\end{lemma}

\begin{lemma}[Lemma 43 in \cite{cw15focs}]\label{lem:lv_lemma43_cw15}
For $r>0$, let $\wh{r} = r/\gamma(A,M,w)$, and let $q\in \mathbb{R}^n$ have
\begin{align*}
q_i = \min \{ 1, \wh{r} \gamma_i (A,M,w) \}.
\end{align*}
Let $S$ be a sampling and rescaling matrix generated using $q$, with weights as usual $w_i' = w_i /q_i$. Let $W\in \mathbb{R}^{d\times z}$, and $\delta >0$. There is an absolute constant $C$ so that for $\wh{r} \geq C z \log(1/\delta) /\epsilon^2$, with probability at least $1-\delta$, we have
\begin{align*}
(1-\epsilon) \| A W \|_{v,w}  \leq \| SA W\|_{v,w'} \leq (1+\epsilon) \| A W \|_{v,w}.
\end{align*}
\end{lemma}

\subsubsection{No dilation and no contraction}\label{sec:lv_l122_dilation_contraction}

\begin{lemma}\label{lem:lv_cost_preserving_sketch}
 Given matrices $A\in \mathbb{R}^{n\times m}$, $U\in \mathbb{R}^{n\times d}$, let $V^* = \underset{\rank-k~ V\in \mathbb{R}^{d\times m}}{\arg\min} \| U V - A \|_v$. If $S\in \mathbb{R}^{s\times n}$ has at most $c_1$-dilation on $UV^*-A$, i.e.,
\begin{align*}
\| S (U V^* - A ) \|_v \leq c_1 \| UV^* - A \|_v,
\end{align*}
and it has at most $c_2$-contraction on $U$, i.e.,
\begin{align*}
\forall x \in \mathbb{R}^d, \| SU x \|_v \geq \frac{1}{c_2} \| U x \|_v,
\end{align*}
then $S$ has at most $(c_2,c_1+\frac{1}{c_2})$-contraction on $(U,A)$, i.e.,
\begin{align*}
\forall ~\rank-k~ V \in \mathbb{R}^{d\times m}, \| SUV -SA \|_v \geq \frac{1}{c_2} \| UV - A\|_v - (c_1+\frac{1}{c_2}) \| UV^* -A \|_v.
\end{align*}
\end{lemma}
\begin{proof}
Let $A\in \mathbb{R}^{n\times m}$, $U\in \mathbb{R}^{n\times d}$ and $S\in \mathbb{R}^{s\times n}$ be the same as that described in the lemma. Let $(V-V^*)_j$ denote the $j$-th column of $V-V^*$. Then $\forall ~ \rank-k ~V\in \mathbb{R}^{d\times m}$,
\begin{align*}
\| SU V - S A\|_v \geq & ~ \| SUV - SUV^*\|_v - \| SU V^* - SA \|_v \\
\geq & ~ \| SUV - SUV^* \|_v - c_1 \| UV^* - A \|_v \\
= & ~ \| SU (V- V^*) \|_v - c_1 \| UV^* - A \|_v \\
= & ~ \sum_{j=1}^m \| SU (V-V^*)_j \|_v  - c_1 \| UV^* - A \|_v \\
\geq & ~ \sum_{j=1}^m \frac{1}{c_2} \| U (V - V^*)_j \|_v - c_1 \| UV^* -A \|_v \\
= & ~ \frac{1}{c_2} \| UV - UV^* \|_v -c_1 \| U V^* - A \|_v \\
\geq & ~ \frac{1}{c_2} \| UV - A \|_v - \frac{1}{c_2} \| UV^* -A \|_v - c_1 \| UV^* -A \|_v \\
= & ~\frac{1}{c_2} \| U V - A \|_v - \left(  (\frac{1}{c_2} +c_2) \| UV^* -A \|_v \right),
\end{align*}
where the first inequality follows by the triangle inequality, the second inequality follows since $S$ has at most $c_1$ dilation on $UV^* - A$, the third inequality follows since $S$ has at most $c_2$ contraction on $U$, and the fourth inequality follows by the triangle inequality.
\end{proof}

\begin{claim}\label{cla:lv_sampling_matrix_markov_bound}
Given matrix $A\in \mathbb{R}^{n\times m}$, for any distribution $p=(p_1,p_2,\cdots,p_n)$ define random variable $X$ such that $X= \| A_i \|_2 /p_i$ with probability $p_i$ where $A_i$ is the $i$-th row of matrix $A$. Then take $m$ independent samples $X^1, X^2, \cdots, X^m$, and let $Y = \frac{1}{m} \sum_{j=1}^m X^j$. We have
\begin{align*}
\Pr[Y \leq 1000 \| A\|_v ] \geq .999.
\end{align*}
\end{claim}
\begin{proof}
We can compute the expectation of $X^j$, for any $j\in [m]$,
\begin{align*}
\E [X^j ] = \sum_{i=1}^n \frac{\| A_i \|_2 }{p_i} \cdot p_i = \| A\|_v.
\end{align*}
Then $\E[Y] = \frac{1}{m} \sum_{j=1}^m \E[X^j] = \| A\|_v$. Using Markov's inequality, we have
\begin{align*}
\Pr[Y \geq \| A\|_v ] \leq .001.
\end{align*}
\end{proof}

\begin{lemma}\label{lem:lv_rank_k_dimensional_subspace}
For any fixed $U^*\in \mathbb{R}^{n\times d}$ and $\rank$-$k$ $V^*\in \mathbb{R}^{d\times m}$ with $d=\poly(k)$, there exists an algorithm that takes $\poly(n,d)$ time to compute a sampling and rescaling diagonal matrix $S\in \mathbb{R}^{n\times n}$ with $s= \poly(k)$ nonzero entries such that, with probability at least $.999$, we have: for all $\rank$-$k$ $V\in \mathbb{R}^{d\times m}$,
\begin{align*}
 \|U^* V^* - U^* V \|_v \lesssim \| S U^* V^* - S U^* V \|_v \lesssim \| U^* V^* - U^* V \|_v.
\end{align*}
\end{lemma}

\begin{lemma}[No dilation]\label{lem:lv_no_dilation}
Given matrices $A\in \mathbb{R}^{n\times m}$, $U^*\in \mathbb{R}^{n\times d}$ with $d=\poly(k)$, define $V^*\in \mathbb{R}^{d\times m}$ to be the optimal solution $\underset{\rank-k~V\in \mathbb{R}^{d\times m}}{\min} \| U^* V - A \|_v$. Choose a sampling and rescaling diagonal matrix $S\in \mathbb{R}^{n\times n}$ with $s=\poly(k)$ according to Lemma~\ref{lem:lv_lemma43_cw15}. Then with probability at least $.99$, we have: for all $\rank$-$k$ $V\in \mathbb{R}^{d\times m}$,
\begin{align*}
\| SU^* V - SA\|_v \lesssim \| U^*V^* - U^* V \|_v + O(1) \| U^* V^* - A \|_v \lesssim \|U^* V -A \|_v.
\end{align*}
\end{lemma}
\begin{proof}
Using Claim~\ref{cla:lv_sampling_matrix_markov_bound} and Lemma~\ref{lem:lv_rank_k_dimensional_subspace}, we have with probability at least $.99$, for all $\rank$-$k$ $V\in \mathbb{R}^{d\times m}$,
\begin{align*}
 & ~ \| SU^* V - SA \|_v \\
\leq & ~ \| SU^* V - SU^* V^* \|_v + \| SU^* V^* - SA \|_v & \text{~by~triangle~inequality} \\
\lesssim & ~ \| SU^* V - SU^* V^* \|_v + O(1) \| U^* V^* - A \|_v & \text{~by~Claim~\ref{cla:lv_sampling_matrix_markov_bound}} \\
\lesssim & ~ \| U^* V - U^* V^*\|_v + O(1) \| U^* V^* -A \|_v & \text{~by~Lemma~\ref{lem:lv_rank_k_dimensional_subspace}} \\
\lesssim & ~ \| U^* V - A \|_v + \|U^* V^* -A \|_v + O(1) \| U^* V^* - A \|_v & \text{~by~triangle~inequality} \\
\lesssim & ~ \|U^* V - A \|_v.
\end{align*}
\end{proof}

\begin{lemma}[No contraction]\label{lem:lv_no_contraction}
Given matrices $A\in \mathbb{R}^{n\times m}$, $U^*\in \mathbb{R}^{n\times d}$ with $d=\poly(k)$, define $V^*\in \mathbb{R}^{d\times m}$ to be the optimal solution $\underset{\rank-k~V\in \mathbb{R}^{d\times m}}{\min} \| U^* V - A \|_v$. Choose a sampling and rescaling diagonal matrix $S\in \mathbb{R}^{n\times n}$ with $s=\poly(k)$ according to Lemma~\ref{lem:lv_lemma43_cw15}. Then with probability at least $.99$, we have: for all $\rank$-$k$ $V\in \mathbb{R}^{d\times m}$,
\begin{align*}
\| U^* V - A \|_v \lesssim \| S U^* V - S A \|_v + O(1) \| U^* V^* - A \|_v.
\end{align*}
\end{lemma}
\begin{proof}
This follows by Lemma~\ref{lem:lv_cost_preserving_sketch}, Claim~\ref{cla:lv_sampling_matrix_markov_bound} and Lemma~\ref{lem:lv_no_dilation}.
\end{proof}

\subsubsection{Oblivious sketches, \textsc{MSketch}}\label{sec:lv_l122_oblivious_sketches}
In this section, we recall a concept called $M$-sketches for $M$-estimators which is defined in \cite{cw15soda}. $M$-sketch is an oblivious sketch for matrices.

\begin{theorem}[Theorem 3.1 in \cite{cw15soda}]
Let $\OPT$ denote $\min_{x\in \mathbb{R}^d} \| A x - b \|_G$. There is an algorithm that in $O(\nnz(A)) + \poly(d\log n)$ time, with constant probability finds $x'$ such that $\|Ax'-b\|_G \leq O(1) \OPT$.
\end{theorem}

\begin{definition}[M-Estimator sketches or \textsc{MSketch} \cite{cw15soda}]\label{def:lv_l122_m_sketches}
Given parameters $N,n,m,b>1$, define $h_{\max} = \lfloor \log_b (n/m) \rfloor$, $\beta = (b-b^{-h_{\max}}) / (b-1)$ and $s=N h_{\max}$.
For each $p\in[n]$, $\sigma_p,g_p,h_p$ are generated (independently) in the following way,
\begin{align*}
\sigma_p &\leftarrow \pm 1, & \mathrm{~chosen~with~equal~probability}, \\
g_p &\in [N], & \mathrm{~chosen~with~equal~probability}, \\
h_{p} &\leftarrow t, & \mathrm{~chosen~with~probability~} 1/(\beta b^t) \mathrm{~for~} t\in \{0,1,\cdots h_{\max}\}.
\end{align*}
For each $p\in[n]$, we define $j_p = g_p + N h_p$.
 Let $w\in \mathbb{R}^{s}$ denote the scaling vector such that, for each $j\in [s]$,
\begin{align*}
w_j =
\begin{cases}
 \beta b^{h_p}, & \mathrm{if~there~exists~}p\in[n]\mathrm{~s.t.} j = j_p,\\
 0 & \mathrm{~otherwise.}
 \end{cases}
\end{align*}
Let $\ov{S}\in \mathbb{R}^{N h_{\max} \times n} $ be such that, for each $j\in [s]$,for each $p\in [n]$,
\begin{align*}
\ov{S}_{j,p} =
\begin{cases}
 \sigma_p,  & \mathrm{~if~} j = g_p + N \cdot h_p, \\
 0,  & \mathrm{~otherwise.}
 \end{cases}
\end{align*}
Let $D_{w}$ denote the diagonal matrix where the $i$-th entry on the diagonal is the $i$-th entry of $w$. Let $S=D_w \ov{S}$.
We say $(\ov{S},w)$ or $S$ is an \textsc{MSketch}.
\end{definition}

\begin{definition}[Tensor $\| \|_{v,w}$-norm]
For a tensor $A\in \mathbb{R}^{d\times n_1 \times n_2}$ and a vector $w\in \mathbb{d}$, we define
\begin{align*}
\| A \|_{v,w} = \sum_{i=1}^d w_i \| A_{i,*,*}\|_F.
\end{align*}
\end{definition}
Let $(\ov{S},w)$ denote an \textsc{MSketch}, and let $S= D_{w} \ov{S}$. If $v$ corresponds to a scale-invariant M-Estimator, then for any three matrices $U,V,W$, we have the following,
\begin{align*}
\| (\ov{S} U ) \otimes V \otimes W \|_{v,w} = \| (D_w \ov{S} U ) \otimes V \otimes W \|_{v} = \| (S U ) \otimes V \otimes W \|_{v}.
\end{align*}

\begin{fact}
For a tensor $A\in \mathbb{R}^{n\times n\times n}$, let $S\in \mathbb{R}^{s\times n}$ denote an \textsc{MSketch} (defined in \ref{def:lv_l122_m_sketches}) with $s=\poly(k,\log n)$. Then $SA$ can be computed in $O(\nnz(A))$ time.
\end{fact}

\begin{lemma}\label{lem:lv_rank_k_dimensional_subspace_oblivious}
For any fixed $U^*\in \mathbb{R}^{n\times d}$ and $\rank$-$k$ $V^*\in \mathbb{R}^{d\times m}$ with $d=\poly(k)$, let $S\in \mathbb{R}^{s\times n}$ denote an \textsc{MSketch} (defined in Definition~\ref{def:lv_l122_m_sketches}) with $s= \poly(k,\log n)$ rows. Then with probability at least $.999$, we have: for all $\rank$-$k$ $V\in \mathbb{R}^{d\times m}$,
\begin{align*}
 \|U^* V^* - U^* V \|_v \lesssim \| S U^* V^* - S U^* V \|_{v} \lesssim \| U^* V^* - U^* V \|_v.
\end{align*}
\end{lemma}

\begin{lemma}[No dilation, Theorem 3.4 in \cite{cw15soda}]\label{lem:lv_no_dilation_oblivious}
Given matrices $A\in \mathbb{R}^{n\times m}$, $U^*\in \mathbb{R}^{n\times d}$ with $d=\poly(k)$, define $V^*\in \mathbb{R}^{d\times m}$ to be the optimal solution to $\underset{\rank-k~V\in \mathbb{R}^{d\times m}}{\min} \| U^* V - A \|_v$. Choose an \textsc{MSketch} $S\in \mathbb{R}^{s\times n}$   with $s=\poly(k,\log n)$ according to Definition~\ref{def:lv_l122_m_sketches}. Then with probability at least $.99$, we have: for all $\rank$-$k$ $V\in \mathbb{R}^{d\times m}$,
\begin{align*}
\| SU^* V - SA\|_{v} \lesssim \| U^*V^* - U^* V \|_v + O(1) \| U^* V^* - A \|_v \lesssim \|U^* V -A \|_v.
\end{align*}
\end{lemma}

\begin{lemma}[No contraction]\label{lem:lv_no_contraction_oblivious}
Given matrices $A\in \mathbb{R}^{n\times m}$, $U^*\in \mathbb{R}^{n\times d}$ with $d=\poly(k)$, define $V^*\in \mathbb{R}^{d\times m}$ to be the optimal solution to $\underset{\rank-k~V\in \mathbb{R}^{d\times m}}{\min} \| U^* V - A \|_v$. Choose an \textsc{MSketch} $S\in \mathbb{R}^{s\times n}$  with $s=\poly(k,\log n)$ according to Definition~\ref{def:lv_l122_m_sketches}. Then with probability at least $.99$, we have: for all $\rank$-$k$ $V\in \mathbb{R}^{d\times m}$,
\begin{align*}
\| U^* V - A \|_v \lesssim \| S U^* V - S A \|_{v} + O(1) \| U^* V^* - A \|_v.
\end{align*}
\end{lemma}

\subsubsection{Running time analysis}\label{sec:lv_l122_time}

\begin{lemma}
Given a tensor $A\in \mathbb{R}^{n\times d\times d}$, let $S\in \mathbb{R}^{s\times n}$ denote an \textsc{MSketch} with $s$ rows. Let $SA$ denote a tensor that has size $s\times d \times d$. For each $i\in \{2,3\}$, let $(SA)_i\in \mathbb{R}^{d \times ds}$ denote a matrix obtained by flattening tensor $SA$ along the $i$-th dimension. For each $i\in \{2,3\}$, let $S_i\in \mathbb{R}^{ds\times s_i}$ denote a CountSketch transform with $s_i$ columns. For each $i\in \{2,3\}$, let $T_i\in \mathbb{R}^{t_i\times d}$ denote a CountSketch transform with $t_i$ rows. Then \\
$\mathrm{(\RN{1})}$ For each $i\in \{2,3\}$, $ (SA)_i S_i$ can be computed in $O(\nnz(A))$ time.\\
$\mathrm{(\RN{2})}$ For each $i\in \{2,3\}$, $T_i (SA)_i S_i$ can be computed in $O(\nnz(A))$ time.
\end{lemma}

\begin{proof}
Proof of Part (\RN{1}). First note that $(SA)_2 S_2$ has size $n\times S_2$. Thus for each $i\in [d], j\in [s_2]$, we have,
\begin{align*}
( (SA)_2 S_2 )_{i,j} = & ~ \sum_{x'=1}^{ds} ((SA)_2)_{i,x'} (S_2)_{x',j} & \text{~by~} (SA)_2 \in \mathbb{R}^{d\times ds}, S_2\in \mathbb{R}^{ds \times s_2}\\
= & ~ \sum_{y=1}^{d} \sum_{z=1}^s  ((SA)_2)_{i,(y-1)s+z} (S_2)_{(y-1)s+z,j} \\
= & ~ \sum_{y=1}^{d} \sum_{z=1}^s  (SA)_{z,i,y} (S_2)_{(y-1)s+z,j} & \text{~by~unflattening} \\
= & ~ \sum_{y=1}^{d} \sum_{z=1}^s  \left( \sum_{x=1}^n S_{z,x} A_{x,i,y} \right)  (S_2)_{(y-1)s+z,j} \\
= & ~ \sum_{y=1}^{d} \sum_{z=1}^s \sum_{x=1}^n    S_{z,x} \cdot A_{x,i,y} \cdot  (S_2)_{(y-1)s+z,j}.
\end{align*}
For each nonzero entry $A_{x,i,y}$, there is only one $z$ such that $S_{z,x}$ is nonzero. Thus there is only one $j$ such that $(S_2)_{(y-1)s+z,j}$ is nonzero. It means that $A_{x,i,y}$ can only affect one entry of $( (SA)_2 S_2 )_{i,j}$.
Thus, $(SA)_2 S_2$ can be computed in $O(\nnz(A))$ time. Similarly, we can compute $(SA)_3S_3$ in $O(\nnz(A))$ time.

Proof of Part (\RN{2}). Note that $T_2 (SA)_2 S_2$ has size $t_2 \times s_2$. Thus for each $i\in [t_2], j\in [s_2]$, we have,
\begin{align*}
( T_2  (SA)_2 S_2  )_{i,j} = & ~ \sum_{x=1}^{d} \sum_{y'=1}^{ds} (T_2)_{i,x} ((SA)_2)_{x,y'} (S_2)_{y',j} & \text{~by~} (SA)_2 \in \mathbb{R}^{d \times ds} \\
= & ~ \sum_{x=1}^{d} \sum_{y=1}^{d} \sum_{z=1}^{s} (T_2)_{i,x} ((SA)_2)_{x,(y-1)s+z}  (S_2)_{(y-1)s+z,j} \\
= & ~ \sum_{x=1}^{d} \sum_{y=1}^{d} \sum_{z=1}^{s} (T_2)_{i,x} (SA)_{z,x,y} (S_2)_{(y-1)s+z,j} & \text{~by~unflattening} \\
= & ~ \sum_{x=1}^{d} \sum_{y=1}^{d} \sum_{z=1}^{s} (T_2)_{i,x} \left(\sum_{w=1}^n S_{z,w} A_{w,x,y}\right) (S_2)_{(y-1)s+z,j}  \\
= & ~ \sum_{x=1}^{d} \sum_{y=1}^{d} \sum_{z=1}^{s} \sum_{w=1}^n (T_2)_{i,x} \cdot S_{z,w} \cdot A_{w,x,y} \cdot (S_2)_{(y-1)s+z,j}.
\end{align*}
For each nonzero entry $A_{w,x,y}$, there is only one $z$ such that $S_{z,w}$ is nonzero. There is only one $i$ such that $(T_2)_{i,x}$ is nonzero. Since there is only one $z$ to make $S_{z,w}$ nonzero, there is only one $j$, such that $(S_2)_{(y-1)s+z,j}$ is  nonzero. Thus, $T_2  (SA)_2 S_2 $ can be computed in $O(\nnz(A))$ time. Similarly, we can compute $T_3 (SA)_3 S_3$ in $O(\nnz(A))$ time.
\end{proof}

\subsubsection{Algorithms}\label{sec:lv_l122_algorithms}
We first give a ``warm-up''  algorithm in Theorem~\ref{thm:lv_l122_polyk_approx_algorithm} by using a sampling and rescaling matrix. Then we improve the running time to be polynomial in all the parameters by using an oblivious sketch, and thus we obtain Theorem~\ref{thm:lv_l122_polyklogn_approx_algorithm}.

\begin{algorithm}[h]\caption{$\ell_1$-Frobenius($\ell_1$-$\ell_2$-$\ell_2$) Low-rank Approximation Algorithm, $\poly(k)$-approximation}\label{alg:lv_l122_polyk_approx_algorithm}
\begin{algorithmic}[1]
\Procedure{\textsc{L122TensorLowRankApprox}}{$A,n,k$} \Comment{Theorem \ref{thm:lv_l122_polyk_approx_algorithm}}
\State $\epsilon\leftarrow 0.1$.
\State $s\leftarrow \poly(k,1/\epsilon)$.
\State Guess a sampling and rescaling matrix $S\in \mathbb{R}^{s\times n}$.
\State $ s_2 \leftarrow s_3 \leftarrow O(k/\epsilon)$.
\State $r \leftarrow s_2 s_3$.
\State Choose sketching matrices $S_2\in \mathbb{R}^{ ns \times s_2 }$, $S_3 \in \mathbb{R}^{ns \times s_3}$.
\State Compute $(SA)_2 S_2, (SA)_3 S_3$.
\State Form $\wt{V}\in \mathbb{R}^{n\times r}$ by repeating $(SA)_2 S_2$ $s_3$ times according to Equation~\eqref{eq:lv_l122_polyk_wtV}.
\State Form $\wt{W}\in \mathbb{R}^{n\times r}$ by repeating $(SA)_3 S_3$ $s_2$ times according to Equation~\eqref{eq:lv_l122_polyk_wtW}.
\State Form objective function $\min_{U\in \mathbb{R}^{n\times r}} \| U \cdot (\wt{V}^\top \odot \wt{W}^\top ) - A_1 \|_F$.
\State Use a linear regression solver to find a solution $\wt{U}$.
\State Take the best solution found over all guesses.
\State \Return $\wt{U}$, $\wt{V}$, $\wt{W}$.
\EndProcedure
\end{algorithmic}
\end{algorithm}

\begin{theorem}\label{thm:lv_l122_polyk_approx_algorithm}
Given a $3$rd order tensor $A\in \mathbb{R}^{n\times n \times n}$, for any $k\geq 1$, let $r=O(k^2)$. There exists an algorithm which takes $n^{\poly(k)} $ time and outputs three matrices $U,V,W\in \mathbb{R}^{n\times r}$ such that
\begin{align*}
\|  U \otimes V \otimes W - A  \|_v \leq \poly(k) \underset{\rank-k~A'}{\min} \| A' - A \|_v,
\end{align*}
holds with probability at least $9/10$.
\end{theorem}

\begin{proof}
We define $\OPT$ as follows,
\begin{align*}
\OPT = \underset{U,V,W \in \mathbb{R}^{n\times k} }{\min}\left\|  U \otimes V\otimes W - A \right\|_v =  \underset{U,V,W \in \mathbb{R}^{n\times k} }{\min}\left\| \sum_{i=1}^k U_i \otimes V_i \otimes W_i - A \right\|_v.
\end{align*}
Let $A_1\in \mathbb{R}^{n\times n^2}$ denote the matrix obtained by flattening tensor $A$ along the $1$st dimension. Let $U^*\in \mathbb{R}^{n\times k}$ denote the optimal solution. We fix $U^*\in \mathbb{R}^{n\times k}$, and consider this objective function,
\begin{align}\label{eq:lv_l122_polyk_fix_Ustar_objective_function}
\min_{V,W \in \mathbb{R}^{n\times k} } \left\|  U^* \otimes V \otimes W - A \right\|_v \equiv \min_{V,W \in \mathbb{R}^{n\times k} } \left\|  U^* \cdot ( V^\top \odot W^\top )- A_1 \right\|_v,
\end{align}
which has cost at most $\OPT$, and where $V^\top \odot W^\top\in \mathbb{R}^{k\times n^2}$ denotes the matrix for which the $i$-th row is a vectorization of $V_i \otimes W_i,\forall i\in [k]$. (Note that $V_i\in \mathbb{R}^n$ is the $i$-th column of matrix $V\in \mathbb{R}^{n\times k}$).
Choose a sampling and rescaling diagonal matrix $S\in \mathbb{R}^{n\times n}$ according to $U^*$, which has $s=\poly(k)$ non-zero entries. Using $S$ to sketch on the left of the objective function when $U^*$ is fixed (Equation~\eqref{eq:lv_l122_polyk_fix_Ustar_objective_function}), we obtain a smaller problem,
\begin{align}\label{eq:lv_l122_polyk_SUVW_minus_SA_lv}
\min_{V,W \in \mathbb{R}^{n\times k} } \left\| (SU^*) \otimes V \otimes W - SA \right\|_v \equiv \min_{V,W \in \mathbb{R}^{n\times k} } \left\|  S U^* \cdot ( V^\top \odot W^\top )- S A_1 \right\|_v.
\end{align}
Let $V',W'$ denote the optimal solution to the above problem, i.e.,
\begin{align*}
V', W' = \underset{V,W\in \mathbb{R}^{n\times k} }{\arg\min} \left\| (SU^*) \otimes V \otimes W -  SA \right\|_v.
\end{align*}
Then using properties (no dilation Lemma~\ref{lem:lv_no_dilation} and no contraction Lemma~\ref{lem:lv_no_contraction}) of $S$, we have
\begin{align*}
\left\|  U^* \otimes V' \otimes W' - A \right\|_v \leq \alpha \OPT.
\end{align*}
where $\alpha$ is an approximation ratio determined by $S$.

By definition of $\| \cdot \|_v$ and $\| \cdot \|_2 \leq \|\cdot \|_1 \leq \sqrt{\text{dim}} \| \cdot \|_2$, we can rewrite Equation~\eqref{eq:lv_l122_polyk_SUVW_minus_SA_lv} in the following way,
\begin{align}\label{eq:lv_l122_polyk_lv_relax_f}
& ~\left\| (SU^*) \otimes V \otimes W - SA \right\|_v \notag \\
= & ~ \sum_{i=1}^{s} \left(  \sum_{j=1}^n \sum_{l=1}^n \left( \left( ( SU^*) \otimes  V \otimes W \right)_{i,j,l} - (SA)_{i,j,l} \right)^2  \right)^{\frac{1}{2}} \notag  \\
\leq & ~ \sqrt{s} \left( \sum_{i=1}^s \sum_{j=1}^n \sum_{l=1}^n \left( \left( ( SU^*) \otimes  V \otimes W \right)_{i,j,l} -  (SA)_{i,j,l} \right)^2 \right)^\frac{1}{2} \notag \\
= & ~ \sqrt{s} \left\|  (SU^*)\otimes V \otimes W - SA \right\|_F.
\end{align}
Given the above properties of $S$ and Equation~\eqref{eq:lv_l122_polyk_lv_relax_f}, for any $\beta\geq 1$, let $V'',W''$ denote a $\beta$-approximate solution of $ \underset{V,W\in \mathbb{R}^{n\times k} }{\min} \left\| (SU^*) \otimes V \otimes W -S A  \right\|_F$, i.e.,

\begin{align}\label{eq:lv_l122_polyk_beta_approximation}
\left\|  (SU^*) \otimes V''\otimes W'' - SA \right\|_F \leq \beta \cdot \underset{V,W\in \mathbb{R}^{n\times k} }{\min} \left\| (SU^*) \otimes V \otimes W -S A  \right\|_F.
\end{align}
Then,
\begin{align}\label{eq:lv_l122_polyk_sqrts_alpha_beta_approximation}
\left\|  U^* \otimes V'' \otimes W'' - A \right\|_v \leq \sqrt{s}\alpha \beta \cdot \OPT.
\end{align}
In the next few paragraphs we will focus on solving Equation~\eqref{eq:lv_l122_polyk_beta_approximation}. We start by fixing $W^*\in \mathbb{R}^{n\times k}$ to be the optimal solution of
\begin{align*}
\min_{V,W\in \mathbb{R}^{n\times k} } \left\|  (SU^*) \otimes V \otimes W - SA \right\|_F.
\end{align*}
 We use $(SA)_2 \in \mathbb{R}^{n\times ns}$ to denote the matrix obtained by flattening the tensor $SA \in \mathbb{R}^{s\times n \times n}$ along the second direction. We use  $Z_2 = (SU^*)^\top \odot (W^*)^\top \in \mathbb{R}^{k\times ns}$ to denote the matrix where the $i$-th row is the vectorization of $(SU^*)_i \otimes W_i^*$. We can consider the following objective function,
\begin{align*}
\min_{V\in \mathbb{R}^{n\times k} } \| V Z_2  - (SA)_2 \|_F. 
\end{align*}
Choosing a sketching matrix $S_2 \in \mathbb{R}^{ns \times s_2}$ with $s_2=O(k/\epsilon)$ gives a smaller problem,
\begin{align*}
\min_{V\in \mathbb{R}^{n\times k} } \| V Z_2 S_2 - (SA)_2 S_2 \|_F.
\end{align*}
Letting $\wh{V} = (SA)_2 S_2 (Z_2 S_2)^\dagger \in \mathbb{R}^{n\times k}$, then
\begin{align}\label{eq:lv_l122_polyk_whVZ2_minus_SA2}
 \| \wh{V} Z_2  - (SA)_2  \|_F \leq & ~ (1+\epsilon) \min_{V\in \mathbb{R}^{n\times k} } \| V Z_2  - (SA)_2 \|_F \notag \\
= &~(1+\epsilon) \min_{V\in \mathbb{R}^{n\times k} } \| V ( (SU^*)^\top \odot (W^*)^\top ) - (SA)_2 \|_F \notag \\
= &~(1+\epsilon) \min_{V\in \mathbb{R}^{n\times k} } \|  (SU^*) \otimes V \otimes W^*  - SA \|_F & \text{~by~unflattening} \notag \\
= &~(1+\epsilon) \min_{V,W\in \mathbb{R}^{n\times k} } \|  (SU^*) \otimes V \otimes W  - SA \|_F. & \text{~by~definition~of~}W^*
\end{align}

We define $D_2\in \mathbb{R}^{n^2 \times n^2}$ to be a diagonal matrix obrained by copying the $n\times n$ identity matrix $s$ times on $n$ diagonal blocks of $D_2$. Then it has $ns$ nonzero entries. Thus, $D_2$ also can be thought of as a matrix that has size $n^2 \times ns$.

We can think of $(SA)_2 S_2 \in \mathbb{R}^{n\times s_2} $ as follows,
\begin{align*}
(SA)_2 S_2 = & ~ (A(S,I,I))_2 S_2 \\
= & ~ \underbrace{A_2}_{n\times n^2} \cdot \underbrace{ D_2}_{n^2 \times n^2} \cdot \underbrace{S_2}_{ns \times s_2} \text{~by~}D_2\text{~can~be~thought~of~as~having~size~}n^2\times ns \\
= & ~ A_2 \cdot
\begin{bmatrix}
c_{2,1} I_n & & & \\
& c_{2,2} I_n & & \\
 &  & \ddots &  \\
 &  &  & c_{2,n} I_n
\end{bmatrix} \cdot S_2
\end{align*}
where $I_n$ is an $n\times n$ identity matrix, $c_{2,i}\geq 0$ for each $i\in [n]$, and the number of nonzero $c_{2,i}$ is $s$.

 For the last step, we fix $SU^*$ and $\wh{V}$. We use $(SA)_3 \in \mathbb{R}^{n\times ns}$ to denote the matrix obtained by flattening the tensor $SA\in \mathbb{R}^{s\times n\times n}$ along the third direction. We use $Z_3 = (SU^*)^\top \odot \wh{V}^\top \in \mathbb{R}^{k\times ns}$ to denote the matrix where the $i$-th row is the vectorization of $(S U^*)_i \otimes \wh{V}_i$. We can consider the following objective function,
\begin{align*}
\min_{W\in \mathbb{R}^{n\times k}} \| W Z_3 - (SA)_3 \|_F. 
\end{align*}
Choosing a sketching matrix $S_3\in \mathbb{R}^{ns \times s_3}$ with $s_3=O(k/\epsilon)$ gives a smaller problem,
\begin{align*}
\min_{W\in \mathbb{R}^{n\times k} } \| W Z_3 S_3 - (SA)_3 S_3 \|_F.
\end{align*}
Let $\wh{W} = (SA)_3 S_3 (Z_3 S_3)^\dagger \in \mathbb{R}^{n\times k}$. Then
\begin{align*}
\| \wh{W} Z_3  - (SA)_3  \|_F \leq & ~ (1+\epsilon) \min_{W\in \mathbb{R}^{n\times k}} \| W Z_3  - (SA)_3 \|_F & \text{~by~property~of~} S_3 \\
= & ~ (1+\epsilon) \min_{W\in \mathbb{R}^{n\times k}} \| W ((SU^*)^\top \odot \wh{V}^\top)  - (SA)_3 \|_F &\text{~by~definition~}Z_3 \\
=& ~ (1+\epsilon) \min_{W\in \mathbb{R}^{n\times k}} \| (SU^*) \otimes \wh{V} \otimes W  - SA \|_F &\text{~by~unflattening~} \\
\leq & ~ (1+\epsilon)^2 \left\|  (SU^*) \otimes V \otimes W - SA \right\|_F. &\text{~by~Equation~\eqref{eq:lv_l122_polyk_whVZ2_minus_SA2}}
\end{align*}

We define $D_3\in \mathbb{R}^{n^2 \times n^2}$ to be a diagonal matrix formed by copying the $n\times n$ identity matrix $s$ times on $n$ diagonal blocks of $D_3$. Then it has $ns$ nonzero entries. Thus, $D_3$ also can be thought of as a matrix that has size $n^2 \times ns$ and $D_3$ is uniquely determined by $S$.

Similarly as to the $2$nd dimension, for the $3$rd dimension, we can think of $(SA)_3 S_3$ as follows,
\begin{align*}
(SA)_3 S_3 = & ~ ( A(S,I,I) )_3  S_3 \\
= & ~ \underbrace{A_3}_{n\times n^2} \cdot \underbrace{D_3}_{n^2\times n^2} \cdot \underbrace{S_3}_{ns \times s_3} & \text{~by~}D_3\text{~can~be~thought~of~as~having~size~}n^2\times ns\\
 = & ~
A_3 \cdot
\begin{bmatrix}
c_{3,1} I_n & & & \\
& c_{3,2} I_n & & \\
 &  & \ddots &  \\
 &  &  & c_{3,n} I_n
\end{bmatrix}
 \cdot S_3
\end{align*}
where $I_n$ is an $n\times n$ identity matrix, $c_{3,i}\geq 0$ for each $i\in [n]$ and the number of nonzero $c_{3,i}$ is $s$.

Overall, we have proved that,
\begin{align}\label{eq:lv_l122_polyk_min_X2_X3}
\min_{X_2,X_3} \|  (SU^*) \otimes (A_2 D_2 S_2 X_2) \otimes (  A_3 D_3 S_3 X_3 ) - S A \|_F \leq (1+\epsilon)^2 \left\|  (SU^*) \otimes V \otimes W - SA \right\|_F,
\end{align}
where diagonal matrix $D_2\in \mathbb{R}^{n^2 \times n^2}$ (with $ns$ nonzero entries) and $D_3\in \mathbb{R}^{n^2 \times n^2}$ (with $ns$ nonzero entries) are uniquely determined by diagonal matrix $S\in \mathbb{R}^{n\times n}$ ($s$ nonzero entries). Let $X'_2$ and $X_3'$ denote the optimal solution to the above problem (Equation~\eqref{eq:lv_l122_polyk_min_X2_X3}). Let $V''=(A_2 D_2 S_2 X_2')\in \mathbb{R}^{n\times k}$ and $W''=(  A_3 D_3 S_3 X_3' ) \in \mathbb{R}^{n\times k}$. Then we have
\begin{align*}
\left\|  U^* \otimes V'' \otimes W'' - A \right\|_v \leq \sqrt{s} \alpha \beta \OPT.
\end{align*}
We construct matrix $\wt{V}\in \mathbb{R}^{n\times s_2 s_3}$ by copying matrix $(SA)_2 S_2\in \mathbb{R}^{n\times s_2}$ $s_3$ times,
\begin{align}\label{eq:lv_l122_polyk_wtV}
\wt{V} =
\begin{bmatrix}
(SA)_2 S_2 & (SA)_2 S_2 & \cdots & (SA)_2 S_2.
\end{bmatrix}
\end{align}
We construct matrix $\wt{W}\in \mathbb{R}^{n\times s_2 s_3}$ by copying the $i$-th column of matrix $(SA)_3 S_3\in \mathbb{R}^{n\times s_3}$ into $(i-1)s_2+1, \cdots, i s_2$ columns of $\wt{W}$,
\begin{align}\label{eq:lv_l122_polyk_wtW}{\small
\wt{W} =
\begin{bmatrix}
( (SA)_3 S_3)_1  \cdots  ( (SA)_3 S_3)_1 &  ( (SA)_3 S_3)_2  \cdots ( (SA)_3 S_3)_2 & \cdots & ( (SA)_3 S_3)_{s_3}  \cdots ( (SA)_3 S_3)_{s_3}.
\end{bmatrix}}
\end{align}

Although we don't know $S$, we can guess all of the possibilities. For each possibility, we can find a solution $\wt{U}\in \mathbb{R}^{n\times s_2 s_3}$ to the following problem,
\begin{align*}
 & ~\min_{U\in \mathbb{R}^{n\times s_2 s_3} } \left\| \sum_{i=1}^{s_2} \sum_{j=1}^{s_3} U_{(i-1)s_3+j} \otimes ( (SA)_2 S_2 )_i \otimes ( (SA)_3 S_3 )_j - A \right\|_v \\
 = & ~ \min_{U\in \mathbb{R}^{n\times s_2 s_3} } \left\| \sum_{i=1}^{s_2} \sum_{j=1}^{s_3} U_{(i-1)s_3+j}  \cdot \vect(   ( (SA)_2 S_2 )_i \otimes ( (SA)_3 S_3 )_j ) - A_1 \right\|_v \\
  = & ~ \min_{U\in \mathbb{R}^{n\times s_2 s_3} } \left\| \sum_{i=1}^{s_2} \sum_{j=1}^{s_3} U_{(i-1)s_3+j}  \cdot (\wt{V}^\top \odot \wt{W}^\top )^{(i-1)s_3+j} - A_1 \right\|_v \\
    = & ~ \min_{U\in \mathbb{R}^{n\times s_2 s_3} } \left\| U  \cdot (\wt{V}^\top \odot \wt{W}^\top ) - A_1 \right\|_v \\
  = & ~ \min_{U\in \mathbb{R}^{n\times s_2 s_3} } \left\| U Z - A_1 \right\|_v \\
  = & ~ \min_{U \in \mathbb{R}^{n\times s_2 s_3} } \sum_{i=1}^{s_2 s_3} \|U^i Z - A_1^i \|_2,
\end{align*}
where the first step follows by flattening the tensor along the $1$st dimension, $U_{(i-1)s_3+j}$ denotes the $(i-1)s_3+j$-th column of $U\in \mathbb{R}^{n\times s_2 s_3}$, $A_1\in\mathbb{R}^{n\times n^2}$ denotes the matrix obtained by flattening tensor $A$ along the $1$st dimension, the second step follows since $\wt{V}^\top \odot \wt{W}^\top \in \mathbb{R}^{s_2s_3 \in n^2}$ is defined to be the matrix where the $(i-1)s_3 +j$-th row is vectorization of $ ( (SA)_2 S_2 )_i \otimes ( (SA)_3 S_3 )_j$, the fourth step follows by defining $Z$ to be $\wt{V}^\top \odot \wt{W}^\top $, and the last step follows by definition of $\| \cdot \|_v$ norm. Thus, we obtain a multiple regression problem and it can be solved directly by using \cite{cw13,nn13}.

Finally, we take the best $\wt{U}, \wt{V}, \wt{W}$ over all the guesses. The entire running time is dominated by the number of guesses, which is $n^{\poly(k)}$. This completes the proof.
\end{proof}

\begin{algorithm}[h]\caption{$\ell_1$-Frobenius($\ell_1$-$\ell_2$-$\ell_2$) Low-rank Approximation Algorithm, $\poly(k,\log n)$-approximation}\label{alg:lv_l122_polyklogn_approx_algorithm}
\begin{algorithmic}[1]
\Procedure{\textsc{L122TensorLowRankApprox}}{$A,n,k$} \Comment{Theorem \ref{thm:lv_l122_polyklogn_approx_algorithm}}
\State $\epsilon\leftarrow 0.1$.
\State $s\leftarrow \poly(k,\log n)$.
\State Choose $S\in \mathbb{R}^{s\times n}$ to be an \textsc{MSketch}. \Comment{Definition~\ref{def:lv_l122_m_sketches}}
\State $ s_2 \leftarrow s_3 \leftarrow O(k/\epsilon)$.
\State $ t_2 \leftarrow t_3 \leftarrow \poly(k/\epsilon)$.
\State $r \leftarrow s_2 s_3$.
\State Choose sketching matrices $S_2\in \mathbb{R}^{ ns \times s_2 }$, $S_3 \in \mathbb{R}^{ns \times s_3}$.
\State Choose sketching matrices $T_2\in \mathbb{R}^{ t_2 \times n}$, $T_3 \in \mathbb{R}^{t_3 \times n}$.
\State Compute $(SA)_2 S_2, (SA)_3 S_3$.
\State Compute $T_2 (SA)_2 S_2, T_3 (SA)_3 S_3$.
\State Form $\wt{V}\in \mathbb{R}^{n\times r}$ by repeating $(SA)_2 S_2$ $s_3$ times according to Equation~\eqref{eq:lv_l122_polyklogn_wtV}.
\State Form $\wt{W}\in \mathbb{R}^{n\times r}$ by repeating $(SA)_3 S_3$ $s_2$ times according to Equation~\eqref{eq:lv_l122_polyklogn_wtW}.
\State Form $\ov{V}\in \mathbb{R}^{t_2\times r}$ by repeating $T_2 (SA)_2 S_2$ $s_3$ times according to Equation~\eqref{eq:lv_l122_polyklogn_ovV}.
\State Form $\ov{W}\in \mathbb{R}^{t_3\times r}$ by repeating $T_3 (SA)_3 S_3$ $s_2$ times according to Equation~\eqref{eq:lv_l122_polyklogn_ovW}.
\State $C\leftarrow A(I,T_2,T_3)$.
\State Form objective function $\min_{U\in \mathbb{R}^{n\times r}} \| U \cdot (\ov{V}^\top \odot \ov{W}^\top ) - C_1 \|_F$.
\State Use linear regression solver to find a solution $\wt{U}$.
\State \Return $\wt{U}$, $\wt{V}$, $\wt{W}$.
\EndProcedure
\end{algorithmic}
\end{algorithm}

\begin{theorem}\label{thm:lv_l122_polyklogn_approx_algorithm}
Given a $3$rd order tensor $A\in \mathbb{R}^{n\times n \times n}$, for any $k\geq 1$, let $r=O(k^2)$. There exists an algorithm which takes $O(\nnz(A) ) + n \poly(k,\log n)$ time and outputs three matrices $U,V,W\in \mathbb{R}^{n\times r}$ such that
\begin{align*}
\|  U \otimes V \otimes W - A  \|_v \leq \poly(k,\log n) \underset{\rank-k~A'}{\min} \| A' - A \|_v
\end{align*}
holds with probability at least $9/10$.
\end{theorem}

\begin{proof}
We define $\OPT$ as follows,
\begin{align*}
\OPT = \underset{U,V,W \in \mathbb{R}^{n\times k} }{\min}\left\|  U \otimes V\otimes W - A \right\|_v =  \underset{U,V,W \in \mathbb{R}^{n\times k} }{\min}\left\| \sum_{i=1}^k U_i \otimes V_i \otimes W_i - A \right\|_v.
\end{align*}
Let $A_1\in \mathbb{R}^{n\times n^2}$ denote the matrix obtained by flattening tensor $A$ along the $1$st dimension. Let $U^*\in \mathbb{R}^{n\times k}$ denote the optimal solution. We fix $U^*\in \mathbb{R}^{n\times k}$, and consider the objective function,
\begin{align}\label{eq:lv_l122_polyklogn_fix_Ustar_objective_function}
\min_{V,W \in \mathbb{R}^{n\times k} } \left\|  U^* \otimes V \otimes W - A \right\|_v \equiv \min_{V,W \in \mathbb{R}^{n\times k} } \left\|  U^* \cdot ( V^\top \odot W^\top )- A_1 \right\|_v,
\end{align}
which has cost at most $\OPT$, and where $V^\top \odot W^\top\in \mathbb{R}^{k\times n^2}$ denotes the matrix for which the $i$-th row is a vectorization of $V_i \otimes W_i,\forall i\in [k]$. (Note that $V_i\in \mathbb{R}^n$ is the $i$-th column of matrix $V\in \mathbb{R}^{n\times k}$).
Choose an (oblivious) \textsc{MSketch} $S\in \mathbb{R}^{s\times n}$  with $s=\poly(k,\log n)$ according to Definition~\ref{def:lv_l122_m_sketches}. Using \textsc{MSketch} $S,w$ to sketch on the left of the objective function when $U^*$ is fixed (Equation~\eqref{eq:lv_l122_polyklogn_fix_Ustar_objective_function}), we obtain a smaller problem,
\begin{align}\label{eq:lv_l122_polyklogn_SUVW_minus_SA_lv}
\min_{V,W \in \mathbb{R}^{n\times k} } \left\| (SU^*) \otimes V \otimes W - SA \right\|_{v} \equiv & ~ \min_{V,W \in \mathbb{R}^{n\times k} } \left\|  S U^* \cdot ( V^\top \odot W^\top )- S A_1 \right\|_{v}.
\end{align}
Let $V',W'$ denote the optimal solution to the above problem, i.e.,
\begin{align*}
V', W' = \underset{V,W\in \mathbb{R}^{n\times k} }{\arg\min} \left\| (SU^*) \otimes V \otimes W -  SA \right\|_{v}.
\end{align*}
Then using properties (no dilation Lemma~\ref{lem:lv_no_dilation_oblivious} and no contraction Lemma~\ref{lem:lv_no_contraction_oblivious}) of $S$, we have
\begin{align*}
\left\|  U^* \otimes V' \otimes W' - A \right\|_v \leq \alpha \OPT.
\end{align*}
where $\alpha$ is an approximation ratio determined by $S$.

By definition of $\| \cdot \|_v$ and $\| \cdot \|_2 \leq \|\cdot \|_1 \leq \sqrt{\text{dim}} \| \cdot \|_2$, we can rewrite Equation~\eqref{eq:lv_l122_polyklogn_SUVW_minus_SA_lv} in the following way,
\begin{align}\label{eq:lv_l122_polyklogn_lv_relax_f}
& ~\left\| (SU^*) \otimes V \otimes W - SA \right\|_v \notag \\
= & ~ \sum_{i=1}^{s} \left(  \sum_{j=1}^n \sum_{l=1}^n \left( \left( ( SU^*) \otimes  V \otimes W \right)_{i,j,l} - (SA)_{i,j,l} \right)^2  \right)^{\frac{1}{2}} \notag  \\
\leq & ~ \sqrt{s} \left( \sum_{i=1}^s \sum_{j=1}^n \sum_{l=1}^n \left( \left( ( SU^*) \otimes  V \otimes W \right)_{i,j,l} -  (SA)_{i,j,l} \right)^2 \right)^\frac{1}{2} \notag \\
= & ~ \sqrt{s} \left\|  (SU^*)\otimes V \otimes W - SA \right\|_F
\end{align}
Using the properties of $S$ and Equation~\eqref{eq:lv_l122_polyklogn_lv_relax_f}, for any $\beta\geq 1$, let $V'',W''$ denote a $\beta$-approximation solution of $ \underset{V,W\in \mathbb{R}^{n\times k} }{\min} \left\| (SU^*) \otimes V \otimes W -S A  \right\|_F$, i.e.,

\begin{align}\label{eq:lv_l122_polyklogn_beta_approximation}
\left\|  (SU^*) \otimes V''\otimes W'' - SA \right\|_F \leq \beta \cdot \underset{V,W\in \mathbb{R}^{n\times k} }{\min} \left\| (SU^*) \otimes V \otimes W -S A  \right\|_F.
\end{align}
Then,
\begin{align}\label{eq:lv_l122_polyklogn_sqrts_alpha_beta_approximation}
\left\|  U^* \otimes V'' \otimes W'' - A \right\|_v \leq \sqrt{s}\alpha \beta \cdot \OPT.
\end{align}

Let $\wh{A}$ denote $SA$. Choose $S_i\in \mathbb{R}^{ns \times s_i}$ to be Gaussian matrix with $s_i=O(k/\epsilon)$, $\forall i \{2,3\}$. By a similar proof as in Theorem~\ref{thm:lv_l122_polyk_approx_algorithm}, we have if $X_2',X_3'$ is a $\beta$-approximate solution to
\begin{align*}
\min_{X_2,X_3} \| (SU^*) \otimes (\wh{A}_2 S_2 X_2) \otimes (\wh{A}_3 S_3 X_3)-  SA\|_F,
\end{align*}
then,
\begin{align*}
\|U^* \otimes (\wh{A}_2 S_2 X_2) \otimes (\wh{A}_3 S_3 X_3) - A\|_v \leq \sqrt{s}\alpha \beta.
\end{align*}

To reduce the size of the objective function from $\poly(n)$ to $\poly(k/\epsilon)$, we use perform an ``input sparsity reduction'' (in Lemma~\ref{lem:f_input_sparsity_reduction}). Note that, we do not need to use this idea to optimize the running time in Theorem~\ref{thm:lv_l122_polyk_approx_algorithm}. The running time of Theorem~\ref{thm:lv_l122_polyk_approx_algorithm} is dominated by guessing sampling and rescaling matrices. (That running time is $\gg \nnz(A)$.) Choose $T_i \in \mathbb{R}^{t_i \times n}$ to be a sparse subspace embedding matrix (CountSketch transform) with $t_i=\poly(k,1/\epsilon)$, $\forall i \in \{2,3\}$. Applying the proof of Lemma~\ref{lem:f_input_sparsity_reduction} here, we obtain,
if $X_2',X_3'$ is a $\beta$-approximate solution to
\begin{align*}
\min_{X_2,X_3} \| (SU^*) \otimes (T_2 (SA)_2 S_2 X_2) \otimes (T_3 (SA)_3 S_3 X_3)-  SA\|_F,
\end{align*}
then,
\begin{align}\label{eq:lv_l122_polyklogn_min_X2_X2_v}
\|U^* \otimes ( (SA)_2 S_2 X_2) \otimes ( (SA)_3 S_3 X_3) - A\|_v \leq \sqrt{s}\alpha \beta.
\end{align}

Similar to the bicriteria results in Section~\ref{sec:f_bicriteria_algorithm}, Equation~\eqref{eq:lv_l122_polyklogn_min_X2_X2_v} indicates that we can construct a bicriteria solution by using two matrices $ (SA)_2 S_2$ and  $(SA)_3 S_3$. The next question is how to obtain the final results $\wh{U},\wh{V},\wh{W}$. We first show how to obtain $\wh{U}$. Then we show to construct $\wh{V}$ and $\wh{W}$.

 To obtain $\wh{U}$, we need to solve a regression problem related to two matrices $\ov{V},\wh{W}$ and a tensor $A(I,T_2,T_3)$. We construct matrix $\ov{V}\in \mathbb{R}^{t_2\times s_2 s_3}$ by copying matrix $T_2 (SA)_2 S_2\in \mathbb{R}^{t_2 \times s_2}$ $s_3$ times,
\begin{align}\label{eq:lv_l122_polyklogn_ovV}
\ov{V} =
\begin{bmatrix}
T_2(SA)_2 S_2 & T_2(SA)_2 S_2 & \cdots & T_2 (SA)_2 S_2
\end{bmatrix}.
\end{align}
We construct matrix $\ov{W}\in \mathbb{R}^{t_3 \times s_2 s_3}$ by copying the $i$-th column of matrix $T_3(SA)_3 S_3\in \mathbb{R}^{t_3 \times s_3}$ into $(i-1)s_2+1, \cdots, i s_2$ columns of $\ov{W}$,
\begin{align}\label{eq:lv_l122_polyklogn_ovW}{
\ov{W} =
\begin{bmatrix}
F_1  \cdots  F_1 &  F_2  \cdots F_2 & \cdots & F_{s_3}  \cdots F_{s_3}
\end{bmatrix}},
\end{align}
where $F=T_3 (SA)_3 S_3$.

Thus, to obtain $\wt{U}\in \mathbb{R}^{s_2 s_3}$, we just need to use a linear regression solver to solve a smaller problem,
\begin{align*}
\min_{U \in \mathbb{R}^{s_2 s_3} } \| U \cdot ( \ov{V}^\top \odot \ov{W}^\top )- A(I,T_2,T_3) \|_F,
\end{align*}
which can be solved in $O(\nnz(A))+ n \poly(k,\log n)$ time. We will show how to obtain $\wt{V}$ and $\wt{W}$.

We construct matrix $\wt{V}\in \mathbb{R}^{n\times s_2 s_3}$ by copying matrix $(SA)_2 S_2\in \mathbb{R}^{n\times s_2}$ $s_3$ times,
\begin{align}\label{eq:lv_l122_polyklogn_wtV}
\wt{V} =
\begin{bmatrix}
(SA)_2 S_2 & (SA)_2 S_2 & \cdots & (SA)_2 S_2.
\end{bmatrix}
\end{align}
We construct matrix $\wt{W}\in \mathbb{R}^{n\times s_2 s_3}$ by copying the $i$-th column of matrix $(SA)_3 S_3\in \mathbb{R}^{n\times s_3}$ into $(i-1)s_2+1, \cdots, i s_2$ columns of $\wt{W}$,
\begin{align}\label{eq:lv_l122_polyklogn_wtW}{
\wt{W} =
\begin{bmatrix}
F_1  \cdots  F_1 &  F_2  \cdots F_2 & \cdots & F_{s_3}  \cdots F_{s_3}
\end{bmatrix}},
\end{align}
where $F=(SA)_3 S_3$.
\end{proof}

\subsection{$\ell_1$-$\ell_1$-$\ell_2$ norm}\label{sec:lu_l112}
Section~\ref{sec:lu_l112_definitions} presents some definitions and useful facts for the tensor $\ell_1$-$\ell_1$-$\ell_2$ norm. We provide some tools in Section~\ref{sec:lu_l112_projection_via_gaussians}. Section~\ref{sec:lu_l112_reduction} presents a key idea which shows we are able to reduce the original problem to a new problem under entry-wise $\ell_1$ norm. Section~\ref{sec:lu_l112_existence_results} presents several existence results. Finally, Section~\ref{sec:lu_l112_algorithm} introduces several algorithms with different tradeoffs.
\subsubsection{Definitions}\label{sec:lu_l112_definitions}
\begin{definition}(Tensor $u$-norm) For an $n\times n \times n$ tensor $A$, we define the $u$-norm of $A$, denoted $\| A \|_u$, to be
\begin{align*}
\left( \sum_{i=1}^n \sum_{j=1}^n M( \| A_{i,j,*}\|_2 ) \right)^{1/p},
\end{align*}
where $A_{i,j,*}$ is the $(i,j)$-th tube of $A$, and $p$ is a parameter associated with the function $M()$, which defines a nice $M$-Estimator.
\end{definition}

\begin{definition}
(Matrix $u$-norm) For an $n\times n$ matrix $A$, we define $u$-norm of $A$, denoted $\| A \|_u$, to be
\begin{align*}
\left( \sum_{i=1}^n  M( \| A_{i,*}\|_2 ) \right)^{1/p},
\end{align*}
where $A_{i,*}$ is the $i$-th row of $A$, and $p$ is a parameter associated with the function $M()$, which defines a nice $M$-Estimator.
\end{definition}

\begin{fact}
For $p=1$, for any two matrices $A$ and $B$, we have $\|A + B \|_u\leq \| A \|_u + \|B \|_u$. For any two tensors $A$ and $B$, we have $\| A + B \|_u \leq \| A \|_u + \| B \|_u$.
\end{fact}

\subsubsection{Projection via Gaussians}\label{sec:lu_l112_projection_via_gaussians}
\begin{definition}
Let $p\geq 1$. Let $\ell_p^{{\cal S}^{n-1}}$ be an infinite dimensional $\ell_p$ metric which consists of a coordinate for each vector $r$ in the unit sphere ${\cal S}^{n-1}$. Define function $f:{\cal S}^{n-1} \rightarrow \mathbb{R}$. The $\ell_1$-norm of any such $f$ is defined as follows:
\begin{align*}
\| f \|_1 =  \left( \int_{r \in {\cal S}^{n-1} } |f(r) |^p  \mathrm{d} r \right)^{1/p}.
\end{align*}
\end{definition}
\begin{claim}
Let $f_v(r) = \langle v, r\rangle$. There exists a universal constant $\alpha_p$ such that
\begin{align*}
\| f_v \|_p = \alpha_p \| v \|_2.
\end{align*}
\end{claim}
\begin{proof}
We have,
\begin{align*}
\| f_v \|_p = & ~ \left( \int_{r\in {\cal S}^{n-1} } | \langle v, r\rangle |^p \mathrm{d} r  \right)^{1/p} \\
= & ~ \left(  \int_{\theta \in {\cal S}^{n-1} } \| v\|_2^p \cdot |\cos\theta|^p \mathrm{d} \theta  \right)^{1/p} \\
= & ~ \| v\|_2 \left( \int_{\theta \in {\cal S}^{n-1} } |\cos\theta|^p \mathrm{d}\theta  \right)^{1/p} \\
= & ~ \alpha_p \| v \|_2.
\end{align*}
This completes the proof.
\end{proof}

\begin{lemma}
Let $G\in \mathbb{R}^{k\times n}$ denote i.i.d. random Gaussian matrices with rescaling. Then for any $v\in \mathbb{R}^n$, we have
\begin{align*}
\Pr[ (1-\epsilon) \| v\|_2 \leq \| G v \|_1 \leq (1+\epsilon) \| v \|_2  ] \geq 1 - 2^{-\Omega(k\epsilon^2)}.
\end{align*}
\end{lemma}
\begin{proof}
For each $i\in [k]$, we define $X_i = \langle v, g_i\rangle$, where $g_i\in \mathbb{R}^n$ is the $i$-th row of $G$. Then $X_i = \sum_{j=1}^n v_j g_{i,j}$ and $\E[|X_i|] = \alpha_p \| v\|_2$. Define $Y = \sum_{i=1}^k |X_i|$. We have $\E[Y] = k\alpha_1 \| v \|_2 = k \alpha_1$.

 We can show
\begin{align*}
\Pr[Y \geq (1+\epsilon) \alpha_1 k ] = & ~ \Pr[ e^{sY} \geq e^{s(1+\epsilon) \alpha_1 k} ] & \text{~for~all~} s>0 \\
\leq & ~ \E[e^{sY}] / e^{s(1+\epsilon) \alpha_1 k} & \text{~by~Markov's~inequality} \\
= & ~ e^{-s(1+\epsilon) \alpha_1 k} \cdot \E[ \prod_{i=1}^k e^{s|X_i|}] & \text{~by~} Y =\sum_{i=1}^k |X_i| \\
= & ~ e^{-s(1+\epsilon) \alpha_1 k} \cdot (\E[ e^{s|X_1|}])^k
\end{align*}
It remains to bound $\E[ e^{s|X_1|}]$. Since $X_1\sim {\cal N}(0,1)$, we have that $X_1$ has density function $e^{-t^2/2}$. Thus, we have,
\begin{align*}
\E[e^{s |X_1| }] =& ~ \frac{1}{\sqrt{2\pi}} \int_{-\infty}^{+\infty} e^{s|t|} \cdot e^{-t^2/2} \mathrm{d} t\\
= & ~ \frac{1}{\sqrt{2\pi}} \int_{-\infty}^{+\infty} e^{s^2/2} \cdot e^{- (|t|-s)^2/2} \mathrm{d} t \\
= & ~ e^{s^2/2} ( \mathrm{erf}( s/\sqrt{2}) + 1) \\
\leq & ~ e^{s^2/2} ( (1-\exp(-2s^2/\pi) )^{1/2} +1)& \text{~by~} 1-\exp(-4x^2/\pi) \geq \mathrm{erf}(x)^2 \\
\leq & ~ e^{s^2/2} ( \sqrt{2/\pi} s +1). & \text{~by~} 1-e^{-x} \leq x
\end{align*}
Thus, we have
\begin{align*}
\Pr[Y \geq (1+\epsilon) \alpha_1 k ]  \leq & ~e^{-s(1+\epsilon)k } e^{ks^2/2} (1+s \sqrt{2/\pi})^k \\
= & ~e^{-s(1+\epsilon) \alpha_1 k } e^{ks^2/2} e^{ k  \cdot \log (1+s \sqrt{2/\pi} ) } \\
\leq & ~e^{-s(1+\epsilon) \alpha_1 k  +ks^2/2 + k  \cdot s \sqrt{2/\pi} } \\
\leq & ~e^{-\Omega(k\epsilon^2)}. & \text{~by~} \alpha_1 \geq \sqrt{2/\pi} \text{~and~setting~} s  = \epsilon
\end{align*}
\end{proof}

\begin{lemma}\label{lem:lu_dvoretsky_for_all}
For any $\epsilon \in (0,1)$, let $k=O(n/\epsilon^2)$. Let $G\in \mathbb{R}^{k\times n}$ denote i.i.d. random Gaussian matrices with rescaling. Then for any $v\in \mathbb{R}^n$,  with probability at least $1-2^{-\Omega(n/\epsilon^2 )}$, we have : for all $v\in \mathbb{R}^n$,
\begin{align*}
(1-\epsilon) \| v\|_2 \leq \| G v \|_1 \leq (1+\epsilon) \| v \|_2.
\end{align*}
\end{lemma}
\begin{proof}
Let ${\cal S}$ denote $\{ y\in \mathbb{R}^n ~|~ \| y\|_2=1\}$. We construct a $\gamma$-net so that for all $y\in {\cal S}$, there exists a vector $w\in {\cal N}$ for which $\|y-w\|_2 \leq \gamma$. We set $\gamma = 1/2$.

For any unit vector $y$, we can write
\begin{align*}
y=y^0 + y^1+y^2 +\cdots,
\end{align*}
where $\| y^i \|_2 \leq 1/2^i$ and $y^i$ is a scalar multiple of a vector in ${\cal N}$. Thus, we have
\begin{align*}
\| G y \|_1 = & ~ \| G (y^0 + y^1 + y^2 +\cdots) \|_1 \\
\leq & ~ \sum_{i=0}^{\infty} \| G y^i \|_1 & \text{~by~triangle~inequality} \\
\leq & ~ \sum_{i=0}^{\infty} (1+\epsilon) \|y^i\|_2 \\
\leq & ~ \sum_{i=0}^{\infty} (1+\epsilon) \frac{1}{2^i} \\
\leq & ~ 1+ \Theta(\epsilon).
\end{align*}
Similarly, we can lower bound $\|Gy\|_1$ by $1-\Theta(\epsilon)$. By Lemma 2.2 in \cite{w14}, we know that for any $\gamma \in (0,1)$, there exists a $\gamma$-net ${\cal N}$ of ${\cal S}$ for which $|{\cal N}|\leq (1 +4/\gamma)^n$.
\end{proof}

\subsubsection{Reduction, projection to high dimension}\label{sec:lu_l112_reduction}
\begin{lemma}
Given a $3$rd order tensor $A\in \mathbb{R}^{n\times n\times n}$, let $S\in \mathbb{R}^{n\times s}$ denote a Gaussian matrix with $s= O(n/\epsilon^2)$ columns. With probability at least $1-2^{-\Omega(n/\epsilon^2)}$, for any $U,V,W\in \mathbb{R}^{n\times k}$, we have
\begin{align*}
(1-\epsilon) \left\|  U \otimes V \otimes W - A\right\|_u \leq \left\| (U \otimes V \otimes W )S -  AS \right\|_1 \leq (1+\epsilon) \left\|  U \otimes V \otimes W - A \right\|_u.
\end{align*}
\end{lemma}
\begin{proof}
By definition of the $\otimes $ product between matrices and $\cdot$ product between a tensor and a matrix, we have $(U\otimes V \otimes W) S = U \otimes V \otimes (SW) \in \mathbb{R}^{n\times n \times s}$. We use $A_{i,j,*} \in \mathbb{R}^n$ to denote the $(i,j)$-th tube (the column in the 3rd dimension) of tensor $A$.
We first prove the upper bound,
\begin{align*}
 \|  (U \otimes V \otimes W)S - AS \|_1
= & ~ \sum_{i=1}^n \sum_{j=1}^n  \left\| \left( (   U \otimes V \otimes W )_{i,j,*} - A_{i,j,*} \right) S \right\|_1 \\
\leq & ~ \sum_{i=1}^n \sum_{j=1}^n  (1+\epsilon) \left\| (   U \otimes V \otimes W )_{i,j,*} - A_{i,j,*}  \right\|_2 \\
= & ~ (1+\epsilon) \left\| U \otimes V \otimes W - A \right\|_u,
\end{align*}
where the first step follows by definition of tensor $\| \cdot \|_u$ norm, the second step follows by Lemma~\ref{lem:lu_dvoretsky_for_all}, and the last step follows by tensor entry-wise $\ell_1$ norm. Similarly, we can prove the lower bound,
\begin{align*}
 \|  (U \otimes V \otimes W)S - AS \|_1  \geq & ~ \sum_{i=1}^n \sum_{j=1}^n  (1-\epsilon) \left\| (   U \otimes V \otimes W )_{i,j,*} - A_{i,j,*}  \right\|_2 \\
 = & ~ (1-\epsilon) \left\| U \otimes V \otimes W - A \right\|_u.
\end{align*}
 This completes the proof.
\end{proof}

\begin{corollary}\label{cor:lu_reduction_to_l1}
For any $\alpha \geq 1$, if $U',V',W'$ satisfy
\begin{align*}
\| (U'\otimes V'\otimes W' -A) S \|_1 \leq \gamma \min_{\rank-k~A_k}\| (A_k - A) S \|_1,
\end{align*}
then
\begin{align*}
\| U'\otimes V'\otimes W' -A \|_u \leq \gamma \frac{1+\epsilon}{1-\epsilon} \min_{\rank-k~A_k} \| A_k - A \|_u.
\end{align*}
\end{corollary}
\begin{proof}
Let $\wh{U}, \wh{V}, \wh{W}$ denote the optimal solution to $\min_{\rank-k~A_k}\| (A_k - A) S \|_1$. Let $U^*,V^*,W^*$ denote the optimal solution to $\min_{\rank-k~A_k} \| A_k - A \|_u$. Then,
\begin{align*}
\| U'\otimes V'\otimes W' -A \|_u \leq & ~ \frac{1}{1-\epsilon} \| ( U'\otimes V'\otimes W' -A ) S \|_1 \\
\leq & ~ \gamma \frac{1}{1-\epsilon} \| ( \wh{U} \otimes \wh{V} \otimes \wh{W} -A ) S \|_1 \\
\leq & ~ \gamma \frac{1}{1-\epsilon} \| ( U^* \otimes V^* \otimes W^* -A ) S \|_1 \\
\leq & ~ \gamma \frac{1+\epsilon}{1-\epsilon} \|  U^* \otimes V^* \otimes W^* -A \|_u,
\end{align*}
which completes the proof.
\end{proof}

\subsubsection{Existence results}\label{sec:lu_l112_existence_results}

\begin{theorem}[Existence results]\label{thm:lu_existence_results}
Given a $3$rd order tensor $A\in \mathbb{R}^{n\times n\times n}$ and a matrix $S\in \mathbb{R}^{n\times \ov{n}}$, let $\OPT$ denote $\min_{\rank-k~A_k\in \mathbb{R}^{n\times n\times n}} \| (A_k - A)S \|_1$, let $\wh{A} = AS\in \mathbb{R}^{n\times n\times \ov{n}}$. For any $k\geq 1$, there exist three matrices $S_1\in \mathbb{R}^{n\ov{n}\times s_1}$, $S_2\in \mathbb{R}^{n\ov{n} \times s_2}$, $S_3 \in \mathbb{R}^{n^2 \times s_3}$ such that
\begin{align*}
\min_{X_1\in \mathbb{R}^{s_1\times k}, X_2\in \mathbb{R}^{s_2\times k}, X_3\in \mathbb{R}^{s_3 \times k}} \left\|   (\wh{A}_1 S_1 X_1) \otimes (\wh{A}_2 S_2 X_2) \otimes ( \wh{A}_3 S_3 X_3) - \wh{A} \right\|_1 \leq \alpha \OPT,
\end{align*}
or equivalently,
\begin{align*}
\min_{X_1\in \mathbb{R}^{s_1\times k}, X_2\in \mathbb{R}^{s_2\times k}, X_3\in \mathbb{R}^{s_3 \times k}} \left\|   \left( (\wh{A}_1 S_1 X_1) \otimes (\wh{A}_2 S_2 X_2) \otimes ( A_3 S_3 X_3) - A \right) S \right\|_1 \leq \alpha \OPT,
\end{align*}
holds with probability $99/100$.

$\mathrm{(\RN{1})}$. Using a dense Cauchy transform,\\
$s_1=s_2=s_3=\wt{O}(k)$, $\alpha = \wt{O}(k^{1.5}) \log^3 n$. 

$\mathrm{(\RN{2})}$. Using a sparse Cauchy transform,\\
$s_1=s_2=s_3=\wt{O}(k^5)$, $\alpha = \wt{O}(k^{13.5}) \log^3 n$. 

$\mathrm{(\RN{3})}$. Guessing Lewis weights,\\
$s_1=s_2=s_3=\wt{O}(k)$, $\alpha = \wt{O}(k^{1.5})$. 
\end{theorem}
\begin{proof}
We use $\OPT$ to denote the optimal cost,
\begin{align*}
\OPT := \underset{\rank-k~A_k \in \mathbb{R}^{n\times n\times  n} }{\min} \| (A_k - A )S \|_1.
\end{align*}
 We fix $V^*\in \mathbb{R}^{n\times k}$ and $W^*\in \mathbb{R}^{n\times k}$ to be the optimal solution to
\begin{align*}
\min_{U,V,W} \| (U \otimes V\otimes W - A) S \|_1.
\end{align*}
 We define $Z_1\in \mathbb{R}^{k\times n\ov{n} }$ to be the matrix where the $i$-th row is the vectorization of $V_i^* \otimes (S W_i^*)$. We define tensor
\begin{align*}
\wh{A} = AS \in \mathbb{R}^{n\times n\times \ov{n} }.
\end{align*}
Then we also have $\wh{A}= A(I,I,S)$ according to the definition of the $\cdot$ product between a tensor and a matrix.

Let $\wh{A}_1\in \mathbb{R}^{n\times n\ov{n}}$ denote the matrix obtained by flattening tensor $\wh{A}$ along the first direction. We can consider the following optimization problem,
\begin{align*}
\min_{U \in \mathbb{R}^{n\times k} } \left\| U Z_1 - \wh{A}_1 \right\|_1.
\end{align*}
Choosing $S_1$ to be one of the following sketching matrices:\\
(\RN{1}) a dense Cauchy transform,\\
(\RN{2}) a sparse Cauchy transform,\\
(\RN{3}) a sampling and rescaling diagonal matrix according to Lewis weights.\\

Let $\alpha_{S_1}$ denote the approximation ratio produced by the sketching matrix $S_1$. We use $S_1\in \mathbb{R}^{n\ov{n} \times s_1}$ to sketch on right of the above problem, and obtain the problem:
\begin{align*}
\min_{U \in \mathbb{R}^{n\times k}} \| U Z_1 S_1 - \wh{A}_1 S_1 \|_1 = \min_{U\in \mathbb{R}^{n\times k}} \sum_{i=1}^n \| U^i Z_1 S_1 - (\wh{A}_1 S_1)^i \|_1,
\end{align*}
where $U^i$ denotes the $i$-th row of matrix $U\in \mathbb{R}^{n\times k}$ and $(\wh{A}_1S_1)^i$ denotes the $i$-th row of matrix $\wh{A}_1 S_1$. Instead of solving it under $\ell_1$-norm, we consider the $\ell_2$-norm relaxation,
\begin{align*}
\underset{U \in \mathbb{R}^{n\times k} }{\min} \| U Z_1 S_1 - \wh{A}_1 S_1 \|_F^2 = \underset{U\in \mathbb{R}^{n\times k} }{\min} \sum_{i=1}^n \| U^i Z_1 S_1 - ( \wh{A}_1 S_1 )^i \|_2^2.
\end{align*}
Let $\wh{U}\in \mathbb{R}^{n\times k}$ denote the optimal solution of the above optimization problem, so that $\wh{U} = \wh{A}_1 S_1 (Z_1 S_1)^\dagger$. We plug $\wh{U}$ into the objective function under the $\ell_1$-norm. By the property of sketching matrix $S_1\in \mathbb{R}^{n\ov{n} \times s_1}$, we have,
\begin{align*}
\| \wh{U} Z_1 - \wh{A}_1 \|_1 \leq \alpha_{S_1} \min_{U\in \mathbb{R}^{n\times k} } \| U Z_1 - \wh{A}_1 \|_1 =  \alpha_{S_1}  \OPT,
\end{align*}
which implies that,
\begin{align*}
\|  \wh{U} \otimes V^* \otimes (SW^*) - \wh{A} \|_1 = \|  (\wh{U} \otimes V^* \otimes W^*) S - \wh{A} \|_1 \leq \alpha_{S_1}  \OPT.
\end{align*}

In the second step, we fix $\wh{U}\in \mathbb{R}^{n\times k}$ and $W^* \in \mathbb{R}^{n\times k}$. Let $\wh{A}_2 \in \mathbb{R}^{n\times n \ov{n} }$ denote the matrix obtained by flattening tensor $\wh{A} \in \mathbb{R}^{n\times n \times \ov{n} }$ along the second direction. We choose a sketching matrix $S_2\in \mathbb{R}^{n \ov{n} \times s_2}$. Let $Z_2 = \wh{U}^\top \odot  (SW^*)^\top \in \mathbb{R}^{k\times n\ov{n}}$ denote the matrix where the $i$-th row is the vectorization of $\wh{U}_i \otimes (SW_i^*)$. Define $\wh{V} = \wh{A}_2 S_2 (Z_2 S_2)^\dagger$. By the properties of sketching matrix $S_2$, we have
\begin{align*}
\| \wh{V} Z_2 - \wh{A}_2 \|_1 \leq \alpha_{S_2} \alpha_{S_1}  \OPT,
\end{align*}

In the third step, we fix $\wh{U}\in \mathbb{R}^{n\times k}$ and $ \wh{V} \in \mathbb{R}^{n\times k}$. Let $\wh{A}_3 \in \mathbb{R}^{\ov{n} \times n^2 }$ denote the matrix obtained by flattening tensor $\wh{A} \in \mathbb{R}^{n\times n \times \ov{n} }$ along the third direction. We choose a sketching matrix $S_3\in \mathbb{R}^{n^2 \times s_3}$. Let $Z_3\in \mathbb{R}^{k\times n^2}$ denote the matrix where the $i$-th row is the vectorization of $\wh{U}_i \otimes \wh{V}_i$. Define $W'=\wh{A}_3 S_3 (Z_3 S_3)^\dagger \in \mathbb{R}^{\ov{n} \times k}$ and $\wh{W} = A_3 S_3 (Z_3 S_3)^\dagger \in \mathbb{R}^{n \times k}$. Then we have,
\begin{align*}
W' = & ~\wh{A}_3 S_3 (Z_3 S_3)^\dagger \\
 = & ~ ( A(I,I,S) )_3 S_3 (Z_3 S_3)^\dagger \\
 = & ~ (S^\top A_3 ) S_3 (Z_3 S_3)^\dagger \\
 = & ~ S^\top \wh{W}
\end{align*}
 By properties of sketching matrix $S_3$, we have
\begin{align*}
\| W' Z_3 - \wh{A}_3 \|_1 \leq \alpha_{S_3} \alpha_{S_2} \alpha_{S_1}  \OPT.
\end{align*}
Replacing $W'$ by $S^\top \wh{W}$, we obtain,
\begin{align*}
\| W' Z_3 - \wh{A}_3 \|_1 = \| S^\top \wh{W} Z_3 - \wh{A}_3 \|_1 = \| S^\top \wh{W} Z_3 - S^\top A_3 \|_1 = \| (\wh{U} \otimes \wh{V} \otimes \wh{W} - A) S \|_1.
\end{align*}
Thus, we have
\begin{align*}
\min_{X_1\in \mathbb{R}^{s_1\times k}, X_2\in \mathbb{R}^{s_2\times k}, X_3\in \mathbb{R}^{s_3 \times k}} \left\|   (\wh{A}_1 S_1 X_1) \otimes (\wh{A}_2 S_2 X_2) \otimes ( \wh{A}_3 S_3 X_3) - \wh{A} \right\|_1 \leq \alpha_{S_1} \alpha_{S_2} \alpha_{S_3} \OPT.
\end{align*}
\end{proof}

\subsubsection{Running time analysis}\label{sec:lv_l122_time}

\begin{fact}
Given tensor $A\in \mathbb{R}^{n\times n \times n}$ and a matrix $B\in \mathbb{R}^{n\times d}$ with $d=O(n)$, let $AB$ denote an $n\times n\times d$ size tensor, For each $i\in [3]$, let $(AB)_i $ denote a matrix obtained by flattening tensor $AB$ along the $i$-th dimension, then
\begin{align*}
(AB)_1\in \mathbb{R}^{n\times nd}, (AB)_2\in \mathbb{R}^{n\times nd}, (AB)_3 \in \mathbb{R}^{d \times n^2}.
\end{align*}
For each $i\in[3]$, let $S_i \in \mathbb{R}^{n d \times s_i}$ denote a sparse Cauchy transform, $T_i \in \mathbb{R}^{t_i\times n}$. 
Then we have,\\
$\mathrm{ (\RN{1}) }$ If $T_1$ denotes a sparse Cauchy transform or a sampling and rescaling matrix according to the Lewis weights, $T_1 (A B)_1 S_1$ can be computed in $O(\nnz(A)d)$ time. Otherwise, it can be computed in $O(\nnz(A)d+ns_1t_1)$.\\
$\mathrm{ (\RN{2}) }$ If $T_2$ denotes a sparse Cauchy transform or a sampling and rescaling matrix according to the Lewis weights, $T_2 (A B)_2 S_2$ can be computed in $O(\nnz(A)d)$ time. Otherwise, it can be computed in $O(\nnz(A)d+ns_2t_2)$.\\
$\mathrm{ (\RN{3}) }$ If $T_3$ denotes a sparse Cauchy transform or a sampling and rescaling matrix according to the Lewis weights, $T_3 (A B)_3 S_3$ can be computed in $O(\nnz(A)d)$ time. Otherwise, it can be computed in $O(\nnz(A)d+ds_3t_3)$.\\
\end{fact}
\begin{proof}
Part (\RN{1}). Note that $T_1 (A B)_1 S_1 \in \mathbb{R}^{t_1 \times s_1}$ and $(AB)_1 \in \mathbb{R}^{n\times nd}$,
for each $i\in [t_1], j\in [s_1]$,
\begin{align*}
(T_1 (A B)_1 S_1)_{i,j} & = ~ \sum_{x=1}^n \sum_{y'=1}^{nd} (T_1)_{i,x} ( (AB)_1 )_{x,y'} (S_1)_{y',j} \\
&  = ~ \sum_{x=1}^n \sum_{y=1}^{n} \sum_{z=1}^d (T_1)_{i,x} ( (AB)_1 )_{x,(y-1)d+z} (S_1)_{(y-1)d+z,j} \\
&  = ~ \sum_{x=1}^n \sum_{y=1}^{n} \sum_{z=1}^d (T_1)_{i,x} (AB)_{x,y,z} (S_1)_{(y-1)d+z,j} \\
&  = ~ \sum_{x=1}^n \sum_{y=1}^{n} \sum_{z=1}^d (T_1)_{i,x} \sum_{w=1}^n (A_{x,y,w} B_{w,z}) (S_1)_{(y-1)d+z,j} \\
&  = ~ \sum_{x=1}^n \sum_{y=1}^{n} (T_1)_{i,x} \sum_{w=1}^n   A_{x,y,w}  \sum_{z=1}^d B_{w,z} (S_1)_{(y-1)d+z,j}.
\end{align*}
We look at a non-zero entry $A_{x,y,w}$ and the entry $B_{w,z}$.
If $T_1$ denotes a sparse Cauchy transform or a sampling and rescaling matrix according to the Lewis weights, then there is at most one pair $(i,j)$ such that $(T_1)_{i,x}A_{x,y,w}B_{w,z}(S_1)_{(y-1)d+z,j}$ is non-zero. Therefore, computing $T_1 (A B)_1 S_1$ only needs $\nnz(A)d$ time. If $T_1$ is not in the above case, since $S_1$ is sparse, we can compute $(A B)_1 S_1$ in $\nnz(A)d$ time by a similar argument. Then, we can compute $T_1(A B)_1 S_1$ in $nt_1s_1$ time.

Part (\RN{2}). It is as the same as Part (\RN{1}).

Part (\RN{3}). Note that $T_3 (A B)_3 S_3 \in \mathbb{R}^{t_3 \times s_3}$ and $(AB)_3 \in \mathbb{R}^{d \times n^2}$. For each $i\in [t_3], j\in [s_3]$,
\begin{align*}
(T_3 (A B)_3 S_3)_{i,j} & = ~ \sum_{x=1}^d \sum_{y'=1}^{n^2} (T_3)_{i,x} ( (AB)_3 )_{x,y'} (S_3)_{y',j} \\
& = ~ \sum_{x=1}^d \sum_{y=1}^{n} \sum_{z=1}^n (T_3)_{i,x} ( (AB)_3 )_{x,(y-1)n+z} (S_3)_{(y-1)n+z,j} \\
& = ~ \sum_{x=1}^d \sum_{y=1}^{n} \sum_{z=1}^n (T_3)_{i,x} (AB)_{y,z,x} (S_3)_{(y-1)n+z,j} \\
& = ~ \sum_{x=1}^d \sum_{y=1}^{n} \sum_{z=1}^n (T_3)_{i,x} \sum_{w=1}^n A_{y,z,w} B_{w,x} (S_3)_{(y-1)n+z,j} \\
\end{align*}
Similar to Part (\RN{1}), if $T_1$ denotes a sparse Cauchy transform or a sampling and rescaling matrix according to the Lewis weights, computing $T_3 (A B)_3 S_3$ only needs $\nnz(A)d$ time. Otherwise, it needs $dt_3s_3+\nnz(A)d$ running time.

\end{proof}

\subsubsection{Algorithms}\label{sec:lu_l112_algorithm}

\begin{algorithm}[!h]\caption{$\ell_1$-$\ell_1$-$\ell_2$-Low Rank Approximation algorithm, input sparsity time}
\begin{algorithmic}[1]
\Procedure{\textsc{L112TensorLowRankApproxInputSparsity}}{$A,n,k$} \Comment{ Theorem~\ref{thm:lu_input_sparsity_time}}
\State $\ov{n} \leftarrow O(n)$.
\State $s_1 \leftarrow s_2 \leftarrow s_3 \leftarrow \wt{O}(k^5)$.
\State Choose $S\in \mathbb{R}^{n \times \ov{n}}$ to be a Gaussian matrix.
\State Choose $S_1 \in \mathbb{R}^{n\ov{n} \times s_1}$ to be a sparse Cauchy transform. \Comment{Part (\RN{2}) of Theorem~\ref{thm:lu_existence_results}}
\State Choose $S_2 \in \mathbb{R}^{n\ov{n} \times s_2}$ to be a sparse Cauchy transform.
\State Choose $S_3 \in \mathbb{R}^{n^2 \times s_3}$ to be a sparse Cauchy transform.
\State Form $\wh{A} = A S$.
\State Compute $\wh{A}_1 S_1$, $\wh{A}_2 S_2$, and $\wh{A}_3 S_3$
\State $Y_1,Y_2,Y_3,C \leftarrow$\textsc{L1PolyKSizeReduction}($\wh{A},\wh{A}_1 S_1, \wh{A}_2S_2, \wh{A}_3 S_3,n,n,\ov{n},s_1,s_2,s_3,k$) \Comment{Algorithm~\ref{alg:l1_polyk_size_reduction}}
\State Create $s_1k+s_2k+s_3k$ variables for each entry of $X_1,X_2,X_3$.
\State Form objective function $\| (Y_1 X_1) \otimes (Y_2 X_2) \otimes (Y_3 X_3) - C\|_F^2$.
\State Run polynomial system verifier.
\State \Return $A_1S_1X_1,A_2S_2X_2,A_3S_3 X_3$
\EndProcedure
\end{algorithmic}
\end{algorithm}

\begin{theorem}\label{thm:lu_input_sparsity_time}
Given a $3$rd order tensor $A\in \mathbb{R}^{n\times n\times n}$, for any $k\geq 1$, there exists an algorithm which takes $O(\nnz(A)n) + \wt{O}(n) \poly(k) + n 2^{\wt{O}(k^2)}$ time and outputs three matrices $U,V,W\in \mathbb{R}^{n\times k}$ such that,
\begin{align*}
 \| U\otimes V \otimes W - A \|_u \leq \poly(k,\log n) \min_{\rank-k~A'} \| A' - A \|_u,
\end{align*}
holds with probability at least $9/10$.
\end{theorem}
\begin{proof}
We first choose a Gaussian matrix $S\in \mathbb{R}^{n\times \ov{n}}$ with $\ov{n}=O(n)$. By applying Corollary~\ref{cor:lu_reduction_to_l1}, we can reduce the original problem to a ``generalized'' $\ell_1$ low rank approximation problem. Next, we use the existence results (Theorem~\ref{thm:lu_existence_results}) and polynomial in $k$ size reduction (Lemma~\ref{lem:l1_polyk_size_reduction}). At the end, we relax the $\ell_1$-norm objective function to a Frobenius norm objective function (Fact~\ref{fac:l1_relax_to_frobenius_norm}).
\end{proof}

\begin{algorithm}[!h]\caption{$\ell_1$-$\ell_1$-$\ell_2$-Low Rank Approximation Algorithm, $\wt{O}(k^{2/3})$}
\begin{algorithmic}[1]
\Procedure{\textsc{L112TensorLowRankApproxK}}{$A,n,k$} \Comment{ Theorem~\ref{thm:lu_best_approximation_ratio}}
\State $\ov{n} \leftarrow O(n)$.
\State $s_1 \leftarrow s_2 \leftarrow s_3 \leftarrow \wt{O}(k)$.
\State Choose $S\in \mathbb{R}^{n \times \ov{n}}$ to be a Gaussian matrix.
\State Guess a diagonal matrix $S_1 \in \mathbb{R}^{n\ov{n} \times s_1}$ with $s_1$ nonzero entries. \Comment{Part (\RN{3}) of Theorem~\ref{thm:lu_existence_results}}
\State Guess a diagonal matrix $S_2 \in \mathbb{R}^{n\ov{n} \times s_2}$ with $s_2$ nonzero entries.
\State Guess a diagonal matrix $S_3 \in \mathbb{R}^{n^2 \times s_3}$ with $s_3$ nonzero entries.
\State Form $\wh{A} = A S$.
\State Compute $\wh{A}_1 S_1$, $\wh{A}_2 S_2$, and $\wh{A}_3 S_3$
\State $Y_1,Y_2,Y_3,C \leftarrow$\textsc{L1PolyKSizeReduction}($\wh{A},\wh{A}_1 S_1, \wh{A}_2S_2, \wh{A}_3 S_3,n,n,\ov{n},s_1,s_2,s_3,k$) \Comment{Algorithm~\ref{alg:l1_polyk_size_reduction}}
\State Create $s_1k+s_2k+s_3k$ variables for each entry of $X_1,X_2,X_3$.
\State Form objective function $\| (Y_1 X_1) \otimes (Y_2 X_2) \otimes (Y_3 X_3) - C\|_1$.
\State Run polynomial system verifier.
\State \Return $A_1S_1X_1,A_2S_2X_2,A_3S_3 X_3$
\EndProcedure
\end{algorithmic}
\end{algorithm}

\begin{theorem}\label{thm:lu_best_approximation_ratio}
Given a $3$rd order tensor $A\in \mathbb{R}^{n\times n\times n}$, for any $k\geq 1$, there exists an algorithm which takes $n^{\wt{O}(k)} 2^{\wt{O}(k^3)}$ time and outputs three matrices $U,V,W\in \mathbb{R}^{n\times k}$ such that,
\begin{align*}
 \| U\otimes V \otimes W - A \|_u \leq O(k^{3/2}) \min_{\rank-k~A'} \| A' - A \|_u,
\end{align*}
holds with probability at least $9/10$.
\end{theorem}
\begin{proof}
We first choose a Gaussian matrix $S\in \mathbb{R}^{n\times \ov{n}}$ with $\ov{n}=O(n)$. By applying Corollary~\ref{cor:lu_reduction_to_l1}, we can reduce the original problem to a ``generalized'' $\ell_1$ low rank approximation problem. Next, we use the existence results (Theorem~\ref{thm:lu_existence_results}) and polynomial in $k$ size reduction (Lemma~\ref{lem:l1_polyk_size_reduction}). At the end, we solve an entry-wise $\ell_1$ norm objective function directly.
\end{proof}

\begin{algorithm}[!h]\caption{$\ell_1$-$\ell_1$-$\ell_2$-Low Rank Approximation Algorithm, Bicriteria Algorithm}
\begin{algorithmic}[1]
\Procedure{\textsc{L112TensorLowRankApproxBicteriteria}}{$A,n,k$} \Comment{ Theorem~\ref{thm:lu_l112_polyklogn_approx_algorithm}}
\State $\ov{n} \leftarrow O(n)$.
\State $s_2 \leftarrow s_3 \leftarrow \wt{O}(k^5)$.
\State $t_2 \leftarrow t_3 \leftarrow \wt{O}(k)$.
\State $r\leftarrow s_2 s_3$.
\State Choose $S\in \mathbb{R}^{n \times \ov{n}}$ to be a Gaussian matrix.
\State Form $\wh{A} = A S \in \mathbb{R}^{n\times n \times \ov{n}}$.
\State Choose a sketching matrix $S_2 \in \mathbb{R}^{n\ov{n} \times s_2}$ with $s_2$ nonzero entries (Sparse Cauchy transform), for each $i\in \{2,3\}$. \Comment{Part (\RN{2}) of Theorem~\ref{thm:lu_existence_results}}
\State Choose a sampling and rescaling diagonal matrix $D_i$ according to the Lewis weights of $\wh{A}_i S_i$ with $t_i$ nonzero entries, for each $i\in \{2,3\}$.
\State Form $\wh{V}\in \mathbb{R}^{n\times r}$ by setting the $(i,j)$-th column to be $(\wh{A}_2 S_2)_i$.
\State Form $\wh{W}\in \mathbb{R}^{n\times r}$ by setting the $(i,j)$-th column to be $(A_3 S_3)_j$. 
\State Form matrix $B\in \mathbb{R}^{r\times t_2t_3}$ by setting the $(i,j)$-th column to be the vectorization of $(T_2 \wh{A}_2 S_2)_i \otimes (T_3 \wh{A}_3 S_3)_j$.
\State Solve $\min_{U} \| U \cdot B - (\wh{A}(I,T_2,T_3))_1 \|_1$.
\State \Return $\wh{U},\wh{V},\wh{W}$
\EndProcedure
\end{algorithmic}
\end{algorithm}

\begin{theorem}\label{thm:lu_l112_polyklogn_approx_algorithm}
Given a $3$rd order tensor $A\in \mathbb{R}^{n\times n\times n}$, for any $k\geq 1$, let $r=\wt{O}(k^2)$. There is an algorithm which takes $O(\nnz(A)n) + \wt{O}(n) \poly(k) $ time and outputs three matrices $U, V, W\in \mathbb{R}^{n\times r}$ such that
\begin{align*}
\| U \otimes V \otimes W - A \|_u \leq \poly(\log n, k) \min_{\rank-k~A_k} \| A_k - A\|_u,
\end{align*}
holds with probability at least $9/10$.
\end{theorem}
\begin{proof}
We first choose a Gaussian matrix $S\in \mathbb{R}^{n\times \ov{n}}$ with $\ov{n}=O(n)$. By applying Corollary~\ref{cor:lu_reduction_to_l1}, we can reduce the original problem to a ``generalized'' $\ell_1$ low rank approximation problem. Next, we use the existence results (Theorem~\ref{thm:lu_existence_results}) and polynomial in $k$ size reduction (Lemma~\ref{lem:l1_polyk_size_reduction}). At the end, we solve an entry-wise $\ell_1$ norm objective function directly.
\end{proof}
\newpage
\section{Weighted Frobenius Norm for Arbitrary Tensors}\label{sec:w}
This section presents several tensor algorithms for the weighted case. For notational purposes, instead of using $U,V,W$ to denote the ground truth factorization, we use $U_1,U_2,U_3$ to denote the ground truth factorization. We use $A$ to denote the input tensor, and $W$ to denote the tensor of weights. Combining our new tensor techniques with existing weighted low rank approximation algorithms \cite{rsw16} allows us to obtain several interesting new results. We provide some necessary definitions and facts in Section~\ref{sec:w_def}. Section~\ref{sec:w_r_distinct_faces} provides an algorithm when $W$ has at most $r$ distinct faces in each dimension. Section~\ref{sec:w_r_distinct_3d_cols} studies relationships between $r$ distinct faces and $r$ distinct columns. Finally, we provides an algorithm with a similar running time but weaker assumption, where $W$ has at most $r$ distinct columns and $r$ distinct rows in Section~\ref{sec:w_r_distinct_2d_cols}. The result in Theorem~\ref{thm:w_r_distinct_columns_rows_tubes} is fairly similar to Theorem \ref{thm:w_r_distinct_2d_cols}, except for the running time. We only put a very detailed discussion in the statement of Theorem \ref{thm:w_r_distinct_2d_cols}. Note that Theorem~\ref{thm:w_r_distinct_columns_rows_tubes} also has other versions which are similar to the Frobnius norm $\rank$-$k$ algorithms described in Section~\ref{sec:intro}. For simplicity of presentation, we only present one clean and simple version (which assumes $A_k$ exists and has factor norms which are not too large). 

\subsection{Definitions and Facts}\label{sec:w_def}

For a matrix $A\in \mathbb{R}^{n\times m}$ and a weight matrix $W\in \mathbb{R}^{n\times m}$, we define $\| W \circ A\|_F$ as follows,
\begin{align*}
\| W \circ A\|_F = \left( \sum_{i=1}^n \sum_{j=1}^m W_{i,j}^2 A_{i,j}^2 \right)^{\frac{1}{2}}.
\end{align*}
For a tensor $A\in \mathbb{R}^{n\times n\times n}$ and a weight tensor $W\in \mathbb{R}^{n\times n\times n}$, we define $\| W \circ A \|_F$ as follows,
\begin{align*}
\| W \circ A\|_F = \left( \sum_{i=1}^n \sum_{j=1}^n \sum_{l=1}^n W_{i,j,l}^2 A_{i,j,l}^2 \right)^{\frac{1}{2}}.
\end{align*}
For three matrices $A\in \mathbb{R}^{n\times m}$, $U\in \mathbb{R}^{n\times k}$, $V\in \mathbb{R}^{k\times m}$ and a weight matrix $W$, from one perspective, we have
\begin{align*}
\| (  UV - A) \circ W \|_F^2 =  ~ \sum_{i=1}^n \| ( U^i V - A^i ) \circ W^i \|_2^2  ~ = \sum_{i=1}^n \| ( U^i V - A^i ) D_{W^i} \|_2^2,
\end{align*}
where $W^i$ denote the $i$-th row of matrix $W$, and $D_{W^i}\in \mathbb{R}^{m\times m}$ denotes a diagonal matrix where the $j$-th entry on diagonal is the $j$-th entry of vector $W^i$. From another perspective, we have
\begin{align*}
\| (  UV - A) \circ W \|_F^2 =  ~ \sum_{j=1}^m \| ( U V_j - A_j ) \circ W_j \|_2^2  ~ = \sum_{j=1}^m \| ( U V_j - A_j ) D_{W_j} \|_2^2,
\end{align*}
where $W_j$ denotes the $j$-th column of matrix $W$, and $D_{W_j}\in \mathbb{R}^{n\times n}$ denotes a diagonal matrix where the $i$-th entry on the diagonal is the $i$-th entry of vector $W_j$.

One of the key tools we use in this section is,
\begin{lemma}[Cramer's rule]\label{lem:w_cramer}
Let $R$ be an $n\times n$ invertible matrix. Then, for each $i\in [n], j\in [n]$,
\begin{align*}
(R^{-1})^j_i = \det(R_{\neg j}^{\neg i} ) /\det(R),
\end{align*}
where $R_{\neg j}^{\neg i}$ is the matrix $R$ with the $i$-th row and the $j$-th column removed.
\end{lemma}

\subsection{$r$ distinct faces in each dimension}\label{sec:w_r_distinct_faces}
Notice that in the matrix case, it is sufficient to assume that $\| A'\|_F$ is upper bounded \cite{rsw16}. Once we have that $\| A'\|_F$ is bounded, without loss of generality, we can assume that $U_1^*$ is an orthonormal basis\cite{cw15focs,rsw16}. If $U_1^*$ is not an orthonormal basis, then let $U_1'R$ denote a QR factorization of $U_1^*$, and then write $U'_2 = RU_2^*$. However, in the case of tensors we have to assume that each factor $\|U^*_i\|_F$ is upper bounded due to border rank issues (see, e.g., \cite{sl08}). 

\begin{algorithm}[h]\caption{Weighted Tensor Low-rank Approximation Algorithm when the Weighted Tensor has $r$ Distinct Faces in Each of the Three Dimensions.}
\begin{algorithmic}
\Procedure{\textsc{WeightedRDistinctFacesIn3Dimensions}}{$A,W,n,r,k,\epsilon$} \Comment{Theorem~\ref{thm:w_r_distinct_columns_rows_tubes}}
\For{$j=1\to 3$}
	\State $s_j \leftarrow O(k/\epsilon)$.
	\State Choose a sketching matrix $S_j \in \mathbb{R}^{n^2 \times s_j}$.
	\For{$i=1\to r$}
		\State Create $k\times s_1$ variables for matrix $P_{i,j}\in \mathbb{R}^{k\times s_j}$.
	\EndFor
	\For{$i=1\to n$}
		\State Write down $(\wh{U}_j)^i = A_i^j D_{W_1^j} S_j P_{j,i}^\top ( P_{j,i} P_{j,i}^\top )^{-1}$.
	\EndFor
\EndFor
\State Form $\| W\circ (  \wh{U}_1 \otimes \wh{U}_2 \otimes \wh{U}_3 - A) \|_F^2$.
\State Run polynomial system verifier.
\State \Return $U_1,U_2,U_3$
\EndProcedure
\end{algorithmic}
\end{algorithm}

\begin{theorem}\label{thm:w_r_distinct_columns_rows_tubes}
 Given a $3$rd order $n\times n\times n$ tensor $A$ and an $n\times n \times n$ tensor $W$ of weights with $r$ distinct faces in each of the three dimensions for which each entry can be written using $O(n^\delta)$ bits, for $\delta >0$, 
define $\OPT={\inf}_{\rank-k~A_k} \| W\circ (A_k - A )\|_F^2$. Let $k\geq 1$ be an integer and let $0 < \epsilon <1$. 

 If $\OPT>0$, and there exists a rank-$k$ $A_k=U_1^*\otimes U_2^* \otimes U_3^*$ tensor (with size $n\times n\times n$) such that $\| W\circ (A_k-A)\|_F^2 = \OPT$, and $\max_{i\in [3]}\|U_i^*\|_F \leq 2^{O(n^\delta)}$, then there exists an algorithm that takes $(\nnz(A)+\nnz(W)  + n 2^{\wt{O}(rk^2/\epsilon)} ) n^{O(\delta)}$ time in the unit cost $\RAM$ model with words of size $O(\log n)$ bits\footnote{The entries of $A$ and $W$ are assumed to fit in $n^{\delta}$ words.} and outputs three $n \times k$ matrices $U_1,U_2,U_3$ such that
\begin{align}\label{eq:w_thm_3cols_1}
\left\| W\circ \left( U_1 \otimes U_2 \otimes U_3 - A \right) \right\|_F^2 \leq (1+\epsilon) \OPT
\end{align}
holds with probability $9/10$. 





\end{theorem}
\begin{proof}
Note that $W$ has $r$ distinct columns, rows, and tubes. Hence, each of the matrices $W_1,W_2,W_3$ $ \in \mathbb{R}^{n\times n^2}$ has at most $r$ distinct columns, and at most $r$ distinct rows. Let $U_1^*,U_2^*,U_3^*\in \mathbb{R}^{n\times k}$ denote the matrices satisfying $\| W\circ (U_1^* \otimes U_2^* \otimes U_3^* -A )\|_F^2 =\OPT$. We fix $U_2^*$ and $U_3^*$, and consider a flattening of the tensor along the first dimension,
\begin{align*}
\min_{U_1 \in \mathbb{R}^{n\times k}} \| (U_1 Z_1 - A_1) \circ W_1 \|_F^2 = \OPT,
\end{align*}
where matrix $Z_1= U_2^{*\top} \odot U_3^{*\top}$ has size $k \times n^2$ and for each $i\in [k]$ the $i$-th row of $Z_1$ is $\vect( (U_2^*)_i \otimes (U_3^*)_i )$. For each $i\in [n]$, let $W_1^i$ denote the $i$-th row of $n\times n^2$ matrix $W_1 $. For each $i\in [n]$, let $D_{W^i_1}$ denote the diagonal matrix of size $n^2 \times n^2$, where each diagonal entry is from the vector $W^i_1 \in \mathbb{R}^{n^2}$. Without loss of generality, we can assume the first $r$ rows of $W_1$ are distinct. We can rewrite the objective function along the first dimension as a sum of multiple regression problems. For any $n\times k$ matrix $U_1$,
\begin{align}\label{eq:w_UZ1_minus_A1_circ_W}
\| (U_1 Z_1 - A_1) \circ W_1 \|_F^2 = \sum_{i=1}^n \| U_1^i Z_1 D_{W_1^i} - A_1^i D_{W_1^i} \|_2^2.
\end{align}
Based on the observation that $W_1$ has $r$ distinct rows, we can group the $n$ rows of $W^1$ into $r$ groups. We use $g_{1,1}, g_{1,2},\cdots, g_{1,r}$ to denote $r$ sets of indices such that, for each $i \in g_{1,j}$, $W_1^i = W_1^j$. Thus we can rewrite Equation~\eqref{eq:w_UZ1_minus_A1_circ_W},
\begin{align*}
\| (U_1 Z_1 - A_1) \circ W_1 \|_F^2 = & ~\sum_{i=1}^n \| U_1^i Z_1 D_{W_1^i} - A_1^i D_{W_1^i} \|_2^2 \\
= & ~\sum_{j=1}^r \sum_{i\in g_{1,j}} \| U_1^i Z_1 D_{W_1^i} - A_1^i D_{W_1^i} \|_2^2.
\end{align*}
We can sketch the objective function by choosing Gaussian matrices $S_1\in \mathbb{R}^{n^2 \times s_1}$ with $s_1 = O(k/\epsilon)$.
\begin{align*}
\sum_{i=1}^n \| U_1^i Z_1 D_{W_1^i} S_1 - A_1^i D_{W_1^i} S_1 \|_2^2.
\end{align*}
Let $\wh{U}_1$ denote the optimal solution of the sketch problem,
\begin{align*}
\wh{U}_1 = \underset{ U_1\in \mathbb{R}^{n\times k} }{\arg\min} \sum_{i=1}^n \| U_1^i Z_1 D_{W_1^i} S_1 - A_1^i D_{W_1^i} S_1 \|_2^2.
\end{align*}
By properties of $S_1$(\cite{rsw16}), plugging $\wh{U}\in \mathbb{R}^{n\times k}$ into the original problem, we obtain,
\begin{align*}
\sum_{i=1}^n \| \wh{U}_1^i Z_1 D_{W_1^i} - A_1^i D_{W_1^i} \|_2^2 \leq (1+\epsilon) \OPT.
\end{align*}
Note that $\wh{U}_1 \in \mathbb{R}^{n\times k}$ also has the following form. For each $i\in [n]$,
\begin{align*}
\wh{U}_1^i = & ~ A_1^i D_{W_1^i} S_1  (Z_1 D_{W_1^i} S_1)^\dagger \\
= & ~ A_1^i D_{W_1^i} S_1  (Z_1 D_{W_1^i} S_1)^\top ( (Z_1 D_{W_1^i} S_1) (Z_1 D_{W_1^i} S_1)^\top )^{-1}.
\end{align*}
Note that $W_1$ has $r$ distinct rows. Thus, we only have $r$ distinct $D_{W_1^i}$. This implies that there are $r$ distinct matrices $Z_1 D_{W_1^i} S_1 \in \mathbb{R}^{k\times s_1}$. Using the definition of $g_{1,j}$, for $j\in [r]$, for each $i\in g_{1,j} \subset [n]$, we have
\begin{align*}
\wh{U}^i_1 = & ~ A_1^i D_{W_1^i} S_1  (Z_1 D_{W_1^i} S_1)^\dagger \\
= & ~  A_1^i D_{W_1^j} S_1  (Z_1 D_{W_1^j} S_1)^\dagger & \text{~by~} W_1^i = W_1^j,
\end{align*}
which means we only need to write down $r$ different $Z_1 D_{W_1^j} S_1$.
 For each $k\times s_1$ matrix $Z_1 D_{W_1^j} S_1$, we create $k\times s_1$ variables to represent it. Thus, we need to create $rks_1$ variables to represent $r$ matrices,
\begin{align*}
\{ Z_1 D_{W_1^1} S_1, Z_1 D_{W_1^2} S_1, \cdots, Z_1 D_{W_1^r} S_1 \}.
\end{align*}
For simplicity, let $P_{1,i}\in \mathbb{R}^{k\times s_1}$ denote $Z_1 D_{W_1^i} S_1$. Then we can rewrite $\wh{U}^i\in \mathbb{R}^k$ as follows,
\begin{align*}
\wh{U}_1^i = & ~ A_1^i D_{W_1^i} S_1  P_{1,i}^\top ( P_{1,i} P_{1,i}^\top )^{-1}.
\end{align*}
If $P_{1,i} P_{1,i}^\top \in \mathbb{R}^{k\times k}$ has rank $k$, then we can use Cramer's rule (Lemma~\ref{lem:w_cramer}) to write down the inverse of $P_{1,i} P_{1,i}^\top$. However, vector $W_1^i$ could have many zero entries. Then the rank of $P_{1,i} P_{1,i}^\top$ can be smaller than $k$. There are two different ways to solve this issue. 

One way is by using the argument from \cite{rsw16}, which allows us to assume that $P_{1,i} P_{1,i}^\top \in \mathbb{R}^{k\times k}$ has rank $k$. 

The other way is straightforward: we can guess the rank. There are $k$ possibilities. Let $t_i \leq k$ denote the rank of $P_{1,i}$. Then we need to figure out a maximal linearly independent subset of rows of $P_{1,i}$. There are $2^{O(k)}$ possibilities. Next, we need to figure out a maximal linearly independent subset of columns of $P_{1,i}$. We can also guess all the possibilities, which is at most $2^{O(k)}$. 
Because we have $r$ different $P_{1,i}$, the total number of guesses we have is at most $2^{O(rk)}$.
Thus, we can write down $( P_{1,i} P_{1,i}^\top )^{-1}$ according to Cramer's rule.

After $\wh{U}_1$ is obtained, we will fix $\wh{U}_1$ and $U_3^*$ in the next round. We consider the flattening of the tensor along the second direction,
\begin{align*}
\min_{U_2 \in \mathbb{R}^{n\times k}} \| (U_2 Z_2 - A_2 ) \circ W_2 \|_F^2, 
\end{align*}
where $n\times n^2$ matrix $A_2$ is obtained by flattening tensor $A$ along the second dimension, $k\times n^2$ matrix $Z_2$ denotes $\wh{U}_1^\top \odot U_3^{*\top}$, and $n\times n^2$ matrix $W_2$ is obtained by flattening tensor $W$ along the second dimension. For each $i\in [n]$, let $W_2^i$ denote the $i$-th row of $n\times n^2$ matrix $W_2$. For each $i\in [n]$, let $D_{W_1^i}$ denote the diagonal matrix which has size $n^2 \times n^2$ and for which each entry is from vector $W_2^i\in \mathbb{R}^{n^2}$. Without loss of generality, we can assume the first $r$ rows of $W_2$ are distinct. We can rewrite the objective function along the second dimension as a sum of multiple regression problems. For any $n \times k$ matrix $U_2$,
\begin{align}
\| (U_2 Z_2 - A_2) \circ W_2 \|_F^2 = \sum_{i=1}^n \| U_2^i Z_2 D_{W_2^i} - A_2^i D_{W_2^i} \|_2^2.
\end{align}
Based on the observation that $W_2$ has $r$ distinct rows, we can group the $n$ rows of $W^2$ into $r$ groups. We use $g_{2,1}, g_{2,2}, \cdots, g_{2,r}$ to denote $r$ sets of indices such that, for each $i\in g_{2,j}$, $W_2^i = W_2^j$. Thus we obtain,
\begin{align*}
\| (U_2 Z_2 - A_2) \circ W_2 \|_F^2 = & ~ \sum_{i=1}^n \| U_2^i Z_2 D_{W_2^i} - A_2^i D_{W_2^i} \|_2^2 \\
= & ~ \sum_{j=1}^r \sum_{i \in g_{2,j}} \| U_2^i Z_2 D_{W_2^i} - A_2^i D_{W_2^i} \|_2^2. 
\end{align*}
We can sketch the objective function by choosing a Gaussian sketch $S_2\in \mathbb{R}^{n^2 \times s_2}$ with $s_2 = O( k/\epsilon )$. Let $\wh{U}_2$ denote the optimal solution to the sketch problem. Then $\wh{U}_2$ has the form, for each $i\in [n]$,
\begin{align*}
\wh{U}_2^i = A_2^i D_{W_2^i} S_2 ( Z_2 D_{W_2^i} S_2)^\dagger.
\end{align*}
Similarly as before, we only need to write down $r$ different matrices $Z_2 D_{W_2^i}S_1$, and for each of them, create $k\times s_2$ variables. Let $P_{2,i}\in \mathbb{R}^{k\times s_2}$ denote $Z_2 D_{W_2^i} S_2$. By our guessing argument, we can obtain $\wh{U}_2$.

In the last round, we fix $\wh{U}_1$ and $\wh{U}_2$. We then write down $\wh{U}_3$. Overall, by creating $l=O(rk^2/\epsilon)$ variables, we have rational polynomials $\wh{U}_1(x)$, $\wh{U}_2(x)$, $\wh{U}_3(x)$. Putting it all together, we can write this objective function,

\begin{align*}
\min_{x\in \mathbb{R}^l} & ~\| ( \wh{U}_1(x) \otimes \wh{U}_2(x) \otimes \wh{U}_3(x) - A ) \circ W\|_F^2. \\
\text{s.t.} & ~ h_{1,i}(x) \neq 0, \forall i \in [r]. \\
& ~ h_{2,i}(x) \neq 0, \forall i \in [r]. \\
& ~ h_{3,i}(x) \neq 0, \forall i \in [r].
\end{align*}
where $h_{1,i}(x)$ denotes the denominator polynomial related to a full rank sub-block of $P_{1,i}(x)$.
 By a perturbation argument in Section 4 in \cite{rsw16}, we know that the $h_{1,i}(x)$ are nonzero. By a similar argument as in Section 5 in \cite{rsw16}, we can show a lower bound on the cost of the denominator polynomial $h_{1,i}(x)$. Thus we can create new bounded variables $x_{l+1}, \cdots, x_{3r+l}$ to rewrite the objective function,

\begin{align*}
\min_{x\in \mathbb{R}^{l+3r}} & ~ q(x) / p(x). \\
\text{s.t.} & ~ h_{1,i}(x) x_{l+i} = 0, \forall i \in [r]. \\
& ~ h_{2,i}(x) x_{l+r+i} = 0, \forall i \in [r]. \\
& ~ h_{3,i}(x) x_{l+2r+i} = 0, \forall i \in [r]. \\
& ~ p(x) = \prod_{i=1}^r h_{1,i}^2(x) h_{2,i}^2(x) h_{3,i}^2(x)
\end{align*}
Note that the degree of the above system is $\poly(kr)$ and all the equality constraints can be merged into one single constraint. Thus, the number of constraints is $O(1)$. The number of variables is $O(rk^2/\epsilon)$. 

Using Theorem~\ref{thm:minimum_positive} and a similar argument from Section 5 of \cite{rsw16}, we have that the minimum nonzero cost is at least $2^{-n^\delta 2^{\wt{O}(rk^2/\epsilon) }}$. Combining the binary search explained in Section~\ref{sec:f}(similar techniques also can be found in Section 6 of \cite{rsw16}) with the lower bound we obtained, we can find the solution for the original problem in time,
\begin{align*}
(\nnz(A) + \nnz(W) + n 2^{\wt{O}(rk^2/\epsilon)}  ) n^{O(\delta)}.
\end{align*}
\end{proof}

\subsection{$r$ distinct columns, rows and tubes}\label{sec:w_r_distinct_3d_cols}

\begin{figure}[!t]
  \centering
    \includegraphics[width=1.0\textwidth]{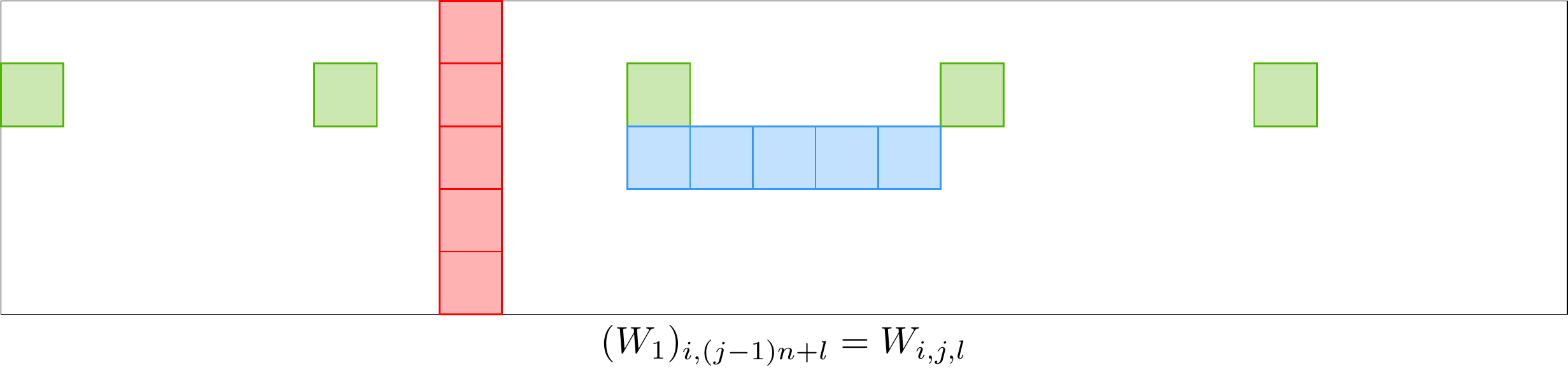}
    \includegraphics[width=1.0\textwidth]{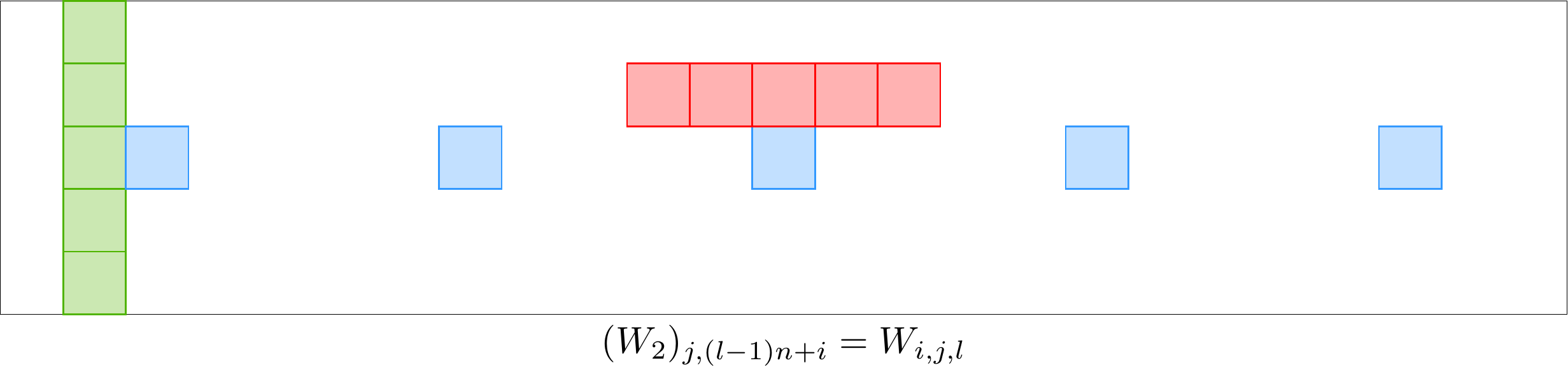}
    \includegraphics[width=1.0\textwidth]{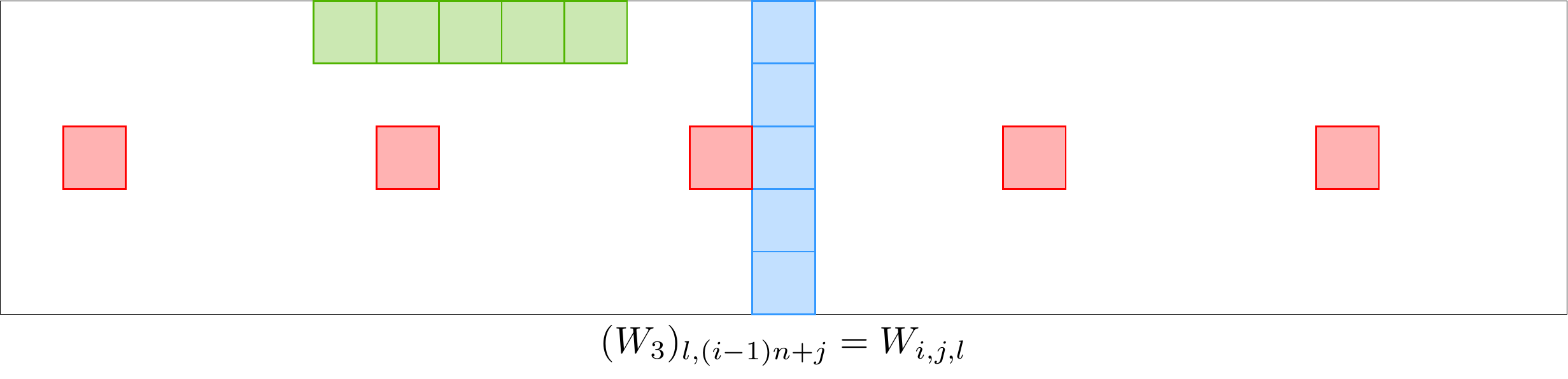}
    \caption{Let $W$ denote a tensor that has columns(red), rows(green) and tubes(blue). For each $i\in [3]$, let $W_i$ denote the matrix obtained by flattening tensor $W$ along the $i$-th dimension. }
\end{figure}

\begin{lemma}\label{lem:distinct_2d_cols_imply_distinct_2d_faces}
Let $W\in \mathbb{R}^{n\times n\times n}$ denote a tensor that has $r$ distinct columns and $r$ distinct rows, then $W$ has \\
$\mathrm{(\RN{1})}$ $r$ distinct column-tube faces.\\
$\mathrm{(\RN{2})}$ $r$ distinct row-tube faces.
\end{lemma}
\begin{proof}
Proof of Part (\RN{1}). Without loss of generality, we consider the first (which is the bottom one) column-row face. Assume it has $r$ distinct rows and $r$ distinct columns. We can re-order all the column-tube faces to make sure that all the $n$ columns in the bottom face have been split into $r$ continuous disjoint groups $C_i$, e.g., $\{C_1, C_2, \cdots, C_r \} = [n]$. Next, we can re-order all the row-tube faces to make sure that all the $n$ rows in the bottom face have been split into $r$ continuous disjoint groups $R_i$, e.g., $\{R_1, R_2, \cdots, R_r \} = [n]$. Thus, the new bottom face can be regarded as $r\times r$ groups, and the number in each position of the same group is the same.

Suppose that the tensor has $r+1$ distinct column-tube faces. By the pigeonhole principle there exist two different column-tube faces belonging to the same group $C_i$, for some $i\in [r]$. Note that these two column-tube faces are the same by looking at the bottom (column-row) face. Since they are distinct faces, there must exist one row vector $v$ which is not in the bottom (column-row) face, and it has a different value in coordinates belong to group $C_i$. Note that, considering the bottom face, for each row vector, it has the same value over coordinates belonging to group $C_i$. But $v$ has different values in coordinates belong to group $C_i$. Also, note that the bottom (column-row) face also has $r$ distinct rows, and $v$ is not one of them. This means there are at least $r+1$ distinct rows, which contradicts that there are $r$ distinct rows in total. Thus, there are at most $r$ distinct column-tube faces.

Proof of Part (\RN{2}). It is similar to Part (\RN{1}).
\end{proof}

\begin{figure}[!t]
  \centering
    \includegraphics[width=1.0\textwidth]{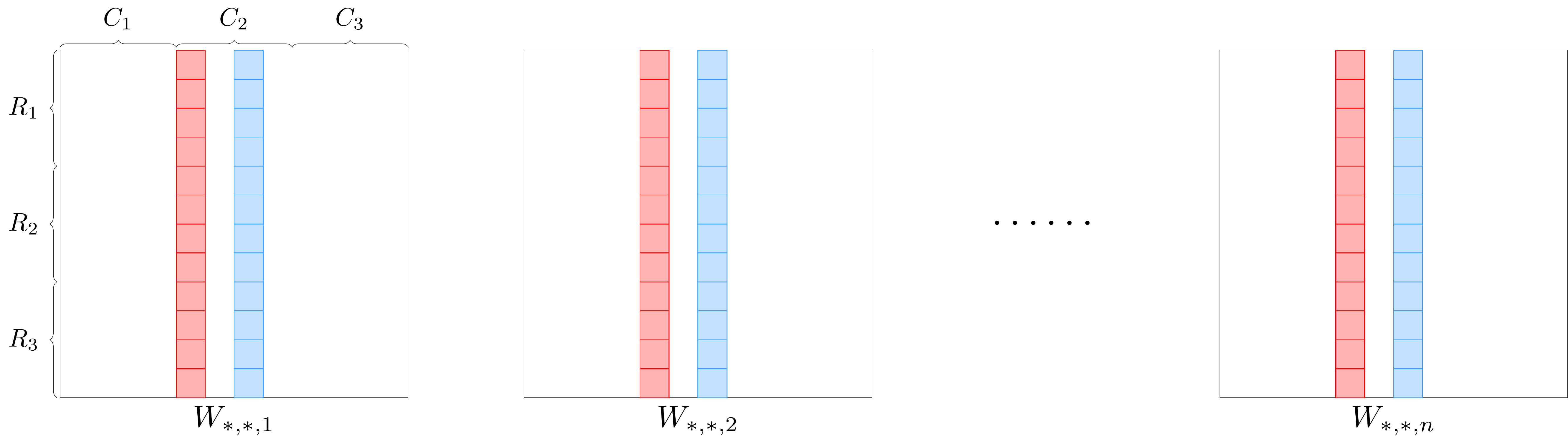}
    \caption{Each face $W_{*,*,i}$ is a column-row face. $W_{*,*,1}$ is the bottom column-row face. $r=3$. The blue blocks represent column-tube faces, the red blocks represent column-tube faces.}
\end{figure}

\begin{corollary}\label{cor:distinct_3d_cols_imply_distinct_3d_faces}
Let $W\in \mathbb{R}^{n\times n\times n}$ denote a tensor that has $r$ distinct columns, $r$ distinct rows, and $r$ distinct rubes. Then $W$ has $r$ distinct column-tube faces, $r$ distinct row-tube faces, and $r$ distinct column-row faces.
\end{corollary}
\begin{proof}
This follows by applying Lemma~\ref{lem:distinct_2d_cols_imply_distinct_2d_faces} twice.
\end{proof}

Thus, we obtain the same result as in Theorem~\ref{thm:w_r_distinct_columns_rows_tubes} by changing the assumption from $r$ distinct faces in each dimension to $r$ distinct columns, $r$ distinct rows and $r$ distinct tubes.

\subsection{$r$ distinct columns and rows}\label{sec:w_r_distinct_2d_cols}

The main difference between Theorem~\ref{thm:w_r_distinct_columns_rows_tubes} and Theorem~\ref{thm:w_r_distinct_2d_cols} is the running time. The first one takes $2^{\wt{O}(rk^2/\epsilon)}$ time and the second one is slightly longer, $2^{\wt{O}(r^2k^2/\epsilon)}$. By Lemma~\ref{lem:distinct_2d_cols_imply_distinct_2d_faces}, $r$ distinct columns in two dimensions implies $r$ distinct faces in two of the three kinds of faces. Thus, the following theorem also holds for $r$ distinct columns in two dimensions.

\begin{algorithm}[h]\caption{Weighted Tensor Low-rank Approximation Algorithm when the Weighted Tensor has $r$ Distinct Faces in Each of the Two Dimensions.}
\begin{algorithmic}
\Procedure{\textsc{WeightedRDistinctFacesIn2Dimensions}}{$A,W,n,r,k,\epsilon$} \Comment{Theorem~\ref{thm:w_r_distinct_2d_cols}}
\For{$j=1\to 3$}
	\State $s_j \leftarrow O(k/\epsilon)$.
	\State Choose a sketching matrix $S_j \in \mathbb{R}^{n^2 \times s_j}$.
	\If {$j \neq 3$}
		\For{$i=1\to r$}
			\State Create $k\times s_1$ variables for matrix $P_{i,j}\in \mathbb{R}^{k\times s_j}$.
		\EndFor
	\EndIf
	\For{$i=1\to n$}
		\State Write down $(\wh{U}_j)^i = A_i^j D_{W_1^j} S_j P_{j,i}^\top ( P_{j,i} P_{j,i}^\top )^{-1}$.
	\EndFor
\EndFor
\State Form $\| W\circ (  \wh{U}_1 \otimes \wh{U}_2 \otimes \wh{U}_3 - A) \|_F^2$.
\State Run polynomial system verifier.
\State \Return $U_1,U_2,U_3$
\EndProcedure
\end{algorithmic}
\end{algorithm}

\begin{theorem}\label{thm:w_r_distinct_2d_cols}
 Given a $3$rd order $n\times n\times n$ tensor $A$ and an $n\times n \times n$ tensor $W$ of weights with $r$ distinct faces in two dimensions (out of three dimensions) such that each entry can be written using $O(n^\delta)$ bits for some $\delta >0$,
define $\OPT={\inf}_{\rank-k~A_k} \| W\circ (A_k - A )\|_F^2$. For any $k\geq 1$ and any $0 < \epsilon <1$. 

$\mathrm{(\RN{1})}$ If $\OPT>0$, and there exists a rank-$k$ $A_k=U_1^*\otimes U_2^* \otimes U_3^*$ tensor (with size $n\times n\times n$) such that $\| W\circ (A_k-A)\|_F^2 = \OPT$, and $\max_{i\in [3]}\|U_i^*\|_F \leq 2^{O(n^\delta)}$, then there exists an algorithm that takes $(\nnz(A)+\nnz(W)  + n 2^{\wt{O}(r^2k^2/\epsilon)} ) n^{O(\delta)}$ time in the unit cost $\RAM$ model with words of size $O(\log n)$ bits\footnote{The entries of $A$ and $W$ are assumed to fit in $n^{\delta}$ words.} and outputs three $n \times k$ matrices $U_1,U_2,U_3$ such that
\begin{align}\label{eq:w_thm_2cols_1}
\left\| W\circ \left( U_1 \otimes U_2 \otimes U_3 - A \right) \right\|_F^2 \leq (1+\epsilon) \OPT
\end{align}
holds with probability $9/10$. 

$\mathrm{(\RN{2})}$ If $\OPT>0$, $A_k$ does not exist, and there exist three $n\times k$ matrices $U_1',U_2',U_3'$ where each entry can be written using $O(n^\delta)$ bits and $\| W \circ( U_1'\otimes U_2'\otimes U_3' - A) \|_F^2 \leq (1+\epsilon/2) \OPT$, then we can find $U,V,W$ such that \eqref{eq:w_thm_2cols_1} holds.

$\mathrm{(\RN{3})}$ If $\OPT=0$, $A_k$ exists, and there exists a solution $U_1^*,U_2^*,U_3^*$ such that each entry of the matrix can be written using $O(n^\delta)$ bits, then we can obtain \eqref{eq:w_thm_2cols_1}. 

$\mathrm{(\RN{4})}$ If $\OPT=0$, and there exist three $n\times k$ matrices $U_1,U_2,U_3$ such that $\max_{i\in [3]}\|U_i^*\|_F \leq 2^{O(n^\delta)}$ and
 \begin{align}\label{eq:w_thm_2cols_2}
\left\| W\circ \left( U_1 \otimes U_2 \otimes U_3 - A \right) \right\|_F^2 \leq (1+\epsilon) \OPT + 2^{-\Omega(n^\delta)},
\end{align}
then we can output $U_1,U_2,U_3$ such that \eqref{eq:w_thm_2cols_2} holds.

$\mathrm{(\RN{5})}$ Further if $A_k$ exists, we can output a number $Z$ for which $\OPT \leq Z \leq (1+\epsilon) \OPT$.

For all the cases, the algorithm succeeds with probability at least $9/10$.
\end{theorem}
\begin{proof}
By Lemma~\ref{lem:distinct_2d_cols_imply_distinct_2d_faces}, we have $W$ has $r$ distinct column-tube faces and $r$ distinct row-tube faces. By Claim~\ref{cla:w_2rlogr_col_row_faces}, we know that $W$ has $R=2^{O(r\log r)}$ distinct column-row faces.

We use the same approach as in proof of Theorem~\ref{thm:w_r_distinct_columns_rows_tubes} (which is also similar to Section 8 of \cite{rsw16}) to create variables, write down the polynomial systems and add not equal constraints. Instead of having $3r$ distinct denominators as in the proof of Theorem~\ref{thm:w_r_distinct_columns_rows_tubes}, we have $2r+R$.

 We create $l=O(rk^2/\eps)$ variables for $\{ Z_1 D_{W_1^1} S_1, Z_1 D_{W_1^2} S_1, \cdots , Z_1 D_{W_1^r} S_1\}$. Then we can write down $\widehat{U}_1$ with $r$ distinct denominators $g_{i}(x)$. Each $g_{i}(x)$ is non-zero in an optimal solution using the perturbation argument in Section 4 in \cite{rsw16}. We create new variables $x_{2l+i}$ to remove the denominators $g_{i}(x)$, $\forall i\in [r]$. Then the entries of $\widehat{U}_1$ are polynomials as opposed to rational functions.

 We create $l=O(rk^2/\eps)$ variables for $\{ Z_2 D_{W_2^1} S_2, Z_2 D_{W_2^2} S_2, \cdots , Z_2 D_{W_2^r} S_2\}$. Then we can write down $\widehat{U}_2$ with $r$ distinct denominators $g_{r+i}(x)$. Each $g_{r+i}(x)$ is non-zero in an optimal solution using the perturbation argument in Section 4 in \cite{rsw16}.  We create new variables $x_{2l+r+i}$ to remove the denominators $g_{r+i}(x)$, $\forall i\in [r]$. Then the entries of $\widehat{U}_2$ are polynomials as opposed to rational functions.

 Using $\widehat{U}_1$ and $\wh{U}_2$ we can express $\widehat{U}_3$ with $R$ distinct denominators $f_i(x)$, which are also non-zero by using the perturbation argument in Section 4 in \cite{rsw16}, and using that $W_3$ has at most this number of distinct rows. Finally we can write the following optimization problem,
\begin{eqnarray*}
\underset{x\in \mathbb{R}^{2l+2r}} {\min } && p(x)/q(x) \\
\mathrm{s.t.} &&  g_{i}(x) x_{2l+i} -1=0, \forall i\in [r]\\
&&  g_{r+i}(x) x_{2l+r+i} -1=0, \forall i\in [r]\\
&& f_j^2(x) \neq 0, \forall j\in [ R ] \\
&& q(x) =  \prod_{j=1}^{ R } f_j^2(x)
\end{eqnarray*}
We then determine if there exists a solution to the above semi-algebraic set in time  
\begin{equation*}
(  \poly(k,r) R  )^{O(rk^2/\eps)} = 2^{ \wt{O}(r^2k^2/\eps)}. 
\end{equation*}
Using similar techniques from Section 5 of \cite{rsw16}, we can show a lower bound on the cost similar to Section 8.3 of \cite{rsw16}, namely, the minimum nonzero cost is at least
\begin{align*}
2^{-n^\delta 2^{\wt{O}(r^2 k^2/\epsilon)}}.
\end{align*}

Combining the binary search explained in Section \ref{sec:f} (a similar techniques also can be found in Section 6 of \cite{rsw16}) with the lower bound we obtained, we can find a solution for the original problem in time
\begin{equation*}
 (\nnz(A) + \nnz(W)  + n  2^{ \widetilde{O}(r^2k^2/\eps)}  ) n^{O(\delta)}.
\end{equation*}
\end{proof}
\begin{remark}
  Note that the running time for the Frobenius norm and for the $\ell_1$ norm are of the form $\poly(n) + \exp(\poly(k/\epsilon))$ rather than $\poly(n) \cdot \exp(k/\epsilon)$. The reason is, we can use an input sparsity reduction to reduce the size of the objective function from $\poly(n)$ to $\poly(k)$.
\end{remark}

\begin{figure}[!t]
  \centering
    \includegraphics[width=1.0\textwidth]{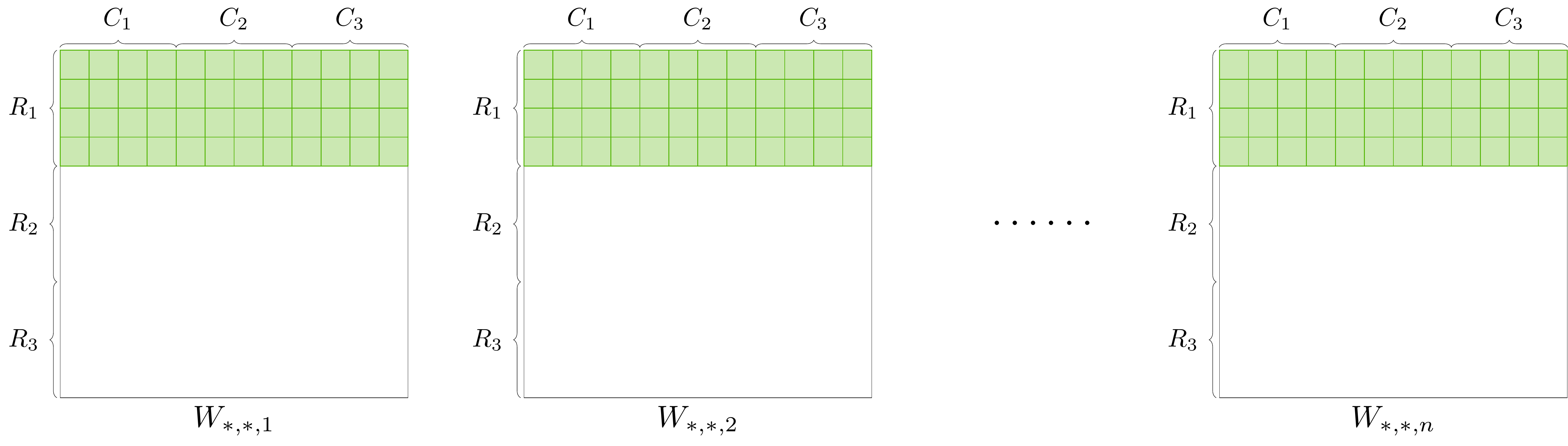}
    \includegraphics[width=1.0\textwidth]{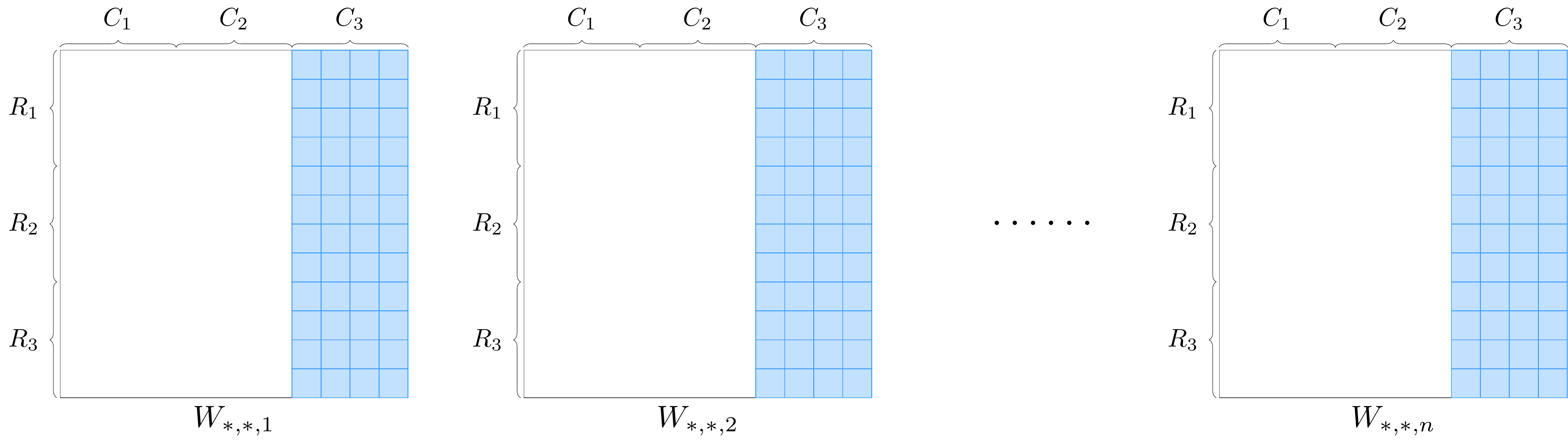}
    \includegraphics[width=1.0\textwidth]{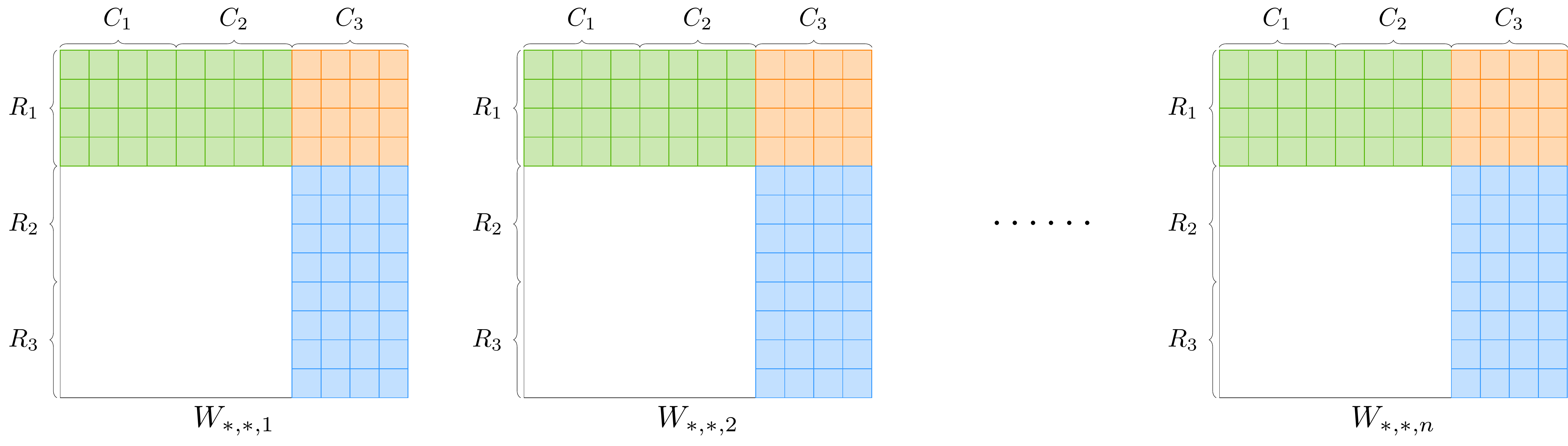}
    \caption{Each face $W_{*,*,i}$ is a column-row face. $W_{*,*,1}$ is the bottom column-row face. $r=3$. The blue blocks represent $|C_3|$ column-tube faces. The green blocks represet $|R_3|$ row-tube faces. In each column-row face, the intersection between blue faces and green faces is a size $|R_3|\times |C_3|$ block, and all the entries in this block are the same.}
\end{figure}

\begin{claim}\label{cla:w_2rlogr_col_row_faces}
Let $W\in \mathbb{R}$ denote a third order tensor that has $r$ distinct columns and $r$ distinct rows. Then it has $2^{O(r\log r)}$ distinct column-row faces. 
\end{claim}
\begin{proof}
By similar arguments as in the proof of Lemma~\ref{lem:distinct_2d_cols_imply_distinct_2d_faces}, the bottom (column-row) face can be split into $r$ groups $C_1, C_2, \cdots, C_r$ based on $r$ columns, and split into $r$ groups $R_1,R_2,\cdots, R_r$ based on rows. Thus, the bottom (column-row) face can be regarded as having $r\times r$ groups, and the number in each position of the same group is the same. 

We can assume that all the $r^2$ blocks in the bottom column-row face have the same size. Otherwise, we can expand the tensor to the situation that all the $r^2$ blocks have the same size. Because this small tensor is a sub-tensor of the big tensor, if the big tensor has at most $t$ distinct column-row faces, then the small tensor has at most $t$ distinct column-row faces.

By Lemma~\ref{lem:distinct_2d_cols_imply_distinct_2d_faces}, we know that the tensor $W$ has at most $r$ distinct column-tube faces and row-tube faces. Because it has $r$ distinct column-tube faces, then all the faces belonging to coordinates in $C_r$ are the same. Thus, all the columns belonging to $C_r$ and in the second column-row face are the same. Similarly, we have that all the rows belonging to $R_r$ and in the second column-row face are the same. Thus we have that all the entries in block $C_R \cup R_r$ and in the second column-row faces are the same. Further, we can conclude, for every column-row face, for every $C_i \cup R_j$ block, all the entries in the same block are the same.

The next observation is, if there exist $r^2+1$ different values in the tensor, then there exist either $r$ distinct columns or $r$ distinct rows. Indeed,
otherwise since we have $r$ distinct columns, each column has at most $r$ distinct entries given our bound on the nunber of distinct rows. Thus, the $r$ distinct columns could have at most $r^2$ distinct entries in total, a contradiction.

For each column-row face, there are at most $r^2$ blocks, and the value in each block can have at most $r^2$ possibilities. Thus, overall we have at most $(r^2)^{r^2} = 2^{O(r^2\log r)}$ column-row faces.

By using different argument, we can improve the above bound. Note that we already show in each column-row face of a tensor, it has $r^2$ blocks, and all the values in each block have to be the same. Since we have $r$ distinct rows, we can fix the those $r$ distinct rows. If we copy row $v$ into one row of $R_i$, then we have to copy row $v$ into every row of $R_i$. This is because if $R_i$ contains two distinct rows, then there must exist a block $C_j$ for which the entries in block $R_i \cup C_j$ are not all the same. Thus, for each row group, all the rows in that group are the same. 

Now, for each column-row face, consider the leftmost $r$ blocks, $R_1 \cup C_1$, $R_2 \cup C_1$, $\cdots$, $R_r \cup C_1$. There are at most $r$ possible values in each block, because we have $r$ distinct rows in total. Overall the total number of possibilities for the leftmost $r$ blocks is at most $(r)^r = 2^{O(r\log r)}$. Once the leftmost $r$ blocks are determined, the remaining $r(r-1)$ are also determined. This completes the proof.

\end{proof}
Also, notice that there is an example that has $2^{\Omega(r\log r)}$ distinct column-row faces. For the bottom column-row faces, there are $r\times r$ blocks for which all the blocks have the same size, the blocks on the diagonal have all $1$s, and all the other blocks contain $0$s everywhere. For the later column-row faces, we can arbitrarily permute this block diagonal matrix, and the total number of possibilities is $\Omega(r!) \geq 2^{\Omega(r\log r)}$.

\newpage
\section{Hardness}\label{sec:hardness}
We first provide definitions and results for some fundamental problems in Section~\ref{sec:hardness_definitions}. Section~\ref{sec:hardness_symmetric_tensor_eigenvalue} presents our hardness result for the symmetric tensor eigenvalue problem. Section~\ref{sec:hardness_singular_spectral_rank1} presents our hardness results for
symmetric tensor singular value problems, computing tensor spectral norm, and rank-$1$ approximation. We improve H{\aa}stad's NP-hardness\cite{h90} result for tensor rank in Section~\ref{sec:hardness_tensor_rank}. We also show a better hardness result for robust subspace approximation in Section~\ref{sec:hardness_robust_subspace_approximation}. Finally, we discuss several other tensor hardness results that are implied by matrix hardness results in Section~\ref{sec:hardness_matrix_extension}.

\subsection{Definitions}\label{sec:hardness_definitions}
We first provide the definitions for \SAT, \ETH, \MAX-\SAT, \MAX-\ESAT and then state some fundamental results related to those definitions.
\begin{definition}[\SAT problem]\label{def:3SAT}
Given $n$ variables and $m$ clauses in a conjunctive normal form \CNF formula with the size of each clause at most $3$, the goal is to decide whether there exists an assignment to the $n$ Boolean variables to make the \CNF formula be satisfied.
\end{definition}

\begin{hypothesis}[Exponential Time Hypothesis ($\mathsf{ETH}$) \cite{ipz98}]
There is a $\delta>0$ such that the \SAT problem defined in Definition \ref{def:3SAT} cannot be solved in $O(2^{\delta n})$ time.
\end{hypothesis}

\begin{definition}[\MAX-$\mathsf{3SAT}$]
Given $n$ variables and $m$ clauses, a conjunctive normal form \CNF formula with the size of each clause at most $3$, the goal is to find an assignment that satisfies the largest number of clauses.
\end{definition}

We use \MAX-\ESAT to denote the version of \MAX-\SAT where each clause contains exactly $3$ literals.
\begin{theorem}[\cite{h01}]
For every $\delta > 0$, it is \NP-hard to distinguish a satisfiable instance of \MAX-\ESAT from an instance where at most a $7/8+\delta$ fraction of the clauses can be simultaneously satisfied.
\end{theorem}

\begin{theorem}[\cite{h01,mr10}]
Assume \ETH holds. For every $\delta>0$, there is no $2^{o(n^{1-o(1)})}$ time algorithm to distinguish a satisfiable instance of \MAX-\ESAT from an instance where at most a fraction $7/8+\delta$ of the clauses can be simultaneously satisfied.
\end{theorem}

We use \MAX-\ESATB to denote the restricted special case of \MAX-\SAT where every variable occurs in at most $B$ clauses. H{\aa}stad \cite{h00} proved that the problem is approximable to within a factor $7/8+1/(64B)$ in polynomial time, and that it is hard to approximate within a factor $7/8+1/(\log B)^{\Omega(1)}$. In 2001, Trevisan improved the hardness result,
\begin{theorem}[\cite{t01}]
Unless \RP=\NP, there is no polynomial time $(7/8+5/\sqrt{B})$-approximate algorithm for \MAX-\ESATB. 
\end{theorem}

\begin{theorem}[\cite{h01,t01,mr10}]\label{thm:hardness_max_esatb}
Unless \ETH fails, there is no $2^{o(n^{1-o(1)})}$ time $(7/8+5/\sqrt{B})$-approximate algorithm for \MAX-\ESATB. 
\end{theorem}

\begin{theorem}[\cite{lms11}]
Unless \ETH fails, there is no $2^{o(n)}$ time algorithm for the Independent Set problem.
\end{theorem}

\begin{definition}[$\mathsf{MAX}$-$\mathsf{CUT}$ decision problem]
Given a positive integer $c^*$ and an unweighted graph $G=(V,E)$ where $V$ is the set of vertices of $G$ and $E$ is the set of edges of $G$, the goal is to determine whether there is a cut of $G$ that has at least $c^*$ edges.
\end{definition}

Note that Feige's original assumption\cite{f02} states that there is no polynomial time algorithm for the problem in Assumption~\ref{ass:def_random_eth}. We do not know of any better algorithm for the problem in Assumption \ref{ass:def_random_eth} and have consulted several experts\footnote{Personal communication with Russell Impagliazzo and Ryan Williams.} about the assumption who do not know a counterexample to it.
\begin{assumption}[Random Exponential Time Hypothesis]\label{ass:def_random_eth}
Let $c>\ln 2$ be a constant. Consider a random \SAT formula on $n$ variables in which each clause has $3$ literals, and in which each of the $8n^3$ clauses is picked independently with probability $c/n^2$. Then any algorithm which always outputs $1$ when the random formula is satisfiable, and outputs $0$ with probability at least $1/2$ when the random formula is unsatisfiable, must run in $2^{c'n}$ time on some input, where $c'>0$ is an absolute constant.
\end{assumption}
The $\mathsf{4SAT}$-version of the above random-\ETH assumption has been used in \cite{gl04} and \cite{rsw16} (Assumption 1.3).

\subsection{Symmetric tensor eigenvalue}\label{sec:hardness_symmetric_tensor_eigenvalue}
\begin{definition}[Tensor Eigenvalue~\cite{hl13}]
An eigenvector of a tensor $A\in \mathbb{R}^{n\times n\times n}$ is a nonzero vector $x\in \mathbb{R}^n$ such that
\begin{align*}
\sum_{i=1}^n \sum_{j=1}^n A_{i,j,k} x_i x_j = \lambda x_k, \forall k\in [n]
\end{align*}
for some $\lambda\in \mathbb{R}$, which is called an eigenvalue of $A$.
\end{definition}

\begin{theorem}[\cite{n03}]\label{thm:hardness_max_independent_set}
Let $G=(V,E)$ on $v$ vertices have stability number (the size of a maximum independent set) $\alpha(G)$. Let $n = v + \frac{v(v-1)}{2}$ and $\mathbb{S}^{n-1} = \{ (x,y) \in \mathbb{R}^v \times \mathbb{R}^{v(v-1)/2} : \| x\|_2^2 + \| y\|_2^2 = 1 \}$. Then,
\begin{align*}
\sqrt{1-\frac{1}{\alpha(G)}} = 3 \sqrt{3/2} \underset{(x,y)\in \mathbb{S}^{n-1} }{\max} \sum_{i< j, (i,j)\notin E} x_i x_j y_{i,j}.
\end{align*}
\end{theorem}

For any graph $G(V,E)$, we can construct a symmetric tensor $A\in \mathbb{R}^{n\times n \times n}$. For any $1\leq i < j < k \leq v$, let
\begin{align*}
A_{i,j,k} = \begin{cases}
1 & 1\leq i < j \leq v, k = v+\phi(i,j), (i,j) \notin E, \\
0 & \text{otherwise},
\end{cases}
\end{align*}
where $\phi(i,j) = (i-1) v - i(i-1)/2 + j-i$ is a lexicographical enumeration of the $v(v-1)/2$ pairs $i<j$. For the other cases $i<k<j$, $\cdots$, $k<j<i$, we set
\begin{align*}
A_{i,j,k} = A_{i,k,j} = A_{j,i,k} = A_{j,k,i} = A_{k,i,j} = A_{k,j,i}.
\end{align*}
If two or more indices are equal, we set $A_{i,j,k}=0$. Thus tensor $T$ has the following property,
\begin{align*}
A(z,z,z) = 6 \sum_{i<j, (i,j)\notin E} x_i x_j y_{i,j},
\end{align*}
where $z = (x,y) \in \mathbb{R}^n$.

Thus, we have
\begin{align*}
\lambda=\max_{z \in \mathbb{S}^{n-1}} A(z,z,z) = \max_{(x,y)\in \mathbb{S}^{n-1} } 6 \sum_{i<j, (i,j)\notin E} x_i x_j y_{i,j}.
\end{align*}
Furthermore, $\lambda$ is the maximum eigenvalue of $A$.

\begin{theorem}\label{thm:exp_hard_eigen}
Unless \ETH fails, there is no $2^{o(\sqrt{n})}$ time to approximate the largest eigenvalue of an $n$-dimensional symmetric tensor within $(1\pm \Theta(1/n))$ relative error.
\end{theorem}
\begin{proof}
The additive error is at least
\begin{align*}
\sqrt{1-1/v} - \sqrt{1-1/(v-1)} = \frac{ 1/(v-1) - 1/v }{ \sqrt{1-1/v} + \sqrt{1-1/(v-1)}  } \gtrsim 1/(v-1) - 1/v \geq 1/v^2.
\end{align*}
Thus, the relative error is $(1\pm \Theta(1/v^2))$. By the definition of $n$, we know $n=\Theta(v^2)$. Assuming \ETH, there is no $2^{o(v)}$ time algorithm to compute the clique number of $\ov{G}$. Because the clique number of $\ov{G}$ is $\alpha(G)$, there is no $2^{o(v)}$ time algorithm to compute $\alpha(G)$. Furthermore, there is no $2^{o(v)}$ time algorithm to approximate the maximum eigenvalue within $(1\pm \Theta(1/v^2))$ relative error. Thus, we complete the proof.
\end{proof}

\begin{corollary}\label{cor:exp_hard_eigen}
Unless \ETH fails, there is no polynomial running time algorithm to approximate the largest eigenvalue of an $n$-dimensional tensor within $(1\pm \Theta(1/\log^{2+\gamma}(n)))$ relative-error, where $\gamma>0$ is an arbitrarily small constant.
\end{corollary}
\begin{proof}
We can apply a padding argument here. According to Theorem~\ref{thm:exp_hard_eigen}, there is a $d$-dimensional tensor such that there is no $2^{o(\sqrt{d})}$ time algorithm that can give a $(1+\Theta(1/d))$ relative error approximation. If we pad $0$s everywhere to extend the size of the tensor to $n=2^{d^{(1-\gamma')/2}}$, where $\gamma'>0$ is a sufficiently small constant, then $\poly(n)=2^{o(\sqrt{d})}$, so $d=\log^{2+O(\gamma')}(n)$. Thus, it means that there is no polynomial running time algorithm which can output a $(1+1/(\log^{2+\gamma}))$-relative approximation to the tensor which has size $n$.

\end{proof}

\subsection{Symmetric tensor singular value, spectral norm and rank-$1$ approximation}\label{sec:hardness_singular_spectral_rank1}

\cite{hl13} defines two kinds of singular values of a tensor. In this paper, we only consider the following kind:
\begin{definition}[$\ell_2$ singular value in~\cite{hl13}]
Given a $3$rd order tensor $A\in \mathbb{R}^{n_1 \times n_2 \times n_3 }$, the number $\sigma\in \mathbb{R}$ is called a singular value and the nonzero $u\in\mathbb{R}^{n_1}$,$v\in \mathbb{R}^{n_2}$,$w\in \mathbb{R}^{n_3}$ are called singular vectors of $A$ if
\begin{align*}
\sum_{j=1}^{n_2} \sum_{k=1}^{n_3} A_{i,j,k} v_j w_k & = ~ \sigma u_i, \forall i \in [n_1] \\
\sum_{i=1}^{n_1} \sum_{k=1}^{n_3} A_{i,j,k} u_i w_k & = ~ \sigma v_j, \forall j \in [n_2] \\
\sum_{i=1}^{n_1} \sum_{j=1}^{n_2} A_{i,j,k} u_i v_j & = ~ \sigma w_k, \forall k \in [n_3].
\end{align*}
\end{definition}

\begin{definition}[Spectral norm~\cite{hl13}]
The spectral norm of a tensor $A$ is:
$$\|A\|_{2}=\underset{x,y,z\neq 0}{\sup} \frac{|A(x,y,z)|}{\|x\|_2 \|y\|_2 \|z\|_2} $$
\end{definition}

Notice that the spectral norm is the absolute value of either the maximum value of $\frac{A(x,y,z)}{\|x\|_2\|y\|_2\|z\|_2}$ or the minimum value of it. Thus, it is an $\ell_2$-singular value of $A$. Furthermore, it is the maximum $\ell_2$-singular value of $A$.

\begin{theorem}[\cite{b38}]
Let $A\in \mathbb{R}^{n\times n\times n}$ be a symmetric $3$rd order tensor. Then,
\begin{align*}
\|A\|_{2}=\underset{x,y,z\neq 0}{\sup} \frac{A(x,y,z)}{\|x\|_2 \|y\|_2 \|z\|_2} = \underset{x\neq 0}{\sup} \frac{|A(x,x,x)|}{ \| x\|_2^3 }.
\end{align*}
\end{theorem}

It means that if a tensor is symmetric, then its largest eigenvalue is the same as its largest singular value and its spectral norm.
Then, by combining with Theorem~\ref{thm:exp_hard_eigen}, we have the following corollary:
\begin{corollary}
Unless \ETH fails,
\begin{enumerate}
\item There is no $2^{o(\sqrt{n})}$ time algorithm to approximate the largest singular value of an $n$-dimensional symmetric tensor within $(1+\Theta(1/n))$ relative-error.
\item There is no $2^{o(\sqrt{n})}$ time algorithm to approximate the spectral norm of an $n$-dimensional symmetric tensor within $(1+\Theta(1/n))$ relative-error.
\end{enumerate}
\end{corollary}
By Corollary~\ref{cor:exp_hard_eigen}, we have:
\begin{corollary}
Unless \ETH fails,
\begin{enumerate}
\item There is no polynomial time algorithm to approximate the largest singular value of an $n$-dimensional tensor within $(1+ \Theta(1/\log^{2+\gamma}(n)))$ relative-error, where $\gamma>0$ is an arbitrarily small constant.
\item There is no polynomial time algorithm to approximate the spectral norm of an $n$-dimensional tensor within $(1+ \Theta(1/\log^{2+\gamma}(n)))$ relative-error, where $\gamma>0$ is an arbitrarily small constant.
\end{enumerate}
\end{corollary}

Now, let us consider Frobenius norm rank-$1$ approximation.
\begin{theorem}[\cite{b38}]
Let $A\in \mathbb{R}^{n\times n\times n}$ be a symmetric $3$rd order tensor. Then,
$$\min_{\sigma\geq 0,\|u\|_2=\|v\|_2=\|w\|_2=1}\|A-\sigma u\otimes v\otimes w\|_F=\min_{\lambda\geq 0,\|v\|_2=1}\|A-\lambda v\otimes v\otimes v\|_F.$$
Furthermore, the optimal $\sigma$ and $\lambda$ may be chosen to be equal.
\end{theorem}

Notice that
\begin{align*}
\|A-\sigma u\otimes v\otimes w\|_F^2=\|A\|_F^2-2\sigma A(u,v,w)+\sigma^2\|u\otimes v\otimes w\|_F^2.
\end{align*}
Then, if $\|u\|_2=\|v\|_2=\|w\|_2=1,$ we have:
\begin{align*}
\|A-\sigma u\otimes v\otimes w\|_F^2=\|A\|_F^2-2\sigma A(u,v,w)+\sigma^2.
\end{align*}
When $A(u,v,w)=\sigma$, then the above is minimized.

Thus, we have:
\begin{align*}
\min_{\sigma\geq 0,\|u\|_2=\|v\|_2=\|w\|_2=1}\|A-\sigma u\otimes v\otimes w\|_F^2+\|A\|_{2}^2=\|A\|_F^2.
\end{align*}
It is sufficient to prove the following theorem:
\begin{theorem}\label{thm:approximate_rank1_is_eth_hard}
Given $A\in\mathbb{R}^{n\times n\times n}$, unless \ETH fails, there is no $2^{o(\sqrt{n})}$ time algorithm to compute $u',v',w'\in\mathbb{R}^n$ such that
\begin{align*}
\|A-u'\otimes v'\otimes w'\|_F^2\leq (1+\varepsilon)\min_{u,v,w\in\mathbb{R}^n}\|A-u\otimes v\otimes w\|_F^2,
\end{align*}
where $\varepsilon=O(1/n^2).$
\end{theorem}

\begin{proof}
Let $A\in\mathbb{R}^{n\times n\times n}$ be the same hard instance mentioned in Theorem~\ref{thm:hardness_max_independent_set}. Notice that each entry of $A$ is either $0$ or $1$. Thus, $\min_{u,v,w\in\mathbb{R}^n}\|A-u\otimes v\otimes w\|_F^2\leq \|A\|_F^2$. Notice that Theorem~\ref{thm:hardness_max_independent_set} also implies that it is hard to distinguish the two cases $\|A\|_2\leq 2\sqrt{2/3}\cdot\sqrt{1-1/c}$ or $\|A\|_2\geq 2\sqrt{2/3}\cdot\sqrt{1-1/(c+1)}$ where $c$ is an integer which is no greater than $\sqrt{n}$. So the difference between $(2\sqrt{2/3}\cdot\sqrt{1-1/c})^2$ and $(2\sqrt{2/3}\cdot\sqrt{1-1/(c+1)})^2$ is at least $\Theta(1/n)$. Since $\|A\|_F^2$ is at most $n$ (see construction of $A$ in the proof of Lemma~\ref{thm:hardness_max_independent_set}), $\Theta(1/n)$ is an $\varepsilon=O(1/n^2)$ fraction of $\min_{u,v,w\in\mathbb{R}^n}\|A-u\otimes v\otimes w\|_F^2$. Because
$$\min_{u,v,w\in\mathbb{R}^n}\|A- u\otimes v\otimes w\|_F^2+\|A\|_{2}^2=\|A\|_F^2,$$
if we have a
$2^{o(\sqrt{n})}$ time algorithm to compute $u',v',w'\in\mathbb{R}^n$ such that
\begin{align*}
\|A-u'\otimes v'\otimes w'\|_F^2\leq (1+\varepsilon)\min_{u,v,w\in\mathbb{R}^n}\|A-u\otimes v\otimes w\|_F^2
\end{align*}
for $\varepsilon=O(1/n^2),$ it will contradict the fact that we cannot distinguish whether $\|A\|_2\leq 2\sqrt{2/3}\cdot\sqrt{1-1/c}$ or $\|A\|_2\geq 2\sqrt{2/3}\cdot\sqrt{1-1/(c+1)}$.
\end{proof}

\begin{corollary}\label{cor:two_to_the_one_over_eps_to_the_forth}
Given $A\in\mathbb{R}^{n\times n\times n}$, unless \ETH fails, for any $\varepsilon$ for which $\frac{1}{2}\geq \varepsilon\geq c/n^2$ where $c$ is any constant, there is no $2^{o(\varepsilon^{-1/4})}$ time algorithm to compute $u',v',w'\in\mathbb{R}^n$ such that
\begin{align*}
\|A-u'\otimes v'\otimes w'\|_F^2\leq (1+\varepsilon)\min_{u,v,w\in\mathbb{R}^n}\|A-u\otimes v\otimes w\|_F^2.
\end{align*}
\end{corollary}
\begin{proof}
If $\varepsilon=\Omega(1/n^2),$ it means that $n=\Omega(1/\sqrt{\varepsilon})$. Then, we can construct a hard instance $B$ with size $m\times m\times m$ where $m=\Theta(1/\sqrt{\varepsilon}),$ and we can put $B$ into $A$, and let $A$ have zero entries elsewhere. Since $B$ is hard, i.e., there is no $2^{o(m^{-1/2})}=2^{o(\varepsilon^{-1/4})}$ running time to compute a rank-$1$ approximation to $B$, this means there is no $2^{o(\varepsilon^{-1/4})}$ running time algorithm to find an approximate rank-$1$ approximation to $A$.
\end{proof}

\begin{corollary}
Unless \ETH fails, there is no polynomial time algorithm to approximate the best rank-$1$ approximation of an $n$-dimensional tensor within $(1+ \Theta(1/\log^{2+\gamma}(n)))$ relative-error, where $\gamma>0$ is an arbitrarily small constant.
\end{corollary}
\begin{proof}
We can apply a padding argument here. According to Theorem~\ref{thm:approximate_rank1_is_eth_hard}, there is a $d$-dimensional tensor such that there is no $2^{o(\sqrt{d})}$ time algorithm which can give a $(1+\Theta(1/d^4))$ relative approximation. Then, if we pad with $0$s everywhere to extend the size of the tensor to $n=2^{d^{(1-\gamma')/2}}$ where $\gamma'>0$ is a sufficiently small constant, then $\poly(n)=2^{o(\sqrt{d})}$, and $d^4=\log^{2+O(\gamma')}(n)$. Thus, it means that there is no polynomial time algorithm which can output a $(1+1/(\log^{2+\gamma}))$-relative error approximation to the tensor which has size $n$.
\end{proof}

\subsection{Tensor rank is hard to approximate}\label{sec:hardness_tensor_rank}

This section presents the hardness result for approximating tensor rank under \ETH.
According to our new result, we notice that not only deciding the tensor rank is a hard problem, but also approximating the tensor rank is a hard problem. This therefore strengthens H{\aa}stad's NP-Hadness \cite{h90} for computing tensor rank. 

\subsubsection{Cover number}
Before getting into the details of the reduction, we provide a definition of an important concept called the ``cover number'' and discuss the cover number for the \MAX-\ESATB problem.
\begin{definition}[Cover number]\label{def:hardness_cover_number}
For any \SAT instance $S$ with $n$ variables and $m$ clauses, we are allowed to assign one of three values $\{0,1,*\}$ to each variable. For each clause, if one of the literals outputs true, then the clause outputs true. For each clause, if the corresponding variable of one of the literals is assigned to $*$, then the clause outputs true. We say $y\in \{0,1\}^n$ is a string, and $z\in \{0,1,*\}^n$ is a star string. For an instance $S$, if there exists a string $y\in \{0,1\}^n$ that causes all the clauses to be true, then we say that $S$ is satisfiable, otherwise it is unsatisfiable. For an instance $S$, let $Z_S$ denote the set of star strings which cause all of the clauses of $S$ to be true. For each star string $z\in \{0,1,*\}^n$, let $\mathrm{star}(z)$ denote the number of $*$s in the star-string $z$. We define the ``cover number'' of instance $S$ to be
\begin{align*}
\mathrm{cover\text{-}number}(S) = \min_{z\in Z_S} \mathrm{star}(z).
\end{align*}
\end{definition}

\begin{figure}[!t]
  \centering
    \includegraphics[width=0.5\textwidth]{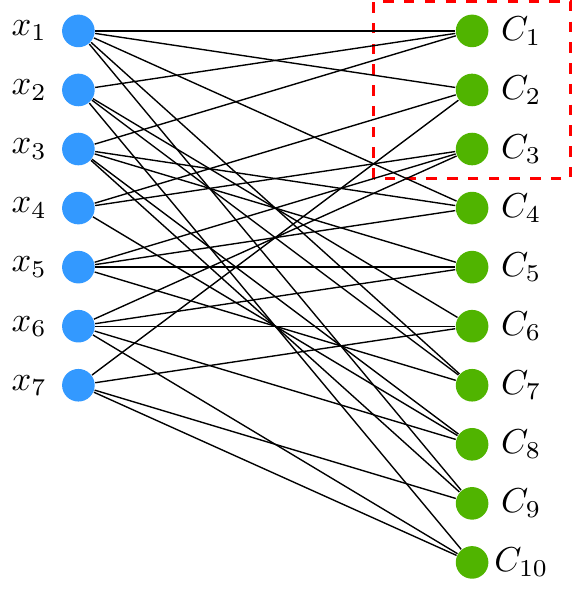}
    \caption{Cover number. For a \SAT instance with $n$ variables and $m$ clauses, we can draw a bipartite graph which has $n$ nodes on the left and $m$ nodes on the right. Each node (blue) on the left corresponds to a variable $x_i$, each node (green) on the right corresponds to a clause $C_j$. If either $x_i$ or $\ov{x}_i$ belongs to clause $C_j$, then we draw a line between these two nodes. Consider an input string $y\in \{0,1\}^7$. There exists some unsatisfied clauses with respect to this input string $y$. For for example, let $C_1$, $C_2$ and $C_3$ denote those unsatisfied clauses. We want to pick a smallest set of nodes on the left partition of the graph to guarantee that for each unsatisfied clause in the right partition, there exists a node on the left to cover it. The cover number is defined to be the smallest such number over all possible input strings.}
\end{figure}

Notice that for a satisfiable \SAT instance $S$, the cover number $p$ is $0$. Also, for any unsatisfiable \SAT instance $S$, the cover number $p$ is at least $1$. This is because for any input string, there exists at least one clause which cannot be satisfied. To fix that clause, we have to assign $*$ to a variable belonging to that clause. (Assigning $*$ to a variable can be regarded as assigning both $0$ and $1$ to a variable)

\begin{lemma}\label{lem:hardness_max_esatb_cover_number}
Let $S$ denote a \MAX-\ESATB instance with $n$ variables and $m$ clauses and $S$ suppose $S$ is at most $7/8+A$ satisfiable, where $A\in (0,1/8)$. Then the cover number of $S$ is at least $(1/8-A)m/B$.
\end{lemma}
\begin{proof}
For any input string $y\in \{0,1\}^n$, there exists at least $(1/8-A)m$ clauses which are not satisfied. Since each variable appears in at most $B$ clauses, we need to assign $*$ to at least $(1/8-A)m/B$ variables. Thus, the cover number of $S$ is at least $(1/8-A)m/B$.
\end{proof}

We say $x_1, x_2, \cdots, x_n$ are variables and $x_1,\ov{x}_1, x_2, \ov{x}_2, \cdots, x_n, \ov{x}_n$ are literals.

\begin{definition}
For a list of clauses $C$ and a set of variables $P$, if for each clause, there exists at least one literal such that the corresponding variable of that literal belongs to $P$, then we say $P$ covers $L$.
\end{definition}

\subsubsection{Properties of \SAT instances}

\begin{fact}
For any \SAT instance $S$ with $n$ variables and $m=\Theta(n)$ clauses, let $c>0$ denote a constant. If $S$ is $(1-c)m$ satisfiable, then let $y\in \{0,1\}^n$ denote a string for which $S$ has the smallest number of unsatisfiable clauses. Let $T$ denote the set of unsatisfiable clauses and let $b$ denote the number of variables in $T$. Then $ \Omega( (cm)^{1/3} )\leq b \leq O( cm )$.
\end{fact}
\begin{proof}
Note that in $S$, there is no duplicate clause. Let $T$ denote the set of unsatisfiable clauses by assigning string $y$ to $S$. First, we can show that any two literals $x_i, \overline{x}_i$ cannot belong to $T$ at the same time. If $x_i$ and $\overline{x}_i$ belong to the same clause, then that clause must be an ``always'' satisfiable clause. If $x_i$ and $\overline{x}_i$ belong to different clauses, then one of the clauses must be satisfiable. This contradicts the fact that that clause belongs to $T$. Thus, we can assume that literals $x_1, x_2, \cdots, x_b$ belong to $T$.

There are two extreme cases: one is that each clause only contains three literals and each literal appears in exactly one clause in $T$. Then $b = 3cm$. The other case is that each clause contains $3$ literals, and each literal appears in as many clauses as possible. Then ${b\choose 3 }=cm$, which gives $b = \Theta( (cm)^{1/3})$. 
\end{proof}

\begin{lemma}\label{lem:hardness_random_3sat_B_is_logn}
For a random \SAT instance, with probability $1-2^{-\Omega( \log n \log \log n)}$ there is no literal appearing in at least $\log n$ clauses.
\end{lemma}
\begin{proof}
By the property of random \SAT, for any literal $x$ and any clause $C$, the probability that $x$ appears in $C$ is $\frac{3}{2n}$, i.e., $\Pr[x \in C] = \frac{3}{2n} = \Theta(1/n)$. Let $p$ denote this probability. For any literal $x$, the probability of $x$ appearing in at least $\log n$ clauses (out of $m$ clauses) is
\begin{align*}
 & ~ \Pr[ ~x \text{~appearing~in} \geq \log n \text{~clauses~}] \\
= & ~ \sum_{i=\log n}^m {m \choose i} p^i (1-p)^{m-i} \\
= & ~ \sum_{i=\log n}^{m/2} {m \choose i} p^i (1-p)^{m-i} + \sum_{i=m/2}^m {m \choose i} p^i (1-p)^{m-i} \\
\leq & ~ \sum_{i=\log n}^{m/2} (e m/i)^i p^i  + \sum_{i=m/2}^m {m \choose i} p^i &\text{~by~} (1-p) \leq 1, {m\choose i} \leq (em/i)^i\\
\leq & ~ (\Theta(1/\log n))^{ \log n} + 2 \cdot (2e)^{m/2} \cdot \Theta(1/n)^{m/2} \\
\leq & ~ 2^{- \Omega( \log n \cdot \log\log n)}.
\end{align*}
Taking a union bound over all the literals, we complete the proof,
\begin{align*}
\Pr[~ \nexists ~ x \text{~appearing~in} \geq \log n \text{~clauses~} ] \geq 1-2^{-\Omega( \log n \log \log n )}.
\end{align*}
\end{proof}

\begin{lemma}
For a sufficiently large constant $c'>0$ and a constant $c>0$, for any random \SAT instance which has $n$ variables and $m=c'n$ clauses, suppose it is $(1-c)m$ satisfiable. Then with probability $1-2^{-\Omega(  \log n \log \log n)}$, for all input strings $y$, among the unsatisfied clauses, each literal appears in $O(\log n)$ places.
\end{lemma}
\begin{proof}
This follows by Lemma~\ref{lem:hardness_random_3sat_B_is_logn}.
\end{proof}

Next, we show how to reduce the $O(\log n)$ to $O(1)$.
\begin{lemma}\label{lem:hardness_random_3sat_B_is_constant}
For a sufficiently large constant $c$, for any random \SAT instance that has $n$ variables and $m=cn$ clauses, for any constant $B \geq 1,b\in (0,1)$, with probability at least $1-\frac{9m}{Bbn}$, there exist at least $(1-b)m$ clauses such that each variable (in these $(1-b)m$  clauses) only appears in at most $B$ clauses (out of these $(1-b)m$ clauses).
\end{lemma}
\begin{proof}
For each $i\in [m]$, we use $z_i$ to denote the indicator variable  such that it is $1$, if for each variable in the $i$th clause, it appears in at most $a$ clauses. Let $B \in [1,\infty)$ denote a sufficiently large constant, which we will decide upon later.

For each variable $x$, the probability of it appearing in the $i$-th clause is $\frac{3}{n}$. Then we have
\begin{align*}
\E[ \text{~\#~clauses~that~contain~}x ] = \sum_{i=1}^m \E[ i\text{-th~clause~contains~}x ] = \frac{3m}{n}
\end{align*}
By Markov's inequality,

\begin{align*}
\Pr[  \text{~\#~clauses~that~contain~}x \geq a ] \leq \E[ \text{~\#~clauses~that~contain~}x ] /B = \frac{3m}{B n}
\end{align*}

By a union bound, we can compute $\E[z_i]$ ,
\begin{align*}
\E[z_i]  = & ~\Pr[z_i=1] \\
\geq & ~ 1 - 3 \Pr[\text{~one~variable~in~}i\text{-th~clause~appearing} \geq B \text{~clauses~}] \\
\geq & ~ 1 - \frac{9m}{B n}.
\end{align*}
Furthermore, we have
\begin{align*}
\E[z] = \E[\sum_{i=1}^m z_i] = \sum_{i=1}^m \E[z_i] \geq (1-\frac{9m}{B n}) m.
\end{align*}
Note that $z\leq m$. Thus $\E[z]\leq m$. Let $b\in (0,1)$ denote a sufficiently small constant.
We can show
\begin{align*}
\Pr[ m - z \geq b m ]\leq & ~ \frac{\E[m-z]}{ bm}  \\
= & ~ \frac{m-\E[z]}{bm} \\
\leq & ~ \frac{m - (1-\frac{9m}{Bn})m }{bm} \\
= & ~ \frac{9m}{Bbn}.
\end{align*}
This implies that with probability at least $1- \frac{9m}{Bbn}$, we have $m-z \leq bm$. Notice that in random-\ETH, $m=cn$ for a constant $c$. Thus, by choosing a sufficiently large constant $B$ (which is a function of $c,b$), we can obtain arbitrarily large constant success probability.
\end{proof}

\subsubsection{Reduction}

We reduce \SAT to tensor rank by following the same construction in \cite{h90}. To obtain a stronger hardness result, we use the property that each variable only appears in at most $B$ (some constant) clauses and that the cover number of an unsatisfiable \SAT instance is large. Note that both \MAX-\ESATB instances and random-\ETH instances have that property. Also each \MAX-\ESATB is also a \SAT instance. Thus if the reduction holds for \SAT, it also holds for \MAX-\ESATB, and similarly for random-\ETH.

\begin{figure}[!t]
  \centering
    \includegraphics[width=0.4\textwidth]{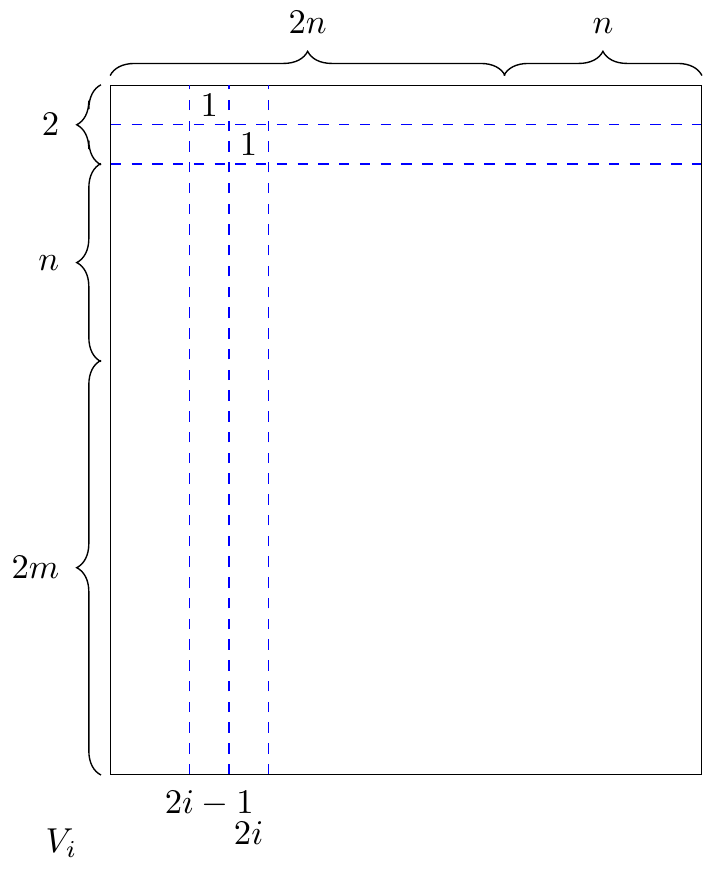}
    \includegraphics[width=0.4\textwidth]{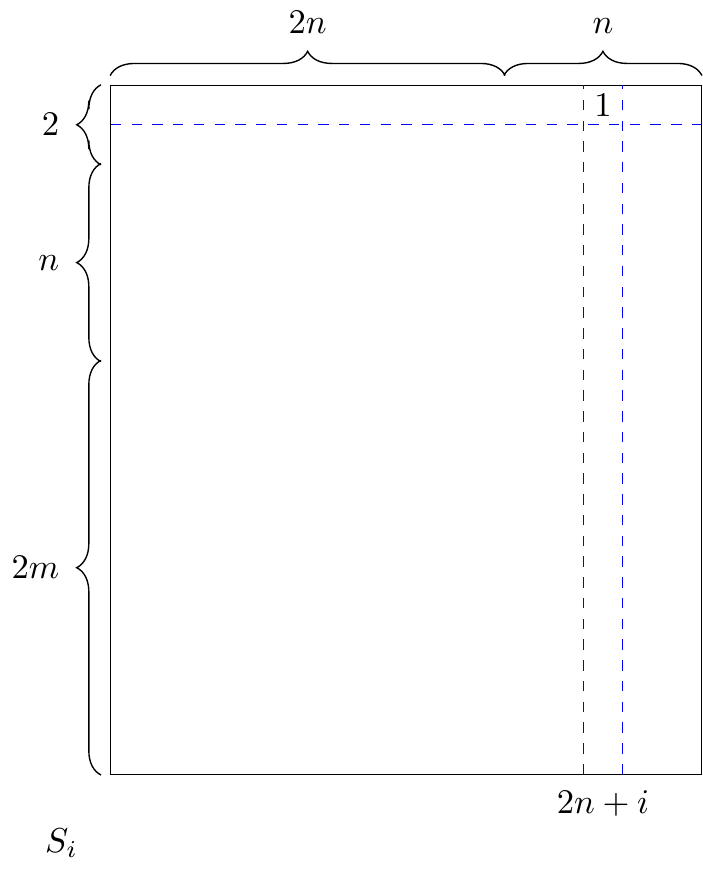}
    \\
    \includegraphics[width=0.4\textwidth]{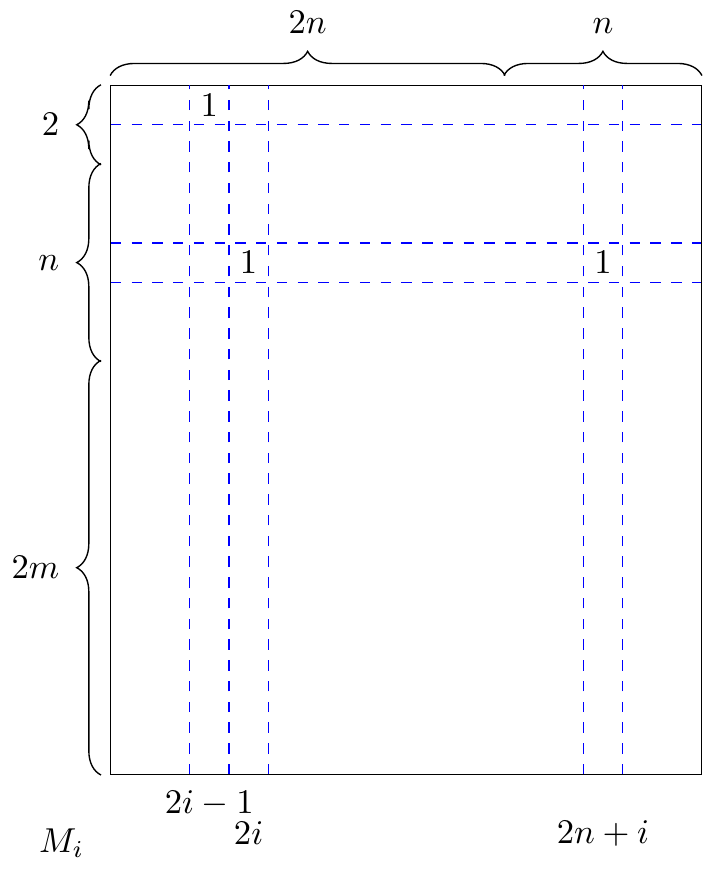}
    \includegraphics[width=0.4\textwidth]{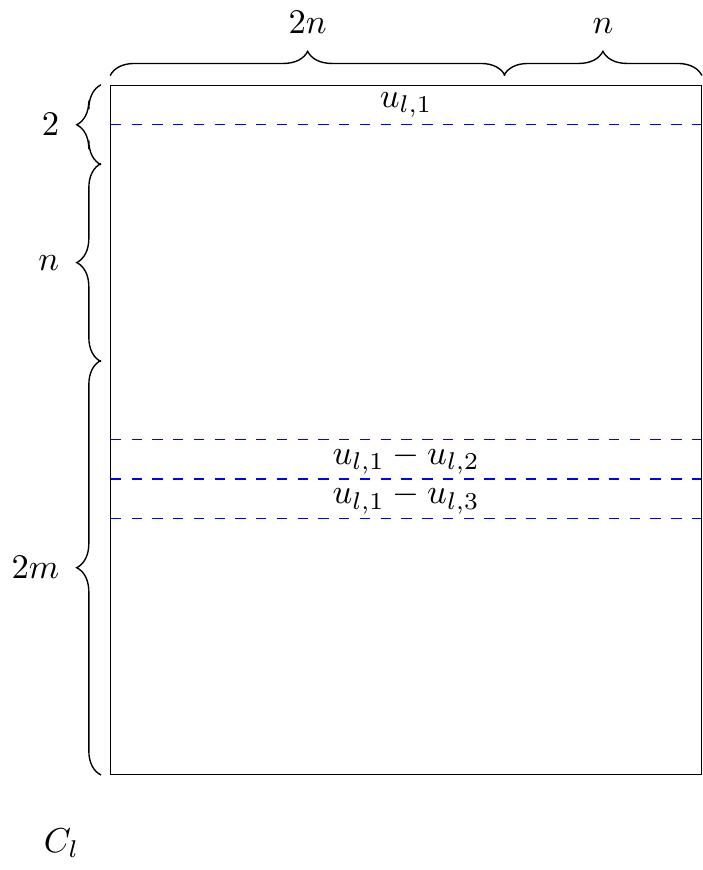}
    \caption{There are $3n+m$ column-row faces, $V_i, \forall i \in [n]$, $S_i, \forall i \in [n]$, $M_i, \forall i \in [n]$, $C_l, \forall l\in [m]$. In face $C_l$, each $u_{l,j}$ is either $x_i$ or $\ov{x}_i$ where $x_i = e_{2i-1}$ and $\ov{x}_{i} = e_{2i-1}+e_{2i}$.}\label{fig:hardness_hastad_construction_V_S_M_C}
\end{figure}

Recall the definition of \SAT: \SAT is the problem of given a Boolean formula of $n$ variables in \CNF form with at most
3 variables in each of the $m$ clauses, is it possible to find a satisfying assignment to the formula? We say $x_1,x_2,\cdots,x_n$ are variables and $x_1,\ov{x}_1, x_2, \ov{x}_2, \cdots, x_n, \ov{x}_n$ are literals. We transform this to the problem of computing the rank of a tensor of size $n_1 \times n_2 \times n_3$ where $n_1 = 2+n+2m$, $n_2=3n$ and $n_3 = 3n+m$. $T$ has the following $n_3$ column-row faces, where each of the faces is an $m_1\times n_2$ matrix,
\begin{itemize}
\item $n$ variable matrices $V_i \in \mathbb{R}^{ n_1 \times n_2 }$. It has a $1$ in positions $(1,2i-1)$ and $(2,2i)$ while all other elements are $0$.
\item $n$ help matrices $S_i \in \mathbb{R}^{ n_1 \times n_2 }$. It has a $1$ position in $(1,2n+i)$ and is $0$ otherwise.
\item $n$ help matrices $M_i \in \mathbb{R}^{ n_1 \times n_2 }$. It has a $1$ in positions $(1,2i-1), (2+i,2i)$ and $(2+i,2n+i)$ and is $0$ otherwise.
\item $m$ clause matrices $C_l \in \mathbb{R}^{ n_1 \times n_2 }$. 
Suppose the clause $c_l$ contains the literals $u_{l,1}, u_{l,2}$ and $u_{l,3}$. For each $j\in [3]$, $u_{l,j} \in \{x_1,x_2,\cdots, x_n, \ov{x}_1, \ov{x}_2,\cdots, \ov{x}_n\}$. Note that $x_i, \ov{x}_i$ are the literals of the \SAT formula. We can also think of $x_i, \ov{x}_i$ as length $3n$ vectors. Let $x_i$ denote the vector that has a $1$ in position $2i-1$, i.e., $x_i=e_{2i-1}$. Let $\ov{x}_i$ denote the vector that has a $1$ in positions $2i-1$ and $2i$, $\ov{x}_i = e_{2i-1}+e_{2i}$.
\begin{itemize}
\item Row $1$ is the vector $u_{l,1} \in \mathbb{R}^{3n}$,
\item Row $2+n + 2l-1$ is the vector $u_{l,1}-u_{l,2} \in \mathbb{R}^{3n}$,
\item Row $2+n + 2l$ is the vector $u_{l,1}-u_{l,3} \in \mathbb{R}^{3n}$.
\end{itemize}
\end{itemize}

First, we can obtain Lemma~\ref{lem:hardness_if_S_satisfied_tensor_rank_at_most_4n_2m} which follows by Lemma 2 in \cite{h90}. For completeness, we provide a proof.
\begin{lemma}\label{lem:hardness_if_S_satisfied_tensor_rank_at_most_4n_2m}
If the formula is satisfiable, then the constructed tensor has rank at most $4n+2m$.
\end{lemma}
\begin{proof}
We will construct $4n+2m$ rank-$1$ matrices $V_i^{(1)}, V_i^{(2)}$, $S_i^{(1)}$, $M_i^{(1)}$, $C_l^{(1)}$ and $C_l^{(2)}$ . Then the goal is to show that for each matrix in the set
\begin{align*}
\{ V_1, V_2,\cdots, V_n ,S_1, S_2,\cdots, S_n ,M_1, M_2, \cdots, M_n, C_1, C_2, \cdots, C_m\},
\end{align*}
 it can be written as a linear combination of these constructed matrices.
\begin{itemize}
\item Matrices $V_i^{(1)}$ and $V_i^{(2)}$. $V_i^{(1)}$ has the first row equal to $x_i$ iff $\alpha_i=1$ and otherwise $\ov{x}_i$. All the other rows are $0$. We set $V_i^{(2)} = V_i - V_i^{(1)}$.
\item Matrices $S_i^{(1)}$. $S_i^{(1)} = S_i$.
\item Matrices $M_i^{(1)}$. \begin{align*} M_i^{(1)}=\begin{cases} M_i - V_i^{(1)} & \text{~if~} \alpha_i =1 \\ M_i - V_i^{(1)} - S_i  & \text{~if~} \alpha_i =0  \end{cases}\end{align*}
\item Matrices $C_l^{(1)}$ and $C_l^{(2)}$. Let $x_i = \alpha_i$ be the assignment that makes the clause $c_l$ true. Then $C_l - V_i^{(1)}$ has rank $2$, since either it has just two nonzero rows (in the case where $x_i$ is the first variable in the clause) or it has three nonzero rows of which two are equal. In both cases we just need two additional rank $1$ matrices.
\end{itemize}
\end{proof}

\begin{figure}[!t]
  \centering
    \includegraphics[width=0.3\textwidth]{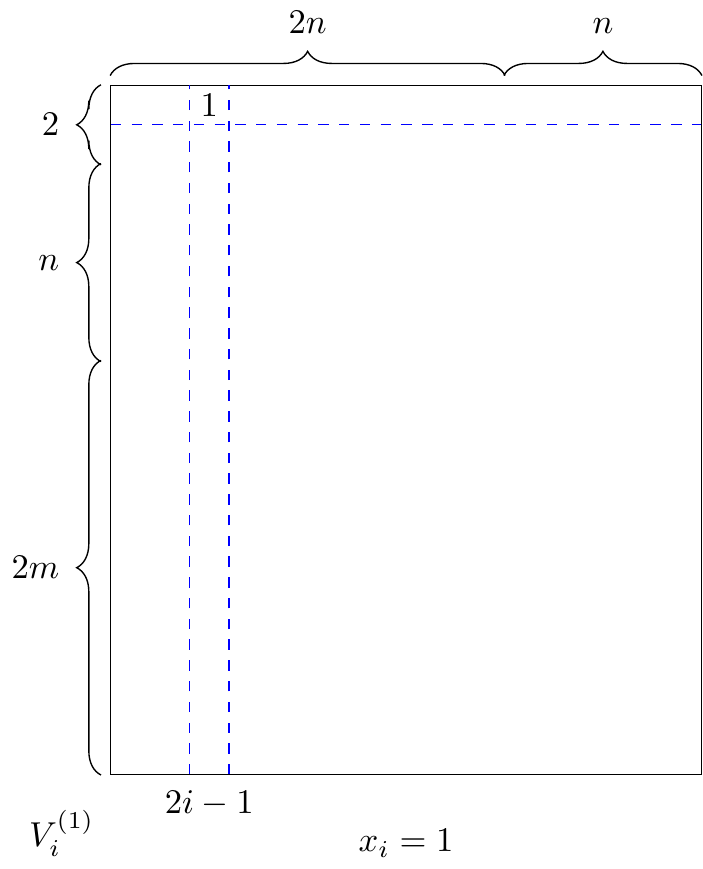}
    \includegraphics[width=0.3\textwidth]{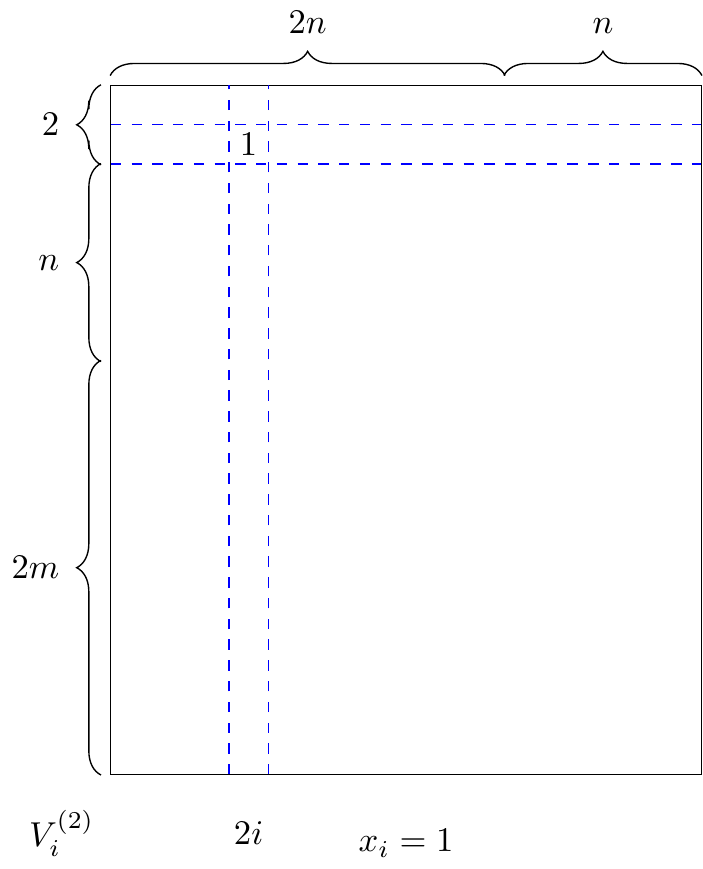}
    \includegraphics[width=0.3\textwidth]{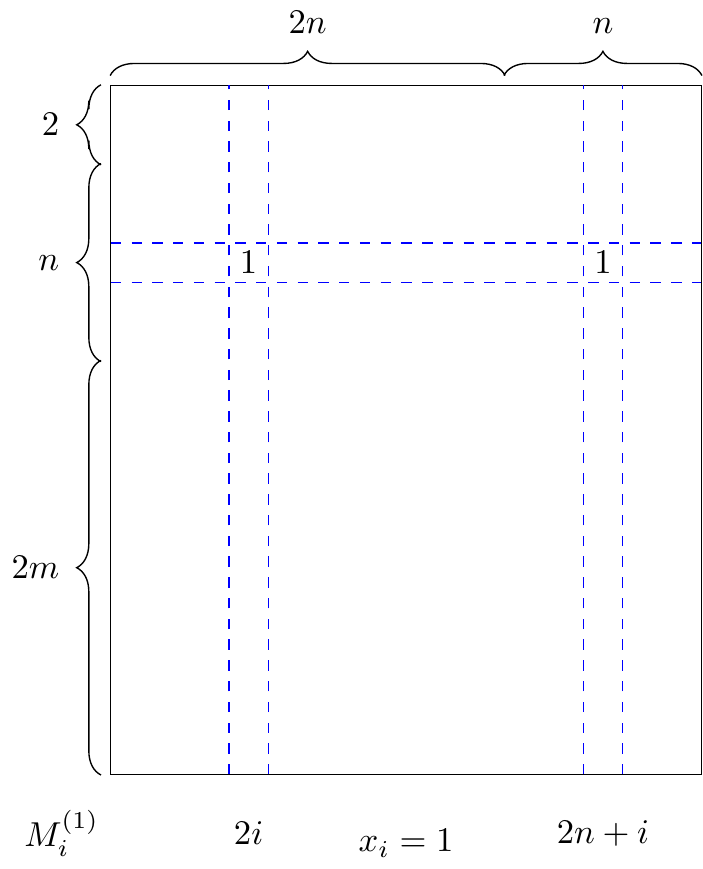}\\

    \includegraphics[width=0.3\textwidth]{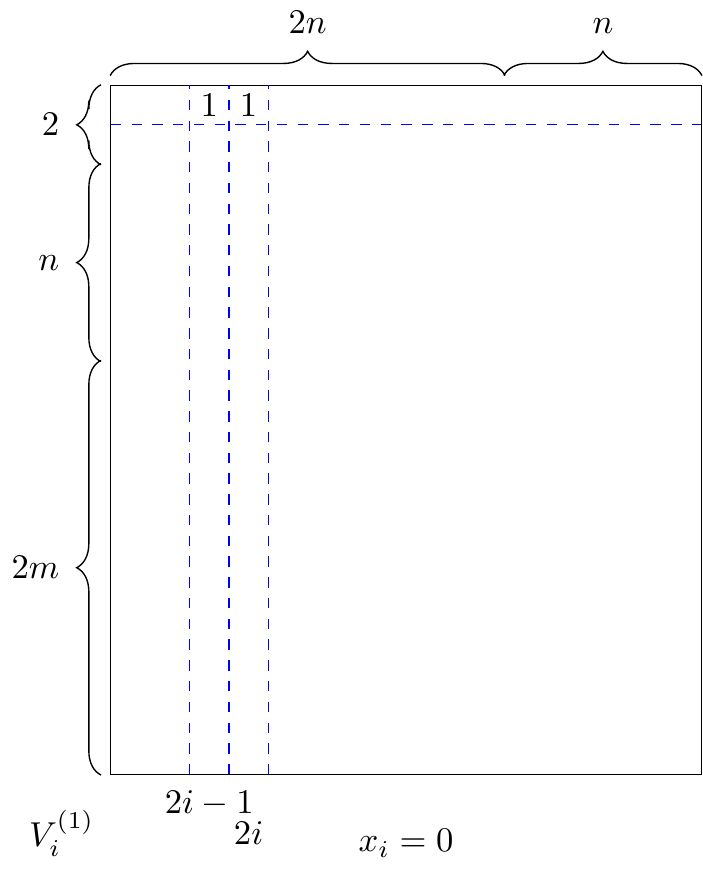}
    \includegraphics[width=0.3\textwidth]{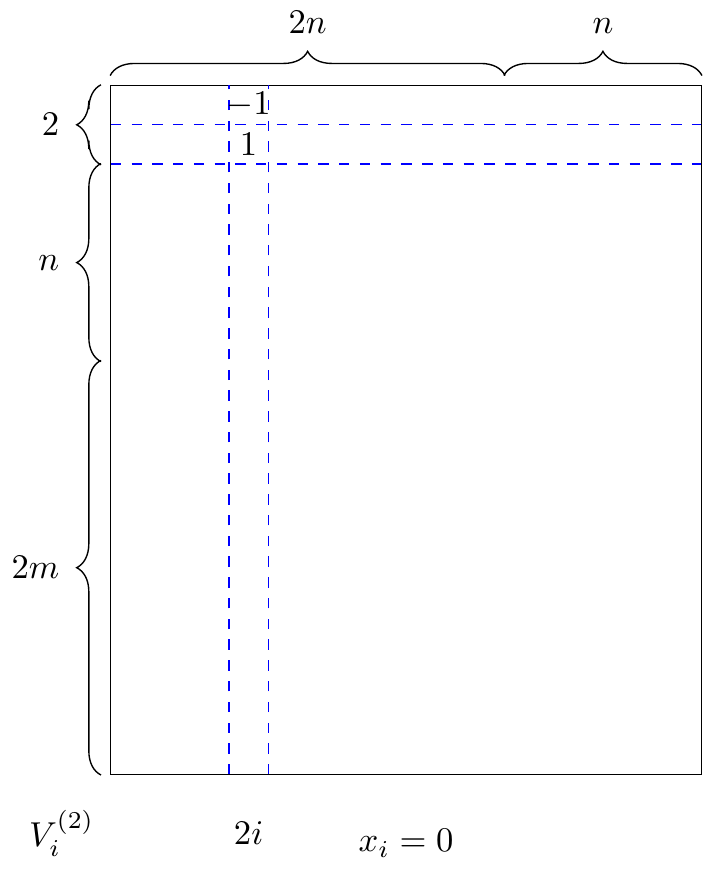}
    \includegraphics[width=0.3\textwidth]{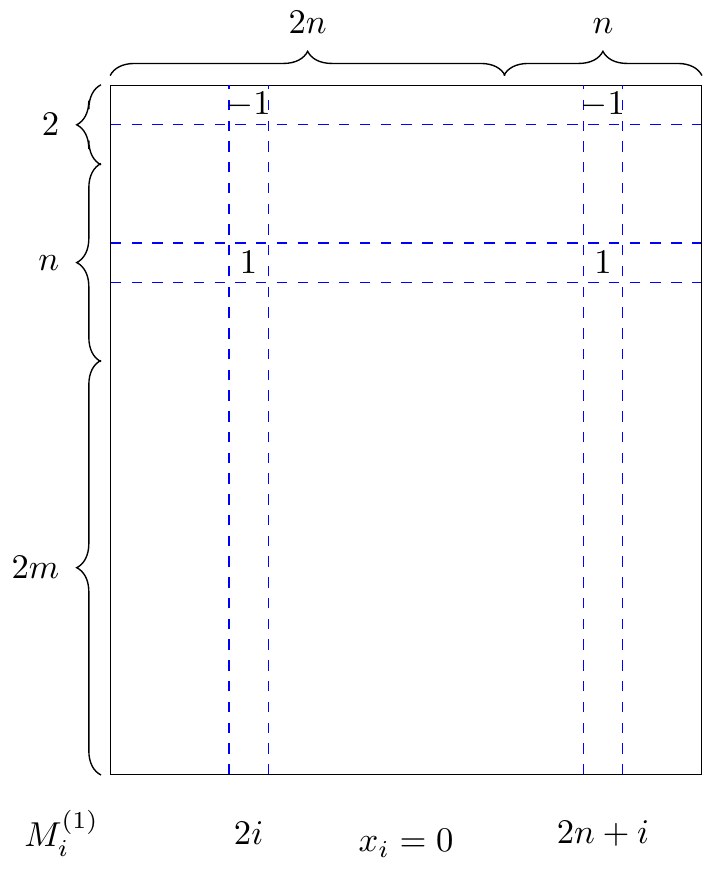}

    \caption{Two possibilities for $V_i^{(1)}, \forall i \in [n]$, $V^{(2)}, \forall i\in [n]$, $M_i^{(1)},\forall i\in [n]$.}
\end{figure}

Once the \SAT instance $S$ is unsatisfiable, then its cover number is at least $1$. For each unsatisfiable \SAT instance $S$ with cover number $p$, we can show that the constructed tensor has rank at most $4n+2m+O(p)$ and also has rank at least $4n+2m+\Omega(p)$. We first prove an upper bound,
\begin{lemma}\label{lem:tensor_rank_is_at_most_4n_2m_p}
For a \SAT instance $S$, let $y\in \{0,1\}$ denote a string such that $S(y)$ has a set $L$ that contains unsatisfiable clauses. Let $p$ denote the smallest number of variables that cover all clauses in $L$. Then the constructed tensor $T$ has rank at most $4n+2m+p$.
\end{lemma}
\begin{proof}
Let $y$ denote a length-$n$ Boolean string $(\alpha_1, \alpha_2,\cdots, \alpha_n)$. Based on the assignment $y$, all the clauses of $S$ can be split into two sets: $L$ contains all the unsatisfied clauses and $\overline{L}$ contains all the satisfied clauses. We use set $P$ to denote a set of variables that covers all the clauses in set $L$. Let $p=|P|$.
We will construct $4n+2m+p$ rank-$1$ matrices $V_i^{(1)}, V_i^{(2)}$,  $S_i^{(1)}$, $M_i^{(1)}$, $\forall i\in [n]$, $C_l^{(1)}$, $C_l^{(2)}$, $\forall l\in [m]$, and $V_j^{(3)}$, $\forall j \in P$. Then the goal is to show that the $V_i,S_i,M_i$ and $C_l$ can be written as linear combinations of these constructed matrices.
\begin{itemize}
\item Matrices $V_i^{(1)}$ and $V_i^{(2)}$. $V_i^{(1)}$ has first row equal to $x_i$ iff $\alpha_i=1$ and otherwise $\ov{x}_i$. All the other rows are $0$. We set $V_i^{(2)} = V_i - V_i^{(1)}$.
\item Matrices $V_j^{(3)}$. For each $j\in P$, $V_j^{(3)}$ has the first row equal to $x_i$ iff $\alpha_i=0$ and otherwise $\ov{x}_i$.
\item Matrices $S_i^{(1)}$. $S_i^{(1)} = S_i$.
\item Matrices $M_i^{(1)}$. \begin{align*} M_i^{(1)}=\begin{cases} M_i - V_i^{(1)} & \text{~if~} \alpha_i =1 \\ M_i - V_i^{(1)} - S_i  & \text{~if~} \alpha_i =0  \end{cases}\end{align*}
\item Matrices $C_l^{(1)}$ and $C_l^{(2)}$.
\begin{itemize}
\item For each $l\notin L$, clause $c_l$ is satisfied according to assignment $y$. Let $x_i = \alpha_i$ be the assignment that makes the clause $c_l$ true. Then $C_l - V_i^{(1)}$ has rank $2$, since either it has just two nonzero rows (in the case where $x_i$ is the first variables in the clause) or it has three nonzero rows of which two are equal. In both cases we just need two additional rank $1$ matrices.
\item For each $l\in L$. It means clause $c_l$ is unsatisfied according to assignment $y$. Let $x_{j_1} = \alpha_{j_1}$, $x_{j_2} = \alpha_{j_2}$, $x_{j_3} = \alpha_{j_3}$ be an assignment that makes the clause $c_l$ false. In other words, one of $j_1, j_2, j_3$ must be $P$ according to the definition that $P$ covers $L$. Then matrix $C_l - V_{j_1}^{(3)}$ has rank $2$, since either it has just two nonzero rows (in the case where $x_{j_1}$ is the first variables in the clause) or it has three nonzero rows of which two are equal. In both cases we just need two additional rank $1$ matrices.
\end{itemize}
\end{itemize}
We finish the proof by taking the $P$ that has the smallest size.
\end{proof}
Further, we have:
\begin{corollary}
For a \SAT instance $S$, let $p$ denote the cover number of $S$, then the constructed tensor $T$ has rank at most $4n+2m+p$.
\end{corollary}
\begin{proof}
This follows by applying Lemma~\ref{lem:tensor_rank_is_at_most_4n_2m_p} to all the input strings and the definition of cover number (Definition~\ref{def:hardness_cover_number}).
\end{proof}

\begin{figure}[!t]
  \centering
    \includegraphics[width=1.0\textwidth]{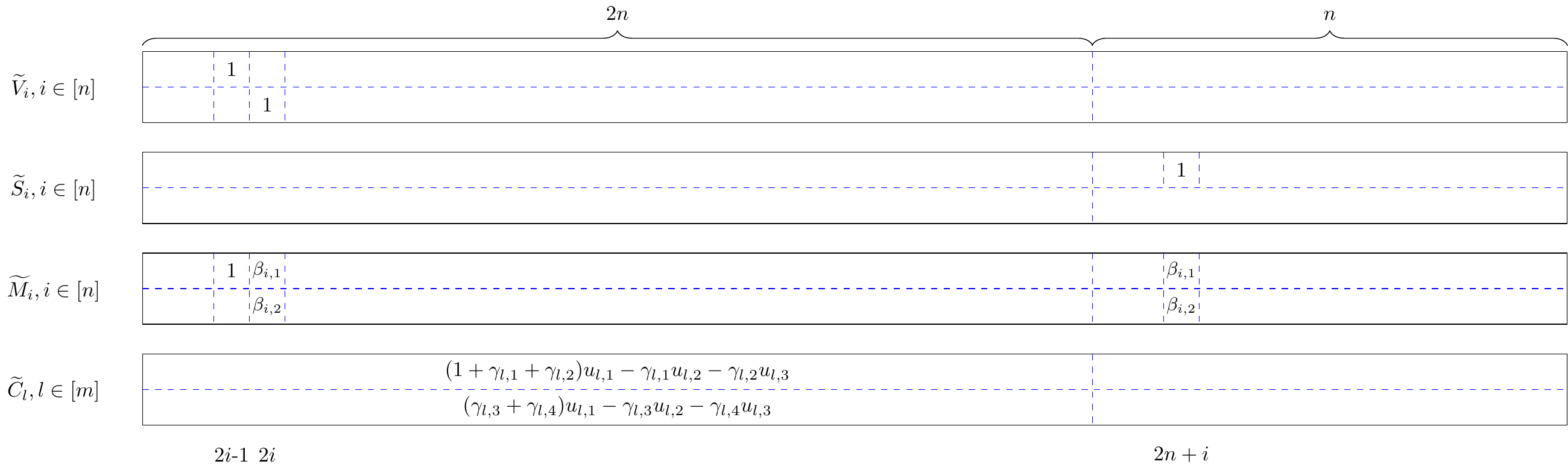}
    \caption{$\wt{V}_i$,$\wt{S}_i$,$\wt{M}_i$,$\wt{C}_l$.}
\end{figure}

We can split the tensor $T\in \mathbb{R}^{(2+n+3m) \times 3n \times (3n+m)}$ into two sub-tensors, one is $T_1 \in \mathbb{R}^{2 \times 3n \times (3n+m)}$ (that contains the first two row-tube faces of $T$ and linear combination of the remaining $2m$ row-tube faces of $T$), and the other is $T_2 \in \mathbb{R}^{(n+2m) \times 3n \times (3n+m)}$ (that contains the next $n+2m$ row-tube faces of $T$). We first analyze the rank of $T_1$ and then analyze the rank of $T_2$.
\begin{claim}
 The rank of $T_2$ is $n+2m$.
\end{claim}
\begin{proof}
According to Figure~\ref{fig:hardness_hastad_construction_V_S_M_C}, the nonzero rows are distributed in $n+m$ fully separated sub-tensors. It is obvious that the rank of each one of those $n$ sub-tensors is $1$, and the rank of each of those $m$ sub-tensors is $2$. Thus, overall, the rank $T_2$ is $n+2m$.
\end{proof}

To make sure $\rank(T) = \rank(T_1) +\rank(T_2)$, the $T_1\in  \mathbb{R}^{2 \times 3n \times (3n+m)}$ can be described as the following $3n+m$ column-row faces, and each of the faces is a $2\times 3n$ matrix.
\begin{itemize}
\item Matrices $\wt{V}_i, \forall i\in[n]$. The two rows are from the first two rows of $V_i$ in Figure~\ref{fig:hardness_hastad_construction_V_S_M_C}, i.e., the first row is $e_{2i-1}$ and the second row is $e_{2i}$.
\item Matrices $\wt{S}_i, \forall i\in[n]$. The two rows are from the first two rows of $S_i$ in Figure~\ref{fig:hardness_hastad_construction_V_S_M_C}, i.e., the first row is $e_{2n+i}$ and the second row is zero everywhere else.
\item Matrices $\wt{M}_i, \forall i\in[n]$. The first row is $e_{2i-1} + \beta_{i,1} (e_{2i}+e_{2n+i})$, while the second row is $\beta_{i,2} (e_{2i} + e_{2n+i})$.
\item Matrices $\wt{C}_l, \forall i\in[m]$. The first row is $(1+\gamma_{l,1}+\gamma_{l,2}) u_{l,1} - \gamma_{l,1} u_{l,2} - \gamma_{l,2} u_{l,3}$ and the second is $ ( \gamma_{l,3} + \gamma_{l,4} ) u_{l,1} - \gamma_{l,3} u_{l,2}- \gamma_{l,4} u_{l,3}$,
\end{itemize}
where for each $i\in [3n]$, we use vector $e_i$ to denote a length $3n$ vector such that it only has a $1$ in position $i$ and $0$ otherwise. $\beta,\gamma$ are variables. The goal is to show a lower bound for,
\begin{align*}
\underset{\beta,\gamma}{\rank} (T_1).
\end{align*}

\begin{figure}[!t]
  \centering
    \includegraphics[width=1.0\textwidth]{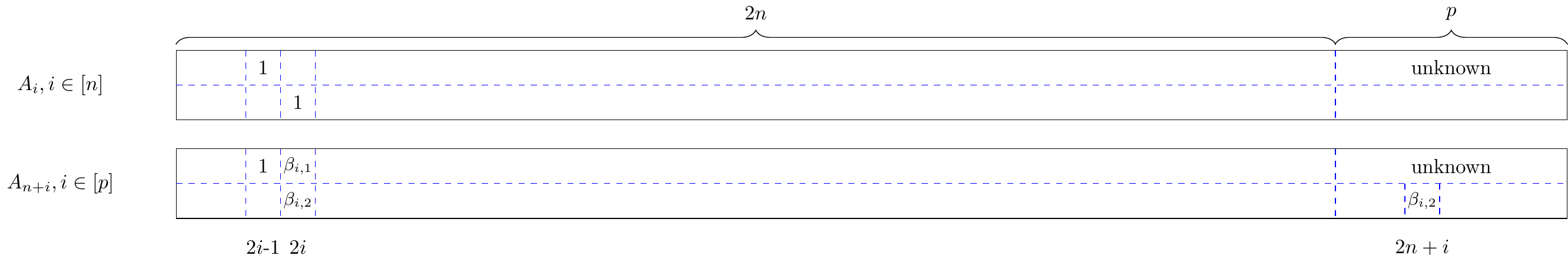}
    \includegraphics[width=0.6\textwidth]{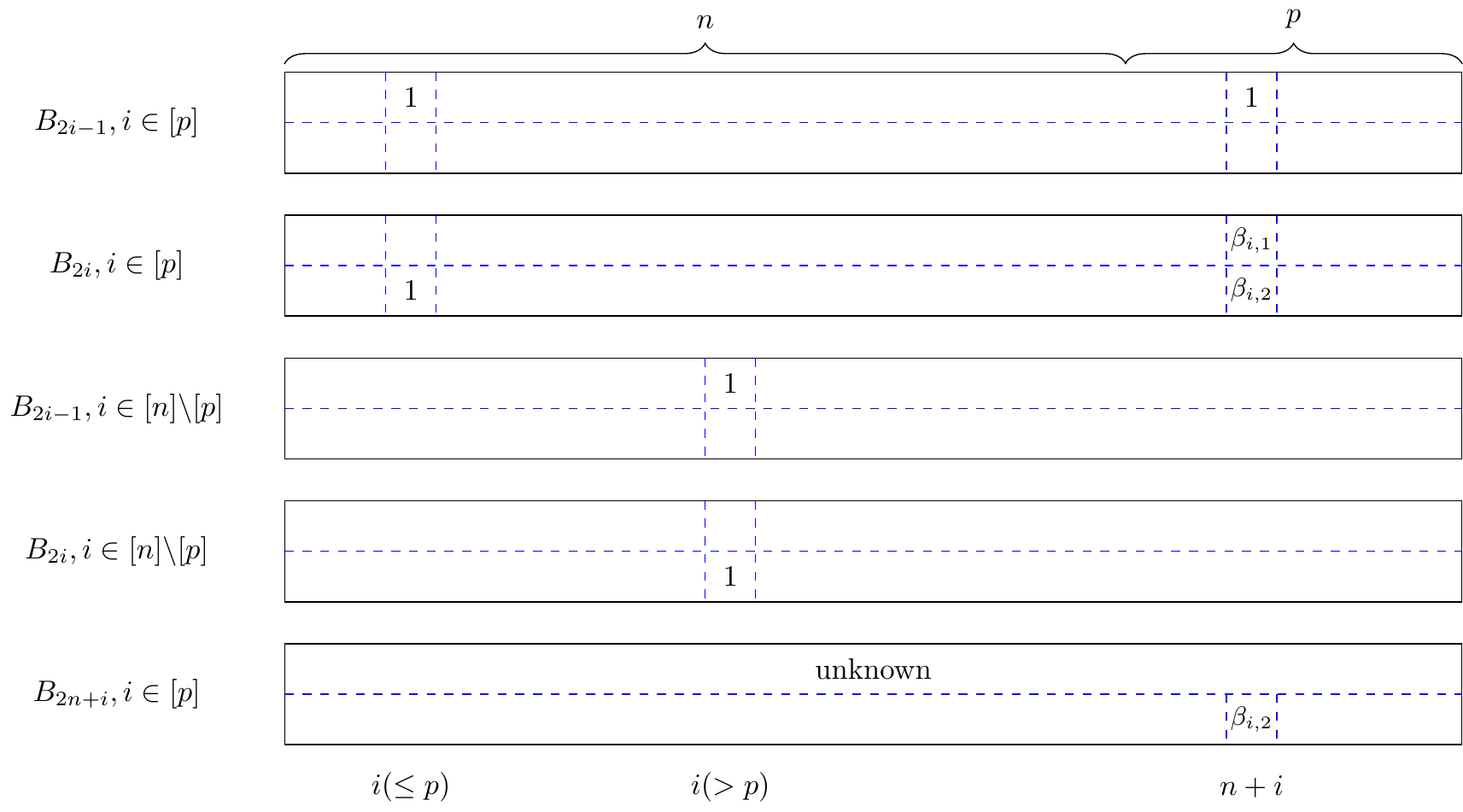}
    \caption{There are $n+p$ matrices $A_i \in \mathbb{R}^{2\times (2n+p)},\forall i\in [n+p]$ and $2n+p$ matrices $B_i\in \mathbb{R}^{2 \times (n+p)},\forall i\in [2n+p]$. Tensor $A$ and tensor $B$ represet the same tensor, and for each $i\in [n+p], j\in [2], l\in [2n+p]$, $(A_i)_{j,l} = (B_{l})_{j,i}$.}
\end{figure}

\begin{lemma}
Let $P$ denote the set $\{ i ~|~ \text{the~second~row~of~matrix~}\wt{M}_i\text{~is~nonzero},\forall i\in [n]\}$. Then the rank of $T_1$ is at least $3n+|P|$.
\end{lemma}
\begin{proof}
We define $p = |P|$. Without loss of generality, we assume that for each $i\in [p]$, the second row of matrix $\wt{M}_i$ is nonzero.

Notice that matrices $\wt{V}_i$, $\wt{S}_i$, $\wt{M}_i$ have size $2\times 3n$, but we only focus on the first $2n+p$ columns. Thus, we have $n+p$ column-row faces (from the $3$rd dimension) $A_j \in \mathbb{R}^{2 \times (2n+p)}$,
\begin{itemize}
\item $A_j$, $1 \leq j \leq n$, $A_j$ is the first $2n+p$ columns of $\wt{V}_j - \sum_{i=1}^n \alpha_{i,j} \wt{S}_i \in \mathbb{R}^{2 \times 3n}$, where $\alpha_{i,j}$ are some coefficients.
\item $A_{n+j}$, $1\leq j \leq p$, $A_j$ is the first $2n+p$ columns of $\wt{M}_j - \sum_{i=1}^n \alpha_{i,n+j} \wt{S}_i \in \mathbb{R}^{2\times 3n}$, where $\alpha_{i,j}$ are some coefficients.
\end{itemize}

Consider the first $2n+p$ column-tube faces (from $2$nd dimension), $B_j$, $\forall j\in [2n+p]$, of $T_1$. Notice that these matrices have size $2\times (n+p)$.
\begin{itemize}
\item $B_{2i-1}$, $1\leq i\leq p$, it has a $1$ in positions $(1,i)$ and $(1,n+i)$.
\item $B_{2i}$, $1\leq i \leq p$, it has $\beta_{i,1}$ in position $(1,n+i)$, $1$ in position $(2,i)$ and $\beta_{i,2}$ in position $(2,n+i)$.
\item $B_{2i-1}$, $p+1\leq i\leq n$, it has $1$ in position $(1,i)$.
\item $B_{2i}$, $p+1\leq i\leq n$, it has $1$ in position $(2,i)$.
\item $B_{2n+i}$, $1 \leq i\leq p$, the first row is unknown, the second row has $\beta_{i,2}$ in position in $(2,n+i)$.
\end{itemize}
It is obvious that the first $2n$ matrices are linearly independent, thus the rank is at least $2n$. We choose the first $2n$ matrices as our basis. For $B_{2n+1}$, we try to write it as a linear combination of the first $2n$ matrices $\{ B_{i}\}_{i\in [2n]}$. Consider the second row of $B_{2n+1}$. The first $n$ positions are all $0$. The matrices $B_{2i}$ all have disjoint support for the second row of the first $n$ columns. Thus, the matrices $B_{2i}$ should not be used. Consider the second row of $B_{2i-1}, \forall i\in [n]$. None of them has a nonzero value in position $n+1$. Thus $B_{2n+1}$ cannot be written as a linear combination of of the first $2n$ matrices. Thus, we can show for any $i\in [p]$, $B_{2n+i}$ cannot be written as a linear combination of matrices $\{ B_{i}\}_{i\in [2n]}$. Consider the $p$ matrices $\{B_{2n+i}\}_{i\in [p]}$. Each of them has a different nonzero position in the second row. Thus these matrices are all linearly independent. Putting it all together, we know that the rank of matrices $\{B_i\}_{i\in [2n+p]}$ is at least $2n+p$.
\end{proof}

Next, we consider another special case when $\beta_{i,2}=0$, for all $i\in [n]$. If we subtract $\beta_{i,1}$ times $\wt{S}_i$ from $\wt{M}_i$ and leave the other column-row faces (from the $3$rd dimension) as they are, and we make all column-tube faces(from the $2$nd dimension) for $j>2n$ identically $0$, then all other choices do not change the first $2n$ column-tube faces (from the $2$nd dimension) and make some other column-tube faces (from the $2$nd dimension) nonzero. Such a choice could clearly only increase the rank of $T$. Thus, we obtain,
\begin{align*}
\rank(T) = 2n + 2m + \min \rank( T_3),
\end{align*}
where $T_3$ is a tensor of size $2\times 2n \times (2n+m)$ given by the following column-row faces (from $3$rd dimension) $A_{i}, \forall i\in [2n+m]$ and each matrix has size $2 \times 2n$ (shown in Figure~\ref{fig:hastad_T3}).
\begin{itemize}
\item $A_{i}$, $i \in [n]$, the first $2n$ columns of $\wt{V}_i$.
\item $A_{n+i}$, $i \in [n]$, the first $2n$ columns of $\wt{M}_i$. The first row is $e_{2i-1} + \beta_{i,1} e_{2i}$, and the second row is $0$.
\item $A_{2n+l}$, $l \in [m]$, the first $2n$ columns of $\wt{C}_l$. The first row is $(1+\gamma_{l,1} + \gamma_{l,2}) u_{l,1} - \gamma_{l,1} u_{l,2} - \gamma_{l,2} u_{l,3}$, and the second row is $(\gamma_{l,3}+\gamma_{l,4}) u_{l,1} - \gamma_{l,3} u_{l,2} - \gamma_{l,4} u_{l,3}$.
\end{itemize}
We can show
\begin{lemma}
Let $p$ denote the cover number of the \SAT instance. $T_3$ has rank at least $2n+\Omega(p)$.
\end{lemma}
\begin{proof}
First, we can show that all matrices $A_{n+i}-A_i$ and $A_{n+i}$ (for all $i\in [n]$ ) are in the expansion of tensor $T_3$. Thus, the rank of $T_3$ is at least $2n$.

We need the following claim:
\begin{claim}\label{cla:hastad_A_2n_l}
For any $l\in [m]$, if $A_{2n+l}$ can be written as a linear combination of $\{ A_{n+i}-A_i \}_{i\in [n]}$ and $\{ A_{n+i} \}_{i\in [n]}$, then the second row of $A_{2n+l}$ is 0, and the first row of one of the $A_{n+i}$ is $u_i$ where $u_i$ is one of the literals appearing in clause $c_l$.
\end{claim}
\begin{proof}
We prove this for the second row first. For each $l\in [m]$, we consider the possibility of using all matrices $A_{n+i}-A_i$ and $A_{n+i}$ to express matrix $A_{2n+l}$. If the second row of $A_{2n+l}$ is nonzero, then it must have a nonzero entry in an odd position. But there is no nonzero in an odd position of the second row of any of matrices $A_{n+i}-A_i$ and $A_{n+i}$. 

For the first row. It is obvious that the first row of $A_{2n+l}$ must have at least one nonzero position, for any $\gamma_{l,1}, \gamma_{l,2}$. Let $u_j$ be a literal belonging to the variable $x_i$ which appears in the first row of $A_{2n+l}$ with a nonzero coefficient. Since only $A_{n+i}$ of all the other $A_{n+s}, \forall s\in [n]$ matrices has nonzero elements in either of the positions $(1,2i-1)$ or $(1,2i)$, then $A_{n+i}$ must be used to cancel these elements. Thus, the first row of $A_{n+i}$ must be a multiple of $u_j$ and since the element in position $(1,2i-1)$ of $A_{n+i}$ is $1$, this multiple must be $1$.

\end{proof}

\begin{figure}[!t]
  \centering
    \includegraphics[width=1.0\textwidth]{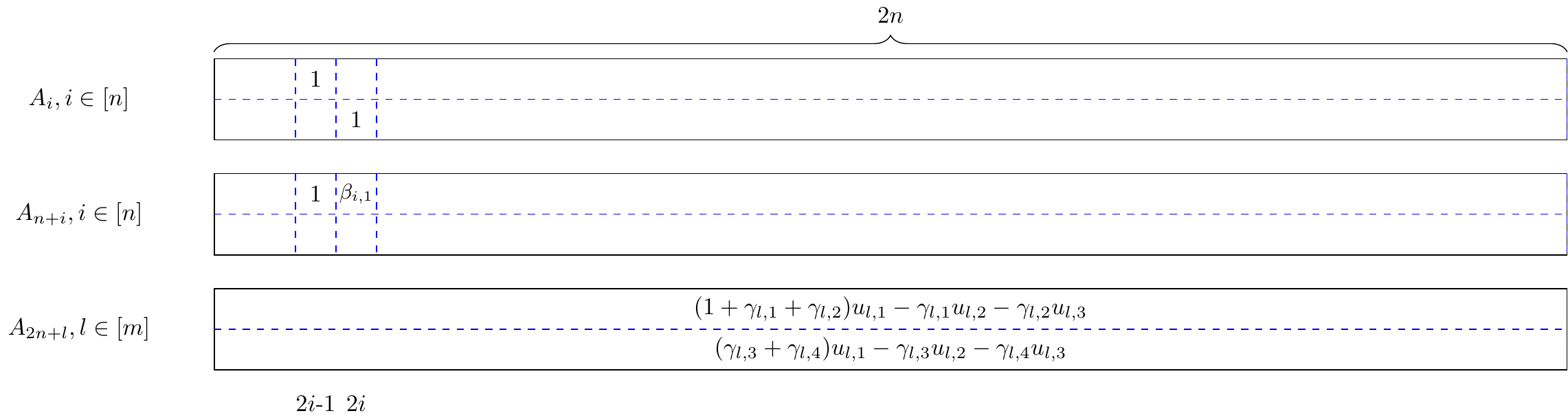}
    \caption{For any $i\in [n]$, $\beta_{i,1}\in \mathbb{R}$, for any $l\in [m]$, $\gamma_{l,1}, \gamma_{l,2} \in \mathbb{R}$, for any $l\in[m]$, if the first literal of clause $l$ is $x_j$, then row vector $u_{l,1} = e_{2i-1}\in \mathbb{R}^{2n}$; if the first literal of clause $l$ is $\overline{x}_j$, then row vector $u_{l,1} = e_{2i-1}+e_{2i}\in \mathbb{R}^{2n}$.}\label{fig:hastad_T3}
\end{figure}
Note that matrices $A_i,\forall i\in [n]$ have the property that, for any matrix in $\{ A_{n+1}, \cdots, A_{2n+m} \}$, it cannot be written as the linear combination of matrices $A_i,\forall i\in [n]$.
Let $\wt{A}\in \mathbb{R}^{(n+m) \times 2n}$ denote a matrix that consists of the first rows of $\{ A_{n+1}, \cdots, A_{2n+m} \}$. According to the property of matrices $A_i,\forall i\in [n]$, and that the rank of a tensor is always greater than or equal to the rank of any sub-tensor, we know that
\begin{align*}
\rank(T_3) \geq n + \min \rank(\wt{A}).
\end{align*}

\begin{claim}\label{cla:hardness_clause_l_imply_row_n_l}
For a \SAT instance $S$, for any input string $y\in \{0,1\}^n$, set $\beta_{*,1}$ to be the entry-wise flipping of $y$, $\mathrm{(\RN{1})}$  if the clause $l$ is satisfied, then the $(n+l)$-th row of $\wt{A}\in \mathbb{R}^{(n+m)\times 2n}$ can be written as a linear combination of the first $n$ rows of $\wt{A}$. $\mathrm{(\RN{2})}$ if the clause $l$ is unsatisfied, then the $(n+l)$-th row of $\wt{A}$ cannot be written as a linear combination of the first $n$ rows of $\wt{A}$.
\end{claim}
\begin{proof}
Part (\RN{1}), consider a clause $l$ which is satisfied with input string $y$. Then there must exist a variable $x_i$ belonging to clause $l$ (either literal $x_i$ or literal $\ov{x}_i$) and one of the following holds: if $x_i$ belongs to clause $l$, then $\alpha_i = 1$; if $\ov{x}_i$ belongs to clause $l$, then $\alpha_i=0$. Suppose clause $l$ contains literal $x_i$. The other case can be proved in a similar way. We consider the $(n+l)$-th row. One of the following assignments $(0,0), (-1,0), (0,-1)$ to $\gamma_{l,1}, \gamma_{l,1}$ is going to set the $(n+l)$-th row of $\wt{A}$ to be vector $e_{2i-1}$. We consider the $i$-th row of $\wt{A}$. Since we set $\alpha_i=1$, then we set $\beta_{i,1}=0$, it follows that the $i$-th row of $A$ becomes $e_{2i-1}$. Therefore, the $(n+l)$-th row of $\wt{A}$ can be written as a linear combination of $\wt{A}$.

Part (\RN{2}), consider a clause $l$ which is unsatisfied with input string $y$. Suppose that clause contains three literals $x_{i_1}, x_{i_2}, x_{i_3}$ (the other seven possibilities can be proved in a similar way). Then for input string $y$, we have $\alpha_{i_1}=0$, $\alpha_{i_2}=0$ and $\alpha_{i_3}=0$, otherwise this clause $l$ is satisfied. Consider $i_1$-th row of $\wt{A}$. It becomes $e_{2i_1 -1} + e_{2i_1}$. Similarly for the $i_2$-th row and $i_3$-th row. Consider the $(n+l)$-th row. We can observe that all of positions $2i_1, 2i_2,2i_3$ must be $0$.
Any linear combination formed by the $i_1,i_2,i_3$-th row of $\wt{A}$ must have one nonzero in one of positions $2i_1, 2i_2,2i_3$. However, if we consider the $(n+l)$-th row of $\wt{A}$, one of the positions $2i_1, 2i_2,2i_3$ must be $0$. Also, the remaining $n-3$ of the first $n$ rows of $\wt{A}$ also have $0$ in positions  $2i_1, 2i_2,2i_3$. Thus, we can show that the $(n+l)$-th row of $\wt{A}$ cannot be written as a linear combination of the first $n$ rows. Similarly, for the other seven cases.
\end{proof}
Note that in order to make sure as many as possible rows in $n+1, \cdots, n+m$ can be written as linear combinations of the first $n$ rows of $\wt{A}$, the $\beta_{i,1}$ should be set to either $0$ or $1$. Also each possibility of input string $y$ is corresponding to a choice of $\beta_{i,1}$. According to the above Claim~\ref{cla:hardness_clause_l_imply_row_n_l}, let $l_0$ denote the smallest number of unsatisfied clauses over the choices of all the $2^n$ input strings. Then over all choices of $\beta,\gamma$, there must exist at least $l_0$ rows of $\wt{A}_{n+1}, \cdots \wt{A}_{n+m}$, such that each of those rows cannot be written as the linear combination of the first $n$ rows.
\begin{claim}
Let $\wt{A}\in \mathbb{R}^{(n+m) \times 2n}$ denote a matrix that consists of the first rows of ${A}_{n+i},\forall i\in[n]$ and ${A}_{n+l},\forall l\in [m]$. Let $p$ denote the cover number of \SAT instance. Then $\min \rank(\wt{A}) \geq n+ \Omega(p)$.
\end{claim}
\begin{proof}


 For any choices of $\{ \beta_{i,1}\}_{i\in [n]}$, there must exist a set of rows out of the next $m$ rows such that, each of those rows cannot be written as a linear combination of the first $n$ rows. Let $L$ denote the set of those rows. Let $t$ denote the maximum size set of disjoint rows from $L$. Since those $t$ rows in $L$ all have disjoint support, they are always linearly independent. Thus the rank is at least $n+t$.

 Note that each row corresponds to a unique clause and each clause corresponds to a unique row. We can just pick an arbitrary clause $l$ in $L$, then remove the clauses that are using the same literal as clause $l$ from $L$. Because each variable occurs in at most $B$ clauses, we only need to remove at most $3B$ clauses from $L$. We repeat the procedure until there is no clause $L$. The corresponding rows of all the clauses we picked have disjoint supports, thus we can show a lower bound for $t$,
 \begin{align*}
t \geq |L| / (3B) \geq l_0/(3B) \geq p / (9B) \gtrsim p,
 \end{align*}
where the second step follows by $|L|\geq l_0$, the third step follows $3l_0\geq p$, and the last step follows by $B$ is some constant.
\end{proof}
Thus, putting it all together, we complete the proof.
\end{proof}

\begin{figure}[!t]
\centering
    \includegraphics[width=1.0\textwidth]{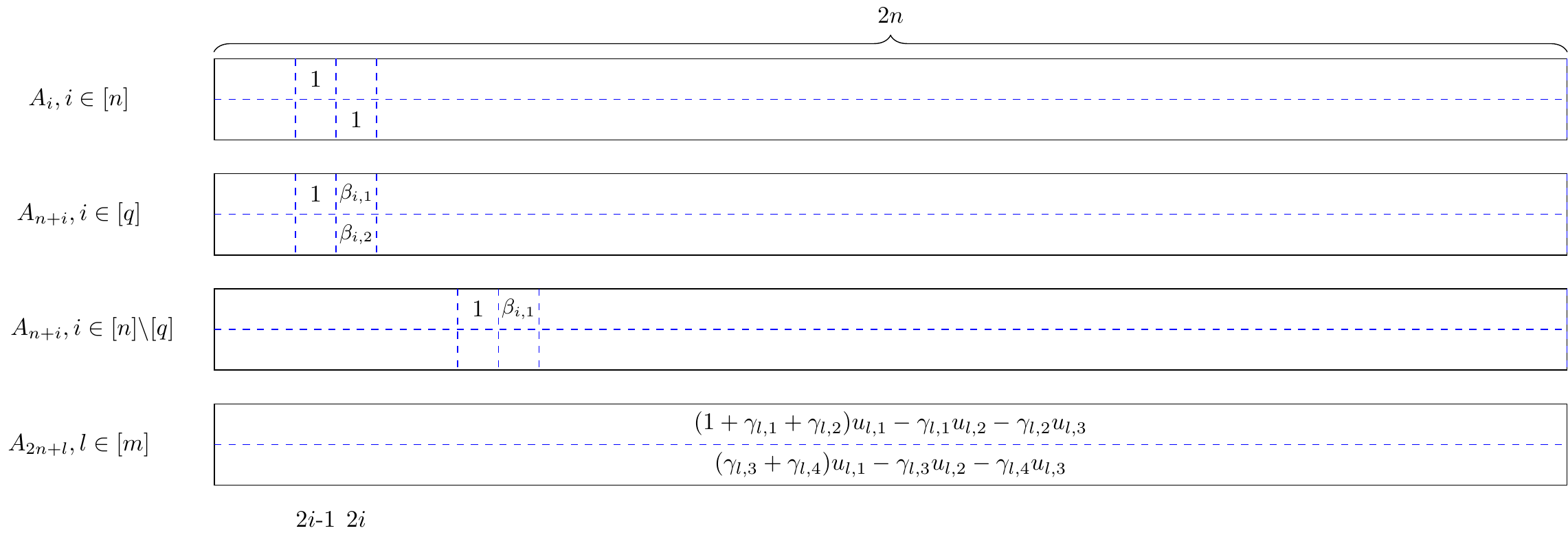}
    \caption{For any $i\in [n]$, $\beta_{i,1}\in \mathbb{R}$. For any $i\in [q]$, $\beta_{i,2}\in \mathbb{R}$. For any $l\in [m]$, $\gamma_{l,1}, \gamma_{l,2} \in \mathbb{R}$. For any $l\in[m]$, if the first literal of clause $l$ is $x_j$, then row vector $u_{l,1} = e_{2i-1}\in \mathbb{R}^{2n}$; if the first literal of clause $l$ is $\overline{x}_j$, then row vector $u_{l,1} = e_{2i-1}+e_{2i}\in \mathbb{R}^{2n}$.}\label{fig:hastad_T4}
\end{figure}

Now, we consider a general case when there are $q$ different $i\in [n]$ satisfying that $\beta_{i,2}\neq 0$. Similar to tensor $T_3$, we can obtain $T_4$ such that,
\begin{align*}
\rank(T) = 2n + 2m + \min \rank (T_4)
\end{align*}
where $T_4$ is a tensor of size $2\times 2n \times (2n+m)$ given by the following column-row faces (from $3$rd dimension) $A_i$, $\forall i\in [2n+m]$ and each matrix has size $2\times 2n$ (shown in Figure~\ref{fig:hastad_T4}).
\begin{itemize}
\item $A_{i}$, $i \in [n]$, the first $2n$ columns of $\wt{V}_i$.
\item $A_{n+i}$, $i \in [q]$, the first $2n$ columns of $\wt{M}_i$. The first row is $e_{2i-1} + \beta_{i,1} e_{2i}$, and the second row is $\beta_{i,2} e_{2i}$.
\item $A_{n+i}$, $i \in \{q+1, \cdots,n\}$, the first $2n$ columns of $\wt{M}_i$. The first row is $e_{2i-1} + \beta_{i,1} e_{2i}$, and the second row is $0$.
\item $A_{2n+l}$, $l \in [m]$, the first $2n$ columns of $\wt{C}_l$. The first row is $(1+\gamma_{l,1} + \gamma_{l,2}) u_{l,1} - \gamma_{l,1} u_{l,2} - \gamma_{l,2} u_{l,3}$, and the second row is $(\gamma_{l,3}+\gamma_{l,4}) u_{l,1} - \gamma_{l,3} u_{l,2} - \gamma_{l,4} u_{l,3}$.
\end{itemize}
Note that modifying $q$ entries(from Figure~\ref{fig:hastad_T3} to Figure~\ref{fig:hastad_T4}) of a tensor can only decrease the rank by $q$, thus we obtain
\begin{lemma}
Let $q$ denote the number of $i$ such that $\beta_{i,2}\neq 0$, and let $p$ denote the cover number of the \SAT instance. Then $T_4$ has rank at least $2n+\Omega(p)-q$.
\end{lemma}

Combining the two perspectives we have
\begin{lemma}
Let $p$ denote the cover number of an unsatisfiable \SAT instance. Then the tensor has rank at least $4n+2m+\Omega(p)$.
\end{lemma}
\begin{proof}
Let $q$ denote the $q$ in Figure~\ref{fig:hastad_T4}. From one perspective, we know that the tensor has rank at least $4n+2m+\Omega(p)-q$. From another perspective, we know that the tensor has rank at least $4n+2m+q$. Combining them together, we obtain the rank is at least $4n+2m+\Omega(p)/2$, which is still $4n+2m+\Omega(p)$.
\end{proof}

\begin{theorem}\label{thm:approximate_tensor_rank_is_eth_hard}
Unless \ETH fails, there is a $\delta>0$ and an absolute constant $c_0>1$ such that the following holds. For the problem of deciding if the rank of a $q$-th order tensor, $q\geq 3$, with each dimension $n$, is at most $k$ or at least $c_0k$, there is no $2^{\delta k^{1-o(1)}}$ time algorithm. 
\end{theorem}

\begin{proof}
The reduction can be split into three parts.\footnote{The first two parts are accomplished by personal communication with Dana Moshkovitz and Govind Ramnarayan.} The first part reduces the \MAX-\SAT problem to the \MAX-\ESAT problem by \cite{mr10}. For each \MAX-\SAT instance with size $n$, the corresponding \MAX-\ESAT instance has size $n^{1+o(1)}$. The second part is by reducing the \MAX-\ESAT problem to \MAX-\ESATB by \cite{t01}. For each \MAX-\ESAT instance with size $n$, the corresponding \MAX-\ESATB instance has size $\Theta(n)$ when $B$ is a constant. The third part is by reducing the \MAX-\ESATB problem to the tensor problem. Combining Theorem~\ref{thm:hardness_max_esatb}, Lemma~\ref{lem:hardness_max_esatb_cover_number} with this reduction, we complete the proof.
\end{proof}

\begin{theorem}\label{thm:approximate_tensor_rank_is_random_eth_hard}
Unless random-\ETH fails, there is an absolute constant $c_0>1$ for which any deterministic algorithm for deciding if the rank of a $q$-th order tensor is at most $k$ or at least $c_0k$, requires $2^{\Omega(k)}$ time.
\end{theorem}
\begin{proof}
This follows by combining the reduction with random-\ETH and Lemma~\ref{lem:hardness_random_3sat_B_is_constant}.
\end{proof}
Note that, if {\bf BPP} = {\bf P} then it also holds for randomized algorithms which succeed with probability $2/3$.

Indeed, we know that any deterministic algorithm requires $2^{\Omega(n)}$ running time on tensors that have size $n\times n\times n$. Let $g(n)$ denote a fixed function of $n$, and $g(n)=o(n)$. We change the original tensor from size $n\times n\times n$ to $2^{g(n)} \times 2^{g(n)} \times 2^{g(n)}$ by adding zero entries. Then the number of entries in the new tensor is $2^{3 g(n)}$ and the deterministic algorithm still requires $2^{\Omega(n)}$ running time on this new tensor. Assume there is a randomized algorithm that runs in $2^{c g(n)}$ time, for some constant $c>3$. Then considering the size of this new tensor, the deterministic algorithm is a super-polynomial time algorithm, but the randomized algorithm is a polynomial time algorithm. Thus, by assuming {\bf BPP} = {\bf P}, we can rule out randomized algorithms, which means Theorem~\ref{thm:approximate_tensor_rank_is_random_eth_hard} also holds for randomized algorithms which succeed with probability $2/3$.

We provide some some motivation for the {\bf BPP} = {\bf P} assumption: this is a standard conjecture in complexity theory, as it is implied by the existence of strong pseudorandom generators or if any problem in deterministic exponential time has exponential size circuits \cite{iw97}.

\subsection{Hardness result for robust subspace approximation}\label{sec:hardness_robust_subspace_approximation}

This section improves the previous hardness for subspace approximation \cite{cw15focs} from $1\pm 1/\poly(d)$  to $1\pm 1/\poly(\log d)$. (Note that, we provide the algorithmic results for this problem in Section~\ref{sec:lvu}.)

\begin{lemma}[\cite{d14}]
For any graph $G$ with $n$ nodes, $m$ edges, for which the maximum degree in graph $G$ is $d$, there exists a $d$-regular graph $G'$ with $2nd-2m$ nodes such that the clique size of $G'$ is the same as the clique size of $G$.
\end{lemma}
\begin{proof}
First we create $d$ copies of the original graph $G$. For each $i\in[n]$, let $v_{i,1}, v_{i,2},\cdots, v_{i,d}$ denote the set of nodes in $G'$ that are corresponding to $v_i$ in $G$. Let $d_{v_i}$ denote the degree of node $v_i$ in graph $G$. In graph $G'$, we create $d-d_{v_i}$ new nodes $v_{i,1}', v_{i,2}', \cdots, v'_{i,d_{v_i}}$ and connect each of them to all of the $v_1, v_2, \cdots, v_d$. Therefore, 1. For each $i\in [n], j\in [d_{v_i}]$, node $v_{i,j}'$ has degree $d$. 2. For each $i\in [n], j\in [d]$, node $v_{i,j}$ has degree $d_{v_i}$ (from the original graph), and $d-d_{v_i}$ degree (from the edges to all the $v_{i,1}', v_{i,2}', \cdots, v'_{i,d_{v_i}}$). Thus, we proved the graph $G$ is $d$-regular.

The number of nodes in the new graph $G'$ is,
\begin{align*}
nd + \sum_{i=1}^n (d - d_{v_i}) = 2nd - \sum_{i=1}^n d_{v_i} = 2nd -2m.
\end{align*}

It remains to show the clique size is the same in graph $G$ and $G'$.
Since we can always reorder the indices for all the nodes, without loss of generality, let us assume the the first $k$ nodes $v_1, v_2,\cdots,v_k$ forms a $k$-clique that has the largest size. It is obvious that the clique size $k'$ in graph $G'$ is at least $k$, since we make $k$ copies of the original graph and do not delete any edges and nodes. Then we just need to show $k' \leq k$. By the property of the construction, the node in one copy does not connect to a node in any other copy. Consider the new nodes we created. For each node $v_{i,j}'$, consider the neighbors of this node. None of them share a edge. Combining the above two properties gives $k'\leq k$. Thus, we finish the proof.
\end{proof}

\begin{figure}[!t]
  \centering
    \includegraphics[width=0.5\textwidth]{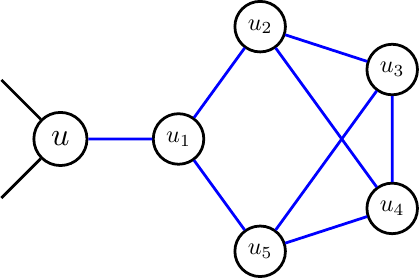}
    \caption{In the original graph $G$, vertex $u$ has degree $2$. We create $5$ new ``artificial'' vertices for $u$ to guarantee that the new graph $G'$ is $3$-regular. This construction was suggested to us by Syed Mohammad Meesum.}
\end{figure}

\begin{theorem}[Theorem 2.6 in~\cite{gjs76}]\label{thm:reduction_form_3sat_to_independent_set}
Any $n$ variable $m$ clauses \SAT instance can be reduced to a graph $G$ with $24m$ vertices,  which is an instance of $10m$-independent set. Furthermore $G$ is a $3$-regular graph.
\end{theorem}
We give the proof for completeness here.
\begin{proof}
Define $o_i$ to be the number of occurrences of $\{x_i,\ov{x}_i\}$ in the $m$ clauses. For each variable $x_i$, we construct $2o_i$ vertices, namely $v_{i,1},v_{i,2},\cdots,v_{i,2o_i}$. We make these $2o_i$ vertices be a circuit, i.e., there are $2o_i$ edges: $(v_{i,1},v_{i,2}),(v_{i,2},v_{i,3}),\cdots,(v_{i,2o_i-1},v_{i,2o_i}),(v_{i,2o_i},v_{i,1}).$ For each clause with $3$ literals $a,b,c$, we create $3$ vertices $v_a,v_b,v_c$ where they form a triangle, i.e., there are edges $(v_a,v_b),(v_b,v_c),(v_c,v_a).$ Furthermore, assume $a$ is the $j^{\text{th}}$ occurrence of $x_i$ (occurrence of $x_i$ means $a=x_i$ or $a=\overline{x}_i$). Then if $a=x_i$, we add edge $(v_a,v_{i,2j})$, otherwise we add edge $(v_a,v_{i,2j-1})$.

Thus, we can see that every vertex in the triangle corresponding to a clause has degree $3$, half of vertices of the circuit corresponding to variable $x_i$ have degree $3$ and the other half have degree $2$. Notice that the maximum independent set of a $2o_i$ circuit is at most $o_i$, and the maximum independent set of a triangle is at most $1$. Thus, the maximum independent set of the whole graph has size at most $m+\sum_{i=1}^n o_i=m+3m=4m.$ Another observation is that if there is a satisfiable assignment for the \SAT instance, then we can choose a $4m$-independent set in the following way: if $x_i$ is true, then we choose all the vertices in set $\{ v_{i,1}, v_{i,3}, \cdots, v_{i,2j-1}, \cdots v_{i,2o_i-1} \}$; otherwise, we choose all the vertices in set $\{ v_{i,2}, v_{i,4}, \cdots, v_{i,2j}, \cdots v_{i,2o_i} \}$. For a clause with literals $a,b,c$: if $a$ is satisfied, it means that $v_{i,t}$ which connected to $v_{a}$ is not chosen in the independent set, thus we can pick $v_a.$

The issue remaining is to reduce the above graph to a $3$ regular graph. Notice that there are exactly $\sum_{i=1}^n o_i=3m$ vertices which have degree $2$. For each of this kind of vertex $u$, we construct $5$ additional vertices $u_1,u_2,u_3,u_4,u_5$ and edges $(u_1,u_2),(u_2,u_3),(u_3,u_4),(u_4,u_5),(u_5,u_1),(u_2,u_4),(u_3,u_5)$ and $(u_1,u)$. Because we can always choose exactly two vertices among $u_1,u_2,\cdots,u_5$ no matter we choose vertex $u$ or not, the value of the maximum independent set will increase the size by exactly $2\sum_{i=1}^n o_i=6m$.

To conclude, we construct a $3$-regular graph reduced from a \SAT instance. The graph has exactly $24m$ vertices. Furthermore, if the \SAT instance is satisfiable, the graph has $10m$-independent set. Otherwise, it does not have a $10m$-independent set.
\end{proof}

\begin{corollary}\label{cor:hardness_kclique}
There is a constant $0<c<1$, such that for any $\varepsilon>0$, there is no $O(2^{n^{1-\varepsilon}})$ time algorithm which can solve $k$-clique for an $n$-vertex $(n-3)$-regular graph where $k=cn$ unless \ETH fails.
\end{corollary}

\begin{proof}
According to Theorem~\ref{thm:reduction_form_3sat_to_independent_set}, for a given $n$ variable $m=O(n)$ clauses \SAT instance, we can reduce it to a $3$-regular graph with $24m$ vertices which is a $10m$-independent set instance. If there exists $\varepsilon>0$ such that we have an algorithm with running time $O(2^{(24m)^{1-\varepsilon}})$ which can solve $10m$-clique for a $24m-3$ regular graph with $24m$ vertices, then we can solve the \SAT problem in $O(2^{n^{1-\varepsilon'}})$ time, where $\varepsilon'=\Theta(\varepsilon)$. Thus, it contradicts \ETH.
\end{proof}

\begin{figure}[!t]
  \centering
    \includegraphics[width=0.8\textwidth]{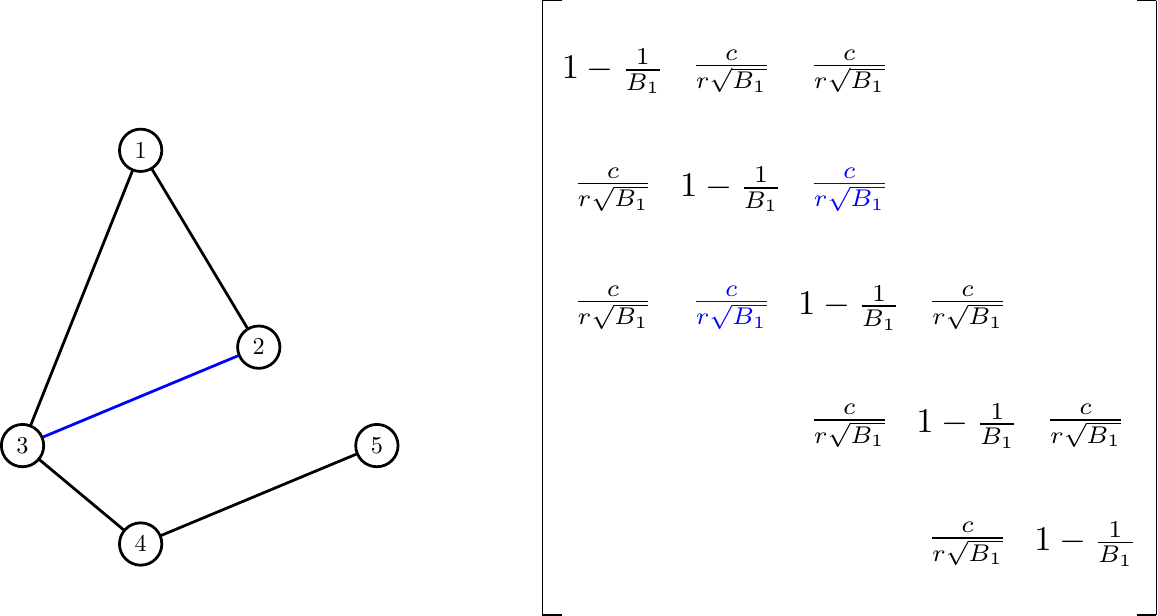}
    \caption{The left graph has $5$ nodes, and we convert it into a $5\times 5$ symmetric matrix.}
\end{figure}

\begin{definition}
Let $V$ be a $k$-dimensional subspace of $\mathbb{R}^d$, represented as the column span of a $d\times k$ matrix with orthonormal columns. We abuse notation and let $V$ be both the subspace and the corresponding matrix. For a set $Q$ of points, let
\begin{align*}
c(Q,V) = \sum_{q \in Q} d(q,V)^p = \sum_{q\in Q} \| q^\top (I - VV^\top) \|_2^p = \sum_{q\in Q} ( \|q\|^2 - \| q^\top V \|^2 )^{p/2},
\end{align*}
be the sum of $p$-th powers of distances of points in $Q$, i.e., $\| Q - QVV^\top \|_v$ with associated $M(x) = |x|^p$.
\end{definition}

\begin{lemma}
For any $k\in [d]$, the $k$-dimensional subspaces $V$ which minimize $c(E,V)$ are exactly the ${n \choose k}$ subspaces formed by taking the span of $k$ distinct standard unit vectors $e_i$, $i\in [d]$. The cost of any such $V$ is $d-k$.
\end{lemma}

\begin{theorem}
Given a set $Q$ of $\poly(d)$ points in $\mathbb{R}^d$, for a sufficiently small $\epsilon = 1/\poly(d)$, it is NP-hard to output a $k$-dimensional subspace $V$ of $\mathbb{R}^d$ for which $c(Q,V) \leq (1+\epsilon) c(Q,V^*)$, where $V^*$ is the $k$-dimensional subspace minimizing the expression $c(Q,V)$, that is $c(Q,V) \geq c(Q,V^*)$ for all $k$-dimensional subspaces $V$.
\end{theorem}

\begin{theorem}\label{thm:ETH_robust_main}
For a sufficiently small $\varepsilon=1/\poly(\log(d))$, there exist $1\leq k\leq d$, unless \ETH fails, there is no algorithm that can output a $k$-dimensional subspace $V$ of $\mathbb{R}^d$ for which $c(Q,V) \leq (1+\epsilon) c(Q,V^*)$, where $V^*$ is the $k$-dimensional subspace minimizing the expression $c(Q,V)$, that is $c(Q,V) \geq c(Q,V^*)$ for all $k$-dimensional subspaces $V$.
\end{theorem}

\begin{proof}
The reduction is from the clique problem of $d$-vertices $(d-3)$-regular graph. We construct the hard instance in the same way as in~\cite{cw15focs}. Given a $d$-vertes $(d-3)$-regular graph graph $G$, let $B_1=d^\alpha,B_2=d^\beta$ where $\beta>\alpha\geq 1$ are two sufficiently large constants. Let $c$ be such that
$$(1-1/B_1)^2+c^2/B_1=1.$$
We construct a $d\times d$ matrix $A$ as the following: $\forall i\in[d]$, let $A_{i,i}=1-1/B_1$ and $\forall i\not=j, A_{i,j}=A_{j,i}=c/\sqrt{B_1 r}$ if $(i,j)$ is an edge in $G$, and $A_{i,j}=A_{j,i}=0$ otherwise. Let us construct $A'\in\mathbb{R}^{2d\times d}$ as follows:
$$A'=\begin{bmatrix} A \\ B_2\cdot I_d\end{bmatrix},$$
where $I_d\in\mathbb{R}^d$ is a $d\times d$ identity matrix.
\begin{claim}[In proof of Theorem 54 in~\cite{cw15focs}]\label{cla:cw15_hardinstance}
Let $V'\in\mathbb{R}^{d\times k}$ satisfy that
 $$c(A',V')\leq (1+1/d^\gamma) c(A',V^*),$$
 where $A'$ is constructed as the above corresponding to the given graph $G$, and $\gamma>1$ is a sufficiently large constant, $V^*$ is the optimal solution which minimizes $c(A',V)$. Then if $G$ has a \kClique, given $V'$, there is a $\poly(d)$ time algorithm which can find the clique which has size at least $k$.
\end{claim}

Now, to apply \ETH here, we only need to apply a padding argument. We can construct a  matrix $A''\in\mathbb{R}^{N\times d}$ as follows:
$$A''=\begin{bmatrix}A'\\A'\\ \cdots \\A'\end{bmatrix}.$$
Basically, $A''$ contains $N/(2d)$ copies of $A'$ where $N=2^{d^{1-\alpha}}$, and $0<\alpha$ is a constant which can be arbitrarily small. Notice that $\forall V\in\mathbb{R}^{d\times k},$
$$c(V,A'')=\sum_{q\in A''} d(q,V)^p=N/(2d)\sum_{q\in A'} d(q,V)^p=N/(2d)c(V,A').$$
So if $V''$ gives a $(1+1/d^\gamma)$ approximation to $A''$, it also gives a $(1+1/d^\gamma)$ approximation to $A'$. So if we can find $V''$ in $\poly(N,d)$ time, we can output a \kClique of $G$ in $\poly(N,d)$ time. But unless \ETH fails, for a sufficiently small constant $\alpha'>0$ there is no $\poly(N,d)=O(2^{d^{1-\alpha'}})$ time algorithm that can output a \kClique of $G$. It means that there is no $\poly(N,d)$ time algorithm that can compute a $(1+1/d^\gamma)=(1+1/\poly(\log(N)))$ approximation to $A''$.
To make $A''$ be a square matrix, we can just pad with $0$s to make the size of $A''$ be $N\times N$. Thus, we can conclude, unless \ETH fails, there is no polynomial algorithm that can compute a $(1+1/\poly(\log(N)))$ rank-$k$ subspace approximation to a point set with size $N$.

\end{proof}

\subsection{Extending hardness from matrices to tensors}\label{sec:hardness_matrix_extension}
In this section, we briefly state some hardness results which are implied by hardness for matrices. The intuition is that, if there is a hard instance for the matrix problem, then we can always construct a tensor hard instance for the tensor problem as follos: the first face of the tensor is the hard instance matrix and it has all $0$s elsewhere. We can prove that the optimal tensor solution will always fit the first face and will have all $0$s elsewhere. Then the optimal tensor solution gives an optimal matrix solution.

\subsubsection{Entry-wise $\ell_1$ norm and $\ell_1$-$\ell_1$-$\ell_2$ norm}
In the following we will show that the hardness for entry-wise $\ell_1$ norm low rank matrix approximation implies the hardness for entry-wise $\ell_1$ norm low rank tensor approximation and asymmetric tensor norm ($\ell_1$-$\ell_1$-$\ell_2$) low rank tensor approximation problems.
\begin{theorem}[Theorem H.13 in~\cite{swz17}]\label{thm:hard_matrix_l1}
Unless \ETH fails, for an arbitrarily small constant $\gamma>0,$ given some matrix $A\in\mathbb{R}^{n\times n}$, there is no algorithm that can compute $\wh{x},\wh{y}\in\mathbb{R}^n$ s.t.
\begin{align*}
\|A-\wh{x}\wh{y}^\top\|_1\leq \left(1+\frac{1}{\log^{1+\gamma}(n)}\right)\min_{x,y\in\mathbb{R}^n}\|A-xy^\top\|_1,
\end{align*}
in $\poly(n)$ time.
\end{theorem}
We can get the hardness for tensors directly.
\begin{theorem}
Unless \ETH fails, for an arbitrarily small constant $\gamma>0,$ given some tensor $A\in\mathbb{R}^{n\times n \times n}$,
\begin{enumerate}
\item there is no algorithm that can compute $\wh{x},\wh{y},\wh{z}\in\mathbb{R}^n$ s.t.
\begin{align*}
\|A-\wh{x}\otimes \wh{y}\otimes \wh{z}\|_1\leq \left(1+\frac{1}{\log^{1+\gamma}(n)}\right)\min_{x,y,z\in\mathbb{R}^n}\|A-x\otimes y\otimes z\|_1,
\end{align*}
in $\poly(n)$ time.
\item there is no algorithm can compute $\wh{x},\wh{y},\wh{z}\in\mathbb{R}^n$ s.t.
\begin{align*}
\|A-\wh{x}\otimes \wh{y}\otimes \wh{z}\|_u\leq \left(1+\frac{1}{\log^{1+\gamma}(n)}\right)\min_{x,y,z\in\mathbb{R}^n}\|A-x\otimes y\otimes z\|_u,
\end{align*}
in $\poly(n)$ time.
\end{enumerate}
\end{theorem}
\begin{proof}
Let matrix $\wh{A}\in\mathbb{R}^{n\times n}$ be the hard instance in Theorem~\ref{thm:hard_matrix_l1}. We construct tensor $A\in\mathbb{R}^{n\times n\times n}$ as follows: $\forall i,j,l\in [n],l\not=1$ we let $A_{i,j,1}=\wh{A}_{i,j},A_{i,j,l}=0.$

Suppose $\wh{x},\wh{y},\wh{z}\in\mathbb{R}^n$ satisfies
\begin{align*}
\|A-\wh{x}\otimes \wh{y}\otimes \wh{z}\|_1\leq \left(1+\frac{1}{\log^{1+\gamma}(n)}\right)\min_{x,y,z\in\mathbb{R}^n}\|A-x\otimes y\otimes z\|_1.
\end{align*}
Then letting $z'=(1,0,0,\cdots,0)^\top$, we have
\begin{align*}
\|A-\wh{x}\otimes \wh{y}\otimes z'\|_1
\leq \|A-\wh{x}\otimes \wh{y}\otimes \wh{z}\|_1
\leq  \left(1+\frac{1}{\log^{1+\gamma}(n)}\right)\min_{x,y,z\in\mathbb{R}^n}\|A-x\otimes y\otimes z\|_1.
\end{align*}
The first inequality follows since $\forall i,j,l\in[n],l\not=1,$ we have $A_{i,j,l}=0$. Let
\begin{align*}
x^*,y^*=\arg\min_{x,y\in\mathbb{R}^n} \|\wh{A}-xy^\top\|_1.
\end{align*}
Then
\begin{align*}
\|A-\wh{x}\otimes \wh{y}\otimes z'\|_1
\leq \left(1+\frac{1}{\log^{1+\gamma}(n)}\right)\|A-\wh{x}\otimes \wh{y}\otimes \wh{z}\|_1
\leq \left(1+\frac{1}{\log^{1+\gamma}(n)}\right)\|A-x^*\otimes y^*\otimes z'\|_1.
\end{align*}
Thus, we have
\begin{align*}
\|\wh{A}-\wh{x}\wh{y}^\top\|_1\leq \left(1+\frac{1}{\log^{1+\gamma}(n)}\right)\|\wh{A}-x^*(y^*)^\top\|_1.
\end{align*}
Combining with Theorem~\ref{thm:hard_matrix_l1}, we know that unless ETH fails, there is no $\poly(n)$ running time algorithm which can output
\begin{align*}
\|A-\wh{x}\otimes \wh{y}\otimes \wh{z}\|_1\leq \left(1+\frac{1}{\log^{1+\gamma}(n)}\right)\min_{x,y,z\in\mathbb{R}^n}\|A-x\otimes y\otimes z\|_1.
\end{align*}

Similarly, we can prove that if $\wt{x},\wt{y},\wt{z}\in\mathbb{R}^n$ satisfies:
\begin{align*}
\|A-\wt{x}\otimes \wt{y}\otimes \wt{z}\|_u\leq \left(1+\frac{1}{\log^{1+\gamma}(n)}\right)\min_{x,y,z\in\mathbb{R}^n}\|A-x\otimes y\otimes z\|_u,
\end{align*}
then
\begin{align*}
\|\wh{A}-\wt{x}\wt{y}^\top\|_1\leq \left(1+\frac{1}{\log^{1+\gamma}(n)}\right)\|\wh{A}-x^*(y^*)^\top\|_1.
\end{align*}
We complete the proof.

\end{proof}

\begin{corollary}
Unless \ETH fails, for arbitrarily small constant $\gamma>0,$
\begin{enumerate}
\item there is no algorithm that can compute $(1+\varepsilon)$ entry-wise $\ell_1$ norm rank-$1$ tensor approximation in $2^{O(1/\varepsilon^{1-\gamma})}$ running time. ($\| \cdot \|_1$-norm is defined in Section~\ref{sec:l1})
\item there is no algorithm that can compute $(1+\varepsilon)$ $\ell_u$-norm rank-$1$ tensor approximation in $2^{O(1/\varepsilon^{1-\gamma})}$ running time. ($\| \cdot \|_u$-norm is defined in Section~\ref{sec:lu_l112})
\end{enumerate}
\end{corollary}

\subsubsection{$\ell_1$-$\ell_2$-$\ell_2$ norm}

\begin{theorem}
Unless \ETH fails, for arbitrarily small constant $\gamma>0,$ given some tensor $A\in\mathbb{R}^{n\times n \times n}$, there is no algorithm can compute $\wh{U},\wh{V},\wh{W}\in\mathbb{R}^{n\times k}$ s.t.
\begin{align*}
\|A-\wh{U}\otimes \wh{V}\otimes \wh{W}\|_v\leq \left(1+\frac{1}{\poly(\log n)}\right)\min_{U,V,W\in\mathbb{R}^{n\times k}}\|A-U\otimes V\otimes W\|_v,
\end{align*}
in $\poly(n)$ running time. ($\| \cdot \|_v$-norm is defined in Section~\ref{sec:lv_l122})
\end{theorem}
\begin{proof}
Let matrix $\wh{A}\in\mathbb{R}^{n\times n}$ be the hard instance in Theorem~\ref{thm:ETH_robust_main}. We construct tensor $A\in\mathbb{R}^{n\times n\times n}$ as follows: $\forall i,j,l\in [n],l\not=1$ we let $A_{i,j,1}=\wh{A}_{i,j},A_{i,j,l}=0.$

Suppose $\wh{U},\wh{V},\wh{W}\in\mathbb{R}^{n\times k}$ satisfies
\begin{align*}
\|A-\wh{U}\otimes \wh{V}\otimes \wh{W}\|_v\leq \left(1+\frac{1}{\poly(\log n)}\right)\min_{U,V,W\in\mathbb{R}^{n\times k}}\|A-U\otimes V\otimes W\|_v.
\end{align*}
Let $W'\in\mathbb{R}^{n\times k}$ be the following:
\begin{align*}
W'=\begin{bmatrix}1&1&\cdots&1\\0&0&\cdots&0\\0&0&\cdots&0\\ \cdots&\cdots&\cdots&\cdots\\0&0&\cdots&0\\\end{bmatrix},
\end{align*}
then we have
\begin{align*}
\|A-\wh{U}\otimes \wh{V}\otimes W'\|_v
\leq \|A-\wh{U}\otimes \wh{V}\otimes \wh{W}\|_v
\leq  \left(1+\frac{1}{\poly(\log n)}\right)\min_{U,V,W\in\mathbb{R}^{n\times k}}\|A-U\otimes V\otimes W\|_v.
\end{align*}
The first inequality follows since $\forall i,j,l\in[n],l\not=1,$ we have $A_{i,j,l}=0$. Let
\begin{align*}
U^*,V^*=\arg\min_{U,V\in\mathbb{R}^{n\times k}} \|\wh{A}-UV^\top\|_v.
\end{align*}
Then
\begin{align*}
\|A-\wh{U}\otimes \wh{V}\otimes W'\|_v
\leq &~ \left(1+\frac{1}{\poly(\log n)}\right)\|A-\wh{U}\otimes \wh{V}\otimes \wh{W}\|_v \\
\leq &~\left(1+\frac{1}{\poly(\log n)}\right)\|A-U^*\otimes V^*\otimes W'\|_v.
\end{align*}
Thus, we have
\begin{align*}
\|\wh{A}-\wh{U}\wh{V}^\top\|_v\leq \left(1+\frac{1}{\poly(\log n)}\right)\|\wh{A}-U^*(V^*)^\top\|_v.
\end{align*}
Combining with Theorem~\ref{thm:ETH_robust_main}, we know that unless \ETH fails, there is no $\poly(n)$ time algorithm which can output
\begin{align*}
\|A-\wh{U}\otimes \wh{V}\otimes \wh{W}\|_v\leq \left(1+\frac{1}{\poly(\log n)}\right)\min_{U,V,W\in\mathbb{R}^{n\times k}}\|A-U\otimes V\otimes W\|_v.
\end{align*}
\end{proof}
\newpage
\section{Hard Instance}\label{sec:hardinstance}

This section provides some hard instances for tensor problems.








\subsection{Frobenius CURT decomposition for $3$rd order tensor}
In this section we will prove that a relative-error Tensor CURT is not possible unless
$C$ has $\Omega(k/\epsilon)$ columns from $A$, $R$ has $\Omega(k/\epsilon)$ rows from $A$, $T$ has $\Omega(k/\epsilon)$ tubes from $A$ and $U$ has rank $\Omega(k)$.

We use a similar construction from \cite{bw14,bdm11,dr10} and extend it to the tensor setting.
\begin{theorem}\label{thm:hardinstance_CURT_3rd_order}
There exists a tensor $A\in \mathbb{R}^{n\times n \times n}$ with the following property. Consider a factorization CURT, with $C\in \mathbb{R}^{n\times c}$ containing $c$ columns of $A$, $R\in \mathbb{R}^{n\times r}$ containing $r$ rows of $A$, $T\in \mathbb{R}^{n\times t}$ containing $r$ tubes of $A$, and $U\in \mathbb{R}^{c\times r \times t}$, such that
\begin{align*}
\left\| A - \sum_{i=1}^n \sum_{j=1}^n \sum_{l=1}^n U_{i,j,l} \cdot C_i \otimes R_j \otimes T_l \right\|_F^2 \leq (1+\epsilon) \| A - A_k \|_F^2.
\end{align*}
Then, for any $\epsilon < 1$ and any $k\geq 1$,
\begin{align*}
c = \Omega(k/\epsilon), ~r = \Omega(k/\epsilon), ~t=\Omega(k/\epsilon) \mathrm{~and~} \rank(U) \geq k/3.
\end{align*}
\end{theorem}
\begin{proof}
For any $i\in [d]$, let $e_i\in \mathbb{R}^{d}$ denote the $i$-th standard basis vector. For $\alpha >0$ and integer $d>1$, consider the matrix $D\in \mathbb{R}^{(d+1) \times (d+1)}$,
\begin{align*}
D = & ~ \begin{bmatrix} e_1 + \alpha e_2 & e_1 + \alpha e_3 & \cdots & e_1 +\alpha e_{d+1} & 0 \end{bmatrix} \\
  = & ~
  \begin{bmatrix}
  1 & 1 & \cdots & 1 & 0\\
  \alpha & & & & 0 \\
  & \alpha & & & 0 \\
  & & \ddots & & \vdots\\
  & & & \alpha & 0
  \end{bmatrix}
\end{align*}
We construct matrix $B\in \mathbb{R}^{(d+1)k/3 \times (d+1)k/3}$ by repeating matrix $D$ $k/3$ times along its main diagonal,
\begin{align*}
B=
\begin{bmatrix}
D & & & \\
& D & & \\
& & \ddots & \\
& & & D
\end{bmatrix}
\end{align*}
Let $m=(d+1)k/3$. We construct a tensor $A\in \mathbb{R}^{n\times n\times n}$ with $n = 3m$ by repeating matrix $B$ three times in the following way,
\begin{align*}
A_{1,j,l} &= B_{j,l}, \forall j,l \in [m] \times [m]\\
A_{m+i,m+1,m+l} &= B_{i,l}, \forall i,l \in [m] \times [m]\\
A_{2m+i,2m+j,2m+1} &= B_{i,j}, \forall j,i \in [m] \times [m]
\end{align*}
and $0$ everywhere else.
We first state some useful properties for matrix $D$,
\begin{align*}
D^\top D = \begin{bmatrix}
1_d 1_d^\top + \alpha^2  I_d & 0 \\
0 & 0
\end{bmatrix}
\in \mathbb{R}^{(d+1)\times (d+1)}
\end{align*}
where
\begin{align*}
\sigma_1^2(D) &= d+\alpha^2, \\
\sigma_i^2(D) &= \alpha^2 , & \forall i=2, \cdots,d \\
\sigma_{d+1}^2(D) &= 0.
\end{align*}
 By definition of matrix $B$, we can obtain the following properties,
\begin{align*}
\sigma_i^2(B) & = d + \alpha^2 , & \forall i =1, \cdots, k/3 \\
\sigma_i^2(B) & = \alpha^2 , & \forall i =k/3+1, \cdots, dk/3 \\
\sigma_i^2(B) & = 0 , & \forall i = dk+1 ,\cdots, dk/3+k/3
\end{align*}
By definition of $A$, we can copy $B$ into three disjoint $n\times n\times n$ sub-tensors on the main diagonal of tensor $A$. Thus, we have
\begin{align*}
\sigma_i^2(A) & = d + \alpha^2 , & \forall i =1, \cdots, k \\
\sigma_i^2(A) & = \alpha^2 , & \forall i =k+1, \cdots, dk \\
\sigma_i^2(A) & = 0 , & \forall i = dk+1 ,\cdots, dk+k
\end{align*}
Let $A_{(k)}$ denote the best rank-$k$ approximation to $A$, and let $D_1$ denote the best rank-$1$ approximation to $D$. Using the above properties, for any $k\geq 1$, we can compute $\| A -  A_{(k)} \|_F^2$,
\begin{align}\label{eq:hardinstance_A_minus_Ak}
\| A -  A_k \|_F^2 = k \| D- D_1 \|_F^2 = k (d-1)\alpha^2.
\end{align}

Suppose we have a CUR decomposition with $c' = o(k/\epsilon)$ columns, $r' = o(k/\epsilon)$ rows or $t'=o(k/\epsilon)$ tubes. Since the tensor is equivalent by looking through any of the $3$ dimensions/directions, we just need to show why the cost will be at least $ (1+\epsilon)\|A-A_k\|_F^2$ if we choose $t = o(k/\epsilon)$ columns and $t=o(k/\epsilon)$ rows.

Let $C\in \mathbb{R}^{n\times c}$ denote the optimal solution. Then it should have the following form,
\begin{align*}
C = \begin{bmatrix}
C_1 & & \\
& C_2 & \\
& & C_3
\end{bmatrix}
\end{align*}
where $C_1\in \mathbb{R}^{m\times c_1}$ contains $c_1$ columns from $A_{1:m,1:m,1:m} \in \mathbb{R}^{m\times m\times m}$, $C_2\in \mathbb{R}^{m\times c_2}$ contains $c_2$ columns from $A_{m+1:2m,m+1:2m,m+1:2m} \in \mathbb{R}^{m\times m\times m}$, $C_3\in \mathbb{R}^{m\times c_3}$ contains $c_3$ columns from $A_{2m+1:3m,2m+1:3m,2m+1:3m} \in \mathbb{R}^{m\times m\times m}$.

Let $R\in \mathbb{R}^{n\times r}$ denote the optimal solution. Then it should have the following form,
\begin{align*}
R = \begin{bmatrix}
R_1 & & \\
& R_2 & \\
& & R_3
\end{bmatrix}
\end{align*}
\begin{align}\label{eq:hardinstance_CC_RR}
\|  A - A(CC^\dagger,RR^\dagger,I) \|_F^2 \geq \| B - R_1 R_1^\dagger B \|_F^2 + \| B - C_2 C_2^\dagger B\|_F^2 + \| B^\top - C_3 C_3^\dagger B^\top \|_F^2.
\end{align}
By the analysis in Proposition 4 of \cite{dv06}, we have
\begin{align}\label{eq:hardinstance_R1R1}
\| B -  R_1 R_1^\dagger B \|_F^2 \geq (k/3)(1+b \cdot \alpha) \|D-D_{(1)}\|_F^2.
\end{align}
and
\begin{align}\label{eq:hardinstance_C2C2}
\| B -  C_2 C_2^\dagger B \|_F^2 \geq (k/3)(1+b \cdot \alpha) \|D-D_{(1)}\|_F^2.
\end{align}
Let $C_3\in \mathbb{R}^{m\times c_3}$ contain any $c_3$ columns from $B^\top$. Note that $C_3$ contains $c_3$($\leq t$) columns from $B^\top$, equivalently $C_2^\top$ contains $c_2$ rows from $B$. Recall that $B$ contains $k$ copies of $D\in \mathbb{R}^{(d+1)\times (d+1)}$ along its main diagonal. 
Even if we choose $t$ columns of $B^\top$, the cost is at least
\begin{align}\label{eq:hardinstance_C3C3}
\| B^\top - C_3 C_3^\dagger B^\top \|_F^2 \geq (k/3)\|D-D_{(t)}\|_F^2 \geq (k/3) (d-t) \alpha^2.
\end{align}

 Combining Equations~\eqref{eq:hardinstance_A_minus_Ak}, \eqref{eq:hardinstance_CC_RR}, \eqref{eq:hardinstance_R1R1}, \eqref{eq:hardinstance_C2C2}, \eqref{eq:hardinstance_C3C3}, $\alpha=\epsilon$ gives,
 \begin{align*}
 & ~\frac{\| A - C C^\dagger A\|_F^2} {\| A - A_{(k)}\|_F^2} \\
 \geq & ~\frac{\| B - R_1 R_1^\dagger B \|_F^2 + \| B - C_2 C_2^\dagger B\|_F^2 + \| B^\top - C_3 C_3^\dagger B^\top \|_F^2}{\| A - A_{(k)} \|_F^2} & \text{~by~Eq.~\eqref{eq:hardinstance_CC_RR}}\\
 \geq & ~ \frac{\| B - R_1 R_1^\dagger B \|_F^2 + \| B - C_2 C_2^\dagger B\|_F^2 + \| B^\top - C_3 C_3^\dagger B^\top \|_F^2}{k(d-1)\alpha^2} & \text{~by~Eq.~\eqref{eq:hardinstance_A_minus_Ak}} \\
\geq & ~ \frac{ 2(k/3) (1+b\epsilon)(d-1)\epsilon^2 + (k/3)(d-t)\epsilon^2 }{ k (d-1)\epsilon^2 }& \text{~by~Eq.~\eqref{eq:hardinstance_R1R1},\eqref{eq:hardinstance_C2C2},\eqref{eq:hardinstance_C3C3} and }\alpha=\epsilon\\
= & ~ \frac{k(d-1)\epsilon^2  + (k/3)(-t+1)\epsilon^2 + 2(k/3)b\epsilon (d-1)\epsilon^2}{k(d-1)\epsilon^2} \\
= &~ 1 + \frac{ (k/3)\epsilon^2 (2b\epsilon(d-1) -t+1) }{k(d-1)\epsilon^2} \\
= &~ 1 + \frac{ 2b\epsilon(d-1) -t +1 }{3(d-1)} \\
\geq & ~ 1 + (b/3)\epsilon & \text{~by~} 2t \leq b \epsilon(d-1)/2 \\
\geq & ~ 1 +\epsilon. & \text{~by~} b>3.
 \end{align*}
 which gives a contradiction.
\end{proof}

\subsection{General Frobenius CURT decomposition for $q$-th order tensor}
In this section, we extend the hard instance for $3$rd order tensors to $q$-th order tensors.
\begin{theorem}\label{thm:hardinstance_CURT_general_order}
For any constant $q\geq 1$, there exists a tensor $A\in \mathbb{R}^{n\times n \times \cdots \times n}$ with the following property. Define
\begin{align*}
\OPT=\min_{\rank-k~A_k\in \mathbb{R}^{c_1 \times c_2 \times \cdots \times c_q}}\| A - A_k \|_F^2.
\end{align*}
Consider a $q$-th order factorization CURT, with $C_1 \in \mathbb{R}^{n\times c_1}$ containing $c$ columns from the $1$st dimension of $A$, $C_2 \in \mathbb{R}^{n\times c_2}$ containing $c_2$ columns from the $2$nd dimension of $A$, $\cdots$, $C_q \in \mathbb{R}^{n\times c_q}$ containing $c_q$ columns from the $q$-th dimension of $A$ and a tensor $U\in \mathbb{R}^{c_1\times c_2 \times \cdots \times c_q}$, such that
\begin{align*}
\left\| A - \sum_{i_1=1}^n \sum_{i_2=1}^n \cdots \sum_{i_q=1}^n  U_{i_1,i_2,\cdots,i_q} \cdot C_{1,i_1} \otimes C_{2,i_2} \otimes \cdots \otimes C_{q,i_q} \right\|_F^2 \leq (1+\epsilon) \OPT.
\end{align*}
There exists a constant $c'<1$ such that for any $\epsilon < c'$ and any $k\geq 1$,
\begin{align*}
c_1 = \Omega(k/\epsilon), ~c_2 = \Omega(k/\epsilon), \cdots, ~c_q=\Omega(k/\epsilon) \mathrm{~and~} \rank(U) \geq c' k.
\end{align*}
\end{theorem}
\begin{proof}
We use the same matrix $D\in \mathbb{R}^{(d+1)\times (d+1)}$ as the proof of Theorem~\ref{thm:hardinstance_CURT_3rd_order}. Then we can construct matrix $B \in \mathbb{R}^{(d+1) k/q \times (d+1)k/q}$ by repeating matrix $D$ $k/q$ times along the its main diagonal,
\begin{align*}
B = \begin{bmatrix}
D & & & \\
& D & & \\
& & \ddots & \\
& & & D
\end{bmatrix}
\end{align*}
Let $m=(d+1)/q$. We construct a tensor $A\in \mathbb{R}^{n \times n \times \cdots \times n}$ with $n= qm$ by repeating the matrix $q$ times in the following way,
\begin{align*}
A_{[1:m],[1:m],1,1,1,\cdots,1,1} & = B, \\
A_{m+1,[m+1:2m],[m+1:2m],m+1,m+1, \cdots,m+1, m+1 } & = B^\top, \\
A_{2m+1,2m+1,[2m+1:3m],[2m+1:3m],2m+1, \cdots,2m+1, 2m+1 } & = B, \\
A_{3m+1,3m+1,3m+1,[3m+1:4m],[3m+1:4m], \cdots,2m+1, 3m+1 } & = B^\top, \\
\cdots & \cdots \cdots \\
A_{(q-2)m+1,(q-2)m+1,(q-2)m+1,(q-2)m+1,(q-2)m+1,\cdots,[(q-2)m+1:(q-1)m], [(q-2)m+1:(q-1)m]} & = B, \\
A_{[(q-1)m+1:qm],(q-1)m+1,(q-1)m+1,(q-1)m+1,(q-1)m+1,\cdots, (q-1)m+1, [(q-1)m+1:qm]} & = B^\top,
\end{align*}
where there are $q/2$ $B$s and $q/2$ $B^\top$s on the right when $q$ is even, and there are $(q+1)/2$ $B$s and $(q-1)/2$ $B$s on the right when $q$ is odd. Note that this tensor $A$ is equivalent if we look through any of the $q$ dimensions/directions.
Similarly as before, we have
\begin{align*}
\| A - A_{(k)} \|_F^2 = k \| D - D_{(1)} \|_F^2 = k(d-1) \alpha^2.
\end{align*}
Suppose there is a general CURT decomposition (of this $q$-th order tensor), with $c_1=c_2=\cdots c_q =o(k/\epsilon)$ columns from each dimension. Let $C_1\in \mathbb{R}^{n\times c_1}, C_2\in \mathbb{R}^{n\times c_2}$, $\cdots$, $C_q\in \mathbb{R}^{n\times c_q}$ denote the optimal solution. Then the $C_i$ should have the following form,
\begin{align*}
C_1 = \begin{bmatrix}
C_{1,1} & & &\\
& C_{1,2} & & \\
& & \ddots &\\
& & & C_{1,q}
\end{bmatrix}
,
C_2 = \begin{bmatrix}
C_{2,1} & & &\\
& C_{2,2} & & \\
& & \ddots &\\
& & & C_{2,q}
\end{bmatrix}
,\cdots,
C_q = \begin{bmatrix}
C_{q,1} & & &\\
& C_{q,2} & & \\
& & \ddots &\\
& & & C_{q,q}
\end{bmatrix}
\end{align*}
(In the rest of the proof, we focus on the case when $q$ is even. Similarly, we can show the same thing when $q$ is odd.)
We have
\begin{align*}
 & ~\| A - A (C_1 C_1^\dagger, C_2C_2^\dagger, \cdots, C_q C_q^\dagger) \|_F^2 \\
\geq & ~ \sum_{i=1}^{q/2} \| B - C_{2i-1,2i-1} C_{2i-1,2i-1}^\dagger B \|_F^2 +  \| B^\top - C_{2i,2i} C_{2i,2i}^\dagger B^\top \|_F^2 \\
\geq & ~ (q/2) \left(  (k/q)(1+ b\alpha) \| D - D_{(1)} \|_F^2  + (k/q) (d-t) \alpha^2 \right) \\
= & ~ (q/2) \left(  (k/q)(1+ b\alpha) (d-1)\alpha^2  + (k/q) (d-t) \alpha^2 \right) \\
\end{align*}
where the second inequality follows by Equations~\eqref{eq:hardinstance_C2C2} and \eqref{eq:hardinstance_C3C3}, and the third step follows by $\| D- D_{(1)}\|_F^2 = (d-1)\alpha^2$.

Putting it all together, we have
\begin{align*}
& ~ \frac{\| A - A (C_1 C_1^\dagger, C_2C_2^\dagger, \cdots, C_q C_q^\dagger) \|_F^2}{\| A - A_{(k)}\|_F^2} \\
\geq & ~ \frac{ (q/2) \left(  (k/q)(1+ b\alpha) (d-1)\alpha^2  + (k/q) (d-t) \alpha^2 \right)  }{k(d-1)\alpha^2} \\
= & ~ \frac{k(d-1)\alpha^2 + (k/2) b\alpha (d-1)\alpha^2 + (k/q) (-t+1)\alpha^2 }{k(d-1)\alpha^2} \\
= & ~ 1 + \frac{(k/2) b\alpha (d-1)\alpha^2 + (k/q) (-t+1)\alpha^2 }{k(d-1)\alpha^2} \\
\leq & ~ 1 + \frac{(k/3) b\alpha (d-1)\alpha^2}{k(d-1)\alpha^2 } \\
= & ~ 1 + (b/3) \epsilon & \text{~by~}\epsilon=\alpha\\
> & ~ 1 + \epsilon & \text{~by~} b>3.
\end{align*}
which leads to a contradiction. Similarly we can show the rank is at least $\Omega(k)$.
\end{proof}
\newpage
\section{Distributed Setting}\label{sec:distributed} 
Input data to large-scale machine learning and data mining tasks may be distributed across different machines. The communication cost becomes the major bottleneck of distributed protocols, and so there is a growing body of work on low rank matrix approximations in the distributed model~\cite{tisseur1999parallel,qu2002principal,bai2005principal,sensors2008,macua2010consensus,fegk13,poulson2013elemental,kvw14,bklw14,blswx16,bwz16,wz16,swz17} and also many other machine learning problems such as clustering, boosting, and column subset selection \cite{bblm14,blgbs15,abw17}. Thus, it is natural to ask whether our algorithm can be applied in the distributed setting. This section will discuss the distributed Frobenius norm low rank tensor approximation protocol in the so-called arbitrary-partition model (see, e.g. \cite{kvw14,bwz16}).

In the following, we extend the definition of the arbitrary-partition model~\cite{kvw14} to fit our tensor setting.
\begin{definition}[Arbitrary-partition model~\cite{kvw14}]\label{def:model_arb}
There are $s$ machines, and the $i^{\mathrm{th}}$ machine holds a tensor $A_i\in\mathbb{R}^{n\times n\times n}$ as its local data tensor. The global data tensor is implicit and is denoted as $A=\sum_{i=1}^s A_i$. Then, we say that $A$ is arbitrarily partitioned into $s$ matrices distributed in the $s$ machines. In addition, there is also a coordinator. In this model, the communication is only allowed between the machines and the coordinator. The total communication cost is the total number of words delivered between machines and the coordinator. Each word has $O(\log(sn))$ bits.
\end{definition}
Now, let us introduce the distributed Frobenius norm low rank tensor approximation problem in the arbitrary partition model:
\begin{definition}[Arbitrary-partition model Frobenius norm rank-$k$ tensor approximation]\label{def:distri_arb}
Tensor $A\in\mathbb{R}^{n\times n\times n}$ is arbitrarily partitioned into $s$ matrices $A_1,A_2,\cdots,A_s$ distributed in $s$ machines respectively, and $\forall i\in[s]$, each entry of $A_i$ is at most $O(\log(sn))$ bits. Given tensor $A$, $k\in \mathbb{N}_+$ and an error parameter $0<\varepsilon<1$, the goal is to find a distributed protocol in the model of Definition~\ref{def:model_arb} such that
\begin{enumerate}
\item Upon termination, the protocol leaves three matrices $U^*,V^*,W^*\in \mathbb{R}^{n\times k}$ on the coordinator.
\item $U^*,V^*,W^*$ satisfies that
$$ \left\|\sum_{i=1}^k U_i^*\otimes V_i^*\otimes W_i^* -A \right\|_F^2\leq (1+\varepsilon) \min_{\rank-k~A'} \|A'-A\|_F^2.$$
\item The communication cost is as small as possible.
\end{enumerate}
\end{definition}

\begin{theorem}
Suppose tensor $A\in\mathbb{R}^{n\times n\times n}$ is distributed in the arbitrary partition model (See Definition~\ref{def:model_arb}). There is a protocol( in Algorithm~\ref{alg:arb}) which solves the problem in Definition~\ref{def:distri_arb} with constant success probability. In addition, the communication complexity of the protocol is $s(\poly(k/\varepsilon)+O(kn))$ words.
\end{theorem}

\begin{proof}
\textbf{Correctness.} The correctness is implied by Algorithm~\ref{alg:f_main_algorithm} and Algorithm~\ref{alg:f_input_sparsity_reduction} (Theorem~\ref{thm:f_main_algorithm}.) Notice that $A_1=\sum_{i=1}^s A_{i,1},A_2=\sum_{i=1}^s A_{i,2},A_3=\sum_{i=1}^s A_{i,3}$, which means that
\begin{align*}
Y_1=T_1A_1S_1,Y_2=T_2A_2S_2,Y_3=T_3A_3S_3,
\end{align*}
and
\begin{align*}
C=A(T_1,T_2,T_3).
\end{align*}
According to line~\ref{sta:compute_Xstars},
\begin{align*}
X^*_1,X^*_2,X^*_3=\underset{X_1,X_2,X_3}{\arg\min}\left\|\overset{k}{\underset{j=1}{\sum}}(Y_1X_1)_j\otimes (Y_2X_2)_j\otimes (Y_3X_3)_j-C \right\|_F.
\end{align*}
According to Lemma~\ref{lem:f_input_sparsity_reduction}, we have
\begin{align*}
&\left\|\overset{k}{\underset{j=1}{\sum}}(T_1A_1S_1X^*_1)_j\otimes (T_2A_2S_2X^*_2)_j\otimes (T_3A_3S_3X^*_3)_j-A(T_1,T_2,T_3)\right\|_F^2\\
\leq& (1+O(\varepsilon))\underset{X_1,X_2,X_3}{\min} \left\|\overset{k}{\underset{j=1}{\sum}}(A_1S_1X_1)_j\otimes (A_2S_2X_2)_j\otimes (A_3Y_3X_3)_j-A \right\|_F^2\\
\leq&(1+O(\varepsilon))\underset{U,V,W}{\min} \left\| \sum_{i=1}^k U_i\otimes V_i\otimes W_i -A\right\|_F^2,
\end{align*}
where the last inequality follows by the proof of Theorem~\ref{thm:f_main_algorithm}. By scaling a constant of $\varepsilon$, we complete the proof of correctness.

\textbf{Communication complexity.} Since $S_1,S_2,S_3$ are $w_1$-wise independent, and $T_1,T_2,T_3$ are $w_2$-wise independent, the communication cost of sending random seeds in line~\ref{sta:send_seed} is $O(s(w_1+w_2))$ words, where $w_1=O(k),w_2=O(1)$ (see~\cite{kvw14,cw13,w14,kn14}). The communication cost in line~\ref{sta:send_YC} is $s\cdot\poly(k/\varepsilon)$ words due to $T_1 A_{i,1}S_1,T_2 A_{i,2}S_2,T_3 A_{i,3}S_3\in\mathbb{R}^{\poly(k/\varepsilon)\times O(k/\varepsilon)}$ and $C_i=A_i(T_1,T_2,T_3)\in\mathbb{R}^{\poly(k/\varepsilon)\times\poly(k/\varepsilon)\times \poly(k/\varepsilon)}$.

Notice that, since $\forall i\in[s]$ each entry of $A_i$ has at most $O(\log(sn))$ bits, each entry of $Y_1,Y_2,Y_3,C$ has at most $O(\log(sn))$ bits. Due to Theorem~\ref{thm:distributed_solve_small}, each entry of $X_1^*,X_2^*,X_3^*$ has at most $O(\log(sn))$ bits, and the sizes of $X_1^*,X_2^*,X_3^*$ are $\poly(k/\varepsilon)$ words. Thus the communication cost in line~\ref{sta:send_X} is $s\cdot\poly(k/\varepsilon)$ words.

Finally, since $\forall i\in[s],U^*_i,V^*_i,W^*_i\in\mathbb{R}^{n\times k}$, the communication here is at most $O(skn)$ words. The total communication cost is $s(\poly(k/\varepsilon)+O(kn))$ words.
\end{proof}

\begin{algorithm}[t!]\caption{Distributed Frobenius Norm Low Rank Approximation Protocol}\label{alg:arb}
\begin{algorithmic}[1]
\Procedure{DistributedFnormLowRankApproxProtocol}{$A$,$\varepsilon$,$k$,$s$}
\State $A\in\mathbb{R}^{n\times n\times n}$ was arbitrarily partitioned into $s$ matrices $A_1,\cdots,A_s\in \mathbb{R}^{n\times n\times n}$ on $s$ machines.
\State \hspace{2cm} {\bf Coordinator} \hspace{4.5cm} {\bf Machines} $i$
\State Chooses a random seed.
\State Sends it to all machines. \label{sta:send_seed}
\State \hspace{4.5cm} $--------->$
\State \hspace{7.5cm} $s_i \leftarrow O(k/\epsilon)$, $\forall i \in [3]$.
\State \hspace{7.5cm} Agree on $S_i\in\mathbb{R}^{n^2\times s_i }$, $\forall i\in [3]$
\State \hspace{7.5cm} which are $w_1$-wise independent random
\State \hspace{7.5cm} $N(0,1/s_i)$ Gaussian matrices.
\State \hspace{7.5cm} $t_i \leftarrow \poly(k/\epsilon)$, $\forall i \in [3]$.
\State \hspace{7.5cm} Agree on $T_i\in\mathbb{R}^{t_i \times n}$, $\forall i\in [3]$
\State \hspace{7.5cm} which are $w_2$-wise independent random
\State \hspace{7.5cm} sparse embedding matrices.
\State \hspace{7.5cm} Compute $Y_{i,1}\leftarrow T_1A_{i,1}S_1,$
\State \hspace{7.5cm} $Y_{i,2} \leftarrow T_2A_{i,2}S_2, Y_{i,3}\leftarrow T_3A_{i,3}S_3$.
\State \hspace{7.5cm} Send $Y_{i,1},Y_{i,2},Y_{i,3}$ to the coordinator.
\State \hspace{7.5cm} Send $C_i\leftarrow A_i(T_1,T_2,T_3)$ to the coordinator.\label{sta:send_YC}
\State \hspace{4.5cm} $<---------$
\State Compute $Y_1\leftarrow\overset{s}{\underset{i=1}{\sum}}Y_{i,1},Y_2\leftarrow \overset{s}{\underset{i=1}{\sum}}Y_{i,2}$,
\State $Y_3\leftarrow\overset{s}{\underset{i=1}{\sum}} Y_{i,3}$, $C\leftarrow \overset{s}{\underset{i=1}{\sum}} C_i$.
\State Compute $X^*_1,X^*_2,X^*_3$ by solving
\State $\underset{X_1,X_2,X_3}{\min}\| (Y_1X_1) \otimes (Y_2X_2) \otimes (Y_3X_3)-C\|_F$ \label{sta:compute_Xstars}
\State Send $X^*_1,X^*_2,X^*_3$ to machines. \label{sta:send_X}
\State \hspace{4.5cm} $--------->$
\State \hspace{7.5cm}Compute $U^*_i \leftarrow A_{i,1}S_1X^*_1,$
\State \hspace{7.5cm}$V^*_i \leftarrow A_{i,2}S_2X^*_2,\ W^*_i\leftarrow A_{i,3}S_3X^*_3$.
\State \hspace{7.5cm}Send $U_i^*$, $V_i^*,W^*_i$ to the coordinator. \label{sta:send_UVWstar}
\State \hspace{4.5cm} $<---------$
\State Compute $U^* \leftarrow \sum_{i=1}^s U_i^*$.
\State Compute $V^* \leftarrow \sum_{i=1}^s V_i^*$.
\State Compute $W^* \leftarrow \sum_{i=1}^s W_i^*$.
\State \Return $U^*$, $V^*$, $W^*$.
\EndProcedure
\end{algorithmic}
\end{algorithm}

\begin{remark}
If we slightly change the goal in Definition~\ref{def:distri_arb} to the following: the coordinator does not need to output $U^*,V^*,W^*$, but each machine $i$ holds $U_i^*,V_i^*,W_i^*$ such that $U^*=\sum_{i=1}^s U_i^*,V^*=\sum_{i=1}^s V_i^*,W^*=\sum_{i=1}^s W_i^*$, then the protocol shown in Algorithm~\ref{alg:arb} does not have to do the line~\ref{sta:send_UVWstar}. Thus the total communication cost is at most $s\cdot \poly(k/\varepsilon)$ words in this setting.
\end{remark}

\begin{remark}
Algorithm~\ref{alg:arb} needs exponential in $\poly(k/\varepsilon)$ running time since it solves a polynomial solver in line~\ref{sta:compute_Xstars}. Instead of solving line~\ref{sta:compute_Xstars}, we can solve the following optimization problem:
\begin{align*}
\alpha^*= \underset{\alpha\in\mathbb{R}^{s_1 \times s_2\times s_3}}{\arg\min}  \left\|\sum_{i=1}^{s_1}\sum_{j=1}^{s_2}\sum_{l=1}^{s_3}\alpha_{i,j,l}\cdot(Y_1)_i\otimes(Y_2)_j\otimes(Y_3)_l-C\right\|_F.
\end{align*}
Since it is actually a regression problem, it only takes polynomial running time to get $\alpha^*$. And according to Lemma~\ref{lem:tensor_regression},
 \begin{align*}
 \sum_{i=1}^{s_1}\sum_{j=1}^{s_2}\sum_{l=1}^{s_3}\alpha^*_{i,j,l}\cdot(Y_1)_i\otimes(Y_2)_j\otimes(Y_3)_l
 \end{align*}
 gives a rank-$O(k^3/\varepsilon^3)$ bicriteria solution.

 Further, similar to Theorem~\ref{thm:f_bicriteria_better}, we can solve
\begin{align*}
\min_{U\in \mathbb{R}^{n\times s_2 s_3}} \left\| \sum_{i=1}^{s_1} \sum_{j=1}^{s_2} U_{i+s_1(j-1)} \otimes (Y_2)_i \otimes (Y_3)_j   -C \right\|_F,
\end{align*}
where $C= \sum_{i} A_i(I,T_2,T_3)$. Thus, we can obtain a $\rank$-$O(k^2/\epsilon^2)$ in polynomial time.
\end{remark}

\begin{remark}
If we select sketching matrices $S_1,S_2,S_3,T_1,T_2,T_3$ to be random Cauchy matrices, then we are able to compute distributed entry-wise $\ell_1$ norm rank-$k$ tensor approximation (see Theorem~\ref{thm:l1_input_sparsity_time}). The communication cost is still $s(\poly(k/\varepsilon)+O(kn))$ words. If we only require a bicriteria solution, then it only needs polynomial running time.
\end{remark}

Using similar techniques as in the proof of Theorem~\ref{thm:f_solving_small_problems}, we can obtain:

\begin{theorem}\label{thm:distributed_solve_small}
Let $\max_{i} \{t_i, d_i\} \leq n$. Given a $t_1 \times t_2 \times t_3$ tensor $A$ and three matrices: a $t_1 \times d_1$ matrix $T_1$, a $t_2 \times d_2$ matrix $T_2$, and a $t_3 \times d_3$ matrix $T_3$. For any $\delta > 0$, if there exists a solution to
\begin{align*}
\min_{X_1,X_2,X_3} \left\| \sum_{i=1}^k (T_1 X_1)_i \otimes (T_2 X_2)_i \otimes (T_3 X_3)_i - A \right\|_F^2 := \OPT,
\end{align*}
and each entry of $X_i$ can be expressed using $O(\log n)$ bits, then there exists an algorithm that takes $ \poly(\log n) \cdot 2^{ O( d_1 k+d_2 k+d_3 k)}$ time and outputs three matrices: $\wh{X}_1$, $\wh{X}_2$, and $\wh{X}_3$ such that $\| (T_1 \wh{X}_1)\otimes (T_2 \wh{X}_2) \otimes (T_3\wh{X}_3) - A\|_F^2 =\OPT$.
\end{theorem}
\newpage
\section{Streaming Setting}\label{sec:streaming} 
One of the computation models which is closely related to the distributed model of computation is the streaming model. There is a growing line of work in the streaming model. Some problems are very fundamental in the streaming model such like Heavy Hitters~\cite{lnnt16,bcinww16,bciw16}, and streaming numerical linear algebra problems~\cite{cw09}. Streaming low rank matrix approximation has been extensively studied by previous work like \cite{ cw09, kl11, gp13, lib13,klmms14,bwz16,swz17}. In this section, we show that there is a streaming algorithm which can compute a low rank tensor approximation.

In the following, we introduce the turnstile streaming model and the turnstile streaming tensor Frobenius norm low rank approximation problem.
The following gives a formal definition of the computation model we study.
\begin{definition}[Turnstile model]\label{def:model_turnstile}
Initially, tensor $A\in\mathbb{R}^{n\times n\times n}$ is an all zero tensor. In the turnstile streaming model, there is a stream of update operations, and the $i^{\text{th}}$ update operation is in the form $(x_i,y_i,z_i,\delta_i)$ where $x_i,y_i,z_i\in[n],$ and $\delta_i\in\mathbb{R}$ has $O(\log n)$ bits. Each $(x_i,y_i,z_i,\delta_i)$ means that $A_{x_i,y_i,z_i}$ should be incremented by $\delta_i$. And each entry of $A$ has at most $O(\log n)$ bits at the end of the stream. An algorithm in this computation model is only allowed one pass over the stream. At the end of the stream, the algorithm stores a summary of $A$. The space complexity of the algorithm is the total number of words required to compute and store this summary while scanning the stream. Here, each word has at most $O(\log(n))$ bits.
\end{definition}
The following is the formal definition of the problem.
\begin{definition}[Turnstile model Frobenius norm rank-$k$ tensor approximation]\label{def:streaming}
Given tensor $A\in\mathbb{R}^{n\times n\times n}$, $k \in \mathbb{N}_+$ and an error parameter $1>\varepsilon>0$, the goal is to design an algorithm in the streaming model of Definition~\ref{def:model_turnstile} such that
\begin{enumerate}
\item Upon termination, the algorithm outputs three matrices $U^*,V^*,W^*\in \mathbb{R}^{n\times k}$.
\item $U^*,V^*,W^*$ satisft that
$$\left\|\sum_{i=1}^k U_i^*\otimes V_i^*\otimes W_i^* - A \right\|_F^2\leq (1+\varepsilon) \min_{\rank-k~A'} \|A'-A\|_F^2.$$
\item The space complexity of the algorithm is as small as possible.
\end{enumerate}
\end{definition}

\begin{theorem}
Suppose tensor $A\in\mathbb{R}^{n\times n\times n}$ is given in the turnstile streaming model (see Definition~\ref{def:model_turnstile}), there is an streaming algorithm (in Algorithm~\ref{alg:turnstile}) which solves the problem in Definition~\ref{def:streaming} with constant success probability. In addition, the space complexity of the algorithm is $\poly(k/\varepsilon)+O(nk/\varepsilon)$ words.
\end{theorem}

\begin{proof}
\textbf{Correctness.} Similar to the distributed protocol, the correctness of this streaming algorithm is also implied by Algorithm~\ref{alg:f_main_algorithm} and Algorithm~\ref{alg:f_input_sparsity_reduction} (Theorem~\ref{thm:f_main_algorithm}.) Notice that at the end of the stream $V_1=A_1S_1 \in \mathbb{R}^{n\times s_1}, V_2=A_2S_2 \in \mathbb{R}^{n\times s_2},V_3=A_3S_3\in \mathbb{R}^{n\times s_3}, C=A(T_1,T_2,T_3)\in \mathbb{R}^{t_1 \times t_2 \times t_3}$. It also means that
\begin{align*}
Y_1=T_1A_1S_1,Y_2=T_2A_2S_2,Y_3=T_3A_3S_3.
\end{align*}
According to line~\ref{sta:stream_compute_Xstars} of procedure \textsc{TurnstileStreaming},
\begin{align*}
X^*_1,X^*_2,X^*_3=\underset{X_1\in \mathbb{R}^{s_1 \times k},X_2 \in \mathbb{R}^{s_2 \times k},X_3 \in \mathbb{R}^{s_3 \times k} }{\arg\min} \left\|\overset{k}{\underset{j=1}{\sum}}(Y_1X_1)_j\otimes (Y_2X_2)_j\otimes (Y_3X_3)_j-C \right\|_F
\end{align*}
According to Lemma~\ref{lem:f_input_sparsity_reduction}, we have
\begin{align*}
& ~ \left\|\overset{k}{\underset{j=1}{\sum}}(Y_1X_1)_j\otimes (Y_2X_2)_j\otimes (Y_3X_3)_j-C\right\|_F^2\\
=& ~ \left\|\overset{k}{\underset{j=1}{\sum}}(T_1A_1S_1X^*_1)_j\otimes (T_2A_2S_2X^*_2)_j\otimes (T_3A_3S_3X^*_3)_j-A(T_1,T_2,T_3) \right\|_F^2\\
\leq& ~ (1+O(\varepsilon))\underset{X_1,X_2,X_3}{\min}\left\|\overset{k}{\underset{j=1}{\sum}}(A_1S_1X_1)_j\otimes (A_2S_2X_2)_j\otimes (A_3Y_3X_3)_j-A \right\|_F^2\\
\leq& ~ (1+O(\varepsilon))\underset{U,V,W}{\min} \left\|\sum_{i=1}^k U_i\otimes V_i\otimes W_i-A \right\|_F^2,
\end{align*}
where the last inequality follows by the proof of Theorem~\ref{thm:f_main_algorithm}. By scaling a constant of $\varepsilon$, we complete the proof of correctness.

\textbf{Space complexity.} Since $S_1,S_2,S_3$ are $w_1$-wise independent, and $T_1,T_2,T_3$ are $w_2$-wise independent, the space needed to construct these sketching matrices in line~\ref{sta:gen_seed1} and line~\ref{sta:gen_seed2} of procedure \textsc{TurnstileStreaming} is $O(w_1+w_2)$ words, where $w_1=O(k),w_2=O(1)$ (see~\cite{kvw14,cw13,w14,kn14}). The cost to maintain $V_1,V_2,V_3$ is $O(nk/\varepsilon)$ words, and the cost to maintain $C$ is $\poly(k/\varepsilon)$ words.

Notice that, since each entry of $A$ has at most $O(\log(sn))$ bits, each entry of $Y_1,Y_2,Y_3,C$ has at most $O(\log(sn))$ bits. Due to Theorem~\ref{thm:distributed_solve_small}, each entry of $X_1^*,X_2^*,X_3^*$ has at most $O(\log(sn))$ bits, and the sizes of $X_1^*,X_2^*,X_3^*$ are $\poly(k/\varepsilon)$ words. Thus the space cost in line~\ref{sta:stream_compute_Xstars} is $\poly(k/\varepsilon)$ words.

The total space cost is $\poly(k/\varepsilon)+O(nk/\varepsilon)$ words.
\end{proof}

\begin{algorithm}[t!]\caption{Turnstile Frobenius Norm Low Rank Approximation Algorithm}\label{alg:turnstile}
\begin{algorithmic}[1]
\Procedure{TurnstileStreaming}{$k$,${\cal S}$}
\State $s_1 \leftarrow s_2 \leftarrow s_3 \leftarrow O(k/\epsilon)$.
\State Construct sketching matrices $S_i\in\mathbb{R}^{n^2\times s_i}, \forall i\in[3]$ where entries of $S_1,S_2,S_3$ are $w_1$-wise independent random $N(0,1/s_i)$ Gaussian variables. \label{sta:gen_seed1}
\State $t_1 \leftarrow t_2 \leftarrow t_3 \leftarrow \poly(k/\epsilon)$.
\State Construct sparse embedding matrices $T_i\in\mathbb{R}^{t_i\times n},\forall i\in[3]$ where entries are $w_2$-wise independent. \label{sta:gen_seed2}
\State Initialize matrices:
\State $V_i\leftarrow \{0\}^{n\times s_i},\forall i\in [3]$.
\State $C\leftarrow \{0\}^{t_1\times t_2\times t_3}$
\For{$i\in [l]$}
    \State Receive update operation $(x_i,y_i,z_i,\delta_i)$ from the data stream ${\cal S}$.
    \For{$r=1\to s_1$}
        \State ${(V_1)}_{x_i,r}\leftarrow {(V_1)}_{x_i,r}+ \delta_i \cdot {(S_1)}_{(y_i-1)n+z_i,r}$.
    \EndFor
    \For{$r=1\to s_2$}
        \State ${(V_2)}_{y_i,r}\leftarrow {(V_2)}_{y_i,r}+ \delta_i \cdot {(S_2)}_{(z_i-1)n+x_i,r}$.
    \EndFor
    \For{$r=1\to s_3$}
        \State ${(V_3)}_{z_i,r}\leftarrow {(V_3)}_{z_i,r}+ \delta_i \cdot {(S_3)}_{(x_i-1)n+y_i,r}$.
    \EndFor
    \For{$r=1\to t_1,p=1\to t_2,q=1\to t_3$}
        \State $C_{r,p,q}\leftarrow C_{r,p,q}+\delta_i\cdot (T_1)_{r,x_i}(T_2)_{p,y_i}(T_3)_{q,z_i}$.
    \EndFor
\EndFor
\State Compute $Y_1\leftarrow T_1 V_1,Y_2\leftarrow T_2V_2,Y_3\leftarrow T_3 V_3$.
\State Compute $X^*_i\in\mathbb{R}^{s_i \times k},\forall i\in [3]$ by solving
\State $\underset{X_1,X_2,X_3}{\min}\|(Y_1X_1)\otimes (Y_2X_2)\otimes (Y_3X_3)-C\|_F$ \label{sta:stream_compute_Xstars}
\State Compute $U^*\leftarrow V_1X^*_1,V^*\leftarrow V_2X^*_2,\ W^*\leftarrow V_3X^*_3$.
\State \Return $U^*,V^*,W^*$
\EndProcedure
\end{algorithmic}
\end{algorithm}

\begin{remark}
In the Algorithm~\ref{alg:turnstile}, for each update operation, we need $O(k/\varepsilon)$ time to maintain matrices $V_1,V_2,V_3$, and we need $\poly(k/\varepsilon)$ time to maintain tensor $C$. Thus the update time is $\poly(k/\varepsilon)$. At the end of the stream, the time to compute
\begin{align*}
X^*_1,X^*_2,X^*_3=\underset{X_1,X_2,X_3\in\mathbb{R}^{O(k/\varepsilon)\times k}}{\arg\min} \left\|\overset{k}{\underset{j=1}{\sum}}(Y_1X_1)_j\otimes (Y_2X_2)_j\otimes (Y_3X_3)_j-C \right\|_F,
\end{align*}
is exponential in $\poly(k/\varepsilon)$ running time since it should use a polynomial system solver. Instead of computing the rank-$k$ solution, we can solve the following:
\begin{align*}
\alpha^*=\underset{\alpha\in\mathbb{R}^{s_1\times s_2 \times s_3}}{\arg\min} \left\|\sum_{i=1}^{s_1}\sum_{j=1}^{s_2}\sum_{l=1}^{s_3}\alpha_{i,j,l}\cdot(Y_1)_i\otimes(Y_2)_j\otimes(Y_3)_l-C\right\|_F
\end{align*}
which will then give
 \begin{align*}
 \sum_{i=1}^{s_1 }\sum_{j=1}^{s_2}\sum_{l=1}^{s_3}\alpha^*_{i,j,l}\cdot(Y_1)_i\otimes(Y_2)_j\otimes(Y_3)_l
 \end{align*}
 to be a rank-$O(k^3/\varepsilon^3)$ bicriteria solution.

Further, similar to Theorem~\ref{thm:f_bicriteria_better}, we can solve
\begin{align*}
\min_{U\in \mathbb{R}^{n\times s_2 s_3}} \left\| \sum_{i=1}^{s_1} \sum_{j=1}^{s_2} U_{i+s_1(j-1)} \otimes (Y_2)_i \otimes (Y_3)_j   -C \right\|_F
\end{align*}
where $C= \sum_{i} A_i(I,T_2,T_3)$. Thus, we can obtain a $\rank$-$O(k^2/\epsilon^2)$ in polynomial time.
\end{remark}

\begin{remark}
If we choose $S_1,S_2,S_3,T_1,T_2,T_3$ to be random Cauchy matrices, then we are able to apply the entry-wise $\ell_1$ norm low rank tensor approximation algorithm (see Theorem~\ref{thm:l1_input_sparsity_time}) in turnstile model.
\end{remark}

\newpage
\section{Extension to Other Tensor Ranks}\label{sec:other_rank}
The tensor rank studied in the previous sections is also called the CP rank or canonical rank. The tensor rank can be thought of as a direct extension of the matrix rank. We would like to point out that there are other definitions of tensor rank, e.g., the tucker rank and train rank. In this section we explain how to extend our proofs to other notions of tensor rank. Section~\ref{sec:tucker} provides the extension to tucker rank, and Section~\ref{sec:train} provides the extension to train rank.

\subsection{Tensor Tucker rank}\label{sec:tucker}

Tensor Tucker rank has been studied in a number of works \cite{kc07,pc08,mh09,zw13,yc14}. We provide the formal definition here:
\subsubsection{Definitions}
\begin{definition}[Tucker rank]
Given a third order tensor $A\in \mathbb{R}^{n \times n \times n}$, we say $A$ has tucker rank $k$ if $k$ is the smallest integer such that there exist three matrices $U,V,W\in \mathbb{R}^{n\times k}$ and a (small) tensor $C\in \mathbb{R}^{k \times k \times k}$ satisfying
\begin{align*}
A_{i,j,l} = \sum_{i'=1}^k \sum_{j'=1}^k \sum_{l'=1}^k C_{i',j',l'} U_{i,i'} V_{j,j'} W_{l,l'}, \forall i,j,l \in [n] \times [n] \times [n],
\end{align*}
or equivalently,
\begin{align*}
A = C(U,V,W).
\end{align*}
\end{definition}

\subsubsection{Algorithm}

\begin{algorithm}[!]\caption{Frobenius Norm Low (Tucker) Rank Approximation}\label{alg:tucker_main_algorithm}
\begin{algorithmic}[1]
\Procedure{\textsc{FLowTuckerRankApprox}}{$A,n,k,\epsilon$} \Comment{Theorem \ref{thm:tucker_main_algorithm}}
\State $s_1\leftarrow s_2 \leftarrow s_3 \leftarrow O(k/\epsilon)$.
\State $t_1\leftarrow t_2 \leftarrow t_3 \leftarrow \poly(k,1/\epsilon)$.
\State Choose sketching matrices $S_1\in \mathbb{R}^{n^2 \times s_1}$, $S_2 \in \mathbb{R}^{n^2 \times s_2}$, $S_3\in \mathbb{R}^{n^2 \times s_3}$. \Comment{Definition \ref{def:fast_gaussian_transform}}
\State Choose sketching matrices $T_1\in \mathbb{R}^{t_1 \times n}$, $T_2 \in \mathbb{R}^{t_2 \times n}$, $T_3 \in \mathbb{R}^{t_3 \times n}$.
\State Compute $A_i S_i,\forall i\in [3]$. \label{sta:tucker_compute_AiSi}
\State Compute $T_i A_i S_i$, $\forall i\in [3]$. \label{sta:tucker_compute_TiAiSi}
\State Compute $B \leftarrow A(T_1,T_2,T_3)$. \label{sta:tucker_compute_AT1T2T3}
\State Create variables for $X_i \in \mathbb{R}^{s_i \times k}, \forall i\in [3]$.
\State Create variables for $C \in \mathbb{R}^{k\times k\times k}$.
\State Run a polynomial system verifier for $\| C( (Y_1 X_1) , (Y_2 X_2),  (Y_3 X_3) ) - B\|_F^2$.\label{sta:tucker_solve_small_problem}
\State \Return $C$,$A_1 S_1 X_1$, $A_2 S_2 X_2$, and $A_3S_3 X_3$.
\EndProcedure
\end{algorithmic}
\end{algorithm}

\begin{theorem}\label{thm:tucker_main_algorithm}
Given a third order tensor $A\in \mathbb{R}^{n\times n\times n}$, for any $k\geq 1$ and $\epsilon \in (0,1)$, there exists an algorithm which takes $O(\nnz(A)) + n \poly(k,1/\epsilon) + 2^{O(k^2/\epsilon +k^3)}$ time and outputs three matrices $U,V,W\in \mathbb{R}^{n\times k}$, and a tensor $C\in \mathbb{R}^{k\times k \times k}$ for which 
\begin{align*}
\left\| C(U,V,W) - A \right\|_F^2 \leq (1+\epsilon) \underset{\tucker~\rank-k~ A_k }{\min} \| A_k - A\|_F^2
\end{align*}
holds with probability $9/10$.
\end{theorem}
\begin{proof}
We define $\OPT$ to be 
\begin{align*}
\OPT=\underset{\tucker~\rank-k~A'}{\min} \| A' -A \|_F^2.
\end{align*}

Suppose the optimal $A_k= C^* (U^*,V^*,W^*).$
We fix $C^* \in \mathbb{R}^{k\times k \times k}$, $V^* \in \mathbb{R}^{n\times k}$ and $W^* \in \mathbb{R}^{n\times k}$. We use $V_1^*, V_2^*, \cdots, V_k^*$ to denote the columns of $V^*$ and $W_1^*, W_2^*, \cdots, W_k^*$ to denote the columns of $W^*$.

We consider the following optimization problem,
\begin{align*}
\min_{U_1, \cdots, U_k \in \mathbb{R}^n } \left\| C^*(U,V^*,W^*) - A \right\|_F^2,
\end{align*}
which is equivalent to
\begin{align*}
\min_{U_1, \cdots, U_k \in \mathbb{R}^n } \left\| U \cdot C^*(I,V^*,W^*) -A \right\|_F^2,
\end{align*}
because $C^*(U,V^*,W^*) = U \cdot C^*(I,V^*,W^*) $ according to Definition~\ref{def:bracket}.

Recall that $C^*(I,V^*,W^*)$ denotes a $k\times n\times n$ tensor. Let $ ( C^*(I,V^*,W^*) )_1$ denote the matrix obtained by flattening $C^*(I,V^*,W^*)$ along the first dimension.
We use matrix $Z_1$ to denote $ ( C^*(I,V^*,W^*) )_1  \in \mathbb{R}^{k\times n^2}$. Then we can obtain the following equivalent objective function,
\begin{align*}
\min_{U \in \mathbb{R}^{n\times k} } \| U Z_1  - A_1 \|_F^2.
\end{align*}
Notice that $\min_{U \in \mathbb{R}^{n\times k} } \| U Z_1  - A_1 \|_F^2=\OPT$, since $A_k=U^*Z_1$.

Let $S_1^\top \in\mathbb{R}^{s_1\times n^2}$ be the sketching matrix defined in Definition~\ref{def:fast_gaussian_transform}, where $s_1=O(k/\epsilon)$. We obtain the following optimization problem,
\begin{align*}
\min_{U \in \mathbb{R}^{n\times k} } \| U Z_1 S_1 - A_1 S_1 \|_F^2.
\end{align*}
Let $ \wh{U} \in \mathbb{R}^{n\times k}$ denote the optimal solution to the above optimization problem. Then $\wh{U} = A_1 S_1 (Z_1 S_1)^\dagger$. By Lemma~\ref{lem:gaussian_count_sketch_for_regression} and Theorem~\ref{thm:multiple_regression_sketch}, we have

\begin{align*}
\| \wh{U} Z_1  - A_1  \|_F^2 \leq (1+\epsilon) \underset{U\in \mathbb{R}^{n\times k}}{\min} \| U Z_1 - A_1 \|_F^2 = (1+\epsilon) \OPT,
\end{align*}

which implies
\begin{align*}
\left\| C^*( \wh{U}, V^* , W^*) - A \right\|_F^2 \leq (1+\epsilon) \OPT.
\end{align*}
To write down $\wh{U}_1, \cdots, \wh{U}_k$, we use the given matrix $A_1$, and we create $s_1 \times k$ variables for matrix $(Z_1 S_1)^\dagger$.

As our second step, we fix $\wh{U} \in \mathbb{R}^{n\times k}$ and $W^* \in \mathbb{R}^{n\times k}$, and we convert tensor $A$ into matrix $A_2$. Let matrix $Z_2$ denote $(C^*(\wh{U},I,W^*))_2 \in \mathbb{R}^{k\times n^2}$. We consider the following objective function,
\begin{align*}
\min_{V \in \mathbb{R}^{n\times k} } \| V Z_2 -A_2  \|_F^2,
\end{align*}
for which the optimal cost is at most $(1+\epsilon) \OPT$.

Let $S_2^\top \in\mathbb{R}^{s_2\times n^2}$ be a sketching matrix defined in Definition~\ref{def:fast_gaussian_transform}, where $s_2=O(k/\varepsilon)$.
We sketch $S_2$ on the right of the objective function to obtain a new objective function,
\begin{align*}
\underset{V\in \mathbb{R}^{n\times k} }{\min} \| V Z_2 S_2 - A_2 S_2 \|_F^2.
\end{align*}
Let $\wh{V} \in \mathbb{R}^{n\times k}$ denote the optimal solution to the above problem. Then $\wh{V} = A_2 S_2 (Z_2 S_2)^\dagger$. By Lemma~\ref{lem:gaussian_count_sketch_for_regression} and Theorem~\ref{thm:multiple_regression_sketch}, we have,
\begin{align*}
\| \wh{V} Z_2 - A_2 \|_F^2 \leq (1+\epsilon ) \underset{V\in \mathbb{R}^{n\times k} }{\min} \| V Z_2  - A_2 \|_F^2 \leq  (1+\epsilon)^2 \OPT,
\end{align*}
which implies
\begin{align*}
\left\| C^*( \wh{U} ,\wh{V}, W^*) - A \right\|_F^2 \leq (1+\epsilon )^2 \OPT.
\end{align*}
To write down $\wh{V}_1, \cdots, \wh{V}_k$, we need to use the given matrix $A_2 \in \mathbb{R}^{n^2 \times n}$, and we need to create $s_2\times k$ variables for matrix $(Z_2 S_2)^\dagger$.

As our third step, we fix the matrices $\wh{U} \in \mathbb{R}^{n\times k}$ and $\wh{V}\in \mathbb{R}^{n \times k}$. We convert tensor $A\in \mathbb{R}^{n\times n \times n}$ into matrix $A_3 \in \mathbb{R}^{n^2 \times n}$. Let matrix $Z_3$ denote $ ( C^*(\wh{U},\wh{V},I) )_3 \in \mathbb{R}^{k\times n^2}$. We consider the following objective function,
\begin{align*}
\underset{W\in \mathbb{R}^{n\times k} }{\min} \| W Z_3 - A_3 \|_F^2,
\end{align*}
which has optimal cost at most $(1+\epsilon)^2 \OPT$.

Let $S_3^\top \in\mathbb{R}^{s_3\times n^2}$ be a sketching matrix defined in Definition~\ref{def:fast_gaussian_transform}, where $s_3=O(k/\varepsilon)$.
We sketch $S_3$ on the right of the objective function to obtain a new objective function,
\begin{align*}
\underset{ W \in \mathbb{R}^{n\times k} }{ \min } \| W Z_3 S_3 - A_3 S_3 \|_F^2.
\end{align*}
Let $\wh{W} \in \mathbb{R}^{n\times k}$ denote the optimal solution of the above problem. Then $\wh{W} = A_3 S_3 (Z_3 S_3)^\dagger$. By Lemma~\ref{lem:gaussian_count_sketch_for_regression} and Theorem~\ref{thm:multiple_regression_sketch}, we have,
\begin{align*}
\| \wh{W} Z_3 - A_3 \|_F^2 \leq (1+\epsilon) \underset{W\in \mathbb{R}^{n\times k} }{\min} \| W Z_3 - A_3 \|_F^2 \leq (1+\epsilon)^3 \OPT.
\end{align*}
Thus, we have
\begin{align*}
\min_{X_1,X_2,X_3} \left\| C^* ( (A_1 S_1 X_1) , (A_2S_2 X_2) , (A_3S_3 X_3) ) - A \right\|_F^2 \leq (1+\epsilon)^3 \OPT.
\end{align*}
Let $V_1=A_1S_1,V_2=A_2S_2$, and $V_3=A_3S_3.$ We then apply Lemma \ref{lem:f_input_sparsity_reduction}, and we obtain $\wh{V}_1,\wh{V}_2,\wh{V}_3,B$. We then apply Theorem~\ref{thm:f_solving_small_problems}. Correctness follows by rescaling $\varepsilon$ by a constant factor.

\paragraph{Running time.} Due to Definition~\ref{def:fast_gaussian_transform}, the running time of line~\ref{sta:train_compute_AiSi} (Algorithm~\ref{alg:tucker_main_algorithm}) is $O(\nnz(A))+n\poly(k,1/\epsilon)$. Due to Lemma~\ref{lem:f_input_sparsity_reduction}, line~\ref{sta:tucker_compute_TiAiSi} and \ref{sta:tucker_compute_AT1T2T3} can be executed in $\nnz(A) + n\poly(k,1/\epsilon)$ time.
The running time of line~\ref{sta:tucker_solve_small_problem} is given by Theorem~\ref{thm:f_solving_small_problems}. (For simplicity, we ignore the bit complexity in the running time.)
\end{proof}

\subsection{Tensor Train rank}\label{sec:train}

\subsubsection{Definitions}

The tensor train rank has been studied in several works \cite{o11,otz11,zwz16,ptbd16}. We provide the formal definition here.
\begin{definition}[Tensor Train rank]
Given a third order tensor $A\in \mathbb{R}^{n \times n \times n}$, we say $A$ has train rank $k$ if $k$ is the smallest integer such that there exist three tensors $U \in \mathbb{R}^{1 \times n \times k}$, $V\in \mathbb{R}^{k\times n \times k}$, $W\in \mathbb{R}^{k\times n\times 1}$ satisfying:
\begin{align*}
A_{i,j,l} = \sum_{i_1=1}^1 \sum_{i_2=1}^k \sum_{i_3=1}^k \sum_{i_4=1}^1 U_{i_1,i,i_2} V_{i_2,j,i_3} W_{i_3,l,i_4}, \forall i,j,l \in [n] \times [n] \times [n],
\end{align*}
or equivalently,
\begin{align*}
A_{i,j,l} =  \sum_{i_2=1}^k \sum_{i_3=1}^k  (U_2)_{i,i_2} (V_2)_{j, i_2 + k(i_3-1) } (W_2)_{l,i_3},
\end{align*}
where $V_2\in \mathbb{R}^{n\times k^2}$ denotes the matrix obtained by flattening the tensor $U$ along the second dimension, and $(V_2)_{i,i_1+k(i_2-1)}$ denotes the entry in the $i$-th row and $i_1+k(i_2-1)$-th column of $V_2$. We similarly define $U_2, W_2 \in \mathbb{R}^{n\times k}$.
\end{definition}

\begin{algorithm}[!]\caption{Frobenius Norm Low (Train) rank Approximation}\label{alg:train_main_algorithm}
\begin{algorithmic}[1]
\Procedure{\textsc{FLowTrainRankApprox}}{$A,n,k,\epsilon$} \Comment{Theorem \ref{thm:train_main_algorithm}}
\State $s_1 \leftarrow s_3 \leftarrow O(k/\epsilon)$.
\State $s_2 \leftarrow O(k^2/\epsilon)$.
\State $t_1\leftarrow t_2 \leftarrow t_3 \leftarrow \poly(k,1/\epsilon)$.
\State Choose sketching matrices $S_1\in \mathbb{R}^{n^2 \times s_1}$, $S_2 \in \mathbb{R}^{n^2 \times s_2}$, $S_3\in \mathbb{R}^{n^2 \times s_3}$. \Comment{Definition \ref{def:fast_gaussian_transform}}
\State Choose sketching matrices $T_1\in \mathbb{R}^{t_1 \times n}$, $T_2 \in \mathbb{R}^{t_2 \times n}$, $T_3 \in \mathbb{R}^{t_3 \times n}$.
\State Compute $A_i S_i,\forall i\in [3]$. \label{sta:train_compute_AiSi}
\State Compute $T_i A_i S_i$, $\forall i\in [3]$. \label{sta:train_compute_TiAiSi}
\State Compute $B \leftarrow A(T_1,T_2,T_3)$. \label{sta:train_compute_AT1T2T3}
\State Create variables for $X_1 \in \mathbb{R}^{s_1 \times k}$.
\State Create variables for $X_3 \in \mathbb{R}^{s_3 \times k}$.
\State Create variables for $X_2 \in \mathbb{R}^{s_2 \times k^2}$.
\State Create variables for $C \in \mathbb{R}^{k\times k\times k}$.
\State Run polynomial system verifier for $\|\sum_{i_2=1}^k \sum_{i_3=1}^k (Y_1 X_1)_{i_2} (Y_2 X_2)_{i_2 + k(i_3-1)}  (Y_3 X_3)_{i_3}  - B\|_F^2$.\label{sta:train_solve_small_problem}
\State \Return $A_1 S_1 X_1$, $A_2 S_2 X_2$, and $A_3S_3 X_3$.
\EndProcedure
\end{algorithmic}
\end{algorithm}

\subsubsection{Algorithm}

\begin{theorem}\label{thm:train_main_algorithm}
Given a third order tensor $A\in \mathbb{R}^{n\times n\times n}$, for any $k\geq 1$, $\epsilon \in (0,1)$, there exists an algorithm which takes $O(\nnz(A)) + n \poly(k,1/\epsilon) + 2^{O(k^4/\epsilon )}$ time and outputs three tensors $U\in \mathbb{R}^{1\times n \times k}$, $V\in \mathbb{R}^{k\times n\times k }$, $W\in \mathbb{R}^{k\times n\times 1}$ such that
\begin{align*}
\left\| \sum_{i=1}^k \sum_{j=1}^k  (U_2)_{i} \otimes (V_2)_{ i + k(j -1) } \otimes (W_2)_{j} - A \right\|_F^2 \leq (1+\epsilon) \underset{\train~\rank-k~ A_k }{\min} \| A_k - A\|_F^2
\end{align*}
holds with probability $9/10$.
\end{theorem}

\begin{proof}

We define $\OPT$ as
\begin{align*}
\OPT=\underset{\train~\rank-k~A'}{\min} \| A' -A \|_F^2.
\end{align*}

Suppose the optimal
\begin{align*}
A_k=  \sum_{i=1}^k \sum_{j=1}^k  U^*_{i} \otimes V^*_{ i + k(j-1)} \otimes W^*_{j}.
\end{align*}
We fix  $V^* \in \mathbb{R}^{n\times k^2}$ and $W^* \in \mathbb{R}^{n\times k}$. We use $V_1^*, V_2^*, \cdots, V_{k^2}^*$ to denote the columns of $V^*$, and $W_1^*, W_2^*, \cdots, W_{k}^*$ to denote the columns of $W^*$.

We consider the following optimization problem,
\begin{align*}
\min_{U \in \mathbb{R}^{n\times k} } \left\|  \sum_{i=1}^k \sum_{j=1}^k  U_{i} \otimes V^*_{ i + k(j-1)} \otimes W^*_{j} - A \right\|_F^2,
\end{align*}
which is equivalent to
\begin{align*}
\min_{U \in \mathbb{R}^{n\times k} } \left\| U \cdot
\begin{bmatrix}
\overset{k}{ \underset{j=1}{\sum} } V_{1 + k(j-1)}^* \otimes W^*_{j} \\
\overset{k}{ \underset{j=1}{\sum} } V_{2 + k(j-1)}^* \otimes W^*_{j} \\
\cdots \\
\overset{k}{ \underset{j=1}{\sum} } V_{k + k(j-1)}^* \otimes W^*_{j} \\
\end{bmatrix} -A \right\|_F^2.
\end{align*}

Let $A_1 \in \mathbb{R}^{n\times n^2}$ denote the matrix obtained by flattening the tensor $A$ along the first dimension.
We use matrix $Z_1 \in \mathbb{R}^{k\times n^2}$ to denote
\begin{align*}
\begin{bmatrix}
\overset{k}{ \underset{j=1}{\sum} } \vect( V_{1 + k(j-1)}^* \otimes W^*_{j} ) \\
\overset{k}{ \underset{j=1}{\sum} } \vect( V_{2 + k(j-1)}^* \otimes W^*_{j} ) \\
\cdots \\
\overset{k}{ \underset{j=1}{\sum} } \vect( V_{k + k(j-1)}^* \otimes W^*_{j} )
\end{bmatrix}.
\end{align*}
 Then we can obtain the following equivalent objective function,
\begin{align*}
\min_{U \in \mathbb{R}^{n\times k} } \| U Z_1  - A_1 \|_F^2.
\end{align*}
Notice that $\min_{U \in \mathbb{R}^{n\times k} } \| U Z_1  - A_1 \|_F^2=\OPT$, since $A_k=U^*Z_1$.

Let $S_1^\top \in\mathbb{R}^{s_1\times n^2}$ be a sketching matrix defined in Definition~\ref{def:fast_gaussian_transform}, where $s_1=O(k/\epsilon)$. We obtain the following optimization problem,
\begin{align*}
\min_{U \in \mathbb{R}^{n\times k} } \| U Z_1 S_1 - A_1 S_1 \|_F^2.
\end{align*}
Let $ \wh{U} \in \mathbb{R}^{n\times k}$ denote the optimal solution to the above optimization problem. Then $\wh{U} = A_1 S_1 (Z_1 S_1)^\dagger$. By Lemma~\ref{lem:gaussian_count_sketch_for_regression} and Theorem~\ref{thm:multiple_regression_sketch}, we have

\begin{align*}
\| \wh{U} Z_1  - A_1  \|_F^2 \leq (1+\epsilon) \underset{U\in \mathbb{R}^{n\times k}}{\min} \| U Z_1 - A_1 \|_F^2 = (1+\epsilon) \OPT,
\end{align*}

which implies
\begin{align*}
\left\| \sum_{i=1}^k \sum_{j=1}^k \wh{U}_i \otimes V^*_{i+ k(j-1)} \otimes W^*_j - A \right\|_F^2 \leq (1+\epsilon) \OPT.
\end{align*}
To write down $\wh{U}_1, \cdots, \wh{U}_k$, we use the given matrix $A_1$, and we create $s_1 \times k$ variables for matrix $(Z_1 S_1)^\dagger$.

As our second step, we fix $\wh{U} \in \mathbb{R}^{n\times k}$ and $W^* \in \mathbb{R}^{n\times k}$, and we convert the tensor $A$ into matrix $A_2$. Let matrix $Z_2 \in \mathbb{R}^{k^2\times n^2}$ denote the matrix where the $(i,j)$-th row is the vectorization of $\wh{U}_i \otimes W^*_j$.
 We consider the following objective function,
\begin{align*}
\min_{V \in \mathbb{R}^{n\times k} } \| V Z_2 -A_2  \|_F^2,
\end{align*}
for which the optimal cost is at most $(1+\epsilon) \OPT$.

Let $S_2^\top \in\mathbb{R}^{s_2\times n^2}$ be a sketching matrix defined in Definition~\ref{def:fast_gaussian_transform}, where $s_2=O(k^2/\epsilon)$.
We sketch $S_2$ on the right of the objective function to obtain the new objective function,
\begin{align*}
\underset{V\in \mathbb{R}^{n\times k} }{\min} \| V Z_2 S_2 - A_2 S_2 \|_F^2.
\end{align*}
Let $\wh{V} \in \mathbb{R}^{n\times k}$ denote the optimal solution of the above problem. Then $\wh{V} = A_2 S_2 (Z_2 S_2)^\dagger$. By Lemma~\ref{lem:gaussian_count_sketch_for_regression} and Theorem~\ref{thm:multiple_regression_sketch}, we have,
\begin{align*}
\| \wh{V} Z_2 - A_2 \|_F^2 \leq (1+\epsilon ) \underset{V\in \mathbb{R}^{n\times k} }{\min} \| V Z_2  - A_2 \|_F^2 \leq  (1+\epsilon)^2 \OPT,
\end{align*}
which implies
\begin{align*}
\left\| \sum_{i=1}^k \sum_{j=1}^k \wh{U}_i \otimes \wh{V}_{i+k(j-1)} \otimes W^* - A \right\|_F^2 \leq (1+\epsilon )^2 \OPT.
\end{align*}
To write down $\wh{V}_1, \cdots, \wh{V}_k$, we need to use the given matrix $A_2 \in \mathbb{R}^{n^2 \times n}$, and we need to create $s_2\times k$ variables for matrix $(Z_2 S_2)^\dagger$.

As our third step, we fix the matrices $\wh{U} \in \mathbb{R}^{n\times k}$ and $\wh{V}\in \mathbb{R}^{n \times k}$. We convert tensor $A\in \mathbb{R}^{n\times n \times n}$ into matrix $A_3 \in \mathbb{R}^{n^2 \times n}$. Let matrix $Z_3\in \mathbb{R}^{k\times n^2}$ denote
\begin{align*}
\begin{bmatrix}
\sum_{i=1}^k \vect(\wh{U}_i \otimes \wh{V}_{i+k \cdot 0}) \\
\sum_{i=1}^k \vect(\wh{U}_i \otimes \wh{V}_{i+k \cdot 1}) \\
\cdots \\
\sum_{i=1}^k \vect(\wh{U}_i \otimes \wh{V}_{i+k \cdot  (k-1)})
\end{bmatrix}
.
\end{align*}
We consider the following objective function,
\begin{align*}
\underset{W\in \mathbb{R}^{n\times k} }{\min} \| W Z_3 - A_3 \|_F^2,
\end{align*}
which has optimal cost at most $(1+\epsilon)^2 \OPT$.

Let $S_3^\top \in\mathbb{R}^{s_3\times n^2}$ be a sketching matrix defined in Definition~\ref{def:fast_gaussian_transform}, where $s_3=O(k/\varepsilon)$.
We sketch $S_3$ on the right of the objective function to obtain a new objective function,
\begin{align*}
\underset{ W \in \mathbb{R}^{n\times k} }{ \min } \| W Z_3 S_3 - A_3 S_3 \|_F^2.
\end{align*}
Let $\wh{W} \in \mathbb{R}^{n\times k}$ denote the optimal solution of the above problem. Then $\wh{W} = A_3 S_3 (Z_3 S_3)^\dagger$. By Lemma~\ref{lem:gaussian_count_sketch_for_regression} and Theorem~\ref{thm:multiple_regression_sketch}, we have,
\begin{align*}
\| \wh{W} Z_3 - A_3 \|_F^2 \leq (1+\epsilon) \underset{W\in \mathbb{R}^{n\times k} }{\min} \| W Z_3 - A_3 \|_F^2 \leq (1+\epsilon)^3 \OPT.
\end{align*}
Thus, we have
\begin{align*}
\min_{X_1,X_2,X_3} \left\| \sum_{i=1}^k \sum_{j=1}^k (A_1 S_1 X_1)_i \otimes (A_2S_2 X_2)_{i+k(j-1)} \otimes ( A_3S_3 X_3 )_j - A \right\|_F^2 \leq (1+\epsilon)^3 \OPT.
\end{align*}
Let $V_1=A_1S_1,V_2=A_2S_2,$ and $V_3=A_3S_3.$ We then apply Lemma \ref{lem:f_input_sparsity_reduction}, and we obtain $\wh{V}_1,\wh{V}_2,\wh{V}_3,B$. We then apply Theorem~\ref{thm:f_solving_small_problems}.  Correctness follows by rescaling $\epsilon$ by a constant factor.

\paragraph{Running time.} Due to Definition~\ref{def:fast_gaussian_transform}, the running time of line~\ref{sta:train_compute_AiSi} (Algorithm~\ref{alg:train_main_algorithm}) is $O(\nnz(A))+n\poly(k,1/\epsilon)$. Due to Lemma~\ref{lem:f_input_sparsity_reduction}, lines~\ref{sta:train_compute_TiAiSi} and \ref{sta:train_compute_AT1T2T3} can be executed in $\nnz(A) + n\poly(k,1/\epsilon)$ time. The running time of $2^{O(k^4/\epsilon)}$ comes from running Theorem~\ref{thm:f_solving_small_problems} (For simplicity, we ignore the bit complexity in the running time.)
\end{proof}

\newpage
\section{Acknowledgments}
The authors would like to thank Udit Agarwal, Alexandr Andoni, Arturs Backurs, Saugata Basu, Lijie Chen, Xi Chen, Thomas Dillig, Yu Feng, Rong Ge, Daniel Hsu, Chi Jin, Ravindran Kannan, J. M. Landsberg, Qi Lei, Fu Li, Syed Mohammad Meesum, Ankur Moitra, Dana Moshkovitz, Cameron Musco, Richard Peng, Eric Price, Govind Ramnarayan, Ilya Razenshteyn, James Renegar, Rocco Servedio, Tselil Schramm, Clifford Stein, Wen Sun, Yining Wang, Zhaoran Wang, Wei Ye, Huacheng Yu, Huan Zhang, Kai Zhong, David Zuckerman for useful discussions.

\newpage
\addcontentsline{toc}{section}{References}
\bibliographystyle{alpha}
\bibliography{ref}
\newpage


\end{document}